\documentclass[reqno,a4paper,11pt]{amsart}

\usepackage[utf8]{inputenc}
\calclayout

\DeclareFontFamily{U}{mathx}{}
\DeclareFontShape{U}{mathx}{m}{n}{<-> mathx10}{}
\DeclareSymbolFont{mathx}{U}{mathx}{m}{n}
\DeclareMathAccent{\widehat}{0}{mathx}{"70}
\DeclareMathAccent{\widecheck}{0}{mathx}{"71}

\usepackage[utf8]{inputenc}
\usepackage{verbatim}
\usepackage[english]{babel}
\usepackage{amsmath}
\usepackage{amsfonts}
\usepackage{mathtools}
\mathtoolsset{showonlyrefs}
\usepackage{upgreek}
\usepackage{bbm}

\usepackage{cite}

\usepackage{amsthm}
\usepackage{bm}
\usepackage{float}
\usepackage{color}
\usepackage{tikz}
\usetikzlibrary{decorations.pathreplacing,decorations.markings,calc}
\usetikzlibrary{arrows}

\makeatletter
\def\grd@save@target#1{%
  \def\grd@target{#1}}
\def\grd@save@start#1{%
  \def\grd@start{#1}}
\tikzset{
  grid with coordinates/.style={
    to path={%
      \pgfextra{%
        \edef\grd@@target{(\tikztotarget)}%
        \tikz@scan@one@point\grd@save@target\grd@@target\relax
        \edef\grd@@start{(\tikztostart)}%
        \tikz@scan@one@point\grd@save@start\grd@@start\relax
        \draw[minor help lines] (\tikztostart) grid (\tikztotarget);
        \draw[major help lines] (\tikztostart) grid (\tikztotarget);
        \grd@start
        \pgfmathsetmacro{\grd@xa}{\the\pgf@x/1cm}
        \pgfmathsetmacro{\grd@ya}{\the\pgf@y/1cm}
        \grd@target
        \pgfmathsetmacro{\grd@xb}{\the\pgf@x/1cm}
        \pgfmathsetmacro{\grd@yb}{\the\pgf@y/1cm}
        \pgfmathsetmacro{\grd@xc}{\grd@xa + \pgfkeysvalueof{/tikz/grid with coordinates/major step}}
        \pgfmathsetmacro{\grd@yc}{\grd@ya + \pgfkeysvalueof{/tikz/grid with coordinates/major step}}
        \foreach \x in {\grd@xa,\grd@xc,...,\grd@xb}
        \node[anchor=north] at (\x,\grd@ya) {\pgfmathprintnumber{\x}};
        \foreach \y in {\grd@ya,\grd@yc,...,\grd@yb}
        \node[anchor=east] at (\grd@xa,\y) {\pgfmathprintnumber{\y}};
      }
    }
  },
  minor help lines/.style={
    help lines,
    step=\pgfkeysvalueof{/tikz/grid with coordinates/minor step}
  },
  major help lines/.style={
    help lines,
    line width=\pgfkeysvalueof{/tikz/grid with coordinates/major line width},
    step=\pgfkeysvalueof{/tikz/grid with coordinates/major step}
  },
  grid with coordinates/.cd,
  minor step/.initial=.2,
  major step/.initial=1,
  major line width/.initial=2pt,
}

\usepackage{graphicx}
\usepackage{subcaption}
\usepackage{enumerate}
\usepackage{bm}
\usepackage{seqsplit}

\usepackage[pdftex,colorlinks]{hyperref}
\hypersetup{
	colorlinks=true,
	linkcolor=blue,
	citecolor=red
}

\makeatletter
\def\l@subsection{\@tocline{2}{0pt}{2.5pc}{5pc}{}}

\DeclareMathOperator{\ai}{Ai}
\DeclareMathOperator{\re}{Re}
\DeclareMathOperator{\im}{Im}
\DeclareMathOperator{\ee}{\rm e}
\DeclareMathOperator{\supp}{supp}

\newcommand{\res}{\mathop{\rm Res}}

\newcommand{\C}{\mathbb{C}}
\newcommand{\R}{\mathbb{R}}
\newcommand{\Z}{\mathbb{Z}}

\newcommand{\E}{\mathbb{E}}
\newcommand{\boh}{\mathit{o}}
\newcommand{\Boh}{\mathcal{O}}

\newcommand{\ii}{\mathrm{i}}
\newcommand{\dd}{\mathrm{d}}
\newcommand*{\deff}{\mathrel{\vcenter{\baselineskip0.5ex \lineskiplimit0pt
                     \hbox{\scriptsize.}\hbox{\scriptsize.}}}%
                     =}
\newcommand*{\revdeff}{=\mathrel{\vcenter{\baselineskip0.5ex \lineskiplimit0pt
                     \hbox{\scriptsize.}\hbox{\scriptsize.}}}%
                     }

\usepackage{dsfont}
\newcommand{\ind}{{\mathds{1}}}
\newcommand{\ptf}{{\rm P34}}

\DeclareMathOperator{\Jb}{J}

\renewcommand{\P}{{\mathbb P}}
\renewcommand{\bm}{\mathbf}
\newcommand{\mcal}{\mathcal}

\newcommand{\msf}{\mathsf}

\newcommand{\wh}{\widehat}
\newcommand{\wc}{\widecheck}
\newcommand{\wt}{\widetilde}

\renewcommand{\sp}{\boldsymbol \sigma}

\newcommand\floor[1]{\left\lfloor #1 \right\rfloor}

\newcommand{\mn}{N} 
\newcommand{\mm}{M} 
\newcommand{\ms}{\alpha} 
\newcommand{\mt}{\beta} 
\newcommand{\mq}{\msf q} 
\newcommand{\mnu}{\nu} 
\newcommand{\mka}{\kappa_{\mq}} 
\newcommand{\mc}{\msf c_{\mnu}} 

\newcommand{\mw}{\omega_\mn} 

\newcommand{\aeq}{{\msf a_{\mnu}}}  

\newcommand{\Hf}{\mathfrak{H}} 

\newcommand{\gw}{\omega} 
\newcommand{\gW}{\Omega} 
\newcommand{\ga}{\msf a} 
\newcommand{\gb}{\msf b} 
\newcommand{\gV}{\msf V} 
\newcommand{\gn}{N} 
\newcommand{\gC}{\msf C} 
\newcommand{\gE}{\msf E} 
\newcommand{\gEO}{{\msf E}_0} 
\newcommand{\gsig}{\sigma} 
\newcommand{\gwd}{\omega^\sigma} 
\newcommand{\gWd}{{\Omega^\sigma}} 
\newcommand{\gequil}{\mu_\gV} 

\newcommand{\gin}{{\scriptscriptstyle({\scriptscriptstyle\rm in})}} 
\newcommand{\gout}{{\scriptscriptstyle({\scriptscriptstyle\rm out})}} 
\newcommand{\pc}{{\scriptscriptstyle({\scriptscriptstyle\rm sb})}} 
\newcommand{\rc}{{\scriptscriptstyle({\scriptscriptstyle\rm c})}} 
\newcommand{\nc}{{\scriptscriptstyle({\scriptscriptstyle\rm sp})}} 
\newcommand{\md}{{\scriptscriptstyle{({\scriptscriptstyle\rm mod})}}} 

\newcommand{\gt}{\tau} 

\newcommand{\sad}{\msf s}

\newcommand{\cP}{{\msf c_{\msf P}}}
\newcommand{\cF}{{\msf c_{\msf F}}}
\newcommand{\zs}{{\zeta_\gt(s)}}
\newcommand{\zn}{{\zeta_\gn(s)}}
\newcommand{\zzn}{{z_\gn(s)}}

\newcommand{\bai}{{\bm{Ai}}}

\renewcommand{\mid}{{\, | \, }}

 \usepackage{todonotes}

\newtheorem{theorem}{Theorem}[section]
\newtheorem{prop}[theorem]{Proposition}
\newtheorem{lemma}[theorem]{Lemma}
\newtheorem{corollary}[theorem]{Corollary}
\newtheorem{mainthm}{Theorem}

\theoremstyle{definition}
\newtheorem{definition}[theorem]{Definition} 

\newtheorem{assumption}[theorem]{Assumptions} 
\newtheorem{rhp}[theorem]{RHP} 

\theoremstyle{remark}
\newtheorem{remark}[theorem]{Remark} 
\numberwithin{equation}{section}

\begin{document}

\title[S6VM and dOPs]{The stochastic six vertex model and discrete orthogonal polynomial ensembles}
\author[P.~Ghosal]{Promit Ghosal$^*$}
\thanks{$^*$Department of Statistics, University of Chicago, USA. \texttt{promit@uchicago.edu}}

\author[G.~Silva]{Guilherme L.~F.~Silva$^\dagger$}
\thanks{$^\dagger$Instituto de Ciências Matemáticas e de Computação, Universidade de S\~ao Paulo (ICMC - USP), São Carlos, São Paulo, Brazil. \texttt{silvag@icmc.usp.br}}


\date{}


\begin{abstract}
Stochastic growth models in the Kardar–Parisi–Zhang (KPZ) universality class exhibit remarkable fluctuation phenomena. While a variety of powerful methods have led to a detailed understanding of their typical fluctuations or large deviations, much less is known about behavior on intermediate, or moderate deviation, scales. Addressing this problem requires refined asymptotic control of the integrable structures underlying KPZ models.

Motivated by this perspective, we study multiplicative statistics of discrete orthogonal polynomial ensembles (dOPEs) in different scaling regimes, with a particular focus on applications to tail probabilities of the height function in the stochastic six-vertex model. For a large class of dOPEs, we obtain robust singular asymptotic estimates for multiplicative statistics critically scaled near a saturated-to-band transition. These asymptotics exhibit universal crossover behavior, interpolating between Airy, Painlevé II, and Bessel-type regimes. 
Our proofs employ the Riemann-Hilbert Problem (RHP) approach to obtain asymptotics for the correlation kernel of a deformed version of the dOPE across the critical scaling windows. These asymptotics are then used on a double integral formula relating this kernel to partition function ratios, which may be of independent interest. At the technical level, the RHP analysis employs a novel parameter-dependent local parametrix, which requires a separate asymptotic analysis of its own. 

Using these results, together with a known identity relating a Laplace-type transform of the stochastic six-vertex model height function to a multiplicative statistic of the Meixner point process, we derive moderate deviation estimates for the height function in both the upper and lower tail regimes, with sharp exponents and constants.

\end{abstract}


\vspace*{-1.6cm}

\maketitle

\setcounter{tocdepth}{1}
\tableofcontents

\section{Introduction}

The six-vertex model is one of the most ubiquitous lattice models in modern mathematical physics, being a cornerstone in the study of two-dimensional lattice models and their phase transitions. From the algebraic structures that arise from the Yang-Baxter equations, the unraveling of rich phase transitions, to its deep connections to random point processes, it displays an overarching range of connections to several areas of mathematics and physics. 

Given a graph with vertices in the square lattice $\Z\times \Z$, a configuration of the six-vertex model in this graph consists of an assignment of orientation to each of the edges, in such a way that the number of arrows pointing inwards equals the number of arrows pointing outwards of each vertex. 

This means that each vertex is assigned one choice among six possible incident edges' configurations, as shown in Figure~\ref{fig:6VMconfigs}.

\begin{figure}[t]
\begin{subfigure}{0.48\textwidth}
    \centering
    \includegraphics[scale=0.68]{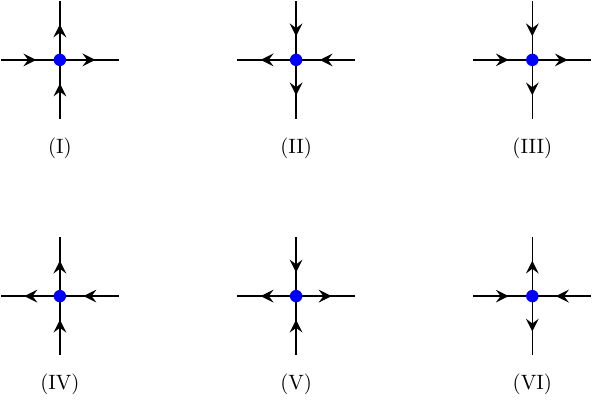}
\end{subfigure}
\begin{subfigure}{0.48\textwidth}
    \centering
    \includegraphics[scale=0.68]{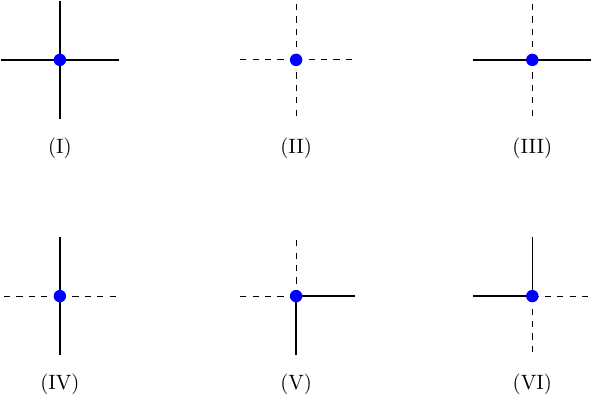}
\end{subfigure}
    \caption{On the left, the six possible configurations of individual vertices in the six-vertex model in its classical formulation. On the right, the corresponding configurations of vertices used in the representation as a random path configuration.}
    \label{fig:6VMconfigs}
\end{figure}

To each one of the 6 possible vertex configurations, one assigns a (positive) weight, and the total weight of a configuration is the product of the weights in each vertex of the configuration. The probability of a configuration is then proportional to the total weight of the configuration. This way, one may view the distribution on the configurations as a probability that depends on 6 initial parameters, the weights of the model. Obviously, depending on the boundary conditions prescribed, one may reduce this number of parameters.

In our case, we are interested in the {\it stochastic} six-vertex model on the quadrant $\Z_{\geq 0}\times \Z_{\geq 0}$ and step initial condition, which may be seen as a random growth model. This version of the model was introduced by Gwa and Spohn in 1992. Rather than viewing it as a model in equilibrium statistical mechanics, it is instructive to describe it as a Markovian evolution of non-intersecting lines, obtained when we replace arrows pointing west and south by empty (or rather dashed) edges, and arrows pointing north or east by solid edges, we refer the reader again to Figure~\ref{fig:6VMconfigs} for the display of all possible configurations at a given vertex. In this version of the model, the vertices are stochastic, meaning that for fixed inputs at a given vertex, the sum of weights over all possible outputs is 1, with each weight being non-negative. By adding arrows to the solid lines, we may view a given configuration of this model as a set of directed up-right random paths, see Figure~\ref{fig:S6VMfullconfig}.
\begin{figure}[t]
    \centering
    \includegraphics[scale=1]{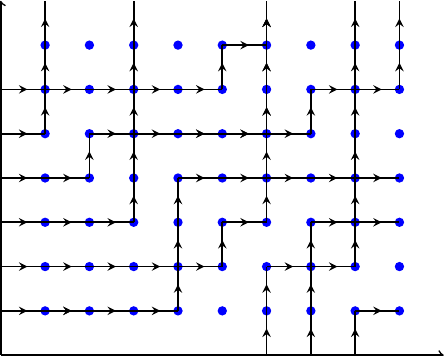}
    \caption{A configuration of the stochastic six-vertex model in the random path representation. }
    \label{fig:S6VMfullconfig}
\end{figure}

The step initial condition is then obtained by setting that paths emerge from every vertex in the vertical axis, and that there are no paths entering from the horizontal axis. Configurations with this initial condition, as well as its stochastic evolution, are depicted in Figure~\ref{fig:MarkovEvol}.

Given a down-right diagonal path on $\mathbb{Z}^2$ and specified boundary conditions along the path, the S6V model generates a measure on the vertices located above and to the right of the path, see Figure~\ref{fig:MarkovEvol} below for more details. Equivalently, it produces a measure on the configuration of solid lines that extend from the boundary inputs and continue upwards and rightwards. The measure is defined recursively: starting with vertices where inputs are provided, the outputs are randomly selected from all possible configurations, with probabilities determined by the associated vertex weights. In this framework, the weights of the six-vertex model may be all described by two independent parameters, which are usually parametrized by $\mq\in (0,1)$ and $u>0$. More details will be given in Section~\ref{sec:S6V} below.

 The height function $\mathfrak{h}(x,y)$ for the S6V model is defined as the net number of particles that have crossed the time-space line between $(0,0)$ and $(x,y)$, where each left-to-right move contributes $+1$ and each right-to-left move contributes $-1$. For a more detailed description of the S6V model, as well as the construction of $\mathfrak{h}(x,y)$ for bi-infinite configurations, we refer to Section~\ref{sec:S6V} below.

The graph of the height function $\mathfrak h(x,y)$ may be viewed as a random interface that encodes all statistics of the model. Naturally, one is asked about its {\it typical behavior}, that is, about a Law of Large Numbers for it, and also about the {\it fluctuations} around this typical behavior. In the seminal work \cite{BorodinCorwinGorin2016}, Borodin, Corwin and Gorin answer both of these questions. Their Law of Large Numbers reads \cite[Theorem~1.1]{BorodinCorwinGorin2016} 
$$
\lim_{N\to +\infty} \frac{\mathfrak h(N x,N y)}{N}= \mcal H(x,y),\quad \text{in probability},
$$
for an explicit limiting function $\mcal H(x,y)$. This deterministic function $\mcal H(x,y)$ is curved in a conic region with vertex at the origin, and it is flat outside this region. The curved region is referred to as the {\it liquid} region, and the flat regions are termed the {\it frozen} regions. A simulation of the S6V model is displayed in Figure~\ref{fig:S6VMsimulation}, where these regions become apparent.

\begin{figure}[t]
    \centering
    \includegraphics[scale=0.3]{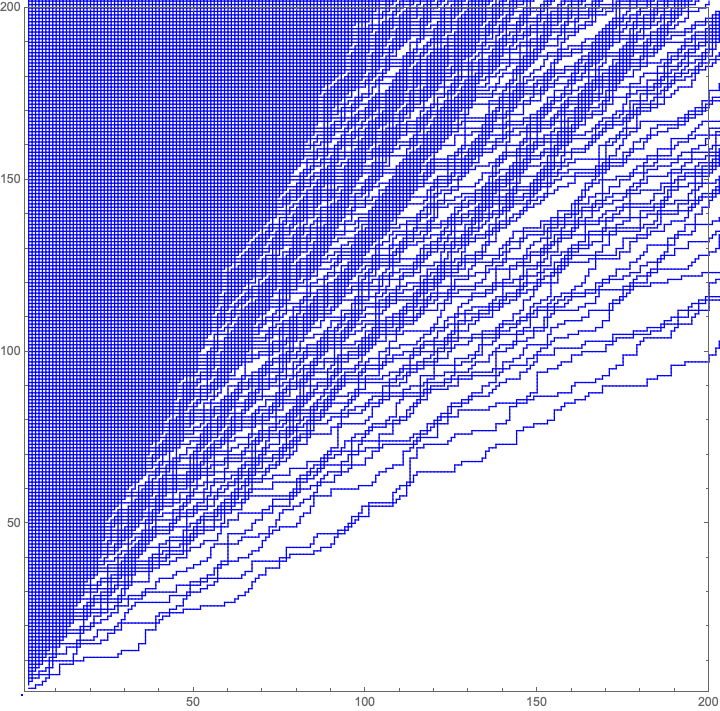}
    \caption{A simulation of the stochastic six-vertex model with narrow wedge initial condition, in a grid of size $200\times 200$, performed with the software Mathematica. Dark (blue) lines correspond to the paths in the model. The region filled with paths in the upper left corner, as well as the empty region in the lower right corner, correspond to the frozen regions. In these regions, the behavior of paths is determined with probability exponentially close to $1$. In the roughly conic region between the frozen ones we see a true stochastic behavior of paths. This region corresponds to the liquid region.
    }
    \label{fig:S6VMsimulation}
\end{figure}

Within the solid region, the fluctuations of $\mathfrak h$ are expected to be exponentially small. On the other hand, for $(x,y)$ within the liquid region the fluctuations are nontrivial, and given by the celebrated {\it Tracy-Widom distribution} $F_{\rm GUE}$ \cite[Theorem~1.2]{BorodinCorwinGorin2016},
\begin{equation}\label{eq:CLTheightfunction}
\lim_{N\to +\infty} \P\left( -\frac{\mathfrak h(Nx, Ny)-N \mathcal H(x,y)}{\sigma N^{1/3}}\leq \msf h \right)=F_{\rm GUE}(\msf h).
\end{equation}
In the convergence above, $\sigma=\sigma(x,y)$ is an explicit scaling constant. 

The original appearance of the Tracy-Widom distribution in the literature goes back to the 1990s \cite{Forrester1993, tracy_widom_level_spacing}, where it was found that it describes the limiting fluctuations for the largest eigenvalue of Hermitian random matrices with independent Gaussian entries, the so-called {\it Gaussian unitary ensemble} of random matrices, in the limit of large matrix size. As shown by Tracy and Widom \cite{tracy_widom_level_spacing}, this distribution can be computed through the Painlevé II equation: it admits the representation
$$
F_{\rm GUE}(s)=\exp\left(-\int_s^\infty (u-s)^2 q(u)^2\dd u\right),
$$
where $q$ is the Hastings-McLeod solution to the Painlevé II equation,
$$
q''(s)=sq(s)+2q(s)^3,\quad \text{with boundary behavior } q(s)\sim \ai(s) \text{ as } s\to +\infty,
$$
$\ai(\cdot)$ being the Airy function. Alternatively, it is given by the Fredholm determinant
$$
F_{\rm GUE}(s)=\det\left(\msf I-\msf A\right)_{L^2(s,\infty)},
$$
where $\msf A:L^2(s,\infty)\to (s,\infty)$ is the integral operator with integral kernel given by the {\it Airy kernel} $\msf A$,
\begin{equation}\label{deff:AiryKernel}
\msf A(x,y)\deff \frac{\ai(x)\ai'(y)-\ai(y)\ai'(x)}{x-y},\quad x\neq y,\qquad \msf A(x,x)\deff (\ai'(x))^2-\ai(x)\ai''(x).
\end{equation}

The scaling exponents $N^1$ and $N^{1/3}$, as well as the appearance of the limit $F_{\rm GUE}$ are not an accident, and instead are manifestations of the fact that the stochastic six-vertex model belongs to the {\it Kardar-Parisi-Zhang} (KPZ) universality class. 

In this paper, the main question we want to address is: how far away from its fluctuation scale can the height function still feel the KPZ universality statistics? 

Our answer will be given through the study of tail probabilities for the height function $\msf h$.  Our first main result for the stochastic six-vertex model deals with the upper tail distribution of the height function.

\begin{mainthm}\label{thm:S6VMlowertail}
    Let $\mathfrak h(x,y)$ be the height function for the stochastic six-vertex model evaluated at a point $(x,y)$ in the liquid region. For any $\varepsilon,\delta>0$ sufficiently small, there exists a (large) constant $\msf h_0>0$ for which
    $$
    \exp\left(-(1+\varepsilon)\frac{4 }{3}\msf h^{3/2}\right)\leq \P\left( \frac{\mathfrak h(Nx, Ny)-N \mathcal H(x,y)}{\sigma N^{1/3}}\leq -\msf h \right)\leq \exp\left(-(1-\varepsilon)\frac{4}{3} \msf h^{3/2}\right), 
    $$
    for every $\msf h$ with $\msf h_0 \leq \msf h\leq  N^{\tfrac{1}{6}-\delta}$.
\end{mainthm}

Our second main result for the stochastic six-vertex model deals with the lower tail distribution of the height function.

\begin{mainthm}\label{thm:S6VMuppertail}
    Let $\mathfrak h(x,y)$ be the height function for the stochastic six-vertex model evaluated at a point $(x,y)$ in the liquid region. For any $\varepsilon,\delta>0$ sufficiently small, there exists a (large) constant $\msf h_0>0$ for which
    $$
    \exp\left(-(1+\varepsilon)\frac{1}{12}\msf h^3\right)\leq \P\left( \frac{\mathfrak h(Nx, Ny)-N \mathcal H(x,y)}{\sigma N^{1/3}}\geq \msf h \right)\leq \exp\left(-(1-\varepsilon)\frac{ 1}{12}\msf h^3\right), 
    $$
    for every $\msf h$ with $\msf h_0 \leq \msf h\leq  N^{\tfrac{1}{6}-\delta}$.
\end{mainthm}

Our methods also extend the tail bounds in Theorems~\ref{thm:S6VMlowertail} and~\ref{thm:S6VMuppertail} to the full range $\msf h_0 \leq \msf h \leq  \msf h^{-1}_0 \mn^{1/6}$ with the same power laws $\msf h^3$ and $\msf h^{3/2}$, at the cost of changing the universal constants $4/3$ and $1/12$ when we go beyond the regime $\msf h \leq \mn^{1/6-\delta}$. See Remark~\ref{rmk:beyond_moderate} for a detailed discussion.

Technically speaking, Theorem~\ref{thm:S6VMlowertail} describes deviations in the direction of the lower tail of the random variable $\mathfrak h(\nu\gn,\gn)$, whereas Theorem~\ref{thm:S6VMuppertail} describes the upper tail of it. The swapped terminology as used here is common in KPZ-type models, and refers more to the lower/upper tail directions of the limiting Tracy-Widom distribution, as we discuss next. We opt to follow this common terminology throughout our paper.

A more precise version of Theorems~\ref{thm:S6VMlowertail} and \ref{thm:S6VMuppertail} that includes more explicit expressions for the scaling parameters is given in Theorems~\ref{thm:S6VMlowertail_precise} and \ref{thm:S6VMuppertail_precise} below.

The Tracy-Widom distribution has tail behavior \cite{tracy_widom_level_spacing}
\begin{equation}\label{eq:tailsTW}
\begin{aligned}
& 1-F_{\rm GUE}(\msf h)=\exp\left({-\frac{2}{3}\msf h^{3/2}(1+\boh(1))}\right), \quad \msf h \to +\infty,\qquad \text{and} \\
& F_{\rm GUE}(-\msf h)=\exp\left({-\frac{1}{12}\msf h^{3}(1+\boh(1))}\right), \quad \msf h\to +\infty.
\end{aligned}
\end{equation}
More refined versions of these estimates are obtained in \cite{BaikBuckinghamDiFranco2008, DeiftItsKrasovsky2008}.

Thus, Theorems~\ref{thm:S6VMlowertail} and \ref{thm:S6VMuppertail} complement the CLT convergence \eqref{eq:CLTheightfunction} drastically, in the sense that $F_{\rm GUE}$ provides a good approximation to the distribution of $\msf h$ not only in the fluctuation limit, but also for the tail probabilities when $|\msf h|$ is large, and deeper down all the way to the moderate deviation regime when $|\msf h| =\Boh(N^{1/6})$. In other words, the KPZ universality limit $F_{\rm GUE}$ still governs the leading behavior of the height function down to the moderate deviations scale.

The investigation of tail probabilities for models in the Kardar-Parisi-Zhang (KPZ) universality class has attracted considerable attention, revealing that the Tracy-Widom distributions governs not only fluctuations at the typical scale but also provide accurate asymptotic descriptions deep into the moderate deviations regime. For unitarily invariant random matrix ensembles with convex analytic potentials, precise moderate deviations in the full upper tail regime are due to Eichelsbacher, Kriecherbauer and Sch{\"u}ler \cite{EichelsbacherKriecherbauerSchuler2016}, strongly based on the earlier work \cite{KriecherbauerSchubertSchulerVenker2015}. 
Upper tail inequalities for GUE with optimal constants were also established by Paquette and Zeitouni \cite{PaquetteZeitouni2017}, following earlier work by Johansson \cite{johansson_2000} using sub-additivity and large deviation arguments. Similar optimal bounds for GOE and LOE were obtained in \cite{ErdosXu2023, BasuBusaniFerrari2022}. See also the review \cite{Ledoux2007} by Ledoux for earlier work surveying various non-optimal tail bounds in random matrix theory. 

Still on the upper tail, some partial results are also known for zero-temperature models beyond the classical invariant ensembles. For $\beta$-ensembles with general Dyson index $\beta$, Ramírez, Rider and Virág \cite{RamirezRiderVirag2006} proved convergence to the stochastic Airy operator spectrum, with sharper upper tail estimates later obtained by Dumaz-Virág \cite{DumazVirag2013}. Ledoux and Rider \cite{LedouxRider10} established tail bounds with correct exponents $3/2$ (upper tail) and $3$ (lower tail), though without optimal constants $2\beta/3$ and $\beta/24$. 

The left tail bounds, as always, require more sophisticated tools. For geometric last passage percolation they were proven by Baik, Deift, McLaughlin, Miller and Zhou \cite{BaikDeiftMcLaughlinMillerZhou2001}, and for longest increasing subsequences (related to Poissonian last passage percolation) by Löwe and Merkl \cite{LoweMerkl2002upper} as well as Löwe, Merkl and Rolles \cite{LoweMerkl2002lower}. For the LUE, sharp moderate deviations across the full lower-tail-to-large-deviations regime follow from the work \cite{DeiftItsKrasovsky2008} by Deift, Its and Krasovsky. A lower bound for the lower tail of the general $\beta$ Laguerre ensembles was obtained by Basu, Ganguly, Hegde and Krishnapur \cite{BasuGangulyHegdeKrishnapur2021} and optimal bounds (with the appropriate exponents and the constants) for both the tails by Baslingker, Basu, Bhattacharjee and Krishnapur \cite{baslingker2025optimal}. We also refer to \cite{ByunSeongMiYang2025, BorotEynardMajumdarNadal2011, MajumdarSchehr2014, LebleSerfaty2017, Forrester2012, DeanMajumdar2006, BorotGuionnet2013, KatzavPerezCastillo2010} and the references therein for important progress on large deviations of extremal eigenvalues in invariant ensembles, whether in the upper or lower tail directions. 

Beyond zero-temperature models, tail estimates for positive-temperature models in the KPZ universality class have been pursued in several recent works, notably for the KPZ equation itself \cite{CorwinGhosal2020,CorwinGhosal2020b, GangulyHegde2022,corwin2024lower}. For last passage percolation, the correct tail exponents $3/2$ and $3$ (without optimal constants) were conditionally obtained via bootstrapping and geodesic geometry methods by Ganguly and Hegde \cite{GangulyHegde2023}. However, prior to our work, no rigorous proof of tight optimal tail asymptotics extending into the moderate deviations regime had been established for any positive-temperature model except the KPZ equation. Our Theorems~\ref{thm:S6VMlowertail} and~\ref{thm:S6VMuppertail} provide the first such result for the stochastic six-vertex model, confirming that the Tracy-Widom universality persists down to the moderate deviations scale $|\msf{h}| = \Boh(N^{1/6})$ and that thermal fluctuations do not alter the leading-order tail behavior given in \eqref{eq:tailsTW}.

Our starting mathematical point in this paper is a relation between the six-vertex model and the so-called {\it Meixner ensemble}. The latter is an interacting particle system $\mcal X$ of $n$ particles on the lattice of strictly positive integers\footnote{In \cite{BorodinOlshanski2017}, the Meixner ensemble has state space $\Z_{\geq 0}$, but for us it is more convenient to shift their lattice and work with the Meixner ensemble in $\Z_{>0}$ instead. } $\Z_{>0}$, whose joint distribution of points may be described by Meixner orthogonal polynomials, and it depends on two parameters $\ms>0$ and $\mt\in (0,1)$. Recall that $\mq\in (0,1)$ and $u>0$ are the parameters of the stochastic six-vertex model. Borodin and Olshanski \cite[Corollary~8.5]{BorodinOlshanski2017} showed the identity of observables
\begin{equation}\label{eq:BorodinOlshanskiidentityintro}
\E_{\rm S6V}\left[ \prod_{j=1}^\infty \frac{1}{1+\zeta\mq^{\mathfrak h(x,y)+j}} \right]=\E_{\rm Meixner}\left[\prod_{\lambda\in \Z_{> 0}\setminus \mcal X} \frac{1}{1+\zeta \mq^{\lambda-1}} \right], \quad \text{for every }\zeta\in \C,
\end{equation}
with correspondence of parameters
\begin{equation}\label{eq:MeixnerS6VMparameters}
n=x,\quad \ms =y-x \quad \text{and} \quad \mt = \msf q^{-1/2}u^{-1}.
\end{equation}
The variable $\zeta \in \C$ may be seen as a spectral variable, and both expectations in \eqref{eq:BorodinOlshanskiidentityintro} may be interpreted as generating functions in this spectral variable. As we will prove in Lemmas~\ref{lem:UpTail} and \ref{lem:LowTail}, if we choose $\zeta=\mq^{-N\mcal H(x,y) -\sigma \gn^{1/3}\msf h }$, then the approximation
\begin{equation}\label{eq:fluctuationsvslaplacetransform}
\P\left( \frac{\mathfrak h(Nx, Ny)-N \mathcal H(x,y)}{\sigma N^{1/3}}\geq \msf h \right)\approx \E_{\rm S6V}\left[ \prod_{j=1}^\infty \frac{1}{1+\zeta\mq^{\mathfrak h(\gn x,\gn y)+j}} \right]
\end{equation}
is valid thanks to soft probabilistic arguments. As a consequence, we learn from \eqref{eq:BorodinOlshanskiidentityintro} that proving Theorems~\ref{thm:S6VMuppertail} and \ref{thm:S6VMlowertail} boils down to obtaining asymptotics for the multiplicative statistic of the Meixner ensemble in the right-hand side of \eqref{eq:BorodinOlshanskiidentityintro}. In fact, we will do more: we will obtain {\it precise and universal asymptotics} for multiplicative statistics of the form
\begin{equation}\label{deff:LNintro}
\msf L_\gn(s) \deff  \E_{\mcal X_n}\left[\prod_{\lambda\in \Z_{>0}\setminus \mcal X_n}\frac{1}{1+\ee^{-t(\lambda-\ga\gn) -s} } \right],\quad n=\lfloor \gn x\rfloor,
\end{equation}
for $\ga,t>0, s\in \R$, and a large family of {\it discrete orthogonal polynomial ensembles} $\mcal X_n$ with $n=\lfloor\gn x\rfloor$ particles. For the record, the statistics needed when applying \eqref{eq:fluctuationsvslaplacetransform} are recovered from the Meixner ensemble with the choice of parameters 
\begin{equation}\label{eq:MeixnertoS6VMcorrespondence}
\begin{aligned}
 & \ms=\gn(y-x),\quad \beta=\mq^{-1/2}u^{-1},\quad  t=\log(\mq^{-1}), \\ 
 & \ga=\ga(x,y)=\mathcal H(x,y),\quad \text{and}\quad s=-\left(1+\sigma\msf h \gn^{1/3}\right) \log(\mq^{-1}).
\end{aligned}
\end{equation}
Discrete orthogonal polynomial ensembles form a class of point processes which are natural discrete analogues of eigenvalues of random matrices. The connection between discrete orthogonal polynomial ensembles and growth models has been extensively used since the early 2000s, but to our knowledge the term discrete orthogonal polynomial ensemble was first coined in the seminal work of Johansson \cite{JohanssonPlancherel2001}, who among other results showed that, for instance, the Meixner ensemble admits a reduction to the classical Plancherel measure on random partitions. Several notable findings in the KPZ universality theory, both old and new, may be traced back to such type of connections. 

The {\it large deviations theory} for discrete orthogonal polynomial ensembles may be considered by now well developed, and allows one to extract asymptotic limits of statistics for these models. In fact, the recent work \cite{DasLiaoMucciconi2025} by Das, Liao and Mucciconi also utilizes the identity \eqref{eq:BorodinOlshanskiidentityintro} as their starting point. Fluctuations of $\mathfrak h$ in the large deviations scale lead to multiplicative statistics of the Meixner which also live in a large deviations (global) scale of the point process. As a consequence, the scaling limit of the Meixner component in the right-hand side of \eqref{eq:BorodinOlshanskiidentityintro} is an immediate consequence of an earlier theory by Johansson \cite{johansson_2000} (see also the more recent work \cite{DasDimitrov22}), and most of the work in \cite{DasLiaoMucciconi2025} lies on deriving properties of this limit, and on obtaining the analogue of the approximation \eqref{eq:fluctuationsvslaplacetransform} in the large deviation scale they focus on.

In contrast, the choice of parameters in \eqref{eq:MeixnerS6VMparameters} that we need to consider implies that the multiplicative statistics $\msf L_\gn$ live on {\it local singular scales}. Still as a consequence of the potential theory developed in \cite{johansson_2000}, points in the Meixner ensemble have three distinguished behaviors in the large. First off, for the value $\ga>0$ in \eqref{eq:MeixnertoS6VMcorrespondence}, random points fill out all sites in $(0,n \ga)\cap \Z_{>0}$ with exponentially large probability. For a second value $\gb>\ga$, all sites in $(\gb n,+\infty)$ remain free of points, also with exponentially large probability. Finally, in the middle interval $(n\ga,n \gb)$, points accumulate with nontrivial fluctuations, following a certain strictly positive density. The intervals $(0,\ga)$, $(\ga,\gb)$ and $(\gb,+\infty)$ are called {\it saturated}, {\it band}, and {\it gap} regions in the limiting (scaled) spectrum of points, respectively. In our scale, the endpoint $\ga$ that marks the transition between the saturated and band regions coincides precisely with the value of $\ga$ in \eqref{eq:MeixnerS6VMparameters}, and the random points may be described globally through a density function $\phi(n z)$, the so-called equilibrium density. This density $\phi(z)$ is constant and strictly positive for $z\in (0,\ga)$, it decays from this constant value to $0$ along the interval $(\ga,\gb)$, and it remains equal to $0$ on the interval $(\gb,+\infty)$.
Roughly, each term in the product in \eqref{deff:LNintro} either goes exponentially to $0$, or converges to $1$. In other words, at a formal level, the term inside the expectation in \eqref{deff:LNintro} reads
\begin{equation}\label{eq:movingwallcharac}
\frac{1}{1+\ee^{-t(\lambda-\ga\gn)-s}}\approx \ind_{(0,\ga-\frac{s}{t\gn})}(\lambda/\gn)\ee^{t(\lambda-\ga\gn)+s}+ \ind_{(\ga-\frac{s}{t\gn},+\infty)}(\lambda/\gn) ,
\end{equation}
where $\ind_J$ denotes the indicator function of the interval $J$. Therefore,
\begin{equation}\label{eq:deltaapproxLN}
\msf L_\gn(s)\approx \E_{\mcal X_n}\left[ \exp\left( \sum_{{ \lambda\in \Z_{>0}\setminus \mcal X_n, \, \lambda\leq \gn \ga-\frac{s}{t} }} \left(t(\lambda-\ga\gn)+s\right)  \right) \right]
\end{equation}

For values $s=\Boh(\gn)$ (corresponding to the large deviations scale $\msf h=\Boh(N^{2/3})$ in the S6VM language), the large $N$ behavior of statistics in \eqref{deff:LNintro} may be determined from the mentioned large deviations theory for orthogonal polynomial ensembles. Thanks to the approximation \eqref{eq:deltaapproxLN}, $\msf L_N(s)$ essentially becomes a gap probability which sits in the global scale $s=\Boh(N)$, and $\log \msf L_\gn(s)/\gn^2$ has leading behavior governed by a modification of the equilibrium problem that incorporates the exponential factor approximation \eqref{eq:deltaapproxLN}.

If, instead, $s$ is of order strictly smaller than $\gn$, $\log \msf L_\gn(s)$ leaves on a scale lower than $\gn^2$, and the mentioned potential-theoretic methods fail to capture its limit. In fact, in such smaller scales of $s$, one sees that $\ind_{(\ga-\frac{s}{t\gn},+\infty)}$ becomes the indicator function of an interval shrinking to the point $\ga$, and the leading behavior of $\msf L_\gn(s)$ comes effectively from the fluctuations of particles near the point $\ga$ rather than from fluctuations of all the particles at the global scale $\Boh(\gn)$.

This effect of local fluctuations is also at the core of our main results on the asymptotics of $\msf L_\gn(s)$. To explain the nature of such asymptotics, it is instructive to look at the empty sites $\mcal X_n^\circ\deff \Z_{>0}\setminus \mcal X_n$ that appear in the expectation \eqref{deff:LNintro} as a new point process; this is the {\it hole} point process. For clarity, for this discussion we will still call these empty sites as random particles. Observe however that $\mcal X_n^\circ$ always has infinitely many particles, and this point process is not an orthogonal polynomial ensemble.

For particles in the process $\mcal X_n^\circ$ the roles of the bands, gaps and saturated regions are interchanged, and the intervals $(0,\ga)$, $(\ga,\gb)$ and $(\gb,+\infty)$ now form gap, band and saturated regions. In particular, there are no particles in $(0,n\ga)\cap \mcal X_n^\circ$ with probability exponentially close to $1$, and the density of particles in $(n\ga,n(\ga+\varepsilon))\cap \mcal X_n^\circ$ is positive. With this interpretation in mind, we may view the indicator function in \eqref{eq:deltaapproxLN} as creating a pushing effect in particles. When $s>0$, essentially all particles in $\mcal X_n^\circ$ are inside $(\ga-s/t\gn,+\infty)$, the indicator function in \eqref{eq:deltaapproxLN} evaluates to $1$, and the corresponding expectation in \eqref{deff:LNintro} becomes close to $1$: it should be close to the gap probability of the interval $(0,\ga-s/t\gn)$. However, when $s<0$, with positive probability one finds particles for which the indicator function ${\ind_{(\ga-\frac{s}{t\gn},+\infty)}(\lambda/\gn) }$ in \eqref{eq:deltaapproxLN} vanishes, one has to account for such cancellations in the expectation \eqref{deff:LNintro}, and the asymptotic behavior of $\msf L_\gn(s)$ changes drastically.

A lot of our work here is to rigorously quantify the heuristics we just described. With our results, we learn that the behavior of $\mathsf{L}_\gn(s)$ is governed by the variation of $\left( 1 + \ee^{-t(\lambda - \gamma \gn) - s} \right)^{-1}$ as $s$ ranges over the interval $(-\delta\,\gn^{1/2},\,\delta\,\gn^{1/2})$, and three distinct asymptotic regimes emerge: (i) \emph{subcritical}, when $s \in (s_0\,\gn^{1/3},\, s_0^{-1}\gn^{1/2})$; (ii) \emph{critical}, when $s \in (-s_0\,\gn^{1/3},\, s_0\,\gn^{1/3})$; and (iii) \emph{supercritical}, when $s \in (-s_0^{-1}\,\gn^{1/2},\, -s_0\,\gn^{1/3})$, where $s_0>0$ is an arbitrary large constant. In terms of the height variable $\msf h=\Boh(s/\gn^{1/3})$ of the six-vertex model in \eqref{eq:MeixnertoS6VMcorrespondence}, the subcritical regime corresponds to the upper moderate deviations regime, the critical regime corresponds to the CLT scale, and the supercritical regime corresponds to the lower moderate deviations regime.

A deeper explanation for the emergence of these three different regimes, and the underlying mechanisms in the corresponding asymptotics, may be described again by random matrix theory. In random matrix theory language, the point $\ga$ is a (regular) soft edge point of the original discrete orthogonal polynomial ensemble $\mcal X_n$, and statistics of $\mcal X_n^\circ$ near $\ga$ may be computed in terms of the Airy kernel \eqref{deff:AiryKernel}. Still in this language, one may view the endpoint $\ga-s/(t\gn)$ in the approximation \eqref{eq:movingwallcharac} as a moving wall. When $s\gg 0$, this moving wall is to the left of the soft edge point $\ga$, and asymptotics of $\msf L_\gn$ are expected to be governed by the Airy kernel as well. However, when $s\ll 0$, this moving wall gets inside the band region $(\ga,\gb)$, and in configurations giving the leading contribution to $\msf L_\gn$ particles are now pushed to the right of this wall at $\ga-s/(t\gn)$. This push phenomenon creates a hard edge effect at $\ga-s/(t\gn)$. Heuristically, one then expects that the leading contributions should be coming from typical hard edge contributions, and governed now by the {\it Bessel kernel},

\begin{equation}\label{deff:BesselKernelintro}
\msf J_0(x,y)\deff \frac{1}{2}\frac{\Jb_0(\sqrt{x})\Jb'_0(\sqrt y)\sqrt y - \Jb_0 (\sqrt y)\Jb'_0(\sqrt x)\sqrt{x} }{x-y}, \quad x\neq y, \qquad \msf J_0(x,x)\deff \frac{1}{4}\left( \Jb_0'(\sqrt x)^2+\Jb_0(\sqrt{x})^2 \right).
\end{equation}
In the formula above, $\Jb_\nu$ denotes the Bessel function of first kind and order $\nu$. Obviously, there is also a transitional asymptotic effect on $\msf L_\gn(s)$ for intermediate values of $s$, which should somehow capture a ``soft-edge to hard-edge" transition. According to random matrix theory, this type of transition should be governed by kernels related to Painlevé equations, as we now describe.

Let $u=u(y)$ be the solution to the Painlevé XXXIV equation
\begin{equation}\label{eq:uppertailPXXXIVtranscendent}
\partial_y^2 u =4u^2+2yu+\frac{(\partial_yu)^2}{2u}\quad \text{with}\quad u(y)=\frac{\ee^{-\frac{4}{3}y^{3/2}}}{4\pi y^{1/2}}\left(1+\Boh(y^{-1/4})\right), \; y\to +\infty.
\end{equation}
This solution is uniquely characterized by its behavior at $+\infty$, it is real-valued and free of poles over the real line, and satisfies
\begin{equation}\label{eq:lowertailPXXXIVtranscendent}
u(y)=-\frac{y}{2}-\frac{1}{8y^2}+\Boh(y^{-7/2}),\quad y\to -\infty.
\end{equation}
With this Painlevé transcendent, we consider a particular solution $\psi=\psi(\zeta,y)$ of the linear Schrödinger equation with potential $2u(y)$ of the form
$$
\partial_y^2 \psi = \left(\zeta+y+2u(y)\right)\psi, \quad \text{with}\quad \psi(\zeta,y)=(1+\Boh(\zeta^{-1}))\frac{\zeta^{-1/4}}{2\sqrt{\pi}}\ee^{-\frac{2}{3}\zeta^{3/2}-y\zeta^{1/2}} \text{ as }\zeta\to +\infty.
$$
The Painlevé XXXIV kernel is then given by
\begin{equation}\label{deff:P34kernel}
\msf K_\ptf(\zeta,\xi\mid y)\deff \frac{\psi(\zeta,y)\partial_y\psi(\xi,y)-\psi(\xi,y)\partial_y \psi(\zeta,y)}{\zeta-\xi},\quad \zeta,\xi\in \R,\; y \in \R,\; \zeta\neq \xi,
\end{equation}
and extended by continuity to $\zeta=\xi$. 

Typically, Painlevé-type kernels are expressed in terms of the matrix-valued solution to the Lax Pair equations characterizing the Painlevé equation itself, rather than in terms of a stand-alone equation. The more direct construction we just provided in terms of the linear Schrödinger equation is proved in Section~\ref{sec:PXXXIV} below.

Still in the context of random matrix theory, other kernels related to the Painlevé XXXIV equation have appeared, to our knowledge, originally in \cite{its_kuijlaars_ostensson_2008}, as the limiting correlation kernel for a random matrix model in a suitable double scaling regime near a regular soft edge. Since then, Painlevé XXXIV kernels (or related quantities) have resurfaced in random matrix theory in other works \cite{zhang_2017, BogatiskiyClaeysIts2016, BothnerShepherd2024}. 

The precise conceptual reason underneath the appearance of the Airy, Bessel and Painlevé XXXIV kernels is explained by the exact solvability structure of discrete orthogonal polynomial ensembles. As we discuss in detail in Section~\ref{sec:dOPE}, correlation functions for discrete orthogonal polynomial ensembles may be computed through the Christoffel-Darboux kernel $\msf K_n$ of orthogonal polynomials with respect to a discrete weight $w(x)$ on $\Z_{>0}$. We will show that if $\Theta=\Theta(\cdot\mid s)$ is any sufficiently regular function depending smoothly on a parameter $s$, then
\begin{equation}\label{eq:defformulaintro}
\log \E_{\mcal X_n}\left[ \prod_{\lambda\in \mcal X_n}\Theta(\lambda\mid s) \right]=-\int_s^\infty \left(\sum_{x\in \Z_{>0}} \msf K_n^\Theta(x,x\mid u)w(x)(\partial_s \Theta)(x\mid s=u) \right)\dd u,
\end{equation}
where $\msf K_n^\Theta(\cdot\mid s)$ is the Christoffel-Darboux kernel associated to the new weight $\Theta(x\mid s)w(x)$, which we view as as a {\it deformation} of the original weight $w(x)$. This formula is not directly applicable to the expectation in \eqref{deff:LNintro} because it runs over the hole point process $\Z_{>0}\setminus \mcal X_n$, which is not an orthogonal polynomial ensemble. Nevertheless, a simple observation (see \eqref{eq:SLrelation} below) will yield that $\msf L_\gn(s)$ is equal to the sum of an explicit growing term and the right-hand side of \eqref{eq:defformulaintro} with choice 
\begin{equation}\label{deff:gsigintro}
\Theta(x\mid s)=\Theta_\gn(x\mid s)\deff 1+\ee^{-t(x-\ga\gn)-s}.
\end{equation}
Most of our effort is devoted to analyzing the right-hand side of \eqref{eq:defformulaintro} with this choice $\Theta=\Theta_\gn$. Notice however that the function $\Theta_\gn$ itself diverges for values $x<\ga \gn-s/t$. In fact, after performing the scaling $\frac{1}{\gn}\Z_{>0}$ of the original lattice, this blow up of the deformation $\Theta_\gn$ contributes nontrivially exactly at the scale of local correlations near the point $\ga$. In turn, this blow up means that the corresponding scaled kernel also blows up, and the asymptotics involved in the formula \eqref{eq:defformulaintro} are in fact {\it singular}. As a consequence, a lot of effort has to be put into the required asymptotic analysis. Notably, first, we have to compute the growing terms coming from the right-hand side of \eqref{eq:defformulaintro} quite precisely, to show that they cancel exactly the explicit blown-up terms previously mentioned and, second, unravel that the next order term yields the leading asymptotics we are seeking for.

Our main findings on asymptotics of $\msf L_\gn$ explain the above heuristics in precise terms, not only for the Meixner ensemble but for any discrete orthogonal polynomial ensemble displaying a saturated-to-band transition in the limiting spectrum.

To state a simplified version of our findings, let $\mathcal X_n$ be a discrete orthogonal polynomial ensemble of $n$ particles on $\Z_{>0}$, for which the associated limiting particle density has the interval $(0,\ga)$ as a saturated region and $(\ga,\gb)$ as a band interval, and let $\msf L_\gn(s)$ be the associated multiplicative statistic \eqref{deff:LNintro}.

\begin{mainthm}\label{thm:multstatsintro}
Let $s<S$ be given numbers with $|s|,|S|\leq \frac{1}{M} \gn^{1/2}$, $M>0$ sufficiently large but fixed.
     For some constant $\msf c>0$ and any $\delta>0$ sufficiently small, the multiplicative statistic $\msf L_\gn$ has asymptotics described by
\begin{equation}\label{eq:LNasymptintrothm}
     \log\frac{\msf L_\gn(S)}{\msf L_\gn(s)}\approx \frac{\msf c}{\gn^{1/3}}\int_s^S \sum_{\substack{x\in \Z\\ |x-\ga\gn|\leq \gn\delta}}\msf H_\gn\left( -\frac{\msf c}{\gn^{1/3}}\left(x-\gn\ga+\frac{v}{t} \right)\mid v \right)\frac{\ee^{-t(x-\gn\ga +v/t)}}{(1+ \ee^{- t(x-\gn\ga +v/t)})^2}\dd v,
     \end{equation}
     as $\gn\to \infty$, where the function $\msf H_\gn(\zeta\mid \sad)$ has the following asymptotic behavior
     \begin{enumerate}[(i)]
         \item For any sufficiently large constant $s_0>0$
         $$
        \msf H_\gn(\zeta\mid s)\approx \msf A\left( \zeta+\frac{\msf c}{t}\frac{s}{\gn^{1/3}} , \zeta+\frac{\msf c}{t}\frac{s}{\gn^{1/3}} \right),
         $$
         uniformly for $s_0\gn^{1/3}\leq s\leq \frac{1}{s_0}\gn^{1/2} $, where $\msf A(x,x)$ is the diagonal of the Airy kernel \eqref{deff:AiryKernel}.
         \item For any $s_0>0$,
         $$
         \msf H_\gn(\zeta\mid s)\approx \msf K_\ptf\left( \zeta,\zeta\mid \frac{\msf c}{t}\frac{s}{\gn^{1/3}} \right),
         $$
         uniformly for $-s_0\gn^{1/3}\leq s\leq s_0\gn^{1/3}$, where $\msf K_\ptf(\cdot,\cdot\mid y)$ is the Painlevé XXXIV kernel in the Painlevé variable $y$.
         \item For any sufficiently large constant $s_0>0$,
         $$
         \msf H_\gn(\zeta\mid s)\approx -\frac{\msf c^2}{t^2}\frac{s^2}{\gn^{2/3}}\msf J_0\left( \frac{\msf c^2}{t^2}\frac{s^2}{\gn^{2/3}} \zeta,\frac{\msf c^2}{t^2}\frac{s^2}{\gn^{2/3}}\zeta \right),
         $$
         uniformly for $-\frac{1}{s_0}\gn^{1/2}\leq s\leq -s_0\gn^{1/3}$, where $\msf J_0$ is the Bessel \eqref{deff:BesselKernelintro}.
     \end{enumerate}
\end{mainthm}

A refined version of Theorem~\ref{thm:multstatsintro}, including estimates for the error terms in each of the cases and an expression for the constant $\msf c$ in terms of the equilibrium measure of the system, is given in Theorem~\ref{thm:multstat_formal} below. In particular, the error terms we obtain are robust enough so that we can in fact integrate these asymptotic identities and recover leading asymptotics for $\log \msf L_\gn(s)$; we will say more about this in a moment.

As we said earlier, the correlation kernel appearing in \eqref{eq:defformulaintro}--\eqref{deff:gsigintro} blows up as $\gn$ grows. The term $\msf H_\gn$ appearing in Theorem~\ref{thm:multstatsintro} rises up as a surviving subleading term in the asymptotics of this correlation kernel, after accounting for explicit cancellations from the mechanisms mentioned earlier.

At a formal level, the exponential factor
\begin{equation*}
\frac{\ee^{- t(x-\gn\ga +v/t)}}{(1+ \ee^{- t(x-\gn\ga +v/t)})^2}=\frac{\ee^{- t|x-\gn\ga +v/t|}}{(1+ \ee^{- t|x-\gn\ga +v/t|})^2},
\end{equation*}
acts as a delta function, in the sense that non-negligible terms in the sum \eqref{eq:LNasymptintrothm} must come from terms $x- \gn\ga+v/t =\Boh(1)$.
A naive Taylor expansion suggests that at these points,
\begin{equation}\label{eq:scalingHNintro}
\msf H_\gn\left( -\frac{\msf c}{t \gn^{1/3}}\left(x-\gn\ga+\frac{v}{t} \right)\mid v \right)\approx \msf H_\gn(0\mid v).
\end{equation}
Sending formally $S\to +\infty$ in \eqref{eq:LNasymptintrothm} and neglecting for a moment boundary terms, the ansatz \eqref{eq:scalingHNintro} yields the approximation
\begin{equation}\label{eq:LNHNintro}
\msf L_\gn(s)\approx -\frac{\msf c}{t\gn^{1/3}}\int_s^\infty \msf H_\gn(0\mid v)\dd v,
\end{equation}
and the asymptotic behavior of $\msf H_\gn$ from Theorem~\ref{thm:multstatsintro} leads to the next main result.

\begin{mainthm}\label{mainthmintrointegrals}
    For a constant $\msf c>0$, the multiplicative statistics $\msf S_\gn$ has asymptotic behavior described by
\begin{enumerate}[(i)]
         \item For any sufficiently large constant $s_0>0$
         $$
        \log \msf L_\gn(s)\approx -\int_{\msf c s/(t\gn^{1/3})}^{+\infty} \msf A\left( y,y\right)\dd y,
         $$
         uniformly for $s_0\gn^{1/3}\leq s\leq \frac{1}{s_0}\gn^{1/2} $, where $\msf A(x,x)$ is the diagonal of the Airy kernel \eqref{deff:AiryKernel}.
         \item For any $s_0>0$,
         $$
         \log \msf L_\gn(s)\approx \log F_{\rm GUE}\left(\frac{\msf c s}{t\gn^{1/3}}\right),
         $$
         uniformly for $-s_0\gn^{1/3}\leq s\leq s_0\gn^{1/3}$, where $F_{\rm GUE}$ is the GUE Tracy-Widom distribution.
         \item For any sufficiently large constant $s_0>0$,
         $$
         \log \msf L_\gn(s)\approx  \msf J_0(0,0) \frac{\msf c^3}{t^3}\frac{|s|^3}{3\gn},
         $$
         uniformly for $-\frac{1}{s_0}\gn^{1/2}\leq s\leq -s_0\gn^{1/3}$, where $\msf J_0$ is the Bessel \eqref{deff:BesselKernelintro}.
     \end{enumerate}
\end{mainthm}

A detailed version of Theorem~\ref{mainthmintrointegrals}, with appropriate bounds for the error terms in the claimed approximations, is given in Theorem~\ref{thm:integrated_formal} below.

Theorem~\ref{mainthmintrointegrals} (i) and (iii) yield leading behavior of $\log\msf L_\gn$ through correlation kernels evaluated along the diagonal, and the appearance of $F_{\rm GUE}$ in (ii) is also explained by a similar phenomenon, as we now explain. 

We will show in Section~\ref{sec:PXXXIV} below the non-trivial identity
\begin{equation}\label{eq:KptfPXXXIV}
\msf K_\ptf(0,0\mid y)=\frac{1}{2}\partial_{y}^2u(y)-3u(y)^2-2yu(y)
\end{equation}
between the PXXXIV kernel and the Painlevé transcendent $u(y)$. It is well understood that solutions to the PXXXIV are in one-to-one correspondence to solutions to the {\it general} Painlevé II equation. In our case, this procedure yields an expression of $u$ in terms of an inhomogeneous Painlevé II equation. However, as we show in Section~\ref{sec:PXXXIV} below, rather than applying the standard identification, we identify the right-hand side of \eqref{eq:KptfPXXXIV} directly with the Hamiltonian of the Hastings-McLeod solution to the homogeneous Painlevé II equation, and so in turn with the Tracy-Widom distribution itself:
\begin{equation}\label{eq:KXXXIVTW}
\msf K_\ptf(0,0\mid y)=\frac{\dd}{\dd y}\log F_{\rm GUE}(y).
\end{equation}
This identity explains how Theorem~\ref{thm:multstatsintro}--(ii) implies Theorem~\ref{mainthmintrointegrals}--(ii).

Once we settle Theorem~\ref{mainthmintrointegrals}, Theorems~\ref{thm:S6VMlowertail} and \ref{thm:S6VMuppertail} become essentially corollaries of Theorem~\ref{mainthmintrointegrals}, as we now indicate. 

We start with the regime $s_0\gn^{1/3}\leq s \leq \gn^{1/2}/s_0$, which is the range in Theorem~\ref{mainthmintrointegrals}--(i). Writing $s= \frac{t}{\msf c} \gn^{1/3}\msf h$ with $\msf h_0\leq \msf h\leq \gn^{1/6}/\msf h_0$ (recall \eqref{eq:MeixnerS6VMparameters}), and accounting for the asymptotic behavior of $\msf A(v,v)$ for $v$ large (see \eqref{eq:diagAiryKernelAsymp} below), we obtain
\begin{equation}\label{eq:heuristicsuppertailLn}
\log\msf L_\gn(s)\approx -\frac{1}{8\pi \msf h}\ee^{-\frac{4}{3}\msf h^{3/2}},\quad \text{that is},\quad \log(1-\msf L_\gn(s))\approx -\frac{4}{3}\msf h^{3/2}.
\end{equation}
Thanks to \eqref{eq:fluctuationsvslaplacetransform} and \eqref{eq:BorodinOlshanskiidentityintro}, this asymptotic approximation eventually leads to Theorem~\ref{thm:S6VMlowertail}.

Next, we look at the regime from Theorem~\ref{mainthmintrointegrals}--(ii). Write, as before, $s= \frac{t}{\msf c} \gn^{1/3}\msf h$, but now with $|\msf h|\leq \msf h_0$ and $\msf h_0$ sufficiently large. In the upper tail regime when $\msf h$ is very large and positive, we use \eqref{eq:tailsTW} and obtain again that $\log(1-\msf L_\gn(s))\approx -\frac{4}{3}\msf h^{3/2}$, in agreement with \eqref{eq:heuristicsuppertailLn}. In the other extreme when $\msf h$ is large but negative, \eqref{eq:tailsTW} now yields
$$
\log \msf L_\gn(s)\approx -\frac{|\msf h|^3}{12},
$$
which is the tail order in Theorem~\ref{thm:S6VMuppertail}. This behavior matches with Theorem~\ref{mainthmintrointegrals}--(iii) once we observe that $\msf J_0(0,0)=1/4$. In other words, it is precisely the Tracy-Widom distribution $F_{\rm GUE}$ in Theorem~\ref{thm:multstatsintro}--(ii) that allows us to interpolate between the asymmetrical upper and lower tail decays in Theorems~\ref{thm:S6VMlowertail} and \ref{thm:S6VMuppertail}.

The heart of the paper lies in the proofs of Theorems~\ref{thm:multstatsintro} and \ref{mainthmintrointegrals}, while Theorems~\ref{thm:S6VMlowertail} and \ref{thm:S6VMuppertail} follow naturally from them. Before turning to the detailed analysis, we briefly outline some core ideas behind the proofs of the first two results.

The approximation \eqref{eq:scalingHNintro} is, at first sight, reasonable to be expected. But it turned out to be rather cumbersome to rigorously derive Theorem~\ref{mainthmintrointegrals} from Theorem~\ref{thm:multstatsintro}. Aside from the technical aspects naturally involved, we believe there is a more conceptual reason underneath this difficulty. After scaling the lattice $\Z\mapsto \gn^{-1}\Z$, the sum in \eqref{eq:LNasymptintrothm} may be asymptotically approximated by
$$
\sum_{k\in \Z }\msf H_\gn\left( \zeta_k\mid v \right)\frac{\ee^{-\gn^{1/3}t\zeta_k}}{(1+\ee^{-\gn^{1/3}t\zeta_k} )^2},\qquad \text{with}\quad \zeta_k\deff \msf -c\gn^{2/3}\left(\frac{k}{\gn}-\ga +\frac{v}{t\gn}\right),\; k\in \Z.
$$
at the cost of exponentially negligible error terms. Such scaling $x\mapsto \zeta$ sends the points in the original lattice $\Z_{>0}$ to scaled points $\zeta_k$ which live in the same local scale of fluctuations near $x=\ga$. Our first try was to approximate this sum by the integral
\begin{equation}\label{eq:sumHNintegralHNintro}
\sum_{k\in \Z }\msf H_\gn\left( \zeta_k\mid v \right)\frac{\ee^{-\gn^{1/3}t\zeta_k}}{(1+\ee^{-\gn^{1/3}t\zeta_k} )^2}\approx \int p(\zeta) \msf H_\gn(\zeta\mid v) \frac{\ee^{-\gn^{1/3}t\zeta}}{(1+\ee^{-\gn^{1/3}t\zeta})^2} \dd \zeta,
\end{equation}
where $p(\zeta)$ is the density of points $\zeta$, which in our case can be computed explicitly. Performing a classical steepest descent argument on the remaining integral, we are left with the same approximation \eqref{eq:LNHNintro} that eventually leads to Theorem~\ref{mainthmintrointegrals}. However, the approximation \eqref{eq:sumHNintegralHNintro} does not actually hold true, and it turns out that the cancellations that lead to the approximation \eqref{eq:LNHNintro} do not come from approximating the sum in \eqref{eq:LNasymptintrothm} by an integral. Instead, the approximation \eqref{eq:LNHNintro} comes from cancellations that emerge from the periodicity of the lattice structure, or in other words from the discrete nature of the original model. At the core of these estimates lies the identity
$$
\sum_{k\in \Z}\frac{\ee^{-tk+u}}{(1+\ee^{-tk+u})^2}=\frac{1}{t}+\frac{8\pi}{t^2}\sum_{k=1}^\infty \frac{k\cos(2\pi k u/t)}{\ee^{2\pi^2 k/t}-\ee^{-2\pi^2k/t}},
$$
which is computed using Poisson's summation formula (see Lemma~\ref{lem:PoissonSummation} below, as well as Remark~\ref{rmk:JacobiTheta} which expresses this identity in terms of a Jacobi Theta function). With this identity at hand, we essentially show that the right-hand side of \eqref{eq:LNasymptintrothm} coincides (up to exponentially small terms) with
$$
\frac{\msf c}{t\gn^{1/3}}\int_{s}^S \msf H_\gn(0\mid v) \dd v+ \frac{8\pi}{t^2}\frac{\msf c}{\gn^{1/3}}\sum_{k=1}^\infty \frac{k}{\ee^{2\pi^2 k/t}-\ee^{-2\pi^2 k/t}}\int_{s}^S \msf H_\gn(0\mid v) \cos\left( \frac{2\pi k}{t}(\gn t \ga -v ) \right)\dd v,
$$
and a simple integration by parts shows that the integrals with trigonometric terms leave on a smaller scale that the first term, leading to \eqref{eq:LNHNintro} and the proof of Theorem~\ref{mainthmintrointegrals}.

The proof of Theorem~\ref{thm:multstatsintro}, in turn, takes up most of the technical effort developed in this paper. As said earlier, it relies on the asymptotic analysis of the right-hand side of \eqref{eq:defformulaintro} with the choice \eqref{deff:gsigintro}. The kernel $\msf K_n^\Theta$ appearing in \eqref{eq:defformulaintro} is expressed in terms of discrete orthogonal polynomials for the deformed weight $\Theta_\gn(x)w(w)$ on $\Z_{>0}$, and the analysis of \eqref{eq:defformulaintro} will be performed through the so-called Riemann-Hilbert approach to (discrete) orthogonal polynomials. The modern asymptotic theory for Riemann-Hilbert problems was initiated by Deift and Zhou in the context of integrable systems \cite{Deift1993, DeiftZhouPII95} and shortly afterwards extended to orthogonal polynomials as well \cite{Deiftetal1999, DMKVZ99OPs}. For discrete orthogonal polynomials, fundamental works include their large-degree asymptotic analysis carried out by Baik, Kriecherbauer, McLaughlin and Miller \cite{BKMMbook} for orthogonality on finite lattices, and by Bleher and Liechty for orthogonality on infinite lattices \cite{BleherLiechtyIMRN2011}, in great part also motivated by the six-vertex model \cite{bleher_liechty_book}. We also refer to \cite{Borodin2000Bessel, WangWong2011, Liechty2012dGaussian, AptekarevTulyakov2012Meixner, LiechtyWang2016, DaiYao2022dLaguerre} for similar works deriving asymptotics for discrete orthogonal polynomials and related quantities through the Riemann-Hilbert approach. 

Aside from technical modifications that we need to cope with because, as said, our deformation $\Theta_\gn$ of the weight blows up, the most fundamental difference of our asymptotic analysis of discrete RHPs with previous literature lies in the construction of the so-called local parametrix, and its needed analysis. In short terms, this local parametrix serves as a local approximation to the orthogonal polynomials near the point $z=\ga$. 

In the classical setup when $\Theta\equiv 1$, corresponding to the asymptotic analysis of discrete orthogonal polynomials with weight $w(x)$, the local parametrix at $z=\ga$ may be constructed with the aid of Airy functions. However, over here this is no longer possible: the effects of the deformation $\Theta_\gn$ are felt through down to the local scale of fluctuations of the model, and the local parametrix has to account for such effects. As a consequence, we need to consider a novel model problem in the construction of the local parametrix. This novel model problem contains two essential information. First off, its jump matrix involves the deformed weight $\Theta_\gn$ in the appropriate local scale, encoded through a function that still depends on $\gn$, and which has infinitely many poles accumulating on the real axis as $\gn$ grows. Second, its jump matrix also involves another function which encodes the fact that the original model was defined through discrete orthogonality rather than continuous. This second function has infinitely many poles on the real axis that become dense as $\gn$ grows. In particular, we have to deal with a local parametrix with non-constant and non-homogeneous jumps, which even at the local scale still depends on the large $\gt=\gn^{1/3}$ parameter, and moreover the jump matrices involved have {\it infinitely many poles} on the complex plane, which interplay all the way throughout our analysis. 

At the end, and in much as a consequence of the presence of the mentioned poles, we have to deal with the nonlinear steepest descent method for Riemann-Hilbert problems at 5 different stages: at the level of the RHP for orthogonal polynomials, at the level of the model problem required, which in itself involves 3 separate asymptotic analysis, one for each of the regimes (i)--(iii) from Theorem~\ref{thm:multstatsintro}, and a fifth and last analysis, necessary for the local parametrix needed to complete regime (iii). We give more details about these aspects, including a comparison with previous RHP literature and possible connections with integrable systems, in Section~\ref{sec:outline} below.

\subsection*{Acknowledgments}\hfill

We are deeply grateful to Jinho Baik for connecting us and encouraging us to work on this project from the very beginning.

G.S. also thanks Mattia Cafasso, Alfredo Deaño, Tom Claeys, Karl Liechty, Leslie Molag, Giulio Ruzza, and Lun Zhang, for various inspiring discussions related to this work.

Parts of this project were carried out during academic visits of G.S. to P.G. at Columbia University, MIT and Brandeis University. He acknowledges the hospitality of the institutions during these visits.

G.S. acknowledges support by the São Paulo Research Foundation (FAPESP),
Brazil, Process Number \# 2019/16062-1 and \# 2020/02506-2, and by the Brazilian National Council for Scientific and Technological Development (CNPq) under
Grant \# 306183/2023-4.

P.G. gratefully acknowledges helpful discussions with Alexei Borodin, Ivan Corwin, Sayan Das, Hindy Drillick, and Alexandre Krajenbrink at various stages of this project. Part of this work was carried out during P.G.’s visit to the Instituto de Ciências Matemáticas e de Computação, Universidade de São Paulo (ICMC–USP), and he sincerely thanks ICMC–USP for its hospitality.

P.G.'s research was partially supported by the National Science Foundation under Award DMS-2153661 (also listed as DMS-2515172).

\section{Statement of precise results}

In this section we introduce the models under study in detail and present our main results in precise terms. We begin with the stochastic six-vertex model, then discuss our findings for discrete orthogonal polynomial ensembles, and finally specialize to the Meixner ensemble which provides the key connection between these two subjects.

\subsection{The stochastic six-vertex model}\label{sec:S6V} \hfill 

The six-vertex model is a fundamental lattice model in statistical mechanics, first introduced by Pauling \cite{Pauling1935} in the context of ice-type models and subsequently studied extensively in the physics and mathematics literature \cite{Lieb1967, baxter1985exactly, reshetikhin2010lectures}. The stochastic version we consider was introduced by Gwa and Spohn \cite{Gwa1992} and has attracted significant recent attention due to its connections with integrable probability and the KPZ universality class \cite{BorodinCorwinGorin2016, Aggarwal2016a, corwin2020stochastic}.

\subsubsection{Definition of the model}

Let $\mathfrak{P}$ be an infinite up-right oriented line ensemble such that from each of the points on the line $\{(1,m): m\in \Z_{\geq }\}$ originates one path (this corresponds to the \textbf{step initial condition}) and they do not share any horizontal edge. 

We endow the space of such path ensemble with a probability measure in the following Markovian way (see Figure~\ref{fig:MarkovEvol}). This measure depends on two parameters $b_1,b_2\in (0,1)$, which may be alternatively described by parameters $\msf q\in (0,1)$ and $u>\msf q^{-1/2}$, see Table~\ref{tbl:table_of_figures} for this correspondence. 

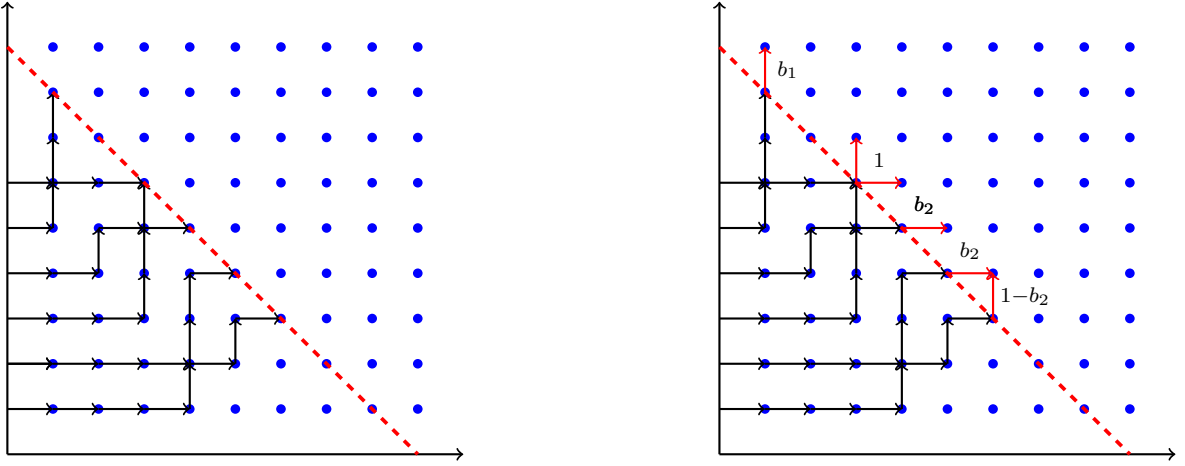
\begin{figure}
\centering
\begin{subfigure}{0.48\textwidth}
 \hspace{0.1\textwidth}
\begin{tikzpicture}[scale=0.6]
\draw[thick, ->] (-4,0) -- (6,0);
\draw[thick, ->] (-4,0) -- (-4,10);
\foreach \x in {-3,...,5}
    \foreach \y in {1,...,9}
    {
    \fill[blue] (\x,\y) circle (3pt);
    }
\draw[thick, ->] (-4,3) -- (-3,3);
\draw[thick, ->] (-4,2) -- (-3,2);    
\draw[thick, ->] (-1,1) -- (0,1);
\draw[thick, ->] (-3,2) -- (-2,2);
\draw[thick, ->] (-4,2) -- (-3,2);
\draw[thick, ->] (-4,1) -- (-3,1);
\draw[thick, ->] (-3,1) -- (-2,1);
\draw[thick, ->] (-2,1) -- (-1,1);
\draw[thick, ->] (0,1) -- (0,2);
\draw[thick, ->] (0,2) -- (1,2);
\draw[thick, ->] (1,2) -- (1,3);
\draw[thick, ->] (1,3) -- (2,3);
\draw[thick, ->] (-2,2) -- (-1,2);
\draw[thick, ->] (-1,2) -- (0,2);
\draw[thick, ->] (0,2) -- (0,3);
\draw[thick, ->] (0,3) -- (0,4);
\draw[thick, ->] (0,4) -- (1,4);
\draw[thick, ->] (-3,3) -- (-2,3);
\draw[thick, ->] (-2,3) -- (-1,3);
\draw[thick, ->] (-1,3) -- (-1,4);
\draw[thick, ->] (-1,4) -- (-1,5);
\draw[thick, ->] (-1,5) -- (0,5);
\draw[thick, ->] (-4,4) -- (-3,4);
\draw[thick, ->] (-3,4) -- (-2,4);
\draw[thick, ->] (-2,4) -- (-2,5);
\draw[thick, ->] (-2,5) -- (-1,5);
\draw[thick, ->] (-1,5) -- (-1,6);
\draw[thick, ->] (-4,5) -- (-3,5);
\draw[thick, ->] (-3,5) -- (-3,6);
\draw[thick, ->] (-3,6) -- (-2,6);
\draw[thick, ->] (-2,6) -- (-1,6);
\draw[thick, ->] (-4,6) -- (-3,6);
\draw[thick, ->] (-3,6) -- (-3,7);
\draw[thick, ->] (-3,7) -- (-3,8);
\draw[domain=-4:5, dashed, variable=\x, red, line width = 1.4pt] plot ({\x},{5-\x});
\end{tikzpicture}
\end{subfigure}
\begin{subfigure}{0.48\textwidth}
 \hspace{0.22\textwidth}
\begin{tikzpicture}[scale=0.6]
\draw[thick, ->] (-4,0) -- (6,0);
\draw[thick, ->] (-4,0) -- (-4,10);
\foreach \x in {-3,...,5}
    \foreach \y in {1,...,9}
    {
    \fill[blue] (\x,\y) circle (3pt);
    }
\draw[thick, ->] (-4,3) -- (-3,3);
\draw[thick, ->] (-4,2) -- (-3,2);
\draw[thick, ->] (-4,1) -- (-3,1);    
\draw[thick, ->] (-1,1) -- (0,1);
\draw[thick, ->] (-3,1) -- (-2,1);
\draw[thick, ->] (-3,2) -- (-2,2);
\draw[thick, ->] (-2,1) -- (-1,1);
\draw[thick, ->] (0,1) -- (0,2);
\draw[thick, ->] (0,2) -- (1,2);
\draw[thick, ->] (1,2) -- (1,3);
\draw[thick, ->] (1,3) -- (2,3);
\draw[thick, ->, red] (2,3) -- (2,4);
\draw[thick, ->] (-2,2) -- (-1,2);
\draw[thick, ->] (-1,2) -- (0,2);
\draw[thick, ->] (0,2) -- (0,3);
\draw[thick, ->] (0,3) -- (0,4);
\draw[thick, ->] (0,4) -- (1,4);
\draw[thick, ->, red] (1,4) -- (2,4);
\draw[thick, ->] (-3,3) -- (-2,3);
\draw[thick, ->] (-2,3) -- (-1,3);
\draw[thick, ->] (-1,3) -- (-1,4);
\draw[thick, ->] (-1,4) -- (-1,5);
\draw[thick, ->] (-1,5) -- (0,5);
\draw[thick, ->, red] (0,5) -- (1,5);
\draw[thick, ->] (-4,4) -- (-3,4);
\draw[thick, ->] (-3,4) -- (-2,4);
\draw[thick, ->] (-2,4) -- (-2,5);
\draw[thick, ->] (-2,5) -- (-1,5);
\draw[thick, ->] (-1,5) -- (-1,6);
\draw[thick, ->] (-4,5) -- (-3,5);
\draw[thick, ->] (-3,5) -- (-3,6);
\draw[thick, ->] (-3,6) -- (-2,6);
\draw[thick, ->] (-2,6) -- (-1,6);
\draw[thick, ->] (-4,6) -- (-3,6);
\draw[thick, ->] (-3,6) -- (-3,7);
\draw[thick, ->] (-3,7) -- (-3,8);
\draw[thick, ->, red] (-1,6) -- (0,6);
\draw[thick, ->, red] (-1,6) -- (-1,7);
\draw[thick, ->, red] (-3,8) -- (-3,9);
\draw[domain=-4:5, dashed, variable=\x, red, line width = 1.4pt] plot ({\x},{5-\x});
\node[] at (0.5,5.5) {$\scriptstyle b_2$};
\node[] at (1.5,4.5) {$\scriptstyle b_2$};
\node[] at (0.5,5.5) {$\scriptstyle b_2$};
\node[] at (2.7,3.5) {$\scriptstyle 1-b_2$};
\node[] at (-0.5,6.5) {$\scriptstyle 1$};
\node[] at (-2.5,8.5) {$\scriptstyle b_1$};
\end{tikzpicture}
\end{subfigure}
\caption{Markovian evolution of upright paths in $\Z^2_{\geq}$, starting from the narrow wedge initial condition.}
\label{fig:MarkovEvol}
\end{figure} 
 
 We assume that the region $T_n\deff\{(x,y)\mid x+y\leq n\}$ is endowed with the probability measure $\mathfrak{P}_n$, and construct a probability measure on $T_{n+1}$. 
 
 Fix a particular configuration on the region $T_n$ which is sampled under the measure $\mathfrak{P}_n$. Fix any vertex $v$ on the line $\{(x,y):x+y=n+1\}$. Denote the number of horizontal (resp. vertical) arrows coming into the site $v$ from the line $\{(x,y)|x+y=n\}$ by $i_1$ (resp. $j_1$). Similarly, denote the number of horizontal (resp. vertical) arrows going out of the the vertex $v$ by $i_2$ (resp. $j_2$). We restrict our discussion to the case $i_1,i_2,j_1,j_2\in \{0,1\}$. Now, the probabilities of the transitions $(i_1,j_1)\to (j_1,j_2)$ is given by the following weights:
\begin{align}\label{eq:}
\P((i_1,0)\to (i_2,0)) &= \frac{1-\msf q^{i_1}\msf q^{-1/2}u}{1-\msf q^{-1/2}u}\mathbbm{1}_{i_1=i_2},\\
\P((i_1,0)\to (i_2,1)) &= \frac{(\msf q^{i_1}-1)\msf q^{-1/2}u}{1-\msf q^{-1/2}u}\mathbbm{1}_{i_1=i_2-1},\\
\P((i_1,1)\to (i_2,1)) &= \frac{ \msf q^{i_1-1}-\msf q^{-1/2}u}{1-\msf q^{-1/2}u}\mathbbm{1}_{i_1=i_2},\\
\P((i_1,1)\to (i_2,0)) &= \frac{1- \msf q^{i_1-1}}{1-\msf q^{-1/2}u}\mathbbm{1}_{i_1=i_2-1}. 
\end{align} 

Note that $\sum_{i_2,j_2} \P((i_1,j_1)\to (i_2,j_2)) =1$ where $i_1,i_2,j_1,j_2 \in \{0,1\}$ in the case of the stochastic six vertex model. This defines a probability measure $\mathfrak{P}_{n+1}$ on the region $\{(x,y):x+y\leq n+1\}$ if all the weights are assumed to be positive. Owing to the order $\mathfrak{P}_n\subset \mathfrak{P}_m$ for all $n\leq m$ (i.e., $\mathfrak{P}_m|T_n = \mathfrak{P}_n$), the infinite volume limit $\mathfrak{P} = \lim_{n\to \infty} \mathfrak{P}_n$ is well defined and is called the \emph{stochastic six vertex model}. See Table~\ref{tbl:table_of_figures} for the weight configuration of individual vertices.

\begin{table}
        \centering
        \begin{tabular}{|c|c|c|c|c|c|}
           \hline

            \begin{tikzpicture}[scale=0.5]
            \fill[blue] (0,0) circle (3pt);
            \draw[dashed] (-1,0) -- (0,0);
            \draw[dashed] (0,0) -- (1,0);
            \draw[dashed] (0,-1) -- (0,0);
            \draw[dashed] (0,0) -- (0,1);
            \end{tikzpicture} 
             & 
            \begin{tikzpicture}[scale=0.5]
            \fill[blue] (0,0) circle (3pt);
            \draw[thick] (-1,0) -- (0,0);
            \draw[thick] (0,0) -- (1,0);
            \draw[thick] (0,-1) -- (0,0);
            \draw[thick] (0,0) -- (0,1);
            \end{tikzpicture} 
 & 
            \begin{tikzpicture}[scale=0.5]
            \fill[blue] (0,0) circle (3pt);
            \draw[dashed] (-1,0) -- (0,0);
            \draw[dashed] (0,0) -- (1,0);
            \draw[thick] (0,-1) -- (0,0);
            \draw[thick] (0,0) -- (0,1);
            \end{tikzpicture} 
 & 
            \begin{tikzpicture}[scale=0.5]
            \fill[blue] (0,0) circle (3pt);
            \draw[thick] (-1,0) -- (0,0);
            \draw[thick] (0,0) -- (1,0);
            \draw[dashed] (0,-1) -- (0,0);
            \draw[dashed] (0,0) -- (0,1);
            \end{tikzpicture} 
&                      
            \begin{tikzpicture}[scale=0.5]
            \fill[blue] (0,0) circle (3pt);
            \draw[dashed] (-1,0) -- (0,0);
            \draw[dashed] (0,0) -- (0,1);
            \draw[thick] (0,-1) -- (0,0);
            \draw[thick] (0,0) -- (1,0);
            \end{tikzpicture} 
 &

            \begin{tikzpicture}[scale=0.5]
            \fill[blue] (0,0) circle (3pt);
            \draw[thick] (-1,0) -- (0,0);
            \draw[thick] (0,0) -- (0,1);
            \draw[dashed] (0,-1) -- (0,0);
            \draw[dashed] (0,0) -- (1,0);
            \end{tikzpicture} 
\\
           \hline
                      
        $1$ & $1$ & $\frac{u\msf q^{-1/2}-\msf q^{-1}}{u\msf q^{-1/2}-1}$  & $\frac{(u\msf q^{1/2}-1)}{u\msf q^{-1/2}-1}$ & $\frac{\msf q^{-1}-1}{u\msf q^{-1/2}-1}$ & $\frac{u\msf q^{-1/2}(1-\msf q)}{u\msf q^{-1/2}-1}$ \\
        
        \hline  
        $1$ &  $1$ & $b_1$ & $b_2$ & $1-b_1$ & $1-b_2$ \\
            \hline
        \end{tabular}
        \caption{Table of vertex weights. The correspondence of weights $(b_1,b_2)\mapsto (\msf q,u)$ with $\msf q\in (0,1)$ and $u>\msf q^{-1/2}$ assumes that $b_1>b_2$.}
        \label{tbl:table_of_figures}
    \end{table}   
  
When needed to specify parameters, we use the notation $\mathfrak{P}(b_1,b_2)$ or $\mathfrak{P}(\msf q,u)$ for the limiting measure $\mathfrak{P}$. The measure $\mathfrak{P}(\msf q,u)$ of the stochastic six vertex on $\Z^2_{\geq 0}$ can also be interpreted using an interacting particle system, see for instance \cite{BorodinCorwinGorin2016} for details.

\subsubsection{The height function}

Each path ensemble $\mathfrak{P}$ is encoded by a \emph{height function} $\mathfrak{h}: \Z_{\geq 1} \times \Z_{\geq 1} \to \Z_{\geq 0}$, defined by setting $\mathfrak{h}(M,N)$ equal to the number of paths passing through or to the right of the vertex $(M,N)$. Equivalently, viewing the model as an interacting particle system where particle positions are given by the intersections of paths with horizontal lines, $\mathfrak{h}(M,N)$ counts the particles weakly to the right of position $M$ at time $N$.

The height function satisfies the boundary conditions $\mathfrak{h}(M,N) = N$ for $M = 1$ (step initial condition) and exhibits a law of large numbers as established by Borodin, Corwin, and Gorin \cite[Theorem~1.1]{BorodinCorwinGorin2016}:
\begin{equation}\label{eq:LLN}
\lim_{N \to +\infty} \frac{\mathfrak{h}(Nx, Ny)}{n} = \mathcal{H}(x,y), \quad \text{in probability},
\end{equation}
for an explicit limiting function $\mathcal{H}(x,y)$. The function $\mathcal{H}(x,y)$ possesses regions of nontrivial curvature and regions where it is affine. 
These are traditionally called the \emph{liquid region} and the \emph{frozen regions}, which we now describe in the setting of the stochastic six vertex model with parameters $(b_1,b_2)$.

\noindent\textbf{Liquid region.}
For fixed vertex weights $(b_1,b_2)$ with $0<b_1,b_2<1$, define the \emph{liquid region} as the set of macroscopic points $(x,y)\in\R_{>0}^2$ where the slope of the limit shape 
\[
(\partial_x \mathcal{H}(x,y),\; \partial_y \mathcal{H}(x,y))
\]
takes values strictly inside the admissible range imposed by arrow conservation.  
Equivalently, this is the region where the local proportions of horizontal and vertical arrows are strictly between $0$ and $1$.  
In probabilistic terms, the paths remain genuinely random, e.g., both horizontal and vertical edges appear with positive asymptotic frequencies determined by $(b_1,b_2)$, and the macroscopic limit shape is \emph{curved}.  
This is the “disordered’’ or “liquid” part of the model, and assuming $b_2<b_1$ it is the conic region given by
$$
\left\{(x,y)\in \R^2_{>0}\;\Big|\; \frac{1-b_1}{1-b_2}<\frac{x}{y}<\frac{1-b_2}{1-b_1} \right\}=\left\{(x,y)\in \R^2_{>0}\;\Big|\; \frac{1}{u\mq^{1/2}}<\frac{x}{y}<u\mq^{1/2} \right\}.
$$

\noindent\textbf{Frozen regions.}
The \emph{frozen regions} are the zones in the first quadrant complementary to the liquid region, where at least one of the local arrow densities reaches an extremal value.  
Equivalently, in these regions the limit shape $\mathcal{H}(x,y)$ is \emph{affine} and the local configuration of edges freezes into a deterministic pattern (all edges pointing in the same macroscopic direction).  
For step initial data, there are two such frozen phases:  
one dominated by vertical edges (close to the $y$–axis) and the other by horizontal edges (close to the $x$–axis).  
In both frozen regions, the height function fluctuations collapse to deterministic (non-random) linear evolution.

When scaled around points in the liquid region, the stochastic six vertex height function exhibits nontrivial random fluctuations:  
in particular, along any characteristic direction inside this region the model is in the KPZ class and one observes 
\begin{align*}
\mathfrak{h}(Nx,Ny)=N\mathcal{H}(x,y)+\Boh(N^{1/3}),
\end{align*}
with fluctuations governed by Tracy–Widom distributions or their inhomogeneous generalizations (depending on direction). In contrast, in the frozen regions the randomness disappears at the macroscopic scale: $
\mathfrak{h}(Nx,Ny)=N\mathcal{H}(x,y)+\Boh(1)$ and the local configuration is deterministic. Fluctuations there are of the order $\Boh(1)$ and do not exhibit KPZ-type behavior. Thus the interface between the curved liquid region and the affine frozen regions describes the transition from KPZ-type random behavior to deterministic crystal-like behavior in the stochastic six vertex model. Our main results about the stochastic six vertex model pertains to the liquid region of the height function.

\subsubsection{Fluctuations and the Tracy-Widom distribution}

Fix $\mnu > 1$ satisfying $1 < \mnu < \mq^{1/2}u$, and consider the height function at the point $(\mnu \mn, \mn)$ for large $\mn$. Define the centering and scaling constants:
\begin{equation}\label{eq:scalingconstants}
\aeq \deff \frac{(1-\sqrt{\mnu\mka})^2}{1-\mka}, \qquad \mc \deff \frac{\mka^{1/2}\mnu^{-1/6}}{1-\mka}\left((1-\sqrt{\mnu\mka})(\sqrt{\mnu/\mka}-1)\right)^{2/3},\quad \text{with}\quad \mka \deff \mq^{-1/2}u^{-1}.
\end{equation}
 The Central Limit Theorem obtained by Borodin, Corwin and Gorin \cite[Theorem~1.2]{BorodinCorwinGorin2016} states that the rescaled height function
\begin{equation}\label{eq:rescaledheight}
\Hf_\mn(\mnu) \deff \frac{\mathfrak{h}(\mnu\mn, \mn) - \aeq\mn}{\mc\mn^{1/3}}
\end{equation}
converges in distribution to the GUE Tracy-Widom distribution:
\begin{equation}\label{eq:TWconvergence}
\lim_{\mn \to +\infty} \P\left(-\Hf_\mn(\mnu) \leq x\right) = F_{\rm GUE}(x).
\end{equation}

Our main results quantify this convergence down to the level of the moderate deviation regimes.

\begin{theorem}[Upper tail for the stochastic six-vertex model]\label{thm:S6VMlowertail_precise}
Let $\mathfrak{h}(\mnu\mn, \mn)$ be the height function for the stochastic six-vertex model with parameters satisfying $\mq^{-1/2}u^{-1} < \mnu < \mq^{1/2}u$ (i.e., $(\mnu\mn, \mn)$ is in the liquid region). For any $\varepsilon,\delta > 0$ sufficiently small, there exists a constant $\msf{h}_0 = \msf{h}_0(\varepsilon, \delta, \mq, u) > 0$ such that
\begin{equation}\label{eq:lowertailbounds}
\exp\left(-(1+\varepsilon)\frac{4}{3}\msf{h}^{3/2}\right) \leq \P\left(\frac{\mathfrak{h}(\mnu\mn, \mn) - \aeq\mn}{\mc\mn^{1/3}} \leq -\msf{h}\right) \leq \exp\left(-(1-\varepsilon)\frac{4}{3}\msf{h}^{3/2}\right),
\end{equation}
uniformly for $\msf{h}_0 \leq \msf{h} \leq \msf{h}_0^{-1}\mn^{\tfrac{1}{6}-\delta}$.
\end{theorem}

\begin{theorem}[Lower tail for the stochastic six-vertex model]\label{thm:S6VMuppertail_precise}
Under the same hypotheses as Theorem~\ref{thm:S6VMlowertail_precise}, for any $\varepsilon,\delta > 0$ sufficiently small, there exists a constant $\msf{h}_0 = \msf{h}_0(\varepsilon, \delta, \mq, u) > 0$ such that
\begin{equation}\label{eq:uppertailbounds}
\exp\left(-(1+\varepsilon)\frac{1}{12}\msf{h}^3\right) \leq \P\left(\frac{\mathfrak{h}(\mnu\mn, \mn) - \aeq\mn}{\mc\mn^{1/3}} \geq \msf{h}\right) \leq \exp\left(-(1-\varepsilon)\frac{1}{12}\msf{h}^3\right),
\end{equation}
uniformly for $\msf{h}_0 \leq \msf{h} \leq \msf{h}_0^{-1}\mn^{\tfrac{1}{6}-\delta}$.
\end{theorem}

\begin{remark}[Extension up to the region $\msf{h} \leq \msf h_0^{-1}\mn^{1/6}$]\label{rmk:beyond_moderate}
Our results in fact give that the conclusions of Theorems~\ref{thm:S6VMlowertail_precise} and~\ref{thm:S6VMuppertail_precise} remain valid for $\msf{h} \leq \msf h_0^{-1}\mn^{1/6}$, with the same exponential growth rates (exponents $3/2$ for the upper tail and $3$ for the lower tail), though the precise values of constants $4/3$ and $1/12$ in the rate functions get lost. More precisely, one can derive from Theorem~\ref{thm:integrated_formal} that for $\msf{h}_0 \leq \msf{h} \leq \msf h_0^{-1}\mn^{1/6}$ with $\msf h_0$ sufficiently large, there exist constants $c_1, c_2 > 0$ (depending on $\varepsilon, \mq, u, \mnu$) such that
\[
\exp(-c_1\msf{h}^{3/2}) \leq \P\left(\frac{\mathfrak{h}(\mnu\mn, \mn) - \aeq\mn}{\mc\mn^{1/3}} \leq -\msf{h}\right) \leq \exp(-c_2\msf{h}^{3/2})
\]
and
\[
\exp(-c_1\msf{h}^3) \leq \P\left(\frac{\mathfrak{h}(\mnu\mn, \mn) - \aeq\mn}{\mc\mn^{1/3}} \geq \msf{h}\right) \leq \exp(-c_2\msf{h}^3).
\]
The restriction $\msf{h} \leq \mn^{1/6-\delta}$ in Theorems~\ref{thm:S6VMlowertail_precise} and~\ref{thm:S6VMuppertail_precise} ensures that we remain in the moderate deviations regime where the Tracy-Widom rate functions with their universal constants $4/3$ and $1/12$ continue to govern the tail behavior.
\end{remark}

For their exploration of the stochastic six-vertex speed process, Drillick and Haunschmid-Sibitz \cite{DrillickHaunschmid-Sibitz2024} obtained {\it one-sided} upper and lower tail bounds (in the CLT scale) for the height function of the narrow wedge stochastic six-vertex model we consider here. For the upper tail, their bound gets the optimal $3/2$-power law, but they do not get the $4/3$ optimal constant as we obtain. For the lower tail, they get a much weaker power law $c\,\msf h$ for some undetermined constant $c$, instead of our optimal $\frac{1}{12}\msf h^3$ factor. Estimates similar to \cite{DrillickHaunschmid-Sibitz2024} have been obtained by Landon and Sosoe \cite{LandonSosoe25} for the stationary stochastic six-vertex model. We are unaware of moderate deviation results for positive temperature models in the KPZ universality class (except the KPZ equation) that are sharp in the power law as we obtain here.

On a somewhat opposite extremal scaling, and as mentioned at the introduction, Das, Liao, and Mucciconi obtained lower tail large deviation results for the narrow wedge stochastic six-vertex model considered here \cite{DasLiaoMucciconi2025}. We discuss related results for other models in the KPZ universality class in Section~\ref{sec:furtherconnection} below.

\subsection{Discrete orthogonal polynomial ensembles}\label{sec:dOPE} \hfill 

The proofs of Theorems~\ref{thm:S6VMlowertail_precise} and \ref{thm:S6VMuppertail_precise} proceed through a connection with discrete orthogonal polynomial ensembles, which we now develop in generality.

\subsubsection{Definition and basic properties}

Following Johansson \cite{JohanssonPlancherel2001}, a \emph{discrete orthogonal polynomial ensemble} (dOPE) on $\Z_{>0}$ is a probability measure on $(\Z_{>0})^\gn$ of the form
\begin{equation}\label{eq:dOPEdef}
\frac{1}{\msf{Z}_\gn(\gw)} \prod_{1 \leq i < j \leq \gn}(\lambda_i - \lambda_j)^2 \prod_{j=1}^\gn \gw(\lambda_j), \qquad \lambda_1, \ldots, \lambda_\gn \in \Z_{>0},
\end{equation}
where $\gw: \Z_{>0} \to [0,+\infty)$ is the \emph{weight function} and
\begin{equation}\label{eq:partitionfunction}
\msf{Z}_\gn(\gw) \deff \sum_{\lambda_1, \ldots, \lambda_\gn \in \Z_{>0}} \prod_{1 \leq i < j \leq \gn}(\lambda_i - \lambda_j)^2 \prod_{j=1}^\gn \gw(\lambda_j)
\end{equation}
is the \emph{partition function}. Such ensembles arise naturally in combinatorics, representation theory, and random matrix theory; see \cite{Konig2005, BorodinGorin2016} for surveys.

A fundamental property is that dOPEs are \emph{determinantal point processes}: the measure \eqref{eq:dOPEdef} admits the representation
\begin{equation}\label{eq:determinantalrep}
\frac{1}{\gn!} \det\left(\msf{K}_\gn(\lambda_i, \lambda_j)\right)_{i,j=1}^\gn \prod_{j=1}^\gn \gw(\lambda_j),
\end{equation}
where $\msf{K}_\gn$ is the \emph{Christoffel-Darboux kernel}
\begin{equation}\label{eq:CDkernel_def}
\msf{K}_\gn(x,y) = \msf{K}_\gn(x,y \mid \gw) \deff \sum_{k=0}^{\gn-1} \gamma_k^2(\gw) \, p_k(x \mid \gw) p_k(y \mid \gw).
\end{equation}
Here $p_k = p_k(\cdot \mid \gw)$ denotes the monic orthogonal polynomial of degree $k$ with respect to $\gw$ on $\Z_{>0}$, and $\gamma_k = \gamma_k(\gw) > 0$ is the inverse of its weighted norm:
\begin{equation}\label{eq:orthogonality}
\sum_{x \in \Z_{>0}} x^j p_k(x) \gw(x) = 0 \quad \text{for } j = 0, \ldots, k-1, \qquad \frac{1}{\gamma_k^2} = \sum_{x \in \Z_{>0}} p_k(x)^2 \gw(x).
\end{equation}
The partition function relates to the norming constants via
\begin{equation}\label{eq:partitionnorming}
\msf{Z}_\gn(\gw) = \gn! \prod_{k=0}^{\gn-1} \gamma_k^{-2}(\gw).
\end{equation}

\subsubsection{Assumptions on the weight}

We consider weights whose rescaled versions become exponential weights in the large-$\gn$ limit.

\begin{assumption}[Assumptions on the weight]\label{assumpt:potential_formal}
The rescaled weight
\begin{equation}\label{eq:rescaledweight}
\gW(x) = \gW_\gn(x) \deff \frac{1}{\gn}\gw(\gn x), \qquad x \in \frac{1}{\gn}\Z_{>0},
\end{equation}
admits the representation
\begin{equation}\label{eq:weightdecomp}
\gW(x) = \ee^{-\gn\gV(x) - 2\gC + \gE(x)}, \qquad x \in \frac{1}{\gn}\Z_{>0},
\end{equation}
where:
\begin{enumerate}[(i)]
\item The \emph{external field} $\gV$ is real-valued, independent of $\gn$, continuous on $[0,+\infty)$, analytic in a complex neighborhood $U$ of $(0,\infty)$, bounded as $x \to 0$ along $U$, and satisfies the growth condition
\begin{equation}\label{eq:growthcondition}
\liminf_{x \to +\infty} \frac{\gV(x)}{x} > 0.
\end{equation}

\item The \emph{correction term} $\gE$ is real-valued, analytic in a complex neighborhood of $(0,+\infty)$, satisfies
\begin{equation}\label{eq:Econvergence}
\gE(x) = \gEO(x) + O(\gn^{-1}) \quad \text{as } \gn \to \infty,
\end{equation}
locally uniformly in a complex neighborhood $U$ of $(0,\infty)$, where $\gEO$ is analytic in $U$, and
\begin{equation}\label{eq:Ebound}
\sup_\gn \sup_{z \in U, |z| \leq \rho} \left(|\gE(z)| + |\gE'(z)|\right) \leq c\rho
\end{equation}
for some constant $c > 0$.

\item The \emph{normalization constant} $\gC$ is independent of $x$ but may depend on $\gn$.
\end{enumerate}
\end{assumption}

\subsubsection{The constrained equilibrium measure}

The large-$\gn$ behavior of dOPEs is governed by a variational problem from logarithmic potential theory. The \emph{constrained equilibrium measure} in the external field $\gV$ is the unique Borel probability measure $\gequil$ on $[0,+\infty)$ minimizing the energy functional
\begin{equation}\label{eq:energyfunctional}
I[\mu] = \iint \log\frac{1}{|x-y|}\,\dd\mu(x)\,\dd\mu(y) + \int \gV(x)\,\dd\mu(x)
\end{equation}
over all Borel probability measures $\mu$ on $[0,\infty)$ satisfying the \emph{upper constraint}
\begin{equation}\label{eq:upperconstraint}
\mu([a,b]) \leq b - a \quad \text{for every } b > a > 0.
\end{equation}

Under Assumption~\ref{assumpt:potential_formal}, the equilibrium measure $\gequil$ exists uniquely, is absolutely continuous with a piecewise real-analytic density, and has support consisting of finitely many intervals. We impose the following regularity conditions.

\begin{assumption}[Assumptions on the equilibrium measure]\label{assumpt:equilmeasure_formal}
The potential $\gV$ satisfies Assumption~\ref{assumpt:potential_formal} and its constrained equilibrium measure $\gequil$ satisfies:
\begin{enumerate}[(i)]
\item \emph{Connected support containing the origin:} $\supp\gequil = [0,\gb]$ for some $\gb > 0$.

\item \emph{Single saturated region:} The upper constraint \eqref{eq:upperconstraint} is active on an interval $[0,\ga]$ with $0 < \ga < \gb$:
\begin{equation}\label{eq:saturatedregion}
\gequil|_{[0,\ga]} = \text{Lebesgue measure on } [0,\ga].
\end{equation}

\item \emph{Regularity at the transition points:} The density of $\gequil$ vanishes as a square root at $\gb$, and $1 - \dd\gequil/\dd x$ vanishes as a square root at $\ga$. 

\item \emph{Regularity of the variational inequalities:} The Euler-Lagrange conditions hold strictly: there exists $\ell = \ell_\gV$ such that
\begin{equation}\label{eq:EulerLagrange}
\int \log\frac{1}{|x-y|}\,\dd \gequil(y) + \frac{1}{2}\gV(x) + \ell
\begin{cases}
< 0, & 0 < x < \ga,\\
= 0, & \ga \leq x \leq \gb,\\
> 0, & x > \gb.
\end{cases}
\end{equation}
\end{enumerate}
\end{assumption}

Under these assumptions, we call $[0,\ga]$ the \emph{saturated region}, $[\ga,\gb]$ the \emph{band region}, and $(\gb,+\infty)$ the \emph{gap region}. The equilibrium measure admits a probabilistic interpretation: for random points $\lambda_1, \ldots, \lambda_\gn$ distributed according to \eqref{eq:dOPEdef}, the empirical measure $\frac{1}{\gn}\sum_j \delta_{\lambda_j/\gn}$ converges almost surely to $\gequil$.

The regularity condition (iii) ensures that near the transition point $\ga$,
\begin{equation}\label{eq:squarerootvanishing}
1-\frac{\dd\gequil}{\dd x}(x)=\frac{\msf c_0}{\pi}(x-\ga)^{1/2}(1+\Boh(x-\ga)),\quad x\searrow \ga,
\end{equation}
for some constant $\msf{c}_0 > 0$. This behavior is crucial for the appearance of Airy-type asymptotics. 

\subsubsection{Deformed weights}

To analyze the six-vertex model through its connection to discrete orthogonal polynomial ensembles, we introduce a deformation of the weight~$\gw$ that encodes the relevant multiplicative observables. For parameters $t>0$ and $s\in\R$, and with the size parameter $\gn$ eventually tending to infinity, define the \emph{deformed weight}
\begin{equation}\label{eq:deformedweight}
\gwd(x)=\gwd(x\mid s)\deff \left(1+\ee^{-t(x-\ga\gn)-s }\right)\gw(x),\qquad x\in \Z_{>0}.
\end{equation}
The dependence of $\gwd$ on $t,s$, and $\gn$ is intentional: the shift $t\ga\gn$ could be absorbed into $s$, but keeping the two parameters separate is crucial for the asymptotic regimes we study, in particular for the connection with the six-vertex model. The dependence on $\gn$ indicates additionally that we are working with a \emph{varying weight}.

Let $\msf Z_\gn^\gsig(s)\deff \msf Z_\gn(\gwd)$ denote the partition function of the discrete orthogonal polynomial ensemble with weight \eqref{eq:deformedweight} (recall~\eqref{eq:dOPEdef}). The undeformed partition function corresponds to the limit $s\to+\infty$:
\[
\msf Z_\gn^\gsig(+\infty)=\msf Z_\gn(\gw).
\]
Here and below, the index ``$\sigma$'' simply indicates that a quantity is defined with respect to the deformed weight~$\gwd$.

\subsubsection{Multiplicative statistics}

The deformation~\eqref{eq:deformedweight} naturally gives rise to the multiplicative statistic
\begin{equation}\label{eq:multstat}
\msf{S}_\gn(s)
   \deff \E_{\mathcal{X}_\gn}\!\left[
      \prod_{\lambda\in\mathcal{X}_\gn}\left(1+\ee^{-t(\lambda-\ga\gn)-s}\right)
   \right],
\end{equation}
where $\mathcal{X}_\gn$ is a random configuration sampled from the discrete orthogonal polynomial ensemble. A direct application of the definition~\eqref{eq:dOPEdef} shows that this expectation is precisely the ratio of deformed to undeformed partition functions:
\begin{equation}\label{eq:SNZNrelation}
\msf S_\gn(s)=\frac{\msf Z_\gn^\gsig(s)}{\msf Z_\gn^\gsig(+\infty)},\qquad
\frac{\msf S_\gn(s)}{\msf S_\gn(S)}
    =\frac{\msf Z_\gn^\gsig(s)}{\msf Z_\gn^\gsig(S)}\quad\text{for any }s,S\in\R.
\end{equation}

A closely related statistic probes the \emph{hole process}:
\begin{equation}\label{eq:Lstatistic}
\msf{L}_\gn(s)\deff \E_{\mathcal{X}_\gn}\!\left[
   \prod_{\lambda\in \Z_{>0}\setminus\mathcal{X}_\gn}
      \frac{1}{1+\ee^{-t(\lambda-\ga\gn)-s}}
\right].
\end{equation}
The statistics $\msf S_\gn(s)$ and $\msf L_\gn(s)$ are linked through the identity
\begin{equation}\label{eq:SLrelation}
\log\msf L_\gn(s)=\log\msf S_\gn(s)
  -\sum_{x=1}^\infty \log\bigl(1+\ee^{-t(x-\ga\gn)-s}\bigr),
\end{equation}
which reflects the complementary contributions of particles and holes.

\subsubsection{The deformation formula}

For the deformed weight $\gwd$, let $p_{\gn,k}^\gsig$ denote the $k$-th monic orthogonal polynomial and let $\gamma_{\gn,k}^\gsig>0$ be its (inverse) norming constant (recall \eqref{eq:orthogonality}). Denote by ${\msf K}_\gn^\gsig$ the associated \emph{deformed Christoffel--Darboux kernel},
\begin{equation}\label{deff:CDkernel}
\msf K_\gn^\gsig(x,y)=\msf K_\gn^\gsig(x,y\mid s)
   \deff \sum_{k=0}^{\gn-1} (\gamma_{\gn,k}^\gsig)^2
         p_{\gn,k}^\gsig(x)\,p_{\gn,k}^\gsig(y).
\end{equation}

The asymptotic analysis of $\msf S_\gn(s)$ and $\msf L_\gn(s)$ and hence their connection to the stochastic six-vertex model, will begin from the properties of these deformed orthogonal polynomials and their kernels.

The starting point for our asymptotic analysis is a formula expressing the ratio of partition functions in terms of an integral involving the deformed Christoffel-Darboux kernel.

\begin{prop}[Deformation formula]\label{prop:deformationformula}
Let $\gw$ satisfy Assumption~\ref{assumpt:potential_formal}, and let $\gwd$ be defined by \eqref{eq:deformedweight}. For any $s \leq S$,
\begin{equation}\label{eq:deformationformula}
\log\frac{\msf{Z}_\gn^\gsig(s)}{\msf{Z}_\gn^\gsig(S)} = \int_s^S \sum_{x \in \Z_{>0}} \frac{\ee^{-t(x - \ga\gn) - u}}{1 + \ee^{-t(x - \ga\gn) - u}} \msf{K}_\gn^\sigma(x, x \mid u)\gw(x)\, \dd u,
\end{equation}
where $\msf{K}_\gn^\sigma(\cdot, \cdot \mid u)$ is the Christoffel-Darboux kernel for the deformed weight $\gwd(\cdot \mid u)$.
\end{prop}

For continuous weights, analogues of this deformation formula have appeared recently in the literature \cite{ClaeysGlesner2023, GhosalSilva2023}, see Remark~\ref{rm:deffformula} below for further details. We show here that Identity~\eqref{eq:deformationformula} is a special case of a generalized deformation formula applicable to a broad class of orthogonality weights; specifically, we can replace Assumption~\ref{assumpt:potential_formal} with substantially weaker regularity assumptions. In Section~\ref{sec:defformulaproof}, we adopt the ideas from \cite{GhosalSilva2023} to state and prove this general version, utilizing only orthogonality properties, regardless of whether it is discrete or continuous orthogonality.

Proposition~\ref{prop:deformationformula}, together with \eqref{eq:SNZNrelation} and \eqref{eq:SLrelation}, explains the appearance of the summand in Theorem~\ref{thm:multstatsintro}: aside from explicit prefactors arising in the derivation, the right-hand side in \eqref{eq:LNasymptintrothm} is essentially the scaling limit of the sum on the right-hand side of \eqref{eq:deformationformula}. To present the full version of Theorem~\ref{thm:multstatsintro}, we first formalize the three distinct regimes that enter into its formulation as a set of assumptions.

\subsubsection{Parameter regimes}

The asymptotic behavior of $\msf{S}_\gn(s)$ depends crucially on the scale of the parameter $s$ relative to $\gn$.

\begin{assumption}[Parameter regimes]\label{assumpt:parameterregimes}
The parameters $t > 0$ and $s \in \R$ in \eqref{eq:deformedweight} satisfy:
\begin{enumerate}[(a)]
\item The parameter $t = t(\gn)$ satisfies $T^{-1} \leq t \leq T$ for some fixed $T > 1$.

\item The parameter $s = s(\gn)$ lies in one of the following regimes:
\begin{enumerate}[(i)]
\item \emph{Subcritical regime:} $s_0\gn^{1/3} \leq s \leq s_0^{-1}\gn^{1/2}$ for {\it some large fixed} $s_0 > 0$.
\item \emph{Critical regime:} $|s| \leq s_0\gn^{1/3}$ {\it for any fixed} $s_0 > 0$.
\item \emph{Supercritical regime:} $-s_0^{-1}\gn^{1/2} \leq s \leq -s_0\gn^{1/3}$ {\it for some large fixed} $s_0 > 0$.
\end{enumerate}
\end{enumerate}
\end{assumption}

When connecting dOPEs back to the stochastic six vertex model, we will choose $\msf h=\Boh(\gn^{1/3}s)$. This way, the subcritical, critical, and supercritical regimes correspond respectively to the upper moderate deviations, fluctuation scale, and lower moderate deviations of the height function in the six-vertex model.

\subsection{Universal asymptotics for multiplicative statistics}\label{sec:universalasymptotics}\hfill 

We now state our main results on the asymptotics of $\msf{S}_\gn(s)$, which hold for any dOPE satisfying Assumptions~\ref{assumpt:potential_formal} and \ref{assumpt:equilmeasure_formal}. These results reveal a remarkable universality: the asymptotic behavior depends on the specific ensemble only through a single constant $\msf{c}_\gV > 0$ determined by the equilibrium measure. 
 Define
\begin{equation}\label{eq:universalconstant}
\msf c_{\gV}\deff \left( \msf c_0 \right)^{2/3},
\end{equation}
where $\msf{c}_0$ is the constant appearing in \eqref{eq:squarerootvanishing}. In the result that follows we will use a local coordinate $z\mapsto \varphi(z)$, which is constructed from the equilibrium measure in a canonical way, and it essentially provides an analytic continuation for the density of $\mu_\gV$, see Proposition~\ref{prop:conformalmap} below for further details. For now it suffices to say that $\varphi$ is a conformal map from a neighborhood of $z=\ga$ with $\varphi(\ga)=0$ and $\varphi'(\ga)=-\msf c_\gV$, it is independent of $\gn$, and it is also independent of the deformation parameters $s,t$ in \eqref{eq:deformedweight}.

Our first main result provides leading order asymptotic estimates for $\msf L_\gn(s)/\msf L_\gn(S)$, when $s,S$ both fall within the same parameter regime. In essence, it says that this quotient has leading asymptotics localized near the transition point $\ga$, and that such asymptotics are governed by Airy, Painlevé, or Bessel-type kernels, depending on whether we are in the sub, critical, or supercritical regimes, respectively. 

\begin{theorem}[Localized asymptotics for multiplicative statistics]\label{thm:multstat_formal}
Let $\gw$ satisfy Assumptions~\ref{assumpt:potential_formal} and \ref{assumpt:equilmeasure_formal}, and let $s, S$ satisfy Assumption~\ref{assumpt:parameterregimes} with $s < S$ and $|s|, |S| \leq M^{-1}\gn^{1/2}$ for some large fixed $M > 0$. Then
\begin{multline}\label{eq:multstat_asymp}
\log\frac{\msf L_\gn(s)}{\msf L_\gn(S)}= \\
     \frac{1}{\gn^{1/3}}\int_s^S \sum_{\substack{x\in \Z\\ |x-\ga\gn|< \gn\delta}}\varphi'\left(\frac{x}{\gn}\right)\msf H_\gn\left( \gn^{2/3}\left( \varphi\left( \frac{x}{\gn} \right) - \varphi\left( \ga-\frac{v}{t\gn}  \right) \right) \mid v \right)\frac{\ee^{-t(x-\gn\ga +v/t)}}{(1+ \ee^{- t(x-\gn\ga +v/t)})^2}\dd v 
     + \msf R,
\end{multline}
as $\gn \to \infty$, for any $\delta > 0$ sufficiently small. The function $\msf H_\gn(\zeta\mid \sad)$ and the error term $\msf R=\msf R(s,S)$ have the following asymptotic behavior in each regime:

\begin{enumerate}[(i)]
\item \textbf{Subcritical regime: }
For some $\eta>0$ and $s,S$ in the subcritical regime, the asymptotic formula
         $$
        \msf H_\gn\left(\zeta-\frac{\msf c_\gV}{t}\frac{s}{\gn^{1/3}} \mid s\right)= \msf A\left( \zeta, \zeta\right)+\Boh\left(\ee^{-\eta s/\gn^{1/3}}\right)
         $$
         is valid uniformly for $\zeta$ in compacts of $\R$, where $\msf A(x,x)$ is the diagonal of the Airy kernel \eqref{deff:AiryKernel}, and the estimate
         $$
         \msf R(S,s)=\Boh\left(\frac{\ee^{-\frac{4}{3}s_\gn^{3/2}}-\ee^{-\frac{4}{3}S_\gn^{3/2}}}{\gn^{1/3}}\right), \qquad s_\gn\deff \frac{\msf c_{\gV}}{t} \frac{s}{\gn^{1/3}},\quad S_\gn\deff \frac{\msf c_{\gV}}{t}\frac{S}{\gn^{1/3}},
         $$
         is also valid.

\item \textbf{Critical regime: }
For $s,S$ in the critical regime, the asymptotic formula
         $$
         \msf H_\gn(\zeta\mid s)= \msf K_\ptf\left( \zeta,\zeta\mid \frac{\msf c_\gV}{t}\frac{s}{\gn^{1/3}} \right)+\Boh(\gn^{-1/3}),
         $$
         is valid uniformly for $\zeta$ in compacts of $\R$, where $\msf K_\ptf(\cdot,\cdot\mid y)$ is the Painlevé XXXIV kernel of \eqref{deff:P34kernel} in the Painlevé variable $y$, and the estimate
         $$
         \msf R(S,s)=\Boh\left( \frac{|S|^{3/2}+|s|^{3/2}}{\gn^{5/6}} \right)
         $$
         is also valid. 

\item \textbf{Supercritical regime: }
For $s,S$ in the supercritical regime, the asymptotic formula
         $$
         \msf H_\gn(\zeta\mid s)= -\frac{\msf c_\gV^2}{t^2}\frac{s^2}{\gn^{2/3}}\msf J_0\left( \frac{\msf c_\gV^2}{t^2}\frac{s^2}{\gn^{2/3}} \zeta,\frac{\msf c_\gV^2}{t^2}\frac{s^2}{\gn^{2/3}}\zeta \right) 
         +\Boh\left( \max\left\{ \frac{|s|}{\gn^{1/3}}, \frac{s^4}{\gn^{5/3}} \right\} \right),
         $$
         is valid uniformly for $|\zeta|\leq \delta$ and $\delta>0$ sufficiently small, where $\msf J_0$ is the Bessel kernel \eqref{deff:BesselKernelintro}, and the estimate
         $$
         \msf R(S,s)=\Boh\left( \frac{|s|^{9/4}-|S|^{9/4}}{\gn^{13/12}} \right)
         $$
         is also valid.
\end{enumerate}
All error terms are uniform in the stated parameter ranges.
\end{theorem}

A Taylor expansion around $x=\ga$ yields
$$
\varphi'\left(\frac{x}{\gn}\right)=-\msf c_\gV +\Boh\left(\frac{x}{\gn}-\ga  \right),
$$
as well as
$$
\varphi\left( \frac{x}{\gn} \right) - \varphi\left( \ga-\frac{v}{t\gn}  \right)=-\frac{\msf c_\gV}{\gn}\left( x-\gn\ga+\frac{v}{t} \right)+\Boh\left( \max\left\{ \left(\frac{x}{\gn}-\ga\right)^2 , \frac{s^2}{\gn^{2}} \right\}\right).
$$
These estimates explain how to go from \eqref{eq:multstat_asymp} to its informal version stated in \eqref{eq:LNasymptintrothm}.

The universal character of random matrix theory statistics is reflected upon Theorem~\ref{thm:multstat_formal} in the leading behavior of $\msf H_\gn$: in the subcritical regime it is described by the Airy kernel, the (regular) soft-edge universal limit of random matrix theory, in the supercritical regime it is described by the Bessel kernel, the (regular) hard-edge universal limit, and in the transitional critical regime it is described by the Painlevé XXXIV kernel, the transitional soft-to-hard edge universal limit of RMT \cite{ClaeysKuijlaars2008HS}. 

In the leading term on the right-hand side of \eqref{eq:multstat_asymp}, however, there are two non-universal quantities. The first one is the conformal map $\varphi$, which is constructed from the equilibrium measure of the system. Even though it is non-universal, it should only be viewed as providing the correct coordinate system for the emergence of universality. The second non-universal factor is the exponential term, and in the local coordinate system $\zeta=\gn^{2/3}\varphi(z)\approx -\gn^{2/3}\msf c_\gV(z-\ga)$ with $z=x/\gn$ and $c=\msf c_\gV/t$ it writes as
$$
\frac{\ee^{-t(x-\gn\ga +v/t)}}{(1+ \ee^{- t(x-\gn\ga +v/t)})^2}\approx \frac{\ee^{c\gn^{1/3}\zeta-v}}{(1+\ee^{c\gn^{1/3}\zeta-v})^2}.
$$
The right-hand side has two competing quantities. As a function of $\zeta$, it converges pointwise to $0$, indicating that this quantity should become trivial in the local scale of fluctuations $\zeta=\Boh(\gn^{2/3})$. However, as a function of $v$, it may well diverge in our regimes of interest. 

Nevertheless, after a careful analysis of all the terms involved, the local asymptotics of Theorem~\ref{thm:multstat_formal} yield explicit asymptotic formulas for $\log\msf L_\gn(s)$. Such formulas are stated in our next result, which is the complete version of Theorem~\ref{mainthmintrointegrals}.

\begin{theorem}[Universal global asymptotics]\label{thm:integrated_formal}
Under the same hypotheses of Theorem~\ref{thm:multstat_formal},
\begin{enumerate}[(i)]
\item \textbf{Subcritical regime:} There exists $\eta>0$ such that for $s$ in the subcritical regime, the estimate
\begin{equation}\label{eq:integratedsubcritical}
\log\msf{L}_\gn(s) = -\left(1+\Boh\left(\ee^{-\eta s/\gn^{1/3}}+\gn^{-\nu}\right)\right)\int_{\msf{c}_\gV s/(t\gn^{1/3})}^\infty \msf{A}(y, y)\,\dd y + \Boh(\ee^{-\eta \gn^{1/4}}),\quad \gn\to \infty,
\end{equation}
is valid for any $\nu\in (0,1/6)$.

\item \textbf{Critical regime:} For $s$ in the critical regime, the estimate
\begin{equation}\label{eq:integratedcritical}
\log\msf{L}_\gn(s) = \left(1+\Boh(\gn^{-\nu})\right)\log F_{\rm GUE}\left(\frac{\msf{c}_\gV s}{t\gn^{1/3}}\right),\quad \gn\to \infty,
\end{equation}
is valid for any $\nu\in(0,1/3)$, where $F_{\rm GUE}$ is the $\beta=2$ Tracy-Widom distribution.

\item \textbf{Supercritical regime:} For $s$ in the supercritical regime, the estimate
\begin{equation}\label{eq:integratedsupercritical}
\log\msf{L}_\gn(s) = -\frac{\msf c_\gV^3 }{12t^3\gn} |s|^3 +\Boh\left(\frac{s^4}{\gn^{3/2}}\right),\quad \gn\to \infty,
\end{equation}
is valid.
\end{enumerate}
\end{theorem}

As mentioned earlier, the appearance of $F_{\rm GUE}$ in (ii) comes from the identity \eqref{eq:KXXXIVTW} between the Hamiltonian formulation for the Tracy-Widom distribution and the value of the Painlevé XXXIV kernel at the origin. The transition between the subcritical and supercritical regimes is mediated by the Tracy-Widom distribution, and the asymptotic behaviors \eqref{eq:tailsTW} show that the critical regime formula (ii) interpolate smoothly between the other two regimes.

A natural analogue of $\msf L_\gn$ is obtained when replacing the symbol $\gsig_\gn$ in the multiplicative statistic by a true characteristic function, say with an endpoint exactly at $\ga -s/(t\gn)$. In such a case, $\msf L_\gn(\sad)$ consists of the quotient of the partition function of the undeformed weight with a Hankel determinant for a truncated weight. In the setup for continuous weights, analogues of Theorem~\ref{thm:integrated_formal} for such type of Hankel determinants have been considered vastly in the literature in recent times, in particular for weights exhibiting Fisher-Hartwig singularities, see for instance \cite{BogatiskiyClaeysIts2016, BerestyckiWebbWong2018, CharlierGharakhloo2021, BothnerShepherd2024} and the references therein. However, to our knowledge hard-to-soft edge transition results have been restricted to the analysis of kernels  \cite{ClaeysKuijlaars2008HS, CharlierLenells2023}. We are unaware of asymptotic results similar to our Theorem~\ref{thm:integrated_formal} covering a soft-to-hard edge phase transition beyond kernel analysis, even for continuous weights.

\subsection{The Meixner ensemble and connection to the six-vertex model}\label{sec:Meixner}\hfill 

We now specialize to the Meixner ensemble, which provides the key link between the general theory of dOPEs and the stochastic six-vertex model.

\subsubsection{Definition of the Meixner ensemble}

The \emph{Meixner ensemble} with $\mn$ particles, strength parameter $\ms > 0$, and shape parameter $\mt \in (0,1)$ is the dOPE on $\Z_{>0}$ with weight function\footnote{The Meixner ensemble is commonly defined on $\Z_{\geq 0}$, but for technical reasons for us it is convenient to define it on $\Z_{>0}$: in this case, after scaling, during the RHP analysis the lower point in the lattice is $1/N$ instead of $0$, which avoids the need of a technical and unimportant analysis around $0$.}
\begin{equation}\label{eq:Meixnerweight}
\mw(x) = \mw(x \mid \ms, \mt) \deff \frac{\Gamma(\ms + x-1)}{\Gamma(\ms)\Gamma(x)}\mt^{x-1}, \qquad x \in \Z_{>0}.
\end{equation}
The corresponding probability measure is
\begin{equation}\label{eq:Meixnermeasure}
\P_{\rm Meixner(\mn, \ms, \mt)}(\lambda_1, \ldots, \lambda_\mn) = \frac{1}{\msf{Z}_\mn(\mw)} \prod_{1 \leq i < j \leq \mn}(\lambda_i - \lambda_j)^2 \prod_{j=1}^\mn \mw(\lambda_j).
\end{equation}

The Meixner polynomials, which are the orthogonal polynomials with respect to this weight, form one of the classical families in the Askey scheme \cite{KoekoeLeskySwarttouw}. They satisfy a three-term recurrence relation with explicit coefficients, enabling detailed asymptotic analysis via classical methods \cite{JinWong1998}. More detailed asymptotic information for them may also be obtained through Riemann-Hilbert methods \cite{BKMMbook, BleherLiechtyIMRN2011}. In our context here, we are interested in the deformations of the Meixner, and we now verify that they indeed satisfy the conditions required for our main results to hold.

\begin{prop}\label{prop:Meixnerverification}
For $\msf q\in (0,1)$,  $u>\msf q^{-1/2}$ and $1<\nu<u\msf q^{1/2}$, consider the Meixner ensemble with parameters
$\alpha = (\nu-1)\,\gn, \beta = \msf q^{-1/2}u^{-1}\in(0,1),
$ and define the rescaled weight
$\gW_\gn(x) \deff \frac{1}{\gn}\mw(\gn x),$ for $ x\in \tfrac{1}{\gn}\Z_{>0}.$
Then Assumptions~\ref{assumpt:potential_formal} and \ref{assumpt:equilmeasure_formal} hold with:
\begin{enumerate}[(i)]
\item External field
\begin{equation}\label{eq:Meixner-V}
\gV(x) = x\log\frac{x}{\beta(\nu+x-1)}
+ (\nu-1)\log\frac{\nu-1}{\nu+x-1}, \qquad x>0.
\end{equation}

\item Saturated region $[0,\ga]$ and band $[\ga,\gb]$ given by
\begin{equation}\label{eq:Meixner-ab}
\ga = \frac{(1-\sqrt{\beta})^2}{1-\beta}\,(\nu-1),
\qquad
\gb = \frac{(1+\sqrt{\beta})^2}{1-\beta}\,(\nu-1).
\end{equation}
\end{enumerate}
In particular, $\gV$ satisfies Assumption~\ref{assumpt:potential_formal}, and the associated constrained equilibrium measure $\gequil$ satisfies Assumption~\ref{assumpt:equilmeasure_formal} with $\msf c_\gV=\mc^{-1}$ and $\mc$ as in \eqref{eq:scalingconstants}.
\end{prop}

The proof follows from standard potential-theoretic arguments; see \cite{johansson_2000, DasDimitrov22}. For the sake of completeness, we outline its proof in Section~\ref{sec:proofpropMeixnerequil} below.

\subsubsection{The Borodin-Olshanski identity}

The connection between the stochastic six-vertex model and the Meixner ensemble was discovered by Borodin and Olshanski \cite{BorodinOlshanski2017}.

\begin{prop}[Borodin-Olshanski \cite{BorodinOlshanski2017}]\label{prop:BorodinOlshanski}
For the stochastic six-vertex model with parameters $(\mq, u)$ satisfying $\mq\in (0,1)$ and $u>\mq^{-1/2}$, and any $\zeta \in \C \setminus \{-\mq^{\Z_{\leq 0}}\}$,
\begin{equation}\label{eq:BOidentity}
\E_{\rm S6V}\left[\prod_{i \geq 1} \frac{1}{1 + \zeta\mq^{\mathfrak{h}(\mn, \mm) + i}}\right] = \E_{\rm Meixner(\mn, \mm-\mn, \mq^{-1/2}u^{-1})}\left[\prod_{x \in X}(1 + \zeta\mq^x)\right] \prod_{i \geq 0} \frac{1}{1 + \zeta\mq^i},
\end{equation}
where the expectation on the right is over the Meixner ensemble with the indicated parameters.
\end{prop}

This identity expresses a $\mq$-Laplace transform of the height function in terms of a multiplicative functional of the Meixner ensemble. The correspondence of parameters is:
\begin{equation}\label{eq:parametercorrespondence}
n = \mn, \qquad \ms = \mm - \mn, \qquad \mt = \mq^{-1/2}u^{-1}.
\end{equation}

\subsubsection{Asymptotics for the six-vertex model via the Meixner ensemble}

Combining Proposition~\ref{prop:BorodinOlshanski} with Theorems~\ref{thm:multstat_formal} and \ref{thm:integrated_formal} and the choice of parameters
\begin{equation}\label{eq:scalingzetasS6V}
\zeta = \mq^{-\aeq\mn - \msf{h}\mc\mn^{1/3}}\qquad \text{and}\qquad s = -(\msf{h}\mc\mn^{1/3} + 1)\log(1/\mq)
\end{equation}
yields the following result, from which Theorems~\ref{thm:S6VMlowertail_precise} and \ref{thm:S6VMuppertail_precise} follow.

\begin{theorem}[Asymptotics of the six-vertex height function]\label{thm:S6VMasymptotics}
Consider the stochastic six-vertex model with parameters $\mq\in (0,1)$ and $u>\mq^{-1/2}$, and fix $\nu$ satisfying $1 < \mnu < \mq^{1/2}u$. Recall $\aeq$ and $\mc$ from \eqref{eq:scalingconstants} and scale $\zeta$ as in \eqref{eq:scalingzetasS6V}.
Then, the following asymptotic estimates hold true.

\begin{enumerate}[(i)]
\item \textbf{Upper tail regime:} For any $\msf h_0>0$ sufficiently large, there exists $\eta>0$ such that the estimate
\begin{equation}\label{eq:S6Vsubcritical}
 \log\E_{\rm S6V}\left[\prod_{i \geq 1} \frac{1}{1 + \zeta\mq^{\mathfrak{h}(\mnu\mn, \mn) + i}}\right] =-\left(1+\Boh\left(\ee^{-\eta s/\gn^{1/3}}+\gn^{-\nu}\right)\right)\int_{\msf{c}_\gV s/(t\gn^{1/3})}^\infty \msf{A}(y, y)\,\dd y + \Boh(\ee^{-\eta \gn^{1/4}}),
\end{equation}
is valid for any $\nu\in (0,1/6)$, uniformly for $\msf h_0\leq \msf h\leq \msf h_0^{-1}\gn^{1/6}$.

\item \textbf{Fluctuations regime:} Given $\msf h_0>0$, the estimate
\begin{equation}\label{eq:S6Vcritical}
 \log\E_{\rm S6V}\left[\prod_{i \geq 1} \frac{1}{1 + \zeta\mq^{\mathfrak{h}(\mnu\mn, \mn) + i}}\right] =\left(1+\Boh(\gn^{-\nu})\right)\log F_{\rm GUE}\left(\frac{\msf{c}_\gV s}{t\gn^{1/3}}\right),
\end{equation}
holds true for any $\nu\in (0,1/3)$, uniformly for $|\msf h|\leq \msf h_0$.

\item \textbf{Lower tail regime:} For any $\msf h_0>0$ sufficiently large, the estimate
\begin{equation}\label{eq:S6Vsupercritical}
 \log\E_{\rm S6V}\left[\prod_{i \geq 1} \frac{1}{1 + \zeta\mq^{\mathfrak{h}(\mnu\mn, \mn) + i}}\right] = -\frac{\mc^{-3} }{12t^3\gn} |s|^3 +\Boh\left(\frac{s^4}{\gn^{3/2}}\right),
\end{equation}
is valid uniformly for $-\msf h_0^{-1}\gn^{1/6}\leq \msf h\leq -\msf h_0$.
\end{enumerate}
\end{theorem}

The derivation of Theorems~\ref{thm:S6VMlowertail_precise} and \ref{thm:S6VMuppertail_precise} from Theorem~\ref{thm:S6VMasymptotics} proceeds via soft probabilistic arguments relating the $\mq$-Laplace transform to tail probabilities. Assuming Theorem~\ref{thm:S6VMasymptotics}, we complete their proofs in Section~\ref{sec:S6Vproofs}.

The starting input in the proof of Theorem~\ref{thm:S6VMasymptotics} is the connection between the stochastic six-vertex model and the Meixner ensemble, as outlined in Proposition~\ref{prop:BorodinOlshanski}. Because of this, the slope parameter $\nu$ must satisfy $\nu>1$. Indeed, the strength parameter $\alpha$ of the associated Meixner orthogonal polynomial ensemble takes the form $(\nu-1) N$, which must be positive. Although this may at first appear to be a restriction, the regime $\nu<1$ can also be investigated, as we explain in the next remark.

\begin{remark}[Extension to $\msf q^{-1/2} u^{-1}<\nu\leq 1$]\label{rmk:extension_nu_leq_1}
While Theorem~\ref{thm:S6VMasymptotics} is stated for the case $\nu > 1$, the techniques developed in this paper extend to the regime $\msf q^{-1/2} u^{-1}<\nu\leq 1$ using a variant of the Borodin-Olshanski identity. Specifically, for $M \leq N$, \cite[Remark~8.6]{BorodinOlshanski2017} establishes that
\begin{equation}\label{eq:BOidentity_shift}
\E_{\rm S6V}\left[\prod_{i \geq 1} \frac{1}{1 + \zeta \mq^{\mathfrak{h}(M,N) + i}}\right]
=
\E_{\rm Meixner^\circ(M-1, N-M+2, \mq^{-1/2}u^{-1})}
\left[\prod_{x \in X+(N-M+1)} \frac{1}{1 + \zeta \mq^x}\right],
\end{equation}
where $X + S$ denotes the deterministic shift of all points in the random configuration $X$ by $S$. The deterministic shift by $N-M+1$ reflects the fact that $\mathfrak{h}(M,N) \geq N-M+1$ in this regime. The asymptotic analysis of the multiplicative statistics in \eqref{eq:BOidentity_shift} requires verifying a modified version of Proposition~\ref{prop:Meixnerverification} for the shifted Meixner ensemble, as detailed in Proposition~\ref{prop:Meixnerverification_shifted} below. In this modified setup, the parameter $\gn$ is replaced by $\widetilde{\gn} = \lfloor\nu\mn\rfloor -1$, after scaling particles start at positions to the right of $\nu^{-1}-1$, the saturated region becomes $[\nu^{-1}-1, \msf{a}]$ where $\msf{a} = \frac{(1-\sqrt{\beta})}{(1-\beta)}(\nu^{-1}-1)$, and the band region becomes $[\msf{a}, \msf{b}]$ where $\msf{b} = \frac{(1+\sqrt{\beta})}{(1-\beta)}(\nu^{-1}-1)$. The external field $\gV(x)$ is modified by shifting its argument by $(\nu^{-1} - 1)$. Once this verification is complete, the general framework of Theorem~\ref{thm:multstat_formal} applies directly, yielding the analogue of Theorem~\ref{thm:S6VMasymptotics} for $\msf q^{-1/2} u^{-1}<\nu\leq 1$. We omit the detailed calculations for brevity, focusing instead on the case $\nu > 1$.
\end{remark}

In the same paper \cite[Remark~8.7]{BorodinOlshanski2017}, the authors showed that  
\begin{equation}\label{eq:BOidentity_Krawtchouk}
\E_{\rm S6V}\left[\prod_{i \geq 1} \frac{1}{1 + \zeta \mq^{\mathfrak{h}(\mn,\mm) + i}}\right]
=
\E_{\rm Krawtchouk(\mm-1, (1+\mq^{-1/2}u^{-1})^{-1}, \mn+\mm-2)}
\left[\prod_{x \in X} \frac{1}{1 + \zeta \mq^x}\right],
\end{equation}
where $\rm Krawtchouk(\mm-1, (1+\mq^{-1/2}u^{-1})^{-1}, \mn+\mm-2)$ is the $(\mm-1)$-particle orthogonal polynomial ensemble on $\{0,\ldots,\mn+\mm-2\}$ with weight function
\[
w(x)
=
\binom{\mn+\mm-2}{x}
(1+\mq^{-1/2}u^{-1})^{-x}
\left( \frac{\mq^{-1/2}u^{-1}}{1+ \mq^{-1/2}u^{-1}} \right)^{\mn+\mm-2-x},
\qquad x \in \{0,\ldots,\mn+\mm-2\}.
\]

One may alternatively use \eqref{eq:BOidentity_Krawtchouk} to derive the asymptotics for the $\mq$-Laplace transform of the stochastic six-vertex model height function. However, we note that such an asymptotic result will not directly follow from our general theorem on multiplicative statistics for discrete orthogonal polynomial ensembles (Theorem~\ref{thm:multstat_formal}), because the Krawtchouk ensemble violates one of the assumptions of that theorem, namely, that the ensemble be supported on a half-infinite integer lattice rather than on a finite interval. Nevertheless, Theorem~\ref{thm:multstat_formal} can be extended to include the Krawtchouk case, with appropriate modifications to the proof.

In \cite{BorodinOlshanski2017}, Borodin and Olshanski also obtain a reduction to \eqref{eq:BOidentity},\eqref{eq:BOidentity_Krawtchouk}, corresponding to the degeneration of the stochastic six-vertex model to the ASEP. We foresee the ideas developed in this paper to also apply to this reduction, and yield similar results for the ASEP. However, such reduction requires analyzing also the limit $t\to 0$ simultaneously with $\gn\to \infty$, so unlike the mentioned Krawtchouk situation the rigorous analysis requires investigation that goes beyond what we do here, and we leave it to a future work.

\subsection{Further connections}\label{sec:furtherconnection} \hfill 

Tail probabilities of stochastic processes are a fundamental tool in understanding the qualitative behaviors of various random models, particularly those in the Kardar-Parisi-Zhang (KPZ) universality class. For models within this class, tail probabilities have been extensively studied through diverse mathematical frameworks, such as random matrix theory, Riemann-Hilbert analysis, Schrödinger operators, weak noise theory, and the Gibbs property of the Brownian ensemble, among others. These approaches have enabled insights into the asymptotic behaviors and extremal statistics of these models.

A prominent example is the last passage percolation model with exponential weights, which belongs to the KPZ universality class. The foundational work of Johansson \cite{johansson_2000} initiated the study of tail probabilities for this model, with large deviation results. About the same time, L{\"o}we and Merkl \cite{LoweMerkl2002upper, LoweMerkl2002lower} obtained optimal moderate deviations for the length of the longest increasing subsequence in random permutations. For last passage percolation with geometric weights, Baik, Deift, McLaughlin, Miller and Zhou \cite{BaikDeiftMcLaughlinMillerZhou2001} obtained optimal moderate deviation estimates utilizing Riemann-Hilbert methods, see also the recent work \cite{ByunCharlierMoreillonSimm2025} by Byun, Charlier, Moreillon and Simm, who improve these bounds in the large deviations scale. Later contributions by Ledoux and Rider \cite{LedouxRider10} established upper bounds for both the upper and lower tail probabilities, valid also in the moderate deviations regime, for beta ensembles, which by distributional identities include some of the mentioned discrete growth models. 

For the KPZ equation itself, studies of tail properties were pioneered by Corwin and the first-named author, in \cite{CorwinGhosal2020}, where the authors derived tight bounds on the lower and upper tail probabilities for the narrow wedge initial condition. A crucial aspect of this analysis involved an identity analogue to \eqref{eq:BOidentity}, connecting the double exponential moments of the KPZ equation, with the expectation of a multiplicative functional of the Airy\(_2\) point process. This connection, originally noted by Borodin and Gorin in \cite{BorodinGorin2016}, facilitated the application of random matrix techniques to derive precise tail bounds, including in the large deviation regime \cite{CorwinGhosal2020}.

In subsequent work, Corwin and the first-named author \cite{CorwinGhosal2020b} leveraged the Brownian Gibbs property of the KPZ line ensemble to generalize these tail probability bounds to a broader set of initial conditions. This research opened the door to numerous further studies on the exact characterization of tail probabilities for the KPZ equation. Notable advancements in this direction include works by Cafasso and Claeys \cite{CafassoClaeys2021}, Cafasso, Claeys, and Ruzza \cite{CafassoClaeysRuzza2021}, Tsai \cite{Tsai22}, Das and Tsai \cite{DasTsai21}, Ghosal and Lin \cite{GhosalLin2023}, and Ganguly and Hegde \cite{GangulyHegde2022}. These studies have employed refined techniques from integrable systems, probabilistic analysis, and determinantal point processes to further elucidate the precise asymptotic behavior and distributional properties of the KPZ equation’s tail probabilities.

As mentioned earlier, Das, Liao and Mucciconi obtained recently lower tail large deviation results for the stochastic six-vertex model \cite{DasLiaoMucciconi2025}. Their starting point is also Borodin-Olshanski's identity \eqref{eq:BOidentity}, but their asymptotic techniques are very distinct from ours. Shortly before, in \cite{DasLiaoMucciconi2023} the same authors also obtained large deviation results for the \(q\)-deformed polynuclear growth (PNG) model, a system introduced in \cite{AggarwalBorodinWheeler2023} that generalizes PNG to include a deformation parameter \(q\), using a similar identity relating the $q$-PNG height function to a multiplicative statistics of a point process (see also \cite{CafassoMucciconiRuzza2026} for further development in this direction). For the six-vertex model speed process, Drillick and Haunschmid-Sibitz \cite{DrillickHaunschmid-Sibitz2024} derived tail bounds by leveraging Borodin-Olshanski's formula combined with concentration inequalities, though their bounds do not capture the precise exponential rate. In related work on discrete models, Corwin and Hegde \cite{corwin2024lower} obtained tail bounds for the $q$-pushTASEP using a combination of coupling arguments with last passage percolation and moment bound techniques, demonstrating yet another approach to understanding tail behavior in integrable stochastic systems.

\subsection{Outline of the proof and organization of the paper}\label{sec:outline}\hfill

The paper is organized naturally into four parts. Each part is largely self-contained, although a few key results and pieces of notation are used across different parts. Conceptually, each part may be viewed as providing a “black box’’ that is invoked in the subsequent parts. These parts are roughly described as follows.

\begin{enumerate}
    \item Part 1: Proof of the tail bounds of S6V; a general deformation formula for the partition functions of OP ensemble; rough estimates for $\msf L_\gn(s)$, when $s\geq \delta \gn^{1/2}$.
    \item Part 2: introduction of a novel Riemann-Hilbert problem, and its asymptotic analysis.
    \item Part 3: analysis of the Riemann-Hilbert problem for the deformed orthogonal polynomials, for $s$ in one of the regimes from Assumption~\ref{assumpt:parameterregimes}.
    \item Part 4: asymptotic analysis of the deformation formula \eqref{eq:deformationformula}.
\end{enumerate}

We now elaborate on each of these parts.

\subsubsection{About Part 1} Part 1 consists of Sections~\ref{sec:S6Vproofs},~\ref{sec:defformulaproof} and~\ref{sec:aprioriest}, which are three independent sections.

Assuming the asymptotics of the multiplicative statistics of Meixner ensemble from Theorem~\ref{thm:integrated_formal}, Section~\ref{sec:S6Vproofs} explains how these results yield the large deviation principles for the stochastic six-vertex model stated in Theorems~\ref{thm:S6VMuppertail_precise} and~\ref{thm:S6VMlowertail_precise}. The proofs of both upper and lower tail results follow a three-step strategy that bridges probability tail bounds with moment asymptotics. The key technical bridge is provided by Lemmas~\ref{lem:UpTail} and~\ref{lem:LowTail}, which relate tail probabilities of the rescaled height fluctuation $\mathfrak h$ to the multiplicative statistic $\E_{\mathrm{S6V}}[\cdot]$ appearing on the left-hand side of \eqref{eq:BOidentity} with a key choice of $\zeta$. The proofs of these comparison lemmas exploit the concentration of the product in the expectation on events where the height function takes specific values; by carefully splitting this product at an appropriate cutoff index and using logarithmic inequalities, we show that deviations of the height function away from the target value lead to exponentially suppressed contributions that are negligible compared to the main deviation exponents of order $\msf{h}^3$ or $\msf{h}^{3/2}$. Once the tail probabilities are expressed in terms of this expectation $\E_{\mathrm{S6V}}[\cdot]$, we apply the asymptotic formulas from Theorem~\ref{thm:S6VMasymptotics} to conclude Theorems~\ref{thm:S6VMlowertail_precise} and \ref{thm:S6VMuppertail_precise}.

Section~\ref{sec:defformulaproof} establishes Proposition~\ref{Prop:deformationPartFctionInt}, which provides the general deformation formula that underlies our entire asymptotic analysis and is of independent interest beyond the specific application to the Meixner ensemble and S6VM model. The setup considers a reference Borel measure $\mu_0$ with finite moments and a one-parameter family of deformed measures. The key result expresses the logarithmic derivative of the partition function as an integral involving the kernel diagonal and the derivative of the weight. This deformation formula allows us to translate the asymptotic analysis of the RHP for deformed orthogonal polynomials (Part 3 \& 4) into asymptotic formulas for the partition function ratios, and hence for the multiplicative statistics that encode the tail probabilities of the S6VM height function. The proof of the deformation formula we provide starts from a known relation between the log derivative of the partition function and a sum of log derivatives of norming constants. The latter are then related to a total weighted integral of a derivative of the deformed kernel. From there, we use integration by parts and the reproducing property of the kernel to arrive at the result. 

While Theorem~\ref{thm:multstat_formal} itself provides precise asymptotics for the multiplicative statistics $\msf S_\gn(s)$ when $|s| \leq \delta\gn^{1/2}$, its proof also requires obtaining a rough bound for the regime $s=\delta\gn^{1/2}$ separately, which Section~\ref{sec:aprioriest} addresses through a probabilistic analysis exploiting the determinantal structure of OPEs.

To obtain such bound, we work directly with determinantal structure. At a scale $s=\delta\gn^{1/2}$, the leading nontrivial contributions to the statistic $\msf S_\gn(s)$ from \eqref{eq:multstat} comes at particles that sit far to the left of the saturated-to-band edge $\ga$. In this region, particles are sitting at every site with very high probability, and we are able to localize $\msf S_\gn(s)$ for contributions coming solely from this region. We then use arguments relying on the determinantal structure, together with basic probabilistic inequalities, to reduce the analysis of $\msf S_\gn(s)$ to the analysis of the counting function of an appropriately chosen interval that sits within the saturated region. By then exploiting the Fredholm series representation, together with known asymptotics for the correlation kernel of the undeformed point process, we are able to obtain sufficient control on this counting function, that ultimately translates into a for $\msf S_\gn(\delta\gn^{1/2})$ which is not sharp, but sufficient for our purposes.

\subsubsection{About Part 2}

Part 2 consists of Sections~\ref{sec:modelRHP}, \ref{sec:relatedRHPs}, \ref{sec:DZModelProblem} and \ref{sec:unwrapsmodel}. As usual in RHP literature, the asymptotic analysis of the RHP for the deformed orthogonal polynomials requires the construction of a so-called local parametrix. In our case, thanks to both the discrete nature of the original RHP and the deformed weight that lives on a critical scale, this local parametrix is based on a novel model RHP, and Part 2 is devoted to studying this RHP, as we now outline.

In Section~\ref{sec:modelRHP} we introduce this model problem in a self-contained language, and obtain some auxiliary results for quantities related to it. But unlike previous literature, this model problem is not independent of the large parameter $\gn$: depends on the parameter $s$ of the deformation, which grows with $\gn$. To our knowledge, the first appearance of such $\gn$-dependent model problems in the context of OPs was observed in our previous work \cite{GhosalSilva2023}. In the mentioned work the model problem used essentially had a regular large $\gn$-limit without the need of performing any other scaling. However, over here this is not quite the case. As a consequence of the fact that we are dealing with weights that blow up with $\gn$, the asymptotic analysis of the model problem is a task on its own, and has to be performed separately for each of the regimes in $s$.

To cope with the asymptotic analysis of the model problem, three different RHPs have to be studied, and in Section~\ref{sec:relatedRHPs} we discuss them. The first two of them are related to the integro-differential Painlevé II, and to the Painlevé XXXIV equations. These RHPs are not new, but we need to tailor several aspects of them to our needs in the asymptotic analysis of the model problem. For the former, we perform a quick asymptotic analysis to show that it is asymptotically described by the Airy RHP. For the latter, we collect some basic facts, in particular showing the relations \eqref{eq:KptfPXXXIV} and \eqref{deff:P34kernel}, which to the best of our knowledge are novel.

But, as it turns out, the analysis of the model problem in the supercritical regime is more subtle, and requires the introduction of an yet new model problem, which we call the modified model problem, and which differs from our initial model problem by its asymptotic behavior at $\infty$. The modified model problem is introduced Section~\ref{sec:modRHPmodel}, and its necessary asymptotic analysis is performed in the same section. As a result, we show that it is asymptotically close to the Bessel RHP. 

In Section~\ref{sec:DZModelProblem}, we perform the asymptotic analysis of the model problem, in which the aforementioned two RHPs are used. In the subcritical regime, we compare it with the RHP for the integro-differential Painlevé II, but in a regime where the latter is asymptotically described by the Airy RHP. In the critical regime, we need the construction of a local parametrix and a global parametrix. The local parametrix is explicitly constructed through Cauchy integrals. The global parametrix, instead, is constructed using the Painlevé XXXIV RHP mentioned. Continuing in Section~\ref{sec:DZModelProblem}, the asymptotic analysis of the model problem in the supercritical regime requires the construction of an explicit global parametrix, and a local parametrix in terms of the modified model problem.

Section~\ref{sec:unwrapsmodel} may be seen as the black-box from part 2: in this section we summarize all the asymptotic results from the previous sections of this part, tailored for direct application in the subsequent parts.

Apart from minor technical assumptions, the model problem that we introduce may be seen as a generalization of an RHP introduced by Cafasso, Claeys and Ruzza for the integro-differential Painlevé II equation \cite{CafassoClaeys2021, CafassoClaeysRuzza2021}. This extension is nontrivial, in the sense that it involves one more function (spectral data) and jumps in additional contours when compared to \cite{CafassoClaeys2021, CafassoClaeysRuzza2021}. Such extra data essentially encodes the fact that, at its core, we are using the model problem for a problem of discrete nature. A detailed comparison is explained in Remark~\ref{rmk:intdiffPII} below. 

The asymptotic analysis of the model problem we just outlined follows the major lines performed in \cite{CafassoClaeysRuzza2021}, and the final outputs are described by the same objects, the Airy, Painlevé XXXIV and Bessel RHPs. However, in virtue of this additional spectral data, several technical constructions have to be performed substantially differently than in the mentioned work, and the quantities we need from the model RHP are different than the ones studied in \cite{CafassoClaeysRuzza2021}. In particular, in \cite{CafassoClaeysRuzza2021} the analysis of the RHP yields asymptotics for their functions of interest in a straightforward manner, as they are contained in a certain residue of the RHP at $\infty$. In contrast, in Part 2 the major quantity of interest for us is essentially the function $\msf H_\gn$ from Theorem~\ref{thm:multstat_formal}, and even after completing the asymptotic analysis of the model problem there is still considerable labor needed to extract asymptotics for $\msf H_\gn$ in the detailed manner we need.

\subsubsection{About Part 3} Part 3 comprises Sections~\ref{sec:RHPOP}, \ref{sec:asymptanaldOP}, and \ref{sec:unwrapOPRHP}. In Section~\ref{sec:RHPOP}, we recall the Riemann-Hilbert problem for discrete orthogonal polynomials, list relevant identities, and implement the scaling corresponding to the lattice transformation $\Z_{>0}\mapsto \frac{1}{\gn}\Z_{>0}$.

In Section~\ref{sec:asymptanaldOP}, we perform the Deift-Zhou nonlinear steepest descent analysis for the discrete orthogonal polynomials associated with the deformed weight $\gwd$ in \eqref{eq:deformedweight}. While the transformations generally follow standard literature \cite{BKMMbook, BleherLiechtyIMRN2011}, we introduce necessary modifications to address specific obstructions. First, the deformation component $(1+\ee^{\bullet})$ of the deformed weight $\gwd$ introduces poles in the complex plane for the inverse $(\gwd)^{-1}$. As $\gn\to \infty$, these poles, after scaling, accumulate near the real axis at $z=\ga-s/(t\gn)$. Consequently, we cannot open lenses around this point, and our transformations must accommodate this singularity structure. Second, due to these poles, the local parametrix near $z=\ga$ must be constructed using the novel model RHP developed in Part 2. Crucially, this local parametrix retains dependence on the large parameter $\gn$, even at the local scale.

The asymptotic analysis concludes at the end of Section~\ref{sec:asymptanaldOP}. Subsequently, in Section~\ref{sec:unwrapOPRHP}, we trace back the explicit transformations. Our primary objective is to evaluate the right-hand side of \eqref{eq:deformationformula}, which requires the values of the kernel $\msf K_\gn^\gsig(x,x\mid u)$ along the entire positive axis. We therefore invert the transformations in all regions of the positive axis, accounting for cancellations between various terms. The final result of Section~\ref{sec:asymptanaldOP} expresses the right-hand side of \eqref{eq:deformationformula} via quantities defined in Part 2, in a format suitable for the application of the framework established therein.

\subsubsection{About Part 4} Part 4 consists of Sections~\ref{sec:AsymptoticMultiplicative} and \ref{sec:ProofOfMain}. Having established the asymptotic behavior of the RHP for deformed orthogonal polynomials and an appropriate expression for the right-hand side of \eqref{eq:deformationformula} in terms of RHP quantities in Part 3, we now turn to the asymptotic analysis of the latter. This part constitutes the final step in proving our main results.

The starting point is the rewriting of formula from Section~\ref{sec:summary_RHP}, which at the end takes the form
$$
\log \frac{\msf Z_\gn^\gsig(s)}{\msf Z^\gsig_\gn(S)}=\msf S(S,s) + \msf R(S,s) + \text{explicit terms},
$$
where $\msf S$ captures the main asymptotic contribution and $\msf R$ represents terms arising from the RHP analysis, which are to be shown to be error term. All these terms are expressed in quantities related to the model problem from Part 2 and the RHP for OPs from part 3, in a format ready for application of the black box summarized in Section~\ref{sec:summary_RHP}. Section~\ref{sec:AsymptoticMultiplicative} is devoted to the asymptotic analysis of these terms, through applications of results from Section~\ref{sec:summary_RHP} together with further asymptotic methods as we now explain. 

We begin by showing that certain explicit terms involving sums over the gap region are exponentially negligible. The analysis then splits into estimating $\msf R$ and $\msf S$ separately, with the key challenge being that their asymptotic behavior depends critically on which regime (subcritical, critical, or supercritical) the parameter $s$ belongs to.

For the error term $\msf R$, we exploit the estimates from Part 3 on the transformed RHP solution. The main observation is that contributions from points far from the band edge $\ga$ decay exponentially due to the Euler-Lagrange conditions. For points near $\ga$ but outside a small neighborhood (the ``outer'' region), we use crude bounds on the model RHP showing exponential decay. The dominant contribution comes from the ``inner'' region near $\alpha$, where we apply the precise asymptotics from Part 2. However, even these inner contributions are shown to be subleading compared to $\msf S$ in each regime. The key technical tool is a careful decomposition of the summation domain combined with the growth/decay estimates on the conformal map $\varphi$ and the model problem controlling the deformation.

\begin{figure}[htbp]
\centering
\begin{tikzpicture}[scale=1.2]
    \draw[thick, ->] (-0.5,0) -- (8,0) node[right] {$x$};
    
    \coordinate (zero) at (0,0);
    \coordinate (alpha) at (3,0);
    \coordinate (beta) at (6,0);
    
    \fill[blue!10] (0,0) rectangle (2.3,-1.5);
    \draw[thick, blue!50] (2.3,0) -- (2.3,-1.5);
    
    \fill[orange!20] (2.3,0) rectangle (3.7,-1.5);
    \draw[thick, orange!70] (3.7,0) -- (3.7,-1.5);
    
    \fill[red!30] (2.7,0) rectangle (3.3,-1.5);
    \draw[thick, red, dashed] (2.7,0) -- (2.7,-1.5);
    \draw[thick, red, dashed] (3.3,0) -- (3.3,-1.5);
    
    \fill[green!15] (3.7,0) rectangle (5.3,-1.5);
    \draw[thick, green!60] (5.3,0) -- (5.3,-1.5);
    
    \fill[blue!10] (5.3,0) rectangle (7.5,-1.5);
    
    \fill (alpha) circle (2pt) node[above] {$\ga$};
    \fill (beta) circle (2pt) node[above] {$\gb$};
    \node at (2.3,0) [above] {$\ga-\epsilon$};
    \node at (3.7,0) [above] {$\ga+\epsilon$};
    \node at (5.3,0) [above] {$\gb-\epsilon$};
    
    \node at (1.15,-0.75) {\small Saturated};
    \node at (1.15,-1.0) {\small region};
    \node at (4.5,-0.75) {\small Band};
    \node at (4.5,-1.0) {\small region};
    \node at (6.4,-0.75) {\small Gap};
    \node at (6.4,-1.0) {\small region};
    
    \draw[<-, thick, red] (3,-1.5) -- (5,-2.2) node[below, align=center] {
        \textbf{Inner region:} \\
        \small Precise asymptotics \\
        \small from RHP applied \\
        \small (still subleading)
    };
    
    \node[orange!70, align=center] at (3,0.7) {\small Outer: crude bounds,};
    \node[orange!70, align=center] at (3,0.4) {\small exponential decay};
    
    \draw[<-, thick, blue!70] (1.15,-1.5) -- (0.5,-2.2) node[below, align=center] {
        \textbf{Far regions:} \\
        \small Exponential decay \\
        \small via Euler-Lagrange
    };
    
    
    
    
\end{tikzpicture}
\caption{Spatial decomposition for analyzing the error term $\msf R$. The summation domain is divided into regions based on distance from the band edge $\ga$. Far regions (saturated, gap and band, depicted in blue and green) contribute exponentially small terms $O(\ee^{-\eta\gn})$ via Euler-Lagrange inequalities. The outer region (orange) near $\ga$ yields exponential decay through crude bounds on the model RHP. The inner region (red) requires precise asymptotics from Part 3, but remains subleading to the main term $\msf S$. The conformal map $\varphi$ and the model problem from Part 2 are essential technical tools.}
\label{fig:decomposition_R}
\end{figure}
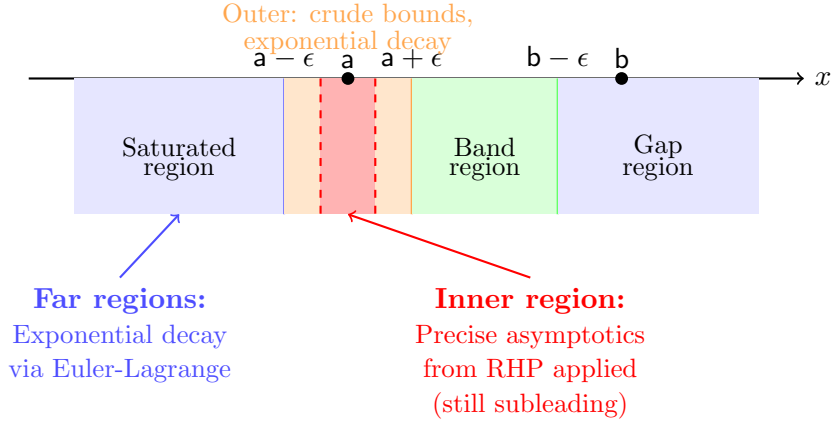

The analysis of $\msf S$ is more delicate and constitutes the heart of Part 4. The quantity $\msf S$ is expressed as a sum involving the function $\msf H_\gn$, whose asymptotics were established in Part 2. Specifically, from the deformation formula we obtain
\begin{equation}
\msf S(S,s) = -\frac{1}{\gn}\int_s^S \sum_{x\in \mcal X^\gin} \zeta'(x)\msf H_\gn(\zeta(x)-\zeta_\gn(v)\mid v) \frac{\ee^{-\gn t(x-z_\gn(v))}}{(1+\ee^{-\gn t(x-z_\gn(v))})^2}\,\dd v.
\end{equation}
As mentioned earlier, the sum over spatial points cannot be approximated by a Riemann integral because the weight factor $\ee^{-\gn t(x-z_\gn(v))}/(1+\ee^{-\gn t(x-z_\gn(v))})^2$ already lives at the correct scale. Instead, we employ a strategy based on Poisson summation.

The key insight is that this weight factor is sharply concentrated around the point $z_\gn(v)$ where the deformation parameter vanishes. Near this point, both the conformal map derivative $\varphi'$ and the function $\msf H_\gn$ are approximately constant, allowing us to factor them out. The remaining sum over lattice points is then exactly computed using the Poisson summation formula (Lemma~\ref{lem:PoissonSummation}), which gives
\begin{equation}
\sum_{k\in\mathbb{Z}}\frac{\ee^{-tk+u}}{(1+\ee^{-tk+u})^2} = \frac{1}{t} + \frac{8\pi}{t^2}\sum_{k=1}^\infty \frac{k\cos(2\pi ku/t)}{\ee^{2\pi^2k/t}-\ee^{-2\pi^2k/t}}.
\end{equation}
The leading term $1/t$ provides the main contribution, while the oscillatory cosine terms are shown to be negligible through integration by parts {\it a là} stationary phase arguments.

For the subcritical regime, where $\gn^{1/3}\ll s  \leq c\gn^{1/2}$, the function $\msf H_\gn$ is asymptotically given by the diagonal Airy kernel $\msf A(\zeta,\zeta)$. We show that contributions from points where $\zeta(x)$ is far from $\zeta_\gn(v)$ are exponentially suppressed by combining the exponential decay in $\msf A$ with the exponential concentration of the weight factor. The dominant contribution comes from points where $|\zeta(x)-\zeta_\gn(v)| \ll \gn^{-\nu}$ for small $\nu>0$. After applying Poisson summation, the oscillatory terms are shown to be negligible via integration by parts. The leading term reduces to an integral of $\msf A(y,y)$ with respect to the rescaled variable $y = \msf c_\varphi v/(t\gn^{1/3})$, as stated in Theorem~\ref{thm:Ssubfinal}.

In the critical regime, where $-C\gn^{1/3}\leq s \leq C\gn^{1/3}$, the function $\msf H_\gn$ is instead approximated by the Painlevé XXXIV kernel $\msf K_{\ptf}(0,0|\zeta_\gn(v))$. The parameter $\zeta_\gn(v)$ remains bounded in this regime, which allows for a simpler error analysis. Similar concentration arguments show that the sum is dominated by points where $|x-z_\gn(v)| \ll \gn^{-2/3-\nu}$, and again Poisson summation combined with integration by parts eliminates the oscillatory terms. The result is Theorem~\ref{thm:Scritfinal}, expressing $\msf S$ as an integral of $\msf K_{\ptf}(0,0|y)$, and therefore in terms of the Tracy-Widom distribution itself.

The supercritical regime, where $ -c\gn^{1/2}\leq s\ll -\gn^{1/3}$, requires yet different estimates. Here $\msf H_\gn$ involves the Bessel kernel $\msf J_0$, but with an additional rescaling through the variable $\xi = \zn^2\zeta$. We decompose the summation domain into regions where $|\xi-\xi_\gn(s)|$ is very small ($\ll \gn^{-1/6}$) versus moderately small ($\ll 1$). For the very small region, the Bessel kernel is approximately constant at $1/4$, while for the moderate region, we use bounds on the oscillatory behavior plus small growth of Bessel functions to show exponential decay when combined with the weight factor. After Poisson summation, the leading term is a cubic polynomial in $s$, as shown in Theorem~\ref{thm:Ssuperfinal}.

A crucial ingredient throughout these calculations is Proposition~\ref{prop:estRfinal}, which provides precise bounds on the error term $\msf R$ in each regime. These bounds are always subleading compared to the main term $\msf S$, confirming that the asymptotics are indeed captured by the latter.

Finally, Section~\ref{sec:ProofOfMain} assembles all the pieces. We combine the asymptotics of $\msf S$ from the three regimes with the error estimates on $\msf R$ and the analysis of the explicit terms. The relation between $\msf L_\gn$ and the partition functions established in \eqref{eq:SNZNrelation} and \eqref{eq:SLrelation} 
allows us to translate results about $\msf Z_\gn^\gsig$ into statements about the multiplicative statistics. Taking appropriate derivatives with respect to the parameter $s$ yields Theorem~\ref{thm:multstat_formal}, while direct integration gives Theorem~\ref{thm:integrated_formal}. The key observation is that certain logarithmic sums telescope or are exponentially small, leaving only the terms involving $\msf S$ and $\msf R$ that we have already analyzed.

\subsection{Basic notations}\hfill 

Values $\varepsilon,\delta,\epsilon>0$ are used to represent small values, independent of other parameters unless otherwise explicitly mentioned. Their values may change from line to line, and in each occurrence they can always be taken sufficiently small. Likewise, $M,\eta>0$ are also used to represent positive constants, independent of other parameters, whose values may change from line to line. 

For us, the set of non-negative integer numbers is $\Z_{\geq 0}\deff \{0,1,2,\hdots\}$, and the set of strictly positive numbers is $\Z_{>0}\deff \{1,2,\hdots\}$. Likewise, we denote $\R_\geq \deff [0,\infty)$ and $\R_{>}\deff (0,\infty)$.

Unless otherwise stated, every branch of $z^\alpha$, as well as of $\log z$, is considered to be the principal value, that is, coming from $\arg z\in (-\pi,\pi)$, with a branch cut along the negative real axis. We reserve the notation $\sqrt{\cdot}$ for the positive square root, which is defined only for non-negative real values.

For us $D_r(p)$ denotes the open disk centered at $p\in \C$ and radius $r>0$. For a set $G\subset \C$, its topological closure is denoted by $\overline G$. The closed disk, however, is denoted by $\overline D_r(p)$ rather than the clumsier notation $\overline{D_r(p)}$. 

During the whole paper, we will be handling both scalar and $2\times 2$ matrix-valued RHPs that depend on complex variables and parameters. Matrices will always be denoted by bold letters, such as $\bm R,\bm G,\bm \Phi$ etc. In particular, we use extensively the notation
$$
\bm E_{11}\deff \begin{pmatrix} 1 & 0 \\ 0 & 0 \end{pmatrix}, \quad
\bm E_{21}\deff \begin{pmatrix} 0 & 0 \\ 1 & 0 \end{pmatrix},\quad 
\bm E_{12}\deff \begin{pmatrix} 0 & 1 \\ 0 & 0 \end{pmatrix}, \quad 
\bm E_{22}\deff \begin{pmatrix} 0 & 0 \\ 0 & 1 \end{pmatrix}.
$$
We also denote the $2\times 2$ identity matrix by $\bm I=\bm E_{11}+\bm E_{22}$, and make extensive use of the Pauli matrices
$$
\sp_1\deff \bm E_{12}+\bm E_{21},\quad \sp_2 \deff -\ii \bm E_{12}+\ii \bm E_{21},\quad \sp_3\deff \bm E_{11}-\bm E_{22},
$$
and the auxiliary matrix
$$
\bm U_0\deff \frac{1}{\sqrt{2}}\left(\bm I+\ii \sp_1\right).
$$

We usually refer to the complex variables as spectral variables, and denote them by letters $z,\zeta,w,\xi$. Derivatives with respect to the spectral variables are always denoted by $'$, whereas derivatives with respect to parameters are denoted using differentials, e.g. $\partial=\partial_y,\partial_s$ etc.

This paper will deal with matrix-valued Riemann-Hilbert problems, and as such it will involve $2\times 2$ matrix-valued functions $\bm M:G\to  \C^{2\times 2}$, with $G$ being a subset of the complex plane, typically a union of domains or a contour. We view the entries of $\bm M$ as functions in $L^p=L^p(\dd\mu)$ with $\mu$ being the planar Lebesgue measure when $G$ is a domain, or $\mu$ being the arc-length measure when $G$ is a union of contours. As such, we use the pointwise norm
$$
|\bm M(z)|\deff \sum_{i,j=1,2}|\bm M_{ij}(z)|,\quad z\in G,
$$
and the induced $L^p$ norms
$$
\|\bm M\|_\infty \deff \sum_{i,j=1,2} \|\bm M_{ij}\|_{L^p}.
$$
We will mostly use $p=\infty$ when $G$ is a union of domains, or $p=1,2,\infty$ when $G$ is a union of contours. 

Often, we will deal with matrices depending on parameters, say $\bm M=\bm M_t$ with $t$ large, and we will seek for precise asymptotics in different $L^p$ norms. For given $p_1,p_2\in [1,\infty]$, a matrix $\bm M_0$ and a positive real-valued function $f(t)$, when we say that $\bm M_t=\bm M_0+\Boh(f(t))$ in $L^{p_1}\cap L^{p_2}(G)$ we mean that
$$
\|\bm M_t-\bm M_0\|_{L^{p_j}(G)}=\Boh(f(t)) \quad \text{as }t\to \infty, \text{ for both }j=1,2.
$$

For a matrix-valued function $\bm M(z)$ which is invertible and has analytic entries, we denote
\begin{equation}\label{deff:Deltaoper}
\bm\Delta \bm M(z)=\bm \Delta_z \bm M(z)\deff \bm M(z)^{-1}\bm M'(z).
\end{equation}
Observe that if $z\mapsto \zeta=\varphi(z)$ is a conformal change of variables and $\bm N(\zeta)\deff \bm M(\varphi(z))$, then
\begin{equation}\label{eq:changeofvariablesDeltaDelta}
\bm\Delta_\zeta \bm N(\zeta) \dd \zeta=\bm\Delta_z \bm M(z)\dd z.
\end{equation}

We will also use certain canonical functions. $\ai$ denotes the Airy function, and for $x,y\in \R$ we use $\msf A(x,y)$ to denote the Airy kernel as in \eqref{deff:AiryKernel}.

A direct calculation with asymptotics of Airy functions show that
\begin{equation}\label{eq:diagAiryKernelAsymp}
\msf A(x,x)=\frac{1}{8\pi x}\ee^{-\frac{4}{3}x^{3/2}}\left(1+\Boh(x^{-3/2})\right),\quad x\to +\infty,
\end{equation}
as well as
\begin{equation}\label{eq:diagAiryKernelAsympNeg}
\msf A(x,x)=\frac{|x|^{1/2}}{\pi }\left(1+\Boh(|x|^{-3/2})\right),\quad x\to -\infty.
\end{equation}

With $\Jb_\nu$ being the Bessel function of first kind and order $\nu>-1$, the Bessel kernel is defined for $x,y>0$ as
\begin{equation}\label{deff:BesselKernel}
\msf J_\nu(x,y)\deff \frac{1}{2}\frac{\Jb_\nu(\sqrt{x})\Jb'_\nu(\sqrt y)\sqrt y - \Jb_\nu (\sqrt y)\Jb'_\nu(\sqrt x)\sqrt{x} }{x-y}, \quad x\neq y, \qquad \msf J_\nu(x,x)\deff \frac{1}{4}\left( \Jb_\nu'(\sqrt x)^2+\Jb_\nu(\sqrt{x})^2 \right).
\end{equation}

We are particularly interested in the Bessel kernel of order $\nu=0$. Even though the Bessel kernel is defined for probabilistic reasons only for $x,y>0$, the function $x\mapsto \msf J_0(x,x)$ is an entire function of $x\in \C$. For this particular case,
\begin{equation}\label{eq:asymptBesseldiag1}
\msf J_0(x,x)=\frac{1}{2\pi x^{1/2}}\left(1+\Boh(x^{-1/2})\right),\quad x\to +\infty,\qquad \msf J_0(x,x)=\frac{1}{4}-\frac{1}{16}x+\Boh(x^2), \quad x\to 0,
\end{equation}
\begin{equation}\label{eq:asymptBesseldiag2}
\msf J_0(x,x)=\frac{1}{8\pi} \frac{\ee^{2|x|^{1/2}}}{|x|}\left(1+\Boh(|x|^{-1/2})\right),\quad x\to -\infty.
\end{equation}


\section{Proof of Theorem~\ref{thm:S6VMlowertail_precise} and~\ref{thm:S6VMuppertail_precise}}\label{sec:S6Vproofs}

Assuming Theorem~\ref{thm:integrated_formal}, the overall outline of the proof of Theorems~\ref{thm:S6VMlowertail_precise} and~\ref{thm:S6VMuppertail_precise} is the following:

\begin{enumerate}
    \item Step (a): verify that the conditions for Theorem~\ref{thm:integrated_formal} are met for the Meixner weight. That is, verify Proposition~\ref{prop:Meixnerverification}.

    \item Step (b): establish upper and lower tail bounds for the tail probabilities of the stochastic six-vertex model in terms of $\mq$-Laplace transform $\E_{\mathrm{S6M}}$ in \eqref{eq:BOidentity}. 

    \item Step (c): Apply Theorem~\ref{thm:S6VMasymptotics} carefully.
\end{enumerate}

Each of the steps above is established in the subsequent subsections.

\subsection{Proof of Proposition~\ref{prop:Meixnerverification}}\label{sec:proofpropMeixnerequil}\hfill 

We first verify Assumptions~\ref{assumpt:potential_formal} followed by Assumption~\ref{assumpt:equilmeasure_formal}.

\medskip\noindent\emph{1. Verification of Assumption~\ref{assumpt:potential_formal}.}
Write
\[
\mw(k)
= \frac{\Gamma(\alpha + k)}{\Gamma(\alpha)\Gamma(k+1)}\,\beta^{\,k-1},
\qquad
\alpha = (\nu-1)\gn,\quad 0<\beta<1,
\]
and set $k=\gn x$ with $x\in \frac1\gn\Z_{>0}$. Then
\begin{equation}\label{eq:Meixner-logW}
\log\gW_\gn(x)
= -\log \gn + \log\Gamma(\alpha+\gn x) - \log\Gamma(\alpha)
- \log\Gamma(\gn x+1) + (\gn x-1)\log\beta.
\end{equation}

We use Stirling’s formula with a remainder that is analytic and uniformly bounded on sectors:
\[
\log\Gamma(z) = \Big(z-\tfrac12\Big)\log z - z + \tfrac12\log(2\pi) + r(z),
\qquad
r(z)=O(1/|z|),
\]
valid uniformly in any fixed sector containing the positive real axis.
Apply this to $\Gamma(\alpha+\gn x)$, $\Gamma(\alpha)$ and $\Gamma(\gn x+1)$.
With $\alpha=(\nu-1)\gn$ we obtain, after an elementary but slightly long algebra (all $\gn$–independent terms in $x$ can be grouped into the normalization $\gC_\gn$),
\[
\frac{1}{\gn}\big(\log\Gamma(\alpha+\gn x)-\log\Gamma(\alpha)-\log\Gamma(\gn x+1)\big)
= F(x) + O(\gn^{-1}),
\]
where
\[
F(x)
= (\nu-1)\log\frac{\nu+x-1}{\nu-1}
+ x\log(\nu+x-1)
- x\log x,
\]
and the $O(\gn^{-1})$ term is analytic in $x$ on any compact subset of a complex neighborhood of $(0,\infty)$, uniformly bounded together with its derivative. Adding the contribution $(\gn x-1)\log\beta$ from \eqref{eq:Meixner-logW}, we get
\[
\frac{1}{\gn}\log\gW_\gn(x)
= F(x) + x\log\beta + O(\gn^{-1}).
\]
Define the external field by $\gV(x)=-F(x)-x\log\beta$, which is precisely the expression in \eqref{eq:Meixner-V}. Then we may rewrite
\[
\gW_\gn(x)
= \exp\!\big(-\gn \gV(x) - 2\gC_\gn + \gE_\gn(x)\big),
\]
where $\gC_\gn$ absorbs all $x$–independent terms and $\gE_\gn(x)$ collects the $O(1)$ remainder coming from the Stirling remainders and the $O(\gn^{-1})$ term multiplied by $\gn$ (hence still $O(1)$). By standard properties of $r(z)$, the functions $\gE_\gn$ are analytic in a complex neighborhood $U$ of $(0,\infty)$ and satisfy the bounds in \eqref{eq:Econvergence}–\eqref{eq:Ebound}.

The analyticity and boundedness of $\gV$ on $[0,\infty)$ and on a complex neighborhood of $(0,\infty)$ are immediate from \eqref{eq:Meixner-V}. For the growth condition, observe that as $x\to\infty$,
\[
\gV(x)
= x\log\frac{x}{\beta(\nu+x-1)} + O(1)
= x(-\log\beta) + O(1),
\]
so $\gV(x)/x\to -\log\beta>0$. Thus Assumption~\ref{assumpt:potential_formal} is satisfied. 

\medskip\noindent\emph{2. Verification of Assumption~\ref{assumpt:equilmeasure_formal}.}
We now consider the constrained logarithmic energy problem with external field $\gV$ and upper constraint $0\leq \dd\mu/\dd x\leq 1$ on $[0,\infty)$, as in the general theory of discrete orthogonal polynomial ensembles. The equilibrium measure $\gequil$ minimizes
\[
I[\mu] = \iint \log\frac{1}{|x-y|}\,\dd\mu(x)\,\dd\mu(y)
+ \int \gV(x)\,\dd\mu(x)
\]
among probability measures $\mu$ on $[0,\infty)$ with density at most $1$. The Euler–Lagrange conditions for this constrained problem read
\begin{equation}\label{eq:EL-Meixner}
2\int \log\frac{1}{|x-y|}\,\dd\gequil(y) + \gV(x) + 2\ell
\begin{cases}
< 0, & 0<x<\ga,\\
= 0, & \ga\leq x\leq\gb,\\
> 0, & x>\gb,
\end{cases}
\end{equation}
for some constant $\ell$, together with the pointwise constraint
\[
0\leq \rho(x)\deff \frac{\dd\gequil}{\dd x}(x) \leq 1.
\]

For the potential \eqref{eq:Meixner-V}, the associated scalar singular integral equation for the Cauchy transform of $\dd\gequil$ can be solved explicitly by the standard one-band ansatz used for discrete orthogonal polynomial ensembles (see, e.g., Baik--Kr\-iecherbauer--McLaughlin--Miller~\cite{BKMMbook}). One looks for a measure $\gequil$ supported on $[0,\gb]$ with a single saturated interval $[0,\ga]$, so that
\[
\rho(x) =
\begin{cases}
1, & 0<x<\ga,\\[0.3em]
\displaystyle \frac{1}{\pi}\arccos\!\Big(
\frac{x(1-\beta) - (\nu-1)(1+\beta)}{2\sqrt{\beta(\nu-1)x}}
\Big),
& \ga < x < \gb,\\[0.7em]
0, & x>\gb,
\end{cases}
\]
for some $0<\ga<\gb$, and then determine $\ga,\gb$ by imposing:
\begin{itemize}
\item[(a)] the endpoint behavior (square-root vanishing of $\rho$ at $\gb$ and of $1-\rho$ at $\ga$);
\item[(b)] the mass constraint $\int_0^\gb \rho(x)\,\dd x =1$;
\item[(c)] the matching of the jump of $G$ with $\gV'(x)$ on $(\ga,\gb)$.
\end{itemize}
Carrying out this computation (see \cite{BKMMbook,BleherLiechtyIMRN2011} for details) yields precisely
\[
\ga = \frac{(1-\sqrt{\beta})^2}{1-\beta}(\nu-1),
\qquad
\gb = \frac{(1+\sqrt{\beta})^2}{1-\beta}(\nu-1),
\]
which are the endpoints stated in \eqref{eq:Meixner-ab}. The explicit expression of $\rho$ shows that
\begin{itemize}
\item $\supp\gequil=[0,\gb]$, so Assumption~\ref{assumpt:equilmeasure_formal}(i) holds.
\item $\rho(x)\equiv 1$ on $(0,\ga)$, i.e.\ $\gequil$ coincides there with Lebesgue measure, so the upper constraint is active exactly on $[0,\ga]$ and Assumption~\ref{assumpt:equilmeasure_formal}(ii) holds.
\item As $x\to\ga^+$, $1-\rho(x)\sim c\sqrt{x-\ga}$, and as $x\to\gb^-$, $\rho(x)\sim c'\sqrt{\gb-x}$, for some $c,c'>0$, giving the required square-root behavior at the transition points.
\end{itemize}
A direct check (or, equivalently, the analysis of the scalar Riemann–Hilbert problem in \cite{BKMMbook,BleherLiechtyIMRN2011}) shows that the Euler–Lagrange inequalities in \eqref{eq:EL-Meixner} are strict away from $[\ga,\gb]$, which is Assumption~\ref{assumpt:equilmeasure_formal}(iii).
Altogether, Assumptions~\ref{assumpt:potential_formal} and \ref{assumpt:equilmeasure_formal} hold for the Meixner ensemble with parameters
$\alpha=(\nu-1)\gn$ and $\beta = q^{-1/2}u^{-1}$, with external field and band/saturated intervals given by \eqref{eq:Meixner-V} and \eqref{eq:Meixner-ab}.

Proposition~\ref{prop:Meixnerverification} was needed so that we can apply Identity~\ref{eq:BOidentity} and recover asymptotics for the $\mq$-Laplace transform $\E_{\mathrm{S6V}}[\cdot]$ from Theorem~\ref{thm:integrated_formal}. But, we stress again, Identity~\ref{eq:BOidentity} is only valid for $1<\nu<\msf q^{1/2}u$. For the complementary regime $\msf q^{-1/2}u^{-1}<\nu\leq 1$, we have to use instead \eqref{eq:BOidentity}, and our next result ensures that the Meixner ensemble on the right-hand side of \eqref{eq:BOidentity} still satisfies our working Assumptions~\ref{assumpt:potential_formal} and \ref{assumpt:equilmeasure_formal}.

\begin{prop}[Meixner ensemble for $\msf q^{-1/2} u^{-1}<\nu\leq 1$]\label{prop:Meixnerverification_shifted}
For $\msf q^{-1/2} u^{-1}<\nu\leq 1$, $\msf q\in (0,1)$ and $u>\msf q^{-1/2}$, consider the shifted Meixner ensemble $\rm Meixner^\circ(M-1, N-M+2, \mq^{-1/2}u^{-1}) + (\mn - \mm+1)$ with parameters
$$
\widetilde{\gn} = \lfloor\nu\mn \rfloor-1, 
\qquad
\alpha = \left(\frac{1}{\nu}-1\right)\widetilde{\gn}, 
\qquad
\beta = \msf q^{-1/2}u^{-1}\in(0,1),
$$
and define the rescaled weight
$$
\gW_{\widetilde{\gn}}(x) \deff \frac{1}{\widetilde{\gn}}\mw(\widetilde{\gn} x),
\qquad 
x\in \tfrac{1}{\widetilde{\gn}}\Z_{>0}.
$$
Then Assumptions~\ref{assumpt:potential_formal} and \ref{assumpt:equilmeasure_formal} hold, as a particle system with particles starting to the right of $\nu^{-1}-1$, and with:
\begin{enumerate}[(i)]
\item External field
\begin{equation}\label{eq:Meixner-V-shifted}
\gV(x) = \left(x - \tfrac{1}{\nu}+1\right)\log\frac{x - \nu^{-1}+1}{\beta x}
+ \left(\tfrac{1}{\nu}-1\right)\log\frac{\nu^{-1}-1}{x}, 
\qquad x>\tfrac{1}{\nu}-1.
\end{equation}
\item Saturated region $\left[\nu^{-1}-1,\, \msf{a}\right]$ and band $[\msf{a},\, \msf{b}]$ given by
\begin{equation}\label{eq:Meixner-ab-shifted}
\msf{a} = \frac{(1-\sqrt{\beta})}{(1-\beta)}(\nu^{-1}-1),
\qquad
\msf{b} = \frac{(1+\sqrt{\beta})}{(1-\beta)}(\nu^{-1}-1).
\end{equation}
\end{enumerate}
In particular, $\gV$ satisfies Assumption~\ref{assumpt:potential_formal}, and the associated constrained equilibrium measure $\gequil$ satisfies Assumption~\ref{assumpt:equilmeasure_formal} with $\msf c_\gV=\mc^{-1}$ and $\mc$ as in \eqref{eq:scalingconstants}.
\end{prop}

The proof of Proposition~\ref{prop:Meixnerverification_shifted} follows the same arguments as the proof of Proposition~\ref{prop:Meixnerverification}, with the appropriate shifts and parameter substitutions. We omit the details for brevity. Furthermore, applying Theorem~\ref{thm:multstat_formal} with Proposition~\ref{prop:Meixnerverification_shifted} in place of Proposition~\ref{prop:Meixnerverification} yields the analogue of Theorem~\ref{thm:S6VMasymptotics} for the regime $\msf q^{-1/2} u^{-1}<\nu\leq 1$.

\subsection{Comparison lemmas for the $\mq$--Laplace transform}\hfill

The goal of this subsection is to establish bounds between the tail probabilities of the S6V model and the $\msf q$-Laplace transform $\E_{\mathrm{S6V}}$ of it that appears in the asymptotic result of Theorem~\ref{thm:S6VMasymptotics}. We use the following shorthand notation below
\begin{equation}\label{deff:InJNqLapl}
\mathcal{I}_n\deff\frac{1}{1+\zeta \mq^{\mathfrak{h}(\nu \mn, \mn)+n}},
\qquad
\mathcal{J}_n\deff \log\bigl(1+\zeta \mq^{\mathfrak{h}(\nu \mn, \mn)+n}\bigr).
\end{equation}

\begin{lemma}[Lower tail comparison via the $\mq$--Laplace transform]\label{lem:UpTail}
Fix $\nu>1$ in the liquid region, $\msf{h}>0$, $\mq\in(0,1)$, and set $\zeta= \mq^{-\aeq\mn-\msf{h}\mc\mn^{1/3}}$ as in \eqref{eq:scalingzetasS6V}.
Then for every $\varepsilon\in(0,1)$ and all $N$ large enough (depending on $\varepsilon$ and $\mq$) one has
\[
\mathbb{P}\bigl(\Hf_\mn(\mnu)\ge (1+\varepsilon)\msf{h}\bigr)
\;\le\;
\exp\Bigl((1-\mq)^{-1}\mq^{\varepsilon \msf{h} \mn^{1/3}}\Bigr)\,
\mathbb{E}_{\mathrm{S6V}}\Bigl[\prod_{n\ge1}\mathcal{I}_n\bigr]
\]
and
\[
\mathbb{P}\bigl(\Hf_\mn(\mnu)\ge (1-\varepsilon)\msf{h}\bigr)
\;\ge\;
\mathbb{E}_{\mathrm{S6V}}\Bigl[\prod_{n\ge1}\mathcal{I}_n\Bigr]
\;-\;
\exp\Bigl(-\tfrac{1}{2}\log(\mq^{-1})(\varepsilon \msf{h}\mc)^2 N^{2/3}\Bigr).
\]
\end{lemma}

\begin{proof}
For the upper bound, define the random variable
\[
Y^{(1)} \deff \mathbf{1}\bigl(\Hf_\mn(\mnu)\ge (1+\varepsilon)\msf{h}\bigr)\prod_{n\ge1}\mathcal{I}_n
 = \mathbf{1}\bigl(\Hf_\mn(\mnu)\ge (1+\varepsilon)\msf{h}\bigr)\exp\Bigl(-\sum_{n\ge1} \mathcal{J}_n\Bigr).
\]
Since on the event $\{\Hf_\mn(\mnu)\ge (1+\varepsilon)\msf{h}\}$, we have $
\mathfrak{h}(\nu \mn,\mn) \ge  \aeq\mn + (1+\varepsilon)\msf{h}\mc \mn^{1/3}$,
for each $n\ge1$ we may write
\[
\zeta \mq^{\mathfrak{h}(\nu \mn,\mn)+n}
 = \mq^{\mathfrak{h}(\nu \mn,\mn)  - \aeq\mn - \msf{h}\mc \mn^{1/3}+n}
 \le \mq^{\varepsilon \msf{h}\mc \mn^{1/3}+n}.
\]
This quantity is in $(0,1)$, and by the inequality $\log(1+z)\le z$ for $z\in(0,1)$ we obtain
\[
\mathcal{J}_n = \log\bigl(1+\zeta \mq^{\mathfrak{h}(\nu \mn,\mn)+n}\bigr)
\;\le\;
\zeta \mq^{\mathfrak{h}(\nu \mn,\mn)+n}
\;\le\;
\mq^{\varepsilon \msf{h}\mc \mn^{1/3}+n}.
\]
Summing over $n$ gives
\[
\sum_{n\ge1}\mathcal{J}_n \;\le\; \mq^{\varepsilon \msf{h}\mc \mn^{1/3}}\sum_{n\ge1} \mq^n
  = \frac{1}{(1-\mq)}\mq^{\varepsilon \msf{h}\mc \mn^{1/3}}.
\]
Substituting this into the definition of $Y^{(1)}$ shows
\[
\exp\Bigl(-(1-\mq)^{-1}\mq^{\varepsilon \msf{h}\mc \mn^{1/3}}\Bigr)\mathbf{1}\bigl(\Hf_\mn(\mnu)\ge (1+\varepsilon)\msf{h}\bigr)\leq Y^{(1)}.
\]
Taking expectation, we obtain
\[
\mathbb{P}\bigl(\Hf_\mn(\mnu)\ge (1+\varepsilon)\msf{h}\bigr)
\le \exp\Bigl(\tfrac{1}{(1-\mq)}\mq^{\varepsilon \msf{h}\mc \mn^{1/3}}\Bigr)\,
\mathbb{E}_{\mathrm{S6V}}\Bigl[Y^{(1)}\Bigr]\le 
\exp\Bigl(\tfrac{1}{(1-\mq)}\mq^{\varepsilon \msf{h}\mc \mn^{1/3}}\Bigr)\,
\mathbb{E}_{\mathrm{S6V}}\Bigl[\prod_{n\ge1}\mathcal{I}_n\Bigr],
\]
which is the first inequality.

For the lower bound, define
\[
Y^{(2)} \deff \mathbf{1}\bigl(\Hf_\mn(\mnu)< (1-\varepsilon)\msf{h}\bigr)\prod_{n\ge1}\mathcal{I}_n,
\qquad
\widetilde{Y}^{(2)} \deff \prod_{n\ge1}\mathcal{I}_n - Y^{(2)}
 = \mathbf{1}\bigl(\Hf_\mn(\mnu)\geq  (1-\varepsilon)\msf{h}\bigr)\prod_{n\ge1}\mathcal{I}_n.
\]
Since $0<\mathcal{I}_n\le1$, we have
\begin{equation}\label{ineqhnY2tildeY2}
\mathbb{P}\bigl(\Hf_\mn(\mnu)\geq  (1-\varepsilon)\msf{h}\bigr)
\ge \mathbb{E}_{\mathrm{S6V}}[\widetilde{Y}^{(2)}]
 = \mathbb{E}_{\mathrm{S6V}}\Bigl[\prod_{n\ge1}\mathcal{I}_n\Bigr]
   - \mathbb{E}_{\mathrm{S6V}}[Y^{(2)}].
\end{equation}
It remains to bound $\mathbb{E}[Y^{(2)}]$ from above. On the event $\{\Hf_\mn(\mnu)<  (1-\varepsilon)\msf{h}\}$ we have
$$
\mathfrak{h}(\nu \mn,\mn) \le  \aeq\mn + (1-\varepsilon)\msf{h}\mc \mn^{1/3},
$$
and we now estimate the product $\prod_{n\ge1}\mathcal{I}_n$ by splitting the sum in $\sum_n\mathcal{J}_n$ at the index $M \deff \lfloor \varepsilon \msf{h}\mc \mn^{1/3} \rfloor$. For $n\leq M$ we use the the basic inequality $\log(1+z)=\log z+\log(1+z^{-1})\geq \log z$, valid for $z\in (0,1)$, and obtain $\mcal J_n\ge ( \varepsilon \msf{h}\mc \mn^{1/3}-n)\log(\mq^{-1})$. When $n>M$, we use $\log(1+z)\geq \frac{1}{2}z$ for $z\geq 0$, obtaining now $\mcal J_n\ge \frac{1}{2}\mq^{n- \varepsilon \msf{h}\mc \mn^{1/3}}$. In total, combining the two estimates we obtain a Gaussian-type lower bound
\[
\sum_{n\ge1}\mathcal{J}_n
\geq \frac{M(M+1)}{2}\log\mq^{-1}+\frac{\mq}{1-\mq}\geq 
\frac{1}{2}\log(\mq^{-1})(\varepsilon  \msf{h}\mc)^2 \mn^{2/3}
\]
for $\gn$ large. This yields 
\[
\mathbb{E}_{\mathrm{S6V}}[Y^{(2)}]
\le
\exp\Bigl(-\tfrac{1}{2}\log(\mq^{-1})(\varepsilon  \msf{h}\mc)^2 \mn^{2/3}\Bigr).
\]
Plugging this inequality into \eqref{ineqhnY2tildeY2} gives the desired lower bound.
\end{proof}

\begin{lemma}[Upper tail comparison via the $\mq$--Laplace transform]\label{lem:LowTail}
Fix $\nu>1$ in the liquid region, $\msf{h}>0$, $\mq\in(0,1)$ and $\zeta= \mq^{-\aeq\mn+\msf{h}\mc\mn^{1/3}}$ as in \eqref{eq:scalingzetasS6V}. 
Then for every $\varepsilon\in(0,1)$ and all $N$ large enough (depending on $\varepsilon$ and $\mq$) one has
\[
\mathbb{P}\bigl(\Hf_\mn(\mnu)\le -(1-\varepsilon)\msf{h}\bigr)
\;\ge\; 1-\mathbb{E}_{\mathrm{S6V}}\Bigl[\prod_{n\ge1}\mathcal{I}_n\Bigr]
\;+\;
\exp\Bigl(-\tfrac{1}{2}\log(\mq^{-1})(\varepsilon \msf{h}\mc)^2 N^{2/3}\Bigr)
\]
and
\[
\mathbb{P}\bigl(\Hf_\mn(\mnu)\le -(1+\varepsilon)\msf{h}\bigr)
\;\le\;
 1-\exp\Bigl((1-\mq)^{-1} \mq^{\varepsilon \msf{h} \mn^{1/3}}\Bigr)\,
\mathbb{E}_{\mathrm{S6V}}\Bigl[\prod_{n\ge1}\mathcal{I}_n\bigr]
\]
\end{lemma}

\begin{proof}
    The argument is completely analogous to that of Lemma~\ref{lem:UpTail}: one repeats the same estimates with the substitutions $\msf{h}\mapsto -\msf{h}$ and $\varepsilon\mapsto -\varepsilon$, which reverse the direction of the inequalities in the bounds for the factors $\mathcal{J}_n$.
Since no new ideas are required, we omit the details for brevity. 
\end{proof}

\subsection{Proof of Theorems~\ref{thm:S6VMlowertail_precise} and \ref{thm:S6VMuppertail_precise}}\hfill 

Recall the notation introduced in \eqref{deff:InJNqLapl}, the rescaled height function $\Hf_\gn$ introduced in \eqref{eq:rescaledheight}, and from now on consider $\zeta$ and $s$ as in \eqref{eq:scalingzetasS6V}. In addition, set
$$
F_\mn\deff \log\E_{\rm S6V}\!\Bigl[\prod_{n\ge1}\mathcal{I}_n\Bigr].
$$

The proofs of the upper and lower tails follow the same strategy:  
(1) rewrite the tail probabilities in terms of the $\mq$--Laplace transform $\E_{\rm S6V}\!\left[\prod_{n\ge1}\mathcal{I}_n\right]=\ee^{F_\gn}$ using Lemmas~\ref{lem:UpTail} and \ref{lem:LowTail};  
(2) apply the asymptotic formulas of Theorem~\ref{thm:S6VMasymptotics} for $F_\gn$ in the appropriate regime of the parameter \(s\);  
(3) convert these asymptotics into the required large-deviation exponents. 

Assuming Theorem~\ref{thm:integrated_formal}, we now complete the proof of Theorems~\ref{thm:S6VMlowertail_precise} and \ref{thm:S6VMuppertail_precise}.

\emph{Step 1. From tail probabilities to the $\mq$--Laplace transform.}  
Lemma~\ref{lem:UpTail} gives, for any fixed $\varepsilon\in(0,1)$,
\begin{align*}
& \mathbb{P}(\Hf_\mn(\nu)\ge(1+\varepsilon)\msf{h})
\;\le\;
\exp\bigl(\mq^{\varepsilon\msf{h}\mn^{1/3}}/(1-\mq)\bigr)\,\ee^{F_\mn},\\
& \mathbb{P}(\Hf_\mn(\nu)\ge(1-\varepsilon)\msf{h})
\;\ge\;
\ee^{F_\mn}-\exp\bigl(-\tfrac12\log(\mq^{-1})(\varepsilon\msf{h}\mc)^2\mn^{2/3}\bigr).
\end{align*}
An entirely analogous pair of inequalities for the upper tail is supplied by Lemma~\ref{lem:LowTail}.  
Thus the problem reduces to determining the large-\(\mn\) asymptotics of \(F_\mn\).

\smallskip
\emph{Step 2. Asymptotics of \(F_\mn\) from Theorem~\ref{thm:S6VMasymptotics}.}
There are three regimes, determined by the size of the parameter \(s\).

\smallskip
\textbf{Supercritical regime} (\(s< -s_0\mn^{1/3}\), equivalently \(\msf{h}>0\) for the upper tail).  
By Theorem~\ref{thm:S6VMasymptotics}(iii),
\[
F_\mn
=
-\frac{\msf c_\gV^3}{12t^3\mn}|s|^3+\Boh\left(\frac{s^4}{\mn^{3/2}}\right).
\]
Using that \(s=-(\msf{h}\mc\mn^{1/3}+1)\log(1/\mq)\) and that \(\mc=\msf c_\gV^{-1}\), the cubic term simplifies to  
\(\,F_\mn=-(1+\boh(1))\,\msf{h}^3/12\).  
Since the extra prefactors in Lemma~\ref{lem:UpTail} are exponentially negligible on the \(\msf{h}^3\)-scale, we obtain
\[
\exp\!\left(-(1+\varepsilon)\frac{\msf{h}^3}{12}\right)
\;\le\;
\mathbb{P}(\Hf_\mn(\nu)\ge \msf{h})
\;\le\;
\exp\!\left(-(1-\varepsilon)\frac{\msf{h}^3}{12}\right),
\]
uniformly for $\msf{h}_0\le\msf{h}\le\msf{h}_0^{-1}\mn^{1/6}$.  
This proves the upper-tail bound of Theorem~\ref{thm:S6VMuppertail_precise}.

\smallskip
\textbf{Subcritical regime} (\(s> s_0\mn^{1/3}\), equivalently \(\msf{h}>0\) for the lower tail).  
Theorem~\ref{thm:S6VMasymptotics}(i) gives
\[
F_\mn
=
-\left(1+\Boh\left(\ee^{-\eta s/\mn^{1/3}}+\mn^{-\nu}\right)\right)
\int_{\msf{c}_\gV s/(t\mn^{1/3})}^\infty \msf{A}(y,y)\,\dd y
+\Boh(\ee^{-\eta\mn^{1/4}}).
\]
Since the lower limit of integration is proportional to \(\msf{h}\) and \(\msf{A}(y,y)\) is the Airy kernel diagonal, the integral equals \((4/3)\,\msf{h}^{3/2}(1+o(1))\).  
Thus $F_\mn=-(4/3)\msf{h}^{3/2}+\boh(1)$.  
Insert this into Lemma~\ref{lem:LowTail} (whose prefactors are negligible on the \(\msf{h}^{3/2}\)-scale).  
We obtain
\[
\exp\!\left(-(1+\varepsilon)\frac{4}{3}\msf{h}^{3/2}\right)
\;\le\;
\mathbb{P}\bigl(\Hf_\mn(\nu)\le-\msf{h}\bigr)
\;\le\;
\exp\!\left(-(1-\varepsilon)\frac{4}{3}\msf{h}^{3/2}\right),
\]
establishing Theorem~\ref{thm:S6VMlowertail_precise}.

\smallskip
\textbf{Critical regime} (\(|s|\le s_0\mn^{1/3}\)).  
Here Theorem~\ref{thm:S6VMasymptotics}(ii) expresses \(F_\mn\) in terms of an integral of the Painlevé~P$_{34}$ kernel.  
The width of this region is \(\Boh(\mn^{1/3})\), and integrating the kernel on this scale contributes only \(\boh(1)\) to \(F_\mn\).  
Thus the subcritical and supercritical asymptotics match smoothly across the critical window and do not affect the dominant large-deviation exponents.

\smallskip
Combining these three steps completes the proof of both upper and lower tail theorems.

\section{A general deformation formula for partition functions of OP ensembles}\label{sec:defformulaproof}

In this section, we prove a more general version of the deformation formula \eqref{eq:deformationformula}. This section is independent of other sections and may be of independent interest, and we start setting up appropriate independent notation.

Start with a Borel measure $\mu$ on $\R$ with all moments finite. We associate to it the orthogonal polynomial $P_k=P_k(\cdot\mid \mu)$ to be the monic polynomial of degree $k$ that satisfies
$$
\int x^j P_k(x)\dd\mu(x)=0,\quad j=0,\hdots,k-1.
$$ 
The corresponding norming constant $\gamma_k=\gamma_k(\mu)>0$ is then defined via the identity
$$
\frac{1}{\gamma_k^2}=\int (P_k(x))^2 \dd\mu(x).
$$
The associated Christoffel-Darboux kernel is
\begin{equation}\label{eq:corrkernelOPs}
K_n(x,y)=K_n(x,y\mid \mu)\deff \sum_{k=0}^{n-1}\pi_k(x)\pi_k(y)=\sum_{k=0}^{n-1}\gamma_k^2P_k(x)P_k(y).
\end{equation}
From the orthogonality for the $P_k$'s, it follows that this kernel satisfies the identity
\begin{equation}\label{eq:reproducingpropkernelnorm}
\int K_n(x,x)\dd\mu(x)=n.
\end{equation}
The partition function of the associated orthogonal polynomial ensemble is
\begin{equation}\label{eq:pttfction}
Z_n=Z_n(\mu)\deff \int_{\R^n} \prod_{1\leq j<k\leq n}(x_k-x_j)^2 \dd\mu(x_1)\cdots \dd\mu(x_n),
\end{equation}
which is a Hankel determinant, and relates to the norming constants via the formula
\begin{equation}\label{eq:identityPartFctionNormCtt}
Z_n=n!\prod_{k=0}^{n-1}\gamma_k^{-2}.
\end{equation}
In particular, for every $k$ it is valid
\begin{equation}\label{eq:relationgammakZk}
\gamma_k^2=\frac{Z_{k+1}}{Z_k}.
\end{equation}
Also, $P_n$ may be computed through Heine's formula
\begin{equation}\label{eq:HeineFormula}
P_n(z)=\frac{1}{n! Z_n}\int_{\R^n} \prod_{j=1}^n (z-x_j)\prod_{1\leq i<j\leq n}(x_j-x_i)^2 \dd\mu(x_1)\cdots \dd\mu(x_n).
\end{equation}
Writing $P_n(x)=x^n+a^{(n)}_{1}x^{n-1}+\cdots + a^{(n)}_{n-1}x+ a^{(n)}_n$, we obtain the expression for $a_k^{(n)}=a_k^{(n)}(\mu)$,
$$
a_k^{(n)}(\mu)=\frac{(-1)^k}{n! Z_n(\mu)}\int_{\R^n} e_k(x_1,\hdots,x_n)\prod_{1\leq i<j\leq n}(x_i-x_j)^2 \dd\mu(x_1)\cdots \dd\mu(x_n),
$$
where $e_k(x_1,\hdots,x_n)$ is the $k$-th elementary symmetric polynomial in $n$ variables. From the straightforward inequality
$$
|e_k(x_1,\hdots,x_n)|=\left|\sum_{1\leq i_1<\cdots < i_k\leq n} x_{i_1}\cdots x_{i_k} \right| \leq \frac{1}{k!}\binom{n}{k}\prod_{j=1}^n(1+|x_j|)
$$
we obtain the bound
\begin{equation}\label{eq:trivialboundcoeffsOPs}
|a_k^{(n)}(\mu)|\leq \frac{1}{n! Z_n(\mu)} \frac{1}{k!}\binom{n}{k} Z_n((1+|x|)\dd\mu(x)).
\end{equation}

We now talk about deformations of measures. Fix an interval $[s_0,s_*]$ with $s_0<s_*\leq +\infty$. When $s_*=+\infty$, we convention that $[s_0,s_*]=[s_0,+\infty)$, and that all the evaluations at $s=s_*$ mean taking limits $s\to +\infty$.

We fix a reference Borel measure $\mu_0$ with all moments finite, and for a function $w=w(x\mid s)$, we consider a deformed measure
$$
\dd\mu_s(x)\deff w(x\mid s)\dd\mu_0(x),\quad s\in [s_0,s_*].
$$
We assume that $x\mapsto x^k w(x\mid s)$ is a non-negative $\mu_0$-integrable function for every $k\geq 0$, so that $\mu_s$ is a measure with finite moments for every fixed $s$. In particular, we may talk about
$$
Z_n(s)=Z_n(\mu_s),\quad \gamma_n(s)=\gamma_n(\mu_s),\quad P_n(x\mid s)=P_n(x\mid \mu_s),\quad K_n(x,y\mid s)=K_n(x,y\mid \mu_s)
$$
as the partition function, norming constants, monic orthogonal polynomials and Christoffel-Darboux kernels, respectively, for the measure $\dd\mu_s(x)$.

We prove
\begin{prop}\label{Prop:deformationPartFctionInt}
For $\partial=\partial_s$ being the derivative with respect to $s$, assume that there exist non-negative functions $f=f(x), F=F(x)$ and a constant $\delta>0$, for which the inequalities
\begin{equation}\label{eq:bounddefweight}
\delta\leq w(x\mid s)\leq f(x),\quad \text{and}\quad |\partial w(x\mid s)|\leq F(x)
\end{equation}
hold true for $\mu_0$-a.e. $x\in \R$ and Lebesgue a.e. $s \in [s_0,s_*]$, and satisfying the condition that 
$$|x|^k f(x), |x|^kF(x)\in L^1(\dd \mu_0, \R),$$ 
for every $k\geq 0$. Then the identity
$$
\log \frac{Z_n(s)}{Z_n(s_*)}=-\int_{s}^{s_*} \int_{-\infty}^\infty K_n(x,x\mid u) (\partial w)(x\mid u) \dd\mu_0(x)\dd u
$$
holds true for every $s\in [s_0,s_*]$.
\end{prop}

\begin{remark}\label{rm:deffformula}
For $\mu_0$ being a specific weighted version of the Lebesgue measure on $\R$ and $w(x\mid s)$ a particular choice of weight,  Proposition~\ref{Prop:deformationPartFctionInt} coincides with \cite[Proposition~9.1]{GhosalSilva2023}. For IIKS integrable kernels, which include the Christoffel-Darboux kernel for OPs for continuous weights (but not for discrete ones as needed in the present work), a similar result was obtained by Claeys and Glesner in \cite{ClaeysGlesner2023}, see the last displayed equation in page 2216 therein.
\end{remark}

Proposition~\ref{prop:deformationformula} follows from Proposition~\ref{Prop:deformationPartFctionInt}, setting $\mu_0(x)$ to be the counting measure of $\Z_{>0}$ weighted by $\gw$ and $w(x\mid s)=(1+\ee^{-t(x-\ga\gn)-s})$, so that $\mu_s$ is the counting measure on $\Z_{>0}$ weighted by $\gwd(x)=(1+\ee^{-t(x-\ga\gn)-s})\gw(x)$. Indeed, in this case $\mu_0$ has all its moments finite because this is an assumption we are working with for $\gw$, the inequalities
$$
1 \leq w(x\mid s) \leq 1+\ee^{t\ga\gn-s}
$$
are immediate, and the term $1+\ee^{t\ga\gn-s}$ is bounded from above as a function of $s$ on any interval of the form $[s_0,+\infty)$, and finally
$$
\left|\partial_s w(x\mid s)\right|\leq \ee^{t\ga\gn-s},
$$
so \eqref{eq:bounddefweight} is satisfied with $f,F$ being constant functions.

\begin{proof}
It suffices to prove the result for $s_*$ finite, the general result then follows taking $s_*\to +\infty$.

First off, we use the orthonormality to obtain the identity
$$
\int \partial(\pi_k(x\mid s)^2)\dd\mu_s(x)=2\int_{-\infty}^\infty \pi_k(x\mid s)(\partial\gamma_k(s) x^k+(\text{pol degree}\leq k-1))\dd\mu_s(x)=2\frac{\partial\gamma_k(s)}{\gamma_k(s)}
$$
We then take logarithms in \eqref{eq:identityPartFctionNormCtt} and differentiate, obtaining
$$
\partial \log Z_n(s)=-2\sum_{k=0}^{n-1}\frac{\partial\gamma_k(s)}{\gamma_k(s)}=-\int_{-\infty}^\infty \partial K_n(x,x\mid s) \dd\mu_s(x).
$$
To our knowledge, the steps we performed so far were first observed by Krasovsky \cite[Equation~(14)]{Krasovsky07}. Next, we integrate from the given value $s$ to $s_*$, obtaining
$$
\log \frac{Z_n(s)}{Z_n(s_*)}=-\int_{s}^{s_*}\int_{-\infty}^\infty K_n(x,x\mid u)\partial w(x\mid u)\dd\mu_0(x)\dd u\\ +\int_{s}^{s_*}\int_{-\infty}^\infty \partial (K_n(x,x\mid u)w(x\mid u))\dd\mu_0(x)\dd u,
$$
where we also used the trivial identity $(\partial K_n)w=\partial(K_n w)-K_n\partial w$.

The final step is to show that the last integral is zero. Assuming that we can apply Fubini's Theorem, 
\begin{equation}\label{eq:applFubini}
\int_{s}^{s_*}\int_{-\infty}^\infty \partial (K_n(x,x\mid u)w(x\mid u))\dd\mu_0(x)\dd u=\int_{-\infty}^\infty\left( K_n(x,x\mid s_*)w(x\mid s_*)-K_n(x,x\mid s)w(x\mid s)\right)\dd\mu_0(x),
\end{equation}
and using \eqref{eq:reproducingpropkernelnorm} we see that the right-hand side above vanishes.

To conclude, we justify why Fubini's Theorem could be applied. To that end, observe first that the identities
$$
\delta^n Z_n(\dd \mu_0)\leq Z_n(s) \leq Z_n(f^n \dd \mu_0)\quad \text{and}\quad |\partial Z_n(s)|\leq Z_n(F^n\dd \mu_0)
$$
are a mere consequence of \eqref{eq:pttfction} and \eqref{eq:bounddefweight}. Thanks to \eqref{eq:relationgammakZk}, we thus see that both $\gamma_k(s)^2$ and $\partial_s(\gamma_k(s)^2)$ are bounded by a uniform constant for $s\in [s_0,s_*]$.

Next, from \eqref{eq:trivialboundcoeffsOPs}, we also see that $P_k(x\mid s)$ is bounded by a factor of the form $M(1+|x|^k)$, for some $M>0$. In an analogous way, proceeding as we did to obtain \eqref{eq:trivialboundcoeffsOPs}, we see that $\partial P_k(x\mid s)$ is bounded by $M' (1+|x|^{k-1})$ for some $M'>0$.

All in all, combining such conclusions with \eqref{eq:corrkernelOPs}, and using again \eqref{eq:bounddefweight}, we obtain that $\partial(K_n w)=\partial K_n w + K_n\partial w$ is bounded by a factor of the form $M(1+|x|^m)f(x)+M'(1+|x|^{m'})F(x)$, for some new values $M,M',m,m'>0$. By assumption and the fact that $s_*$ is finite, these factors are $\dd\mu_0(x)\times \dd s$ integrable on $\R\times [s_0,s_*]$, and therefore the application of Fubini's Theorem in \eqref{eq:applFubini} is justified.
 
\end{proof}

\section{A priori estimates for multiplicative statistics}\label{sec:aprioriest}

Theorem~\ref{thm:multstat_formal} deals with asymptotics for $\msf L_\gn(s)$ when $|s|\leq \delta \gn^{1/2}$. However, we will also need to obtain rough estimates when $s\gg \gn^{1/2}$, and we pursue them in this section.

Set
$$
\kappa(S)\deff \floor{\gn \left(\ga -\frac{S}{2t\gn}\right)},\quad S>0.
$$
The mentioned needed estimate is summarized with the following result.

\begin{theorem}\label{thm:aprioriest}
For any $\delta>0$ sufficiently small, there exists $\eta>0$ such that the estimate
$$
\msf S_N(S)=\left[\prod_{x=1}^{\kappa(S)}\left(1+\ee^{-t(x-\gn(\ga-\frac{S}{\gn t}))}\right)\right] \left(1+\Boh\left( \ee^{-\eta S^{3/2}\gn^{-1/2}} \right)\right),\quad \text{as}\quad \gn\to +\infty,
$$
holds true uniformly for $\delta\gn^{1/2}\leq S\leq \delta \gn$.
\end{theorem}

The proof of Theorem~\ref{thm:aprioriest} is split into different lemmas and will be concluded at the end of the present section. From now on we always assume that $\delta>0$ is fixed in such a way that
$$
0<\kappa(S)<\ga,\quad \text{for every } S\in [\delta\gn^{1/2},\delta \gn].
$$
Such $\delta>0$ will be made sufficiently smaller whenever needed.

For $\mcal X_\gn$ the discrete orthogonal polynomial ensemble for the weight $\gw$ in $\Z_{>0}$, let
$$
\#(S)\deff |\mcal X_\gn\cap (0,\kappa(S))|
$$
be the counting function of the interval $(0,\kappa(S))$. The random variable $\#$ will be used in the proofs that follow.

\begin{lemma}\label{lem:upperboundLNsat}
For any $S\in [\delta\gn^{1/2},\delta\gn]$, the inequality
$$
\msf S_\gn(S)\leq \prod_{x=1}^{\kappa(S)}\left(1+\ee^{-t(x-\gn(\ga-\frac{S}{\gn t}))}\right) \exp\left( \ee^{-\frac{S}{2}+\frac{1}{t}}\frac{\ee^{-t}}{1-\ee^{-t}} \right)
$$
holds true.
\end{lemma}
\begin{proof}
    Having in mind that $\#\leq \kappa(S)$, we expand in conditional expectation,
    $$
    \msf S_\gn(S)=\sum_{k=0}^{\kappa(S)}\E\left[ \prod_{x\in \mcal X_\gn}\left(1+\ee^{-t(x-\gn(\ga-\frac{S}{\gn t}))}\right)\Big| \#(S)=k \right] \P\left(\#(S)=k\right).
    $$
    Split the product inside the expectation into two products, one over $\mcal X_\gn\cap (0,\kappa(S)]$, and the second one over $\mcal X_\gn\cap (\kappa(S),+\infty)$.
    The particles $x\in \mcal X_\gn$ take values on the integers, and the exponential inside the expectation is a decreasing function of $x$. Therefore, once $\#=k$, the largest value of the product over $\mcal X_\gn\cap (0,\kappa(S)]$ is attained when all the $k$ particles in this set are at the positions $1,2,\hdots, k$. The remaining $N-k$ particles have to be in the interval $(\kappa(S),\infty)$, and once again the corresponding product is maximized when they are positioned at the points $\kappa(S)+1,\hdots, \kappa(S)+N-k$. Hence,
    $$
    \msf S_\gn(S)\leq \sum_{k=0}^{\kappa(S)} \prod_{x=1}^k \left(1+\ee^{-t(x-\gn(\ga-\frac{S}{\gn t}))}\right)\prod_{x=\kappa(S)+1}^{\kappa(S)+N-k}\left(1+\ee^{-t(x-\gn(\ga-\frac{S}{\gn t}))}\right)\P\left( \#(S)=k \right).
    $$
    The first product above increases when we increase $k$, so each of these products is bounded by the product with $k=\kappa(S)$. Next, we rewrite
    $$
    \prod_{x=\kappa(S)+1}^{\kappa(S)+N-k}\left(1+\ee^{-t(x-\gn(\ga-\frac{S}{\gn t}))}\right)=\prod_{x=1}^{N-k}\left(1+\ee^{-t(x+\kappa(S)-\gn(\ga-\frac{S}{\gn t}))}\right)
    $$
    From the basic inequality $1+y\leq e^y$ for $y\geq 0$, we learn that
    $$
    \prod_{x=1}^{N-k}\left(1+\ee^{-t(x+\kappa(S)-\gn(\ga-\frac{S}{\gn t}))}\right)\leq \prod_{x=1}^{\infty}\left(1+\ee^{-t(x+\kappa(S)-\gn(\ga-\frac{S}{\gn t}))}\right)
    \leq \exp\left( \ee^{-t\left[\kappa(S)-\gn\left(a-\frac{S}{\gn t}\right)\right]} \sum_{x=1}^\infty \ee^{-tx} \right).
    $$
    Now,
    $$
    \kappa(S)-\gn\left(\ga-\frac{S}{\gn t}\right)=\gn\left(\ga-\frac{S}{2\gn t}\right)-\gn\left(\ga-\frac{S}{\gn t}\right)+\kappa(S)-\gn\left(\ga-\frac{S}{2\gn t}\right)\geq \frac{S}{2t}-1.
    $$
    Hence, we just obtained the inequality
    $$
    \msf S_N(S)\leq \prod_{x=1}^{\kappa(S)}\left(1+\ee^{-t(x-\gn(\ga-\frac{S}{\gn t}))}\right) \exp\left( \ee^{-\frac{S}{2}+\frac{1}{t}}\frac{\ee^{-t}}{1-\ee^{-t}} \right)\sum_{k=0}^{\kappa(S)}\P\left(\#(S)=k\right),
    $$
    and the remaining sum is identically $1$, concluding the proof.
\end{proof}

The next result gives a lower bound for $\msf S_\gn$.

\begin{lemma}\label{lem:lowerboundLNsat}
    For any $S\in [\delta\gn^{1/2},\delta\gn]$, the inequality
    $$
    \msf S_\gn(S)\geq \prod_{x=1}^{\kappa(S)}\left(1+\ee^{-t(x-\gn(\ga-\frac{S}{\gn t}))}\right) \P(\#(S)=\kappa(S))
    $$
    holds true.
\end{lemma}
\begin{proof}
    Applying again conditional expectation, we see that
    $$
    \msf S_\gn(S)\geq \E\left[ \prod_{x\in \mcal X_\gn}\left(1+\ee^{-t(x-\gn(\ga-\frac{S}{\gn t}))}\right)\Big| \#(S)=\kappa(S) \right] \P\left(\#(S)=\kappa(S)\right).
    $$
    Using that all the terms in the product are greater than $1$, we further bound
    $$
    \msf S_\gn(S)\geq \E\left[ \prod_{x\in \mcal X_\gn\cap (0,\kappa(S))}\left(1+\ee^{-t(x-\gn(\ga-\frac{S}{\gn t}))}\right)\Big| \#(S)=\kappa(S) \right] \P\left(\#(S)=\kappa(S)\right].
    $$
    There are exactly $\kappa(S)$ integer numbers in the interval $(0,\kappa(S)]$, so under the conditional $\#(S)=\kappa(S)$ the position of the particles is now deterministic, and we obtain that
    $$
    \E\left[ \prod_{x\in \mcal X_\gn\cap (0,\kappa(S))}\left(1+\ee^{-t(x-\gn(\ga-\frac{S}{\gn t}))}\right)\Big| \#(S)=\kappa(S) \right]=\prod_{x=1}^{\kappa(S)}\left(1+\ee^{-t(x-\gn(\ga-\frac{S}{\gn t}))}\right),
    $$
    concluding the proof.
\end{proof}

Finally, it remains to estimate the probability of the counting function that appears in the last result. For that, we will explore the underlying determinantal structure. First of all, we use Markov's inequality and for any given $v>0$ we express
\begin{equation}\label{eq:markovsat}
\P(\#(S)\leq \kappa(S)-1)=\P\left(\ee^{-v \#(S)}\leq \ee^{-v(\kappa(S)-1)}\right)\leq \ee^{v(\kappa(S)-1)}\E\left[\ee^{-v\#(S)} \right].
\end{equation}
The expectation on the right-hand side is the moment generating function of the counting function $\#(S)$. From the theory of determinantal point processes we obtain
\begin{equation}\label{eq:detreprMGF}
\E\left[\ee^{-v\#(S)} \right]=\det \left(\msf I-(1-\ee^{-u})\msf K_\gn\right)_{\ell^2(\Z_{>0}\cap (0,\kappa(S)])},
\end{equation}
where $\msf K_\gn$ is the correlation kernel of the Meixner point process, obtained when $\gsig\equiv 1$, or in other words by making $s\to +\infty$ in \eqref{deff:CDkernel}. We now estimate this moment generating function exploring this determinantal representation.
\begin{lemma}\label{lem:estMGFsat}
Given any $\delta>0$, there exist $\eta>0$ and $v_0>0$ such that the estimate
$$
\E\left[\ee^{-v\#(S)} \right]=\ee^{-v \kappa(S) }\left(1 + \Boh(\ee^{-\eta \gn^{-1/2}S^{3/2}})\right)
$$
holds true uniformly for $\delta \gn^{1/2}\leq S\leq \delta \gn$ and uniformly for $0\leq v\leq  v_0 S^{3/2}\gn^{-1/2}$.  
\end{lemma}

\begin{proof}
We first recall some facts about the kernel $\msf K_\gn$. Since $\kappa(\delta)<\ga$, the interval $(0,\kappa(\delta)]$ is inside the saturated region of $\mcal X_\gn$, where the position of the particles is essentially deterministic, up to exponentially small corrections. Concretely, this means that the kernel $\msf K_\gn$ writes as
\begin{equation}\label{eq:satKnRn}
\msf K_\gn(x,y)=\delta_{x,y}-\msf R_\gn(x,y),\quad x,y\in (0,\kappa(\delta)]\cap \Z_{>0},
\end{equation}
where $\msf R_\gn$ satisfies an estimate of the form
$$
|\msf R_\gn(x,y)|\leq M\ee^{-\gn f(\kappa(S))}, \quad \text{uniformly for } x,y\in \Z_{>0}\cap [0,\kappa(S)],
$$
for some value $M>0$, and the function $f(x)$ is strictly positive on $[0,\ga)$ with 
$$
f(x)=c(\ga-x)^{3/2}(1+\Boh(\ga-x))
\quad \text{as}\quad x\nearrow \ga,
$$ 
where $c>0$. Under our working assumptions and the condition $S=\delta\gn$, the estimate we just described can be obtained with a canonical Riemann-Hilbert/nonlinear steepest descent analysis, and $f$ appears as the $+$-boundary value of the $\phi$-function that appears in this analysis. In the larger regime $\delta\gn^{1/2}\leq \delta \gn$ we are working here, the standard analysis may be modified, with a shrinking neighborhood for the Airy parametrix at $z=\msf a$, and the estimate is still valid at the cost of making $f$ a positive multiple of the $\phi$-function. The arguments needed in this RHP analysis are technical, somewhat lengthy, but standard, and we skip details.

With the local behavior of $f$ in mind, the estimate on $\msf R_\gn$ updates to
\begin{equation}\label{eq:estRNsat}
|\msf R_\gn(x,y)|\leq M\ee^{-\eta\gn^{-1/2}S^{3/2}}, \quad \text{uniformly for } x,y\in \Z_{>0}\cap [0,\kappa(S)].
\end{equation}
Back to the estimate of the moment generating function, we use the Fredholm determinant representation \eqref{eq:detreprMGF}. We view \eqref{eq:satKnRn} as the operator identity $\msf K_\gn=\msf I-\msf R_\gn$ on the finite-dimensional space $\ell^2(\Z_{>0}\cap (0,\kappa(S)])$ of dimension $\kappa(S)$, and manipulate
$$
\ee^{v\kappa(S)}\E\left[\ee^{-v\#(S)} \right]=\ee^{v\kappa(S)}\det \left(\msf I-(1-\ee^{-v})\msf K_\gn\right)=\det\left( \ee^v\left(\msf I-(1-\ee^{-v})\msf K_\gn\right) \right)=\det\left(\msf I-(\ee^v-1)\msf R_\gn\right),
$$
where in the above all determinants are on $\ell^2(\Z_{>0}\cap (0,\kappa(S)])$. The Fredholm determinant expansion of the right-hand side yields
$$
\ee^{v\kappa(S)}\E\left[\ee^{-v\#(S)} \right]=
1+\sum_{k=1}^{\kappa(S)}\frac{(-1)^k}{k!}
\sum_{x_1,\hdots, x_k=1}^{\kappa(S)} (\ee^v-1)^k \det\left(\msf R_\gn(x_i,x_j)\right)_{i,j=1}^k 
$$
Hadamard's inequality combined with \eqref{eq:estRNsat} yields that
$$
|\det\left(\msf R_\gn(x_i,x_j)\right)_{i,j=1}^k |\leq M^k \ee^{-\eta k \gn^{-1/2}S^{3/2}},
$$
and we obtain that
\begin{equation}\label{eq:estsatdet}
\left|\ee^{v\kappa(S)}\E\left[\ee^{-v\#(S)} \right]-1\right|
\leq 
\sum_{k=1}^{\kappa(S)}\frac{1}{k!}\left[M\kappa(S)(\ee^v-1)\ee^{-\eta \gn^{-1/2}S^{3/2}}\right]^k
\end{equation}
For $S$ in the working range, we choose $v_0=\eta/2$, so that for some new value $\eta'>0$,
$$
M\kappa(S)(\ee^v-1)\ee^{-\eta \gn^{-1/2}S^{3/2}}\leq M\kappa(S)(\ee^{-\frac{\eta}{2} \gn^{-1/2}S^{3/2}}-\ee^{-\eta \gn^{-1/2}S^{3/2}}) \leq \ee^{-\eta' \gn^{-1/2}S^{3/2}},
$$
which implies that the right-hand side of \eqref{eq:estsatdet} is bounded by $\ee^{-\eta' \gn^{-1/2}S^{3/2}}$, concluding the proof.
\end{proof}

We are finally ready to conclude this section with the final proof.

\begin{proof}[Proof of Theorem~\ref{thm:aprioriest}]
Having in mind that $\#(S)\leq \kappa(S)$, a combination of Lemmas~\ref{lem:lowerboundLNsat} and \ref{lem:upperboundLNsat} gives that
$$
1-\P\left(\#(S)\leq \kappa(S)-1 \right)\leq \dfrac{\msf L_\gn(S)}{\prod_{x=1}^{\kappa(S)}\left(1+\ee^{-t(x-\gn(\ga-\frac{S}{\gn t}))}\right) }\leq \exp\left( \ee^{-\frac{S}{2}+\frac{1}{t}}\frac{\ee^{-t}}{1-\ee^{-t}} \right).
$$
In our range of values of $S$, we in fact have that $S$ is larger than $\eta S^{3/2}\gn^{-1/2}$, for some $\eta>0$, and therefore
$$
\exp\left( \ee^{-\frac{S}{2}+\frac{1}{t}}\frac{\ee^{-t}}{1-\ee^{-t}} \right)=1+\Boh\left(\ee^{-\eta S^{3/2}\gn^{-1/2}}\right).
$$
On the other hand, thanks to \eqref{eq:markovsat} and Lemma~\ref{lem:estMGFsat},
$$
\P\left(\#(S)\leq \kappa(S)-1 \right)=\ee^{-v}\left(1+\Boh(\ee^{-\eta S^{3/2}\gn^{-1/2}} )\right),
$$
This last estimate is uniform for $v\leq v_0 S^{3/2}\gn^{-1/2}$, and so we choose $v= v_0 S^{3/2}\gn^{-1/2}$ to conclude the proof.
\end{proof}

The following result is the main output we will need for later.
\begin{corollary}\label{cor:roughbound}
    For any $\delta>0$, there exists $\eta>0$ such that the estimate
    $$
    \log \msf L_\gn(\delta \gn^{1/2})=\Boh\left( \ee^{-\eta \gn^{1/4}} \right),\quad \gn\to\infty,
    $$
    is valid
\end{corollary}

\begin{proof}
    Using \eqref{eq:SLrelation} and Theorem~\ref{thm:aprioriest} with $S=\delta\gn^{1/2}$, we obtain
    $$
    \log\msf L_\gn(\delta\gn^{1/2})=\sum_{x=\kappa(S)+1}^\infty \log\left(1+\ee^{-t(x-\gn(\ga-S/(t\gn)))}\right)+\Boh(\ee^{-\eta \gn^{1/4}}),\quad \gn\to \infty.
    $$
    A simple argument shows that the remaining sum is in fact $\Boh(\ee^{-\eta \gn})$, for some $\eta>0$, concluding the proof.
\end{proof}

\section{The model problem}\label{sec:modelRHP}

At the core of our arguments lies in a RHP - the model RHP - that we introduce and study in-depth in this section. This model problem is an extension of several previously known RHPs from the literature, and it depends on two given power series $\msf P$ and $\msf F$ defined near the origin. In Section~\ref{sec:admdata} we introduce the necessary conditions in such power series, and obtain some straightforward but important estimates on related quantities. Next, in Section~\ref{sec:introofmodelproblem} we then introduce the model problem itself, and quantities related to it that will be of particular relevance for us.

\subsection{The class of admissible data for the model problem}\label{sec:admdata}
\hfill

As mentioned, the model problem we are about to introduce will depend on a pair of power series that we introduce next.

\begin{definition}\label{deff:admissibleFP}
We say that a pair of functions $(\msf P,\msf F)$ is {\it admissible} if they are defined and analytic on a disk $D_\delta(0)$, are real-valued over $(-\delta,\delta)$, and
\begin{align*}
& \msf P(0)=0,\quad \msf c_{\msf P}\deff \msf P'(0)>0, \\
& \msf F(0)=0, \quad \msf c_{\msf F}\deff \msf F'(0)>0.
\end{align*}
\end{definition}

Fix a value $\lambda_0\in (0,\delta/2)$, which will be made arbitrarily small but fixed. Given $\lambda\in [-\lambda_0,\lambda_0]$, we introduce the contours
\begin{equation}\label{deff:Gammalambda}
\Gamma_k^\lambda\deff \lambda +(0, \ee^{\pi \ii k/3}\infty),\quad k=0,\hdots ,5,\quad \Gamma^\lambda\deff \bigcup_{k=0}^5 \Gamma_k^\lambda,
\end{equation}
and the sectors
\begin{equation}\label{deff:Omegalambda}
\Omega_j^\lambda \deff \left\{ \zeta\in \C\mid \arg(\zeta-\lambda )\in \left( \frac{\pi j}{3}, \frac{\pi(j+1)}{3} \right) \right\}.
\end{equation}
We orient $\Gamma_0^\lambda$, $\Gamma_1^\lambda$ and $\Gamma_5^\lambda$ from the origin to $\infty$, and $\Gamma_2^\lambda$, $\Gamma_3^\lambda$ and $\Gamma_4^\lambda$ from $\infty$ to the origin. These sets are displayed in Figure~\ref{Fig:ContourModel}.

\begin{figure}[t]
\centering
\includegraphics[scale=1]{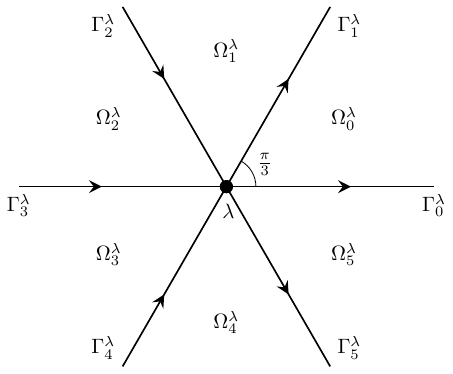}
\caption{The contours $\Gamma_j^\lambda$ and regions $\Omega_j^\lambda$ introduced in \eqref{deff:Gammalambda} and \eqref{deff:Omegalambda}, respectively. The angle between two consecutive $\Gamma_j^\lambda$'s is always $\frac{\pi}{3}$.}
\label{Fig:ContourModel}
\end{figure}

The model problem will be defined on $\Gamma^\lambda$ for appropriate choices of $\lambda\in [-\lambda_0,\lambda_0]\subset D_\delta(0)$, and will involve admissible functions $(\msf F,\msf P)$ which are originally defined only on a neighborhood $D_\delta(0)$ of the origin. By fixing the pair $(\msf F,\msf P)$ we are in principle also fixing the radius of convergence $\delta>0$ of their power series, and we now extend these functions to a full neighborhood of $\Gamma^\lambda$ in an appropriate way. Nevertheless, we stress that along the way it will turn out convenient to choose the values $\delta>0$ and $\lambda_0\in (0,\delta/2)$ to be arbitrarily small, but they should always be seen as fixed.

\begin{prop}\label{prop:extAdmFunc}
Given any pair of admissible functions $(\msf P,\msf F)$, we can choose $\delta>0$ and $\lambda_0>0$ sufficiently small, such that $(\msf P,\msf F)$ admit $C^\infty$ extensions from $D_\delta(0)$ to $\C$ which are real-valud over the real line, and with the following additional properties.

\begin{enumerate}[(i)]
    \item The function $\msf P$ is increasing over the real line, and 
    \begin{align*}
    & \re \msf P(w) \geq \frac{1}{2} \msf c_{\msf P}\re w,  \quad \text{for }\re w>0,\qquad \text{and} \\ 
    & \re \msf P(w) \leq  -\frac{1}{2} \msf c_{\msf P}|\re w|,  \quad \text{for }\re w<0,
    \end{align*}
    are valid.

\item The inequalities
    \begin{align*}
    & \im \msf F(w) \geq \frac{1}{2} \msf c_{\msf F}\im w,  \quad \text{for }\im w>0,\qquad \text{and} \\ 
    & \im \msf F(w) \leq  -\frac{1}{2} \msf c_{\msf F}|\im w|,  \quad \text{for }\im w<0,
 \end{align*}
    are valid.

\item Set $\epsilon_j=1$ for $j=1,2$, and $\epsilon_j=-1$ for $j=4,5$. There exists $\eta>0$ such that the inequalities
\begin{equation}\label{eq:estF32prec}
\frac{4}{3}\re \left(w^{3/2}\right) + \epsilon_j\im \msf F\left(w\right)\geq \eta |w-\lambda|, \quad w\in \Gamma_{j}^{\lambda},
\end{equation}
are valid whenever $\lambda \in [-\lambda_0,\lambda_0]$ and $\gt>0$ is chosen sufficiently large.

\end{enumerate}
\end{prop}

\begin{proof}

We verify that properties (i)--(ii) are valid for $w\in D_\delta(0)$, eventually taking $\delta,\lambda_0$ sufficiently small. Once this is established, the fact that we can extend such properties to the whole complex plane is a cumbersome but straightforward construction using basic interpolation techniques, and we skip it.

Both $\msf P$ and $\msf F$ are real-valued on $\R\cap D_\delta(0)$, by their very definition as real power series near the origin. Furthermore, since we are assuming $\msf c_{\msf P}=\msf P'(0)>0$ and $\msf c_{\msf F}=\msf F'(0)>0$, the claimed properties (i) and (ii) are valid in $D_\delta(0)$, eventually making $\delta>0$ sufficiently small. Observe that these properties considered so far do not depend on $\lambda$.

Property (iii), however, also depends on $\lambda$ as it is a property along $\Gamma_j^{\lambda}$. Once again, we prove it for $w\in \Gamma_j^{\lambda}\cap D_\delta(0)$, and because we are assuming that $\lambda \in [-\lambda_0,\lambda_0]$, the proof for general $w\in \Gamma_j^{\lambda}$ follows by adjusting $\lambda_0>0$ if necessary.

For the sake of clarity in the arguments, we prove \eqref{eq:estF32prec} for $j=1,2$, the proof for $j=4,5$ is analogous. 

Write $w\in D_\delta(0)$ with $\im w>0$ as $w=r\ee^{\ii \theta}$, $r=|w|$, $\theta=\arg w \in (0,\pi)$. Using (ii), we bound
$$
\frac{4}{3}\re \left(w^{3/2}\right)+\im \msf F\left(w\right) \geq \frac{4}{3}\re \left(w^{3/2}\right)+\frac{1}{2}\cF \im w = \frac{4}{3}r^{3/2}\cos\left(\frac{3}{2}\theta\right)+\frac{1}{2}\cF r \sin \theta.
$$
Using that $\cos(3\theta/2)\geq -2\sin \theta$ for $\theta\in [0,\pi]$, we obtain
$$
\frac{4}{3}\re \left(w^{3/2}\right)+\im \msf F\left(w\right)\geq 
 \left(\frac{1}{2}\cF-\frac{8}{3}r^{1/2}\right)r\sin\theta.
$$
Thus, by making $\delta>0$ sufficiently small, we can make sure that $\cF/2-8r^{1/2}/3\geq \cF/4$ for every $r\in [0,\delta]$, concluding that
$$
\frac{4}{3}\re \left(w^{3/2}\right)+\im \msf F\left(w\right)\geq \frac{1}{8}\cF r \sin\theta =\frac{1}{8}\cF \im w,\quad \text{for every }w\in D_\delta(0).
$$
The proof is then completed by noticing that for $w\in \Gamma^\lambda_1\cup \Gamma_2^\lambda$, we have $\im w=\im (w-\lambda)=\frac{\sqrt{3}}{2}|w-\lambda|$.
\end{proof}

From now on, given admissible functions $\msf P,\msf F$, we already assume that they have been extended to $\C$, as ensured by Proposition~\ref{prop:extAdmFunc}.

For real-valued parameters $s \in \R, u,\gt>0$, we introduce scaled versions of admissible functions:
\begin{equation}\label{eq:PtauFtau}
\msf P_\gt(\zeta)\deff -\frac{s}{\gt} +\gt^2\msf P\left(\frac{\zeta}{\gt^2}\right),\quad \msf F_\gt(\zeta)\deff -u \gt^2 +\gt^2 \msf F\left(\frac{\zeta}{\gt^2}\right),\quad \zeta\in \C.
\end{equation}

For us, we allow $u>0$ to vary with $\gt$ as $\gt\to +\infty$, but remaining in a fixed compact set of $(0,+\infty)$, but we never stress this fact further. In our original context of dOPEs and the stochastic six-vertex model, we will in fact choose $\gt=\gn^{1/3}$, but in our study of the model problem we opt for using this parameter $\gt$ as a dummy parameter, which may be seem as a (large) time parameter. The value $u$ itself will later be chosen as $2\pi\ga$.

The parameter $s\in \R$ will also be allowed to vary with $\gt$, but within the following regime: for {\it any} $T>0$ sufficiently large but fixed, we assume that
\begin{equation}\label{eq:critregime}
-\frac{1}{T} \gt^{3/2} \leq s\leq \frac{1}{T}\gt^{3/2}.
\end{equation}
Under the identification $\gt=\gn^{1/3}$, such restriction is simply capturing that we are working under Assumptions~\ref{assumpt:parameterregimes}

When we say that a property holds uniformly for $s$, we always mean that the property holds uniformly for $s$ as in \eqref{eq:critregime}, with $T$ fixed but possibly made sufficiently large. Later on, we will further split this regime into three distinct sub-regimes in terms of $\gt$, to reflect the sub-regimes from Assumptions~\ref{assumpt:parameterregimes}.

Throughout the whole paper, a zero of $\msf P_\gt$ on the real line will play a major role, and we now introduce it.

\begin{prop}\label{prop:zeroPt} 
Suppose that $s$ is within the regime \eqref{eq:critregime}. The equation
$$
\msf P_\gt(\zeta)=0
$$
has exactly one solution $\zeta=\zeta_\gt(s)$ on the real axis, and for $\zeta\in \Gamma^{\zeta_\gt(s)}$,
$$
\re\msf P_\gt(\zeta)>0 \quad \text{if } \re\zeta>\zeta_\gt(s),\qquad \text{and}\qquad \re\msf P_\gt(\zeta)<0 \quad \text{if } \re\zeta<\zeta_\gt(s).
$$
Furthermore, this solution $\zeta_\gt(s)$ is a real analytic function of $s\in \R$, and it satisfies
\begin{equation}\label{eq:zetasseries}
\zeta_\gt(s)=\frac{1}{\msf c_{\msf P}}\frac{s}{\gt}\left(1+\Boh(s \gt^{-3})\right),\quad \gt\to \infty,
\end{equation}
uniformly for $s$ within the regime \eqref{eq:critregime}, and where we recall that $\msf c_{\msf P}>0$ is as in Definition~\ref{deff:admissibleFP}.
\end{prop}

\begin{proof}
The zeros of $\msf P_\gt$ are solutions to 
\begin{equation}\label{eq:zerosPgt}
\msf P\left(\frac{\zeta}{\gt^{2}}\right)=\frac{s}{\gt^3}
\end{equation}
Thanks to Proposition~\ref{prop:extAdmFunc}--(i), the function $\zeta\mapsto \msf P$ is strictly increasing, becomes negative for very negative values of $\zeta$ and positive for very positive values of $\zeta$. Thus, there is exactly one such solution $\zeta_\gt(s)$, and we may invert $\msf P$ on the real axis and conclude that
$$
\zeta_\gt(s)=\gt^2\msf P^{-1}\left(\frac{s}{\gt^3}\right)
$$
The local behavior of $\msf P$ near the origin (see Definition~\ref{deff:admissibleFP}) ensures that $\msf P$ is an analytic invertible function in a neighborhood of the origin, say $D_\varepsilon(0)$. Thanks to \eqref{eq:critregime}, by making $T$ sufficiently large we can be sure that $|s/\gt^3|<\varepsilon$ for $\gt$ sufficiently large, and the Implicit Function Theorem gives us that $\zeta_\gt(s)$ is an analytic function of $s$ that satisfies \eqref{eq:zetasseries} uniformly as claimed.

Finally, the inequalities on $\re\msf P_\gt(\zeta)$ also follow from the fact that $\msf P$ is increasing over the real line.
\end{proof}

We will also need a certain control on the values of the functions $\msf P_\gt$ and $\msf F_\gt$ over the contour $\Gamma^{\zeta_\gt(s)}$. We collect these properties as separate results. Many of the properties that we now collect are based on the following observation. For $s$ within the critical regime, a Taylor expansion shows that the approximations
\begin{equation}\label{eq:TaylorExpPgtFgt}
\msf P_\gt(\zeta)=\msf c_{\msf P}(\zeta-\zeta_\gt(s))+\Boh\left(\frac{s^2}{\gt^4}\right)\qquad \text{and}\qquad \msf F_\gt(\zeta)=-u\gt^2+\msf c_{\msf F}\zeta +\Boh\left(\frac{s^2}{\gt^4}\right)
\end{equation}
are valid as $\gt\to \infty$, uniformly for $(1+|\zeta_\gt(s)|)^{-1}\zeta$ in compact subsets of $\C$. Thus, in essence, in the analysis that follows we are always guided by these linear approximations to $\msf P_\gt$ and $\msf F_\gt$. 

For instance, this approximation indicates that
$$
\frac{1}{1+\ee^{\gt\msf P_\gt(\zeta)}}\approx \frac{1}{1+\ee^{\gt \msf c_{\msf F}(\zeta-\zeta_\gt(s))} }
$$
Hence, any disc centered at $\zeta_\gt(s)$ of radius $\Boh(1)$ contains $\Boh(\gt)$ poles of the left-hand side. We will have to control the location of such poles, and to do it we first establish a general lemma.

\begin{lemma}\label{lem:controlexppoles}
    Let $G\subset \C$ be an arbitrary subset of the complex plane. Suppose that
    \begin{equation}\label{eq:infGlattice}
    \inf_{\substack{ w\in G \\ k\in \Z }} |w-(2k+1)\pi \ii| \geq \epsilon,
    \end{equation}
    for some $\epsilon>0$. Then there exists $\delta>0$ such that
    $$
    \inf_{w\in G}|1+\ee^{w}|\geq \delta.
    $$
\end{lemma}
\begin{proof}
    Suppose, to the contrary, that there exists a sequence $(w_n)\subset G$ for which
    $$
    |1+\ee^{w_n}|\to 0,\quad n\to\infty.
    $$
    Write $w_n=x_n+\ii y_n$, $x_n,y_n\in \R$. From the basic inequality $|\ee^{x_n}-1|=||\ee^{w_n}|-|-1||\leq |\ee^{w_n}+1|$ we learn that $x_n\to 0$ as well, and from \eqref{eq:infGlattice} we see that there exists $n_0>0$ such that
    \begin{equation}\label{eq:infboundtheta}
    \inf_{ k\in \Z } |y_n-(2k+1)\pi \ii| \geq \frac{\epsilon}{2}, \quad \text{ for every }n\geq n_0.
    \end{equation}
    From a Taylor expansion,
    $$
    \ee^{w_n}=\ee^{x_n}\ee^{\ii y_n}=(1+\Boh(x_n))\ee^{\ii y_n}=\ee^{\ii y_n}+\Boh(x_n),
    $$
    and therefore
    \begin{equation}\label{eq:convangles}
    |1+\ee^{\ii y_n}|\to 0
    \end{equation}
    we must have $|1+\ee^{\ii y_n}|\to 0$. We view $(y_n)$ as a sequence of angles, and denote $\theta_n=y_n \mod 2\pi $, with $\theta_n\in [0,2\pi)$. By passing to a subsequence if necessary, we may assume that $(\theta_n)$ converges in $[0,2\pi]$, and from \eqref{eq:convangles} we learn that $\theta_n\to -\pi $. In particular, this convergence ensures that
    $$
    \inf_{k\in \Z}|y_n-(2k+1)\pi |=\inf_{k\in \Z}|\theta_n-(2k+1)\pi |\to 0,\quad n\to \infty,
    $$
    where the first equality above is valid because $\theta_n=y_n \mod 2\pi$. But this last conclusion is in contradiction with \eqref{eq:infboundtheta}, concluding the proof.
\end{proof}

With the next result, we ensure that the factor $(1+\ee^{\gt\msf P_\gt(\zeta)})^{-1}$ remains bounded along the contour $\Gamma^{\zeta_\gt(s)}$, in spite of having poles accumulating near $\zeta_\gt(s)$ as $\gt\to \infty$.

\begin{prop}\label{prop:nondegenerates}
Suppose that $s$ satisfies \eqref{eq:critregime}.

Then, there exist $\gt_0>0$ and $\delta>0$ such that
\begin{equation}\label{eq:nondegsregcritconsq}
 |1+\ee^{\gt \msf P_\gt(\zeta)}|\geq \delta,\quad \gt\geq \gt_0,\quad \text{for every } \zeta\in \Gamma^{\zeta_\gt(s)},
 \end{equation}
 and for every $\gt\geq \gt_0$. 
\end{prop}

\begin{proof}
The proof will be split into two parts. In the first part, we prove that if $|\re \zeta|$ is sufficiently large in an explicit manner, then the mild bound from Proposition~\ref{prop:extAdmFunc} will give that $|1+\ee^{\gt \msf P_\gt(\zeta)}|$ has to be bounded from below. In the second part, we consider $|\re \zeta|$ to be of moderate size, and then use Lemma~\ref{lem:controlexppoles} in combination with the power series \eqref{eq:TaylorExpPgtFgt} to conclude the proof.

Proceeding with the first part, assume that $\zeta\in \Gamma^\zs$ satisfies $|\re \zeta|\geq M|\zeta_\gt(s)|+1$ with $M=\frac{8}{\msf c_{\msf P}}+2$, say first with $\re \zeta \geq M|\zeta_\gt(s)|+1>0$. Using Proposition~\ref{prop:extAdmFunc}--(i) we compute
\begin{align}
\re \msf P_\gt(\zeta) & \geq -\frac{s}{\gt}+\frac{1}{2}\cP \re \zeta =\frac{1}{4}\cP\re(\zeta-\zs) +\frac{1}{4}\cP \re\zeta  +\frac{1}{4}\cP\zs -\frac{s}{\gt} \nonumber\\
& \geq \frac{1}{4}\cP\re(\zeta-\zs) +\frac{1}{4}\cP M |\zs| +\frac{1}{4}\cP\zs -\frac{s}{\gt}  \nonumber\\
& = \frac{1}{4}\cP\re(\zeta-\zs) +\frac{1}{4}\cP (2|\zs|-\zs)+2|\zs|-\frac{s}{\gt} \nonumber\\ 
& >  \frac{1}{4}\cP\re(\zeta-\zs) + 2|\zs|-\frac{s}{\gt}. \label{eq:ineqPzs1}
\end{align}

From the expansion \eqref{eq:zetasseries} we learn that $2|\zs|-\frac{s}{\gt}\geq 0$ for $\gt$ sufficiently large, and therefore we just learnt that in this case
\begin{equation}\label{eq:ineqPzs2}
\re \msf P_\gt(\zeta)\geq \frac{1}{4}\cP \re(\zeta-\zs)\geq \frac{1}{4}\cP,
\end{equation}
where in the last step we used that $\re \zeta \geq M |\zs|+1\geq  \zs+1$. Hence, still in the case $\re\zeta \geq M|\zs|+1$ we learn that
$$
|1+\ee^{\gt \msf P_\gt(\zeta)}|\geq \ee^{\gt\cP/4}-1,
$$
and the right-hand side can be made strictly positive by choosing $\gt>0$ sufficiently large. In a completely analogous way, in the case when $\re \zeta \leq -(M|\zeta_\gt(s)|+1)<0$ we use again Proposition~\ref{prop:extAdmFunc}--(i) to obtain 
\begin{equation}\label{eq:ineqPzs3}
\re\msf P_\gt(\zeta)\leq - \frac{1}{4}\cP |\re(\zeta-\zs)| \leq -\frac{1}{4}\cP,
\end{equation}
and therefore
$$
|1+\ee^{\gt \msf P_\gt(\zeta)}|\geq 1- \ee^{-\gt\cP/4},
$$
which again can be made strictly positive. This finishes the first part of the proof.

For the second part, assume that $\zeta\in \Gamma^\zs$ with $|\re \zeta |\leq (M|\zeta_\gt(s)|+1)$.  Our next goal is to show that for such values of $\zeta$, the condition \eqref{eq:infGlattice} for $w=\gt\msf P_\gt(\zeta)$ is satisfied for some $\epsilon>0$.

If $\zeta\in \R$, then because $\msf P_\gt$ is real-valued we simply bound $|1+\ee^{\gt\msf P_\gt(\zeta)}|=1+\ee^{\gt\msf P_\gt(\zeta)}\geq 1$, and from now on we assume that $\zeta \in \Gamma^\zs\setminus \R$.

Because $\Gamma^\zs\setminus \R$ consists of straight lines with positive angle from the real axis, we can find a new value $M'>0$, independent of other parameters, such that
$$
|\zeta|\leq M'|\re \zeta|\leq M'(M|\zs|+1 ).
$$
For later use, denote
$$
G=\{\zeta\in \Gamma^\zs\setminus \R \mid |\zeta|\leq M'|\re \zeta|\leq M'(M|\zs|+1 )\}.
$$
In particular, for $\zeta\in G$ we have that $|\zs|^{-1}\zeta$ remains within a compact set of the complex plane, and the expansion \eqref{eq:TaylorExpPgtFgt} can be used. We write $\zeta\in G$ as $\zeta=\zs+r\omega $, with $r=|\zeta-\zs|$, $|\omega|=1$ with $|\re\omega|=\frac{1}{2}$ and $|\im \omega|=\frac{\sqrt{3}}{2}$. In particular, we may write $\im \omega=\varsigma \re w$ with $\varsigma\in \{\sqrt{3},-\sqrt{3}\}$. In these coordinates, the expansion \eqref{eq:TaylorExpPgtFgt} then says that
$$
\im \msf P_\gt(\zeta)=\varsigma \re \msf P_\gt(\zeta)+\Boh\left(\frac{s^2}{\gt^4}\right).
$$
Hence, we can choose $T>0$ such that for $s$ within the regime \eqref{eq:critregime},
\begin{equation}\label{eq:boundRePImP}
|\im \msf P_\gt(\zeta)-\varsigma\re \msf P_\gt(s)|\leq \frac{1}{10 \gt},\quad \zeta\in G.
\end{equation}
The value $1/10$ here is chosen just for concreteness, any sufficiently small value would suffice for what follows.

Now, using the inequality $|w|\geq \frac{1}{2}(|\re w|+|\im w|), w\in \C$, we write for $k\in \Z$ and $\zeta\in G$,
$$
|\gt \msf P_\gt(\zeta)-(2k+1)\pi \ii |\geq \frac{1}{2}\left(\Big||(2k+1)\pi \ii|-|\gt \im \msf P_\gt(\zeta) |\Big|+|\gt \re \msf P_\gt(\zeta)|  \right)
$$
If, say, $|\gt \re \msf P_\gt(\zeta)|\geq \frac{1}{10\gt}$, then the right-hand side above is obviously bounded by $1/10$. If, on the other hand, $|\gt \re \msf P_\gt(\zeta)|\leq \frac{1}{10\gt}$, then \eqref{eq:boundRePImP} ensures that $|\im \msf P_\gt|\leq \frac{1+|\varsigma|}{10\gt }=\frac{1+\sqrt{3}}{10\gt }$, and we bound the right-hand side above obtaining
$$
|\gt \msf P_\gt(\zeta)-(2k+1)\pi \ii |
\geq \frac{1}{2}|2k+1|\pi -|\gt \im \msf P_\gt(\zeta)|
\geq \frac{1}{2}\pi -\frac{1+\sqrt{3}}{10}\geq 1
$$
We just verified that 
$$
\sup_{\substack{\zeta \in G \\ k\in \Z}} |\gt \msf P_\gt(\zeta)-(2k+1)\pi \ii |\geq \frac{1}{10}.
$$
The proof is now finally completed using Lemma~\ref{lem:controlexppoles}.
\end{proof}

\begin{remark}
    The proof of Proposition~\ref{prop:nondegenerates} actually shows an additional fact that will be useful later: given $\varepsilon>0$, there exist $\eta>0$ and $C>0$ such that
    \begin{equation}\label{eq:ineqPrezetazs}
    \begin{aligned}
        & |\ee^{\gt\msf P_\gt(\zeta)}|\geq C\ee^{\gt \eta |\re (\zeta-\zs)|},\quad \text{for every } \zeta \in \Gamma^\zs \text{ with } \re(\zeta-\zs)\geq 0, \\
        & |\ee^{\gt\msf P_\gt(\zeta)}|\leq C\ee^{-\gt \eta |\re (\zeta-\zs)|},\quad \text{for every } \zeta \in \Gamma^\zs \text{ with } \re(\zeta-\zs)\leq  0,
    \end{aligned}
    \end{equation}
    and these inequalities are valid for every $s$ within the regime \eqref{eq:critregime}.
    
    Indeed, if $|\re \zeta|\geq M|\zs|+1$, then the developments in \eqref{eq:ineqPzs1}, \eqref{eq:ineqPzs2} and \eqref{eq:ineqPzs3} already show that 
    \begin{align*}
    & \ee^{\gt\re \msf P_\gt(\zeta)}\geq \ee^{\frac{1}{4}\gt\cP\re(\zeta-\zs)}, \quad\text{if } \re(\zeta-\zs)>0,\qquad \text{and} \\
    & \quad \ee^{\gt\re \msf P_\gt(\zeta)}\leq \ee^{-\frac{1}{4}\gt\cP|\re(\zeta-\zs)|},\quad  \text{if } \re(\zeta-\zs)<0.
    \end{align*}
    If, on the other hand, $|\re \zeta|\leq M|\zs|+1$, the expansion \eqref{eq:TaylorExpPgtFgt} yields
    \begin{equation}
    \re \msf P_\gt(\zeta)-\cP \re(\zeta-\zs)=\Boh\left(\frac{s^2}{\gt^4}\right).
    \end{equation}
    For $s$ within the critical regime, the error term on the right-hand side is $\Boh(\gt^{-1})$. Multiplying across by $\gt$ and taking exponentials, we see that in this case
    $$
    |\ee^{\gt \msf P_\gt(\zeta)}|= \ee^{\Boh(1)} \ee^{\cP \re(\zeta-\zs)},
    $$
    and the claim \eqref{eq:ineqPrezetazs} follows.
\end{remark}

The next result is the analogue of Proposition~\ref{prop:nondegenerates} for factors of the form $(1+\ee^{ \pm \ii \gt\msf F_\gt})^{-1}$. In virtue of the factors $\ii$ and $u$ in the exponent, we cannot ensure this term is bounded away from zero everywhere, but only outside a small neighborhood of the point $\zeta_\gt(s)$. Nevertheless, this will suffice for our purposes.

\begin{prop}\label{prop:controlFzetas}
Given $\varepsilon>0$, there exist $\gt_0>0$ and $\delta>0$ such that
 \begin{equation}\label{eq:nondegtauregcritconsq}
 |1-\ee^{\mp \ii \gt \msf F_\gt(\zeta)}|\geq \delta,\quad \text{for every } \zeta\in \Gamma^{\zeta_\gt(s)} \text{ with } \pm \im \zeta>0 \text{ and } |\zeta-\zeta_\gt(s)|\geq \varepsilon,
 \end{equation}
 for every $\gt>\gt_0$ and for every $s$ within the regime \eqref{eq:critregime}.
\end{prop}

\begin{proof}
Write $\zeta = \zs+r \omega \in \Gamma^\zs\setminus \R$ with $r>0$, $\im \omega =\varsigma \re \omega $, $|\omega |=1$ and $\varsigma\in \{\sqrt{3},-\sqrt{3}\}$. We verify the claim for $\im \zeta=r\im \omega>0$, so that $\im \omega =\sqrt{3}/2$. The claim when $\zeta=r\im \omega=-r \sqrt{3}/2<0$ follows analogously.

Using Proposition~\ref{prop:extAdmFunc}--(ii), we compute that
\begin{equation}\label{eq:estImFImzeta}
\im \msf F_\gt(\zeta)=\gt^2 \im \msf F\left( \frac{\zeta}{\gt^2} \right)\geq \frac{1}{2}\cF \frac{\sqrt{3}}{2}r\geq \cF r.
\end{equation}
Therefore, using that $r=|\zeta-\zs|\geq \varepsilon$,
$$
|1-\ee^{-\gt \ii \msf F_\gt(\zeta)}|\geq |\ee^{-\gt \ii \msf F_\gt(\zeta) }|-1=\ee^{\gt \im \msf F_\gt(\zeta)}-1\geq \ee^{\gt \cF r}-1\geq \ee^{\gt \cF\varepsilon}-1
$$
Hence, by choosing $\gt\geq \gt_0$ with $\gt_0>0$ sufficiently large, we can make sure that the right-hand side above is strictly positive, concluding the proof.
\end{proof}

\begin{remark}
Equation~\eqref{eq:estImFImzeta} actually shows that the inequality
\begin{equation}\label{eq:ineqFabszetazs}
|\ee^{\pm \ii \gt \msf F_\gt(\zeta)}|\leq \ee^{-\cF \gt |\zeta-\zs|},\quad \zeta\in \Gamma^\zs \text{ with } \pm \im \zeta>0,
\end{equation}
is valid for every $\gt>0$ and every $s$.
\end{remark}

At the technical level, the so-called small norm estimates that will appear in the asymptotic analysis of $\bm \Phi_\gt$ will involve the scalar factors in the jump $\bm J_\gt$ from \eqref{eq:jumpJt}, and before starting with the asymptotic analysis of $\bm \Phi_\gt$ we collect such estimates with the next result.

\begin{prop}\label{prop:canonicestPF}
  Suppose that $s$ is in the regime \eqref{eq:critregime}. There exists $\eta=\eta(T)>0$ such that the functions $\msf P_\gt$ and $\msf F_\gt$ satisfy the following estimates as $\gt \to +\infty$.
\begin{enumerate}[(i)]
    \item The estimate
    $$
    \frac{1}{1+\ee^{\gt \msf P_\gt(\zeta)}}=\Boh\left(\ee^{-\gt\eta|\zeta-\zeta_\gt(s)|}\right)
    $$
    holds true uniformly for $\zeta\in \Gamma_0^{\zeta_\gt(s)}\cup \Gamma_1^\zs\cup \Gamma_5^\zs$.

   \item The estimates
   $$
   \ee^{\gt\msf P_\gt(\zeta)}=\Boh\left(\ee^{-\gt\eta|\zeta-\zeta_\gt(s)|}\right), \quad  \frac{1}{1+\ee^{\gt\msf P_\gt(\zeta)}}=1+\Boh\left(\ee^{-\gt\eta|\zeta-\zeta_\gt(s)|}\right)
   $$
   hold true uniformly for $\zeta\in \Gamma_2^{\zeta_\gt(s)}\cup \Gamma_3^\zs\cup \Gamma_4^\zs$.

   \item Set $\epsilon_j\deff 1$ for $j=1,2$, and $\epsilon_j\deff -1$ for $j=4,5$. The estimates
   $$ 
   \ee^{\epsilon_j\ii\gt \msf F_\gt(\zeta)}=\Boh\left(\ee^{-\gt\eta|\zeta-\zeta_\gt(s)|}\right) \qquad \text{and}\qquad \ee^{-\left(\frac{4}{3}\zeta^{3/2}-\epsilon_j \gt \msf F_\gt(\zeta)\right)}=\Boh\left(\ee^{-\gt\eta|\zeta-\zeta_\gt(s)|}\right),
   $$
   hold true uniformly for $\zeta\in \Gamma_j^{\zeta_\gt(s)}$ and $j=1,2,4,5$, as well as for any $\rho>0$, the estimate
   $$
   \frac{1}{1-\ee^{\epsilon_j \gt \msf F_\gt(\zeta)}}=\Boh\left(\ee^{-\gt\eta|\zeta-\zeta_\gt(s)|}\right),
   $$
    holds true uniformly for $\zeta\in \Gamma_j^{\zeta_\gt(s)}\setminus D_\rho(\zs)$, $j=1,2,4,5$.

\end{enumerate}
\end{prop}

\begin{proof}
Thanks to Proposition~\ref{prop:nondegenerates} and Lemma~\ref{lem:basicestimate}, we see that in order to estimate terms of the form $1/(1+\ee^{\gt\msf P_\gt(\zeta)})$ along $\Gamma^\zs$, it suffices to provide estimates for $\ee^{\gt\msf P_\gt(\zeta)}$. Having in mind that for $\zeta\in \Gamma^\zs$, we have $|\zeta-\zs|=\varsigma |\re(\zeta-\zs)|$ with $\varsigma\in \{1/2, 1\}$, Equations~\eqref{eq:ineqPrezetazs} provide the estimates claimed in (i) and (ii).

The first and third estimates in (iii) follow similarly, applying Lemma~\ref{lem:basicestimate} and \eqref{eq:ineqFabszetazs}.

It remains to prove the second estimate in (iii). For that, we use \eqref{eq:PtauFtau} and write
$$
\re \left(\frac{4}{3}\zeta^{3/2} - \ii \epsilon_j\gt \msf F_\gt(\zeta)  \right)= \gt^3 \left( \re \left(\frac{\zeta}{\gt^2}\right)^{3/2}+\epsilon_j \im \msf F\left(\frac{\zeta}{\gt^2}\right) \right)
$$
Now we apply Proposition~\ref{prop:extAdmFunc}--(iii) with $\lambda=\zs/\gt^2=\Boh(\gt^{-1/2})$ and $w=\zeta/\gt^2 \in \Gamma^\lambda_j$.
    
\end{proof}

\subsection{The model problem}\label{sec:introofmodelproblem}  \hfill 

We are finally ready to introduce the model problem. The RHP we are about to introduce depends on a parameter $\gt>0$, which will eventually be made sufficiently large. In everything that follows, we assume that the value $\lambda\in \R$ in \eqref{deff:Gammalambda} satisfies 
\begin{equation}\label{eq:deltalambdarestrict}
\frac{|\lambda|}{\gt^2}\leq \lambda_0,\quad \text{with}\quad \lambda_0\leq \frac{\delta}{2},
\end{equation}
for the value $\lambda_0>0$ from Proposition~\ref{deff:admissibleFP}. Such condition ensures, for instance, that when evaluating $\msf P(\zeta/\gt^2)$ and $\msf F_\gt(\zeta/\gt^2)$ at $\zeta=\lambda$, the particular value of the argument $\lambda/\gt^2$ falls within the disk where $\msf P$ and $\msf F$ are analytic. We will eventually make such $\delta$ arbitrarily small (but fixed) throughout the whole analysis.

Let $(\msf P,\msf F)$ be a pair of admissible functions as in Definition~\ref{deff:admissibleFP} and Proposition~\ref{prop:extAdmFunc}, $\msf P_\gt$ and $\msf F_\gt$ as in \eqref{eq:PtauFtau}, and $\Gamma^\lambda$ as in \eqref{deff:Gammalambda}. Consider the following RHP, the {\it model problem}.

\begin{rhp}\label{rhp:modelPhi} 
For a parameter $\gt>0$, find a $2\times 2$ matrix-valued function $\bm \Phi_\gt$ with the following properties.
\begin{enumerate}[(1)]
\item $\bm \Phi_\gt$ is analytic on $\C\setminus \Gamma^\lambda$.
\item The matrix $\bm \Phi_\gt$ has continuous boundary values $\bm \Phi_{\gt,\pm}$ along $\Gamma^\lambda \setminus \{\lambda\}$, and they are related by $\bm \Phi_{\gt,+}(\zeta)=\bm \Phi_{\gt,-}(\zeta)\bm J_{\gt}(\zeta)$, $\zeta\in \Gamma^{\lambda}\setminus \{\lambda\}$, where
\begin{equation}\label{eq:jumpJt}
\bm J_{\gt}(\zeta)\deff
\begin{dcases}
\bm I+\frac{1}{1+\ee^{\gt\msf P_\gt(\zeta)}}\bm E_{12}, & \zeta\in \Gamma_0^{\lambda}, \\
\bm I-\frac{1}{(1-\ee^{-\ii \gt\msf F_\gt(\zeta)})(1+\ee^{\gt\msf P_\gt(\zeta)})} \bm E_{12}, & \zeta\in \Gamma_1^{\lambda}, \\
\left(\bm I+(1+\ee^{\gt\msf P_\gt(\zeta)})\bm E_{21}\right)\left(\bm I-\frac{1}{(1-\ee^{-\ii \gt\msf F_\gt(\zeta)})(1+\ee^{\gt\msf P_\gt(\zeta)})} \bm E_{12}\right), & \zeta\in \Gamma_2^\lambda, \\
\frac{1}{1+\ee^{\gt\msf P_\gt(\zeta)}}\bm E_{12}-(1+\ee^{\gt\msf P_\gt(\zeta)})\bm E_{21}, & \zeta\in \Gamma_3^\lambda, \\
\left(\bm I-\frac{1}{(1-\ee^{\ii \gt\msf F_\gt(\zeta)})(1+\ee^{\gt\msf P_\gt(\zeta)})} \bm E_{12}\right)\left(\bm I+(1+\ee^{\gt\msf P_\gt(\zeta)}) \bm E_{21}\right), & \zeta\in \Gamma_4^\lambda,\\
\bm I-\frac{1}{(1-\ee^{\ii \gt\msf F_\gt(\zeta)})(1+\ee^{\gt\msf P_\gt(\zeta)})} \bm E_{12}, & \zeta\in \Gamma_5^\lambda.
\end{dcases}
\end{equation}
\item As $\zeta\to \infty$, 
\begin{equation}\label{eq:asymptPhit}
\bm \Phi_\gt(\zeta)=\left(\bm I+\Boh(\zeta^{-1})\right)\zeta^{-\sp_3/4}\bm U_0\ee^{-\frac{2}{3}\zeta^{3/2}\sp_3}.
\end{equation}
\item As $\zeta\to \lambda$, 
$$
\bm \Phi_\gt(\zeta)=
\begin{cases}
    \Boh(1), & \text{if } \ee^{\ii \gt\msf F_\gt(\lambda)}\neq 1, \\
    \Boh(1), & \text{if } \ee^{\ii \gt\msf F_\gt(\lambda)}=1  \text{ and } \zeta\notin \Omega_1^\lambda\cup \Omega_4^\lambda, \\
    \Boh\begin{pmatrix}
       1 & (\zeta-\lambda)^{-1} \\ 1 & (\zeta-\lambda)^{-1} 
    \end{pmatrix}, & \text{if } \ee^{\ii \gt\msf F_\gt(\lambda)}=1  \text{ and } \zeta\in \Omega_1^\lambda\cup \Omega_4^\lambda.
\end{cases}
$$
\end{enumerate}
\end{rhp}

The conditions in RHP~\ref{rhp:modelPhi} (1)--(3) are usual conditions when using model RHPs within the Deift-Zhou nonlinear steepest descent analysis. The form of condition RHP~\ref{rhp:modelPhi} (4) is unusual, in the sense that it depends on the nature of the point $\lambda$, but it is needed because when $ \ee^{\ii \gt\msf F_\gt(\lambda)}=1$ the jump matrix $\bm J_\gt$ becomes unbounded as $\zeta\to \lambda$ along $\Gamma^\lambda_1\cup\Gamma^\lambda_2\cup\Gamma^\lambda_4\cup\Gamma^\lambda_5$. This condition will also appear naturally later on, when we discuss the asymptotic analysis for the discrete OPs. Although this condition will follow us throughout the whole analysis, it poses only a technical issue but not conceptual ones, and will more or less resolve itself automatically.

From the very Definition~\ref{deff:admissibleFP} of admissible functions $\msf F, \msf P$ and \eqref{eq:PtauFtau}, there exists $\varepsilon>0$ for which $1/(1+\ee^{\gt\msf P_\gt})$ is analytic on $D_\varepsilon(\zs)$, and the term $1/(1+\ee^{\pm \ii \gt \msf F_\gt})$ is either analytic on $D_\varepsilon(\zs)$ or it has $\zs$ as its unique pole in $D_\varepsilon(\zs)$ (compare with \eqref{eq:expsigmaPsigmaF}). This value $\varepsilon>0$ obviously depends on $\gt$, but for now it suffices to consider $\gt$ as fixed.

Recall that the sectors $\Omega_j^\zs$ are displayed in Figure~\ref{Fig:ContourModel}. In the neighborhood $D_\varepsilon(\zs)$, we modify
\begin{equation}\label{deff:Phimod}
\bm\Phi_\gt^\md(\zeta)\deff 
\begin{dcases}
    \bm \Phi_\gt(\zeta)\left(\bm I-\frac{\bm E_{12}}{1+\ee^{\gt\msf P_\gt(\zeta)}}\right), & \zeta\in \Omega_0^\zs\cap D_\varepsilon(\zs), \\
    \bm \Phi_\gt(\zeta)\left(\bm I+\frac{\bm E_{12}}{(1-\ee^{-\ii \gt\msf F_\gt(\zeta)})(1+\ee^{\gt\msf P_\gt(\zeta)})}\right)\left(\bm I-\frac{\bm E_{12}}{1+\ee^{\gt\msf P_\gt(\zeta)}}\right), & \zeta\in \Omega_1^\zs\cap D_\varepsilon(\zs), \\
    \bm \Phi_\gt(\zeta)\left( \bm I+(1+\ee^{\gt\msf P_\gt(\zeta)})\bm E_{21} \right)\left( \bm I-\frac{\bm E_{12}}{1+\ee^{\gt\msf P_\gt(\zeta)}}\right), & \zeta\in \Omega_2^\zs\cap D_\varepsilon(\zs), \\
    \bm \Phi_\gt(\zeta)(\bm I-(1+\ee^{\gt\msf P_\gt(\zeta)})\bm E_{21}), & \zeta\in \Omega_3^\zs\cap D_\varepsilon(\zs), \\
    \bm \Phi_\gt(\zeta)\left(\bm I- \frac{\bm E_{12}}{(1-\ee^{\ii \gt\msf F_\gt(\zeta)})(1+\ee^{\gt\msf P_\gt(\zeta)})}\right), & \zeta\in \Omega_4^\zs\cap D_\varepsilon(\zs), \\
    \bm \Phi_{\gt}(\zeta), & \zeta\in \Omega^\zs_5\cap D_\varepsilon(\zs),
\end{dcases}
\end{equation}
and make
\begin{equation}\label{deff:bmHtau}
\mcal H_\gt(\zeta)\deff \left[\bm\Delta_\zeta\bm\Phi_{\gt}^\md(\zeta+\zs)\right]_{21,-},\quad \zeta\in (-\varepsilon,\varepsilon),
\end{equation}
where the operator $\bm\Delta_\zeta$ is defined as in \eqref{deff:Deltaoper}. Naturally $\bm\Phi_\gt^\md(\zeta)=\bm\Phi_\gt^\md(\zeta\mid s)$ and $\mcal H_\gt(\zeta)=\mcal H_\gt(\zeta\mid s)$ also depend on $s$, but we omit from our notation. From \eqref{deff:Phimod}, we obtain
$$
\mcal H_\gt(\zeta)\deff \left[\bm\Delta_\zeta\bm\Phi_{\gt}(\zeta+\zs)\right]_{21,-},\quad 0<\zeta<\varepsilon,
$$
This function $\mcal H_\gt$ is central to our work: with the appropriate identification of $\gt$ with $\gn$, it will turn out to be given by $\mcal H_\gt=2\pi\ii \msf H_\gn$, with $\msf H_\gn$ being the function appearing in Theorem~\ref{thm:multstat_formal}. We now collect some properties that will be relevant for later.

\begin{prop}\label{prop:fundamentalbmHtau}
Assume that $s$ is within the regime \eqref{eq:critregime}. 

\begin{enumerate}[(i)]
\item The function $\bm \Phi_\gt^\md$ is analytic in a neighborhood of $\zeta=0$, and $\bm\Phi_{\gt,-}(\zeta)=\bm\Phi_\gt^\md (\zeta)$ for $\zeta>0$.

\item There exists $\kappa>0$, independent of $s,\gt$, such that $\mcal H_\gt$ extends to a matrix-valued real-analytic function in the real interval $(-\kappa\gt^2,\kappa \gt^2)$. Furthermore, it verifies the relation
\begin{multline}\label{eq:bmHmdPhi}
\mcal H_\gt(x-\zs)=\\ 
\begin{dcases}
    \left[\bm \Delta_\zeta \bm \Phi_{\gt,-}(x)\right]_{21}, & x>\zs, \\ 
    -(1+\ee^{\gt\msf P_\gt(x)})^2 \left[ \left( \bm I+\frac{\bm E_{12}}{1+\ee^{\gt \msf P_\gt(x)}} \right)\bm \Delta_\zeta \bm \Phi_{\gt,-}(x)\left( \bm I-\frac{\bm E_{12}}{1+\ee^{\gt \msf P_\gt(x)}} \right) \right]_{12}-\frac{\dd}{\dd \zeta}(\ee^{\gt\msf P_\gt(\zeta)})\Big|_{\zeta=x}, & x<\zs.
\end{dcases}
\end{multline}

\end{enumerate}

\end{prop}

\begin{proof}
The relation between $\bm\Phi_\gt^\md$ and $\bm\Phi_{\gt,-}$ claimed in (i) is immediate from \eqref{deff:Phimod}.

The definition of $\bm \Phi_\gt^\md$ is performed in such a way that $\bm\Phi_\gt^\md$ is the analytic continuation of $\bm \Phi_\gt$ around the point $\zs$, starting from the sector $\Omega_0^\zs$. In other words, using the very definition of $\bm J_\gt$ as the jump of $\bm \Phi_\gt$ in \eqref{eq:jumpJt}, it follows that $\bm \Phi_\gt^\md$ has no jumps around the point $\zs$, and as such $\zs$ is an isolated singularity of $\bm \Phi^\md$. Now, sending $\zeta\to \zs$, say along $\overline{\Omega_5^\zs}\cap(\zs,\zs+\varepsilon)$, we see from RHP~\ref{rhp:modelPhi}--(4) that this isolated singularity is removable. Thus, $\bm \Phi_\gt^\md$ is analytic near $\zs$, proving (i).

Next, RHP~\ref{rhp:modelPhi}--(2) shows that $\bm \Phi_\gt$ has an analytic continuation from the lower half plane to any interval $J\subset \R$, as long as the functions $\zeta\mapsto \msf P_\gt(\zeta)$ and $\zeta\mapsto 1/(1+\ee^{\gt{\msf P_\gt(\zeta)}})$ are analytic in a neighborhood of this interval $J$. From the definition of admissibility and \eqref{eq:PtauFtau}, we see that these latter functions are analytic on intervals of size $\Boh(\gt^2)$. This argument suffices to show that for some $\kappa>0$, $\mcal H_\gt$ is real-analytic in $(-\kappa\gt^2,\kappa \gt^2)$.

Equation~\eqref{eq:bmHmdPhi} for $x>\zs$ follows directly from \eqref{deff:Phimod}, \eqref{deff:bmHtau}, once one notices that for $x>\zs$, we have $x-\ii \eta\in \Omega^\zs_5$ for $\eta>0$ sufficiently small.

Now, for $x<\zs$, we have $x-\ii\eta\in \Omega^\zs_3$ for $\eta>0$ sufficiently small, and from the very definition of $\bm \Phi_\gt^\md$ in $\Omega_3^\zs$ we write
\begin{multline*}
\left( \bm I+\frac{\bm E_{12}}{1+\ee^{\gt \msf P_\gt}} \right)\bm \Delta_\zeta \bm \Phi_{\gt,-}\left( \bm I-\frac{\bm E_{12}}{1+\ee^{\gt \msf P_\gt}} \right) =
(\ee^{\gt \msf P_\gt(\zeta)})' \left[ \bm E_{21}+\frac{\bm \sp_3}{1+\ee^{\gt\msf P_\gt}}-\frac{\bm E_{12}}{(1+\ee^{\gt\msf P_\gt})^2} \right] \\
+
\left( \bm I+\frac{\bm E_{12}}{1+\ee^{\gt \msf P_\gt}} \right)\left( \bm I-(1+\ee^{\gt\msf P_\gt})\bm E_{21} \right)\bm \Delta_\zeta \bm \Phi_{\gt,-}^\md\left( \bm I+(1+\ee^{\gt\msf P_\gt})\bm E_{21} \right)\left( \bm I-\frac{\bm E_{12}}{1+\ee^{\gt \msf P_\gt}} \right).
\end{multline*}
Taking the $(1,2)$-entry, and using the algebraic identity
\begin{equation}\label{eq:coolmatrixidentity}
\left[\left(\bm I+\frac{1}{\beta} \bm E_{12}\right)\left(\bm I-\beta \bm E_{21}\right)M \left(\bm I+\beta \bm E_{21}\right)\left(\bm I-\frac{1}{\beta} \bm E_{12}\right) \right]_{12}=
-\frac{1}{\beta^2}M_{21},
\end{equation}
which is valid for any $2\times 2$ matrix $M$ and any scalar $\beta\neq 0$, we obtain
\begin{equation*}
\left[\left( \bm I+\frac{\bm E_{12}}{1+\ee^{\gt \msf P_\gt}} \right)\bm \Delta_\zeta \bm \Phi_{\gt,-}\left( \bm I-\frac{\bm E_{12}}{1+\ee^{\gt \msf P_\gt}} \right) \right]_{12}=
-\frac{(\ee^{\gt \msf P_\gt})'}{(1+\ee^{\gt\msf P_\gt})^2}-\frac{1}{(1+\ee^{\gt\msf P_\gt})^2}\left[ \Delta_\zeta \bm \Phi_{\gt,-}^\md \right]_{21},
\end{equation*}
which is equivalent to the second line of \eqref{eq:bmHmdPhi}.

\end{proof}

The solution $\bm \Phi_\gt=\bm \Phi_\gt(\cdot \mid \msf P,\msf F,\lambda,s,u)$ depends on several input data. In what follows, we will carry out its asymptotic analysis under the following assumption.

\begin{assumption}\label{assumptions:modelcriticalregime}
We assume
\begin{enumerate}[(a)]
\item $(\msf P,\msf F)$ is a pair of fixed admissible functions.
\item $s$ is allowed to vary with $\gt$, but within the regime \eqref{eq:critregime}.
\item $u$ is allowed to vary with $\gt$, but within compacts of $(0,+\infty)$.
\item The choice $\lambda=\zeta_\gt(s)$, with $\zeta_\gt(s)$ as in Proposition~\ref{prop:zeroPt}. 
\end{enumerate}
\end{assumption}

\begin{remark}\label{rmk:intdiffPII}
It is instructive to keep in mind the approximations \eqref{eq:TaylorExpPgtFgt}, which yield in essence that
\begin{equation}\label{eq:expsigmaPsigmaF}
\frac{1}{1+\ee^{\gt \msf P_\gt(\zeta)}}\approx \frac{1}{1+\ee^{\gt \msf c_{\msf P}\zeta-s}}\quad \text{and}\quad \frac{1}{1-\ee^{\pm \ii \gt\msf F_\gt(\zeta)}}\approx \frac{1}{1-\ee^{\pm \ii u \gt^3\pm \ii \gt \msf c_{\msf F}\zeta}}.
\end{equation}

To illustrate our next steps, let us formally replace these approximations in \eqref{eq:jumpJt}. Recalling the sets $\Omega_j^\lambda$ from \eqref{deff:Gammalambda}, if one transforms
\begin{equation}\label{eq:formaltildePsi0}
\widetilde{\bm \Phi}_\gt(\zeta)=
\begin{dcases}
    \bm\Phi_\gt(\zeta)\left(\bm I+\frac{1}{1+\ee^{\gt \msf c_{\msf P}\zeta-s}}\frac{1}{1-\ee^{- \ii u \gt^3- \ii \gt \msf c_{\msf F}\zeta}}\bm E_{12}\right), & z\in \Omega_1^{\zeta_\gt(s)}\\
    \bm\Phi_\gt(\zeta)\left(\bm I-\frac{1}{1+\ee^{\gt \msf c_{\msf P}\zeta-s}}\frac{1}{1-\ee^{\ii u \gt^3+ \ii \gt \msf c_{\msf F}\zeta}}\bm E_{12}\right), & z\in \Omega_4^{\zeta_\gt(s)} \\
    \bm \Phi_\gt(\zeta), & \text{elsewhere},
\end{dcases}
\end{equation}
then a direct calculation shows that $\widetilde{\bm \Phi}_\gt$ has no jumps along $\Gamma^\lambda_1\cup \Gamma^{\lambda}_5$, and in the remaining parts of $\Gamma^\lambda$ the jump matrix for $\widetilde{\bm \Phi}_\gt$ is
\begin{equation}\label{eq:formaltildePsi1}
\begin{dcases}
\bm I+\frac{1}{1+\ee^{\gt \msf c_{\msf P}\zeta-s}}\bm E_{12}, & \zeta\in \Gamma_0^{\lambda}, \\
\bm I+(1+\ee^{\gt \msf c_{\msf P}\zeta-s})\bm E_{21}, & \zeta\in \Gamma_2^\lambda\cup \Gamma_4^\lambda, \\
\frac{1}{1+\ee^{\gt \msf c_{\msf P}\zeta-s}}\bm E_{12}-(1+\ee^{\gt \msf c_{\msf P}\zeta-s})\bm E_{21}, & \zeta\in \Gamma_3^\lambda.
\end{dcases}
\end{equation}
Also formally ignoring possible poles of the transformation factors in \eqref{eq:formaltildePsi0}, the transformation $\bm \Phi_\gt\mapsto \wt{\bm \Phi}_\gt$ does not affect the behavior near $\infty$, that is, we still have
\begin{equation}\label{eq:formaltildePsi2}
\wt{\bm \Phi}_\gt(\zeta)=\left(\bm I+\Boh(\zeta^{-1})\right)\zeta^{-\sp_3/4}\bm U_0\ee^{-\frac{2}{3}\zeta^{3/2}\sp_3},\quad \zeta\to \infty.
\end{equation}

Conditions \eqref{eq:formaltildePsi1} and \eqref{eq:formaltildePsi2} would then say that $\wt{\bm \Phi}_\gt$ is precisely the solution to the RHP associated to the integro-differential Painlevé II equation. This RHP first appeared in \cite{CafassoClaeys2021}, in connection with the narrow wedge solution to the KPZ equation, but was realized to be connected to the integro-differential Painlevé II in \cite{CafassoClaeysRuzza2021}, see also \cite[Sections~4 and 5]{GhosalSilva2023} for the formulation as it appears here.

However, the transformation \eqref{eq:formaltildePsi0} involves factors with poles: the factor $1/(1+\ee^{\gt \msf c_{\msf P}\zeta-s})$ has infinitely many poles at the points $\zeta=(s+(2k+1)\pi \ii)/\gt$, $k\in \Z$, which are all located in the regions $\Gamma_1^{\zeta_\gt(s)}$ and $\Gamma_4^{\zeta_\gt(s)}$ where the transformation \eqref{eq:formaltildePsi0} is nontrivial. In other words, the matrix $\wt{\bm \Phi}_\gt$ has now poles at these points $\zeta=(s+(2k+1)\pi \ii)/\gt$, and a proper RHP for it would have to involve a pole condition at these points, and therefore $\wt{\bm \Phi}_\gt$ is not rigorously comparable with the int-diff PII RHP from the cited works. 
\end{remark}

Nevertheless, it will turn out that $\bm \Phi_\gt$ is asymptotically comparable with the int-diff PII. Our next goal is to extract asymptotics of $\bm \Phi_\gt$ as $\gt\to \infty$. Such asymptotics does depend on how $s$ grows with $\gt$, and will be further split into three distinct asymptotic regimes, depending on whether $s/\gt$ remains bounded, it is very positive or it is very negative. Each of these three regimes will require a separate asymptotic analysis, and it will turn out that in each of them the asymptotic behavior of $\bm \Phi_\gt$ has essentially two nontrivial terms, one in direct correspondence with the analogue asymptotic regime of the int-diff PII RHP from \cite{CafassoClaeysRuzza2021}, and one corresponding to the presence of poles of the factor $(1+\ee^{\gt \msf c_{\msf P}\zeta-s})^{-1}$.

We now move on to the mentioned asymptotic analysis.

\section{Analysis of related RHPs}\label{sec:relatedRHPs}

The asymptotic analysis of the RHP~\ref{rhp:modelPhi}, under Assumptions~\ref{assumptions:modelcriticalregime}, will display three distinct behaviors, depending on how $s$ grows with $\gt$.

\begin{definition}\label{deff:subregimescritical}
The regime \eqref{eq:critregime} splits in the following subregimes.
\begin{enumerate}
\item {\it Subcritical regime}: for any $T>0$ {\it sufficiently large}, but fixed, we assume
$$
 T \gt \leq s \leq \frac{1}{T}\gt^{3/2}
$$

\item {\it Critical regime:} for any $T>0$ fixed, we assume
$$
-T \gt\leq s \leq  T\gt .
$$

\item {\it Supercritical regime:} for any $T>0$ {\it sufficiently large}, but fixed, we assume
$$
-\frac{1}{T}\gt^{3/2}\leq s \leq -T\gt.
$$
\end{enumerate}
\end{definition}

With the identification $\gt=\gn^{1/3}$, the subcritical, critical, and supercritical regimes introduced in Definition~\ref{deff:subregimescritical} coincide with the regimes in Assumptions~\ref{assumpt:parameterregimes}. 

The three different subregimes specified in Definition~\ref{deff:subregimescritical} are somewhat explained in terms of the parameter $\zeta_\gt(s)$ introduced in Proposition~\ref{prop:zeroPt}: the regimes (1), (2) and (3) correspond to whether $\zeta_\gt(s)\gg 1$, $\zeta_\gt(s)=\Boh(1)$ or $\zeta_\gt(s)\ll -1$ as $\gt\to \infty$.

In the next section we carry out the asymptotic analysis in each of these regimes. The phenomena that we will observe in each regime may be explained replacing $\bm \Phi_\gt$ by $\wt{\bm \Phi}_\gt$ from \eqref{eq:formaltildePsi0} and observing the following. In the critical regime, the approximation \eqref{eq:formaltildePsi1} suggests that as $\gt \to \infty$,
$$
\bm J_{\bm \gt}\approx 
\begin{cases}
    \bm I, & \zeta \in \Gamma_0^{\zeta_\gt(s)}=(\zeta_\gt(s),+\infty), \\
    \bm I + \bm E_{21}, & \zeta \in \Gamma_2^{\zeta_\gt(s)}\cup \Gamma_4^{\zeta_\gt(s)}, \\ 
    \bm E_{12}-\bm E_{21}, & \zeta\in \Gamma_3^{\zeta_\gt(s)}=(-\infty,-\zeta_\gt(s)).
\end{cases}
$$
This approximate jump matrix and asymptotic behavior of $\bm \Phi_\gt$ as $\zeta\to \infty$ coincides with the RHP for the Painlevé XXXIV equation. As we will see, in the critical regime the PXXXIV RHP indeed provides a good approximation to $\bm \Phi_\gt$ away from the point $\zeta_\gt(s)$. Near the point $\zeta_\gt(s)$, however, we have to devise a local parametrix which, in essence, creates a good approximation to the pole behavior that emerges from the transformation $\bm \Phi_\gt\mapsto \wt{\bm \Phi}_\gt$. When performing this approximation, the Painlevé variable $y$ corresponds to $y=\zeta_\gt(s)\approx \frac{s}{\msf c_{\msf P}\gt}$.

When moving from the critical to, say, the subcritical regime, the Painlevé variable $y=\zeta_\gt(s)$ becomes very large. Rather than still using the PXXXIV RHP as a good approximation for $\bm \Phi_\gt$, for technical reasons we have to use instead a RHP related to the integro-differential Painlevé II equation, which in turn in the same asymptotic regime is close to the Airy RHP. Once again, such global parametrix will provide a good approximation for $\bm \Phi_\gt$ only away from (a scaling of) $\zs$, and near this point we have again to construct a local parametrix to essentially handle the poles coming from the factors $(1+\ee^{\gt\msf P_\gt})^{-1}$ and $1/(1+\ee^{\pm \ii \gt \msf F_\gt})^{-1}$ present in the jump. In the supercritical regime, we use the global parametrix for PXXXIV in the regime $y\to -\infty$, and locally near $\zs$ we have to use a RHP obtained from a  modification of the model problem itself, whose study we postpone to Section~\ref{sec:modRHPmodel}.

The main goal of the current section is to discuss the RHPs just mentioned, obtaining properties of them that will be needed in the analysis of the model problem.

\subsection{The Painlevé XXXIV RHP}\label{sec:PXXXIV}\hfill 

The discussion on the Painlevé XXXIV RHP follows \cite{XiaoBoXuZhao2018} closely.

Recall that
$$
\Gamma_2^0= (\infty \ee^{2\pi \ii/3},0],\quad \Gamma_3^0= (-\infty,0],\quad \Gamma_4^0= (\infty \ee^{4\pi \ii/3},0],\quad  \text{and set}\quad \Gamma_{\bm \Psi}\deff \Gamma_2^0\cup\Gamma_3^0\cup\Gamma_4^0.
$$
The three rays $\Gamma_2^0$, $\Gamma_3^0$ and $\Gamma_4^0$ are oriented from $\infty$ to the origin. Also, denote 
\begin{equation}\label{deff:sectorsOmegaj}
\begin{aligned}
& \Omega_0\deff \{\zeta\in \C\mid |\arg \zeta|<2\pi/3\},\\
& \Omega_1\deff \{\zeta\in \C\mid 2\pi/3<\arg\zeta<\pi\}=\Omega_2^0, \\
& \Omega_2\deff \{\zeta\in \C\mid -\pi<\arg\zeta<-2\pi/3\}=\Omega_3^0.
\end{aligned}
\end{equation}
Consider the following RHP.
\begin{rhp}\label{rhp:PXXXIV}
Find a $2\times 2$ matrix-valued function $\bm \Psi$ with the following properties.
\begin{enumerate}[(1)]
\item ${\bm \Psi}$ is analytic on $\C\setminus \Gamma_{\bm \Psi}$.
\item ${\bm \Psi}_{+}(\zeta)=\bm\Psi_{-}(\zeta)\bm J_{0}(\zeta)$, $\zeta\in \Gamma_{\bm \Psi}$, where 
\begin{equation}\label{eq:deffJPXXXIV}
\bm J_0(\zeta)\deff
\begin{cases}
\bm I+\bm E_{21}, & \zeta\in \Gamma_2^0\cup \Gamma_4^0, \\
\bm E_{12}-\bm E_{21}, & \zeta\in \Gamma_3^0.
\end{cases}
\end{equation}
\item As $\zeta\to \infty$, 
\begin{equation}\label{eq:asymptPXXXIVinfinity}
\bm\Psi(\zeta)=(\bm I+\ii q\bm E_{21})\left(\bm I+\frac{\bm\Psi_1^{(\infty)}}{\zeta}+\Boh(\zeta^{-2})\right)\zeta^{-\sp_3/4}\bm U_0\ee^{-(\frac{2}{3}\zeta^{3/2}+y\zeta^{1/2})\sp_3},
\end{equation}
where $q\in \C$ is a free parameter independent of $\zeta$ and that will be chosen appropriately in a moment, $y\in \C$ is the Painlevé variable, and $\bm\Psi_1^{(\infty)}=\bm\Psi_1^{(\infty)}(y)$ is a $2\times 2$ matrix independent of $\zeta$.
\item There exists a matrix-valued function $\bm\Psi_0$ which is analytic on a disk $D_\varepsilon(0)$ centered at the origin, and constant matrices $\bm M_j$, $j=0,1,2$, for which
\begin{equation}\label{eq:localbehP34RHP}
\bm\Psi(\zeta)=\bm\Psi_0(\zeta)\left(\bm I+\frac{1}{2\pi \ii }\log(\zeta)\bm E_{12}\right)\bm M_j,\quad \zeta\in D_\varepsilon(0)\cap \Omega_j.
\end{equation}
\end{enumerate}
\end{rhp}

This is the PXXXIV RHP from \cite[Section~2]{XiaoBoXuZhao2018}, with monodromy data
$$
\alpha=\omega =0.
$$
The solution $\bm\Psi$ and also the term $\bm\Psi_0$ both depend on $y$ as well, and when needed we write $\bm\Psi(\zeta)=\bm\Psi(\zeta\mid y)$ and $\bm\Psi_0(\zeta)=\bm\Psi_0(\zeta\mid y)$ etc. In particular, \cite[Proposition~1]{XiaoBoXuZhao2018} ensures that this RHP is solvable for every $y\in \R$. From standard methods, one shows that $y\mapsto \bm \Psi(\zeta\mid y)$ is continuous, hence bounded for $y$ within compacts of the real line and $\zeta$ within compacts of $\C\setminus \Gamma_{\bm \Psi}$. Likewise, the matrix $\bm \Psi_0$ remains bounded for $y$ in compacts of the real line and $\zeta$ in a small neighborhood of the origin.

The function $u=u(y)$ defined by 
$$
u(y)\deff -\frac{y}{2}-\ii \partial_y (\bm\Psi^{(\infty)}_1(y))_{12},
$$
satisfies the PXXXIV equation
\begin{equation}\label{eq:PXXXIVparticular}
\partial_{yy}u=4u^2+2yu+\frac{(\partial_yu)^2}{2u}.
\end{equation}
Alternatively, $u(y)$ can also be related to $\bm\Psi$ through the limit
$$
u(y)=\ii \lim_{\zeta\to 0} \zeta\left[\bm\Psi'(\zeta\mid y)\bm\Psi(\zeta\mid y)^{-1}\right]_{12},
$$
where here and from now on, the symbol $'$ denotes derivative with respect to the (spectral) variable $\zeta$. This solution has no poles on the real axis and has asymptotic behavior
\begin{equation}\label{eq:asymptPXXXIV}
u(y)=\frac{\ee^{-\frac{4}{3}y^{3/2}}}{4\pi y^{1/2}}\left(1+\Boh(y^{-1/4})\right),\quad y\to +\infty,\qquad \text{and}\qquad u(y)=-\frac{y}{2}-\frac{1}{8y^2}+\Boh(y^{-7/2}), \quad y\to -\infty.
\end{equation}
Such asymptotics are stated and proven as above in \cite[Theorem~2]{XiaoBoXuZhao2018}, but follow closely the similar calculations from \cite{ItsKuijlaarsOstensson2009} performed for slightly different monodromy data.

The solution $u(y)$ is classically related to a non-homogeneous Painlevé II equation. With the change of variables $U(y)=2^{1/3}u(-2^{-1/3}y)$, the matrix $U$ takes the form
$$
U(y)=\partial_y p(y)+p(y)^2+\frac{y}{2},
\quad \text{where } p \text{ solves the Painlevé II equation}\quad 
\partial_{yy}p=yp+2p^3-\frac{1}{2}.
$$

This Painlevé II equation is a so-called inhomogeneous one, with inhmogeneity parameter $1/2$. It is however possible to relate $u$ to the Hastings-McLeod solution to the homogeneous Painlevé II equation, as we now explain.

From \eqref{eq:asymptPXXXIV} we see that $u(y)$ is strictly positive for large positive values of $y$. Therefore, also for large values of $y>0$, we may define a function $p(y)$ by the relation
\begin{equation}\label{eq:PIIPXXXIV}
p(y)^2=u(y),\quad \text{with asymptotics}\quad p(y)=\frac{\ee^{-\frac{2}{3}y^{3/2}}}{2\sqrt{\pi}y^{1/4}}\left(1+\Boh(y^{-1/4})\right),\quad y\to +\infty.
\end{equation}

We differentiate the defining relation $p^2=u$ of $p$ and substitute into the PXXXIV Equation \eqref{eq:PXXXIVparticular}. A direct calculation then reduces \eqref{eq:PXXXIVparticular} to
$$
\partial_{yy}p=2p^3+yp.
$$
This is precisely the homogeneous PII equation, and the asymptotics \eqref{eq:PIIPXXXIV} shows that $p(y)$ is in fact the celebrated Hastings-McLeod solution to PII. 

We need additional information about the RHP above. First of all, the matrices in $\bm\Psi.4$ are determined from
\begin{equation}\label{eq:PXXXIVmon}
\bm M_1=\bm I-\bm E_{21},\quad \bm M_0=\bm M_1(\bm I+\bm E_{21})=\bm I,\quad \bm M_2=\bm M_0(\bm I+\bm E_{21})=\bm I+\bm E_{21}.
\end{equation}
Hence,
$$
\bm M_0=\bm I,\quad \bm M_2=\bm I+\bm E_{21}.
$$
Second, the keen reader will observe that the PXXXIV above does not depend on the factor $q$ appearing in $\bm\Psi{\bm .3}$. Nevertheless, for any such choice of $q$, possibly depending on $y$, the relations between $u$ and $\bm\Psi$ that we just described still hold true. For a particular choice of $q$ in $\bm\Psi{\bm .3},$ it is shown in \cite[Equation~(45) {\it et seq.}]{XiaoBoXuZhao2018} that $\bm\Psi$ is a solution to the Lax pair system
\begin{equation}\label{eq:LaxPairPXXXIV}
\bm\Psi'(\zeta\mid y)=\bm A(\zeta\mid y)\bm\Psi(\zeta\mid y),\quad \partial_y\bm\Psi(\zeta\mid y)=\bm B(\zeta\mid y)\bm\Psi(\zeta\mid y),
\end{equation}
with
\begin{equation*}
\begin{aligned}
& \bm A(\zeta\mid y)\deff \frac{\bm A_{-1}(y)}{\zeta}+\bm A_0(y)-\ii\zeta \bm E_{21}, \\
& \bm A_{-1}(y)\deff \frac{\partial_y u(y)}{2}\sp_3 -\ii u(y)\bm E_{12}-\frac{\ii \partial_y u(y)^2}{4u(s)} \bm E_{21}, \\
& \bm A_0(y)\deff \ii \bm E_{12}-\ii (u(y)+y)\bm E_{21}, 
\end{aligned} \qquad\qquad 
\begin{aligned}
&\bm B(\zeta\mid y)\deff \bm B_0(y)-\ii \zeta \bm E_{21}, \\
&\bm B_0(y)\deff \ii \bm E_{12}-i(2u(y)+y)\bm E_{21}.
\end{aligned}
\end{equation*}

The relevant function $q$ for obtaining this Lax pair is given in terms of $u$, namely satisfying the relation 
\begin{equation}\label{eq:relationquPXXXIV}
2u(y)=-y-2\partial_y q(y).
\end{equation}
Hence, for us $q$ depends on the Painlevé variable $y$, and from now on we always assume this choice and write $q=q(y)$. Nevertheless, the specific form of $q$ will not be important, only the fact that it leads to the Lax pair as stated. Using this function $q(y)$, for later reference it is also useful to express \eqref{eq:asymptPXXXIVinfinity} as
\begin{equation}\label{eq:asymptPXXXIVinfinity2}
\bm\Psi(\zeta)=\left(\bm I+\ii \left(q(y) +\frac{y^2}{4} \right)\bm E_{21}\right)\left(\bm I+\Boh(\zeta^{-1})\right)\zeta^{-\sp_3/4}\bm U_0\ee^{-\frac{2}{3}(\zeta+y)^{3/2}\sp_3},\quad \zeta\to \infty.
\end{equation}

We need some additional relation between the solution $u$ and the solution $\bm\Psi$ that we state as a separate result. For its statement, we introduce
\begin{equation}\label{deff:mcalHPXXXIV}
\mcal H(\zeta)=\mcal H(\zeta\mid y)\deff
\begin{cases}
\left[\bm\Delta_\zeta \bm\Psi(\zeta\mid y)\right]_{21}, & \zeta >0, \\
[(\bm I+\bm E_{21})\bm\Delta_\zeta \bm\Psi_-(\zeta\mid y)(\bm I-\bm E_{21})]_{21}, & \zeta<0,
\end{cases} 
\end{equation}
This function is the analogue of $\mcal H_\gt$ which we introduced in terms of the model problem in \eqref{deff:bmHtau}.

\begin{prop}\label{prop:RHPlimitPXXXIV}
The function $\zeta\mapsto \mcal H(\zeta)$ extends to an analytic function for $\zeta \in \C$ in a neighborhood of the origin, which satisfies
$$
\mcal H(0\mid y)=\pi \ii (\partial_{yy}u(y)-6u(y)^2-4yu(y)).
$$
Furthermore, the following estimates hold true. 

\begin{enumerate}[(i)]
    \item Given $y_0>0$ and $\zeta_0>0$, there exists $M>0$ such that
\begin{equation}\label{eq:boundPXXXIVHder}
|\partial_\zeta\mcal H(\zeta)|\leq M
\end{equation}
for $|\zeta|\leq \zeta_0$ and $|y|\leq y_0$.

\item Given $y_0>0$ there exist $\eta>0$ and $M>0$ such that
\begin{equation}\label{eq:boundPXXXIVH1}
|\mcal H(\zeta)|\leq M \ee^{-\eta \zeta^{3/2}-2y\zeta^{1/2}}
\end{equation}
for $|y|\leq y_0$ and $\zeta\geq 0$.
\item Given $y_0>0$ there exist $\eta>0$ and $M>0$ such that
\begin{equation}\label{eq:boundPXXXIVH2}
|\mcal H(\zeta)|\leq M(1+|\zeta|^{3/2}) 
\end{equation}
for $|y|\leq y_0$ and $\zeta\leq  0$.
\end{enumerate}

\end{prop}

\begin{proof}
The matrix $\bm \Psi_0(\zeta)=\bm \Psi_0(\zeta\mid u)$ is the matrix from \eqref{eq:localbehP34RHP} which is complex-analytic in a neighborhood of $\zeta=0$. Furthermore, from the general theory of RHPs this matrix is also smooth in $y$. Let us expand it in power series near $\zeta=0$,
$$
\bm\Psi_0(\zeta)=\bm\Psi_0^{(0)}+\bm\Psi_1^{(0)}\zeta+\Boh(\zeta^2),\quad \zeta\to 0.
$$
Because $\det\bm\Psi_0 \equiv 1$, we must have $\det\bm\Psi_0^{(0)}\equiv 1$ as well, and both of these matrices are invertible, with inverses which are smooth functions of $y$.

Using that $\bm M_0=\bm I$ (see \eqref{eq:PXXXIVmon}), we compute from RHP~\ref{rhp:PXXXIV}--(4) that for $\zeta>0$,
\begin{equation}\label{eq:limitPXXXIVaux1b}
\mcal H(\zeta\mid y)=\left[\bm \Psi(\zeta\mid y)^{-1}\bm \Psi'(\zeta\mid y)\right]_{21}=\left[\bm \Psi_0(\zeta\mid y)^{-1}\bm \Psi_0'(\zeta\mid y)\right]_{21},
\end{equation}
and therefore
\begin{equation}\label{eq:limitPXXXIVaux1}
\lim_{\zeta\searrow 0}\mcal H(\zeta\mid y)=\left[ \left(\bm\Psi_0^{(0)}\right)^{-1}\bm\Psi_1^{(0)}\right]_{21}.
\end{equation}
This shows that $\zeta\mapsto \mcal H(\zeta\mid y)$ extends analytically to $\zeta=0$ from the right.
In a similar way, since $\bm M_2=\bm I+\bm E_{21}$, we use RHP~\ref{rhp:PXXXIV}--(4) again to get for $\zeta<0$ that
\begin{multline*}
\mcal H(\zeta\mid y)=[(\bm I+\bm E_{21})\bm\Psi_-(\zeta\mid y)^{-1}\bm\Psi'_-(\zeta\mid y)(\bm I-\bm E_{21})]_{21} \\ 
\begin{aligned}
 & =  \left[\left(\bm I-\frac{1}{2\pi\ii }(\log\zeta)_-\bm E_{12}\right)\bm\Psi_0(\zeta\mid y)^{-1}\left(\bm\Psi_0(\zeta\mid y)(\bm I+\frac{1}{2\pi \ii }(\log\zeta)_-\bm E_{12})\right)'\right]_{21}
 & = \left[ \bm\Psi_0(\zeta\mid y)^{-1}\bm\Psi'_0(\zeta\mid y)  \right]_{21}.
\end{aligned}
\end{multline*}
The right-hand side of this expression is analytic in $\zeta\in \C$ and smooth in $y$, and agrees with \eqref{eq:limitPXXXIVaux1b}. This shows that $\mcal H(\zeta\mid y)$ is analytic in $\zeta$ and smooth in $y$, and also shows the bound \eqref{eq:boundPXXXIVHder}. 

The bounds \eqref{eq:boundPXXXIVH1} and \eqref{eq:boundPXXXIVH2} follow from the continuity of $\mcal H$ as a function of $\zeta\in \R$ and $y\in \R$ that we just established, together with the asymptotics \eqref{eq:asymptPXXXIVinfinity}.

It remains to compute $\mcal H(0\mid y)$, which we do so using \eqref{eq:limitPXXXIVaux1} and with the help of the Lax pair \eqref{eq:LaxPairPXXXIV}. The identity $\bm\Psi'=\bm A\bm\Psi$ yields
$$
\bm\Psi_0'(\zeta\mid y)\left(\bm I+\frac{\log \zeta}{2\pi \ii}\bm E_{12}\right)+\frac{1}{2\pi \ii \zeta}\bm\Psi_0(\zeta\mid y)\bm E_{12}=\bm A(\zeta\mid y)\bm\Psi_0(\zeta\mid y)\left(\bm I+\frac{\log \zeta}{2\pi \ii}\bm E_{12}\right).
$$
Expanding both sides as $\zeta\to 0$, the terms of order $\zeta^{-1}$ and $\zeta^0$ yield the relations
\begin{equation}\label{eq:Am1aux3}
\bm\Psi_0^{(0)}\bm E_{12}=2\pi \ii \bm A_{-1}\bm\Psi_0^{(0)}\quad \text{and}\quad \bm\Psi_1^{(0)}\left(\bm I+\frac{1}{2\pi\ii}\bm E_{12}\right)=\bm A_0\bm\Psi_0^{(0)}+\bm A_{-1}\bm\Psi_1^{(0)},
\end{equation}
which are identities between functions of the variable $y$.
The first equation gives the relation 
$$
(\bm\Psi_0^{(0)})^{-1}\bm A_{-1}=\frac{1}{2\pi \ii}\bm E_{12}(\bm\Psi_0^{(0)})^{-1}.
$$
Multiplying the second equation in \eqref{eq:Am1aux3} by $(\bm\Psi_0^{(0)})^{-1}$ to the left and using this previous identity, we obtain
\begin{equation}\label{eq:Am1aux4}
\lim_{\zeta\nearrow 0}[(\bm I+\bm E_{21})\bm\Psi_-(\zeta\mid y)^{-1}\bm\Psi'_-(\zeta\mid y)(\bm I-\bm E_{21})]_{21} =[(\bm\Psi_0^{(0)})^{-1}\bm\Psi_1^{(0)}]_{21}=
[(\bm\Psi_0^{(0)})^{-1}\bm A_0\bm\Psi_0^{(0)}]_{21}.
\end{equation}

Our next goal is to express the entries of $\bm\Psi_0^{(0)}$ in terms of the entries of the $\bm A_k$'s. The first row of the first equation in \eqref{eq:Am1aux3} yields the relations
\begin{equation}\label{eq:Am1aux2}
(\bm A_{-1})_{11}(\bm\Psi_0^{(0)})_{11}+(\bm A_{-1})_{12}(\bm\Psi^{(0)}_0)_{21}=0,\quad (\bm A_{-1})_{11} \bm (\bm\Psi_0^{(0)})_{12}+(\bm A_{-1})_{12}(\bm\Psi_0^{(0)})_{22}=\frac{1}{2\pi \ii}(\bm\Psi_0^{(0)})_{11}.
\end{equation}
We see these relations as equations on the unknows $(\bm\Psi_0^{(0)})_{jk}$. Let us assume for a moment that $(\bm A_{-1})_{12}=-\ii u(y)\neq 0$, which is true for all but a discrete set of $y\in \R$ because $u$ is analytic in $y$. Solving the first equation for $(\bm\Psi_0^{(0)})_{21}$ we obtain
\begin{equation}\label{eq:Am1aux}
(\bm\Psi_0^{(0)})_{21}=-(\bm\Psi_0^{(0)})_{11}\frac{(\bm A_{-1})_{11}}{(\bm A_{-1})_{12}}.
\end{equation}
Also, the second equation writes as
$$
(\bm\Psi_0^{(0)})_{11}=2\pi \ii(\bm A_{-1})_{12}\left((\bm\Psi_0^{(0)})_{22}+\frac{(\bm A_{-1})_{11}}{(\bm A_{-1})_{12}}(\bm\Psi_0^{(0)})_{12}\right).
$$
On the other hand, the condition $\det \bm\Psi_0^{(0)}=1$ combined with \eqref{eq:Am1aux} yields the relation
$$
(\bm\Psi_0^{(0)})_{11}\left((\bm\Psi_0^{(0)})_{22}+\frac{(\bm A_{-1})_{11}}{(\bm A_{-1})_{12}}(\bm\Psi_0^{(0)})_{12}\right)=1.
$$
Thus, we concluded the relations
$$
\left((\bm\Psi_0^{(0)})_{11}\right)^2=2\pi \ii {(\bm A_{-1}})_{12}\quad \text{and}\quad \left((\bm\Psi_0^{(0)})_{21}\right)^2=\left((\bm\Psi_0^{(0)})_{11}\frac{(\bm A_{-1})_{11}}{(\bm A_{-1})_{12}}\right)^2=2\pi \ii\frac{\left((\bm A_{-1})_{11}\right)^2}{(\bm A_{-1})_{12}}.
$$
With the help of these equations we express the right-hand side of \eqref{eq:Am1aux4} solely in terms of the entries of $\bm A_0$ and $\bm A_{-1}$, which leads to the claimed result.
\end{proof}

It is instructive to interpret $\mcal H(\zeta\mid y)$ in terms of a correlation kernel. For that, introduce
\begin{multline}\label{eq:P34kernelRHP}
\msf K_\ptf(\zeta,\xi)=\msf K_\ptf(\zeta,\xi\mid y) \\ \deff \frac{1}{2\pi\ii}\frac{1}{\zeta-\xi}
    \left[ (\bm I+\chi_{(-\infty,0)}(\xi)\bm E_{21})\bm \Psi_-(\xi\mid y)^{-1}\bm\Psi_-(\zeta\mid y)(\bm I-\chi_{(-\infty,0)}(\zeta)\bm E_{21}) \right]_{21},\quad \zeta,\xi\in \R,\quad \zeta\neq \xi.
\end{multline}
Extending $\msf K_\ptf$ to the diagonal $\zeta=\xi$ by continuity, we then obtain that
\begin{equation}\label{eq:identitymcalHKptf}
\mcal H(\zeta\mid y)=2\pi \ii\, \msf K_\ptf(\zeta,\zeta\mid y),\quad \zeta\in \R,\quad y\in \R.
\end{equation}
Proposition~\ref{prop:RHPlimitPXXXIV} then shows that
\begin{equation*}
\msf H(y)\deff \msf K_{\ptf}(0,0\mid y)=\frac{1}{2}\partial_{yy} u(y)-3u(y)^2-2y u(y)
\end{equation*}
in terms of the Painlevé XXXIV transcendent $u(y)$, which is the same as \eqref{eq:KptfPXXXIV}. Substituting $u=p^2$, we obtain 
\begin{equation}\label{eq:hamiltonianPII}
\msf H(y)=(\partial_y p(y))^2-p(y)^4-yp(y)^2,
\end{equation}
which is the Hamiltonian of the PII equation. A direct calculation using the PII equation shows that this Hamiltonian satisfies $\partial_y \msf H(y)=-p(y)^2$, from which it follows in particular that the Tracy-Widom distribution is recovered from it through
\begin{equation}\label{eq:FGUEHamilt}
\partial_y \log F_{\rm GUE}(y)= \msf H(y).
\end{equation}

To finish this section, we now show the characterization of $\msf K_{\ptf}$ claimed in \eqref{deff:P34kernel}. For that, let us define
$$
\bm\Psi^\md(\zeta)\deff 
\begin{cases}
\bm \Psi(\zeta),& \zeta\in \Omega_0,\\
\bm \Psi(\zeta)(\bm I+\bm E_{21}), & \zeta\in \Omega_1, \\ 
\bm\Psi(\zeta)(\bm I-\bm E_{}21), & \zeta\in \Omega_2.
\end{cases}
$$

Then $\bm \Psi^\md$ is analytic on $\C\setminus (-\infty,0]$, it has jump
$$
\bm \Psi_+^\md(\zeta)=\bm\Psi^\md_-(\zeta)(\bm I+\bm E_{12}),\quad \zeta<0,
$$
and it coincides with $\bm \Psi$ on $\Omega_0$. In particular, the asymptotic behavior of $\bm\Psi^\md$ along the positive real axis is also given by the right-hand side of \eqref{eq:asymptPXXXIVinfinity}, from which we conclude that
\begin{equation}\label{eq:psiP34asympt}
\psi(\zeta,y)\deff \frac{1}{\sqrt{2\pi}}\left(\bm \Psi^\md(\zeta\mid y)\right)_{11}=(1+\Boh(\zeta^{-1}))\frac{\zeta^{-1/4}}{2\sqrt{\pi}}\ee^{-\frac{2}{3}\zeta^{3/2}-y\zeta^{1/2}},\quad \zeta\to +\infty.
\end{equation}

The transformation $\bm\Psi\mapsto\bm\Psi^\md$ does not change the Lax pair system \eqref{eq:LaxPairPXXXIV}, and the equation $\partial_y \bm\Psi^\md=\bm B \bm \Psi^\md$ implies that $\bm\Psi^\md$ is of the form
$$
\bm \Psi^\md(\zeta\mid y) =     \sqrt{2\pi}
\begin{pmatrix}
\psi(\zeta,y) & \wt\psi(\zeta,y) \\ 
-\ii\partial_y \psi(\zeta,y) & -\ii\partial_y \wt\psi(\zeta,y)
\end{pmatrix}.
$$
Plugging this identity into \eqref{eq:P34kernelRHP}, we see that the kernel $\msf K_\ptf$ takes the form expressed in \eqref{deff:P34kernel}. Still from the second equation in \eqref{eq:LaxPairPXXXIV}, we obtain that $\partial_{yy}\bm\Psi^\md=(\bm B+\partial_y\bm B)\bm\Psi^\md$, and this identity implies 
\begin{equation}\label{eq:psiP34PDE}
\partial_y^2\psi=(\zeta+y+2u(y))\psi.
\end{equation}

Equations~\eqref{eq:psiP34asympt} and \eqref{eq:psiP34PDE} complete the description claimed in \eqref{deff:P34kernel}. 

\subsection{A simplification of the model RHP and the integro-differential Painlevé II equation}\label{sec:subcrintdiffPIIsimplif}\hfill 

For this section, we set
$$
\Gamma_{\wh{\bm\Psi}}\deff \Gamma^\zs_0\cup \Gamma^\zs_2\cup\Gamma^\zs_3\cup\Gamma^\zs_4.
$$
Consider the following RHP.
\begin{rhp}\label{rhp:modelhatPsi}
Find a $2\times 2$ matrix-valued function $\wh{\bm \Psi}$ with the following properties.
\begin{enumerate}[(1)]
\item $\wh{\bm \Psi}$ is analytic on $\C\setminus \Gamma_{\wh{\bm \Psi}}$.
\item The matrix $\wh{\bm \Psi}$ has continuous boundary values $\wh{\bm \Psi}_{\pm}$ along $\Gamma_{\wh{\bm \Psi}} \setminus \{\zs\}$, and they are related by $\wh{\bm \Psi}_{+}(w)=\wh{\bm \Psi}_{-}(w)\bm J_{\wh{\bm \Psi}}(w)$, $w\in \Gamma_{\wh{\bm \Psi}} \setminus \{\zs\}$, where
\begin{equation}
\bm J_{\wh{\bm \Psi}}(w)\deff
\begin{dcases}
\bm I+\frac{1}{1+\ee^{\gt\msf P_\gt(w)}}\bm E_{12}, & w\in \Gamma_0^\zs, \\
\bm I+(1+\ee^{\gt\msf P_\gt(w)})\bm E_{21}, & w\in \Gamma_2^\zs\cup\Gamma_4^\zs, \\
\frac{1}{1+\ee^{\gt\msf P_\gt(w)}}\bm E_{12}-(1+\ee^{\gt\msf P_\gt(w)})\bm E_{21}, & w\in \Gamma_3^\zs.
\end{dcases}
\end{equation}
\item As $w\to \infty$, 
\begin{equation}\label{eq:asymptildePsi}
\wh{\bm \Psi}(w)=\left(\bm I+\Boh(w^{-1})\right)w^{-\sp_3/4}\bm U_0\ee^{-\frac{2}{3}w^{3/2}\sp_3}.
\end{equation}
\item As $w\to \zs$, the matrix $\wh{\bm \Psi}$ remains bounded.
\end{enumerate}
\end{rhp}

The RHP for the matrix $\wh{\bm \Psi}$ differs from the one for $\bm\Psi_\gt$ in that there is no jump for $\wh{\bm \Psi}$ along $\Gamma_1^\zs\cup\Gamma_5^\zs$, and the admissible function $\msf F$ does not appear in the RHP formulation for $\wh{\bm \Psi}$. We will use the matrix $\wh{\bm \Psi}$ in the analysis of the model problem in the subcritical regime. 

Similarly as for $\bm\Psi$ and $\wt{\bm \Psi}$, we introduce
\begin{equation}\label{deff:whmcalH}
\wh{\mcal H}(w)=\wh{\mcal H}(w\mid s)\deff 
\begin{dcases}
\left[\bm\Delta_w \wh{\bm\Psi}_-(w)\right]_{21}, & w>\zs, \\ 
\left[\bm\Delta_w \left(\wh{\bm\Psi}_-(w)\left(\bm I-(1+\ee^{\gt\msf P_\gt(w)})\bm E_{21}\right)\right)\right]_{21}, & w<\zs. \\ 
\end{dcases}
\end{equation}

\begin{prop}\label{prop:analyticwhH}
    For some $\kappa^+,\kappa^->0$, the function $x\mapsto \wh{\mcal H}(x)$ is analytic on the interval $(-\kappa^-\gt^2,\kappa^+\gt)$.
\end{prop}
\begin{proof}
The proof is a direct calculation from the RHP for $\wh{\bm \Psi}$, and the fact that $\msf P_\gt$ - and hence the jump matrix for $\wh{\bm \Psi}$ - is analytic on a disk of growing radius of order $\Boh(\gt)$.
\end{proof}

If we make the formal replacement $\gt\msf P_\gt (\zeta)\mapsto -s+\msf c_{\msf P}w$, the solution $\wh{\bm \Psi}$ was studied in \cite{CafassoClaeysRuzza2021}, where it was shown to be related to the integro-differential Painlevé II equation and related to the KdV equation, see also \cite{GhosalSilva2023}. In our situation, this RHP plays a marginal role, and only its limit as $s\to +\infty$ becomes of greater relevance. For that reason, we refer to \cite{CafassoClaeysRuzza2021, GhosalSilva2023} for further details on connections with integrable systems.

When $s\to +\infty$ within the subcritical regime from Definition~\ref{deff:subregimescritical}, we essentially have the approximation $\ee^{\gt\msf P_\gt(w)}\approx 0$. Consequently, $\wh{\bm \Psi}$ should reduce to the RHP with jumps given by the replacement $\ee^{\gt\msf P_\gt(w)}\mapsto 0$, which is the Airy RHP, as we will now discuss.

From \eqref{eq:PtauFtau}, we see that $\msf P_\gt$ is analytic on a disk of growing radius of order $O(\gt^2)$. In particular, since in the subcritical regime we have $\zs=\Boh(\gt^{1/2})$, $\msf P_\gt$ is analytic on the growing disk $D_\gt(\zs)$ - the choice of radius $\gt$ is made for concreteness, and any radius of order $O(\gt^\eta)$ with $1/2 <\eta< 2$ would suffice for what follows.

We will make a deformation of the RHP~\ref{rhp:modelhatPsi} within the disk $D_\gt(\zs)$, and we now describe the needed geometry. The contour $\Gamma_2^\zs$ intersects the imaginary axis at the point $w_0\deff \ii\sqrt{3}\zs\in D_\gt(\zs)$. We denote
\begin{equation}\label{deff:gammaplus}
\gamma^+\deff (\Gamma_2^\zs\cap \{w\in \C\mid \re w<0\})\cup [w_0,0].
\end{equation}
That is, when we walk along $\gamma^+$ starting from $\infty$, we first walk along $\Gamma_2^\zs$ on the left half plane, until we reach the imaginary axis, and then we walk vertically downwards along the imaginary axis, reaching then the origin. Let also $U^+$ be the triangular domain determined by the edges $w=w_0,w=0$ and $w=\zs$, and $\gamma^-$ and $U^-$ be the sets obtained from $\gamma^+$ and $U^+$ under complex conjugation, respectively. This construction is illustrated in Figure~\ref{Fig:RHPintdiffPIIcontours}.

\begin{figure}[t]
\centering
\begin{subfigure}{.5\textwidth}
\includegraphics[scale=1]{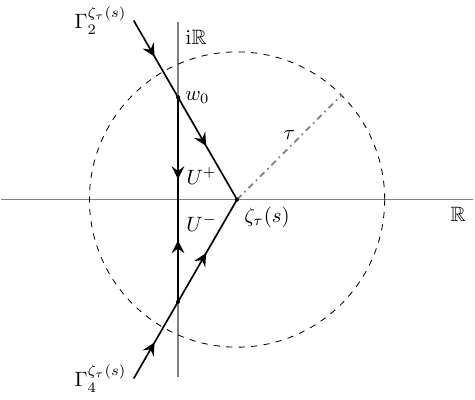}    
\end{subfigure}%
\begin{subfigure}{.5\textwidth}
\centering
\includegraphics[scale=1]{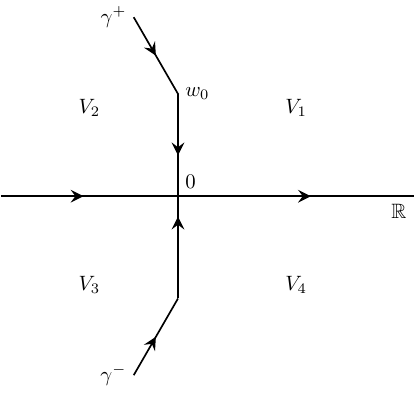}    
\end{subfigure}
\caption{On the left, the construction of the regions $U^\pm$ used for the transformation $\wh{\bm\Psi}\mapsto \wh{\bm\Psi}^\md$. On the right, the regions determined by the contour $\Gamma^\md_{\wh{\bm \Psi}}=\R\cup\gamma^+\cup\gamma^-$. }
\label{Fig:RHPintdiffPIIcontours}
\end{figure}

By construction we know that $U^\pm\subset D_\gt(\zs)$, and as such the function $\msf P_\gt$ is analytic on a neighborhood of $U^\pm$. We transform
\begin{equation}\label{deff:whPsimd}
\wh{\bm \Psi}^\md(w)\deff 
\begin{cases}
    \wh{\bm\Psi}(w)\left(\bm I\pm (1+\ee^{\gt\msf P_\gt(w)}))\bm E_{21}\right), & w\in U^\pm, \\
    \wh{\bm\Psi}(w), & \text{elsewhere}.
\end{cases}
\end{equation}
Set
$$
\Gamma^\md_{\wh{\bm \Psi}}\deff \R\cup\gamma^+\cup\gamma^-,
$$
with orientation as displayed in Figure~\ref{Fig:RHPintdiffPIIcontours}.
Then $\wh{\bm \Psi}^\md$ satisfies the following RHP.

\begin{rhp}\label{rhp:modelhatPsimd}
Find a $2\times 2$ matrix-valued function $\wh{\bm \Psi}^\md$ with the following properties.
\begin{enumerate}[(1)]
\item $\wh{\bm \Psi}^\md$ is analytic on $\C\setminus \Gamma^\md_{\wh{\bm \Psi}}$.
\item The matrix $\wh{\bm \Psi}^\md$ has continuous boundary values $\wh{\bm \Psi}_{\pm}^\md$ along $\Gamma_{\wh{\bm \Psi}}^\md \setminus \{0\}$, and they are related by $\wh{\bm \Psi}_{+}^\md(w)=\wh{\bm \Psi}_{-}^\md(w)\bm J_{\wh{\bm \Psi}}^\md(w)$, $w\in \Gamma_{\wh{\bm \Psi}}^\md \setminus \{0\}$, where
\begin{equation}
\bm J_{\wh{\bm \Psi}}^\md(w)\deff
\begin{dcases}
\bm I+\frac{1}{1+\ee^{\gt\msf P_\gt(w)}}\bm E_{12}, & w\in (0,\infty), \\
\bm I+(1+\ee^{\gt\msf P_\gt(w)})\bm E_{21}, & w\in \gamma^+\cup\gamma^-, \\
\frac{1}{1+\ee^{\gt\msf P_\gt(w)}}\bm E_{12}-(1+\ee^{\gt\msf P_\gt(w)})\bm E_{21}, & w\in (-\infty,0).
\end{dcases}
\end{equation}
\item As $w\to \infty$, 
\begin{equation}
\wh{\bm \Psi}^\md(w)=\left(\bm I+\Boh(w^{-1})\right)w^{-\sp_3/4}\bm U_0\ee^{-\frac{2}{3}\zeta^{3/2}\sp_3}.
\end{equation}
\item As $w\to 0$, the matrix $\wh{\bm \Psi}^\md$ remains bounded.
\end{enumerate}
\end{rhp}

To taylor $\wh{\bm \Psi}$ to our later needs, we now perform a quick asymptotic analysis of it, comparing it with the Airy RHP.

The bare Airy parametrix $\bai_0$ is given in \eqref{eq:bai1}, and it is used to construct the solution $\bai$ to the classical Airy RHP, namely RHP~\ref{rhp:airy}. We need to deform the RHP for $\bai$, and for that we repeat the construction of $\bai$ but using slightly different sectors on the plane. With $V_j$'s being the connected components of $\C\setminus \Gamma_{\wh{\bm \Psi}}^\md$, labelled as shown in Figure~\ref{Fig:RHPintdiffPIIcontours}, we define
\begin{equation}\label{deff:whAi}
\wh{\bai}(w)\deff
\ee^{\pi \ii\sp_3/4}\bai_0(w)\times
\begin{cases}
\bm I, & w\in V_1, \\
(\bm I-\bm E_{21}), & w\in V_2, \\ 
(\bm I-\bm E_{12})(\bm I+\bm E_{21}), & w\in V_3, \\
(\bm I-\bm E_{12}), & w\in V_4. \\ 
\end{cases}
\end{equation}
Then $\wh\bai$ satisfies the RHP obtained from RHP~\ref{rhp:modelhatPsimd} with the replacement $\ee^{\gt\msf P_\gt(w)}\mapsto 0$. We emphasize that $\wh\bai$ is bounded on bounded subsets of $\C$, and the asymptotic behavior
\begin{equation}\label{eq:asymptbehwhbai}
\wh\bai(w)=\left(\bm I+\Boh(w^{-1})\right)w^{-\sp_3/4}\bm U_0\ee^{-\frac{2}{3}w^{3/2}\sp_3},\quad w\to \infty,
\end{equation}
still holds true.

We transform
\begin{equation}\label{deff:whR}
\wh{\bm R}(w)\deff \wh{\bm \Psi}^\md(w)\wh\bai(w)^{-1},\quad w\in \C\setminus \Gamma^\md_{\wh{\bm \Psi}}.
\end{equation}
From RHP~\ref{rhp:modelhatPsimd}, we obtain that $\wh{\bm R}$ satisfies the following RHP.

\begin{rhp}\label{rhp:modelhatR}
Find a $2\times 2$ matrix-valued function $\wh{\bm R}$ with the following properties.
\begin{enumerate}[(1)]
\item $\wh{\bm R}$ is analytic on $\C\setminus \Gamma^\md_{\wh{\bm \Psi}}$.
\item The matrix $\wh{\bm R}$ has continuous boundary values $\wh{\bm R}_{\pm}$ along $\Gamma_{\wh{\bm \Psi}} \setminus \{0\}$, and they are related by $\wh{\bm R}_{+}(w)=\wh{\bm R}_{-}(w)\bm J_{\wh{\bm R}}(w)$, $w\in \Gamma_{\wh{\bm \Psi}}^\md \setminus \{0\}$, where
\begin{equation}
\bm J_{\wh{\bm R}}(w)\deff
\begin{dcases}
\bm I+\left(\frac{1}{1+\ee^{\gt\msf P_\gt(w)}}-1\right) \wh\bai_+(w)\bm E_{12}\wh\bai_+(w)^{-1}, & w\in (0,\infty), \\
\bm I+\ee^{\gt\msf P_\gt(w)}\wh\bai_\pm(w)\bm E_{21}\wh\bai_\pm(w)^{-1}, & w\in \gamma^\pm, \\
\bm I+\wh\bai_+(w)\left[\ee^{\gt\msf P_\gt(w)}\bm E_{11}+\left(\frac{1}{1+\ee^{\gt\msf P_\gt(w)}}-1\right)\bm E_{22}\right]\wh\bai_+(w)^{-1}, & w\in (-\infty,0).
\end{dcases}
\end{equation}
\item As $w\to \infty$, 
\begin{equation}
\wh{\bm R}=\bm I+\Boh(w^{-1}).
\end{equation}
\item As $w\to 0$, the matrix $\wh{\bm R}$ remains bounded.
\end{enumerate}
\end{rhp}

To conclude the asymptotic analysis, we prove that the jump matrix $\bm J_{\wh{\bm R}}$ is asymptotically close to the identity in the appropriate sense.

\begin{prop}\label{prop:whRestimate}
    There exists $\eta>0$ such that the estimate
    $$
    \|\bm J_{\wh{\bm R}}-\bm I\|_{L^1\cap L^\infty(\Gamma_{\wh{\bm R}}^\md)}=\Boh(\ee^{-\eta s/\gt})
    $$
    holds true as $\gt\to \infty$, uniformly for $s$ within the subcritical regime from Definition~\ref{deff:subregimescritical}.
\end{prop}

\begin{proof}
We estimate the jump in each part of $\Gamma_{\wh{\bm \Psi}}^\md$ separately. 

To analyze the jump on $(0,\infty)$, we split it into two parts, namely $w>\zs/2$ and $w\leq \zs/2$. The asymptotic behavior \eqref{eq:asymptbehwhbai} shows that
\begin{align}
& |\wh\bai_+(w)\bm E_{12}\wh\bai_+(w)|^{-1}\leq C\ee^{-3w}\leq C\ee^{-w-\zs},\quad w\geq \frac{\zs}{2},\qquad \text{and} \\
&|\wh\bai_+(w)\bm E_{12}\wh\bai_+(w)|^{-1}\leq C,\quad 0\leq w\leq \frac{\zs}{2}, 
\end{align}
for some constant $C>0$. The linear factor $3w$ in the exponent above is not optimal, but sufficient for our purposes.

On the other hand,
\begin{align}
& \frac{1}{1+\ee^{\gt\msf P_\gt(w)}}\leq 1,\quad w\geq \frac{\zs}{2},\qquad \text{and} \\
& \frac{1}{1+\ee^{\gt\msf P_\gt(w)}}-1=\Boh(\ee^{-\gt\eta \zs/2}),\quad 0\leq w\leq \frac{\zs}{2}, 
\end{align}
where for the last estimate we used Proposition~\ref{prop:canonicestPF}--(ii). These estimates combined imply that
    $$
    \|\bm J_{\wh{\bm R}}-\bm I\|_{L^1\cap L^\infty(0,\infty)}=\Boh(\ee^{-\eta\zs}).
    $$

Next, we prove estimates along $\gamma^+$, the estimate along $\gamma^-$ being analogous. The asymptotic behavior \eqref{eq:asymptbehwhbai} now yields that
$$
|\wh\bai_\pm(w)\bm E_{21}\wh\bai_\pm(w)^{-1}|=\Boh(1+|w|^{1/2}\ee^{-\frac{4}{3}w^{3/2}}),\quad w\in \gamma^+.
$$
In particular, this implies that for some $C>0$, 
\begin{equation}\label{eq:whbaigammaplus}
|\wh\bai_\pm(w)\bm E_{21}\wh\bai_\pm(w)^{-1}|\leq C \ee^{-|w|},
\end{equation}
for every $w\in \gamma^+$.

To continue, recall that $\gamma^+$ was given in \eqref{deff:gammaplus}. On one hand, if $w\in \gamma^+\cap \Gamma_2^\zs$, then certainly $|w-\zs|\geq |\zs|=\zs$, and Proposition~\ref{prop:canonicestPF}--(ii) yields that
\begin{equation}\label{eq:whbaigammaplus2}
|\ee^{\gt\msf P_\gt(w)}|\leq C\ee^{-\eta\gt\zs },\quad w\in \gamma^+\cap \Gamma_2^\zs,
\end{equation}
for some constants $\eta,C>0$. On the other hand, for $w\in \gamma^+\setminus \Gamma_2^\zs=[0,w_0)$, we use \eqref{eq:PtauFtau} and see that
$$
\msf P_\gt(w)=-\frac{s}{\gt}+\msf c_{\msf P}w +\Boh\left(\frac{\zs^2}{\gt^2}\right),\quad \gt\to \infty,
$$
with uniform error term for $w\in [0,w_0]$. Now, $\frac{\zs^2}{\gt^2}=\Boh(s^2/\gt^4)$, and because we are in the subcritical regime we can make sure that this term is sufficiently small, which yields that
\begin{equation}\label{eq:whbaigammaplus3}
\re \msf P_\gt(w)\leq -\frac{s}{\gt}+\msf c_{\msf P}\re w+1\leq -\msf c_{\msf P}\zs+2.
\end{equation}
Combining \eqref{eq:whbaigammaplus}, \eqref{eq:whbaigammaplus2} and \eqref{eq:whbaigammaplus3}, we see that
$$
|\bm J_{\wh{\bm R}}(w)-\bm I|\leq C\ee^{-|w|-\eta \gt\zs},
$$
for every $w\in \gamma^+$. This inequality implies the claim on $L^1\cap L^\infty(\Gamma_{\wh{\bm R}}^\md)$. 

Finally, for $w<0$, we see from \eqref{eq:asymptbehwhbai} that
$$
\wh\bai_+(w)\bm E_{jj}\wh\bai_+(w)^{-1}=\Boh(1+|w|^{1/2}),
$$
uniformly for $w\in (-\infty,0]$, whereas having in mind that $|w-\zs|=\zs-|w|$, we use once again Proposition~\ref{prop:canonicestPF}--(ii) and obtain that 
$$
|\ee^{\gt\msf P_\gt(w)}|+\left|\frac{1}{1+\ee^{\gt\msf P_\gt(w)}}-1\right|\leq C \ee^{-\eta\gt\zs-\eta|w| } ,
$$
for some constant $C>0$, which finalizes the proof.
\end{proof}

Thanks to Proposition~\ref{prop:whRestimate}, the general small norm theory of RHPs yields the required estimate for $\wh{\bm R}$. We will use these estimates later.

\begin{theorem}\label{thm:smallnormhatR}
There exists $\eta>0$ such that the estimates
$$
\|\wh{\bm R}-\bm I\|_{L^\infty(\C\setminus \Gamma_{\wh{\bm \Psi}}^\md)}=\Boh\left( \ee^{-\eta \zs} \right),\quad \|\wh{\bm R}_\pm-\bm I\|_{L^2(\Gamma_{\wh{\bm \Psi}}^\md)}=\Boh\left( \ee^{-\eta \zs} \right)
$$
as well as
$$
\|\partial_w \wh{\bm R} \|_{L^\infty(\C\setminus \Gamma_{\wh{\bm \Psi}}^\md )}=\Boh\left( \ee^{-\eta \zs} \right),
$$
are valid as $\gt\to \infty$, uniformly for $s$ within the subcritical regime.
\end{theorem}

For later, we also encode an useful analytic continuation of $\wh{\bm R}$.

\begin{prop}\label{prop:analcontwhR}
    For some $\kappa>0$, the function
    $$
    \wh{\bm R}^\md (w)\deff 
    \begin{dcases}
    \wh{\bm R}(w), & w\in V_4,\; |w|< \kappa\gt^2\\
    \wh{\bm R}(w)\left(\bm I-\ee^{\gt\msf P_\gt(w)}\wh\bai(w)\bm E_{21}\wh\bai(w)^{-1}\right)& w\in V_3,\; |w|<\kappa\gt^2,
    \end{dcases}
    $$
    extends analytically to $\{w\in \C\mid -\pi<\arg w<0$, $|w|<\kappa \gt^2\}$.
\end{prop}
\begin{proof}
    All the terms involved are analytic on each of the sectors $V_3$ and $V_4$ with $|w|<\kappa \gt^2$. The jump condition for $\wh{\bm R}$ from RHP~\ref{rhp:modelhatR} then ensures the claimed expression is analytic across $\gamma_-$ as well.
\end{proof}

\subsection{A modification of the model problem for the supercritical regime}\label{sec:modRHPmodel}\hfill 

In this section we introduce and study a modification of the model RHP \ref{rhp:modelPhi}, which will be essential in the analysis of the RHP \ref{rhp:modelPhi} itself in the supercritical regime. 

Let $\msf F$ and $\msf P$ be admissible functions in the sense of Definition~\ref{deff:admissibleFP}, with corresponding point $\zeta_\gt(s)$ as in Proposition~\ref{prop:zeroPt}. In this section we are only interested in $s$ within the supercritical regime from Definition~\ref{deff:subregimescritical}. Under such assumption, \eqref{eq:zetasseries} ensures the existence of $M=M(T)>0$ depending only on $T$ for which
\begin{equation}\label{eq:boundzetagtnc}
-\frac{1}{M} \gt^{1/2}\leq \zeta_\gt(s)\leq -M.
\end{equation}
The analysis that follows will make extensive use of this restriction on $\zs$.

Introduce
\begin{equation}\label{deff:mcalPF}
\begin{aligned}
& \mcal P_\gt(w)\deff -s+\gt^3\msf P\left(\frac{w}{\gt^3}+\frac{\zeta_\gt(s)}{\gt^2}\right)=\gt\msf P_\gt\left(\frac{w}{\gt}+\zeta_\gt(s)\right),\quad w\in \C,\\
& \mcal F_\gt(w)\deff -u\gt^3+\gt^3\msf F\left(\frac{w}{\gt^3}+\frac{\zeta_\gt(s)}{\gt^2}\right)=\gt\msf F_\gt\left(\frac{w}{\gt}+\zeta_\gt(s)\right), \quad w\in \C.
\end{aligned}
\end{equation}

The new model RHP we are about to state has jump matrices given in terms of $\mcal F_\gt$ and $\mcal P_\gt$, and hence it depends on the admissible functions $\msf F, \msf P$ as well as on the variables $\gt,s,u$. Furthermore, it will also depend on a new variable $y>0$. For its statement, recall that $\Gamma_0$ was introduced in \eqref{deff:Gammalambda}.

\begin{rhp}\label{rhp:modelUpsilon}
Find a $2\times 2$ matrix-valued function ${\bm \Upsilon}_\gt$ with the following properties.
\begin{enumerate}[(1)]
\item ${\bm \Upsilon}_\gt$ is analytic on $\C\setminus \Gamma^0$.
\item The matrix ${\bm \Upsilon}_\gt$ has continuous boundary values ${\bm \Upsilon}_{\gt,\pm}$ along $\Gamma^0 \setminus \{0\}$, and they are related by ${\bm \Upsilon}_{\gt,+}(w)={\bm \Upsilon}_{\gt,-}(w)\bm J_{\bm \Upsilon}(w)$, $w\in \Gamma^{0}\setminus \{0\}$, where
\begin{equation}
\bm J_{\bm \Upsilon}(w)\deff
\begin{dcases}
\bm I+\frac{1}{1+\ee^{{\mcal P}_\gt(w)}}\bm E_{12}, & w\in \Gamma_0^{0}, \\
\bm I-\frac{1}{(1-\ee^{-\ii {\mcal F}_\gt(w)})(1+\ee^{{\mcal P}_\gt(w)})} \bm E_{12}, & w\in \Gamma_1^0, \\
\left(\bm I+(1+\ee^{{\mcal P}_\gt(w)})\bm E_{21}\right)\left(\bm I-\frac{1}{(1-\ee^{-\ii{\mcal F}_\gt(w)})(1+\ee^{{\mcal P}_\gt(w)})} \bm E_{12}\right), & w\in \Gamma_2^0, \\
\frac{1}{1+\ee^{{\mcal P}_\gt(w)}}\bm E_{12}-(1+\ee^{{\mcal P}_\gt(w)})\bm E_{21}, & w\in \Gamma_3^0, \\
\left(\bm I-\frac{1}{(1-\ee^{\ii {\mcal F}_\gt(w)})(1+\ee^{{\mcal P}_\gt(w)})} \bm E_{12}\right)\left(\bm I+(1+\ee^{{\mcal P}_\gt(w)}) \bm E_{21}\right), & w\in \Gamma_4^0,\\
\bm I-\frac{1}{(1-\ee^{\ii {\mcal F}_\gt(w)})(1+\ee^{{\mcal P}_\gt(w)})} \bm E_{12}, & w\in \Gamma_5^0.
\end{dcases}
\end{equation}
\item As $w\to \infty$, 
\begin{equation}
\bm \Upsilon_\gt(w)=\left(\bm I+\Boh(w^{-1})\right)w^{-\sp_3/4}\bm U_0\ee^{yw^{1/2}\sp_3}.
\end{equation}

\item As $w\to 0$,
$$
\bm \Upsilon_\gt(w)=
\begin{cases}
    \Boh(1), & \text{if } \ee^{\ii \mcal F_\gt(0)}\neq 1, \\
    \Boh(1), & \text{if } \ee^{\ii \mcal F_\gt(0)}=1  \text{ and } w\notin \Omega_1^0\cup \Omega_4^0, \\
    \Boh\begin{pmatrix}
        1 & w^{-1} \\ 1 & w^{-1}
    \end{pmatrix}, & \text{if } \ee^{\ii \mcal F_\gt(0)}=1  \text{ and } w\in \Omega_1^0\cup \Omega_4^0.
\end{cases}
$$
\end{enumerate}
\end{rhp}

As said, this new model problem $\bm \Upsilon_\gt$ depends on all the parameters $s,y,u$. When needed to specify one of these parameters, we write $\bm \Upsilon_\gt(\cdot)=\bm \Upsilon_\gt(\cdot\mid s=s_0), \bm \Upsilon_\gt(\cdot\mid y=y_0)$ etc.

The structure of the jump of $\bm \Upsilon_\gt$ is the same as the jump for the model problem $\bm \Phi_\gt$ from RHP~\ref{rhp:modelPhi}: the corresponding jumps are related by a scaling
\begin{equation}\label{eq:jumpPhijumpUpsilonrel}
\bm J_{\gt}(\zeta)=\bm J_{\bm \Upsilon}(\gt\zeta).
\end{equation}
However, the behavior of $\bm \Upsilon_\gt$ and $\bm \Phi_\gt$ at $\infty$ are different.

Introduce
\begin{equation}\label{deff:wtmcalH}
\wt{\mcal H}_\gt(w)=\wt{\mcal H}_\gt(w\mid y)\deff \left[\bm\Delta_w \bm\Upsilon_\gt(w\mid y)\right]_{21,-},\quad w>0,y>0,
\end{equation}
and for sufficiently small $\varepsilon>0$, also define
\begin{equation}\label{deff:Upsilonmod}
\bm\Upsilon_\gt^\md(w)\deff 
\begin{dcases}
    \bm\Upsilon_\gt(w)\left(\bm I-\frac{\bm E_{12}}{1+\ee^{\gt\mcal P_\gt(w)}}\right), & w\in \Omega_0^0\cap D_{\varepsilon \gt^3}(0), \\
    \bm\Upsilon_\gt(w)\left(\bm I+\frac{\bm E_{12}}{(1-\ee^{-\ii \gt\mcal F_\gt(w)})(1+\ee^{\gt\mcal P_\gt(w)})}\right)\left(\bm I-\frac{\bm E_{12}}{1+\ee^{\gt\mcal P_\gt(w)}}\right), & w\in \Omega_1^0\cap D_{\varepsilon \gt^3}(0), \\
    \bm\Upsilon_\gt(w)\left( \bm I+(1+\ee^{\gt\mcal P_\gt(w)})\bm E_{21} \right)\left( \bm I-\frac{\bm E_{12}}{1+\ee^{\gt\mcal P_\gt(w)}}\right), & w\in \Omega_2^0\cap D_{\varepsilon \gt^3}(0), \\
    \bm\Upsilon_\gt(w)(\bm I-(1+\ee^{\gt\mcal P_\gt(w)})\bm E_{21}), & w\in \Omega_3^0\cap D_{\varepsilon \gt^3}(0), \\
    \bm\Upsilon_\gt(w)\left(\bm I- \frac{\bm E_{12}}{(1-\ee^{\ii \gt\mcal F_\gt(w)})(1+\ee^{\gt\mcal P_\gt(w)})}\right), & w\in \Omega_4^0\cap D_{\varepsilon \gt^3}(0), \\
    \bm\Upsilon_\gt(w), & \zeta\in \Omega^0_5\cap D_{\varepsilon \gt^3}(0).
\end{dcases}
\end{equation}
The next result is the analogue of Proposition~\ref{prop:fundamentalbmHtau}.

\begin{prop}\label{prop:fundHPsimod}
Suppose that $y>0$ and $s$ is within the supercritical regime from Definition~\ref{deff:subregimescritical}. The following properties hold.
\begin{enumerate}[(i)]
    \item The function $\bm\Upsilon_\gt^\md$ is the analytic continuation of $\bm\Upsilon_\gt$ from $\Omega_5^0$ to a full neighborhood of the origin.

    \item The function $\wt{\mcal H}_\gt(w)$ extends analytically from $w>0$ to $-\delta \gt^3<w<\delta \gt^3 $, for some $\delta>0$ independent of $s,y,\gt$.
\end{enumerate}
\end{prop}

\begin{proof}
   The expressions in \eqref{deff:Upsilonmod} that define $\bm\Upsilon_\gt^\md$ are obtained from the right-hand side of \eqref{deff:Phimod} when we replace $\bm\Phi_\gt, \gt\msf P_\gt,\gt\msf F_\gt$ by $\bm\Upsilon_\gt,\mcal P_\gt$ and $\mcal F_\gt$, respectively. With a change of variables, the jumps of $\bm{\Upsilon}_\gt$ and $\bm \Phi_\gt$ coincide, see \eqref{eq:jumpPhijumpUpsilonrel}. This means that the same arguments used in the proof of Proposition~\ref{prop:fundamentalbmHtau} also show that $\bm\Upsilon_\gt^\md$ is analytic in a neighborhood of the origin, proving (i).

   The proof of (ii) then follows by the identity
   $$
   \wt{\mcal H}_\gt(w\mid y)=\left[\bm\Delta_w \bm\Upsilon_\gt^\md(w\mid y)\right]_{21},\quad w>0,
   $$
   the analyticity of $\bm\Upsilon_\gt^\md$, and the fact that $\mcal P$ and $\mcal F$ are analytic in a neighborhood of the origin of radius $\Boh(\gt^3)$.
\end{proof}

To better understand the different nature between $\bm \Upsilon_\gt$ and $\bm\Phi_\gt$, for $s$ within the supercritical regime we state the approximation
$$
\mcal P_\gt(w)\approx \msf c_{\msf P}w\quad \text{and}\quad  \mcal F_\gt(w)\approx u_\gt+\msf c_{\msf F}w,\quad u_\gt\deff -u\gt^3+\gt\cF\zeta_\gt(s),
$$
which is formally valid as $\gt\to \infty$. Thus, $\bm \Upsilon_\gt$ depends on $\gt$ in an oscillatory way: the factors $\ee^{\pm \ii \mcal F_\gt(w)}\approx \ee^{\pm \ii u_\gt}\ee^{\pm\ii \msf c_{\msf F}w}$ oscillate with $\gt$ for any value of $w$. This is in contrast with the critical regime, where in the appropriate scale we had $|\ee^{\pm \ii \gt \msf F_\gt(w)}|\approx \ee^{\mp \gt \cF \im w }$ so the factors involving $\msf F_\gt$ were decaying as $\gt\to \infty$.

We are interested in two major properties of $\bm \Upsilon_\gt$: (1) the existence of $\bm \Upsilon_\gt$, and (2) its behavior as $\gt\to \infty$ and simultaneously $y\to 0^+$ in a proper sense. The main outcome of the current section is to carry out the analysis of (2). As a consequence, (1) will follow for $y$ sufficiently small. We state it as a formal result, the proof follows from (2) in a standard manner and we skip it.

\begin{theorem}\label{thm:existencemodelUpsilon}
There exists $y_0>0$ and $\gt_0>0$ sufficiently large such that a solution $\bm \Upsilon_\gt$ of the RHP~\ref{rhp:modelUpsilon} exists for every $\gt\geq \gt_0$, every $y\in (y_0/\gt^{3/2},1/y_0)$, and every $s$ within the supercritical regime.
\end{theorem}

We now carry out the asymptotic analysis for $\bm \Upsilon_\gt$ as $\gt\to \infty$ and $y$ is small. This asymptotic analysis is similar to the one in \cite[Section~7]{CafassoClaeysRuzza2021}. The functions $\mcal F_\gt$ and $\mcal P_\gt$ - and hence the jump matrix for $\bm \Upsilon_\gt$ - depend on the function $\zeta_\gt(s)$. We assume that for some fixed but arbitrarily large constant $y_0>0$,
\begin{equation}\label{eq:scalingy}
\frac{y_0}{\gt^{3/2}}\leq y\leq \frac{1}{y_0} ,\quad \text{and }s \text{ is within the supercritical regime.}
\end{equation}

\subsubsection{Scaling step}\label{sec:Upsilonscaling} As a first step, we scale
\begin{equation}\label{deff:whUpsi}
\bm V_\gt(\xi)\deff y^{-\sp_3/2}\bm{\Upsilon}_\gt\left(\frac{\xi}{y^2}\right),\quad \xi\in \C\setminus \Gamma^0.
\end{equation}
Denote the induced scaled functions by
\begin{equation}\label{deff:wtcalPF}
\wt{\mcal P}_\gt(\xi)\deff \mcal P_\gt\left(\frac{\xi}{y^2}\right)=\gt \msf P_\gt\left( \frac{\xi}{y^2\gt}+\zeta_\gt(s) \right),\quad \wt{\mcal F}_\gt(\xi)\deff \mcal F_\gt\left(\frac{\xi}{y^2}\right)=\gt \msf F_\gt\left( \frac{\xi}{y^2\gt}+\zeta_\gt(s) \right),\quad \xi\in \Gamma^0.
\end{equation}
These functions admit an expansion of the form
\begin{equation}\label{eq:expwtmcPF}
\begin{aligned}
& \wt{\mcal P}_\gt(\xi)=\frac{\xi}{y^2}\msf P'\left( \frac{\zs}{\gt^2} \right)\left(1+\Boh\left(\frac{\xi}{y^2\gt^3}\right)\right),\qquad \text{and}\\
& \wt{\mcal F}_\gt(\xi)=-u\gt^3+\gt^3\msf F\left( \frac{\zs}{\gt^2} \right)+\frac{\xi}{y^2}\msf F'\left( \frac{\zs}{\gt^2} \right)\left(1+\Boh\left(\frac{\xi}{y^2\gt^3}\right)\right),
\end{aligned}
\end{equation}
valid as $\gt\to\infty$ uniformly for $\xi=o(y^2\gt^3)$ and $s,y$ as in \eqref{eq:scalingy}.

It then follows that $\bm V_\gt$ satisfies the following RHP.

\begin{rhp}\label{rhp:tildeV}
Find a $2\times 2$ matrix-valued function $\bm V_\gt$ with the following properties.
\begin{enumerate}[(1)]
\item $\bm V_\gt$ is analytic on $\C\setminus \Gamma^0$.
\item The matrix $\bm V_\gt$ has continuous boundary values $\bm V_{\gt,\pm}$ along $\Gamma^0 \setminus \{0\}$, and they are related by $\bm V_{\gt,+}(\xi)=\bm V_{\gt,-}(\xi)\bm J_{\bm V}(\xi)$, $\xi\in \Gamma^{0}\setminus \{0\}$, where

\begin{equation}\label{eq:jumphatUpsilon}
\bm J_{\bm V}(\xi)\deff \bm J_{\bm \Upsilon}(y^{-2}\xi)=
\begin{dcases}
\bm I+\frac{1}{1+\ee^{\wt{\mcal P}_\gt(\xi)}}\bm E_{12}, & \zeta\in \Gamma_0^{0}, \\
\bm I-\frac{1}{(1-\ee^{-\ii \wt{\mcal F}_\gt(\xi)})(1+\ee^{\wt{\mcal P}_\gt(\xi)})} \bm E_{12} , & \xi\in \Gamma_1^{0}, \\
\left(\bm I+(1+\ee^{\wt{\mcal P}_\gt(\xi)}) \bm E_{21} \right)
\left( \bm I- \frac{1}{(1-\ee^{-\ii\gt\wt{\mcal F}_\gt(\xi)})(1+\ee^{\wt{\mcal P}_\gt(\xi)})} \bm E_{12}\right), & \xi\in \Gamma_2^0, \\
\frac{1}{1+\ee^{\wt{\mcal P}_\gt(\xi)}}\bm E_{12}-(1+\ee^{\wt{\mcal P}_\gt(\xi)})\bm E_{21}, & \xi\in \Gamma_3^0, \\
\left(\bm I-\frac{1}{(1-\ee^{\ii\gt\wt{\mcal F}_\gt(\xi)})(1+\ee^{\wt{\mcal P}_\gt(\xi)})}\bm E_{12}\right)\left(\bm I+(1+\ee^{\wt{\mcal P}_\gt(\xi)})\bm E_{21}\right), & \xi\in \Gamma_4^0,\\
\bm I-\frac{1}{(1-\ee^{\ii \wt{\mcal F}_\gt(\xi)})(1+\ee^{\wt{\mcal P}_\gt(\xi)})} \bm E_{12} , & \xi\in \Gamma_5^0.
\end{dcases}
\end{equation}
\item As $\xi\to \infty$, 
\begin{equation}\label{eq:asympthatUpsilon}
\bm V_\gt(\xi)=\left(\bm I+\Boh(\xi^{-1})\right)\xi^{-\sp_3/4}\bm U_0\ee^{\xi^{1/2}\sp_3}.
\end{equation}

\item As $\xi\to 0$,
\begin{equation}\label{eq:behhatUpsilonorigin}
\bm V_\gt(\xi)=
\begin{cases}
    \Boh(1), & \text{if } \ee^{\ii \wt{\mcal F}_\gt(0)}\neq 1, \\
    \Boh(1), & \text{if } \ee^{\ii \wt{\mcal F}_\gt(0)}=1  \text{ and } |\arg\xi|\notin [\frac{\pi}{3},\frac{2\pi}{3}], \\
    \Boh\begin{pmatrix}
        \xi^{-1} & 1 \\ \xi^{-1} & 1
    \end{pmatrix}, & \text{if } \ee^{\ii \wt{\mcal F}_\gt(0)}=1  \text{ and } |\arg\xi|\in (\frac{\pi}{3},\frac{2\pi}{3}).
\end{cases}
\end{equation}
\end{enumerate}
\end{rhp}

The proof that $\bm V_\gt$ given in \eqref{deff:whUpsi} indeed satisfies this RHP is immediate from the RHP~\ref{rhp:modelUpsilon} satisfied by $\bm \Upsilon_\gt$, and we skip it.

\subsubsection*{Construction of global parametrix} As we will see, 
$$
\bm J_{\bm V}
\to 
\begin{cases}
    \bm I, & w\in \Gamma^0\cup \Gamma_1^0\cup \Gamma_5^0,\\
    \bm I+\bm E_{21}, & w\in \Gamma_2^0\cup \Gamma_4^0, \\
    \bm E_{12}-\bm E_{21}, & w\in \Gamma_3^0.
\end{cases}
$$
Hence, with
$$
\Gamma_{\bm B}\deff \Gamma_2^0\cup\Gamma_3^0\cup\Gamma_4^0,
$$
the global parametrix $\bm B$ we require is the solution to the following RHP.

\begin{rhp}\label{rhp:bessel}
Find a $2\times 2$ matrix-valued function $\bm B$ with the following properties.
\begin{enumerate}[(1)]
\item The matrix $\bm B$ is analytic on $\C\setminus \Gamma_{\bm B}$.
\item The matrix $\bm B$ has continuous boundary values $\bm B_\pm$ along $\Gamma_{\bm B}\setminus \{0\}$, and they are related by $\bm B_+(\xi)=\bm B_-(\xi)\bm J_{\bm B}(\xi)$, $\xi\in \Gamma_{\bm B}\setminus \{0\}$, where
\begin{equation}\label{eq:jumpBessel}
\bm J_{\bm B}(\xi)\deff 
\begin{dcases}
\bm I+\bm E_{21}, & \xi\in \Gamma_2^0\cup\Gamma_4^0, \\
\bm E_{12}-\bm E_{21}, & \xi\in \Gamma_3^0.
\end{dcases}
\end{equation}
\item As $\xi\to \infty$, 
\begin{equation}\label{eq:asymptBessel}
\bm B(\xi)=\left(\bm I+\Boh(\xi^{-1})\right)\xi^{-\sp_3/4}\bm U_0\ee^{\xi^{1/2}\sp_3}.
\end{equation}

\item As $\xi\to 0$,
\begin{equation}\label{eq:behBesselParOrigin}
\bm B(\xi)=\bm B_0(\xi)\left(\bm I+\frac{1}{2\pi \ii }\log (\xi) \bm E_{12}\right)\bm M_j, \quad \xi\in \Omega_j, j=0,1,2,
\end{equation}
where $\bm B_0$ is analytic and invertible in a neighborhood of the origin, and the matrices $\bm M_j$ and sectors $\Omega_j$ are as in \eqref{eq:PXXXIVmon} and \eqref{deff:sectorsOmegaj}, respectively.
\end{enumerate}
\end{rhp}

Obtaining the solution to this RHP is standard, and it is constructed explicitly in terms of Bessel functions, namely
\begin{equation}\label{deff:BesselB}
\bm B(\xi)=\left( \bm I+\frac{3\ii}{8}\bm E_{21} \right)
\begin{pmatrix}
    I_0(\xi^{1/2}) & \ii K_0(\xi^{1/2}) \\ \ii \xi^{1/2}I_0'(\xi^{1/2}) & -\xi^{1/2}K'_0(\xi^{1/2})
\end{pmatrix}
\pi^{\sp_3/2}
\times 
\begin{cases}
\bm I, & \xi\in \Omega_0, \\    
\bm I-\bm E_{21}, & \xi\in \Omega_1, \\
\bm I+\bm E_{21}, & \xi\in \Omega_2, \\
\end{cases}
\end{equation}
where $I_\nu$ and $K_\nu$ are the modified Bessel functions of first and second kind, respectively. For later use, we recall that the Bessel functions appearing above are solutions to the Bessel differential equation,
$$
w^2y''+wy-(w^2+\nu^2)w=0,\quad y=I_\nu(w), K_\nu(w),\; \nu\in \C,
$$
and they admit the expansions
\begin{equation}\label{eq:seriesBesselI0}
I_0(\xi^{1/2})=\sum_{k=0}^\infty \frac{\xi^k}{4^k(k!)^2}\quad \text{and}\quad K_0(\xi^{1/2})=-\left(\frac{1}{2}\log\xi -\log 2+\gamma\right)I_0(\xi^{1/2})+\sum_{k=1}^\infty \frac{\gamma_k\xi^k}{4^k(k!)^2},
\end{equation}
where here 
$$
\gamma_n\deff \sum_{k=1}^n \frac{1}{k}, \qquad \text{and}\qquad  \gamma\deff \lim_{n\to\infty} \left(-\log n+\sum_{k=1}^n \frac{1}{k}\right),
$$ 
is the Euler's constant, and the power series in \eqref{eq:seriesBesselI0} are absolutely convergent for $\xi\in \C$. In particular, $\bm B_0$ is in fact an entire function.

In a moment, we will also use $I_1$, which admits the expansion
\begin{equation}\label{eq:seriesBesselI1}
I_1(\xi^{1/2})=\frac{1}{2}\xi^{1/2}\sum_{k=0}^\infty \frac{\xi^k}{4^k k!(k+1)!},\quad \xi\to 0,
\end{equation}
the series on the right-hand side is an absolutely convergent series on $\C$, and it satisfies the identity
$$
I_0'(\xi)=I_1(\xi),\quad \xi>0.
$$

A direct calculation shows that
$$
\left[\bm B(\xi)^{-1}\bm B(v)\right]_{21}=\pi \ii \left( I_0(\sqrt{\xi})I_0'(\sqrt{v})\sqrt{v}-I_0(\sqrt{v})I_0(\sqrt{\xi})\sqrt{\xi} \right),\quad \xi,v>0, 
$$
and in turn this identity implies that
\begin{equation}\label{eq:DeltawB}
[\bm\Delta_\xi \bm B(\xi)]_{21}=\frac{\pi \ii}{2}\left( I_1(\sqrt{\xi})^2-I_0(\sqrt{\xi})^2 \right),\quad \xi>0.
\end{equation}
The right-hand side is an entire function, and thus it provides an analytic extension of the (boundary value of the) left-hand side to the whole real line.

Recall that the Bessel kernel $\msf J_\nu(\xi,v)$ was defined for $\xi,v>0$ in \eqref{deff:BesselKernel}. Using the series expansion
$$
\Jb_0(\xi^{1/2})= \sum_{k=0}^\infty \frac{(-1)^k \xi^k}{4^k (k!)^2},
$$
it is straightforward to show that $\msf J_0(\xi,\xi)$ extends to an analytic function of $\xi\in \C$, and furthermore that the identity
\begin{equation}\label{eq:analyticextensionbmBJbessel}
[\bm\Delta_\xi \bm B(\xi)]_{21}=-2\pi \ii \msf J_0(-\xi,-\xi),\quad \xi>0
\end{equation}
holds true. In particular, with this identity we extend $[\bm\Delta_\xi \bm B(\xi)]_{21,-}$ to the whole real line.

\subsubsection{Construction of the local parametrix} Recall the scaled functions $\wt{\mcal F}_\gt$ and $\wt{\mcal P}_\gt$ introduced in \eqref{deff:wtcalPF}, set
$$
{\msf b}_{\bm \Upsilon}(\xi)\deff 
\begin{dcases}
\frac{1}{1+\ee^{\wt{\mcal P}_\gt(\xi)}}, & \xi\in \Gamma_0^0, \\
-\frac{1}{(1-\ee^{-\ii \wt{\mcal F}_\gt(\xi)}) (1+\ee^{\wt{\mcal P}_\gt(\xi)})}, & \xi\in \Gamma_1^0\cup \Gamma_2^0, \\
\frac{1}{1+\ee^{\wt{\mcal P}_\gt(\xi)}}-1, & \xi\in \Gamma_0^3, \\
-\frac{1}{(1-\ee^{\ii \wt{\mcal F}_\gt(\xi)})(1+\ee^{ \wt{\mcal P}_\gt(\xi)})}, & \xi\in \Gamma_4^0\cup \Gamma_5^0,
\end{dcases}
$$
and define
\begin{equation}\label{eq:BUpsilon}
\msf B_{\bm \Upsilon}(\xi)\deff \frac{1}{2\pi \ii}\int_{\Gamma^0}\frac{\msf b_{\bm \Upsilon}(v)}{v-\xi}\dd v,\quad \xi\in \C\setminus \Gamma^0.
\end{equation}
This scalar function $\msf B_{\bm \Upsilon}$ depends on $\gt$, $s$ and $y$, which we recall we are assuming that they satisfy \eqref{eq:scalingy}.

The behavior of $\wt{\mcal P}$ and $\wt{\mcal F}$ near $\infty$, as described by the relations \eqref{deff:wtcalPF} and Proposition~\ref{prop:canonicestPF}, ensure that $\msf b_{\bm \Upsilon}(v)$ decays exponentially fast as $v\to \infty$ along $\Gamma^0$, so the integration is well defined near $\infty$.
Care must be taken at $v=0$: when $\ee^{\ii \wh{\mcal F}_\gt(0)}=1$, corresponding to $\gt\msf F_\gt(\zeta_\gt(s))\in 2\pi\Z$, the factor $\msf b_{\bm \Upsilon}$ has a simple pole sitting in the contour $\Gamma^0$. In this case, we interpret the integral along each of the contours $\Gamma^0_1\cup\Gamma^0_2$ and $\Gamma^0_4\cup\Gamma^0_5$ in the sense of principal value. Since the density $\msf{b}_{\bm \Upsilon}$ is meromorphic along each of these contours, with only a simple pole at $v=0$, then this principal value integral is well defined as a finite complex number.

We collect some properties of this function.
\begin{lemma}\label{lem:propBnc}
Suppose that we are in the supercritical regime from Definition~\ref{deff:subregimescritical}. The function $\msf B_{\bm \Upsilon}$ satisfies the following properties.
\begin{enumerate}[(i)]
\item The function $\msf B_{\bm \Upsilon}$ has continuous boundary values $\msf B_{\bm \Upsilon,\pm}$ along $\Gamma^0\setminus \{0\}$, and they satisfy
$$
\msf B_{\bm \Upsilon,+}(\xi)=\msf B_{\bm \Upsilon,-}(\xi) \msf b_{\bm \Upsilon}(\xi),\quad \xi\in \Gamma^0\setminus \{0\}.
$$

\item As $\xi\to 0$, while $\gt,s$ are kept fixed,
\begin{equation}\label{eq:behBpcorigin}
\msf B_{\bm \Upsilon}(\xi)=
\begin{cases}
 \Boh(\log \xi), & \text{if } \ee^{\ii \wt{\mcal F}_\gt(0)}\neq 1, \\
 \Boh(\log \xi), & \text{if } \ee^{\ii \wt{\mcal F}_\gt(0)}= 1 \text{ and } |\arg \xi |\notin  [\frac{\pi}{3},\frac{2\pi}{3}], \\ 
 \Boh(\xi^{-1}), & \text{if } \ee^{\ii \wt{\mcal F}_\gt(0)}= 1 \text{ and } |\arg \xi |\in  [\frac{\pi}{3},\frac{2\pi}{3}], \\ 
\end{cases}
\end{equation}

\item As $\tau \to +\infty$,
$$
\msf B_{\bm \Upsilon}(\xi)=\Boh\left( y^2\right),
$$
with uniform error term for $\xi$ within compact subsets of $\C\setminus \{0\}$.
\end{enumerate}

\end{lemma}
\begin{proof}
The proof of (i) follows directly from Plemelj's formula.

For the proof of (ii), thanks to Proposition~\ref{prop:canonicestPF}, in this case we see that the matrix $\msf b_{\bm \Upsilon}$ is meromorphic along each of the contours $\Gamma^0_1\cup\Gamma^0_2$ and $\Gamma^0_4\cup\Gamma^0_5$, so that the limit in (i) corresponds to the limit of four Cauchy integrals, one along the contour $\Gamma_0^0=(0,+\infty)$, one along the contour $\Gamma_3^0=(-\infty,0)$, and the other two over $\Gamma^0_1\cup\Gamma^0_2$ and $\Gamma^0_4\cup\Gamma^0_5$. Along the latter two, each density is meromorphic, with a simple pole at $\zeta=0$ when $\ee^{\ii\gt\wt{\msf F}_\gt(0)}=1$, and no other poles. The property (ii) then follows from basic properties of Cauchy transforms and deformation of contours.

Part (iii) is based on a estimate of contours in the style of the classical steepest descent method, and we only indicate the steps. The very definition of $\wt{\mcal P}_\gt$ and $\wt{\mcal F}_\gt$ in \eqref{deff:wtcalPF} and Proposition~\ref{prop:canonicestPF} ensure that the integrand $\msf b_{\bm \Upsilon}$ is exponentially decaying for $\xi$ away from the origin, and this way for some fixed small neighborhood $D_\delta$ of the origin and some $\eta>0$, we can estimate
$$
\msf B_{\bm \Upsilon}(\xi)=\frac{1}{2\pi \ii}\int_{\Gamma^0\cap D_\delta} \frac{\msf b_{\bm \Upsilon}(v)}{v-\xi}\dd v +\Boh(\ee^{-\eta/y^2})\quad \text{as}\quad \gt\to \infty,
$$
where the error term is uniform for $\xi$ in compacts of $\C\setminus D_\delta$, and uniform for $s$, $y$ as in \eqref{eq:scalingy}.

Next, for the integral over $D_\delta$ we perform the change of variables $v\mapsto r=y^2v$. The expansion \eqref{eq:expwtmcPF} ensures the resulting integrand remains bounded as $\gt\to \infty$, whereas the change of differentials $\dd v=y^2\dd r$ gives the claimed error estimate.

\end{proof}

Recalling the sectors $\Omega_j^0$ displayed in Figure~\ref{Fig:ContourModel}, we finally introduce the needed local parametrix as
\begin{equation}\label{deff:PLUpsi}
\begin{aligned}
& \bm P_{\bm \Upsilon}(\xi)\deff \bm B(\xi)\bm L_{\bm \Upsilon}(\xi),\qquad \text{with }\bm B \text{ as in \eqref{deff:BesselB} and} \\
& \bm L_{\bm \Upsilon}(\xi)\deff 
\begin{cases}
(\bm I+\msf B_{\bm \Upsilon}(\xi)\bm E_{12}), & \xi\in \Omega_0^0\cup \Omega_1^0\cup\Omega_4^0\cup \Omega_5^0, \\
(\bm I+\bm E_{21})(\bm I+\msf B_{\bm \Upsilon}(\xi)\bm E_{12})(\bm I-(1+\ee^{\wt{\mcal P}_\gt(\xi)})\bm E_{21}), & \xi\in \Omega_2^0, \\
(\bm I-\bm E_{21})(\bm I+\msf B_{\bm \Upsilon}(\xi)\bm E_{12})(\bm I+(1+\ee^{\wt{\mcal P}_\gt(\xi)})\bm E_{21}), & \xi\in \Omega_3^0, 
\end{cases}
\end{aligned}
\end{equation}

The introduction of the factor $\msf B_{\bm \Upsilon}$ is done in order to accomplish 
\begin{equation}\label{eq:jumpPUpsi}
\bm P_{\bm \Upsilon,+}(\xi)=\bm P_{\bm \Upsilon,-}(\xi)\bm J_{{\bm V}}(\xi),\quad \xi\in \Gamma^0\setminus \{0\}.
\end{equation}
Indeed, this jump property follows from a straightforward calculation using Lemma~\ref{lem:propBnc}--(i), \eqref{eq:jumphatUpsilon} and \eqref{eq:jumpBessel}. 

Furthermore, as a consequence of \eqref{eq:behBpcorigin} we learn that as $\xi\to 0$,
\begin{equation}\label{eq:behPUpsiorigin}
\bm P_{\bm \Upsilon}(\xi)=
\begin{cases}
 \Boh(\log \xi), & \text{if } \ee^{\ii \wt{\mcal F}_\gt(0)}\neq 1, \\
 \Boh(\log \xi), & \text{if } \ee^{\ii \wt{\mcal F}_\gt(0)}= 1 \text{ and } |\arg \xi |\notin  [\frac{\pi}{3},\frac{2\pi}{3}], \\ 
 \Boh
\begin{pmatrix}
 1 & \xi^{-1} \\ 1 & \xi^{-1}   
\end{pmatrix}
 , & \text{if } \ee^{\ii \wt{\mcal F}_\gt(0)}= 1 \text{ and } |\arg \xi |\in  [\frac{\pi}{3},\frac{2\pi}{3}]. 
\end{cases}
\end{equation}

We also need to control the behavior of $\bm L_{\bm \Upsilon}$ for $\xi$ within a positive distance from the origin. Using \eqref{eq:expwtmcPF} we see that for any compact $K\subset \C\setminus \{0\}$, there exists $\eta=\eta_K>0$ independent of the other parameters, such that
$$
\re \wt{\mcal P}_\gt(\xi)\geq -\eta \frac{1}{y^2}|\xi|, \text{ for } \xi\in K\cap (\Omega_2^0\cup\Omega_3^0).
$$
With this estimate, we ensure the exponential decay of the factors involving $\wt{\mcal P}_\gt$ and $\wt{\mcal F}_\gt$ in \eqref{deff:PLUpsi}, and when combined with Lemma~\ref{lem:propBnc}--(iii) we obtain that
\begin{equation}\label{eq:boundLUpsilonboundary}
\bm L_{\bm \Upsilon}(\xi)=\bm I+\Boh\left(y^2\right),\quad \gt\to \infty,
\end{equation}
uniformly for $s,y$ as in \eqref{eq:scalingy}, and uniformly for $\xi$ within compact subsets of $\C\setminus \{0\}$.

\subsubsection{Small norm theory} To conclude the asymptotic analysis, we fix $\varepsilon>0$ and introduce
\begin{equation}\label{deff:RUpsi}
{\bm R}_{\bm\Upsilon}(\xi)\deff 
\begin{cases}
    \bm V_\gt(\xi)\bm B(\xi)^{-1}, & \xi\in \C\setminus (\Gamma^0\cup \overline D_\varepsilon(0)), \\
    \bm V_\gt(\xi)\bm P_{\bm \Upsilon}(\xi)^{-1}, & \xi\in D_\varepsilon(0)\setminus \Gamma^0. \\
\end{cases}
\end{equation}
With the contour
$$
\Gamma_{\bm R_{\bm \Upsilon}}\deff (\Gamma^0\cup \partial D_\varepsilon)\setminus D_\varepsilon,
$$
and where we orient $\partial D_\varepsilon$ in the clockwise direction, it follows that $\bm R_{\bm \Upsilon}$ satisfies the following RHP.

\begin{rhp}\label{rhp:RUpsi}
Find a $2\times 2$ matrix-valued function $\bm R_{\bm \Upsilon}$ with the following properties.
\begin{enumerate}[(1)]
\item $\bm R_{\bm \Upsilon}$ is analytic on $\C\setminus \Gamma_{\bm R_{\bm \Upsilon}}$.
\item The matrix $\bm R_{\bm \Upsilon}$ has continuous boundary values $\bm R_{\bm \Upsilon,\pm}$ along $\Gamma_{\bm R_{\bm \Upsilon}}$, and they are related by the jump relation $\bm R_{\bm \Upsilon,+}(\xi)=\bm R_{\bm \Upsilon,-}(\xi)\bm J_{\bm R_{\bm \Upsilon}}(\xi)$, $\xi\in \Gamma_{\bm R_{\bm \Upsilon}}$, with
\begin{equation}\label{eq:jumpRUpsi}
\bm J_{\bm R_{\bm \Upsilon}}(\xi)\deff 
\begin{dcases}
\bm I+\frac{\bm B(\xi)\bm E_{12}\bm B(\xi)^{-1}}{1+\ee^{\wt{\mcal P}_\gt(\xi)}}, & \xi\in \Gamma_0^{0}\setminus D_\varepsilon, \\
\bm I-\frac{\bm B(\xi)\bm E_{12}\bm B(\xi)^{-1}}{(1-\ee^{-\ii\wt{\mcal F}_\gt(\xi)})(1+\ee^{\wt{\mcal P}_\gt(\xi)})}, 
& \xi\in \Gamma_1^{0}\setminus D_\varepsilon, \\
\left(\bm I+\ee^{\wt{\mcal P}_\gt(\xi)}\bm B_+(\xi)\bm E_{21}\bm B_+(\xi)^{-1}\right)
\left(\bm I- \frac{\bm B_+(\xi)\bm E_{12}\bm B_+(\xi)^{-1}}{(1-\ee^{-\ii\wt{\mcal F}_\gt(\xi)})(1+\ee^{\wt{\mcal P}_\gt(\xi)})} \right)
, & \xi\in \Gamma_2^0\setminus D_\varepsilon, \\
\bm I+ \bm B_-(\xi)\left[\left(\frac{1}{1+\ee^{\wt{\mcal P}_\gt(\xi)}}-1\right)\bm E_{11}-\ee^{\wt{\mcal P}_\gt(\xi)}\bm E_{22}\right]\bm B_-(\xi)^{-1}, & \xi\in \Gamma_3^0\setminus D_\varepsilon, \\
\left(\bm I- \frac{\bm B_-(\xi)\bm E_{12}\bm B_-(\xi)^{-1}}{(1-\ee^{\ii\wt{\mcal F}_\gt(\xi)})(1+\ee^{\wt{\mcal P}_\gt(\xi)})} \right)
\left(\bm I+\ee^{\wt{\mcal P}_\gt(\xi)}\bm B_-(\xi)\bm E_{21}\bm B_-(\xi)^{-1}\right), & \xi\in \Gamma_4^0\setminus D_\varepsilon,\\
\bm I-\frac{\bm B(\xi)\bm E_{12}\bm B(\xi)^{-1}}{(1-\ee^{\ii\wt{\mcal F}_\gt(\xi)})(1+\ee^{\wt{\mcal P}_\gt(\xi)})}, 
 & w\in \Gamma_5^0\setminus D_\varepsilon, \\
\bm B(\xi)\bm L_{\bm \Upsilon}(\xi)\bm B(\xi)^{-1}, & w\in \partial D_\varepsilon.
\end{dcases}
\end{equation}
\item As $\xi\to \infty$, 
$$
\bm R_{\bm \Upsilon}(\xi)=\bm I+\Boh(\xi^{-1}).
$$
\end{enumerate}
\end{rhp}

The verification that $\bm R_{\bm \Upsilon}$ indeed solves this RHP is straightforward using its definition in \eqref{deff:RUpsi}, the RHP~\ref{rhp:bessel} for $\bm B$, and \eqref{deff:PLUpsi}, \eqref{eq:jumpPUpsi}, \eqref{eq:behPUpsiorigin}.

Finally, the key estimates to finish the current asymptotic analysis are provided by the next result.

\begin{prop}\label{prop:jumpconvRUP}
The estimate
$$
\left\| \bm J_{\bm R_{\bm \Upsilon}} -\bm I \right\|_{L^1\cap L^\infty(\Gamma_{\bm R_{\bm \Upsilon}})}=\Boh\left( y^2 \right)
$$
holds true as $\gt\to \infty$, uniformly for $s, y$ satisfying \eqref{eq:scalingy}.
\end{prop}

\begin{proof}

Throughout this analysis, we use that $\bm B(\xi)$, as well as its boundary values along $\Gamma_{\bm B}$, remain bounded on compact components of $\C\setminus \{0\}$. This fact combined with \eqref{eq:asymptBessel} shows that for any $\varepsilon>0$ the estimate
\begin{equation}\label{eq:basicBbesselboundary}
\bm B(\xi)=\Boh(1+|\xi|^{1/4})\ee^{\xi^{1/2}\sp_3}
\end{equation}
is valid uniformly for $\xi$ on $\C\setminus D_\varepsilon$.

We verify the decay of the jump in each component of the jump contour $\Gamma_{\bm R_{\bm\Upsilon}}$. 

For $\xi\in \partial D_\varepsilon$, which is compact, we estimate using \eqref{eq:basicBbesselboundary} that $\bm B(\xi)^{\pm 1}=\Boh(1)$ uniformly, and when combined with \eqref{eq:boundLUpsilonboundary} we see that
$$
\left\| \bm J_{\bm R_{\bm \Upsilon}} -\bm I \right\|_{L^1\cap L^\infty(\partial D_\varepsilon)}=\left\|\bm B( \bm L_{\bm \Upsilon}-\bm I)\bm B^{-1} \right\|_{L^1\cap L^\infty(\partial D_\varepsilon)}=\Boh(y^{2}).
$$

For the remaining jumps, we will now show exponential decay of the $L^1\cap L^\infty$ norms.

On $\Gamma_0^0\setminus D_\varepsilon$ we obtain from \eqref{eq:basicBbesselboundary}
$$
\bm B(\xi)\bm E_{12}\bm B(\xi)^{-1}=\Boh(w^{1/2}\ee^{2\xi^{1/2}}),
$$
and using \eqref{deff:wtcalPF} and Proposition~\ref{prop:canonicestPF}--(i) we also see that
$$
\frac{1}{1+\ee^{\wt{\mcal P}_\gt(\xi)}}=\Boh(\ee^{-\eta|\xi|/y^2}),
$$
uniformly for $y,s$ as in \eqref{eq:scalingy}. Such estimate ensures
$$
\left\| \bm J_{\bm R_{\bm \Upsilon}} -\bm I \right\|_{L^1\cap L^\infty(\Gamma_0^0\setminus D_\varepsilon})=\Boh(\ee^{-\eta/y^2}),
$$
for a possibly different $\eta>0$.
For the analysis on $\Gamma_1^0\setminus D_\varepsilon$, we need to consider three terms. The first term is
\begin{equation}\label{eq:boundBBesselGamma1}
\bm B(\xi)\bm E_{12}\bm B(\xi)^{-1}=\Boh(|\xi|^{1/2} \ee^{2\xi^{1/2}})
\end{equation}
as before. The second term is estimated as
\begin{equation}\label{eq:boundwtmcalPGamma1}
\frac{1}{1+\ee^{\wt{\mcal P}_\gt}(\xi)}=\Boh(\ee^{-\eta |\xi|/y^2}),
\end{equation}
where we used \eqref{deff:wtcalPF} and Proposition~\ref{prop:canonicestPF}--(i). 

Finally, for the last term, namely the one involving $\wt{\mcal F}_\gt$, we split into two cases, namely whether $|\xi|\geq y^2\gt \varepsilon$ or not. For the case $|\xi|\geq y^2\gt \varepsilon$ we use \eqref{deff:wtcalPF} again and apply Proposition~\ref{prop:canonicestPF}--(iii), obtaining
\begin{equation}\label{eq:boundwtFGamma1}
\frac{1}{1+\ee^{-\ii\wt{\mcal F}_\gt(\xi)}}=\Boh(\ee^{-\eta|\xi|/y^2}),
\end{equation}
for a possibly different constant $\eta>0$. Next, for the remaining case $\varepsilon\leq |\xi|\leq y^2\gt \varepsilon$, we see from \eqref{eq:expwtmcPF} that
$$
\re\left(\ii \wt{\mcal F}_\gt(\xi) \right)=\frac{1}{y^2}\msf F'\left( \frac{\zeta_\gt(s)}{\gt^2} \right)\left(1+\Boh\left( \frac{1}{\gt^2} \right)\right)\im \xi=\frac{\msf c_{\msf F}}{y^2}\left(1+\Boh\left( \frac{1}{\gt^2} \right)\right)\im \xi.
$$
Since $\msf c_{\msf F}>0$ and we are assuming also that $\xi\in \Gamma_1^0\setminus D_\varepsilon$, the bound above shows that \eqref{eq:boundwtFGamma1} extends from $|\xi|\geq y^2\gt \varepsilon$ to the whole set $\Gamma^1_0\setminus D_\varepsilon$. All in all, \eqref{eq:boundBBesselGamma1}, \eqref{eq:boundwtmcalPGamma1} and \eqref{eq:boundwtFGamma1} combined show that
$$
\left\| \bm J_{\bm R_{\bm \Upsilon}} -\bm I \right\|_{L^1\cap L^\infty(\Gamma^0_1\setminus D_\varepsilon)}=\Boh(\ee^{-\eta/y^2}),
$$
uniformly for $s,y$ as in \eqref{eq:scalingy}. In a similar manner, one can show that this bound also holds true along $L^1\cap L^\infty(\Gamma_j^0\setminus D_\varepsilon)$ with $j=2,3,4,5$, we skip details.
\end{proof}

With Proposition~\ref{prop:jumpconvRUP} established, we immediately obtain
\begin{theorem}\label{thm:smallnormRUp}
The estimates
$$
\|\bm R_{\bm \Upsilon}-\bm I\|_{L^\infty(\C\setminus \Gamma_{\bm R_{\bm \Upsilon}} )}=\Boh\left( y^2 \right),\quad \|\bm R_{\bm \Upsilon,\pm}-\bm I\|_{L^2(\Gamma_{\bm R_{\bm \Upsilon}})}=\Boh\left( y^2 \right),
$$
as well as
$$
\|\partial_\zeta\bm R_{\bm \Upsilon}\|_{L^\infty(\C\setminus \Gamma_{\bm R_{\bm \Upsilon}} )}=\Boh\left( y^2\right),\quad \|\partial_\zeta\bm R_{\bm \Upsilon,\pm}\|_{L^2(\Gamma_{\bm R_{\bm \Upsilon}})}=\Boh\left(y^2 \right),
$$
are valid as $\gt\to \infty$, uniformly for $s, y$ satisfying \eqref{eq:scalingy}.
\end{theorem}

Theorem~\ref{thm:smallnormRUp} finally completes all the steps in the asymptotic analysis of $\bm \Upsilon_\gt$ in the regime \eqref{eq:scalingy}.

\subsubsection{Summary of the asymptotic analysis of $\bm\Upsilon_\gt$} We now summarize the asymptotic formulas resulting from the asymptotic analysis we just carried out, and that will be needed for later on.

Next, we now assume that $y, s$ satisfy \eqref{eq:scalingy}. Unfolding the transformations $\bm\Upsilon_\gt\mapsto \bm V_\gt\mapsto \bm R_{\bm \Upsilon}$, we obtain that
\begin{equation}\label{eq:Upstauasymptoticlarge}
\bm \Upsilon_{\gt,-}(w)=y^{\sp_3/2}\bm R_{\bm \Upsilon,-}(\xi)\bm B_-(\xi),\quad \xi=y^2w,\quad \xi\in \R\setminus D_\varepsilon(0),
\end{equation}
as well as
\begin{equation}\label{eq:Upstauasymptotic}
\bm \Upsilon_{\gt,-}(w)=y^{\sp_3/2}\bm R_{\bm \Upsilon}(\xi)\bm B_-(\xi)\bm L_{\bm\Upsilon,-}(\xi),\quad \xi\in (-\varepsilon,\varepsilon),
\end{equation}
where we recall that $\bm L_{\bm\Upsilon}$ is given in \eqref{deff:PLUpsi}, and we emphasize that its expression depends on whether $\xi>0$ or $\xi<0$.
Recall that $\wt{\mcal H}$ was introduced in \eqref{deff:wtmcalH}. From the expressions \eqref{eq:Upstauasymptoticlarge}--\eqref{eq:Upstauasymptotic}, for a given $\delta>0$ sufficiently small and $0<w<\gt^3\delta $ we obtain
$$
\wt{\mcal H}_\gt(w)=y^2\left[\bm \Delta_\xi \bm B(\xi) \right]_{21} + y^2\left[ \bm B(\xi)^{-1}\bm\Delta_\xi \bm R_{\bm \Upsilon}(\xi)\bm B(\xi) \right]_{21,-},\quad \xi=y^2w.
$$
Using RHP~\ref{rhp:bessel}--(4), we re-express this identity as
\begin{equation}\label{eq:mcalHBB0Bessel}
\wt{\mcal H}_\gt(w)=y^2\left[\bm \Delta_\xi \bm B(\xi) \right]_{21} + y^2\left[ \bm B_0(\xi)^{-1}\bm\Delta_\xi \bm R_{\bm \Upsilon}(\xi)\bm B_0(\xi) \right]_{21,-},\quad \xi=y^2w,\quad 0<w<\gt^3\delta.
\end{equation}
All the terms in this identity extend analytically to $w<0$: this is true for $\bm R_{\bm \Upsilon}$ because RHP~\ref{rhp:RUpsi} ensures it has no jumps near $\xi=0$, the analyticity of $\bm B_0$ is part of RHP~\ref{rhp:bessel}--(4), the analyticity of the remaining term involving $\bm B=\bm B_-$ on the right-hand side is discussed in \eqref{eq:seriesBesselI1} {\it et seq.}, and the analyticity of $\wt{\mcal H}_\gt$ is ensured by Proposition~\ref{prop:fundHPsimod}. Thus, \eqref{eq:mcalHBB0Bessel} is in fact valid in a full growing neighborhood $(-\delta \gt^3,\delta \gt^3)$ of $w=0$.

\section{Asymptotic analysis of the model problem}\label{sec:DZModelProblem}

In this section we carry out the asymptotic analysis of the model problem, RHP~\ref{rhp:modelPhi}, in each of the regimes singled out in Definition~\ref{deff:subregimescritical}.

\subsection{Asymptotic analysis of the model problem in the critical regime}\label{sec:asymptregcrit}\hfill

We now carry out the asymptotic analysis of the model problem in the critical regime of Definition~\ref{deff:subregimescritical}. In this regime, the expansion \eqref{eq:zetasseries} ensures the existence of $M=M(T)>0$ depending only on $T$ for which
\begin{equation}\label{eq:boundzetagtrc}
|\zs|\leq M.
\end{equation}

In this regime, the asymptotic analysis as $\gt\to \infty$ does not require any further rescaling of $\bm \Phi_\gt$. As we will see, for $\zeta$ away from $\zs$, the model problem $\bm \Phi_\gt$ may be well approximated by the Painlevé XXXIV directly. Near the point $\zeta_\gt(s)$, we need to construct a local parametrix, which is the major technical step we overcome in this section.

\subsubsection{Construction of the global parametrix}
The global parametrix is
\begin{equation}\label{eq:deffPsisregcrit}
\bm G^\rc(\zeta)\deff \bm \Psi(\zeta-\zeta_\gt(s)\mid y=\zeta_\gt(s)),
\end{equation}
where $\bm \Psi$ is the solution to the PXXXIV RHP~\ref{rhp:PXXXIV}. This matrix is analytic on $\C\setminus (\Gamma^{\zeta_\gt(s)}_2\cup\Gamma^{\zeta_\gt(s)}_3\cup\Gamma^{\zeta_\gt(s)}_4)$, and satisfies the jump
$$
\bm G^\rc_+(\zeta)=\bm G^\rc_-(\zeta)\bm J_{\bm G^\rc}(\zeta), \quad \zeta\in \Gamma^{\zeta_\gt(s)}_2\cup \Gamma^{\zeta_\gt(s)}_3\cup \Gamma^{\zeta_\gt(s)}_4,  
$$
with
\begin{equation}\label{eq:jumpPsirc}
\bm J_{\bm G^\rc}(\zeta)\deff \bm J_{{\bm \Psi}}(\zeta-\zeta_\gt(s))=
\begin{cases}
    \bm I+\bm E_{21}, & \zeta\in \Gamma^{\zeta_\gt(s)}_2\cup \Gamma^{\zeta_\gt(s)}_4, \\
    \bm E_{12}-\bm E_{21}, & \zeta\in \Gamma^{\zeta_\gt(s)}_3.
\end{cases}
\end{equation}
For later use, we also record its asymptotics
\begin{equation}\label{eq:asymptPsihatinfinity2}
\bm G^\rc(\zeta)=\left(\bm I+\ii \left( q(\zeta_\gt(s))+\frac{\zeta_\gt(s)^2}{4} \right)\bm E_{21} \right)\left(\bm I+\Boh(\zeta^{-1})\right)\zeta^{-\sp_3/4}\bm U_0\ee^{-\frac{2}{3}\zeta^{3/2}\sp_3},
\end{equation}
see \eqref{eq:relationquPXXXIV} and \eqref{eq:asymptPXXXIVinfinity2}, as well as the behavior
\begin{equation}\label{eq:localbehhatPsi}
\bm G^\rc(\zeta)=\bm G^\rc_0(\zeta)\left(\bm I+\frac{1}{2\pi \ii}\log(\zeta-\zeta_\gt(s))\bm E_{21}\right)\bm M_j,\quad \text{as } \zeta\to \zeta_\gt(s) \text{ along } \Omega_j,
\end{equation}
where $\bm G^\rc_0(\zeta)={\bm \Psi}_0(\zeta-\zeta_\gt(s)\mid y=\zeta_\gt(s))$ is analytic near $\zeta=\zeta_\gt(s)$ and smooth in $y$, and with matrices $\bm M_1=\bm I,\bm M_j=\bm I+(-1)^j\bm E_{21}, j=1,2,$ see \eqref{eq:localbehP34RHP} and \eqref{eq:PXXXIVmon}.

\subsubsection{Construction of the local parametrix} The local parametrix needed in the critical regime is a matrix-valued function $\bm P^\rc$, which is piecewise analytic on a disk $D_R(\zeta_\gt(s))$ around $\zeta_\gt(s)$, with the same jumps as $\bm \Phi_\gt(\zeta)$ along the contour $\Gamma^{\zeta_\gt(s)}\cap D_R(\zeta_\gt(s))$, and asymptotically matching (a shift of) the solution $\bm \Psi$ to the PXXXIV RHP~\ref{rhp:PXXXIV} on the boundary of a neighborhood of $\zeta_\gt(s)$.

The constructive procedure to get to its form is a bit cumbersome, but based on several techniques common in Riemann-Hilbert literature. The basic idea is that, although the jump matrix for $\bm 
\Phi_\gt$ has non-trivial components in all its entries, it can essentially be brought to the form $\bm I+(\ast)\bm E_{12}$ after some formal manipulations. In the following, we just describe the final result of this procedure.

For the calculations that come next, recall that the contours $\Gamma_j^{\zeta_\gt(s)}$ and sectors $\Omega_j^{\zeta_\gt(s)}$ were introduced in \eqref{deff:Gammalambda} and \eqref{deff:Omegalambda}, respectively, and the functions $\msf F_\gt$ and $\msf P_\gt$ are given in \eqref{eq:PtauFtau}. As a construction similar to \eqref{eq:BUpsilon}, set
\begin{equation}\label{deff:jumpjrc}
\msf j_\gt(\zeta)\deff 
\begin{dcases}
\frac{1}{1+\ee^{\gt\msf P_\gt(\zeta)}}, & \zeta\in \Gamma_0^{\zeta_\gt(s)}, \\
-\frac{1}{(1-\ee^{-\ii \gt\msf F_\gt(\zeta)}) (1+\ee^{\gt\msf P_\gt(\zeta)})}, & \zeta\in \Gamma_1^{\zeta_\gt(s)}\cup \Gamma_2^{\zeta_\gt(s)}, \\
\frac{1}{1+\ee^{\gt\msf P_\gt(\zeta)}}-1, & \zeta \in \Gamma_3^{\zeta_\gt(s)}, \\
-\frac{1}{(1-\ee^{\ii \gt \msf F_\gt(\zeta)})(1+\ee^{\gt \msf P_\gt(\zeta)})}, & \zeta\in \Gamma_4^{\zeta_\gt(s)}\cup \Gamma_5^{\zeta_\gt(s)},
\end{dcases}
\end{equation}
and define
\begin{equation}\label{eq:deffmsfBrc}
\msf B^\rc(\zeta)\deff \frac{1}{2\pi \ii}\int_{\Gamma^\zs}\frac{\msf j_\gt(w)}{w-\zeta}\dd w,\quad \zeta\in \C\setminus \Gamma^{\zeta_\gt(s)}.
\end{equation}

The growth of $\re \msf P$ and $\im \msf F$ near $\infty$ as implied by Proposition~\ref{prop:extAdmFunc}--(iii),(iv), ensure that $\msf j_\gt(w)$ decays exponentially fast as $w\to \infty$ along $\Gamma^{\zeta_\gt(s)}$, so that the integrals defining $\msf B^\rc$ are absolutely convergent near $\infty$. 

Care must be taken at $\zeta=0$ when $\ee^{\ii \gt {\msf F}_\gt(\zeta_\gt(s))}=1$, corresponding to $\gt\msf F_\gt(\zeta_\gt(s))\in 2\pi\Z$. Recalling RHP~\ref{rhp:modelPhi}--(4), this latter condition is the same that produces a singular behavior of the model problem at $\zeta_\gt(s)$. In this case, we interpret the integral along each of the contours $\Gamma^{\zeta_\gt(s)}_1\cup\Gamma^{\zeta_\gt(s)}_2$ and $\Gamma_3^{\zeta_\gt(s)}\cup\Gamma^{\zeta_\gt(s)}_4$ in the sense of principal value. Since the density ${\msf j}_\gt$ is meromorphic along each of these contours, with only a simple pole at $\zeta=\zeta_\gt(s)$, this principal value integral is well defined as a finite complex number.

The next lemma is an analogue of Lemma~\ref{lem:propBnc}.
\begin{lemma}\label{lem:estBrc}
Suppose that we are in the critical regime from Definition~\ref{deff:subregimescritical}. The function $\msf B^\rc$ satisfies the following properties.
\begin{enumerate}[(i)]
\item For $\zeta\in \Gamma^{\zeta_\gt(s)}\setminus \{\zeta_\gt(s)\}$, the scalar function $\msf B^\rc$ satisfies
$$
\msf B^\rc_+(\zeta)=\msf B^\rc_-(\zeta)\msf j_\gt(\zeta).
$$

\item As $\zeta\to {\zeta_\gt(s)}$ while $\gt,s$ are kept fixed,
\begin{equation}\label{eq:behBrcorigin}
\msf B^\rc(\zeta)=
\begin{cases}
 \Boh(\log(\zeta-\zs)), & \text{if}\quad \gt \msf F_\gt(\zeta_\gt(s))\notin 2\pi \Z, \\
 \Boh(\log(\zeta-\zs)), & \text{if}\quad \gt \msf F_\gt(\zeta_\gt(s))\in 2\pi \Z \quad \text{and}\quad \zeta\in \Omega_j^{\zeta_\gt(s)}, \; j=0,2,3,5, \\ 
 \Boh\left((\zeta-\zeta_\gt(s))^{-1}\right), & \text{if}\quad  \gt \msf F_\gt(\zeta_\gt(s))\in 2\pi \Z \quad \text{and}\quad \zeta\in \Omega_1^{\zeta_\gt(s)}\cup\Omega_4^{\zeta_\gt(s)}. 
\end{cases}
\end{equation}

\item As $\tau \to +\infty$,
$$
\msf B^\rc(\zeta)=\Boh\left( \gt^{-1}\right)
$$
with uniform error term for $\zeta$ within compact subsets of $\C\setminus \{{\zeta_\gt(s)}\}$.
\end{enumerate}
\end{lemma}
\begin{proof}
The claimed property (i) is a consequence of Plemelj's formula.

Next, we prove (ii). Thanks to Proposition~\ref{prop:canonicestPF}, in this case we see that the matrix ${\msf j}_\gt$ is meromorphic along bounded arcs of each of the contours $\Gamma^0_0=[0,+\infty),\Gamma^0_3=(-\infty,0]$, $\Gamma^0_1\cup\Gamma^0_2$ and $\Gamma^0_4\cup\Gamma^0_5$, and it is $C^\infty$ along unbounded components. Thus, the limit in (i) corresponds to the limit of four Cauchy integrals near a point where each density is meromorphic, with a simple pole at $\zeta=\zeta_\gt(s)$ when $\ee^{\ii\gt {\msf F}_\gt(\zeta_\gt(s))}=1$, and no other poles. The property (ii) then follows from basic properties of Cauchy transforms and deformation of contours.

Finally, for the proof of (iii), we observe that Proposition~\ref{prop:canonicestPF} provides the estimate
$$
|\msf j_\gt(w)|=\Boh\left( \ee^{-\eta |w-\zs|} \right),
$$
uniformly for $w\in \Gamma^\zs$. Furthermore, from the very definition of $\msf P_\gt$ and $\msf F_\gt$ in \eqref{eq:PtauFtau}, it follows that given any compact $K$, there exists $\gt_0>0$ such that $\msf j_\gt(w)$ has an analytic continuation to a neighborhood of any $z\in K\cap \Gamma^\zs\setminus \{\zs\}$, for any $\gt\geq \gt_0$. Part (iii) then follows from standard arguments using the classical steepest descent method for oscillatory integrals, having in mind that the major contribution for each integral along $\Gamma^\zs_j$ will always come from the value $w=\zs$.
\end{proof}

Finally, with
$$
\bm B^\rc(\zeta)\deff \bm I+\msf B^\rc(\zeta)\bm E_{12}, \quad \zeta\in \C\setminus \Gamma^{\zeta_\gt(s)},
$$
and $\bm G^\rc$ as in \eqref{eq:deffPsisregcrit}, we introduce
\begin{equation}\label{deff:localparamRegimeII}
\bm P^\rc(\zeta)=\bm G^\rc(\zeta)\times
\begin{cases}
\bm B^\rc(\zeta), & \zeta\in \Omega_j^\zs,\; j=0,1,4,5, \\
(\bm I+\bm E_{21})\bm B^\rc(\zeta)(\bm I-(1+\ee^{\gt\msf P_\gt(\zeta)})\bm E_{21}), & \zeta\in \Omega_2^{\zeta_\gt(s)}, \\
(\bm I-\bm E_{21})\bm B^\rc(\zeta)(\bm I+(1+\ee^{\gt\msf P_\gt(\zeta)})\bm E_{21}), & \zeta\in \Omega_3^{\zeta_\gt(s)}.
\end{cases}
\end{equation}

The construction of $\bm P^\rc$ is as defined in \eqref{deff:localparamRegimeII} for the following reasons. First of all, thanks to the jump \eqref{eq:jumpPsirc} and the jump for $\bm B^\rc$ obtained from \eqref{deff:jumpjrc} and Lemma~\ref{lem:estBrc}--(i), the relation
\begin{equation}\label{eq:jumpsPrc}
\bm P^\rc_+(\zeta)=\bm P^\rc_-(\zeta){\bm J}_{\gt}(\zeta),\quad \zeta\in \Gamma^{\zeta_\gt(s)}\cap D_\delta(\zeta_\gt(s))\setminus \{\zeta_\gt(s)\},
\end{equation}
follows after a cumbersome but straightforward calculation, where we recall that $\bm J_\gt$ is the jump of the model problem given in \eqref{eq:jumpJt}.

Moreover, using \eqref{eq:localbehhatPsi} and \eqref{eq:behBrcorigin}, we obtain the behavior as $\zeta\to \zeta_\gt(s)$
\begin{equation}\label{eq:behPrcorigin}
\bm P^\rc(\zeta)=
\begin{dcases}
 \Boh((\log(\zeta-\zeta_\gt(s)))^2), & \text{ if } \gt \msf F_\gt(\zeta_\gt(s))\notin 2\pi \Z, \\
 \Boh((\log(\zeta-\zeta_\gt(s)))^2), & \text{ if } \gt \msf F_\gt(\zeta_\gt(s))\in 2\pi \Z  \text{ and } \zeta\in \Omega_j^{\zeta_\gt(s)}, j=0,2,3,5, \\ 
 \Boh\left(\frac{\log(\zeta-\zeta_\gt(s))}{\zeta-\zeta_\gt(s)}\right), & \text{ if } \gt \msf F_\gt(\zeta_\gt(s))\in 2\pi \Z  \text{ and } \zeta\in \Omega_1^{\zeta_\gt(s)}\cup\Omega_4^{\zeta_\gt(s)}. \\ 
\end{dcases}
\end{equation}

Furthermore, thanks to Lemma~\ref{lem:estBrc}--(iii) and the estimates for $\msf P_\gt,\msf F_\gt$ provided by Proposition~\ref{prop:canonicestPF}, we obtain that
$$
\bm P^\rc(\zeta)=\bm G^\rc(\zeta)\left(\bm I+\Boh(\gt^{-1})\right),
$$
uniformly for $\zeta$ in compacts of $\C\setminus \{\zeta_\gt(s)\}$, and uniformly for $s$ within the critical regime. Still in the critical regime, the factor $y=\zeta_\gt(s)$ remains bounded, and the matrix $\bm \Psi(\zeta)=\bm \Psi(\zeta\mid y)$ is also uniformly bounded for $\zeta$ in compacts of $\C\setminus 0$ and $y$ in compacts of the real line. Thus, the matrix $\bm G^\rc$ is also uniformly bounded for $\zeta$ in compacts of $\C\setminus \{\zeta_\gt(s)\}$ and $s$ within the critical regime. Hence, we can commute the terms in the right-hand side above, obtaining that
\begin{equation}\label{eq:regcritLocParEstBoundary}
\bm P^\rc(\zeta)=\left(\bm I+\Boh(\gt^{-1})\right)\bm G^\rc(\zeta),
\end{equation}
valid uniformly for $\zeta$ in compacts of $\C\setminus \{\zeta_\gt(s)\}$ and uniformly for $s$ within the critical regime.

\subsubsection{Small norm analysis} To conclude the asymptotic analysis, we fix $r>0$ and make
\begin{equation}\label{deff:Rrc}
\bm R^\rc(\zeta)\deff
\left(\bm I+\ii \left(q(\zeta_\gt(s)) +\frac{\zeta_\gt(s)^2}{4} \right)\bm E_{21}\right)\times 
\begin{cases}
\bm \Phi_\gt(\zeta) \bm P^\rc(\zeta)^{-1}, & \zeta\in D_r(\zeta_\gt(s))\setminus \Gamma^{\zeta_\gt(s)}, \\ 
\bm \Phi_\gt(\zeta) \bm G^\rc(\zeta)^{-1},& \zeta\in \C\setminus (\overline{D_r(\zeta_\gt(s))} \cup \Gamma^{\zeta_\gt(s)}),
\end{cases}
\end{equation}
where we recall that $q=q(y)$ is the same factor that appears in \eqref{eq:asymptPsihatinfinity2}.

Denote
$$
\Gamma^\rc\deff \left(\Gamma^{\zeta_\gt(s)}\cup \partial D_r(\zeta_\gt(s))\right)\setminus D_r(\zeta_\gt(s)),\quad \Gamma^\rc_j\deff \Gamma^{\zeta_\gt(s)}_j\setminus\overline D_r(\zeta_\gt(s)) =\Gamma^{\zeta_\gt(s)}_j\cap \Gamma^\rc,
$$
with $D_r(\zeta_\gt(s))$ being oriented in the clockwise direction.

The matrix $\bm R^\rc$ was defined as in \eqref{deff:Rrc} so that it is the solution to the following RHP.
\begin{rhp}\label{rhp:Rrc}
Find a $2\times 2$ matrix-valued function $\bm R^\rc$ with the following properties.
\begin{enumerate}[(1)]
\item $\bm R^\rc$ is analytic on $\C\setminus \Gamma^\rc$.
\item The matrix $\bm R^\rc$ has continuous boundary values $\bm R_\pm^\rc$ along $\Gamma^\rc$, and they are related by the jump relation $\bm R^\rc_+(\zeta)=\bm R^\rc_-(\zeta)\bm J^\rc(\zeta)$, $\zeta\in \Gamma^\rc$, with
\begin{multline}\label{eq:jumpRrc}
\bm J^\rc(\zeta)\deff  \\
\begin{dcases}
\bm G^\rc(\zeta)\bm P^\rc(\zeta)^{-1}, & \zeta\in \partial D_r(\zeta_\gt(s)), \\
\bm I+\frac{1}{1+\ee^{\gt\msf P_\gt(\zeta)}}\bm G^\rc(\zeta)\bm E_{12}\bm G^\rc(\zeta)^{-1}, & \zeta\in \Gamma_0^\rc, \\
\bm I-\frac{\bm G^\rc(\zeta)\bm E_{12}\bm G^\rc(\zeta)^{-1}}{(1-\ee^{-\ii \gt\msf F_\gt(\zeta)})(1+\ee^{\gt\msf P_\gt(\zeta)})} , & \zeta\in \Gamma_1^\rc, \\
\left(\bm I+\ee^{\gt \msf P_\gt(\zeta)}\bm G^\rc_+(\zeta)\bm E_{21}\bm G^\rc_+(\zeta)^{-1}\right)\left(\bm I-\frac{\bm G^\rc_+(\zeta)\bm E_{12}\bm G^\rc_+(\zeta)^{-1}}{(1-\ee^{-\ii \gt\msf F_\gt(\zeta)})(1+\ee^{\gt\msf P_\gt(\zeta)})}\right), & \zeta\in \Gamma_2^\rc, \\
\bm I+\bm G^\rc_{-}(\zeta)\left[\left(\frac{1}{1+\ee^{\gt\msf P_\gt(\zeta)}}-1\right)\bm E_{11}-\ee^{\gt\msf P_\gt(\zeta)}\bm E_{22}\right]\bm G^\rc_-(\zeta)^{-1}, & \zeta\in \Gamma_3^\rc , \\
\left(\bm I-\frac{\bm G^\rc_-(\zeta)\bm E_{12}\bm G^\rc_-(\zeta)^{-1}}{(1-\ee^{\ii \gt\msf F_\gt(\zeta)})(1+\ee^{\gt\msf P_\gt(\zeta)})}\right)\left(\bm I+\ee^{\gt \msf P_\gt(\zeta)}\bm G^\rc_-(\zeta)\bm E_{21}\bm G^\rc_-(\zeta)^{-1}\right), & \zeta\in \Gamma_4^\rc,\\
\bm I-\frac{\bm G^\rc(\zeta)\bm E_{12}\bm G^\rc(\zeta)^{-1}}{(1-\ee^{\ii \gt\msf F_\gt(\zeta)})(1+\ee^{\gt\msf P_\gt(\zeta)})}, & \zeta\in \Gamma_5^\rc.
\end{dcases}
\end{multline}
\item As $\zeta\to \infty$, 
$$
\bm R^\rc(\zeta)=\bm I+\Boh(\zeta^{-1}).
$$
\item The matrix $\bm R^\rc$ remains bounded near each point in the finite set $\partial D_r(\zeta_\gt(s))\cap \Gamma^{\zeta_\gt(s)}$.
\end{enumerate}
\end{rhp}

We now summarize how to obtain this RHP for $\bm R^\rc$. In what follows, recall that $\bm J_\gt$ is the jump matrix for the model problem $\bm \Phi_\gt$, whose explicit expression is found in \eqref{eq:jumpJt}.

The jumps for $\bm R^\rc$ outside $\overline D_r(\zeta_\gt(s))$ are obtained from the relations
$$
\bm J^\rc(\zeta)=\bm G^\rc_-(\zeta){\bm J}_\gt(\zeta)\bm J_{\bm G^\rc}(\zeta)^{-1}\bm G^\rc_-(\zeta)^{-1}=\bm G^\rc_+(\zeta)\bm J_{\bm G^\rc}(\zeta)^{-1}{\bm J}_\gt(\zeta)\bm G^\rc_+(\zeta)^{-1}.
$$
With these identities in mind, the matrix $\bm J^\rc$ along $\Gamma^\rc\setminus \overline D_r(\zeta_\gt(s))$ is computed from \eqref{eq:jumpPsirc} and \eqref{eq:jumpJt}. The jump on $\partial D_r(\zeta_\gt(s))$ is computed directly from \eqref{deff:Rrc}.

Both $\bm P^\rc$ and $\bm G^\rc$ have jumps on the arcs of $\Gamma^{\zeta_\gt(s)}\cap D_r(\zeta_\gt(s))$, so the matrix $\bm R^\rc$ may have jumps along these arcs as well. However, with the aid of the identity
$$
\bm R^\rc_-(\zeta)^{-1}\bm R^\rc_+(\zeta)=\bm P^\rc_-(\zeta){\bm J}_\gt(\zeta)\bm J_{\bm R^\rc}(\zeta)^{-1}\bm P^\rc_-(\zeta)^{-1},\quad \zeta\in \Gamma^{\zeta_\gt(s)}\cap D_r(\zeta_\gt(s))\setminus \{\zeta_\gt(s)\},
$$
and \eqref{eq:jumpsPrc}, we see that $\bm R^\rc$ does not have jumps inside $D_r(\zeta_\gt(s))$, and therefore the point $\zeta=\zeta_\gt(s)$ is an isolated singularity of $\bm R^\rc$.  A combination of \eqref{eq:behPrcorigin} with the local behavior of $\bm \Phi_\gt$ as $\zeta\to \zeta_\gt(s)$ ensures that $\zeta_\gt(s)$ is not an essential singularity. Finally, sending $\zeta\to \zeta_\gt(s)$ along $\Omega_0^{\zeta_\gt(s)}$ then ensures that this singularity is in fact removable. This argument shows that $\bm R^\rc$ is analytic inside $D_r(\zeta_\gt(s))$.

Also, with \eqref{eq:asymptPhit} and \eqref{eq:asymptPsihatinfinity2} in mind, it is straightforward to show RHP~\ref{rhp:Rrc}--(3). Finally, the condition on RHP~\ref{rhp:Rrc}--(4) follows because $\bm P^\rc$, $\bm \Phi_\gt$, $\bm G^\rc$ and their boundary values are all bounded away from $\zeta_\gt(s)\notin \partial D_r(\zeta_\gt(s))\cap \Gamma^{\zeta_\gt(s)}$.

To conclude the asymptotic analysis, we now show that $\bm R^\rc$ is indeed close to the identity matrix.

\begin{prop}\label{prop:estJrc}
    The estimate
    $$
    \|\bm J^\rc-\bm I\|_{L^1\cap L^\infty(\Gamma^\rc)}=\Boh(\gt^{-1})
    $$
    holds true as $\gt\to \infty$, uniformly for $s$ within the critical regime from Definition~\ref{deff:subregimescritical}.
\end{prop}

Along the way of the proof of Proposition~\ref{prop:estJrc}, and also later, we will need to estimate terms involving $\bm G^\rc$. We use the fact that $\bm G^\rc$ and its boundary values $\bm G^\rc_\pm$ are uniformly bounded for $\zeta$ in compacts away from $\zeta=\zeta_\gt(s)\notin \C\setminus \overline D_r(\zeta_\gt(s))$ and $s$ within the critical regime. For $\zeta$ in unbounded components, we obtain estimates for $\bm G^\rc$ using \eqref{eq:asymptPsihatinfinity2}. All in all, and having in mind that $q(y)$ is bounded for $y=\zeta_\gt(s)$ in compacts, it means that the estimates
\begin{equation}\label{eq:fundboundPXXXIV}
\bm G^\rc(\zeta)=\Boh\left(\zeta^{1/4}\right)\ee^{-\frac{2}{3}\zeta^{3/2}\sp_3}\quad \text{and}\quad \bm G^\rc(\zeta)^{-1}=\Boh\left(\zeta^{1/4}\right)\ee^{\frac{2}{3}\zeta^{3/2}\sp_3}
\end{equation}
are valid uniformly for $\zeta\in \Gamma^{\zeta_\gt(s)}\setminus D_r(\zeta_\gt(s))$  (including for unbounded $\zeta$), and also uniformly for $s$ within the critical regime. The same bound also holds true for $\bm G^\rc_{\pm}$. 

\begin{proof}[Proof of Proposition~\ref{prop:estJrc}]
We estimate the jumps on \eqref{eq:jumpRrc} case by case.

The estimate \eqref{eq:regcritLocParEstBoundary} immediately implies that $\bm J^\rc=\bm I+\Boh(\gt^{-1})$ in $L^1\cap L^\infty(\partial D_r(\zeta_\gt(s)))$.

Next, for $j=1,2$ we use Proposition~\ref{prop:canonicestPF}--(iii) and \eqref{eq:fundboundPXXXIV} to obtain
$$
\frac{\bm G_+^\rc(\zeta)\bm E_{12}\bm G_+^\rc(\zeta)^{-1}}{1-\ee^{-\ii\gt\msf F_\gt(\zeta)}}=\Boh(\zeta^{1/2})\frac{\ee^{\ii \gt \msf F_\gt(\zeta)-\frac{4}{3}\zeta^{3/2}}}{1-\ee^{\ii \gt\msf F_\gt(\zeta)}} =\Boh(\zeta^{1/2}\ee^{-\eta \gt |\zeta-\zs|}),\quad \zeta\in \Gamma_j^\zs\setminus \overline D_r(\zs),
$$
and when combined with the estimates from Proposition~\ref{prop:canonicestPF}--(i),(ii) we get that $\bm J^\rc=\bm I+\Boh(\ee^{-\eta\gt})$ in $L^1\cap L^\infty(\Gamma_j^\zs\setminus \overline D_r(\zs))$, $j=1,2$, for some new value $\eta>0$. In a completely similar manner, the same estimate also holds true in $L^1\cap L^\infty(\Gamma_j^\zs\setminus \overline D_r(\zs))$ for $j=4,5$.

Finally, the estimate $\bm J^\rc=\bm I+\Boh(\ee^{-\eta\gt})$ in $L^1\cap L^\infty(\Gamma_3^\zs\setminus \overline D_r(\zs))=(-\infty,\zs-r)$ is a direct consequence of \eqref{eq:fundboundPXXXIV} and Proposition~\ref{prop:canonicestPF}--(i).
\end{proof}

As a consequence of the estimate for $\bm J^\rc$, we obtain
\begin{theorem}\label{thm:smallnormregcrit}
The estimates
$$
\|\bm R^\rc-\bm I\|_{L^\infty(\C\setminus \Gamma^\rc )}=\Boh(\gt^{-1}),\quad \|\bm R^\rc_\pm-\bm I\|_{L^2(\Gamma^\rc)}=\Boh(\gt^{-1})
$$
as well as
$$
\|\partial_\zeta\bm R^\rc\|_{L^\infty(\C\setminus \Gamma^\rc )}=\Boh(\gt^{-1}),\quad \|\partial_\zeta\bm R^\rc_\pm\|_{L^2(\Gamma^\rc)}=\Boh(\gt^{-1})
$$
are valid as $\gt\to \infty$, uniformly for $s$ within the critical regime from Definition~\ref{deff:subregimescritical}.
\end{theorem}
\begin{proof}
The result is a consequence of Proposition~\ref{prop:estJrc} and the general small norm theory for RHPs, we skip details.
\end{proof}

\subsubsection{Unwrap of the transformations in the critical regime} Recall that the operator $\bm\Delta_z$ in an arbitrary variable $z$ was introduced in \eqref{deff:Deltaoper}.

Unwrapping the transformations, we obtain the following identities. For any $r>0$, we obtain that the equality
\begin{equation}\label{eq:unwrapPhircout}
\bm \Phi_{\gt,-}(\zeta)=\left(\bm I-\ii \left(q(\zs) +\frac{\zs^2}{4} \right)\bm E_{21}\right)\bm R^\rc_-(\zeta)\bm\Psi_-(\zeta-\zs\mid y=\zs), 
\end{equation}
holds true for $\zeta\in \R\setminus (\zs-r,\zs+r)$. Additionally, the identities
\begin{equation}\label{eq:unwrapPhircinaux1}
\bm \Phi_{\gt,-}(\zeta)=\left(\bm I-\ii \left(q(\zs) +\frac{\zs^2}{4} \right)\bm E_{21}\right)\bm R^\rc_-(\zeta)\bm\Psi_-(\zeta-\zs\mid y=\zs)(\bm I+\msf B^\rc(\zeta)\bm E_{12})
\end{equation}
for $\zs<\zeta<\zs+r$, and
\begin{multline}\label{eq:unwrapPhircinaux2}
\bm \Phi_{\gt,-}(\zeta)=\left(\bm I-\ii \left(q(\zs) +\frac{\zs^2}{4} \right)\bm E_{21}\right)\bm R^\rc_-(\zeta)\bm\Psi_-(\zeta-\zs\mid y=\zs)(\bm I-\bm E_{21})\\ 
\times (\bm I+\msf B^\rc_-(\zeta)\bm E_{12})(\bm I+(1+\ee^{\gt\msf P_\gt(\zeta)})\bm E_{21}), 
\end{multline}
for $\zs-r<\zeta<\zs$, also hold.

We use these identities to obtain rough bounds for $\bm\Phi_\gt$ and $\bm \Delta_\zeta\bm\Phi_\gt$, as well as a suitable representation of the latter.

Starting with the former, observing that $\zs$ remains bounded for $s$ within the critical regime, and using Theorem~\ref{thm:smallnormnegcrit}, we estimate
$$
\left(\bm I-\ii \left(q(\zs) +\frac{\zs^2}{4} \right)\bm E_{21}\right)\bm R^\rc_-(\zeta)=\Boh(1),\quad \gt\to \infty,
$$
valid uniformly for $\zeta\in \R$. Equation~\eqref{eq:unwrapPhircout} updates to
\begin{equation}\label{eq:unwrapPhircinaux3}
\bm \Phi_{\gt,-}(\zeta)=\Boh(1)\bm\Psi_-(\zeta-\zs\mid y=\zs),
\end{equation}
valid as $\zeta\to \infty$ uniformly for $|\zeta-\zs|>r$.

Next, we use the behavior \eqref{eq:localbehP34RHP} of $\bm \Psi$, and the boundedness of $\bm \Psi_0$ for $y=\zs$ bounded, and express \eqref{eq:unwrapPhircinaux1} and \eqref{eq:unwrapPhircinaux2} in an unified way as
\begin{equation}\label{eq:unwrapPhircinaux4}
\bm \Phi_{\gt,-}(\zeta)=\Boh(1)\left[\bm I+\left(\msf B^\rc(\zeta)+\frac{1}{2\pi\ii}\log\zeta\right)\bm E_{12}\right]\left(\bm I+\chi_{(-r,0)}(\zeta-\zs)(1+\ee^{\gt\msf P_\gt(\zeta)})\bm E_{21}\right),
\end{equation}
valid as $\gt \to\infty$ uniformly for $|\zeta-\zs|<r$.

Accounting for the asymptotic behavior \eqref{eq:asymptPXXXIVinfinity}, we summarize \eqref{eq:unwrapPhircinaux3} and \eqref{eq:unwrapPhircinaux4} in the more compact form
\begin{multline}\label{eq:unwrapPhircinaux5}
\bm \Phi_{\gt,-}(\zeta)=\Boh(1+|\zeta|^{1/4})\ee^{-(\frac{2}{3}\zeta_-^{3/2}+\zs\zeta_-^{1/2})\sp_3} \\ 
\times\left[\bm I+\chi_{(-r,r)}(\zeta-\zs)\left(\msf B^\rc(\zeta)+\frac{1}{2\pi\ii}\log\zeta\right)\bm E_{12}\right]\left(\bm I+\chi_{(-r,0)}(\zeta-\zs)(1+\ee^{\gt\msf P_\gt(\zeta)})\bm E_{21}\right),
\end{multline}
which is valid as $\gt\to \infty$, uniformly for $\zeta\in \R$ and uniformly for $\zs$ in the critical regime.

Next, we move to estimates for $\bm \Delta_\zeta\bm\Phi_\gt$. A direct calculation from \eqref{eq:unwrapPhircout} yields
\begin{equation}\label{eq:DeltaPhigtrc1}
\bm \Delta_\zeta\bm\Phi_{\gt,-}(\zeta+\zs)=\bm\Delta_\zeta \bm\Psi_-(\zeta\mid \zs)+\bm \Psi_-(\zeta\mid \zs)^{-1}\bm\Delta_\zeta\bm R_-^\rc(\zeta+\zs)\bm\Psi_-(\zeta\mid \zs),
\end{equation}
valid for $\zeta\in \R\setminus (-r,r)$.

Thanks to Theorem~\ref{thm:smallnormregcrit}, we obtain that $\bm\Delta_\zeta\bm R_-(\zeta+\zs)=\Boh(\gt^{-1})$ uniformly for $\zeta\in \R\setminus (-r,r)$. Furthermore, due to the fact that $y=\zs$ remains bounded for $s$ within the critical regime, the matrix $\bm \Psi_-$ is bounded for $\zeta$ in compact subsets of $\R\setminus (-r,r)$, also uniformly in $s$ within the critical regime. Finally, as $\zeta\to \infty$, the factor $\bm \Psi_-$ and its derivative that appear above can be estimated from \eqref{eq:asymptPXXXIVinfinity}. All in all, these arguments show that for some $\eta>0$, the rough estimate
$$
\bm \Delta_\zeta\bm\Psi_-(\zeta)=\ee^{\left(\frac{2}{3}-\eta\right)\zeta^{3/2}\sp_3}\Boh(\zeta)\ee^{-\left(\frac{2}{3}-\eta\right)\zeta^{3/2}\sp_3},
$$
valid uniformly for $\zeta\in \R\setminus (-r,r)$ and uniformly for $s$ within the critical regime as $\gt\to \infty$. 

Thus, using this bound in \eqref{eq:DeltaPhigtrc1}, we obtain
\begin{equation}\label{eq:asymptailsPhirc}
\bm\Delta_\zeta\bm\Phi_{\gt,-}(\zeta+\zs)=\ee^{\left(\frac{2}{3}-\eta\right)\zeta^{3/2}\sp_3}\Boh(\zeta)\ee^{-\left(\frac{2}{3}-\eta\right)\zeta^{3/2}\sp_3},
\end{equation}
which again is valid uniformly for $\zeta\in \R\setminus (-r,r)$ and uniformly for $s$ within the critical regime as $\gt\to \infty$. Such asymptotic formula may be differentiated term by term.

Next, Equation~\eqref{eq:unwrapPhircinaux1} yields 
$$
\left[ \bm\Delta_\zeta \bm\Phi_{\gt,-}(\zeta+\zs) \right]_{21}=\left[ \bm\Delta_\zeta \bm\Psi(\zeta\mid \zs)\right]_{21}+\left[\bm\Psi_-(\zeta\mid \zs)^{-1}\bm\Delta_\zeta\bm R^\rc(\zeta+\zs)\bm\Psi_-(\zeta\mid \zs)\right]_{21},
$$
for $\zeta\in (0,r)$. The first term on the right-hand side coincides with \eqref{deff:mcalHPXXXIV}. The second term can be written in terms of the analytic function $\bm\Psi_0$ from \eqref{eq:localbehhatPsi}, where we recall that therein $\bm M_0=\bm I$, see \eqref{eq:PXXXIVmon}. In summary, we then obtain the identity
\begin{equation}\label{eq:PhigtmcalHPsi0}
\left[ \bm\Delta_\zeta \bm\Phi_{\gt,-}(\zeta+\zs) \right]_{21}=\mcal H(\zeta\mid \zs)+\left[\bm\Psi_0(\zeta\mid \zs)^{-1}\bm\Delta_\zeta\bm R^\rc(\zeta+\zs)\bm\Psi_0(\zeta\mid \zs)\right]_{21},
\end{equation}
which is valid for $\zeta\in (0,r)$, and where we emphasize that $r>0$ can be made arbitrarily small. Observe that the terms on the right-hand side are all analytic on a full disk $D_r(0)$, fact which will be useful later.

The asymptotic formulas \eqref{eq:unwrapPhircinaux5} and \eqref{eq:asymptailsPhirc} and the identity \eqref{eq:PhigtmcalHPsi0} with the asymptotic decay of $\bm R^\rc$ ensured by Theorem~\ref{thm:smallnormregcrit} are the main outcomes that we will use later.

\subsection{Asymptotic analysis of the model problem in the subcritical regime}\label{sec:asymptrposcrit} \hfill   

As the next step we carry out the asymptotic analysis in the subcritical regime as stated in Definition~\ref{deff:subregimescritical}. Under this regime, \eqref{eq:zetasseries} ensures the existence of $M=M(T)>0$ depending only on $T$ for which
\begin{equation}\label{eq:boundzetagtpc}
M\leq \zs \leq \frac{1}{M}\gt^{1/2}.
\end{equation}
We use these bounds on $\zeta_\gt(s)$ extensively in the analysis that follows.

\subsubsection{Scaling step} As a first step for the asymptotic analysis, we make a scaling of the model problem. Denote
\begin{equation}\label{eq:shiftedPFPC}
\wh{\msf P}_\gt(\zeta)\deff \msf P_\gt\left(\zeta_\gt(s)\zeta\right),\quad \wh{\msf F}_\gt(\zeta)\deff \msf F_\gt\left(\zeta_\gt(s)\zeta\right),
\end{equation}
and set
$$
\bm \Phi^\pc_\gt(\zeta)\deff \zeta_\gt(s)^{\sp_3/4}\bm\Phi_\gt(\zeta_\gt(s)\zeta).
$$
Then the matrix-valued function $\bm\Phi_\gt^\pc$ satisfies a RHP which is obtained from a shift and scaling of the RHP~\ref{rhp:modelPhi} for $\bm \Phi_\gt$, namely the following.

\begin{rhp}\label{rhp:modelPhiPC}
Find a $2\times 2$ matrix-valued function $\bm \Phi^\pc_\gt$ with the following properties.
\begin{enumerate}[(1)]
\item $\bm \Phi^\pc_\gt$ is analytic on $\C\setminus \Gamma^1$.
\item The matrix $\bm \Phi^\pc_\gt$ has continuous boundary values $\bm \Phi^\pc_{\gt,\pm}$ along $\Gamma^1 \setminus \{1\}$, and they are related by $\bm \Phi^\pc_{\gt,+}(\zeta)=\bm \Phi^\pc_{\gt,-}(\zeta)\wh{\bm J}_{\gt}(\zeta)$, $\zeta\in \Gamma^{1}\setminus \{1\}$, where
\begin{equation}\label{eq:jumpJpc}
\wh{\bm J}_{\gt}(\zeta)\deff
\begin{dcases}
\bm I+\frac{1}{1+\ee^{\gt\wh{\msf P}_\gt(\zeta)}}\bm E_{12}, & \zeta\in \Gamma_0^{1}, \\
\bm I-\frac{1}{(1-\ee^{-\ii \gt\wh{\msf F}_\gt(\zeta)})(1+\ee^{\gt\wh{\msf P}_\gt(\zeta)})} \bm E_{12}, & \zeta\in \Gamma_1^{1}, \\
\left(\bm I+(1+\ee^{\gt\wh{\msf P}_\gt(\zeta)})\bm E_{21}\right)\left(\bm I-\frac{1}{(1-\ee^{-\ii \gt\wh{\msf F}_\gt(\zeta)})(1+\ee^{\gt\wh{\msf P}_\gt(\zeta)})} \bm E_{12}\right), & \zeta\in \Gamma_2^1, \\
\frac{1}{1+\ee^{\gt\wh{\msf P}_\gt(\zeta)}}\bm E_{12}-(1+\ee^{\gt\wh{\msf P}_\gt(\zeta)})\bm E_{21}, & \zeta\in \Gamma_3^1, \\
\left(\bm I-\frac{1}{(1-\ee^{\ii \gt\wh{\msf F}_\gt(\zeta)})(1+\ee^{\gt\wh{\msf P}_\gt(\zeta)})} \bm E_{12}\right)\left(\bm I+(1+\ee^{\gt\wh{\msf P}_\gt(\zeta)}) \bm E_{21}\right), & \zeta\in \Gamma_4^1,\\
\bm I-\frac{1}{(1-\ee^{\ii \gt\wh{\msf F}_\gt(\zeta)})(1+\ee^{\gt\wh{\msf P}_\gt(\zeta)})} \bm E_{12}, & \zeta\in \Gamma_5^1.
\end{dcases}
\end{equation}
\item As $\zeta\to \infty$, 
\begin{equation}\label{eq:asymptPhipc}
{\bm \Phi}^\pc_\gt(\zeta)=\left(\bm I+\Boh(\zeta^{-1})\right)\zeta^{-\sp_3/4}\bm U_0\ee^{-\frac{2}{3}\zeta_\gt(s)^{3/2}\zeta^{3/2}\sp_3}.
\end{equation}

\item As $\zeta\to 1$,
$$
{\bm \Phi}^\pc_\gt(\zeta)=
\begin{cases}
    \Boh(1), & \text{if } \ee^{\ii \gt\wh{\msf F}_\gt(1)}\neq 1, \\
    \Boh(1), & \text{if } \ee^{\ii \gt\wh{\msf F}_\gt(1)}=1  \text{ and } \zeta\notin \Omega_1^1\cup \Omega_1^4, \\
    \Boh\begin{pmatrix}
         1 & (\zeta-1)^{-1} \\  1 & (\zeta-1)^{-1}
    \end{pmatrix}, & \text{if } \ee^{\ii \gt\wh{\msf F}_\gt(1)}=1  \text{ and } \zeta\in \Omega_1^1\cup \Omega_1^4.
\end{cases}
$$
\end{enumerate}
\end{rhp}

To complete the asymptotic analysis for $\bm \Phi^\pc_\gt$ we need to construct a global parametrix.

\subsubsection{The global parametrix} For the construction of the global parametrix we will use the model problem $\wh{\bm \Psi}$ which solves RHP~\ref{rhp:modelhatPsi}, setting
\begin{equation}\label{deff:Gpc}
\bm G^\pc(\zeta)\deff \zs^{\sp_3/4}\wh{\bm \Psi}(\zs \zeta),\quad \zeta\in \C\setminus (\R\cup \Gamma^1_2\cup\Gamma^1_4).
\end{equation}

From the RHP for $\wh{\bm\Psi}$ we obtain the following properties for $\bm G^\pc$. First of all, $\bm G^\pc$ is analytic on $\C\setminus (\R\cup \Gamma_3^1\cup\Gamma_4^1)$. Along this discontinuity contour it possesses the jump relations
\begin{equation}\label{eq:jumpGpc}
\bm G^\pc_+(\zeta)=\bm G^\pc_-(\zeta)\bm J_{\bm G^\pc}(\zeta),\quad 
\bm J_{\bm G^\pc}(\zeta)\deff 
\begin{dcases}
\bm I+\frac{1}{1+\ee^{\gt\wh{\msf P}_\gt(\zeta)}}\bm E_{12}, & \zeta\in \Gamma_1^{0}=(1,+\infty),\\
\bm I+(1+\ee^{\gt\wh{\msf P}_\gt(\zeta)})\bm E_{21}, & \zeta\in \Gamma_2^1\cup \Gamma_4^1, \\
\frac{1}{1+\ee^{\gt\wh{\msf P}_\gt(\zeta)}}\bm E_{12}-(1+\ee^{\gt\wh{\msf P}_\gt(\zeta)})\bm E_{21}, & \zeta\in \Gamma_3^1=(-\infty,1).
\end{dcases}
\end{equation}
Note that this jump matrix coincides with the jump matrix \eqref{eq:jumpJpc} along $\R$, and along $\Gamma^1_2\cup\Gamma^1_4$ this jump matrix appears as one of the factors in \eqref{eq:jumpJpc}.

Next, RHP~\ref{rhp:modelhatPsi}--(3) yields the behavior,
\begin{equation}\label{eq:asympGpc}
\bm G^\pc(\zeta)=\left(\bm I+\Boh(\zeta^{-1})\right)\zeta^{-\sp_3/4}\bm U_0\ee^{-\frac{2}{3}\zs^{3/2}\zeta^{3/2}},\quad \zeta\to \infty.
\end{equation}
whereas RHP~\ref{rhp:modelhatPsi}--(4) gives that the estimate
\begin{equation}\label{eq:behGPCorigin}
\bm G^\pc(\zeta)=\Boh(1),\quad \zeta\to 1,
\end{equation}
is valid.

Finally, given any $R>0$, we can always choose $\varepsilon>0$ to make sure that $|\zs \zeta|\geq R$ for every $\zeta\in \C$ with $|\zeta|\geq \varepsilon$ and for every $s$ within the subcritical regime. In particular, we can ensure that $\zs \zeta$ falls within the neighborhood at $\infty$ where the expansion \eqref{eq:asymptildePsi} is valid, and consequently we estimate
\begin{equation}\label{eq:behGPCtau}
\bm G^\pc_\pm(\zeta)=\left(\bm I+\Boh\left(\frac{1}{\zs^{1/2}}\right)\right)\zeta^{-\sp_3/4} \bm U_0 \ee^{-\frac{2}{3}\zs^{3/2}\zeta^{3/2}\sp_3},\quad \gt\to \infty,
\end{equation}
uniformly for $\zeta\in \Gamma^1\setminus D_\varepsilon$, uniformly also for $s$ within the subcritical regime. 

As a final remark, the matrix $\wh{\bm \Psi}$ is bounded on bounded subsets of the plane: this follows from the fact that it is approximated by $\wh\bai$ (see for instance Theorem~\ref{thm:smallnormhatR}). Thus, the rough estimate
$$
\bm G^\pc(\zeta)=\Boh(1+\zeta^{1/4})\ee^{-\frac{2}{3}\zs^{3/2}\zeta^{3/2}\sp_3},\quad \gt\to \infty,
$$
is also valid, now uniformly for $\zeta\in \C$, including for boundary values along $\Gamma^1$, and uniformly also for $s$ within the subcritical regime.

\subsubsection{Conclusion of the asymptotic analysis in the subcritical regime} Moving towards the end of the asymptotic analysis in the subcritical regime, introduce
\begin{equation}\label{eq:deffRpc}
\bm R^\pc(\zeta)\deff \bm \Phi^\pc_\gt (\zeta)\bm G^\pc(\zeta)^{-1},\quad \zeta\in \C\setminus \Gamma^1,\quad \zeta\in \C\setminus \Gamma^1.
\end{equation}

Set
$$
\Gamma^\pc\deff \Gamma^1\setminus \R.
$$
Combining RHP~\ref{rhp:modelPhiPC} with the properties of $\bm G^\pc$, we obtain that $\bm R^\pc$ solves the following RHP.

\begin{rhp}\label{rhp:Rpc}
Find a $2\times 2$ matrix-valued function $\bm R^\pc$ with the following properties.
\begin{enumerate}[(1)]
\item $\bm R^\pc$ is analytic on $\C\setminus \Gamma^\pc$.
\item The matrix $\bm R^\pc$ has continuous boundary values $\bm R_\pm^\pc$ along $\Gamma^\pc$, and they are related by the jump relation $\bm R^\pc_+(\zeta)=\bm R^\pc_-(\zeta)\bm J^\pc(\zeta)$, $\zeta\in \Gamma^\pc$, with
\begin{equation}\label{eq:jumpRPC1}
\bm J^\pc(\zeta)\deff 
\bm I-\frac{\bm G_\pm^\pc(\zeta)\bm E_{12} \bm G_\pm^\pc(\zeta)^{-1}}{(1-\ee^{\mp \ii \gt\wh{\msf F}_\gt(\zeta)})(1+\ee^{\gt\wh{\msf P}_\gt(\zeta)})} , \quad \zeta\in \Gamma^\pc \text{ with }\pm \re\zeta>0.
\end{equation}
\item As $\zeta\to \infty$, 
$$
\bm R^\pc(\zeta)=\bm I+\Boh(\zeta^{-1}).
$$
\item As $\zeta\to 1$,
$$
\bm R^\pc(\zeta)=
\begin{cases}
    \Boh(1), & \text{if } \ee^{\ii \gt\wh{\msf F}_\gt(0)}\neq 1, \\
    \Boh(1), & \text{if } \ee^{\ii \gt\wh{\msf F}_\gt(0)}=1  \text{ and } \zeta\notin \Omega_1^1\cup\Omega_1^4, \\
    \Boh \left((\zeta-1)^{-1}\right), & \text{if } \ee^{\ii \gt\wh{\msf F}_\gt(0)}=1  \text{ and } \zeta\in \Omega_1^1\cup\Omega_1^4.
\end{cases}
$$
\end{enumerate}
\end{rhp}

The proof that $\bm R^\pc$ as we defined in \eqref{eq:deffRpc} indeed solves the RHP~\ref{rhp:Rpc} is as follows. The jump matrix for $\bm R^\pc$ is as usual computed from 
$$
\bm J^\pc=\bm G^\pc_-\wh{\bm J}_{\gt}(\bm J_{\bm G^\pc})^{-1}(\bm G_-^\pc)^{-1}=\bm G^\pc_+(\bm J_{\bm G^\pc})^{-1}\wh{\bm J}_{\gt}(\bm G_+^\pc)^{-1}.
$$
With \eqref{eq:jumpJpc} and \eqref{eq:jumpGpc} in mind, expression \eqref{eq:jumpRPC1} follows. Identities \eqref{eq:jumpJpc} and \eqref{eq:jumpGpc} also say that the jumps of $\bm G^\pc$ and $\bm \Phi^\pc_\gt$ coincide on $\R$, which explains why the obvious analyticity of $\bm R^\pc$ on $\C\setminus \Gamma^1$ extends across $\R$ as well, and RHP~\ref{rhp:Rpc}--(1),(2) are established. Property RHP~\ref{rhp:Rpc}--(3) follows from a simple combination of RHP~\ref{rhp:modelPhiPC}--(3) and \eqref{eq:asympGpc}, and RHP~\ref{rhp:Rpc}--(4) is a consequence of RHP~\ref{rhp:modelPhiPC}--(4) and \eqref{eq:behGPCorigin}.

RHP~\ref{rhp:modelPhi}--(4) states that $\bm \Phi_\gt$ may have a pole behavior as $\zeta\to \lambda=\zs$ when $\ee^{\ii \gt \msf F_\gt(\zs)}=1$. In the asymptotic analysis of the critical and supercritical regimes, this issue of possible pole behavior is sorted out on its own, once we construct a local parametrix. However, over here we do not construct a local parametrix, and this pole behavior remains, as reflected in RHP~\ref{rhp:Rpc}. Even though one could prove that $\bm R^\pc$ is already close to the identity when $\ee^{\ii \gt\wh{\msf F}_\gt(0)}\neq 1$, we want to handle both cases $\ee^{\ii \gt\wh{\msf F}_\gt(0)}=1$ and $\ee^{\ii \gt\wh{\msf F}_\gt(0)}\neq 1$ simultaneously, and for that we have to make another transformation, with the goal of sorting out such pole behavior.

As explained in the discussion following \eqref{eq:formaltildePsi0}, the factor $1/(1+\ee^{\ii \gt \msf P_\gt(\zeta)})$ has poles accumulating from the upper and lower half planes towards $\lambda=\zs$ as $\gt\to \infty$. After the scaling $\zeta\mapsto \zs\zeta$, such poles are turned into poles of $1/(1+\ee^{\gt \wh{\msf P}_\gt(\zeta)})$. These poles are determined by the equation $\gt\wh{\msf P}_\gt(\zeta)=(2k+1)\pi \ii, k\in \Z$, and thanks to \eqref{eq:TaylorExpPgtFgt}, \eqref{eq:zetasseries} and \eqref{eq:shiftedPFPC} they take the form
$$
\zeta=1+\frac{1}{s}(2k+1)\pi \ii +\Boh\left( \frac{s}{\gt^3} \right),\quad k\in \Z.
$$

In particular, such expansion shows that $1/(1+\ee^{\gt\wh{\msf P}_\gt})$ has no poles on any disk of the form $D_{\varepsilon s^{-1}}(1)$, for any $\varepsilon\in (0,1)$ and any $s$ within the subcritical regime.

Thus, we fix $\varepsilon\in (0,1)$ and transform
\begin{equation}\label{eq:RPCtowhRPC}
\wh{\bm R}^\pc(\zeta)= 
\begin{dcases}
\bm R^\pc(\zeta)
\left(\bm I+\frac{\bm G^\pc(\zeta)\bm E_{12} \bm G^\pc(\zeta)^{-1}}{(1-\ee^{-\ii \gt\wh{\msf F}_\gt(\zeta)})(1+\ee^{\gt\wh{\msf P}_\gt(\zeta)})}\right), & \zeta\in \Omega_1^1\cap D_{\varepsilon s^{-1}}(1), \\
\bm R^\pc(\zeta)
\left(\bm I-\frac{\bm G^\pc(\zeta)\bm E_{12} \bm G^\pc(\zeta)^{-1}}{(1-\ee^{\ii \gt\wh{\msf F}_\gt(\zeta)})(1+\ee^{\gt\wh{\msf P}_\gt(\zeta)})}\right),& \zeta\in \Omega_4^1\cap D_{\varepsilon s^{-1}}(1), \\
\bm R^\pc(\zeta), & \text{elsewhere}.
\end{dcases}
\end{equation}

By our choice of $\varepsilon$, the factor $1/(1+\ee^{\gt\wh{\msf P}_\gt})$ is analytic on the sectors $\Omega_1^1\cap D_{\varepsilon s^{-1}}(1)$ and $\Omega_4^1\cap D_{\varepsilon s^{-1}}(1)$. Thanks to Proposition~\ref{prop:extAdmFunc}--(ii), we also see that the factors $1/(1-\ee^{\pm\ii \gt\wh{\msf F}_\gt})$ do not have poles on the corresponding sectors. Finally, recalling \eqref{deff:Gpc} we are sure that $\bm G^\pc$ is analytic on the sectors $\Omega_1^1\cap D_{\varepsilon s^{-1}}(1)$ and $\Omega_4^1\cap D_{\varepsilon s^{-1}}(1)$ as well. Hence, the transformation \eqref{eq:RPCtowhRPC} involves only analytic factors within $\Omega_1^1\cap D_{\varepsilon s^{-1}}(1)$ and $\Omega_4^1\cap D_{\varepsilon s^{-1}}(1)$.

Set
$$
\wh\Gamma^\pc\deff \left(\Gamma^\pc\setminus D_{\varepsilon s^{-1}}(1)\right)\cup (\partial D_{\varepsilon s^{-1}}(1)\cap \Omega_1^1)\cup (\partial D_{\varepsilon s^{-1}}(1)\cap \Omega_4^1),
$$
with orientation as displayed in Figure~\ref{Fig:RHPRPCmod}.

\begin{figure}[t]
\centering
\includegraphics[scale=1]{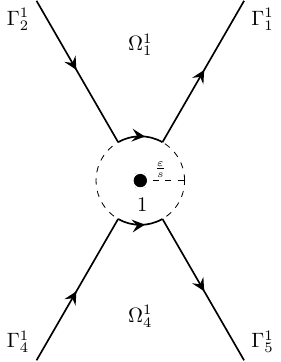}
\caption{Geometric details for the transformation $\bm R^\pc\mapsto \wh{\bm R}^\pc$. In solid, the contours that constitute the jump contour $\wh\Gamma^\pc$.}
\label{Fig:RHPRPCmod}
\end{figure}

Then $\wh{\bm R}^\pc$ satisfies the following RHP.

\begin{rhp}\label{rhp:whRpc}
Find a $2\times 2$ matrix-valued function $\wh{\bm R}^\pc$ with the following properties.
\begin{enumerate}[(1)]
\item $\wh{\bm R}^\pc$ is analytic on $\C\setminus \wh\Gamma^\pc$.
\item The matrix $\wh{\bm R}^\pc$ has continuous boundary values $\wh{\bm R}_\pm^\pc$ along $\wh\Gamma^\pc$, and they are related by the jump relation $\wh{\bm R}^\pc_+(\zeta)=\wh{\bm R}^\pc_-(\zeta)\wh{\bm J}^\pc(\zeta)$, $\zeta\in \wh{\Gamma}^\pc$, with
\begin{equation}\label{eq:jumpwhRPC}
\wh{\bm J}^\pc(\zeta)\deff 
\bm I-\frac{\bm G_\pm^\pc(\zeta)\bm E_{12} \bm G_\pm^\pc(\zeta)^{-1}}{(1-\ee^{\mp\ii \gt\wh{\msf F}_\gt(\zeta)})(1+\ee^{\gt\wh{\msf P}_\gt(\zeta)})} , \quad \zeta\in \wh{\Gamma}^\pc \text{ with } \pm \re\zeta>0.
\end{equation}
\item As $\zeta\to \infty$, 
$$
\wh{\bm R}^\pc(\zeta)=\bm I+\Boh(\zeta^{-1}).
$$
\item $\wh{\bm R}^\pc(\zeta)=\Boh(1)$ as $\zeta\to 1$.
\end{enumerate}
\end{rhp}

The verification of properties in \ref{rhp:whRpc}--(1),(3),(3) is standard as a direct consequence of the RHP~\ref{rhp:Rpc} satisfied by $\bm R^\pc$ and the transformation \eqref{eq:RPCtowhRPC}. 

The verification of \ref{rhp:whRpc}--(4) for $\wh{\bm R}$ as in \eqref{eq:RPCtowhRPC} is a consequence of the following simple observation: the transformation \eqref{eq:RPCtowhRPC} removes the jumps of $\bm R^\pc$ in a neighborhood of $\zeta=1$. Hence, $\wh{\bm R}$ has an isolated singularity at $\zeta=1$. Now using RHP~\ref{rhp:Rpc}--(4) we see that this singularity cannot be essential, and evaluating the limit of $\wh{\bm R}^\pc$ as $\zeta\to 1$ along, say, the real axis, and using again RHP~\ref{rhp:Rpc}--(4), we see that this singularity is in fact removable.

We are finally ready to establish the asymptotic estimates needed to conclude the asymptotic analysis.

\begin{prop}\label{prop:estJpc}
    The estimate
    $$
    \|\wh{\bm J}^\pc-\bm I\|_{L^1\cap L^\infty(\wh\Gamma^\pc)}=\Boh\left(\ee^{-\zs^{3/2}}\right)
    $$
    holds true as $\gt\to \infty$, uniformly for $s$ within the subcritical regime from Definition~\ref{deff:subregimescritical}.
\end{prop}

\begin{proof}
We use \eqref{eq:behGPCtau} and for $\zeta\in \wh{\Gamma}^\pc$ we estimate as $\gt\to \infty$,
\begin{equation}\label{eq:whRest1}
\frac{\bm G_\pm^\pc(\zeta)\bm E_{12} \bm G_\pm^\pc(\zeta)^{-1}}{(1-\ee^{-\ii \gt\wh{\msf F}_\gt(\zeta)})(1+\ee^{\gt\wh{\msf P}_\gt(\zeta)})}=\Boh(\zeta^{1/4})
\frac{\ee^{\pm \ii \gt\wh{\msf F}_\gt(\zeta)- \frac{4}{3}\zs^{3/2}\zeta^{3/2}}}{(1-\ee^{\pm \ii \gt\wh{\msf F}_\gt(\zeta)})(1+\ee^{\gt\wh{\msf P}_\gt(\zeta)})},\quad \pm \re\zeta>0,
\end{equation}
where the error term is uniform for $\zeta\in \wh{\Gamma}^\pc$ and $s$ within the subcritical regime. To give estimates for the explicit scalar term on the right-hand side, we split into two parts, namely with $|\zeta-1|\leq \varepsilon$ and $|\zeta-1|\geq \varepsilon$, where for convenience $\varepsilon$ is the same used in the disk $D_{\varepsilon s^{-1}}(1)$ in \eqref{eq:RPCtowhRPC}, which will eventually be chosen sufficiently small.

For $\zeta\in \wh{\Gamma}_\pc\setminus D_\varepsilon(1)$, Proposition~\ref{prop:canonicestPF}--(ii),(iii) yield the uniform estimate
$$
\ee^{\pm \ii \gt\wh{\msf F}_\gt(\zeta)- \frac{4}{3}\zs^{3/2}\zeta^{3/2}}=\Boh\left(\ee^{-\gt \zs\eta |\zeta-1|}\right)\quad \text{and}\quad \frac{1}{(1-\ee^{\pm \ii \gt\wh{\msf F}_\gt(\zeta)})(1+\ee^{\gt\wh{\msf P}_\gt(\zeta)})}=\Boh(1).
$$
When combined with \eqref{eq:whRest1} and \eqref{eq:jumpwhRPC}, these estimates are enough to show that
\begin{equation}\label{eq:whRPJumpCauxest01}
\|\wh{\bm J}^\pc-\bm I\|_{L^1\cap L^\infty(\wh\Gamma^\pc\setminus D_\varepsilon(1))}=\Boh(\ee^{-\eta \gt \zs})=\Boh(\ee^{-\eta s}),
\end{equation}
uniformly for $s$ within the subcritical regime, where for the last identity we used that $\zs=\Boh(s/\gt)$, see \eqref{eq:zetasseries}.

For the estimate on the remaining part of the contour, that is, for $\zeta\in \wh{\Gamma}^\pc$ with $\varepsilon s^{-1}\leq |\zeta-1|\leq \varepsilon$, we use again Proposition~\ref{prop:canonicestPF}--(ii), (iii) and now estimate
$$
\ee^{\pm \ii \gt\wh{\msf F}_\gt(\zeta)}=\Boh(\ee^{-\eta \gt\zs \varepsilon/s})=\Boh(\ee^{-\eta \varepsilon})\quad \text{and likewise}\quad  \ee^{ \ii \gt\wh{\msf P}_\gt(\zeta)}=\Boh(\ee^{-\eta \varepsilon}).
$$
Hence, by making $\varepsilon>0$ sufficiently small, we can make sure that $|1-\ee^{\pm \ii \gt\wh{\msf F}_\gt(\zeta)}|\geq \delta, |1+\ee^{ \ii \gt\wh{\msf P}_\gt(\zeta)}|\geq \delta$ for some $\delta>0$, and consequently
 $$
 \frac{\ee^{\pm \ii \gt\wh{\msf F}_\gt(\zeta)}}{(1-\ee^{\pm \ii \gt\wh{\msf F}_\gt(\zeta)})(1+\ee^{\gt\wh{\msf P}_\gt(\zeta)})}=\Boh(1),\quad \pm \re \zeta>0,
 $$
uniformly for $\zeta\in \wh{\Gamma}^\pc$ with $\varepsilon s^{-1}\leq |\zeta-1|\leq \varepsilon$. On the other hand, by continuity we can choose $\varepsilon>0$ such that
$$
\ee^{-\frac{4}{3}\zs^{3/2}\zeta^{3/2}}=\Boh(\ee^{-\zs^{3/2}}),
$$
uniformly for $|\zeta-1|\leq \varepsilon$. Combining these estimates, we obtain
\begin{equation}\label{eq:whRPJumpCauxest02}
\|\wh{\bm J}^\pc-\bm I\|_{L^\infty(\wh\Gamma^\pc\cap D_\varepsilon(1))}=\Boh(\ee^{-\zs^{3/2}}),
\end{equation}
uniformly within the subcritical regime. The result is now completed combining \eqref{eq:whRPJumpCauxest01} and \eqref{eq:whRPJumpCauxest02}, and observing that $\zs^{3/2}\ll s$ for $s$ in the subcritical regime.
\end{proof}

As a consequence, we obtain the estimates which were the final goal of the current section.
\begin{theorem}\label{thm:smallnormpc}
The estimates
$$
\|\wh{\bm R}^\pc-\bm I\|_{L^\infty(\C\setminus \wh{\Gamma}^\pc )}=\Boh(\ee^{-\zs^{3/2}}),\quad \|\wh{\bm R}^\pc_\pm-\bm I\|_{L^2(\wh{\Gamma}^\pc)}=\Boh(\ee^{-\zs^{3/2}})
$$
as well as
$$
\|\partial_\zeta\wh{\bm R}^\pc\|_{L^\infty(\C\setminus \wh{\Gamma}^\pc )}=\Boh(\ee^{-\zs^{3/2}}),\quad \|\partial_\zeta\wh{\bm R}^\pc_\pm\|_{L^2(\wh{\Gamma}^\pc)}=\Boh(\ee^{-\zs^{3/2}}),
$$
are valid as $\gt\to \infty$, uniformly for $s$ within the subcritical regime from Definition~\ref{deff:subregimescritical}.
\end{theorem}
\begin{proof}
The result is a consequence of Proposition~\ref{prop:estJpc} and the general small norm theory for RHPs.
\end{proof}

Since $\bm R^\pc=\wh{\bm R}^\pc$ along the real axis, we conclude
\begin{corollary}\label{cor:estRpc}
    The estimates
    $$
    \|{\bm R}^\pc_\pm-\bm I\|_{L^2(\R)}=\Boh(\ee^{-\zs^{3/2}})\quad \text{and}\quad \|\partial_\zeta{\bm R}^\pc_\pm\|_{L^2(\R)}=\Boh(\ee^{-\zs^{3/2}})
    $$
    are valid as $\gt\to\infty$, uniformly for $s$ within the subcritical regime from Definition~\ref{deff:subregimescritical}.
\end{corollary}

\subsubsection{Unwrap of the transformations in the subcritical regime} We unwrap the transformation $\bm \Phi_\gt\mapsto \bm\Phi_\gt^\pc\mapsto \bm R^\pc\mapsto \wh{\bm R}^\pc$ performed in Section~\ref{sec:asymptrposcrit}, arriving at the expression
\begin{equation}\label{eq:unwrapmodelPhipc001}
\bm\Phi_{\gt,-}(\zs\zeta)=\zs^{-\sp_3/4}\wh{\bm R}_{-}^\pc(\zeta)\zs^{\sp_3/4}\wh{\bm\Psi}_-(\zs\zeta)\quad \zeta\in \R\setminus \{1\}.
\end{equation}
Having in mind \eqref{deff:bmHtau}, \eqref{eq:unwrapmodelPhipc001} implies
\begin{multline}
\zs \mcal H_\gt(\zs(\zeta-1))=
\\ \left[\zs(\bm\Delta_\zeta\wh{\bm\Psi})_-(\zs \zeta)\right]_{21}+\left[\wh{\bm\Psi}_-(\zs \zeta)^{-1}\zs^{-\sp_3/4}(\bm\Delta_\zeta \wh{\bm R}^\pc_-)(\zeta)\zs^{\sp_3/4}\wh{\bm\Psi}_-(\zs \zeta)\right]_{21},\quad \zeta>1.
\end{multline}
By the conditions in RHP~\ref{rhp:whRpc}, the only jumps of the function $\wh{\bm R}^\pc$ on a neighborhood of $\zeta=1$ are across the real axis. As a consequence, $\wh{\bm R}^\pc_-$ extends analytically from $\zeta>1$ to $\zeta<1$. Also, the first term on the right-hand side coincides with $\wh{\mcal H}(\zs \zeta)$ from \eqref{deff:whmcalH}, and therefore it also extends analytically from $\zeta>1$ to $\zeta<1$. We just concluded that the identity
\begin{equation}\label{eq:unwrappc}
\mcal H_\gt(\zs(\zeta-1))=\wh{\mcal H}(\zs \zeta)+\frac{1}{\zs}\left[\wh{\bm\Psi}_-(\zs \zeta)^{-1}\zs^{-\sp_3/4}(\bm\Delta_\zeta \wh{\bm R}^\pc_-)(\zeta)\zs^{\sp_3/4}\wh{\bm\Psi}_-(\zs \zeta)\right]_{21}
\end{equation}
holds true for $\zeta \in (-\kappa \gt^2/\zs,\kappa\gt^2/\zs)$, for some $\kappa>0$. We will use this identity later, together with Theorem~\ref{thm:smallnormpc}.

\subsection{Asymptotic analysis of the model problem in the supercritical regime}\label{sec:asymptrnegcrit} \hfill   

We now carry out the asymptotic analysis in the supercritical regime as stated in Definition~\ref{deff:subregimescritical}. We emphasize that under this regime, $\zs$ satisfies the inequalities \eqref{eq:boundzetagtnc}, which will be used extensively in the analysis that follows.

For the asymptotic analysis in the supercritical regime, we will first rescale $\bm \Phi_\gt$ appropriately, and from that on the analysis will follow more or less the same steps as the analysis of the regular regime, namely constructing local and global parametrices, and then applying the small norm theory. The main distinctions are that the global parametrix is rather explicit, given only in terms of algebraic and transcendental functions, and that the local parametrix has a similar non-constant-jump appeal as the model problem itself. 

Now to the details. 

\subsubsection{Scaling step} We first perform a scaling of the model problem,
$$
\bm \Phi_\gt^\nc(\zeta)\deff \left(\bm I+\frac{\ii}{4}\zeta_\gt(s)^2\bm E_{21}\right)  \bm \Phi_\gt\left(\zeta+\zeta_\gt(s)\right),\quad \zeta\in \C\setminus \Gamma^0.
$$
The function ${\bm \Phi}_\gt^\nc$ satisfies a RHP which is obtained from a shift of the RHP for $\bm \Phi_\gt$. For later convenience, we state this RHP explicitly. For its statement and in what follows, we set
\begin{equation}\label{eq:shiftedPFNC}
\wt{\msf P}_\gt(\zeta)\deff \msf P_\gt\left(\zeta+\zeta_\gt(s)\right),\quad \wt{\msf F}_\gt(\zeta)\deff \msf F_\gt\left(\zeta+\zeta_\gt(s)\right), 
\end{equation}

\begin{rhp}\label{rhp:modelPhinc}
Find a $2\times 2$ matrix-valued function ${\bm \Phi}_\gt^\nc$ with the following properties.
\begin{enumerate}[(1)]
\item ${\bm \Phi}_\gt^\nc$ is analytic on $\C\setminus \Gamma^0$.
\item The matrix ${\bm \Phi}_\gt^\nc$ has continuous boundary values ${\bm \Phi}_{\gt,\pm}^\nc$ along $\Gamma^0 \setminus \{0\}$, and they are related by ${\bm \Phi}_{\gt,+}^\nc(\zeta)={\bm \Phi}_{\gt,-}^\nc(\zeta)\wt{\bm J}_\gt(\zeta)$, $\zeta\in \Gamma^{0}\setminus \{0\}$, where
\begin{equation}\label{eq:jumptildeJt}
\wt{\bm J}_\gt(\zeta)\deff
\begin{dcases}
\bm I+\frac{1}{1+\ee^{\gt\wt{\msf P}_\gt(\zeta)}}\bm E_{12}, & \zeta\in \Gamma_0^{0}, \\
\bm I-\frac{1}{(1-\ee^{-\ii \gt\wt{\msf F}_\gt(\zeta)})(1+\ee^{\gt\wt{\msf P}_\gt(\zeta)})} \bm E_{12}, & \zeta\in \Gamma_1^0, \\
\left(\bm I+(1+\ee^{\gt\wt{\msf P}_\gt(\zeta)})\bm E_{21}\right)\left(\bm I-\frac{1}{(1-\ee^{-\ii \gt\wt{\msf F}_\gt(\zeta)})(1+\ee^{\gt\wt{\msf P}_\gt(\zeta)})} \bm E_{12}\right), & \zeta\in \Gamma_2^0, \\
\frac{1}{1+\ee^{\gt\wt{\msf P}_\gt(\zeta)}}\bm E_{12}-(1+\ee^{\gt\wt{\msf P}_\gt(\zeta)})\bm E_{21}, & \zeta\in \Gamma_3^0, \\
\left(\bm I-\frac{1}{(1-\ee^{\ii \gt\wt{\msf F}_\gt(\zeta)})(1+\ee^{\gt\wt{\msf P}_\gt(\zeta)})} \bm E_{12}\right)\left(\bm I+(1+\ee^{\gt\wt{\msf P}_\gt(\zeta)}) \bm E_{21}\right), & \zeta\in \Gamma_4^0,\\
\bm I-\frac{1}{(1-\ee^{\ii \gt\wt{\msf F}_\gt(\zeta)})(1+\ee^{\gt\wt{\msf P}_\gt(\zeta)})} \bm E_{12}, & \zeta\in \Gamma_5^0.
\end{dcases}
\end{equation}
\item As $\zeta\to \infty$, 
\begin{equation}\label{eq:asymptPhinc}
{\bm \Phi}_\gt^\nc(\zeta)=\left(\bm I+\Boh(\zeta^{-1})\right)\zeta^{-\sp_3/4}\bm U_0\ee^{-\left(\frac{2}{3}\zeta^{3/2}+\zeta_\gt(s)\zeta^{1/2}\right)\sp_3}.
\end{equation}

\item As $\zeta\to 0$,
$$
{\bm \Phi}_\gt^\nc(\zeta)=
\begin{cases}
    \Boh(1), & \text{if } \ee^{\ii \gt\wt{\msf F}_\gt(0)}\neq 1, \\
    \Boh(1), & \text{if } \ee^{\ii \gt\wt{\msf F}_\gt(0)}=1  \text{ and } \zeta\notin \Omega_1^0\cup \Omega_4^0, \\
    \Boh\begin{pmatrix}
        1 & \zeta^{-1} \\ 1 & \zeta^{-1} 
    \end{pmatrix}, & \text{if } \ee^{\ii \gt\wt{\msf F}_\gt(0)}=1  \text{ and } \zeta\in \Omega_1^0\cup \Omega_4^0.
\end{cases}
$$
\end{enumerate}
\end{rhp}

We now move on to the construction of global and local parametrices.

\subsubsection{Construction of the global parametrix} 
As we will see later on, the convergence
$$
\wt{\bm J}_\gt\to 
\begin{cases}
\bm E_{12}-\bm E_{21}, & \zeta\in \Gamma^0_3, \\
\bm I, & \text{ elsewhere in } \Gamma^0,
\end{cases}
$$
holds true in the appropriate sense. Thus, the global parametrix $\bm G^\nc$ should be the $2\times 2$ matrix with this jump, and matching the leading asymptotics of $\bm \Phi_\gt^\nc$ as $\zeta\to \infty$. This lead us immediately to the choice
\begin{equation}\label{deff:Gnc}
\bm G^\nc(\zeta)\deff \zeta^{-\sp_3/4}\bm U_0\ee^{-\left(\frac{2}{3}\zeta^{3/2}+\zeta_\gt(s)\zeta^{1/2}\right)\sp_3},\quad \zeta\in \C\setminus \Gamma^0_0=\C\setminus (-\infty,0].
\end{equation}
The matrix $\bm G^\nc$ is analytic on $\C\setminus (-\infty,0]$, and for later we record the jump relation for it,
\begin{equation}\label{eq:jumpGnc}
\bm G^\nc_+(\zeta)=\bm G^\nc_-(\zeta)(\bm E_{12}-\bm E_{21}),\quad \zeta\in (-\infty,0).
\end{equation}

\subsubsection{The local parametrix} The rescaled model problem $\bm \Phi^\nc$ will be well-approximated by $\bm G^\nc$ for values of $\zeta$ away from the origin. Near the origin, we will need to construct a local parametrix.

Introduce
\begin{equation}\label{deff:Anc}
\bm A^\nc(\zeta)\deff \zeta^{-\sp_3/4}\bm U_0 \ee^{-\frac{2}{3}\zeta^{3/2}\sp_3}\bm U_0^{-1}\zeta^{\sp_3/4},\quad \zeta\in \C.
\end{equation}
Clearly $\bm A^\nc$ is analytic on $\C\setminus (-\infty,0]$. A standard calculation shows that $\bm A^\nc$ has no jump along $(-\infty,0]$, and that the (isolated) singularity at $\zeta=0$ is removable. Hence, the matrix $\bm A$ is in fact analytic on $\C$.

Recall that the matrix $\bm \Upsilon_\gt$ is the solution to the RHP~\ref{rhp:modelUpsilon}. We fix $\varepsilon>0$ and define the local parametrix needed as
\begin{equation}\label{deff:Pnc}
\bm P^\nc(\zeta)\deff \bm A^\nc(\zeta)\gt^{\sp_3/4}\bm \Upsilon_\gt\left(w=\gt\zeta\mid y=-\frac{\zeta_\gt(s)}{\gt^{1/2}} \right),\quad \zeta\in D_\varepsilon\setminus \Gamma^0.
\end{equation}
Thanks to Theorem~\ref{thm:existencemodelUpsilon}, and with \eqref{eq:boundzetagtnc} in mind, we know that this parametrix $\bm P^\nc$ exists. With the scaling $w=\gt \zeta$ and the relation \eqref{eq:jumpPhijumpUpsilonrel} in mind, it is immediate that
\begin{equation}\label{eq:Pncjump}
\bm P^\nc_+(\zeta)=\bm P^\nc_-(\zeta)\wt{\bm J}_\gt(\zeta),\quad \zeta\in D_\varepsilon\cap \Gamma^0,
\end{equation}
with $\wt{\bm J}_\gt$ as in \eqref{eq:jumptildeJt}. Finally, with the choice of $y$ as in \eqref{deff:Pnc} and using RHP~\ref{rhp:modelUpsilon}--(3), it follows that the estimate
\begin{align}
\bm P^{\nc}(\zeta)& = \bm A^\nc(\zeta)\gt^{\sp_3/4}\left(\bm I+\Boh(\gt^{-1})\right)\gt^{-\sp_3/4}\zeta^{-\sp_3/4}\bm U_0 \ee^{-\zs \zeta^{1/2}\sp_3}  \nonumber \\
& =\left(\bm I+\Boh(\gt^{-1/2})\right)\bm G^\nc(\zeta), \label{eq:Pncasympt}
\end{align}
as $\gt\to \infty$ holds uniformly for $\zeta\in \partial D_\varepsilon$, for any $\varepsilon>0$ fixed.

\subsubsection{Conclusion of the asymptotic analysis in the supercritical regime} To conclude the asymptotic analysis, we fix $\varepsilon>0$ and transform
\begin{equation}\label{deff:Rnc}
\bm R^\nc(\zeta)\deff 
\bm \Phi^\nc_\gt(\zeta)\times 
\begin{cases}
    \bm G^\nc(\zeta)^{-1}, &\zeta\in \C\setminus (\Gamma^0\cup \overline D_\varepsilon), \\
    \bm P^\nc (\zeta)^{-1}, & \zeta\in D_\varepsilon \setminus \Gamma^0.
\end{cases}
\end{equation}

With
$$
\Gamma^\nc \deff (\Gamma^0\cup \partial D_\varepsilon)\setminus D_\varepsilon,
$$
and with $\partial D_\varepsilon$ oriented in the clockwise direction, this matrix $\bm R^\nc$ is the solution to the following RHP.
\begin{rhp}\label{rhp:Rnc}
Find a $2\times 2$ matrix-valued function $\bm R^\nc$ with the following properties.
\begin{enumerate}[(1)]
\item $\bm R^\nc$ is analytic on $\C\setminus \Gamma^\nc$.
\item The matrix $\bm R^\nc$ has continuous boundary values $\bm R_\pm^\nc$ along $\Gamma^\nc$, and they are related by the jump relation $\bm R^\nc_+(\zeta)=\bm R^\nc_-(\zeta)\bm J^\nc(\zeta)$, $\zeta\in \Gamma^\rc$, with
\begin{equation}\label{eq:jumpRnc1}
\bm J^\nc(\zeta)\deff 
\begin{dcases}
\bm I+\frac{\ee^{-\frac{4}{3}\zeta^{3/2}-2\zeta_\gt(s)\zeta^{1/2}}}{1+\ee^{\gt\wt{\msf P}_\gt(\zeta)}}\zeta^{-\sp_3/4}\bm U_0\bm E_{12}\bm U_0^{-1}\zeta^{\sp_3/4}, & \zeta\in \Gamma_0^{0}\setminus D_\varepsilon, \\
\bm I-\frac{\ee^{-\frac{4}{3}\zeta^{3/2}-2\zeta_\gt(s)\zeta^{1/2}}}{(1-\ee^{-\ii \gt\wt{\msf F}_\gt(\zeta)})(1+\ee^{\gt\wt{\msf P}_\gt(\zeta)})} \zeta^{-\sp_3/4}\bm U_0\bm E_{12} \bm U_0^{-1}\zeta^{\sp_3/4}, & \zeta\in \Gamma_1^{0}\setminus D_\varepsilon, \\
\begin{multlined}[b]
\left(\bm I+(1+\ee^{\gt\wt{\msf P}_\gt(\zeta)})\ee^{\frac{4}{3}\zeta^{3/2}+2\zeta_\gt(s)\zeta^{1/2}} \zeta^{-\sp_3/4}\bm U_0  \bm E_{21}\bm U_0^{-1}\zeta^{\sp_3/4}\right) \\
\times \left(\bm I-\frac{\ee^{-\frac{4}{3}\zeta^{3/2}-2\zeta_\gt(s)\zeta^{1/2}}}{(1-\ee^{-\ii \gt\wt{\msf F}_\gt(\zeta)})(1+\ee^{\gt\wt{\msf P}_\gt(\zeta)})}\zeta^{-\sp_3/4}\bm U_0  \bm E_{12}\bm U_0^{-1}\zeta^{\sp_3/4}\right) , 
\end{multlined}
& \zeta\in \Gamma_2^0\setminus D_\varepsilon, 
\end{dcases}
\end{equation}
and
\begin{equation}\label{eq:jumpRnc2}
\bm J^\nc(\zeta)\deff 
\begin{dcases}
\bm I+\zeta_-^{-\sp_3/4}\bm U_0\left[\left(\frac{1}{1+\ee^{\gt\wt{\msf P}_\gt(\zeta)}}-1\right)\bm E_{11}+\ee^{\gt\wt{\msf P}_\gt(\zeta)} \bm E_{22}\right]\bm U_0^{-1}\zeta_-^{\sp_3/4}, & \zeta\in \Gamma_3^0\setminus D_\varepsilon, \\
\begin{multlined}[b]
\left(\bm I-\frac{\ee^{-\frac{4}{3}\zeta^{3/2}-2\zeta_\gt(s)\zeta^{1/2}}}{(1-\ee^{\ii \gt\wt{\msf F}_\gt(\zeta)})(1+\ee^{\gt\wt{\msf P}_\gt(\zeta)})}\zeta^{-\sp_3/4}\bm U_0  \bm E_{12}\bm U_0^{-1}\zeta^{\sp_3/4}\right) \\
\times \left(\bm I+(1+\ee^{\gt\wt{\msf P}_\gt(\zeta)})\ee^{\frac{4}{3}\zeta^{3/2}+2\zeta_\gt(s)\zeta^{1/2}} \zeta^{-\sp_3/4}\bm U_0  \bm E_{21}\bm U_0^{-1}\zeta^{\sp_3/4}\right), 
\end{multlined}
& \zeta\in \Gamma_4^0\setminus D_\varepsilon,\\
\bm I-\frac{\ee^{-\frac{4}{3}\zeta^{3/2}-2\zeta_\gt(s)\zeta^{1/2}}}{(1-\ee^{\ii \gt\wt{\msf F}_\gt(\zeta)})(1+\ee^{\gt\wt{\msf P}_\gt(\zeta)})}\zeta^{-\sp_3/4}\bm U_0  \bm E_{12}\bm U_0^{-1}\zeta^{\sp_3/4}, & \zeta\in \Gamma_5^0\setminus D_\varepsilon, \\
\bm P^\nc(\zeta)\bm G^\nc(\zeta)^{-1}, & \zeta\in \partial D_\varepsilon.
\end{dcases}
\end{equation}
\item As $\zeta\to \infty$, 
$$
\bm R^\nc(\zeta)=\bm I+\Boh(\zeta^{-1}).
$$
\end{enumerate}
\end{rhp}

The verification that $\bm R^\nc$ as introduced in \eqref{deff:Rnc} indeed solves RHP~\ref{rhp:Rnc} is as follows. The analyticity properties of $\bm G^\nc,\bm P^\nc$ and $\bm \Phi^\nc$ imply that $\bm R^\nc$ has jumps on $\Gamma^0\cap \partial D_\varepsilon$. Thanks to the relations \eqref{eq:Pncjump} and RHP~\ref{rhp:modelPhinc}--(2), the matrix $\bm R^\nc$ has in fact no jumps on $(D_\varepsilon\cap \Gamma^0)\setminus \{0\}$, and thus has an isolated singularity at $\zeta=0$. Sending $\zeta\to 0$ along, say, the positive real axis, the conditions RHP~\ref{rhp:modelPhinc}--(4) and RHP~\ref{rhp:modelUpsilon}--(4) show that this singularity is removable, concluding the verification of RHP~\ref{rhp:Rnc}--(1).

The verification of RHP~\ref{rhp:Rnc}--(2) is a straightforward calculation using \eqref{deff:Gnc}, \eqref{eq:jumpGnc} and \eqref{eq:jumptildeJt}.

Finally, RHP~\ref{rhp:Rnc}--(3) follows from \eqref{eq:asymptPhinc} and \eqref{deff:Gnc}.

To complete the asymptotic analysis we provide the convergence estimates for the jump matrix.

\begin{prop}\label{prop:jumpconvRnc}
The estimate
$$
\left\| \bm J^\nc -\bm I \right\|_{L^1\cap L^\infty(\Gamma^\nc)}=\Boh(\gt^{-1/2})
$$
holds true as $\gt\to \infty$, uniformly for $s$ within the supercritical regime from Definition~\ref{deff:subregimescritical}.
\end{prop}

\begin{proof}
We start by recalling the relations \eqref{eq:shiftedPFNC} between $\wt{\msf  P}_\gt$ and $\msf P_\gt$, and between $\wt{ \msf F}_\gt$ and $\msf F_\gt$. These relations allow us to apply Proposition~\ref{prop:canonicestPF} in order to obtain estimates for terms involving $\ee^{\gt\wt{\msf P}_\gt}$ and $\ee^{\pm \ii \gt \wt{\msf F}_\gt}$, and will be used in the sequel without further mention.

The decay 
\begin{equation}\label{eq:algdecayJnc}
\bm J^\nc=\bm I+\Boh(\gt^{-1/2}) \quad \text{along}\quad \partial D_\varepsilon
\end{equation} 
is a direct consequence of the definition of $\bm J^\nc$ in \eqref{eq:jumpRnc2} and the relation \eqref{eq:Pncasympt}. We now show that in the remaining pieces of $\Gamma^\nc$, the jump matrix $\bm J^\nc$ decays exponentially fast to the identity matrix in both $L^1$ and $L^\infty$ norms. The jumps on all of these components consist of a multiplication of scalar factors with matrix factors of the form 
$$
\zeta_-^{\sp_3/4}\bm U_0\bm E_{ij}\bm U_0^{-1}\zeta_-^{\sp_3/4}.
$$
These matrix factors are roughly $\Boh(\zeta^{1/4})$ in any component of $\Gamma_0\setminus \overline D_\varepsilon$, including in unbounded parts of these components. We will show in the rest of the proof that the scalar factors are all $\Boh(\ee^{-\gt \eta|\zeta|})$ for some constant $\eta>0$, uniformly for $\zeta \in \Gamma_0\setminus \overline D_\varepsilon$, including in unbounded components of this contour. Combining this exponential decay of the scalar parts with the algebraic decay of the matrix parts then yield that ${\bm J}^\nc$ decay exponentially fast to the identity in $L^1\cap L^\infty(\Gamma^\nc\setminus \partial D_\varepsilon)$, and when combined with \eqref{eq:algdecayJnc} the proof is completed.

Proposition~\ref{prop:canonicestPF} provides the estimate $1/(1+\ee^{\gt\wt{\msf P}_\gt(\zeta)})=\Boh(\ee^{-\gt |\zeta|})$ along $\Gamma_0^0\setminus D_\varepsilon=[\varepsilon,\infty)$, hence the exponential decay in $L^1\cap L^\infty$ follows along this contour.

Along $\Gamma_0^1\setminus D_\varepsilon$, Proposition~\ref{prop:canonicestPF}--(i) yields $1/(1+\ee^{\gt \wt{\msf P}_\gt(\zeta)})=\Boh(\ee^{-\gt |\zeta|})$, and Proposition~\ref{prop:canonicestPF}--(iii) gives that
$$
\frac{\ee^{-\frac{4}{3}\zeta^{3/2}-2\zs \zeta^{1/2}}}{1-\ee^{-\ii \gt \wt{\msf F}_\gt(\zeta)}}=\frac{\ee^{-\frac{4}{3}\zeta^{3/2}+\ii \gt \wt{\msf F}_\gt(\zeta)-2\zs \zeta^{1/2}}}{1-\ee^{\ii \gt \wt{\msf F}_\gt(\zeta)}}=\Boh\left( \ee^{-\gt \eta |\zeta|-2\zs \re \zeta^{1/2}} \right).
$$
Likewise, and having in mind that $\re (\zeta^{3/2})=0$ and $\re \zeta^{1/2}=|\zeta|^{1/2}/2$ for $\zeta\in \Gamma_2^0$, Proposition~\ref{prop:canonicestPF}--(ii) we obtain the estimate
$$
(1+\ee^{\gt \wt{\msf P}_\gt(\zeta)})\ee^{\frac{4}{3}\zeta^{3/2}+2\zs \zeta^{1/2}}=\Boh((1+\ee^{-\gt \eta |\zeta|})\ee^{\zs |\zeta|^{1/2}}),\quad \zeta\in \Gamma_2^0\setminus D_\varepsilon.
$$
As we are in the supercritical regime, we are certain that $\zs<0$ for $\gt$ sufficiently large, and moreover $\zs=\Boh(\gt^{1/2})$ (see \eqref{eq:boundzetagtnc}). These last two displayed estimates show the claimed exponential convergence of the scalar factors along $(\Gamma_1^0\cup \Gamma_2^0)\setminus D_\varepsilon(0)$. The analysis on $(\Gamma_4^0\cup \Gamma_5^0)\setminus D_\varepsilon(0)$ is similar.

At last, for the jump on $\Gamma_3^0\setminus D_\varepsilon$, Proposition~\ref{prop:canonicestPF}--(ii) implies that 
$$
\ee^{\gt \wt{\msf P}_\gt(\zeta)}=\Boh(\ee^{-\gt \eta|\zeta|}) \quad \text{and}\quad \frac{1}{1+\ee^{\gt \wt{\msf P}_\gt(\zeta)}}-1 = \Boh(\ee^{-\gt \eta|\zeta|}),
$$
uniformly for $\zeta\in \Gamma_0^3=(-\infty,0)$. As we are only interested in values $\zeta\in \Gamma^3_0\setminus D_\varepsilon=(-\infty-\varepsilon]$, these estimates provide exponential decay of the scalar factors of $\bm J^\nc$ in $L^1\cap L^\infty$ along $\Gamma_3^0\setminus D_\varepsilon$, concluding the proof.
\end{proof}

Thanks to Proposition~\ref{prop:jumpconvRnc}, we obtain
\begin{theorem}\label{thm:smallnormnegcrit}
The estimates
$$
\|\bm R^\nc-\bm I\|_{L^\infty(\C\setminus \Gamma^\nc )}=\Boh(\gt^{-1/2}),\quad \|\bm R^\nc_\pm-\bm I\|_{L^2(\Gamma^\nc)}=\Boh(\gt^{-1/2})
$$
as well as
$$
\|\partial_\zeta\bm R^\nc\|_{L^\infty(\C\setminus \Gamma^\nc )}=\Boh(\gt^{-1/2}),\quad \|\partial_\zeta\bm R^\nc\|_{L^2(\Gamma^\nc)}=\Boh(\gt^{-1/2})
$$
are valid as $\gt\to \infty$, uniformly for $s$ within the supercritical regime.
\end{theorem}

\begin{proof}
    The proof is a direct consequence of Proposition~\ref{prop:jumpconvRnc} and the small norm theory for RHPs.
\end{proof}

\subsubsection{Unwrap of the transformations in the supercritical regime} We now unwrap the transformations in a form appropriate for what will be needed later.

For $\zeta$ away from the origin, the unfolding of the transformations yields
\begin{equation}\label{eq:unfoldPhicritreg001}
\bm \Phi_{\gt,-}(\zeta+\zs)=\left(\bm I-\frac{\ii \zs}{4}\bm E_{21}\right)\bm R^{\nc}_-(\zeta)\bm G^\nc_-(\zeta),\quad \zeta\in \R\setminus (-\delta,\delta).
\end{equation}
Therefore
\begin{equation}\label{eq:mcalHncunfolding1}
\mcal H_\gt(\zeta)=\left[\bm\Delta_\zeta\bm G^\nc(\zeta)\right]_{21}+\left[\bm G^\nc(\zeta)^{-1}\bm\Delta_\zeta\bm R^\nc_-(\zeta)\bm G^\nc(\zeta)\right]_{21},\quad \delta<\zeta<\kappa\gt^2,
\end{equation}
as well as
\begin{multline}\label{eq:mcalHncunfolding2}
\mcal H_\gt(\zeta)=  -\frac{\dd}{\dd\zeta}\left(\ee^{\gt\wt{\msf P}_\gt(\zeta)}\right) -\left(1+\ee^{\gt\wt{\msf P}_\gt(\zeta)}\right)^2\left[\left(\bm I+\frac{\bm E_{12}}{1+\ee^{\gt\wt{\msf P}_\gt(\zeta)}}\right)\bm\Delta_\zeta\bm G^\nc_-(\zeta)\left(\bm I-\frac{\bm E_{12}}{1+\ee^{\gt\wt{\msf P}_\gt(\zeta)}}\right)\right]_{12} \\ 
 -\left(1+\ee^{\gt\wt{\msf P}_\gt(\zeta)}\right)^2 \left[\left(\bm I+\frac{\bm E_{12}}{1+\ee^{\gt\wt{\msf P}_\gt(\zeta)}}\right)\bm G^\nc_-(\zeta)^{-1}\bm\Delta_\zeta\bm R^\nc_-(\zeta)\bm G^\nc_-(\zeta)\left(\bm I-\frac{\bm E_{12}}{1+\ee^{\gt\wt{\msf P}_\gt(\zeta)}}\right)\right]_{12},
\end{multline}
valid for $-\kappa\gt^2<\zeta<-\delta$, and where we recall that $\wt{\msf P}_\gt(\zeta)$ was given in \eqref{eq:shiftedPFNC}.

For $\zeta$ near the origin, say $|\zeta|<\delta$, the unfolding gives
\begin{equation}\label{eq:NCbmPhiunwrapnearorigin}
\bm \Phi_{\gt,-}(\zeta+\zs)=\left(\bm I-\frac{\ii \zs}{4}\bm E_{21}\right)\bm R^{\nc}(\zeta) \bm A^\nc(\zeta)\gt^{\sp_3/4}\bm \Upsilon_{\gt,-}(w\mid y),\quad w=\gt\zeta, y =-\frac{\zs}{\gt^{1/2}}.
\end{equation}

We look at this expression for $0<\zeta<\delta$. Recalling that the functions $\mcal H_\gt$ and $\wt{\mcal H}_\gt$ were introduced in \eqref{deff:bmHtau} and \eqref{deff:wtmcalH}, respectively, and the function $\bm\Upsilon_\gt^\md$ from Proposition~\ref{prop:fundHPsimod}, we obtain the identity
$$
\mcal H_\gt(\zeta\mid s)=\gt \wt{\mcal H}_\gt(w\mid y)+\left[\bm \Upsilon_{\gt}^\md(w)^{-1}\gt^{-\sp_3/4}\bm\Delta_\zeta \left( \bm R^\nc(\zeta)\bm A^\nc(\zeta) \right)\gt^{\sp_3/4}\bm \Upsilon_{\gt}^\md(w)\right]_{21},
$$
valid again for $0<\zeta<\delta$ and with the identifications $w=\gt \zeta$ and $y=-\zs/\gt^{1/2}$. All the terms involved are analytic in a full neighborhood of the origin, so this identity extends to $-\delta<\zeta<\delta$ as well. With \eqref{deff:Upsilonmod} in mind, we just obtained
\begin{equation}\label{eq:HtildeH01}
\mcal H_\gt(\zeta\mid s)=\gt \wt{\mcal H}_\gt(w\mid y)+
\left[\bm \Upsilon_{\gt}(w)^{-1}\gt^{-\sp_3/4}\bm\Delta_\zeta \left( \bm R^\nc(\zeta)\bm A^\nc(\zeta) \right)\gt^{\sp_3/4}\bm \Upsilon_{\gt}(w)\right]_{21}, \quad 0<\zeta<\delta,
\end{equation}
and
\begin{multline}\label{eq:HtildeH02}
\mcal H_\gt(\zeta\mid s)=\gt \wt{\mcal H}_\gt(w\mid y)+
\Big[\left(\bm I+(1+\ee^{\mcal P_\gt(w)})\bm E_{21}\right)\bm \Upsilon_{\gt}(w)^{-1}\gt^{-\sp_3/4} \\ 
\times \bm\Delta_\zeta \left( \bm R^\nc(\zeta)\bm A^\nc(\zeta) \right)\gt^{\sp_3/4}\bm \Upsilon_{\gt}(w)\left(\bm I-(1+\ee^{\mcal P_\gt(w)})\bm E_{21}\right)\Big]_{21}, 
\end{multline}
the latter valid for $-\delta<\zeta<0$, and where in these two identities we remind the reader that $w=\gt \zeta$, $\mcal P_\gt=\mcal P_\gt(\cdot\mid s)$ is as in \eqref{deff:mcalPF}, and $\bm\Upsilon_\gt=\bm\Upsilon_\gt(\cdot\mid y)$ with $y=-\zs/\gt^{1/2}$.

\section{Consequences of the asymptotic analysis of the model problem}\label{sec:unwrapsmodel}

For later reference, we summarize the main findings that we already obtained along the way of the previous calculations, and that we will need later on.

We start with estimates on the model problem $\bm \Phi_\gt$ itself. 

\begin{prop}\label{prop:fundestphitau}
    The following estimates are valid.
    \begin{enumerate}[(i)]
        \item For any $r>0$, the estimate
\begin{multline}\label{eq:Phiunifasymprc}
\bm \Phi_{\gt,-}(\zeta)=\Boh(1+|\zeta|^{1/4})\ee^{-(\frac{2}{3}\zeta_-^{3/2}+\zs\zeta_-^{1/2})\sp_3} \\ 
\times\left[\bm I+\chi_{(-r,r)}(\zeta-\zs)\left(\msf B^\rc(\zeta)+\frac{1}{2\pi\ii}\log\zeta\right)\bm E_{12}\right]\left(\bm I+\chi_{(-r,0)}(\zeta-\zs)(1+\ee^{\gt\msf P_\gt(\zeta)})\bm E_{21}\right),
\end{multline}
        is valid as $\gt\to \infty$, uniformly for $\zeta\in \R$ and uniformly for $s$ within the critical regime, where we recall that $\msf B^\rc$ is given in \eqref{eq:deffmsfBrc}.

        \item For any $M>0$, the estimate
        $$
        \bm \Phi_{\gt,-}(\zeta)=\Boh(1)( \chi_{\{|w|\leq M\}}(\zeta)+ \chi_{\{|w|> M\}}(\zeta)\zeta_-^{-\sp_3/4})\bm U_0\ee^{-\frac{2}{3}\zeta^{3/2}_-\sp_3}    
    \left(\bm I+(1+\ee^{\gt\msf P_{\gt}(\zeta)})\chi_{(0,\zs)}(\zeta)\bm E_{21}\right),
        $$
        is valid uniformly for $\zeta\in \R$ and uniformly for $s$ within the subcritical regime.

        \item For any $\delta>0$, $\bm \Phi_{\gt,-}$ has the following estimates valid as $\gt\to \infty$, uniformly for $s$ within the supercritical regime:
        \begin{enumerate}[(a)]
            \item For $\zeta\in \R$ with $|\zeta|>\delta$,
            $$
            \bm \Phi_{\gt,-}(\zeta+\zs)=\left(\bm I-\frac{\ii\zs}{4}\bm E_{21}\right)\Boh(1)\zeta^{-\sp_3/4}\bm U_0\ee^{-(\frac{2}{3}\zeta_-^{3/2}+\zs \zeta_-^{1/2})\sp_3}.
            $$
            \item For $\zeta\in \R$ with $\delta/\zs^2\leq |\zeta|\leq \delta$,
            $$
            \bm \Phi_{\gt,-}(\zeta+\zs)=\left(\bm I-\frac{\ii\zs}{4}\bm E_{21}\right)
            \Boh(1)|\zs|^{\sp_3/4}\left(\bm I+\Boh\left(\frac{s}{\gt^2}\right)\right)\bm B_-(\zs^2\zeta),
            $$
            where $\bm B$ is the Bessel parametrix from RHP~\ref{rhp:bessel}.
            \item For $\zeta\in \R$ with $|\zeta|\leq \delta/\zs^2$,
            $$
            \bm \Phi_{\gt,-}(\zeta+\zs)=\left(\bm I-\frac{\ii\zs}{4}\bm E_{21}\right)\Boh(1)|\zs|^{\sp_3/4}\left(\bm I+\Boh\left(\frac{s}{\gt^2}\right)\right)\bm B_-(\zs^2\zeta)\bm L_{\bm \Upsilon,-}(\zs^2\zeta),
            $$
            where again $\bm B$ is the Bessel parametrix from RHP~\ref{rhp:bessel}, and $\bm L_{\bm\Upsilon}$ is given in \eqref{deff:PLUpsi}.
        \end{enumerate}
    \end{enumerate}

    All the estimates above may be differentiated term by term.
\end{prop}
\begin{proof}
    The estimate (i) was already proven in \eqref{eq:unwrapPhircinaux5}.

For (ii), recall that the unwrap $\bm \Phi_\gt\mapsto \bm\Phi_\gt^\pc\mapsto \bm R^\pc\mapsto \wh{\bm R}^\pc$ yielded \eqref{eq:unwrapmodelPhipc001}, and we further work on it. First, observe that $\wh{\bm R}_{-}^\pc$ may be estimated from Corollary~\ref{cor:estRpc}. Moving further and also unfolding the transformations $\wh{\bm \Psi}\mapsto \wh{\bm \Psi}^\md \mapsto \wh{\bm R}$ from \eqref{deff:whPsimd} and \eqref{deff:whR}, and estimating $\wh{\bm R}$ using Theorem~\ref{thm:smallnormhatR}, we obtain
\begin{align*}
    \bm \Phi_{\gt,-}(\zeta)& = \left(\bm I+\Boh\left(\zs^{1/2}\ee^{-\zs^{3/2}}\right)\right)\wh{\bai}_-(\zeta)    
    \left(\bm I+(1+\ee^{\gt\msf P_{\gt}(\zeta)})\chi_{(0,\zs)}(\zeta)\bm E_{21}\right) \\ 
    & \Boh(1)\wh{\bai}_-(\zeta)    
    \left(\bm I+(1+\ee^{\gt\msf P_{\gt}(\zeta)})\chi_{(0,\zs)}(\zeta)\bm E_{21}\right).
\end{align*}
The very definitions of $\wh{\bai}$ in \eqref{deff:whAi} and of $\bai$ in \eqref{eq:bai2} show that $\wh{\bai}_-(\zeta)=\bai_-(\zeta)$ for $\zeta \in \R$. Having in mind that $\bai_-$ is bounded on compacts and its asymptotics as $\zeta\to \infty$ is as in \eqref{eq:asympAiryparam}, part (ii) follows.

Part (iii)--(a) is a direct consequence of \eqref{eq:unfoldPhicritreg001} combined with the uniform boundedness for $\bm R^\nc$ provided by Theorem~\ref{thm:smallnormnegcrit} and the explicit expression for $\bm G^\nc$ in \eqref{deff:Gnc}.

For both (iii)-(b) and (iii)-(c), our starting point is \eqref{eq:NCbmPhiunwrapnearorigin}. Using Theorem~\ref{thm:smallnormnegcrit} to bound $\bm R^\nc$ therein, and also that $\bm A^\nc$ is analytic near the origin and independent of $\gt$ (see \eqref{deff:Anc}), we simplify \eqref{eq:NCbmPhiunwrapnearorigin} to
$$
\bm \Phi_{\gt,-}(\zeta+\zs)=\left(\bm I-\frac{\ii\zs}{4}\bm E_{21}\right)\Boh(1)\gt^{\sp_3/4}\bm \Upsilon_{\gt,-}(\gt \zeta\mid y=-\zs/\gt^2),
$$
valid for $|\zeta|\leq \delta$, for any $\delta>0$ fixed. In the regime in (iii)-(b), asymptotics for $\bm \Upsilon_{\gt,-}$ are now provided by \ref{eq:Upstauasymptoticlarge} and Theorem~\ref{thm:smallnormRUp}, whereas in the regime in (iii)--(c) asymptotics for $\bm \Upsilon_{\gt,-}$ are concluded from \eqref{eq:Upstauasymptotic} and Theorem~\ref{thm:smallnormRUp}. 
\end{proof}

Next, we collect estimates for the function $\mcal H_\gt$ from \eqref{deff:bmHtau}. As we will show later, these estimates are at the core of the proof of Theorem~\ref{thm:multstat_formal}.

\begin{theorem}\label{prop:fundmcalHallregimes}
    Fix $\kappa>0$ sufficiently small. The following estimates are valid.
    \begin{enumerate}[(i)]
        \item For any $\eta\in (0,1)$ sufficiently small, the estimate
        $$
        \mcal H_\gt(\zeta-\zs\mid s)=2\pi \ii\, \msf A(\zeta,\zeta)+\Boh\left(\ee^{-\eta\zs-\frac{4}{3}\zeta^{3/2}_+}\right),
        $$
        holds true as $\gt\to \infty$, uniformly for $\zeta\in [-\kappa\gt^2,\kappa\gt^2]$ and uniformly for $s$ within the subcritical regime, where $\msf A$ is the Airy kernel \eqref{deff:AiryKernel}.
        \item For any $L>0$ there exists $\eta>0$ such that the estimates
        \begin{align*}
        & \mcal H_\gt(\zeta-\zs\mid s)=\Boh(\ee^{-\eta (\zeta-\zs)}), \qquad L<\zeta-\zs<\kappa \gt^2, \\
        & \mcal H_\gt(\zeta-\zs\mid s)=\Boh(|\zeta-\zs|), \qquad -\kappa \gt^2<\zeta-\zs<-L,
        \end{align*}
        and
        $$
        \mcal H_\gt(\zeta-\zs\mid s)= 2\pi \ii\, \msf K_\ptf(\zeta-\zs,\zeta-\zs\mid \zs) +\Boh(\gt^{-1}), \quad |\zeta-\zs|\leq L,
        $$
        are valid as $\gt \to \infty$, uniformly for $s$ within the critical regime, and uniformly for $\zeta$ in each of the corresponding sets, and where $\msf K_\ptf$ is the Painlevé XXXIV kernel from \eqref{eq:P34kernelRHP}--\eqref{eq:identitymcalHKptf}.
        \item For any $\delta>0$ sufficiently small, there exists $\eta>0$ such that the estimates
        $$
        \mcal H_\gt(\zeta-\zs\mid s)\frac{\ee^{\gt {\msf P}(\zeta\mid s)}}{(1+\ee^{\gt {\msf P}(\zeta\mid s)})^2}=\Boh(\ee^{-\eta\gt|\zeta-\zs|}), \quad \zeta\in [-\kappa\gt^2,\kappa\gt^2]\setminus [\zs-\delta,\zs+\delta],
        $$
        and
        \begin{multline*}
        \mcal H_\gt(\zeta\mid s)=-2 \pi \ii \zs^2 \msf J_0(-\xi,-\xi) \\ 
        +\Boh\left((1+|\xi|^{1/2})\ee^{2\xi_+^{1/2}}\left(|\zs|+|\zs|^4\gt^{-1}\right)\right), \quad \xi=\zs^2\zeta,\quad |\zeta|\leq \delta,
        \end{multline*}
        are valid as $\gt\to \infty$, uniformly for $s$ within the supercritical, where $\msf J_0$ is the Bessel kernel from \eqref{deff:BesselKernel}.
    \end{enumerate}
    Each of the estimates above may be differentiated term by term.
\end{theorem}

\begin{proof}
For the claim (i), our starting point is \eqref{eq:unwrappc}. Theorem~\ref{thm:smallnormpc} ensures that $\bm\Delta_w \wh{\bm R}_-^\pc(w)=\Boh(\ee^{-\zs^{3/2}})$ uniformly on the interval of integration above. Then, using the asymptotics \eqref{eq:asymptildePsi} we are sure that the rough estimate
\begin{multline}\label{eq:estPsihatRpc01}
\left[\wh{\bm\Psi}_-(\zs w)^{-1}\zs^{-\sp_3/4}(\bm\Delta_w \wh{\bm R}^\pc_-)(w)\zs^{\sp_3/4}\wh{\bm\Psi}_-(\zs w)\right]_{21} =\\ 
\begin{dcases}
\Boh\left(\ee^{-\zs^{3/2}}\ee^{-4\zs^{3/2} w^{3/2}/3}(1+|w|^{1/2})\zs^{1/2}\right), & w\geq 0,\\
\Boh\left(\ee^{-\zs^{3/2}}(1+|w|^{1/2})\zs^{1/2}\right), & w\leq 0, 
\end{dcases}
\end{multline}
is valid uniformly for $w=\Boh(\gt^2/\zs)$. Setting $\zeta=\zs w$ and using this estimate in \eqref{eq:unwrappc}, we obtain
$$
\mcal H_\gt(\zeta-\zs)=\wh{\mcal H}(\zeta)+E_\gt^\pc(\zeta),\quad \gt\to\infty,
$$
where the error term $E_\gt^\pc$ satisfies
$$
E_\gt^\pc(\zeta)=
\begin{cases}
    \Boh\left(\ee^{-\zs^{3/2}}\ee^{-4\zeta^{3/2}/3}(\zs^{-1/2}+\zs^{-1} |\zeta|^{1/2})  \right), & 0\leq \zeta\leq \kappa\gt^2,\\ 
    \Boh\left(\ee^{-\zs^{3/2}}(\zs^{-1/2}+\zs^{-1} |\zeta|^{1/2})  \right), & -\kappa\gt^2\leq \zeta\leq 0,
\end{cases}
$$
valid as $\gt\to \infty$ uniformly for $s$ in the subcritical regime.

We now express $\wh{\mcal H}$ in terms of the Airy kernel. Recalling once again \eqref{deff:whmcalH} and using \eqref{deff:whPsimd} and \eqref{deff:whR}, we write
\begin{align}
\wh{\mcal H}(\zeta\mid s) & =\left[\bm \Delta_\zeta \wh\bai(\zeta)\right]_{21,-}+\left[\wh\bai(\zeta)^{-1}\bm\Delta_\zeta \wh{\bm R}(\zeta)\wh\bai(\zeta)\right]_{21,-} \\ 
& = \left[\bm\Delta_\zeta\bai_0(\zeta)\right]_{21,-}+\left[\wh\bai(\zeta)^{-1}\bm\Delta_\zeta \wh{\bm R}(\zeta)\wh\bai(\zeta)\right]_{21,-},\quad \zeta>\zs. \label{eq:whmcalHwhRwhAi}
\end{align}

By RHP~\ref{rhp:modelhatR}, the definition of $\bai_0$ in \eqref{eq:bai1} and Proposition~\ref{prop:analyticwhH}, every term in this identity is analytic on the interval $(0,+\infty)$, so this identity extends to $\zeta>0$. The term $\wh\bai$ can be analytically continued to $\zeta<0$ using \eqref{deff:whAi}, and its analytic continuation is
$$
\wh\bai (\zeta)(\bm I -\bm E_{21}).
$$
Thanks to Proposition~\ref{prop:analcontwhR}, the analytic continuation of $\wh{\bm R}_-$ from $\zeta>0$ to $\zeta<0$ is
$$
\wh{\bm R}(\zeta)\left(\bm I+\ee^{\gt\msf P_\gt(\zeta)}\wh\bai(\zeta)\bm E_{21}\wh\bai(\zeta)^{-1}\right).
$$
Using these last two expressions, we analytically continue the right-hand side of \eqref{eq:whmcalHwhRwhAi} from $\zeta>0$ to $\zeta<0$. As a result, after a cumbersome but straightforward calculation we obtain the identity
\begin{multline}\label{eq:analyticcontwhmcalH}
    \wh{\mcal H}(\zeta\mid s)=\\ 
    \begin{dcases}
        \left[\bm\Delta_\zeta\bai_0(\zeta)\right]_{21,-}+\left[\wh\bai(\zeta)^{-1}\bm\Delta_\zeta \wh{\bm R}(\zeta)\wh\bai(\zeta)\right]_{21,-},& \zeta>0, \\
        \begin{aligned}[b]
        & \left[\bm\Delta_\zeta\bai_0(\zeta)\right]_{21,-} -\ee^{2\gt\msf P_\gt(\zeta)}\left(\wh\bai(\zeta)^{-1}\partial_\zeta\wh\bai(\zeta) \right)_{21,-} \\ 
        & +\ee^{\gt\msf P_\gt(\zeta)}\left[ (\bm I+\bm E_{21})[\wh\bai(\zeta)^{-1}\partial_\zeta\wh\bai(\zeta),\bm E_{21}](\bm I-\bm E_{21})\right]_{21,-} \\
        & +\left[(\bm I+(1-\ee^{\gt\msf P_\gt(\zeta)})\bm E_{21})\wh\bai(\zeta)^{-1}\bm\Delta_\zeta\wh{\bm R}(\zeta)\wh\bai(\zeta)(\bm I-(1-\ee^{\gt\msf P_\gt(\zeta)})\bm E_{21})\right]_{21,-}, 
        \end{aligned}
        & -\kappa\gt^2<\zeta<0,
    \end{dcases}
\end{multline}
which is valid for some $\kappa>0$. 

We now estimate the terms on the right-hand side above. From asymptotics of Airy functions and Theorem~\ref{thm:smallnormhatR}, we see that
$$
\left|\left[\wh\bai^{-1}\bm\Delta_\zeta \wh{\bm R}(\zeta)\wh\bai(\zeta)\right]_{21,-}\right|=\Boh(\ee^{-\eta\zs}(1+|\zeta|^{1/2})\ee^{-4\zeta^{3/2}/3}),\quad \zeta>0,\quad \gt\to \infty,
$$
for some $\eta>0$. Next, using the fact that in the subcritical regime we are considering here we have $\zs>0$, and thus $|\zeta-\zs|=|\zeta|+|\zs|$ for $\zeta\leq 0$, we obtain again from Proposition~\ref{prop:canonicestPF} that
\begin{multline*}
\left| \ee^{2\gt\msf P_\gt(\zeta)}\left(\wh\bai(\zeta)^{-1}\partial_\zeta\wh\bai(\zeta) \right)_{21,-}\right| \\ +\left|\ee^{\gt\msf P_\gt(\zeta)}\left[ (\bm I+\bm E_{21})[\wh\bai(\zeta)^{-1}\partial_\zeta\wh\bai(\zeta),\bm E_{21}](\bm I-\bm E_{21})\right]_{21,-}\right|=\Boh(\ee^{-\eta\gt\zs}),\quad \gt\to \infty.
\end{multline*}
Also, using again asymptotics of Airy functions, we obtain for $\zeta\leq 0$,
$$
\left[(\bm I+(1-\ee^{\gt\msf P_\gt(\zeta)})\bm E_{21})\wh\bai(\zeta)^{-1}\bm\Delta_\zeta\wh{\bm R}(\zeta)\wh\bai(\zeta)(\bm I-(1-\ee^{\gt\msf P_\gt(\zeta)})\bm E_{21})\right]_{21,-}
=\Boh(\ee^{-\eta\zs}(1+|\zeta|^{1/2})).
$$
All in all, we conclude that
$$
\mcal H_\gt(\zeta-\zs)=\left[\bm\Delta_\zeta \bai_0(\zeta)\right]_{21,-}+\wh E_\gt^\pc(\zeta),\quad \gt\to \infty,
$$
where this new error term satisfies
$$
\wh E_\gt^\pc(\zeta)=
\begin{cases}
    \Boh\left( \ee^{-\frac{4}{3}\zeta^{3/2}}\ee^{-\eta\zs}(1+|\zeta|^{1/2}) \right),& 0\leq  \zeta\leq \kappa\gt^2,\\
    \Boh\left( \ee^{-\eta\zs}(1+|\zeta|^{1/2}) \right),& -\kappa\gt^2\leq \zeta\leq 0.
\end{cases}
$$
The result now follows reducing $\eta$, and also using \eqref{eq:bai1} and \eqref{deff:AiryKernel}.

Next, we move to part (ii), taking as a starting point \eqref{eq:bmHmdPhi}. 

The case $\zeta>\zs+L$ is a direct consequence of \eqref{eq:asymptailsPhirc}, and we now focus on $\zeta<\zs-L$.

To bound the term involving $\bm\Delta_\zeta$ in \eqref{eq:bmHmdPhi} we use \eqref{eq:asymptailsPhirc}, and also the fact that $\zeta_-^{3/2}\in \ii \R$ for $\zeta<0$. As a consequence, for $\zeta<\zs-L$ we update \eqref{eq:bmHmdPhi} to
$$
\mcal H_\gt(\zeta-\zs) = -\gt \msf P_\gt'(\zeta) \ee^{\gt\msf P_\gt(\zeta)} -(1+\ee^{\gt\msf P_\gt(\zeta)})^2\Boh(\zeta).
$$
Equation~\eqref{eq:ineqPrezetazs} ensures that $\ee^{\gt\msf P_\gt(\zeta)}\leq M\ee^{-\eta \gt |\zeta-\zs|}$, uniformly in $s$ within the critical regime. Moreover, a direct calculation shows that
$$
\gt\msf P_\gt'(\zeta)=\gt \msf P'\left(\frac{\zeta}{\gt^2}\right).
$$
By making $\kappa>0$ sufficiently small, we can make sure that whenever $-\kappa\gt^2<\zeta-\zs<-L$, the quotient $\zeta/\gt^2$ falls within the disk of analyticity of $\msf P$ (see Definition~\ref{deff:admissibleFP}), and therefore $\gt\msf P_\gt'(\zeta)=\Boh(\gt)$ uniformly for $\zeta\in J$ as $\gt\to\infty$ and $s$ in the critical regime. The second claim in (ii) then follows.

Next, we treat the third estimate in (ii). Combining the expression of $\mcal H_\gt$ from the first line of \eqref{eq:bmHmdPhi} with \eqref{eq:PhigtmcalHPsi0}, we obtain the identity
\begin{equation}\label{eq:estHtHerror1}
\mcal H_\gt(\zeta)=\mcal H(\zeta\mid \zs)+\left[\bm\Psi_0(\zeta\mid \zs)^{-1}\bm\Delta_\zeta\bm R^\rc(\zeta+\zs)\bm\Psi_0(\zeta\mid \zs)\right]_{21},
\end{equation}
which is an identity initially valid for $\zeta\in (0,L)$, with $L$ sufficiently small. However, all the terms involved in this identity are analytic near $\zeta=0$: the analyticity of the left-hand side is ensured by Proposition~\ref{prop:fundamentalbmHtau}, the analyticity of $\mcal H$ is ensured by Proposition~\ref{prop:RHPlimitPXXXIV}, the shift of $\bm R^\rc$ is analytic by RHP~\ref{rhp:Rrc}--(1) (observe that $\Gamma^\rc$ does not intersect a neighborhood of $\zs$), and $\bm \Psi_0$ is analytic by RHP~\ref{rhp:PXXXIV}--(4). Thus, the identity above extends analytically to a full neighborhood of the origin. The $L^\infty$ estimates for $\bm R^\pc$ and its $\zeta$-derivative from Theorem~\ref{thm:smallnormregcrit} then ensure that
\begin{equation}\label{eq:estHtHerror2}
\left[\bm\Psi_0(\zeta\mid \zs)^{-1}\bm\Delta_\zeta\bm R^\rc(\zeta+\zs)\bm\Psi_0(\zeta\mid \zs)\right]_{21}=\Boh(\gt^{-1}),\quad \gt\to \infty,
\end{equation}
uniformly for $\zeta\in [-L,L]$ and uniformly for $s$ within the critical regime, and the result now follows using \eqref{eq:identitymcalHKptf}.
    
We now move to the proof of (iii). Set
$$
J\deff J^+\cup J^-,\quad J^+\deff (\delta,\kappa\gt^2],\quad J^-\deff [-\kappa\gt^2,-\delta).
$$
We first analyze the case $\zeta\in J^+$. As a starting point, we use the representation \eqref{eq:mcalHncunfolding1}. An explicit calculation from \eqref{deff:Gnc} yields
\begin{equation}\label{eq:DeltaGnc}
\bm\Delta_\zeta \bm G^\nc(\zeta)=-\left(\zeta^{1/2}+\frac{\zs}{2\zeta^{1/2}}\right)\sp_3 +\frac{1}{4\zeta}\ee^{\left(\frac{2}{3}\zeta^{3/2}+\zs\zeta^{1/2}\right)\sp_3}\sp_2 \ee^{-\left(\frac{2}{3}\zeta^{3/2}+\zs\zeta^{1/2}\right)\sp_3},\quad \zeta\in \C\setminus (-\infty,0].
\end{equation}
Hence,
$$
\left[\bm\Delta_\zeta \bm G^\nc(\zeta)\right]_{21}=\Boh(\ee^{-\frac{4}{3}\zeta^{3/2}-2\zs\zeta^{1/2}}),\quad \zeta\in J^+,
$$
and in a similar manner, in combination with Theorem~\ref{thm:smallnormnegcrit},
$$
\left[\bm G^\nc(\zeta)^{-1} \bm\Delta_\zeta \bm R_-^\nc(\zeta)\bm G^\nc(\zeta) \right]_{21}=\Boh(\gt^{-1/2}\ee^{-\frac{4}{3}\zeta^{3/2}-2\zs\zeta^{1/2}}),\quad \zeta\in J^+.
$$
With \eqref{eq:mcalHncunfolding1} in mind and using Proposition~\ref{prop:canonicestPF}--(i), we obtain
$$
\mcal H_\gt(\zeta)\frac{\ee^{\gt{\msf P}_\gt(\zeta+\zs\mid s)}}{\big(1+\ee^{\gt{\msf P}_\gt(\zeta+\zs\mid s)}\big)^2}=\Boh\left( \ee^{-\frac{4}{3}\zeta^{3/2}-2\zs\zeta^{1/2}-\eta\gt \zeta} \right),\quad \zeta\in J^+,
$$
where the error term is uniform for $\zeta\in J^+$ and $s$ within the supercritical regime. 

Having in mind that $\zs<0$ and $\zs=\Boh(s/\gt)$ so that $\zs/\gt=\boh(1)$, we express
$$
-2\zs\zeta^{1/2}-\eta\gt \zeta=\gt\zeta \left( \frac{2|\zs|}{\gt\zeta^{1/2}}-\eta \right)\leq -\frac{\eta}{2}\gt\zeta , \quad \zeta\in J^+,
$$
from which we conclude that for a new constant $\eta>0$,
\begin{equation}\label{msfDnc2}
\mcal H_\gt(\zeta\mid s)\frac{\ee^{\gt{\msf P}(\zeta+\zs\mid s)}}{(1+\ee^{\gt{\msf P}(\zeta+\zs\mid s)})^2}=\Boh(\ee^{-\eta \gt \zeta}),\quad \gt\to \infty,\quad \zeta\in J^+,
\end{equation}
where the error term is uniform in $\zeta$, and also uniform for $s$ within the supercritical regime.

Next, for the estimate for $\zeta \in J^-$, we start with the representation \eqref{eq:mcalHncunfolding2} for $\mcal H_\gt$. A simple estimate using Proposition~\ref{prop:canonicestPF}--(ii) yields
\begin{equation}\label{msfDnc3}
\frac{\dd }{\dd\zeta}\left(\ee^{\gt\wt{\msf P}_\gt(\zeta)}\right)\frac{\ee^{\gt\wt{\msf P}_\gt(\zeta)}}{(1+\ee^{\gt\wt{\msf P}_\gt(\zeta)})^2}=\Boh(\ee^{-\eta\gt|\zeta|}),\quad \gt\in J^-,
\end{equation}
for some $\eta>0$, uniformly for $s$ within the supercritical regime. 

Next, from \eqref{eq:DeltaGnc} and again Proposition~\ref{prop:canonicestPF}--(ii) we compute the rough estimate
$$
\left[ \left(\bm I+\frac{\bm E_{12}}{1+\ee^{\gt\wt{\msf P}_\gt(\zeta)}}\right)\bm\Delta_\zeta\bm G^\nc_-(\zeta)\left(\bm I-\frac{\bm E_{12}}{1+\ee^{\gt\wt{\msf P}_\gt(\zeta)}}\right) \right]_{12}=\Boh(\zeta^{1/2}),\quad \zeta\in J^-,
$$
and then again applying Proposition~\ref{prop:canonicestPF}--(ii),
 \begin{equation}\label{msfDnc4}
\left[ \left(\bm I+\frac{\bm E_{12}}{1+\ee^{\gt\wt{\msf P}_\gt(\zeta)}}\right)\bm\Delta_\zeta\bm G^\nc_-(\zeta)\left(\bm I-\frac{\bm E_{12}}{1+\ee^{\gt\wt{\msf P}_\gt(\zeta)}}\right) \right]_{12}\ee^{\gt\wt{\msf P}_\gt(\zeta)}=\Boh(\ee^{-\eta\gt|\zeta|}),\quad \gt\to \infty, \quad \zeta\in J^-,
\end{equation}
uniformly for $s$ within the supercritical regime.
Finally, using now \eqref{deff:Gnc} and Theorem~\ref{thm:smallnormnegcrit}, we get the rough bound
$$
\left[\left(\bm I+\frac{\bm E_{12}}{1+\ee^{\gt\wt{\msf P}_\gt(\zeta)}}\right)\bm G^\nc_-(\zeta)^{-1}\bm\Delta_\zeta\bm R^\nc_-(\zeta)\bm G^\nc_-(\zeta)\left(\bm I-\frac{\bm E_{12}}{1+\ee^{\gt\wt{\msf P}_\gt(\zeta)}}\right)\right]_{12}=\Boh(\zeta^{1/2}\gt^{-1/2}),
$$
and then once again 
\begin{equation}\label{msfDnc5}
\left[\left(\bm I+\frac{\bm E_{12}}{1+\ee^{\gt\wt{\msf P}_\gt(\zeta)}}\right)\bm G^\nc_-(\zeta)^{-1}\bm\Delta_\zeta\bm R^\nc_-(\zeta)\bm G^\nc_-(\zeta)\left(\bm I-\frac{\bm E_{12}}{1+\ee^{\gt\wt{\msf P}_\gt(\zeta)}}\right)\right]_{12}\ee^{\gt\wt{\msf P}_\gt(\zeta)}=\Boh(\ee^{-\eta\gt |\zeta|}),\quad \zeta\in J^-,
\end{equation}
as $\gt\to \infty$, uniformly for $s$ within the supercritical regime.

Recalling \eqref{eq:mcalHncunfolding2}, we summarize \eqref{msfDnc2}, \eqref{msfDnc3}, \eqref{msfDnc4}, \eqref{msfDnc5} into the estimate
    \begin{equation}\label{eqestmcalHtauspcrtfinalJ}
    {\mcal H}_\gt(\zeta)\frac{\ee^{\gt\wt{\msf P}(\zeta\mid s)}}{(1+\ee^{\gt\wt{\msf P}(\zeta\mid s)})^2}=\Boh(\ee^{-\eta\gt|\zeta|}),\quad \gt\to \infty,\quad \zeta\in J,
    \end{equation}
    uniformly for $s$ within the supercritical regime. Recalling that $\wt{\msf P}_\gt(\zeta+\zs)=\msf P_\gt(\zeta)$ (see \eqref{eq:shiftedPFNC}), this last estimate is the same as the claim in (iii) with $\delta=L$.

Finally, we move on to the proof of the second estimate in (iii). The starting points for it are equations \eqref{eq:HtildeH01} and \eqref{eq:HtildeH02}. Recall that in Section~\ref{sec:Upsilonscaling} we performed the change of variables $\xi=y^2 w$ in the RHP for $\bm Y_\gt(\cdot\mid y)$, and in Equation~\eqref{deff:Pnc} we later identified $w=\gt\zeta$ and $y=-\zs/\gt^{1/2}$. Unfolding these two transformations, we identify $\zeta=\xi/(\gt y^2)=\xi/\zs^2$. In the formulas that will follow, we use $\xi,w,\zeta$ as variables for simplicity, and they always correspond to the just mentioned correspondences.

We first re-express \eqref{eq:HtildeH01} using \eqref{eq:Upstauasymptoticlarge}, \eqref{eq:Upstauasymptotic} and \eqref{eq:mcalHBB0Bessel}, which leads to
\begin{multline}\label{eq:mcalHbmB}
\mcal H_\gt(\zeta\mid s)=\zs^2[\bm\Delta_\xi \bm B(\xi)]_{21} +\zs^2 \left[ \bm B_0(\xi)^{-1}\bm\Delta_{\xi}\bm R_{\bm \Upsilon}(\xi)\bm B_0(\xi) \right]_{21} \\ 
+ \left[ \bm B(\xi)^{-1}\bm R_{\bm \Upsilon}(\xi)^{-1} |\zs|^{-\sp_3/2}\bm\Delta_\zeta \left(\bm R^\nc(\zeta)\bm A^\nc(\zeta)\right)|\zs|^{\sp_3/2}\bm R_{\bm \Upsilon}(\xi)\bm B(\xi)  \right]_{21,-}, \quad 0<\zeta<\delta.
\end{multline}
In the expression above, we already accounted for the explicit expression for $\bm L_{\bm \Upsilon}$ in \eqref{deff:PLUpsi}. 

The function $\bm A^\nc$ is analytic in a neighborhood of the origin and independent of $\gt$ (see \eqref{deff:Anc}), and therefore it remains bounded as $\gt\to \infty$. The functions $\bm R_{\bm \Upsilon}(\xi)$ and $\bm R^\nc(\zeta)$ both decay uniformly to the identity matrix (see Theorem~\ref{thm:smallnormnegcrit} and Theorem~\ref{thm:smallnormRUp}) and therefore
\begin{multline*}
\left[ \bm B(\xi)^{-1}\bm R_{\bm \Upsilon}(\xi)^{-1} |\zs|^{-\sp_3/2}\bm\Delta_\zeta \left(\bm R^\nc(\zeta)\bm A^\nc(\zeta)\right)|\zs|^{\sp_3/2}\bm R_{\bm \Upsilon}(\xi)\bm B(\xi)  \right]_{21,-} \\
= \left[ \bm B(\xi)^{-1}\Boh(1)|\zs|^{-\sp_3/2} \Boh(1)|\zs|^{\sp_3/2}\Boh(1) \bm B(\xi)  \right]_{21,-} = \left[ \bm B(\xi)^{-1} \Boh(|\zs|) \bm B(\xi)  \right]_{21,-}
\end{multline*}
We are interested in values $0<\zeta<\delta$, which correspond to $0<\xi <\delta \zs^2$, hence also possibly growing with $\gt$. For bounded values of $\xi>0$, the expansion in RHP~\ref{rhp:bessel}--(4) ensures that 
$$
\left[ \bm B(\xi)^{-1} \Boh(|\zs|) \bm B(\xi)  \right]_{21,-}=\left[ \bm B_0(\xi)^{-1} \Boh(|\zs|) \bm B_0(\xi)  \right]_{21}=\Boh(|\zs|).
$$
For unbounded values of $\xi>0$, we use RHP~\ref{rhp:bessel}--(3) instead, and obtain that
$$
\left[ \bm B(\xi)^{-1} \Boh(|\zs|) \bm B(\xi)  \right]_{21,-}=\left[ \ee^{-\xi^{1/2}\sp_3}\bm U_0^{-1}\xi^{\sp_3/4} \Boh(|\zs|)\xi^{-\sp_3/4}\bm U_0\ee^{\xi^{1/2}\sp_3} \right]_{21}=
\Boh( |\zs| \xi^{1/2} \ee^{2\xi^{1/2}} ).
$$
All in all, these bounds lead to 
\begin{multline*}
\left[ \bm B(\xi)^{-1}\bm R_{\bm \Upsilon}(\xi)^{-1} |\zs|^{-\sp_3/2}\bm\Delta_\zeta \left(\bm R^\nc(\zeta)\bm A^\nc(\zeta)\right)|\zs|^{\sp_3/2}\bm R_{\bm \Upsilon}(\xi)\bm B(\xi)  \right]_{21,-}\\ 
=\Boh((1+|\xi|^{1/2})\ee^{2\xi^{1/2}}|\zs|),
\end{multline*}
valid as $\gt\to \infty$, uniformly for $0<\zeta<\delta$ and uniformly for $s$ within the supercritical regime.
In a completely analogous way, we obtain
$$
[\bm B_0(\xi)^{-1}\bm\Delta_\xi\bm R_{\bm Y}(\xi)\bm B_0(\xi)]_{21}=\Boh((1+|\xi|^{1/2})\ee^{2\xi^{1/2}}|\zs|^2 \gt^{-1}),\quad \gt\to \infty,
$$
where the term $\zs^2\gt^{-1}=y^2$ comes from Theorem~\ref{thm:smallnormRUp}, and we emphasize that the error term is uniform for $0<\zeta<\delta$ and also uniform for $s$ within the supercritical regime. Recalling that $\zs=\Boh(s\gt^{-1})$ and $s$ is within the supercritical regime, and using also the identity \eqref{eq:analyticextensionbmBJbessel}, we update \eqref{eq:mcalHbmB} to
\begin{equation}\label{eq:auxxineg07}
\mcal H_\gt(\zeta\mid s)=-2\pi \ii \zs^2 \msf J_0(-\xi,-\xi)+\Boh\left( (1+|\xi|^{1/2})\ee^{2\xi^{1/2}}(|\zs|+|\zs|^4\gt^{-1} )\right),\quad \gt\to \infty,
\end{equation}
uniformly for $0<\zeta<\delta$ and uniformly for $s$ within the supercritical regime. Observe that $|\zs|^4/\gt \leq C |\zs|$ for $\gt\leq |s|\leq \gt^{4/3}$, whereas $C|\zs|^4/\gt \geq |\zs|$ for $\gt^{4/3}\leq |s|\leq \gt^{3/2}$, but for now it suffices to keep the error term as presented above. This estimate concludes the second estimate in (iii) for when $\zeta\geq 0$.

The analysis for $-\delta<\zeta<0$ is similar. We re-express \eqref{eq:HtildeH02} using again \eqref{eq:Upstauasymptoticlarge}, \eqref{eq:Upstauasymptotic} and \eqref{eq:mcalHBB0Bessel}, as well as \eqref{eq:analyticextensionbmBJbessel} and the definition of $\bm L_{\bm \Upsilon}$ in \eqref{deff:PLUpsi}, which now leads to
\begin{multline}\label{eq:auxxineg04}
\mcal H_\gt(\zeta\mid s)=-2\pi \ii \zs^2\msf J_0(-\xi,-\xi)+\zs^2[\bm B_0(\xi)^{-1}\bm \Delta_\xi\bm R_{\bm \Upsilon}(\xi)\bm B_0(\xi)]_{21,-} \\ 
+ \left[ (\bm I+\bm E_{21})\bm B(\xi)^{-1}\bm R_{\bm \Upsilon}(\xi)^{-1}|\zs|^{-\sp_3/2}\bm\Delta_\zeta(\bm R^\nc(\zeta)\bm A^{\nc }(\zeta))|\zs|^{\sp_3/2}\bm R_{\bm \Upsilon}(\xi)\bm B(\xi)(\bm I-\bm E_{21}) \right]_{21,-},
\end{multline}
valid for $-\varepsilon<\xi<0$, corresponding to $-\varepsilon/\zs^2<\zeta<0$, as well as
\begin{multline}\label{eq:auxxineg05}
\mcal H_\gt(\zeta\mid s)=-2\pi \ii \zs^2\msf J_0(-\xi,-\xi)+\zs^2[\bm B_0(\xi)^{-1}\bm \Delta_\xi\bm R_{\bm \Upsilon}(\xi)\bm B_0(\xi)]_{21,-} + \Big[ (\bm I+(1+\ee^{\wt{\mcal P}_\gt(\xi)})\bm E_{21})\bm B(\xi)^{-1}  \\
\times \bm R_{\bm \Upsilon}(\xi)^{-1}|\zs|^{-\sp_3/2} \bm\Delta_\zeta(\bm R^\nc(\zeta)\bm A^{\nc }(\zeta))|\zs|^{\sp_3/2}\bm R_{\bm \Upsilon}(\xi)\bm B(\xi)(\bm I-(1+\ee^{\wt{\mcal P}_\gt(\xi)})\bm E_{21}) \Big]_{21,-},
\end{multline}
which is valid for $-\delta \zs^2<\xi<-\varepsilon$, that is, $-\delta<\zeta<-\varepsilon/\zs^2$. For this last identity, we recall that $\wt{\mcal P}_\gt$ was introduced in \eqref{deff:wtcalPF}.

We now estimate the terms on the right-hand side exactly as before. Observing now that $\xi^{1/2}_-\in \ii \R$, from Theorem~\ref{thm:smallnormRUp} and RHP~\ref{rhp:bessel}--(3),(4) we bound
\begin{equation}\label{eq:auxxineg03}
[\bm B_0(\xi)^{-1}\bm \Delta_\zeta\bm R_{\bm \Upsilon}(\xi)\bm B_0(\xi)]_{21,-}=\Boh((1+|\xi|^{1/2})|\zs|^2\gt^{-1}),\quad \gt\to \infty,
\end{equation}
uniformly for $-\delta \zs^2<\xi<0$ and $s$ within the supercritical regime.

Next, using RHP~\ref{rhp:bessel}--(4) and the boundedness of $\bm R^\nc(\zeta),\bm A^\nc(\zeta)$, for $-\varepsilon<\xi<0$ we obtain
\begin{multline}\label{eq:auxxineg01}
\left[ (\bm I+\bm E_{21})\bm B(\xi)^{-1}\bm R_{\bm \Upsilon}(\xi)^{-1}|\zs|^{-\sp_3/2}\bm\Delta_\zeta(\bm R^\nc(\zeta)\bm A^{\nc }(\zeta))|\zs|^{\sp_3/2}\bm R_{\bm \Upsilon}(\xi)\bm B(\xi)(\bm I-\bm E_{21}) \right]_{21,-}\\
= \left[(\bm I-\frac{\log \xi}{2\pi \ii} \bm E_{12})\bm B_0(\xi)^{-1}\Boh(|\zs|)\bm B_0(\xi)(\bm I+\frac{\log \xi}{2\pi \ii} \bm E_{12})\right]_{21,-}=\Boh(|\zs|).
\end{multline}
Next, to treat the remaining interval $-\delta\zs^2<\xi<-\varepsilon$, we use \eqref{eq:expwtmcPF} to obtain
\begin{equation}\label{eq:wtmsfPexpansion}
\wt{\mcal P}_\gt(\xi)=\frac{\gt}{\zs^2}\msf c_{\msf P}\xi+\Boh\left( \frac{\zs}{\gt} \right),
\end{equation}
uniformly in the mentioned interval, and also uniformly for $s$ within the supercritical regime. This estimate yields that
\begin{equation}\label{eq:estexpwtmcalP}
\ee^{\wt{\mcal P}_\gt(\xi)}=\left(1+\Boh\left(\frac{\zs}{\gt}\right)\right)\ee^{\msf c_{\msf P}\gt\xi/\zs^2},\quad |\xi|\leq \delta \zs^2.
\end{equation}
In particular, $\ee^{\gt{\mcal P}(\xi)}$ remains bounded for $-\delta \zs^2<\xi<-\varepsilon$. Using then RHP~\ref{rhp:bessel}--(3) and the boundedness of $\bm R_{\bm \Upsilon}(\xi)$, $\bm R^\nc(\zeta)$ and $\bm A^\nc(\zeta)$, we get
\begin{multline}\label{eq:auxxineg02}
\Big[ (\bm I+(1+\ee^{\wt{\mcal P}_\gt(\xi)})\bm E_{21})\bm B(\xi)^{-1}\bm R_{\bm \Upsilon}(\xi)^{-1}|\zs|^{-\sp_3/2} \\
\bm\Delta_\zeta(\bm R^\nc(\zeta)\bm A^{\nc }(\zeta))|\zs|^{\sp_3/2}\bm R_{\bm \Upsilon}(\xi)\bm B(\xi)(\bm I-(1+\ee^{\wt{\mcal P}_\gt(\xi)})\bm E_{21}) \Big]_{21,-}=
\Boh(|\zs|(1+  |\xi|^{1/2})).
\end{multline}
valid uniformly for $-\delta \zs^2<\xi<-\varepsilon$, and uniformly for $s$ within the supercritical regime.

In summary, combining \eqref{eq:auxxineg03}, \eqref{eq:auxxineg01} and \eqref{eq:auxxineg02} into \eqref{eq:auxxineg04} and \eqref{eq:auxxineg05}, we obtain the second estimate in (iii) for $\zeta<0$, concluding the proof.
\end{proof}

Finally, we will also need a rough control of the behavior of the model problem as $\zeta\to \infty$, valid also with uniform control as $\gt\to \infty$.

\begin{prop}\label{prop:asymptmodelproblemmatching}
    For $\sad$ within any of the subcritical or critical regimes, the asymptotic behavior \eqref{eq:asymptPhit} is valid with error term $\Boh(\zeta^{-1})$ being uniform in $s$. 
    
    On the other hand, in the supercritical regime \eqref{eq:asymptPhit} takes the more precise form
    $$
    \bm \Phi_\gt(\zeta)=(\zeta-\zs)^{-\sp_3/4}\bm U_0\left(\bm I+\Boh\left(\frac{\zs^2}{\zeta^{1/2}}\right)\right)\ee^{-\frac{2}{3}\zeta^{3/2}\sp_3},
    $$
    uniformly for $s$ within the supercritical regime, and valid for $|\zeta|\geq M\zs^4$, for some $M>0$ sufficiently large but independent of $s,\gt$.
\end{prop}

\begin{proof}
    Suppose first that $s$ is within the subcritical regime. We unwrap the transformations from Sections~\ref{sec:asymptrposcrit} and \ref{sec:subcrintdiffPIIsimplif}, obtaining that for $\zeta$ sufficiently large,
    $$
    \bm \Phi_\gt(\zeta)=\zs^{-\sp_3/4}\bm R^\pc \left(\frac{\zeta}{\zs}\right)\zs^{\sp_3/4} \wh{\bm R}(\zeta)\wh{\bai}(\zeta),
    $$
    where $\bm R^\pc$ solves RHP~\ref{rhp:Rpc}, $\wh{\bm R}$ solves RHP~\ref{rhp:modelhatR}, and $\wh{\bai}$ is as in \eqref{eq:asymptbehwhbai}. Using Corollary~\ref{cor:estRpc} and Theorem~\ref{thm:smallnormhatR}, we see that as $\zeta\to \infty$,
    $$
    \zs^{-\sp_3/4}\bm R^\pc \left(\frac{\zeta}{\zs}\right)\zs^{\sp_3/4}=\bm I+\Boh(\ee^{-\eta\zs^{3/2}}\zeta^{-1}),\quad \text{and}\quad \wh{\bm R}(\zeta)=\bm I+\Boh(\ee^{-\eta\zs}\zeta^{-1}),
    $$
    for some $\eta>0$, where the error term is now uniform for $s$ in the subcritical regime. The result now follows from \eqref{eq:asymptbehwhbai}.

    The result for $s$ in the critical regime follows similarly, simply unfolding the transformations performed in Section~\ref{sec:asymptregcrit}, we skip the details.

    For $s$ in the supercritical regime, the unfolding of the transformations in Section~\ref{sec:asymptrnegcrit} unravel that for $\zeta$ away from $\zs$, 
    \begin{equation}\label{eq:Phigtspasympt01}
    \bm \Phi_\gt(\zeta)=\left(\bm I-\frac{\ii}{4}\zs^2\bm E_{21}\right)\bm R^\nc(\zeta-\zs)(\zeta-\zs)^{-\sp_3/4}\bm U_0 \ee^{-(\frac{2}{3}(\zeta-\zs)^{3/2}+\zs(\zeta-\zs)^{1/2})\sp_3},
    \end{equation}
    where $\bm R^\nc$ satisfies RHP~\ref{rhp:Rnc} and the factors to its right are coming from $\bm G^\nc$ in \eqref{deff:Gnc}.

    Using Theorem~\ref{thm:smallnormnegcrit} we see still for $\zeta$ large, meaning $|\zeta-\zs|\geq M$ for some $M>0$ sufficiently large, but fixed and independent of $s,\gt$,
    $$
    \bm R^\nc(\zeta-\zs)(\zeta-\zs)^{-\sp_3/4}\bm U_0
    =(\zeta-\zs)^{-\sp_3/4}\bm U_0\left(\bm I+\Boh\left(\frac{1}{\gt^{1/2}(\zeta-\zs)^{1/2}}\right)\right),
    $$
    as well as for $|\zeta|\geq M|\zs|^{4}$, for any $M>0$ sufficiently large but independent of $s,\gt$,
    $$
    \ee^{-(\frac{2}{3}(\zeta-\zs)^{3/2}+\zs(\zeta-\zs)^{1/2})\sp_3}=
    \left(\bm I+\Boh\left(\frac{\zs^2}{\zeta^{1/2}}\right) \right)\ee^{-\frac{2}{3}\zeta^{3/2}\sp_3}.
    $$
    The result then follow plugging these last two estimates into \eqref{eq:Phigtspasympt01}, and commuting the term $(\zeta-\zs)^{-\sp_3/4}$ with the matrix $\bm I-\ii(\zs^2/4)\bm E_{21}$, which again produces an error of order $\Boh(\zeta^2/\zeta^{1/2})$.
\end{proof}

This last result concludes the analysis of all the quantities we will need from the model problem, and we now finally turn our attention to the analysis of the discrete OPs themselves.

\section{The Riemann-Hilbert Problem for discrete orthogonal polynomial ensembles}\label{sec:RHPOP}

The RHP for continuous weight functions was first introduced by Fokas, Its and Kitaev \cite{FokasItsKitaev92}, and in combination with the Deift-Zhou nonlinear steepest descent method for its asymptotic analysis \cite{deift_book, Deiftetal1999, Deift1993, DeiftZhouPII95} it has become an ubiquitous tool in the asymptotic analysis of models which are solvable by orthogonal polynomials. 

The RHP for discrete orthogonality was first introduced by Borodin and Boyarchenko \cite{Borodin2003a}. In the context of the six-vertex model, this discrete RHP for OPs was notably used by Bleher and Liechty \cite{bleher_liechty_PDWBC, bleher_liechty_book, bleher_liechty_CPAM}, see also the recent follow-up \cite{GorinLiechty25} by Gorin and Liechty. However, they consider the six-vertex model with boundary conditions and lattices different than the ones we consider here, corresponding to different classes of weights of orthogonality. In related contexts, the discrete RHP for OPs was used by Baik, Kriecherbauer, McLaughlin and Miller \cite{BKMMbook} motivated by random tiling models, and by Liechty and Wang \cite{LiechtyWang2016} and Buckingham and Liechty \cite{BuckinghamLiechty19} in the context of non-intersecting Brownian bridges on the unit circle. The foundational work of Bleher and Liechty \cite{BleherLiechtyIMRN2011} computes asymptotics for discrete OPs for weights on the full real line; even though our setup here is different in various manners, the streamlined presentation of the asymptotic analysis in \cite{BleherLiechtyIMRN2011} was very helpful to us. 

Recall that we are interested in OPs for a deformed discrete weight $\gwd$ as in \eqref{eq:deformedweight}, with the undeformed component $\gw$ satisfying Assumptions~\ref{assumpt:potential_formal}. The Riemann-Hilbert problem for the orthogonal polynomials for the discrete weight $\gwd$ is the following.
\begin{rhp}\label{rhp:OP}
Find a matrix-valued function $\bm D: \C\setminus \Z_{>0}\to \C^{2\times 2}$ with the following properties:
\begin{enumerate}[(1)]
\item $\bm D$ is analytic.
\item $\bm D$ has simple poles at each point $k\in\Z_{>0}$, and its residue satisfies
$$
\res_{z=k}\bm D(z)=\lim_{z\to k} \bm D(z)\bm E_{12}\gwd(k).
$$

\item For some function $r:\Z_{>0}\to [0,+\infty)$ satisfying $r(z)\to 0$ as $z\to \infty$, the matrix $\bm D$ admits an asymptotic expansion of the form
$$
\bm D(z)\sim\left(\bm I+\sum_{j=1}^\infty \frac{\bm D_j}{z^j}\right)
\begin{pmatrix}
z^\gn & 0 \\ 0 & z^{-\gn}
\end{pmatrix},
\quad 
z\to \infty \mbox{ with } z\in \C\setminus \bigcup_{k\in \Z_{>0}} D_{r}(k),
$$
where the radius $r=r(k)$ may depend on the center of the disk $D_r(k)$, $\bm I$ is the $2\times 2$ identity matrix, and $\bm D_k$ is a $2\times 2$ matrix independent of $z$.
\end{enumerate}
\end{rhp}

The corresponding Christoffel-Darboux kernel $\msf K_\gn^\gsig$ from \eqref{deff:CDkernel} expresses as
$$
\msf K_\gn^\gsig(x,y)=\msf K_\gn^\gsig(x,y\mid s)=
-\frac{1}{x-y}\bm e_2^T\bm D(y\mid s)^{-1}\bm D(x\mid s)\bm e_1,\quad x\neq y,
$$
and in the confluent limit $y\to x$ this identity becomes
\begin{equation}\label{eq:CDkerneldiagRHPD}
\msf K_\gn^\gsig(x,x\mid s)=-\bm e_2^T\bm D(x\mid s)^{-1}\bm D'(x\mid s)\bm e_1,\quad x\in \C.
\end{equation}
As a consequence, \eqref{eq:deformationformula} yields the relation
\begin{equation}\label{eq:deffformula2}
\log \frac{\msf Z_\gn^\gsig(S)}{\msf Z^\gsig_\gn(s)}
=
\int_{s}^{S} \sum_{x\in \Z_{>0}} \ee^{-t(x-a\gn)-v }\gw(x\mid v) \bm e_2^T\bm D(x\mid v)^{-1}\bm D'(x\mid v)\bm e_1\,  \dd v.
\end{equation}

Our next goal is thus to perform the Deift-Zhou nonlinear steepest descent method for this RHP. The steps in the asymptotic analysis that we will perform here are inspired by \cite{BKMMbook, BleherLiechtyIMRN2011}. However, many technical adaptations in this analysis are necessary, in virtue of the presence of the factor $\gsig$ in the weight of orthogonality, and the corresponding poles of the factor $1/\gsig$ that appears in this analysis. In a more fundamental contrast with previous works in the literature, we will use the model problem RHP~\ref{rhp:modelPhi} in the construction of the so-called local parametrix.


\subsection{Scaling of the RHP}\label{sec:scalingRHP}\hfill

As a first step, we scale the nodes of the RHP. For that, recall that the scaled weight $\gW$ is as in Assumptions~\ref{assumpt:potential_formal}, and introduce
\begin{equation}\label{deff:gsiggW}
\gsig(x)=\gsig(x\mid s,t)\deff 1+\ee^{-t\gn(x-\ga)-s} \qquad \text{and}\qquad  \gWd(x)\deff \frac{1}{\gn}\gwd(\gn x)=\gsig(x)\gW(x).
\end{equation}
To streamline notation we mostly write $\sigma=\sigma(x)$, writing the other parameters explicitly as $\sigma=\sigma(x\mid s),\sigma=\sigma(x\mid s,t)$ etc, solely when necessary to avoid confusion. Also, when we need to evaluate one parameter to a specific value, say $t$ at a value $\bullet$, we write this evaluation as $\sigma(x\mid t=\bullet)=\sigma(x\mid \bullet,s)$ etc.

With $\bm D$ being the solution to the RHP~\ref{rhp:OP}, make the change
$$
\bm X(z)=\gn^{-\gn\sp_3}\bm D(\gn z), \quad z\in \C\setminus \frac{1}{\gn}\Z_{>0}.
$$
Then $\bm X$ satisfies the following RHP.
\begin{rhp}\label{rhp:OPX}
Find a $2\times 2$ matrix-valued function $\bm X:\C\setminus \frac{1}{N}\Z_{>0}\to \C^{2\times 2}$ with the following properties.
\begin{enumerate}[(1)]
\item $\bm X$ is analytic on $\C\setminus \frac{1}{\gn}\Z_{>0}$.
\item $\bm X$ has simple poles at each point $x_k\deff \frac{1}{\gn}k\in \frac{1}{\gn}\Z_{>0}$, and its residue satisfies
$$
\res_{z=x_k}\bm X(z)=\lim_{z\to x_k} \bm X(z)\bm E_{12}\gWd(x_k).
$$
\item For some function $r:\frac{1}{n}\Z_{> 0}\to [0,+\infty)$ satisfying $r(z)\to 0$ as $z\to \infty$, the matrix $\bm X$ admits an asymptotic expansion of the form
$$
\bm X(z)\sim\left(\bm I+\sum_{k=1}^\infty \frac{\bm X_k}{z^k}\right)
\begin{pmatrix}
z^\gn & 0 \\ 0 & z^{-\gn}
\end{pmatrix},
\quad 
z\to \infty \mbox{ with } z\in \C\setminus \bigcup_{k=1}^\infty D_{r(x_k)}(x_k),
$$
where for each $k\in \Z_{>0}$, $\bm X_k$ is a $2\times 2$ matrix independent of $z$.
\end{enumerate}
\end{rhp}

Formula~\eqref{eq:CDkerneldiagRHPD} becomes
\begin{equation}\label{eq:KNsigX}
\msf K_\gn^\gsig(x,x\mid s)=-\frac{1}{\gn}\bm e_2^T \bm X\left(\frac{x}{\gn}\mid s\right)^{-1}\bm X'\left(\frac{x}{\gn}\mid s\right)\bm e_1,
\end{equation}
whereas~\eqref{eq:deffformula2} then updates to
\begin{equation}\label{eq:deffformula3}
\log \frac{\msf Z_\gn^\gsig(S)}{\msf Z^\gsig_\gn(s)}
=
\int_{s}^{S} \sum_{x\in \frac{1}{\gn}\Z_{>0}} \frac{\gsig(x\mid u)-1}{\gsig(x\mid u)}\gWd(x\mid u) \bm e_2^T\bm X(x)^{-1}\bm X'(x)\bm e_1\,  \dd u.
\end{equation}

In what follows, we will need to consider $1/\gsig(z)$ for values of $z$ away from certain subintervals of the real axis, and for that reason it is convenient to single out certain properties of the poles of this quotient, that is, of the zeros of $\gsig$.

\begin{prop}\label{prop:zerossigma}
The zeros of the function $\gsig$ are at the points
\begin{equation}\label{eq:zerosrqmw}
w_k\deff \frac{(2k+1)\pi \ii}{\gn}+z_\gn(s),\; k\in \Z,\quad z_\gn(s)\deff -\frac{s}{t\gn}+\ga.
\end{equation}

In particular, under Assumptions~\ref{assumpt:parameterregimes} the following holds.
\begin{enumerate}[(i)]
\item There exists $\delta>0$ independent of $s,t,\gn$ for which
$$
\sigma(z)\neq 0 \quad \text{for}\quad \re z>\ga+\delta.
$$
\item For any $\theta_0\in (0,\pi/2)$, there exist $\rho_0>0$ and $\delta>0$ for which
$$
|\sigma(z_\gn(s)\pm \rho \ee^{\ii \theta})|\geq \delta\quad \text{for any } \theta \in [-\theta_0,\theta_0] \text{ and any } \rho \in [0,\rho_0].
$$
The values $\rho_0$ and $\delta$ may depend on $s$ but are independent of $t$.
\end{enumerate}
\end{prop}
\begin{proof}
It follows directly from the definition of $\gsig$.
\end{proof}

\section{Asymptotic analysis of the RHP for dOPs}\label{sec:asymptanaldOP}

We now apply the Deift-Zhou steepest descent method to the RHP $\bm X$ for OPs described in Section~\ref{sec:scalingRHP}. The method itself consists in performing several transformations on the original RHP, with the goal of obtaining a RHP which can be solved by perturbation theory. 

The first transformations are in a sense algebraic in nature, and will be carried out in a uniform manner, regardless of the subcritical regime of interest. These steps follow closely the already established roadmap of the Deift-Zhou method for discrete RHPs \cite{BKMMbook, BleherLiechtyIMRN2011}: moving from discrete to a continuous RHP, introduction of the $g$-function and opening of lenses. The main difference in our work here is that we have to cope with the poles of the factor $1/\gsig$ throughout all these steps.

As one of the final steps of the analysis, we will need to construct the so-called local and global parametrices. They will be constructed with the same auxiliary functions and model problem, regardless of the regime considered, but their asymptotic behavior will differ depending on which of the subregimes from Assumptions~\ref{assumpt:parameterregimes} that we are considering. 

The last step of the asymptotic analysis is the application of the small norm theory, as usual.

\subsection{From discrete to continuous RHP}\hfill

The first transformation consists of replacing the pole conditions for $\bm X$ to jump conditions for a new matrix $\bm Y$, and it relies on the structure of the support of the equilibrium measure $\gequil$.

Set
$$
\Pi_{\gn}(z)\deff\frac{\sin (\pi \gn z)}{\pi \gn}=\frac{\ee^{\pi \ii \gn z}-\ee^{-\pi \ii \gn z}}{2\pi \ii\gn },
$$
which is an analytic function on a neighborhood of the real line that satisfies
$$
\Pi_\gn(x_k)=0,\qquad \Pi_\gn(z)\neq 0 \text{ for } z\in \R_>\setminus \frac{1}{\gn}\Z,\qquad \Pi'_\gn(x_k)=(-1)^k,
$$
where we recall that $x_k=\frac{1}{\gn}k\in \frac{1}{\gn} \Z_{> 0}$ are the (simple) poles of the matrix $\bm X$.

We also need to introduce certain contours and regions that will be used for the first transformation. These transformations will depend on the position of the point $z_\gn(s)$ introduced in \eqref{eq:zerosrqmw}, which we always assume to satisfy
$$
|z_\gn(s)-\ga|<\varepsilon,
$$
for some $\varepsilon>0$ independent of $\gn$, which can be made arbitrarily small but will be kept fixed. Recalling \eqref{eq:zerosrqmw}, we see that under Assumptions~\ref{assumpt:parameterregimes} this is always true for $\gn$ sufficiently large and independent of the parameters $s,t$, and $\varepsilon>0$ can be chosen to ensure that $z_\gn(s)>0$ as well.

Fix some $\delta>0$ and consider contours $\mcal L^\Delta_\pm$ and $\mcal L^\nabla_\pm$ as follows. 

The contour $\mcal L^\Delta_+$ starts at the origin and moves upwards to a point $\ii\delta$. It then turns right, moving torwards the real axis, up until it reaches the point $z_\gn(s)$, making an angle $\pi/3$ with the interval $[0,z_\gn(s)]$. 

The contour $\mcal L^\nabla_+$ starts from $z_\gn(s)$ on the upper half plane, with angle $\pi/3$ with the interval $[z_\gn(s),\infty)$, and then extends horizontally to $\infty$ running asymptotically parallel with the positive real axis, within a positive but small distance from it. 

We then set $\mcal L_-^\Delta$ and $\mcal L_-^\nabla$ to be the contours obtained from $\mcal L_+^\Delta$ and $\mcal L_+^\nabla$ by complex conjugation, respectively. These contours are displayed in Figure~\ref{Fig:RHP1stTransfMod}.
\begin{figure}[t]
\centering
\begin{minipage}{.5\textwidth}
\centering
\includegraphics[scale=0.8]{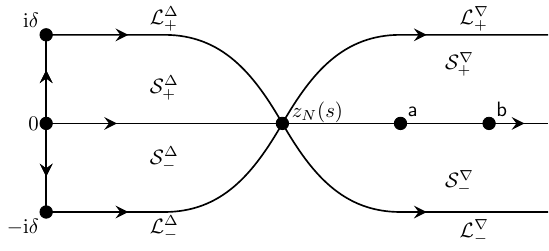}
\end{minipage}%
\begin{minipage}{.5\textwidth}
\centering
\includegraphics[scale=0.8]{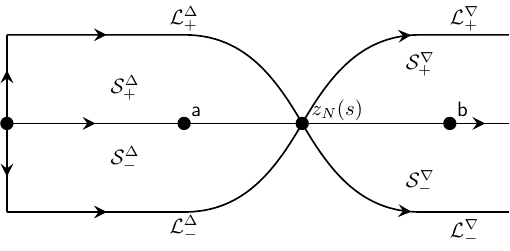}
\end{minipage}
\caption{The contours and regions used for the transformation $\bm X\mapsto \bm Y$ of the RHP, corresponding to $s>0$ (left) and $s<0$ (right).}
\label{Fig:RHP1stTransfMod}
\end{figure}

Also, we consider $\mcal S^\Delta_\pm$ to be the bounded region determined by $\mathcal L_\pm^\Delta$ and $[0,z_\gn(s)]$, and $\mcal S^\nabla_\pm$ to be the region between $\mathcal L_\pm^\nabla$ and $[z_\gn(s),+\infty)$. These regions are also displayed in Figure~\ref{Fig:RHP1stTransfMod}. Observe in particular that $\ga$ is on the interior of the region determined by either $\mcal L^\Delta_+\cup \mcal L^\Delta_-$ or $\mcal L^\nabla_+\cup \mcal L^\nabla_-$, depending on whether $s>0$ or $s<0$, we refer again to Figure~\ref{Fig:RHP1stTransfMod}. Nevertheless, for the first transformation this distinction can still be treated in an uniform manner.

Using these neighborhoods, we transform
\begin{equation}\label{eq:transfXtoY}
\bm Y(z)=
\begin{dcases}
\bm X(z)\left(\bm I - \frac{\gWd(z)}{\Pi_\gn(z)} \ee^{\pm \ii \pi \gn z}\bm E_{12} \right),& z\in \mcal S^\nabla_\pm, \\
\bm X(z)\Pi_\gn(z)^{-\sp_3}\left(\bm I - \frac{1}{\gWd(z)\Pi_\gn(z)} \ee^{\pm \ii \pi \gn z}\bm E_{21} \right),& z\in \mcal S^\Delta_\pm, \\
\bm X(z), & \mbox{otherwise}.
\end{dcases}
\end{equation}
The function $\Pi_\gn$ has no zeros outside the real axis. Using Proposition~\ref{prop:zerossigma} we see that $\gWd$ has no zeros on $\mcal S_\pm^\Delta$ for any $\gn>0$, so the transformation above is well defined for any $\delta>0$. We will eventually take $\delta>0$ sufficiently small but independent of $\gn$.

The properties of $\Pi_\gn$ listed earlier ensure that the poles coming from $1/\Pi_\gn$ cancel the poles of $\bm X$, with the cost of creating jump conditions, so that $\bm Y$ satisfies a continuous RHP that we describe next.

Set
$$
\Gamma_{\bm Y}\deff\mathcal L_+^\Delta\cup \mathcal L_-^\Delta\cup \R_>\cup \mathcal L_+^\nabla\cup \mathcal L_-^\nabla,\quad  \mcal S^\Delta\deff \mcal S^\Delta_+\cup \mcal S^\Delta_-,\quad \mcal S^\nabla\deff \mcal S^\nabla_+\cup \mcal S^\nabla_-,\quad \mcal S\deff \mcal S^\Delta\cup\mcal S^\nabla,
$$
with the orientation of the arcs of $\Gamma_\bm Y$ as shown in Figure~\ref{Fig:RHP1stTransfMod}. The matrix $\bm Y$ satisfies the following Riemann-Hilbert problem.
\begin{rhp}\label{rhp:OPY}
Find a $2\times 2$ matrix-valued function $\bm X:\C\setminus \Gamma_{\bm Y}\to \C^{2\times 2}$ with the following properties.
\begin{enumerate}[(1)]
\item $\bm Y$ is analytic on $\C\setminus \Gamma_{\bm Y}$.
\item $\bm Y$ has continuous boundary values $\bm Y_\pm$ along $\Gamma_{\bm Y}$, and they are related by
$$
\bm Y_+(z)=\bm Y_-(z)\bm J_{\bm Y}(z),\quad z\in \Gamma_{\bm Y},
$$
where the jump matrix $\bm J_{\bm Y}$ is given by
$$
\bm J_{\bm Y}(z)\deff 
\begin{dcases}
\bm I - \frac{2\pi \ii \gn}{\gWd(z)}\bm E_{21},& z\in (0,z_\gn(s)), \\
\bm I - 2\pi \ii \gn\gWd(z)\bm E_{12},& z\in (z_\gn(s),\infty), \\
\bm I \pm \frac{\gWd(z)}{\Pi_\gn(z)}\ee^{\pm \ii \pi \gn z}\bm E_{12},& z\in \mathcal L^\nabla_\pm, \\
\Pi_\gn(z)^{\pm \sp_3}\pm \frac{\ee^{\pm \pi \ii \gn z}}{\gWd(z)}\bm E_{21}  & z\in \mathcal L^\Delta_\pm.
\end{dcases}
$$
\item As $z\to \infty$, 
$$
\bm Y(z)=(\bm I+\Boh(z^{-1}))
\begin{pmatrix}
z^n & 0 \\
0 & z^{-n}
\end{pmatrix}.
$$

\item As $z\to z_\gn(s)$, the matrix $\bm Y$ has the following behavior,
$$
\bm Y(z)=
\begin{cases}
    \Boh(1), & \text{if }z_\gn(s)\in \R_>\setminus \frac{1}{\gn}\Z_{>0}, \\
    \Boh(1), & \text{if }z_\gn(s)\in \frac{1}{\gn}\Z_{>0} \text{ and } z\to z_\gn(s) \text{ along } \mcal S, \\ 
    \Boh
    \begin{pmatrix}
        1 & (z-z_{\gn}(s))^{-1} \\
        1 & (z-z_{\gn}(s))^{-1}
    \end{pmatrix},
     & \text{if }z_\gn(s)\in \frac{1}{\gn}\Z_{>0} \text{ and } z\to z_\gn(s) \text{ along } \C\setminus \mcal S.
\end{cases}
$$
\end{enumerate}
\end{rhp}

For $\gsig\equiv 1$, the form of the transformation $\bm Y\mapsto \bm X$ is standard in the RH analysis of discrete OPs, see for instance \cite{BKMMbook, BleherLiechtyIMRN2011}. Due to the presence of the factor $\gsig$ we needed to consider the sets $\mcal S^\Delta$ and $\mcal S^\nabla$ in such a way that their boundaries emerge from the point $z_\gn(s)$ rather than from the endpoint $\ga$ of $\supp\gequil$. As a consequence, in the case $z_\gn(s)\in \frac{1}{\gn}\Z_{>0}$ we do not cancel the pole behavior of $\bm X$ when approaching $z_\gn(s)$ from outside $\mcal S^\Delta\cup\mcal S^\nabla$, and the somewhat nonstandard conditions in RHP~\ref{rhp:OPY}--(4) have to be placed. The exact form of RHP~\ref{rhp:OPY}--(4) may be computed directly or, alternatively, in a more systematic way with the aid of Lemma~\ref{lem:resRHPaux}. 

\subsection{Introduction of the g-function}\hfill

For the second transformation, we use the equilibrium measure $\gequil$ from Assumptions~\ref{assumpt:equilmeasure_formal} and define
$$
g(z)\deff\int \log(z-x)\dd\gequil(x),\quad z\in \C\setminus (-\infty,\gb],
$$
where the branch of the argument is the principal one. This is the so-called $g$-function associated to $\gequil$.

\begin{prop}\label{prop:properties_g_function}
The $g$-function is analytic on $\C\setminus (-\infty,\gb]$ and satisfies the following properties
\begin{enumerate}[(a)]
\item As $z\to \infty$ away from $(-\infty,\gb]$,
$$
g(z)=\log z +\Boh(z^{-1}).
$$

\item For $x<0$,
$$
g_+(x)-g_-(x)=2\pi \ii.
$$
\item For any $x>0$,
$$
g_+(x)+g_-(x)=-2U^{\gequil}(x),\quad U^{\gequil}(x)\deff \int\log\frac{1}{|x-y|}\dd\gequil(y).
$$
In particular, $g_++g_-$ is purely real over the positive real axis. In general, for any $z\in \C\setminus \R$,
$$
\re g(z)=-U^{\gequil}(z).
$$
\item The variational conditions
$$
\gV(x)+2\ell-g_+(x)-g_-(x)=2U^{\gequil}(x)+\gV(x)+2\ell
\begin{cases}
>0, & x>\gb, \\
= 0, & \ga\leq x \leq \gb, \\
<0, & 0<x<\ga
\end{cases}
$$
hold true, where $\ell=\ell_\gV$ is the same constant that appears in Assumptions~\ref{assumpt:equilmeasure_formal}.
\item For $x\in (0,\gb)$,
$$
g_+(x)-g_-(x)=2\pi \ii \gequil([x,\gb])=2\pi \ii (1-\gequil([0,x])).
$$
In particular this difference is purely imaginary with positive imaginary part, and on the interval $(0,\ga)$ it reduces to
$$
g_+(x)-g_-(x)=2\pi \ii \left(1-x\right),\quad 0<x<\ga.
$$
\end{enumerate}
\end{prop}

The proof of Proposition~\ref{prop:properties_g_function} is standard in RHP literature, see for instance \cite[Section~3]{BleherLiechtyIMRN2011}, and we skip it.

For the transformations that follow, it is convenient to introduce
\begin{equation}\label{eq:deffzmaxzmin}
\wh z_\gn(s)\deff \min \{z_\gn(s),\ga\}=
\begin{cases}
\ga, & \text{if }s<0,\\
z_\gn(s), & \text{if }s>0,
\end{cases}
\quad \text{and}\quad
\wc z_\gn(s)\deff \max \{z_\gn(s),\ga\}=
\begin{cases}
z_\gn(s), & \text{if }s<0,\\
\ga, & \text{if }s>0.
\end{cases}
\end{equation}
With $\gC \in \R$ being the constant from Assumptions~\ref{assumpt:potential_formal}, the second transformation is
$$
\bm T(z)=\ee^{(-\gn\ell +\gC )\sp_3}\bm Y(z)\ee^{-\gn g(z)\sp_3}\ee^{(\gn\ell -\gC)\sp_3},\quad z\in \C\setminus \Gamma_{\bm T}, \quad \Gamma_{\bm T}\deff\Gamma_{\bm Y}.
$$
We stress that this transformation does not affect the jump contour, and we refer again to Figure~\ref{Fig:RHP1stTransfMod} for the display of this contour.

Using the jump properties of the $g$-function as described by Proposition~\ref{prop:properties_g_function} , it follows that $\bm T$ satisfies the following Riemann-Hilbert problem.
\begin{rhp}\label{rhp:OPT} 
Find a $2\times 2$ matrix-valued function $\bm T$ with the following properties.
\begin{enumerate}[(1)]
\item $\bm T$ is analytic on $\C\setminus \Gamma_{\bm T}$.
\item The matrix $\bm T$ has continuous boundary values $\bm T_\pm$ along $\Gamma_{\bm T}$, which satisfy the jump relation $\bm T_+(z)=\bm T_-(z)\bm J_{\bm T}(z)$, $z\in \Gamma_{\bm T}$, with
$$
\bm J_{\bm T}(z)\deff 
\begin{dcases}
\ee^{2\pi \ii \gn z\sp_3}-\frac{2\pi \ii \gn }{\gWd(z)}\ee^{-\gn (g_+(z)+g_-(z)-2\ell)-2\gC}\bm E_{21}, 
& 0<z<\wh z_\gn(s), \\
\bm I-2\pi \ii \gn \gWd(z) \ee^{\gn(g_+(z)+g_-(z)-2\ell)+2\gC}\bm E_{12}, 
& z>\gb,\\
\ee^{-\gn(g_+(z)-g_-(z))\sp_3}-2\pi \ii \gn \gWd(z)\ee^{\gn(g_+(z)+g_-(z)-2\ell)+2\gC}\bm E_{12},
& \wc z_\gn(s)<z<b,
\end{dcases}
$$
and
$$
\bm J_{\bm T}(z)\deff 
\begin{dcases}
\ee^{-\gn(g_+(z)-g_-(z))\sp_3}-\frac{2\pi \ii\gn}{\gWd(z)}\ee^{-\gn(g_+(z)+g_-(z)-2\ell)-2\gC}\bm E_{21},
& \text{if } s>0 \text{ and } z\in (\wh z_\gn(s),\wc z_\gn(s)), \\
\ee^{2\pi \ii \gn z\sp_3}-2\pi \ii \gn \gWd(z)\ee^{\gn(g_+(z)+g_-(z)-2\ell)+2\gC}\bm E_{12}, 
& \text{if }s<0 \text{ and } z\in (\wh z_\gn(s),\wc z_\gn(s)),
\end{dcases}
$$
which covers the jumps on the real axis, and
$$
\bm J_{\bm T}(z)=
\begin{dcases}
\Pi_\gn(z)^{\pm \sp_3}\pm \frac{\ee^{-\gn(2g(z)-2\ell\mp \pi \ii z)-2\gC}}{\gWd(z)}\bm E_{21},
& z\in \mcal S_\pm^\Delta, \\
\bm I\pm \frac{\gWd(z)}{\Pi_\gn(z)}\ee^{\gn(2g(z)-2\ell\pm \pi \ii z)+2\gC}\bm E_{12}, 
& z\in \mcal S_\pm^\nabla,
\end{dcases}
$$
which covers the jumps outside the real axis.
\item As $z\to \infty$,
$$
\bm T(z)=\bm I+\Boh(z^{-1}).
$$

\item As $z\to z_\gn(s)$, the matrix $\bm T$ has the following behavior,
$$
\bm T(z)=
\begin{cases}
    \Boh(1), & \text{if }z_\gn(s)\in \R_>\setminus \frac{1}{\gn}\Z_{>0}, \\
    \Boh(1), & \text{if }z_\gn(s)\in \frac{1}{\gn}\Z_{>0} \text{ and } z\to z_\gn(s) \text{ along } \mcal S, \\ 
    \Boh
    \begin{pmatrix}
        1 & (z-z_{\gn}(s))^{-1} \\
        1 & (z-z_{\gn}(s))^{-1}
    \end{pmatrix},
     & \text{if }z_\gn(s)\in \frac{1}{\gn}\Z_{>0} \text{ and } z\to z_\gn(s) \text{ along } \C\setminus \mcal S.
\end{cases}
$$
\end{enumerate}
\end{rhp}

To further help with the rewriting of the jumps above, we introduce the $\phi$-function
\begin{equation}\label{def:phi_function}
\phi(z)\deff \int_{\gb}^z\left(C^{\gequil}(s)+\frac{\gV'(s)}{2}\right)\dd s,\quad z\in \C\setminus (-\infty,\gb],\quad \text{where } C^\mu(z)\deff \int \frac{\dd\mu(s)}{s-z}, \; z\in \C\setminus \supp\mu,
\end{equation}
and where the path of integration in the definition of $\phi$ does not cross $(-\infty,\gb]$. This function $\phi$ is well defined and analytic in $\C\setminus (-\infty,\gb]$, and its relevant properties for the subsequent analysis are listed in the next result.

\begin{prop}\label{prop:properties_phi_function}
We have
$$
\phi(z)=-g(z)+\frac{1}{2}\gV(z)+\ell,\quad  z\in \C\setminus (-\infty,\gb].
$$
Also,
$$
\phi_+(z)+\phi_-(z)=0,\quad \ga<z<\gb,\qquad \phi_+(z)-\phi_-(z)=-2\pi \ii \gequil((z,\gb)),\quad 0<z<\gb.
$$
In particular,
$$
\phi_+(z)-\phi_-(z)=2\pi \ii \left(z-1\right),\quad 0<z<\ga.
$$
\end{prop}
\begin{proof}
Simply note that for $z\in \C\setminus (-\infty,\gb]$, 
$$
g'(z)=\int \frac{\dd \gequil (x)}{z-x}=-C^{\gequil}(z)=-\phi'(z)+\frac{\gV'(z)}{2}, 
$$
so for some constant $c$ 
$$
\phi(z)=-g(z)+\frac{\gV(z)}{2}+c.
$$
Taking the limit $z\to \gb$ and using properties of the $g$-function, we get that
\begin{align*}
0=2\phi(\gb) & =-2g_{\pm}(\gb)+\gV(\gb)+2c \\
				& =-g_+(\gb)-g_-(\gb)+\gV(\gb)+2c \\
				& =2U^{\gequil}(\gb)+\gV(\gb)+2c\\
				& = -2\ell +2c,
\end{align*}
where for the last identity we used the variational equality. From this, all the remaining properties follow from the properties of $g$.
\end{proof}

In terms of the $\phi$ function, the jump matrix for $\bm T$ rewrites as
$$
\bm J_{\bm T}(z)=
\begin{dcases}
\ee^{2\pi \ii \gn z\sp_3}-\frac{2\pi \ii \gn }{\gsig(z)\ee^{\gE(z)}}\ee^{\gn (\phi_+(z)+\phi_-(z))}\bm E_{21}, 
& 0<z<\wh z_\gn(s), \\
\bm I-2\pi \ii \gn \gsig(z)\ee^{\gE(z)} \ee^{-2\gn \phi(z)}\bm E_{12}, 
& z>\gb,\\
\ee^{\gn(\phi_+(z)-\phi_-(z))\sp_3}-2\pi \ii \gn \gsig(z)\ee^{\gE(z)}\bm E_{12},
& \wc z_\gn(s)<z<b,
\end{dcases}
$$
and
$$
\bm J_{\bm T}(z)=
\begin{dcases}
\ee^{\gn(\phi_+(z)-\phi_-(z))\sp_3}-\frac{2\pi \ii\gn}{\gsig(z)\ee^{\gE(z)} }\bm E_{21},
& \text{if } s<0 \text{ and } z\in (\wh z_\gn(s),\wc z_\gn(s)), \\
\ee^{2\pi \ii \gn z\sp_3}-2\pi \ii \gn \gsig(z)\ee^{\gE(z)}  \ee^{-\gn(\phi_+(z)-\phi_-(z))}\bm E_{12}, 
& \text{if }s>0 \text{ and } z\in (\wh z_\gn(s),\wc z_\gn(s)),
\end{dcases}
$$
which covers the jumps on the real axis, and
$$
\bm J_{\bm T}(z)=
\begin{dcases}
\Pi_\gn(z)^{\pm \sp_3}\pm \frac{\ee^{\gn(2\phi(z)\pm \pi \ii z)}}{\gsig(z)\ee^{\gE(z)} }\bm E_{21},
& z\in \mcal L_\pm^\Delta, \\
\bm I\pm \frac{\gsig(z)\ee^{\gE(z)} }{\Pi_{\gn}(z)}\ee^{-\gn(2\phi(z)\mp \pi \ii z)}\bm E_{12},
& z\in \mcal L_\pm^\nabla,
\end{dcases}
$$
which covers the jumps outside the real axis.

\subsection{Opening of lenses}\hfill

The next transformation is the opening of lenses. Consider the vertical segment $[\gb, \gb \pm \ii \delta]$ which connects the point $\gb$ to the boundary $\mcal L_\pm^{\nabla}$. This segment splits $\mcal S^\nabla_\pm$ into two regions, and we denote by $\mcal G_\pm^\nabla$ the bounded one amongst these regions. Additionaly, we also consider the region $\mcal T_\pm$, which is the region delimited by the contours $\mcal L_\pm^\Delta, \mcal L_\pm^\nabla$ and the horizontal straight line $[\pm \ii \delta,\pm \ii \delta+\infty]$. These regions are displayed in Figure~\ref{Fig:RHPSTransf}.

\begin{figure}[t]
\centering
\includegraphics[scale=1]{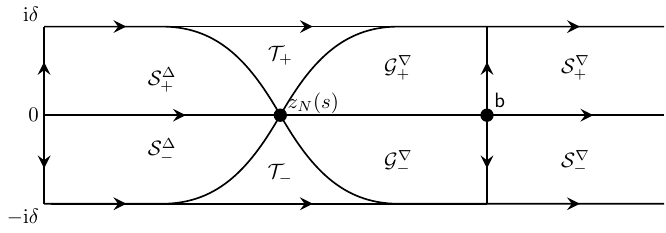}
\caption{The regions used for the transformation $\bm T\mapsto \bm S$ of the RHP. For this transformation, there is no distinction between the cases $z_\gn(s)<\ga$ and $z_\gn(s)>\ga$}
\label{Fig:RHPSTransf}
\end{figure}

Make the transformation
$$
{\bm S}(z)\deff
\begin{dcases}
(-2\pi \ii\gn)^{-\sp_3/2}\bm T(z)\left(\bm I\pm \dfrac{\ee^{2\gn \phi(z)}}{2\pi \ii \gn \gsig(z)\ee^{\gE(z)}}\bm E_{21}\right)(-2\pi \ii\gn)^{\sp_3/2}, & z\in \mcal G^\nabla_\pm, \\
(-2\pi \ii\gn)^{-\sp_3/2}\bm T(z)(\mp 2\pi \ii \gn)^{-\sp_3}\ee^{\mp \pi \ii \gn z\sp_3}  (-2\pi \ii\gn)^{\sp_3/2}, & z\in \mcal S^\Delta_\pm, \\
(-2\pi \ii\gn)^{-\sp_3/2}\bm T(z)(1-\ee^{\pm 2\pi \ii \gn z})^{-\sp_3}(-2\pi \ii\gn)^{\sp_3/2}, & z\in \mcal T_\pm,\\
(-2\pi \ii\gn)^{-\sp_3/2}\bm T(z)(-2\pi \ii\gn)^{\sp_3/2}, & \text{elsewhere}.
\end{dcases}
$$
Set
$$
\Gamma_{\bm S}\deff \Gamma_{\bm T}\cup [\gb,\gb+\ii \delta]\cup [\gb,\gb-\ii\delta],\quad \mcal G\deff \mcal G_+^\Delta\cup\mcal G_-^\Delta, \quad \mcal T\deff \mcal T_+\cup \mcal T_-,
$$
with the orientation of the segments $[\gb,\gb \pm \ii \delta]$ taken from $\gb$ to $\gb\pm\ii\delta$, and the remaining parts of $\Gamma_{\bm S}$ inheriting the orientation from $\Gamma_{\bm T}$, see Figure~\ref{Fig:RHPSTransf}.

Using the properties of the $\phi$-function just proved in Proposition~\ref{prop:properties_phi_function}, we arrive at the following RHP for $\bm S$.

\begin{rhp}\label{rhp:OPS} 
Find a $2\times 2$ matrix-valued function $\bm S$ with the following properties.
\begin{enumerate}[(1)]
\item $\bm S$ is analytic on $\C\setminus \Gamma_{\bm S}$.
\item The matrix $\bm S$ has continous boundary values $\bm S_\pm$ along $\Gamma_{\bm S}$, satisfying the relation $\bm S_+(z)=\bm S_-(z)\bm J_{\bm S}(z)$, $z\in \Gamma_{\bm S}$. Recalling \eqref{eq:deffzmaxzmin}, this jump matrix is given by
$$
\bm J_{\bm S}(z)\deff
\begin{dcases}
-\bm I-\dfrac{\ee^{ \gn(\phi_+(z)+\phi_-(z)) }}{\gsig(z)\ee^{\gE(z)} }\bm E_{21},
& 0<z<\wh z_\gn(s), \\
\gsig(z)\ee^{\gE(z)}\bm E_{12}-\dfrac{1}{ \gsig(z)\ee^{\gE(z)}}\bm E_{21}, & \wc z_\gn(s)<z<\gb,\\
\bm I+\gsig(z)\ee^{\gE(z)}\ee^{- 2\gn\phi(z) }\bm E_{12}, & z>\gb, \\
-\bm \ee^{\gn(\phi_+(z)-\phi_-(z)-2\pi \ii z)\sp_3}-\frac{1}{\gsig(z)\ee^{\gE(z)}}\bm E_{21}, & 
 z\in (\wh z_\gn(s),\wc z_\gn(s)) \; (\text{when } s<0), \\
\frac{ \gsig(z)\ee^{\gE(z)}}{\ee^{\gn(\phi_+(z)+\phi_-(z))}}\bm E_{12}-\dfrac{\ee^{\gn(\phi_+(z)+\phi_-(z))}}{ \gsig(z)\ee^{\gE(z)}}\bm E_{21}, &  
 z\in (\wh z_\gn(s),\wc z_\gn(s)) \; (\text{when } s>0), 
\end{dcases}
$$
which covers the jumps on the real axis, 
$$
\bm J_{\bm S}(z)\deff
\begin{dcases}
\bm I+\frac{\ee^{2\gn \phi(z)}}{\gsig(z)\ee^{\gE(z)}}\frac{1}{1-\ee^{\pm 2\pi \ii \gn z}}\bm E_{21}, & z\in \partial \mcal S_\pm^\Delta\cap \partial \mcal T_\pm, \\
\left(\bm I+\ee^{-2\gn (\phi(z)-\pi \ii z)}\gsig(z)\ee^{\gE(z)}\bm E_{12}\right)
\left(\bm I+\frac{\ee^{2\gn \phi(z)}}{\gsig(z)\ee^{\gE(z)}}\frac{1}{1-\ee^{2\pi \ii \gn z}}\bm E_{21}\right), & z\in \partial \mcal G_+^\nabla\cap \partial \mcal T_+,\\
\left(\bm I+\frac{\ee^{2\gn \phi(z)}}{\gsig(z)\ee^{\gE(z)}}\frac{1}{1-\ee^{-2\pi \ii \gn z}}\bm E_{21}\right)
\left(\bm I+\ee^{-2\gn (\phi(z)+\pi \ii z)}\gsig(z)\ee^{\gE(z)}\bm E_{12}\right), & z\in \partial \mcal G_-^\nabla\cap \partial \mcal T_-,
\end{dcases}
$$
which covers the jumps along the contours emanating from $z_\gn(s)$ away from the real axis, 
$$
\bm J_{\bm S}(z)\deff 
\begin{dcases}
\bm I- \frac{\gsig(z)\ee^{\gE(z)- 2\gn( \phi(z)\mp \pi \ii z) }}{1-\ee^{\pm 2\pi \ii \gn z}}\bm E_{12}, 
& z\in [\gb\pm \ii \delta,\pm \ii \delta+\infty), \\
\bm I+\frac{\ee^{2\gn \phi(z)}}{\gsig(z)\ee^{\gE(z)}}\bm E_{21}, & z\in (\gb,\gb \pm \ii \delta), 
\end{dcases}
$$
which covers the jumps along the contours emanating from $\gb$ away from the real axis, and finally
$$
\bm J_{\bm S}(z)\deff
\begin{dcases}
(1-\ee^{\pm 2\pi \ii \gn z})^{\pm\sp_3}+\frac{\ee^{2\gn \phi(z)}}{\gsig(z)\ee^{\gE(z)}}\bm E_{21}, 
& z\in \partial \mcal S_\pm^\Delta\setminus \partial \mcal T_\pm , \\
(1-\ee^{\pm 2\pi \ii z})^{\pm\sp_3}, & z\in \partial \mcal T_{\pm}\cap [\pm \ii\delta,\pm \ii\delta+\infty), \\ 
\begin{pmatrix}
1 & \dfrac{\gsig(z) \ee^{\gE(z) -2\gn(\phi(z)-\pi \ii z)  }}{1-\ee^{2\pi \ii \gn z}} \\
\dfrac{\ee^{2\gn \phi(z)}}{\gsig(z)\ee^{\gE(z)}} & 1-\dfrac{1}{1-\ee^{-2\pi \ii \gn z}}
\end{pmatrix},
& z\in \partial\mcal G_+^\nabla\cap [\ii \delta,\ii \delta+\infty), \\
\begin{pmatrix}
1 - \dfrac{1}{1-\ee^{2\pi \ii \gn z}} & \dfrac{\gsig(z) \ee^{\gE(z) -2\gn(\phi(z)+\pi \ii z)  }}{1-\ee^{-2\pi \ii \gn z}} \\
\dfrac{\ee^{2\gn \phi(z)}}{\gsig(z)\ee^{\gE(z)}} & 1
\end{pmatrix},
& z\in \partial\mcal G_-^\nabla\cap [-\ii \delta,-\ii \delta+\infty), \\
\end{dcases}
$$
which covers the remaining parts of the jump contour.

\item As $z\to \infty$,
$$
\bm S(z)=\bm I+\Boh(z^{-1}).
$$

\item As $z\to z_\gn(s)$, 
$$
\bm S(z)=
\begin{cases}
    \Boh(1), & \text{if }z_\gn(s)\in \R_>\setminus \frac{1}{\gn}\Z_{>0}, \\
    \Boh(1), & \text{if }z_\gn(s)\in \frac{1}{\gn}\Z_{>0} \text{ and } z\to z_\gn(s) \text{ along } \mcal S\cup \mcal G, \\ 
    \Boh
    \begin{pmatrix}
        (z-z_{\gn}(s))^{-1} & 1 \\
        (z-z_{\gn}(s))^{-1} & 1
    \end{pmatrix},
     & \text{if }z_\gn(s)\in \frac{1}{\gn}\Z_{>0} \text{ and } z\to z_\gn(s) \text{ along } \mcal T.
\end{cases}
$$
\end{enumerate}
\end{rhp}

The next proposition ensures that the jumps of $\bm S$ away from $[0,\gb]$ are exponentially small as $\gn \to +\infty$.

\begin{prop}\label{prop:decayphi}
    For $\delta>0$ sufficiently small, there exists a constant $\eta>0$ such that the function $\phi$ satisfies the following properties.
    \begin{enumerate}[(i)]
        \item For $\delta\leq z\leq  \ga-\delta$, we have $\re\left(\phi_+(z)+\phi_-(z)\right)\leq -\eta$.
        \\
        
        \item For $z\geq \msf b+\delta$, we have $\re \phi(z)\geq \eta$.\\
        
        \item For $z\in \partial \mcal S_\pm^\Delta\setminus (D_\delta(\ga)\cup \R)$, we have $\re \phi(z)\leq -\eta$.\\
        
        \item For $z\in (\partial \mcal G_\pm^\nabla\cup \partial S_\pm^\nabla)\setminus (D_\delta(\ga)\cup D_\delta(\gb)\cup \R)$, we have $\re \phi(z) \leq -\eta$, as well as $\re(\phi(z)\mp \pi \ii z)\geq \eta$.       
    \end{enumerate}
\end{prop}
\begin{proof}
    With standard arguments, all these properties follow from the properties of the $g$-function as ensured by Proposition~\ref{prop:properties_g_function}, and we skip the details.
\end{proof}

Thanks to Proposition~\ref{prop:decayphi}, the jump matrix $\bm J_{\bm S}(z)$ decays to the identity matrix as $\gn\to \infty$, as long as $z$ stays within a positive distance from the interval $[0,\gb]$. We state this fact as a formal result.

\begin{prop}\label{prop:decayJSOPoutparam}
    There exists $\delta>0$ and $\eta>0$ such that
    $$
    \|\bm J_{\bm S}-\bm I\|_{L^1\cap L^\infty( \Gamma_{\bm S}\setminus ( [0,b]\cup D_\delta(\ga)\cup D_\delta(\gb)  )  )}= \Boh(\ee^{-\eta n}),\quad n\to\infty,
    $$
    uniformly for $s$ within any of the regimes from Assumptions~\ref{assumpt:parameterregimes}.
\end{prop}

\begin{proof}
    The proof follows by inspection, using Proposition~\ref{prop:decayphi} in each component of $\bm J_{\bm S}$ as defined in RHP~\ref{rhp:OPS}--(2). We skip details.
\end{proof}

As usual in the asymptotic analysis of OPs with the RHP approach, we now proceed to the construction of global and local parametrices to handle the jumps of $\bm S$ on $[0,\gb]$ and near $z_\gn(s)$ and $\gb$, which are not decaying to the identity uniformly.

\subsection{The global parametrix}\label{sec:globalPOP}\hfill

As we will see later, the jump matrix for $\bm S$ decays to the identity matrix away from the interval $[0,\gb]$. If we ignore these decaying jumps, we are lead to considering the following RHP, the {\it global parametrix RHP}.

\begin{rhp}\label{rhp:OPG} 
Find a $2\times 2$ matrix-valued function $\bm G$ with the following properties.
\begin{enumerate}[(1)]
\item $\bm G$ is analytic on $\C\setminus [0,\gb]$.
\item $\bm G_+(z)=\bm G_-(z)\bm J_{\bm G}(z)$, $z\in (0,\gb)\setminus\{\zzn\}$, where
$$
\bm J_{\bm G}(z)=
\begin{dcases}
-\bm I,& 0<z<\zzn, \\
\gsig(z)\ee^{\gE(z)}\bm E_{12}-\dfrac{1}{\gsig(z)\ee^{\gE(z)}}\bm E_{21}, & \zzn<z<\gb.
\end{dcases}
$$
\item As $z\to \infty$,
$$
\bm G(z)=\bm I+\Boh(z^{-1}).
$$
\item The entries of $\bm G$ are locally integrable or square-integrable near $z=0,\zzn,\gb$.
\end{enumerate}
\end{rhp}

The RHP above could have been stated replacing the endpoint $\zzn$ by $\ga$. However, it turns out convenient to state it this way, to ensure that certain quantities remain bounded later on. Since this global parametrix will only be used in the conclusion of the RHP analysis for values away from $\ga$, and $\zzn\to \ga$, such distinction does not play a prominent role and, as said, it is only convenient for asymptotic estimates that will be neded later.

To solve for $\bm G$ explicitly, we first transform its jumps to constant jumps, seeking for a solution of the form
$$
\bm G(z)\deff \ee^{-\msf g(\infty)\sp_3}\bm H(z)\ee^{\msf g(z)\sp_3}.
$$
This way, we enforce $\msf g$ to solve the scalar RHP:
\begin{rhp}\label{rhp:OPscalarg} 
Find a function $\msf g:\C\setminus [0,\gb]\to \C$ with the following properties.
\begin{enumerate}[(1)]
\item $\msf g$ is analytic on $\C\setminus [0,\gb]$,
\item $\msf g_+(z)+\msf g_-(z)=-\gE(z)-\log\gsig(z)$ for $\zzn<z<\gb$, \newline and \newline $\msf g_+(z)-\msf g_-(z)=-\pi \ii$ for $0<z<\zzn$.
\item $\ee^{\msf g(z)}$ has integrable singularities near $z=0$, and it is bounded near $z=\zzn, \gb$.
\item As $z\to \infty, $ $\msf g(z)=\msf g(\infty)+\Boh(z^{-1})$, for some finite constant $\msf g(\infty)\in \C$.
\end{enumerate}
\end{rhp}
and $\bm H$ should solve the following RHP:
\begin{rhp}\label{rhp:OPH} 
Find a $2\times 2$ matrix-valued function $\bm H$ with the following properties.
\begin{enumerate}[(1)]
\item $\bm H$ is analytic on $\C\setminus [\zzn,\gb]$,
\item $\bm H_+(z)=\bm H_-(z)(\bm E_{12}-\bm E_{21})$, $\zzn<z<\gb$.
\item $\bm H$ has square-integrable singularities at the endpoints $z=\zzn,\gb$,
\item $\bm H(z)=\bm I+\Boh(z^{-1})$ as $z\to\infty$.
\end{enumerate}
\end{rhp}

Observe that $\gsig$ is strictly positive over the real axis, so its real logarithm as used above is well defined.

\subsubsection{Solution of the RHP for H} The solution $\bm H$ is found in a standard way: after diagonalizing the jump matrix $\bm J_{\bm H}$, solving the RHP for $\bm H$ is essentially reduced to solving a scalar RHP, see \cite[Section~7.3]{deift_book} for details. It takes the form
\begin{equation}\label{eq:solution_H}
\bm H(z) \deff
\bm U_0^{-1}\, \msf m(z)^{\sigma_3/4}\bm U_0,\quad 
\msf m(z)\deff\frac{z-\zzn}{z-\gb},\quad
 \bm U_0\deff\frac{1}{\sqrt{2}}
 \begin{pmatrix}
 1 & \ii \\ \ii & 1
\end{pmatrix},
\end{equation}
with the principal branch of the root. In particular, this function satisfies
$$
\msf m(z)^{1/4}_+=\ii \msf m(z)^{1/4}_-, \quad \zzn<z<\gb.
$$

For later use, we also record that $\bm H$ satisfies the identities
\begin{equation}\label{eq:solution_H2}
\bm H(z) =
\bm U_0\msf m(z)^{-\sigma_3/4}\bm U_0^{-1}\quad\text{and}\quad 
\sigma_1\bm H(z)\sigma_1=\bm U_0\msf m(z)^{\sp_3/4}\bm U_0^{-1}.
\end{equation}

\subsubsection{Solution of the scalar RHP} The scalar Riemann-Hilbert problem for $\msf g$ is solved with standard methods, we briefly outline them.
First, make the ansatz that $\msf g=\msf g_1+\msf g_2$, where $\msf g_1$ is analytic on $\C\setminus [0,\zzn]$, $\msf g_2$ is analytic on $\C\setminus [\zzn,\gb]$, and both are $\Boh(1)$ as $z\to\infty$. This way, the jump conditions for $\msf g$ translate into
$$
\msf g_{1,+}(z)-\msf g_{1,-}(z)=-\pi \ii, \quad 0<z<\zzn, 
$$
and
$$
\msf g_{2,+}(z)+\msf g_{2,-}(z)=-\gE(z)-\log \gsig(z)-2\msf g_1(z), \quad \zzn<z<\gb.
$$
The RHP for $\msf g_1$ can be solved using Plemelj's formula (or by inspection) and yields
$$
\msf g_1(z)=\frac{1}{2}\log \left(\frac{z}{z-\zzn}\right),\quad z\in \C\setminus [0,\zzn],
$$
where the branch of the root is the principal one, so that $\left(\log \left(\frac{z}{z-\zzn}\right)\right)_\pm=\log\left|\frac{z}{z-\zzn}\right|\mp \pi \ii$ for $0<z<\zzn$.

Having $\msf g_1$ at hand, we now solve for $\msf g_2$. In contrast with $\msf g_1$ whose jump comes as a difference, the jump for $\msf g_2$ comes as a sum. But multiplying $\msf g_2$ by $((z-\zzn)(z-\gb))^{1/2}$ (with principal branch), we turn the jump for $\msf g_2$ from a sum into a difference, and solve the resulting RHP once again using Plemelj's formula. The final outcome is then that
\begin{multline*}
\msf g(z)=\frac{((z-\zzn)(z-\gb))^{1/2}}{2\pi}\int_{\zzn}^{\gb}\frac{\gE(x)+\log\gsig(x)+\log\left(\frac{x}{x-\zzn}\right)}{\sqrt{(\gb-x)(x-\zzn)}}\frac{\dd x}{x-z} \\ 
+\frac{1}{2}\log\left(\frac{z}{z-\zzn}\right),\quad z\in \C\setminus [0,\gb],
\end{multline*}
where $( \cdot )^{1/2}$ is the principal branch of the square root, and $\sqrt{ \cdot }$ is the standard positive root for positive real numbers. This solution satisfies
\begin{equation}\label{eq:deffginfty}
\msf g(z)=\msf g(\infty)+\Boh(z^{-1}), \,  z\to \infty, \quad \text{with}\quad 
\msf g(\infty)\deff -\frac{1}{2\pi}\int_{\zzn}^{\gb}\frac{\gE(x)+\log\gsig(x)+\log\left(\frac{x}{x-\zzn}\right)}{\sqrt{(\gb-x)(x-\zzn)}} \dd x.
\end{equation}

From the procedure outlined above, we are ensured that $\msf g$ satisfies the conditions of the RHP~\ref{rhp:OPscalarg}, except that it is not immediate whether it satisfies the condition RHP~\ref{rhp:OPscalarg}--(3). To verify that condition, we re-express $\msf g$. When $z\to 0$, it is immediate that
$$
\msf g(z)=\frac{1}{2}\log z+\Boh(1),\quad \text{and therefore}\quad \ee^{\msf g(z)}=z^{1/2}(\eta+\boh(1)),\; z\to 0,
$$
for some constant $\eta\neq 0$.

As for the remaining endpoints, a residue calculation shows that
\begin{multline*}
\int_{\zzn}^{\gb}\frac{\log\frac{x}{x-\zzn}}{\sqrt{(\gb-x)(x-\zzn)}}\frac{\dd x}{x-z} \\
\begin{aligned}
& =-\pi \res_{s=z} \frac{\log\left(\frac{s}{s-\zzn}\right)}{((s-\zzn)(s-\gb))^{1/2}(s-z)}
+\frac{1}{2\ii} \int_0^{\zzn} \frac{\left(\log\left(\frac{s}{s-\zzn}\right)\right)_+-\left(\log\left(\frac{s}{s-\zzn}\right)\right)_-}{ ((s-\zzn)(s-\gb))^{1/2}}\frac{\dd s}{s-z} \\
& =-\frac{\pi \log\left(\frac{z}{z-\zzn}\right)}{((z-\zzn)(z-\gb))^{1/2}}+\pi \int_0^{\zzn} \frac{1}{\sqrt{(\zzn-s)(\gb-s)}}\frac{\dd s}{s-z} ,\quad z\in \C\setminus [0,\gb].
\end{aligned} 
\end{multline*}
We emphasize that in the last equality we also used that $((s-\zzn)(s-\gb))^{1/2}=\sqrt{(\zzn-s)(\gb-s)}$ for $0<s<\zzn$. Using this identity, we simplify $\msf g$ to
\begin{multline}\label{eq:finalmsfgOPG}
\msf g(z)= 
\frac{((z-\zzn)(z-\gb))^{1/2}}{2\pi } \int_{\zzn}^{\gb} \frac{\gE(x)+\log\gsig(x)}{\sqrt{(\gb-x)(x-\zzn)}}\frac{\dd x}{x-z} \\ 
+ \frac{((z-\zzn)(z-\gb))^{1/2}}{2 } \int_0^{\zzn} \frac{1}{\sqrt{(\gb-x)(\zzn-x)}} \frac{\dd x}{x-z},\quad z\in \C\setminus [0,\gb].
\end{multline}
The remaining integrals on the right-hand side are Cauchy transforms of densities that blow up as square root at the endpoints $\ga,\gb$ and, therefore, the integrals too blow up as square roots. This way it follows that $\msf g$ is bounded as $z\to \ga,\gb$.

For later purposes, it is convenient to split the parts of $\msf g$ that depend on $s$ in a nontrivial way. Recall that $\gsig(x)=\gsig(x\mid s)$ depends on $s$, which gives rise to the dependence of $\msf g$ on $s$ as well. Set
$$
\msf g_\gn(z,s)\deff \frac{((z-\zzn)(z-\gb))^{1/2}}{2\pi } \int_{\zzn}^{\gb} \frac{\log\gsig(x\mid s)}{\sqrt{(\gb-x)(x-\zzn)}}\frac{\dd x}{x-z},
$$
as well as 
\begin{multline}\label{deff:g0}
\msf g_0(z)=\frac{((z-\zzn)(z-\gb))^{1/2}}{2}\\ 
\times \left(\frac{1}{\pi} \int_{\zzn}^{\gb} \frac{\gE(x)}{\sqrt{(\gb-x)(x-\zzn)}}\frac{\dd x}{x-z}
+ \int_0^{\zzn} \frac{1}{\sqrt{(\gb-x)(\zzn-x)}} \frac{\dd x}{x-z}\right).
\end{multline}
This way, the decomposition
\begin{equation}\label{eq:msfgaltrepr}
\msf g(z)=\msf g_{\gn}(z,s)+\msf g_0(z)
\end{equation}
follows, and we now provide estimates for each of these terms, as well as for $\msf g(\infty)$, as $\gn\to \infty$.

The factor $\msf g_0$ and $\msf g(\infty)$ still depend on $\gn$ through $\gE$, as well as through the endpoint $\zzn$ of the interval of integration. Likewise, $\msf g(\infty)$ depends on $\gE$ and $\zzn$, and also in the factor $\gsig$ which remains bounded in the interval of integration in \eqref{eq:deffginfty}. A simple argument then shows that
\begin{equation}\label{eq:boundg0Nlimit}
\msf g_0(z)=\Boh(1),\quad \text{as well as}\quad \msf g(\infty)=\Boh(1)\quad \text{as } \gn\to \infty,
\end{equation}
where the estimate for $\msf g_0$ is valid uniformly for $z$ within compacts of $\C$ that do not contain the points $\gb$ and $0$, and both error terms are uniform for $s$ within any of the regimes. Furthermore, $\msf g_0$ is analytic on both the upper and lower half planes, and a direct calculation shows that 
\begin{equation}\label{eq:jumpsg0}
\msf g_{0,+}(z)-\msf g_{0,-}(z)=-\pi \ii, \; 0<z<\zzn,\quad \msf g_{0,+}(z)+\msf g_{0,-}(z)=-\msf E(z), \; \zzn<z<\gb.
\end{equation}
From standard properties of Cauchy transforms we also obtain
\begin{equation}\label{eq:boundg0ata}
\msf g_0(z)=\Boh(1),\quad z\to \zzn.
\end{equation}

Next, we provide an estimate on $\msf g_\gn(z,s)$, which we encode in the next lemma.

\begin{lemma}\label{lem:basicestimatemsfgmfrC}
    For $\beta>-1$, set
    $$
    \msf L_\beta\deff \int_0^\infty v^\beta \log(1+\ee^{-v})\dd v.
    $$
    Suppose that a function $f$ is smooth on a real neighborhood of the interval $[\ga,\gb]$. Then the estimate
    \begin{equation}\label{eq:easyestimate}
    \int_{\zzn}^\gb (x-\zzn)^\beta f(x)\log \gsig(x\mid s)\, \dd x=\frac{f(\zzn)}{t^{1+\beta}\gn^{1+\beta}}\left( \msf L_\beta+\Boh(\gn^{-1}) \right)
    \end{equation}
    holds true uniformly for $s$ within any of the regimes from Assumptions~\ref{assumpt:parameterregimes}.
\end{lemma}
\begin{proof}
    Let us write the integral on the left-hand side of \eqref{eq:easyestimate} as
    $$
    f(\zzn) \int_{\zzn}^\gb (x-\zzn)^\beta\log \gsig(x\mid s) \dd x +\int_{\zzn}^\gb (x-\zzn)^\beta\log \gsig(x\mid s) (f(x)-f(\zzn)) \dd x.
    $$
    From the smoothness of $f$, we obtain that $|f(x)-f(\zzn)|\leq M (x-\zzn)$ for some $M$ independent of $s,t$, and every $x$ in the interval of integration. A standard argument {\it a là} Watson's Lemma for exponential integrals shows that
    $$
    \int_{\zzn}^\gb (x-\zzn)^\nu\log \gsig(x\mid s) \dd x=\frac{1}{t^{1+\nu}  \gn^{1+\nu}} \msf L_\nu+\Boh(\ee^{-\eta \gn t (\gb-\zzn) }),
    $$
    and the lemma follows from this estimate applied to both $\nu=\beta$ and $\nu=\beta+1$.
\end{proof}

With an application of Lemma~\ref{lem:basicestimatemsfgmfrC}, we obtain the estimate
\begin{equation}\label{eq:estmsfgN}
\msf g_\gn(z,s)=-\frac{1}{t^{1/2}\gn^{1/2}}\frac{\msf L_{-1/2}}{2\pi\sqrt{\gb-\zzn}}\left(\frac{z-\gb}{z-\zzn}\right)^{1/2}\left(1+\Boh(\gn^{-1})\right)
\end{equation}
where the error term is uniform for $z$ in compacts that do not contain $\ga$ and $\gb$, and also uniform for $s$ in any of the regimes from Assumptions~\ref{assumpt:parameterregimes}.

As a consequence, we obtain that the estimate
\begin{equation}\label{eq:estimagegGlobalfinal}
\msf g(z)=\msf g_0(z) + \Boh(\gn^{-1/2}),
\end{equation}
is valid uniformly for $z$ in compacts of $\C\setminus\{\ga,\gb,0\}$, and also uniformly for $s$ within any of the regimes from Assumptions~\ref{assumpt:parameterregimes}.

\subsubsection{Final form of the global parametrix} Summarizing, the solution to the RHP~\ref{rhp:OPG} takes the form
\begin{equation}\label{eq:globalGOPfinal}
\bm G(z)=\ee^{-\msf g(\infty)\sp_3}\bm U_0^{-1}\msf m(z)^{\sp_3/4}\bm U_0\ee^{\msf g(z)\sp_3}=\ee^{-\msf g(\infty)\sp_3}\bm H(z)\ee^{\msf g(z)\sp_3},
\end{equation}
with $\bm H$ as in \eqref{eq:solution_H}, $\msf g$ as in \eqref{eq:finalmsfgOPG}, or alternatively as in \eqref{eq:msfgaltrepr}, and $\msf g(\infty)$ as in \eqref{eq:deffginfty}. For later convenience, we book-keep the error term when moving between $\msf g$ and $\msf g_0$ in this formula. We keep denoting by $U_{\ga}$ a small neighborhood of $\ga$, sufficiently small such that it does not contain $0$ and $\msf b$ in its closure.

The following observations will be crucial for later. First off, fix a neighborhood $U_\ga$ of $\ga$ that does not contain $\gb$. From \eqref{eq:estimagegGlobalfinal},
$$
\ee^{\msf g(z)\sp_3}=\left(\bm I+\Boh\left(\gn^{-1/2}\right)\right)\ee^{\msf g_0(z)\sp_3},\quad \gn\to \infty,
$$
uniformly for $s,t$ as in Assumptions~\ref{assumpt:parameterregimes}, and also uniformly for $z$ in compacts of $U_\ga\setminus \{\ga\}$. The matrix $\bm H$ is bounded for $z$ on compacts of $U_\ga\setminus \{\ga\}$, and it is independent of $s,\gn,t$. So it can be commuted with the error above. Likewise, $\msf g(\infty)$ remains bounded as $\gn \to\infty$ (see \eqref{eq:boundg0Nlimit}), and we obtain
    \begin{equation}\label{eq:estG}
    \bm G(z)=\ee^{-\msf g(\infty)\sp_3}\bm H(z)\ee^{\msf g(z)\sp_3}=\left(\bm I+\Boh\left(\gn^{-1/2}\right)\right)\ee^{-\msf g(\infty)\sp_3}\bm H(z)\ee^{\msf g_0(z)\sp_3},\quad \gn\to \infty,
    \end{equation}
again valid uniformly for $s,t$ as in Assumptions~\ref{assumpt:parameterregimes} and uniformly for $z$ in compacts of $U_\ga\setminus \{\ga\}$.

Second, expressing $\msf g=\msf g_1+\msf g_2$, we may write 
\begin{equation}\label{eq:RHPOPGexpansion}
\bm G(z)=\ee^{-\msf g(\infty)\sp_3}\Boh(1)\ee^{\msf g_1(z)\sp_3}=\Boh(1)z^{\sp_3/2},\quad \gn\to \infty,
\end{equation}
where this expansion is valid uniformly for $z$ in closed sets of $\C\setminus \{\ga,\gb\}$, and it may be differentiated term by term.

\subsection{Construction of the local parametrix near $z=\ga$}\label{sec:localPOPa}  \hfill

The local parametrix near $z=\ga$ is the solution to a local RHP with the same jumps as $\bm S$ in a neighborhood $U_\ga$ of $z=\ga$, and matching $\bm G$ as $\gn \to \infty$ for values $z\in \partial U_\ga$. This neighborhood $U_\ga$ of $\ga$ will be made as small as needed, it is independent of $\gn$ and $s$, but it is not necessarily a disk. Since $z_\gn(s)\to \ga$ as $\gn\to \infty$, we assume without loss of generality that $\gn$ is sufficiently large so that $\wh z_\gn (s), \wc z_\gn (s)\in U_\ga$.

We assume that the neighborhood $U_\ga$ is sufficiently small so that $|z-\ga|<\varepsilon$ for every $z\in U_\ga$, with $2\varepsilon<\min\{\delta,|\ga-\gb|,|\ga|\}$. This ensures that $\Gamma_{\bm S}\cap U_\ga$ consists of subarcs of the six contours of $\Gamma_{\bm S}$ that emanate from $z_\gn(s)$, see Figure~\ref{fig:RHPLPdashed}.
\begin{figure}[t]
\centering
\includegraphics[scale=1]{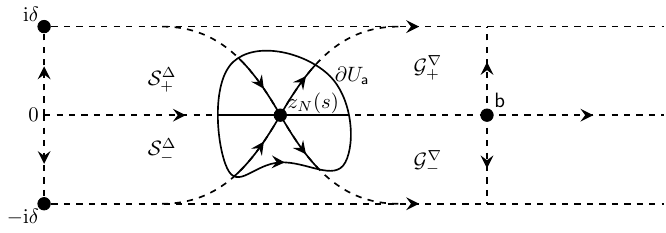}
\caption{The region $U_{\msf a}$ used for the local parametrix near $\ga$, and the jump contours for this local parametrix in solid inside $U_{\msf a}$, compare with Figure~\ref{Fig:RHPSTransf}. Although not displayed, the point $\msf a$ is inside $U_{\msf a}$, either to the left or to the right of $z_\gn(s)$, depending on whether $s<0$ or $s>0$, respectively.}
\label{fig:RHPLPdashed}
\end{figure}

The initial local parametrix problem is the following. We seek for a matrix $\bm P$ which solves the following RHP.
\begin{rhp}\label{rhp:OPP} 
Find a $2\times 2$ matrix-valued function $\bm P$ with the following properties.
\begin{enumerate}[(1)]
\item $\bm P$ is analytic on $U_{\ga}\setminus\Gamma_{\bm S}$.
\item $\bm P_+(z)=\bm P_-(z)\bm J_{\bm S}(z)$, $z\in \Gamma_{\bm S}\cap U_{\ga}$
\item For $z\in \partial U_{\ga}$, 
$$
\bm P(z)=(\bm I+\boh(1))\bm G(z),\quad \gn\to\infty.
$$
\item As $z\to z_\gn(s)$, 
$$
\bm P(z)=
\begin{cases}
    \Boh(1), & \text{if }z_\gn(s)\in \R_>\setminus \frac{1}{\gn}\Z_{>0}, \\
    \Boh(1), & \text{if }z_\gn(s)\in \frac{1}{\gn}\Z_{>0} \text{ and } z\to z_\gn(s) \text{ along } \mcal S\cup \mcal G, \\ 
    \Boh
    \begin{pmatrix}
        (z-z_{\gn}(s))^{-1} & 1 \\
        (z-z_{\gn}(s))^{-1} & 1
    \end{pmatrix},
     & \text{if }z_\gn(s)\in \frac{1}{\gn}\Z_{>0} \text{ and } z\to z_\gn(s) \text{ along } \mcal T.
\end{cases}
$$
\end{enumerate}
\end{rhp}

The first step is to remove the $\phi$-function from the jumps of $\bm P$. For that, introduce
\begin{equation}\label{deff:tildephi}
\wt\phi(z)\deff 
\begin{cases}
\pi\ii(z-1)-\phi(z), & \im z>0, \\
-\pi\ii(z-1)-\phi(z), & \im z<0.
\end{cases}
\end{equation}

\begin{prop}\label{prop:tildephiexpansions}
The function $\wt\phi$ satisfies
$$
\wt\phi_+(z)-\wt\phi_-(z)=0,\quad \ga-\varepsilon<z<\ga,\qquad \text{and}\qquad \wt\phi_+(z)+\wt\phi_-(z)=0,\quad \ga<z<\ga+\varepsilon.
$$
Furthermore, for the constant $\msf c_0>0$ from \eqref{eq:squarerootvanishing}, it satisfies
$$
\wt\phi(z)=\frac{2\msf c_0}{3} (\ga-z)^{3/2}(1+\Boh(z-\ga)),\quad z\to \ga, \; z\in U_{\ga}\setminus (\ga,\ga+\varepsilon).
$$
\end{prop}
\begin{proof}
    The proof is essentially a consequence of the relation of $\wt \phi$ with $\phi$ from \eqref{deff:tildephi} and Proposition~\ref{prop:properties_phi_function}. The exact value $\frac{2\pi \msf c_0}{3}$ of the constant upfront is obtained from \eqref{def:phi_function} and \eqref{eq:squarerootvanishing}.
\end{proof}

The jumps of $\bm S$ near $z=\ga$ are  then expressed in terms of $\wt\phi$ as
$$
\bm J_{\bm S}(z)=
\begin{dcases}
-\bm I-\dfrac{\ee^{-\gn(\wt\phi_+(x)+\wt\phi_-(x))}}{\gsig(z)\ee^{\gE(z)}}\bm E_{21},
& 0<z<\wh z_\gn(s), \\
\gsig(z)\ee^{\gE(z)}\bm E_{12}-\dfrac{\ee^{-\gE(z)}}{\gsig(z)}\bm E_{21}, 
& z> \wc z_\gn(s), \\
-\ee^{-\gn(\wt\phi_+(z)-\wt\phi_-(z))\sp_3} -\frac{1}{\gsig(z)\ee^{\gE(z)}}\bm E_{21}, 
& z\in (\wh z_\gn(s),\wc z_\gn(s)) \; ( \text{if } s<0), \\
\frac{\gsig(z)\ee^{\gE(z)}}{\ee^{-\gn(\wt\phi_+(z)+\wt\phi_-(z))}}\bm E_{12}-\frac{\ee^{-\gn(\wt\phi_+(z)+\wt\phi_-(z))}}{\gsig(z)\ee^{\gE(z)}} \bm E_{21}, 
& z\in (\wh z_\gn(s),\wc z_\gn(s)) \; ( \text{if } s>0), \\
\end{dcases}
$$
which covers the jumps on $U_{\msf a}\cap \R$, and
$$
\bm J_{\bm S}(z)=
\begin{dcases}
\bm I-\frac{\ee^{-2\gn \wt\phi(z)}}{\gsig(z)\ee^{\gE(z)}}\frac{1}{1-\ee^{\mp 2\pi \ii \gn z}}\bm E_{21}, & z\in \partial \mcal S_\pm^\Delta\cap \partial \mcal T_\pm, \\
\left(\bm I+\ee^{2\gn \wt\phi(z)}\gsig(z)\ee^{\gE(z)}\bm E_{12}\right)
\left(\bm I-\frac{\ee^{-2\gn \wt\phi(z)}}{\gsig(z)\ee^{\gE(z)}}\frac{1}{1-\ee^{-2\pi \ii \gn z}}\bm E_{21}\right), & z\in \partial \mcal S_+^\nabla\cap \partial \mcal T_+,\\
\left(\bm I-\frac{\ee^{-2\gn \wt\phi(z)}}{\gsig(z)\ee^{\gE(z)}}\frac{1}{1-\ee^{2\pi \ii \gn z}}\bm E_{21}\right)
\left(\bm I+\ee^{2\gn \wt\phi(z)}\gsig(z)\ee^{\gE(z)}\bm E_{12}\right), & z\in \partial \mcal S_-^\nabla\cap \partial \mcal T_-,
\end{dcases}
$$
which covers the remaining jumps on $U_\msf a$.

We now transform
\begin{equation}\label{deff:LPOP}
\bm L(z)\deff  \bm P(z)\ee^{(\gn \wt\phi(z)+\gE(z)/2\pm\pi\ii/2)\sp_3},\quad z\in U_{\ga}\setminus \Gamma_{\bm S},\quad \pm \im z>0.
\end{equation}
Strictly speaking, this transformation $\bm L$ is not necessary, as we could (and in fact will) construct $\bm P$ directly. However, for the sake of clarity we prefer to proceed the discussion with $\bm L$ instead, as its jump connects to the jump of the model problem in a more clear manner.
Then $\bm L$ should solve the following problem.

\begin{rhp}\label{rhp:OPL} 
Find a $2\times 2$ matrix-valued function $\bm L$ with the following properties.
\begin{enumerate}[(1)]
\item $\bm L$ is analytic on $U_{\ga}\setminus \Gamma_{\bm S}$.
\item $\bm L_+(z)=\bm L_-(z)\bm J_{\bm L}(z)$, $z\in \Gamma_{\bm S}\cap U_{\ga}$, with
$$
\bm J_\bm L(z)\deff
\begin{dcases}
\bm I-\dfrac{1}{\gsig(z)}\bm E_{21}, & z\in (0,z_\gn(s))\cap U_{\ga}, \\
\gsig(z)\bm E_{12}-\dfrac{1}{\gsig(z)}\bm E_{21}, & z\in (z_\gn(s),+\infty)\cap U_{\ga},
\end{dcases}
$$
and
$$
\bm J_{\bm L}(z)=
\begin{dcases}
\bm I+\frac{1}{\gsig(z)}\frac{1}{1-\ee^{\mp 2\pi \ii \gn z}}\bm E_{21}, & z\in \partial \mcal S_\pm^\Delta\cap \partial \mcal T_\pm, \\
\left(\bm I-\gsig(z)\bm E_{12}\right)
\left(\bm I+\frac{1}{\gsig(z)}\frac{1}{1-\ee^{-2\pi \ii \gn z}}\bm E_{21}\right), & z\in \partial \mcal S_+^\nabla\cap \partial \mcal T_+,\\
\left(\bm I+\frac{1}{\gsig(z)}\frac{1}{1-\ee^{2\pi \ii \gn z}}\bm E_{21}\right)
\left(\bm I-\gsig(z)\bm E_{12}\right), & z\in \partial \mcal S_-^\nabla\cap \partial \mcal T_-,
\end{dcases}
$$
\item For $z\in \partial U_\ga$ with $\pm \im z>0$,
$$
\bm L(z)=(\bm I+\boh(1))\bm G(z)\ee^{(\gn \wt\phi(z)+\gE(z)/2\pm\pi\ii/2)\sp_3}, \quad \gn\to\infty.
$$
\item As $z\to z_\gn(s)$, 
$$
\bm L(z)=
\begin{cases}
    \Boh(1), & \text{if }z_\gn(s)\in \R_>\setminus \frac{1}{\gn}\Z_{>0}, \\
    \Boh(1), & \text{if }z_\gn(s)\in \frac{1}{\gn}\Z_{>0} \text{ and } z\to z_\gn(s) \text{ along } \mcal S\cup \mcal G, \\ 
    \Boh
    \begin{pmatrix}
        (z-z_{\gn}(s))^{-1} & 1 \\
        (z-z_{\gn}(s))^{-1} & 1
    \end{pmatrix},
     & \text{if }z_\gn(s)\in \frac{1}{\gn}\Z_{>0} \text{ and } z\to z_\gn(s) \text{ along } \mcal T.
\end{cases}
$$
\end{enumerate}
\end{rhp}

We will now construct this matrix $\bm L$, matching it with the model RHP~\ref{rhp:modelPhi}. Several ingredients will be needed, namely
\begin{enumerate}[(i)]
    \item A change of coordinates $z\mapsto \zeta$, mapping $U_{\ga}$ to a new $\zeta$-plane, and with a ``blown-up" effect that identifies $\partial U_{\ga}$ in the $z$-plane with the point at $\infty$ in the $\zeta$-plane.
    \item A pair of admissible functions $(\msf P,\msf F)=(\msf P_0,\msf F_0)$ in the sense of Definition~\ref{deff:admissibleFP}.
    \item An identification of $z_\gn(s)$ with the central point $\lambda=\zs$ used for the fomulation of RHP~\ref{rhp:modelPhi}.
    \item We also need to make sure that $\Gamma_{\bm S}\cap U_\ga$, under the change $z\mapsto \zeta$, matches the jump contour $\Gamma^\lambda=\Gamma^\zs$.
    \item Ingredients (i)--(iv) just outlined will ensure that, under $z\mapsto \zeta$, conditions RHP~\ref{rhp:OPL}--(1),(2),(4) match with conditions RHP~\ref{rhp:modelPhi}--(1),(2),(4). To conclude RHP~\ref{rhp:OPL}--(3) matching it with RHP~\ref{rhp:modelPhi}--(3) , we will also need to construct a matrix-valued function $\bm A$ with specific properties (the so-called analytic prefactor).
\end{enumerate}
    
We now proceed to accomplish each step above. For what comes next, we identify $\gt =\gn^{1/3}$. Under this identification, the regimes from Assumptions~\ref{assumpt:parameterregimes} are equivalent to the ones from Definition~\ref{deff:subregimescritical}.

For (i), we will need a conformal map constructed in a canonical way, as given by the next proposition.

\begin{prop}\label{prop:conformalmap}
The function
\begin{equation}\label{deff:conformalmapvarphi}
\varphi(z)\deff \left(\frac{3}{2}\wt\phi(z)\right)^{2/3}
\end{equation}
is a conformal map from a neighborhood of $\ga$ to a neighborhood of the origin, with an expansion of the form
\begin{equation}\label{eq:expansionvarphi}
\varphi(z)=-\msf c_\varphi (z-\ga)(1+\Boh(z-\ga)),\quad z\to \ga, \quad \msf c_{\varphi}\deff  \left(\msf c_0\right)^{2/3}>0,
\end{equation}
and where $\msf c_0>0$ is as in \eqref{eq:squarerootvanishing}.
\end{prop}
\begin{proof}
The result is a consequence of Proposition~\ref{prop:tildephiexpansions}.
\end{proof}

\begin{remark}\label{rmk:constants}
The constant $\msf c_\varphi$ coincides with the constant $\msf c_\gV$ introduced earlier in \eqref{eq:universalconstant}. But from now on we prefer to adopt the notation $\msf c_\varphi$, to make it transparent that this constant emerges from the local expansion of the conformal map $\varphi$.    
\end{remark}

To accomplish (i), we now set
\begin{equation}\label{deff:conformalzeta}
\zeta=\gn^{2/3}\varphi(z),\qquad \text{so that}\quad \gn\wt\phi(z)=\frac{2}{3}\zeta^{3/2}, \quad \text{and} \quad \zeta>0 \text{ for } z<\ga.
\end{equation}
Thanks to Proposition~\ref{prop:conformalmap}, $z\mapsto \zeta$ constructed this way is a bona fide change of coordinates.

Next, to accomplish (ii) we set
$$
 \msf P_0(w)\deff -t(\varphi^{-1}(w)-\ga )\quad \text{and}\quad \msf F_0(w)\deff -2\pi \varphi^{-1}(w)+2\pi \ga.
$$
These functions are defined in terms of the inverse conformal map $\varphi^{-1}$, and therefore are analytic in a neighborhood of the origin and satisfy
$$
\msf P_0(w)=\frac{t}{\msf c_\varphi}w(1+\Boh(w)), \quad \msf F_0(w)=\frac{2\pi}{\msf c_\varphi}w+\Boh(w^2),\quad w\to 0.
$$
Hence, the pair $(\msf P_0,\msf F_0)$ is admissible in the sense of Definition~\ref{deff:admissibleFP}.
The functions $\msf P_\gt$ and $\msf F_\gt$ from \eqref{eq:PtauFtau} are identified with the functions $\msf P_\gn$ and $\msf F_\gn$ given explicitly by
\begin{equation}\label{deff:msfPt}
\msf P_\gn(\zeta)=\msf P_\gn(\zeta\mid s)\deff -\frac{s}{\gn^{1/3}}+\gn^{2/3}\msf P_0\left(\frac{\zeta}{\gn^{2/3}}\right),\quad \msf F_\gn(\zeta)\deff -2\pi \ga \gn^{2/3}+ \gn^{2/3}\msf F_0\left(\frac{\zeta}{\gn^{2/3}}\right),
\end{equation}
so that under the change of coordinates $z\mapsto \zeta=\zeta(z)$
\begin{equation}\label{eq:gsigPzeta}
\gsig(z\mid s)=1+\ee^{-\gn t(z-\zzn)}=1+\ee^{\gn^{1/3}\msf P_\gn(\zeta) }\qquad \text{and}\qquad \ee^{2 \pi \ii \gn z} =\ee^{-\ii \gn^{1/3}\msf F_\gn(\zeta)}.
\end{equation}
With these identifications, we have the correspondence of parameters $\gt=\gn^{1/3}$ as previously said, $s$ in this section remains the same as in the previous sections, and the values $\msf c_{\msf P},\msf c_{\msf F},\msf u$ from Definition~\ref{deff:admissibleFP} and \eqref{eq:PtauFtau} are 
\begin{equation}\label{eq:corrparameters}
\msf c_{\msf P}=\frac{t}{\msf c_\varphi},\quad \msf c_{\msf F}=\frac{2\pi}{\msf c_{\msf \varphi}} \qquad \text{and}\qquad \msf u=2\pi \ga.
\end{equation}
This concludes the completion of ingredient (ii).

Still under the correspondence $z=z(\zeta)\mapsto \zeta=\zeta(z)$, and taking into account the definition of $z_\gn(s)$ from Proposition~\ref{prop:zerossigma},
\begin{equation}\label{eq:znzetat}
\zzn=z(\zn),\qquad \text{for the value } \zn=\zs=\gn^{2/3}\varphi(\zzn) \text{ for which } \quad \msf P_\gn(\zn)=0.
\end{equation}
Hence, $\zn$ as above is identified with $\lambda=\zs$ from Proposition~\ref{prop:zeroPt} for $\msf P_\gt=\msf P_\gn$, and ingredient (iii) is completed.

Ingredient (iv) can be accomplished by simply making sure that the intersection of the boundaries $\partial \mcal G^\nabla_\pm,\partial \mcal S^\Delta_\pm $ of the lenses with $U_\ga$ are contained in the image of $\Gamma^{\lambda=\zs}$ under $\varphi$.

For ingredient (v), recall that the matrix $\bm H$ and the function $\msf g_0$ were defined in \eqref{eq:solution_H} and \eqref{deff:g0}, respectively. For $z$ in a neighborhood of $z=\ga$, we introduce
$$
\bm A_0(z)\deff \ee^{-\msf g(\infty)\sp_3}\bm H(z)\ee^{(\msf E(z)+2\msf g_0(z)\pm \pi \ii)\sp_3/2}\bm U_0^{-1}\left(\varphi(z)-\varphi(\zzn)\right)^{-\sp_3/4}, \quad \pm \im z>0,
$$
and
\begin{equation}\label{eq:analprefac}
\bm A(z)\deff \bm A_0(z) N^{-\sp_3/6},\quad \pm \im z>0.
\end{equation}

\begin{prop}\label{prop:boundA0prefactor}
The matrix-valued functions $\bm A_0$ and $\bm A$ are analytic on a neighborhood of $z=\ga$. Furthermore, 
$$
\bm A_0(z)=\Boh(1) \quad \text{and}\quad \bm A_0'(z)=\Boh(1) \quad \text{as}\quad \gn \to \infty,
$$
where the error terms are uniform for $z$ in a neighborhood of $z=\ga$.
\end{prop}
\begin{proof}
A routine calculation using the jumps for $\bm H$ and $\msf g_0$ (recall for instance \eqref{eq:jumpsg0}) shows that the jump $\bm A_{0,-}(z)^{-1}\bm A_{0,+}(z)$ along each of the intervals $(\zzn -\delta,\zzn )$ and $(\zzn,\zzn+\delta)$ is the identity matrix. This means that $\bm A_0$ is in fact analytic across each of these intervals, and as such it has an isolated singularity at $z=\zzn$. Again from the very definition of $\bm A_0$ and \eqref{eq:boundg0ata}, we then verify that this singularity is in fact removable. This shows the analyticity of $\bm A_0$, and the analyticity of $\bm A$ follows from it.

The claimed boundedness for $\bm A_0$ follows from its expression, the Maximum Principle, and the fact that all the terms involved are bounded in $\gn$ for $z\in U_\ga$. The boundedness of its derivative then follows from Cauchy's integral formula and again the Maximum Principle.
\end{proof}

Finally, with all these notations and identifications, we are ready to obtain $\bm L$ or, equivalently, $\bm P$ from \eqref{deff:LPOP}. 

\begin{theorem}
    The matrix-valued function
\begin{equation}\label{deff:modelPendpointacrit}
\bm P(z)\deff \bm A(z)\sp_1\bm\Phi_{\gt}(\gt^{2}\varphi(z))\sp_1 \ee^{-(\gn\wt\phi(z)\pm\pi\ii/2+\gE(z)/2)\sp_3},\quad \pm\im z>0, \quad z\in U_\ga,\quad \gt=\gn^{1/3},
\end{equation}
solves RHP~\ref{rhp:OPP}, with RHP~\ref{rhp:OPP}--(3) being improved to one of the following two matching conditions.
\begin{enumerate}[(i)]
    \item For $s$ in either the subcritical or critical regimes, the matching condition
\begin{equation}\label{eq:matchingPOPa}
\bm P(z)=\left(\bm I+\Boh(\gn^{-1/3})\right)\bm G(z),\quad \gn \to \infty,
\end{equation}
is valid uniformly, also uniformly for $z\in \partial U_\ga$.

\item For $s$ in the supercritical regime, the matching condition
\begin{equation}\label{eq:matchingPOPasupercr}
\bm P(z)=\left(\bm I+\Boh\left(\frac{s^2}{\gn}\right)\right)\bm G(z),\quad \gn \to \infty,
\end{equation}
is valid uniformly, also uniformly for $z\in \partial U_\ga$.
\end{enumerate}
\end{theorem}

\begin{proof}
With the correspondence \eqref{deff:msfPt}--\eqref{eq:corrparameters} in mind, and stressing that the conformal map $z\mapsto \zeta$ has a rotation effect by $\pi$ (see \eqref{eq:expansionvarphi}), the jump matrix $\bm J_{\bm L}$ writes as
$$
\bm J_{\bm L}(z)=\sigma_1\bm J_{\gt}(\zeta)^{-1}\sigma_1,
$$
where we recall that $\bm J_\gt$ is precisely the jump matrix for the model RHP~\ref{rhp:modelPhi}. Since $\bm A$ is analytic, and accounting for the relation \eqref{deff:LPOP} and the RHP~\ref{rhp:modelPhi} for $\bm \Phi_\gt$, this is enough to show that $\bm P$ determined from \eqref{deff:modelPendpointacrit} indeed solves RHP~\ref{rhp:OPP}--(1),(2),(4).

It remains to verify the matching condition \eqref{eq:matchingPOPa}, and we start with (i). The very definition of $\bm A$, the identity \eqref{deff:conformalmapvarphi} and Proposition~\ref{prop:asymptmodelproblemmatching}  when combined yield
\begin{multline*}
\bm P(z)=\ee^{-\msf g(\infty)\sp_3}\bm H(z)\ee^{(\gE(z)+2\msf g_0(z)\pm \pi \ii)\sp_3/2}\bm U_0^{-1}\left(\varphi(z)-\varphi(\zzn)\right)^{-\sp_3/4}N^{-\sp_3/6} \\
\times \left(\bm I+\Boh\left(\gn^{-2/3}\right)\right)\gn^{\sp_3/6}\varphi(z)^{\sp_3/4}\bm U_0\ee^{-(\gE(z)\pm \pi \ii)\sp_3/2}
\end{multline*}
where the error is uniform for $s$ in the sub or critical regime, and $z\in \partial U_\ga$. The terms $\bm U_0$ and $(\varphi(z)-\varphi(\zzn))$ remain bounded as $\gn\to \infty$ and $z\in \partial U_\ga$, so they commute with the error term above, simplying it to
\begin{multline*}
\bm P(z)=\ee^{-\msf g(\infty)\sp_3}\bm H(z)\ee^{(\gE(z)+2\msf g_0(z)\pm \pi \ii)\sp_3/2} \left(\bm I+\Boh\left(\gn^{-1/3}\right)\right)\\ 
\times \bm U_0^{-1}\left(1-\frac{\varphi(\zzn)}{\varphi(z)}\right)^{-\sp_3/4}\bm U_0\ee^{-(\gE(z)\pm \pi \ii)\sp_3/2}
\end{multline*}
Using that $\zzn=\ga+\Boh(s/\gn)$ and the fact that $\varphi$ is continuous and independent of $s$, we estimate $\varphi(\zzn)/\varphi(z)=\Boh(s/\gn)=\Boh(\gn^{-1/2})$, and this last estimate simplifies to
$$
\bm P(z)=\ee^{-\msf g(\infty)\sp_3}\bm H(z)\ee^{\msf g_0(z)\sp_3}\left(\bm I+\Boh(\gn^{-1/3})\right).
$$
The proof is finally completed observing that $\bm H$, $\msf g_0(z)$ and $\msf g(\infty)$ are all bounded, and using \eqref{eq:estG}.

Finally, for (ii), that is, for $s$ in the supercritical regime, we start observing that for $z\in \partial U_\ga$,
$$
|\zeta(z)|=N^{2/3}|\varphi(z)|\geq \delta N^{2/3}\geq M|\zzn|^{1/4},
$$
where the constants $M,\delta$ are uniform in $s$ within the supercritical regime, and where we also used that $|\zzn|\leq \frac{1}{s_0}\gn^{1/6}$ for $s$ in the supercritical regime, for $s_0>0$ sufficiently large. This means that we can apply Proposition~\ref{prop:asymptmodelproblemmatching}, obtaining now the estimate
\begin{align*}
\bm P(z) & =
\ee^{-\msf g(\infty)\sp_3}\bm H(z)\ee^{(\gE(z)+2\msf g_0(z)\pm \pi \ii)\sp_3/2}
\left(\bm I+\Boh\left(\frac{\zzn^2}{\gn^{1/3}}\right)\right)\ee^{-(\gE(z)\pm \pi \ii)\sp_3/2}\\
& =
\ee^{-\msf g(\infty)\sp_3}\bm H(z)\ee^{\msf g_0(z)\sp_3}
\left(\bm I+\Boh\left(\frac{\zzn^2}{\gn^{1/3}}\right)\right),\\
\end{align*}
where for the last identity we used again that $\gE$ is bounded for $z\in \partial U_\ga$. The poof is now completed using again \eqref{eq:estG} and the estimate $\zzn^2/\gn^{1/3}=\Boh(s^2/\gn)$.
\end{proof}

\subsection{Construction of the local parametrix near $z=\gb$}\label{sec:localPOPb}\hfill

As it was also the case for $z=\gb$, near $z=\ga$ the jumps of $\bm J_{\bm S}$ do not convergence uniformly to the identity and a local parametrix is needed. This jump matrix $\bm J_{\bm S}$ still involves the factor $\gsig$ deforming the original weight of orthogonality. But unlike near $z=\ga$, this factor $\gsig$ becomes nonvanishing in a fixed neighborhood of $z=\gb$, so we can remove this factor with simple considerations, and the local parametrix can then be constructed using Airy functions in a canonical way. 

We only describe the outcome.

The bare Airy parametrix $\bai_0$ is given in \eqref{eq:bai1}, and it is used to construct the solution $\bai$ to the classical Airy RHP, namely RHP~\ref{rhp:airy}. We need to deform the RHP for $\bai$, in a similar way as we have already done in \eqref{deff:whAi}. Define
\begin{equation}\label{deff:wtAi}
\wt{\bai}(\xi)\deff
\ee^{\pi \ii\sp_3/4}\bai_0(\xi)\times
\begin{cases}
\bm I, & 0< \arg \xi < \frac{\pi}{2} , \\
(\bm I-\bm E_{21}), & \frac{\pi}{2}<\arg \xi< \pi, \\ 
(\bm I-\bm E_{12})(\bm I+\bm E_{21}), & -\pi<\xi<-\frac{\pi}{2}, \\
(\bm I-\bm E_{12}), & -\frac{\pi}{2}<\xi<0. \\ 
\end{cases}
\end{equation}
Then $\wt{\bai}$ is analytic on $\C\setminus (\ii\R \cup \R)$, and its jumps are given by
$$
\bm I+\bm E_{12} \text{ on } (0,+\infty), \quad \bm I+\bm E_{21} \text{ on }\ii\R, \quad \text{and}\quad \bm E_{12}-\bm E_{21} \text{ on }(-\infty,0).
$$
In the above, we assume that $(0,\pm \ii \infty)$ are oriented from $0$ towards $\infty$, and the subintervals of $\R$ inherit the natural orientation from $\R$. Furthermore
$$
\wt\bai(\xi)=\xi^{-\sp_3/4}\bm U_0\left(\bm I+\Boh(\xi^{-3/2})\right)\ee^{-\frac{2}{3}\xi^{3/2}\sp_3},\quad \xi\to \infty.
$$

Set
$$
\wt\varphi(z)\deff \left( \frac{3}{2}\phi(z) \right)^{2/3},
$$
which is a conformal map near $z=\gb$. Recalling that $\msf g_0$ was introduced in \eqref{deff:g0} and that $\bm H$ was given in \eqref{eq:solution_H}, introduce
$$
\wt{\bm A}(z)\deff \bm G(z)\gsig(z)^{\sp_3/2}\ee^{\gE(z)\sp_3/2}\bm U_0^{-1}(\gn^{2/3}\wt\varphi(z))^{\sp_3/4}.
$$
A direct calculation, using \eqref{eq:jumpsg0} and RHP~\ref{rhp:OPG}, shows that $\wt{\bm A}$ is analytic near $z=\gb$. Finally, fix a small neighborhood $U_\gb$ of $\gb$ and set
\begin{equation}\label{deff:parmPbdOP}
{\bm P}(z)\deff \wt{\bm A}(z)\wt\bai(\xi)\gsig(z)^{-\sp_3/2}\ee^{-\gE(z)\sp_3/2}\ee^{-\gn \phi(z)\sp_3},\quad z\in U_\gb\setminus \Gamma_{\bm S},\quad \xi\deff \gn^{2/3}\wt\varphi(z).
\end{equation}
Then ${\bm P}$ is analytic on $z\in U_\gb\setminus \Gamma_{\bm S}$, and its jumps along $U_\gb\setminus \Gamma_{\bm S}$ are given precisely by $\bm J_{\bm S}$. Furthermore, it satisfies the matching condition
\begin{equation}\label{eq:matchingPGOPb}
{\bm P}(z)=\left( \bm I+\Boh(\gn^{-1}) \right)\bm G(z),\quad z\in \partial U_\gb, \quad \gn\to \infty.
\end{equation}

\subsection{Last step and small norm theory} \hfill 

In summary, the global parametrix $\bm G$ was constructed in Section~\ref{sec:globalPOP}, culminating in the expression \eqref{eq:globalGOPfinal}, the local parametrix $\bm P$ on a neighborhood $U_\ga$ of $\ga$ was constructed in Section~\ref{sec:localPOPa}, resulting in the expression \eqref{deff:modelPendpointacrit}, and the local parametrix $\bm P$ on a neighborhood $U_\gb$ of $\gb$ was constructed in Section~\ref{sec:localPOPb}, resulting in the expression \eqref{deff:parmPbdOP}. 

The last transformation $\bm S\mapsto \bm R$ involves the functions just mentioned, and it is given by
\begin{equation}\label{eq:transfSRRHPOP}
\bm R(z)\deff
\begin{cases}
\bm S(z)\bm P(z)^{-1},\quad z\in U_\ga\cup U_\gb\setminus \Gamma_{\bm S}, \\
\bm S(z)\bm G(z)^{-1},\quad z\in \C\setminus \left(\overline U_\ga \cup \overline U_\gb\cup \Gamma_{\bm S}\right). \\
\end{cases}
\end{equation}
Because the jumps of $\bm S$ and $\bm G$ coincide on $(0,\gb)$, it follows that $\bm R$ is analytic across this interval. Likewise, the jumps of $\bm S$ and $\bm P$ inside $U_\ga$ and $U_\gb$ coincide, so $\bm R$ has no jumps there as well. For these considerations and for what follows, we make sure that $\gn$ is chosen sufficiently large, in such a way that $\zzn\in U_\ga$; thanks to \eqref{eq:zerosrqmw} and Assumptions~\ref{assumpt:parameterregimes}, we can always do so uniformly in $s$, for a neighborhood $U_\ga$ of $\ga$ which is independent of $s$.

In particular, we see that $\bm R$ has no jumps near $\ga$, $\gb$ and $\zzn$. Since $\bm P$ and $\bm S$ remain bounded as $z\to \ga,\gb$, naturally $\bm R$ too remains bounded near these points. Even though $\bm P$ and $\bm S$ may blow up as $z\to \zzn$, we do know that these two functions remain bounded as $z\to \zzn$ along the real axis (see RHP~\ref{rhp:OPS} and RHP~\ref{rhp:OPP}--(4)). Therefore, $\bm R$ has a removable singularity at $z=\zzn$.

Set
$$
\Gamma_{\bm R}\deff (\Gamma_{\bm S}\cup \partial U_\ga\cup \partial U_\gb)\setminus ([0,\gb]\cup U_\ga\cup U_\gb).
$$
Summarizing, it follows that $\bm R$ solves the following RHP.
\begin{rhp}\label{rhp:OPRfinal} 
Find a $2\times 2$ matrix-valued function $\bm R$ with the following properties.
\begin{enumerate}[(1)]
\item $\bm R$ is analytic on $\C\setminus \Gamma_{\bm R}$.
\item The matrix $\bm R$ has continous boundary values $\bm R_\pm$ along $\Gamma_{\bm R}$, satisfying the relation $\bm R_+(z)=\bm R_-(z)\bm J_{\bm R}(z)$, $z\in \Gamma_{\bm R}$, with
\begin{equation}
\bm J_{\bm R}(z)\deff 
\begin{cases}
    \bm P(z)\bm G(z)^{-1}, & z\in \partial U_\ga\cup \partial U_\gb, \\ 
    \bm G(z)\bm J_{\bm S}(z)\bm G(z)^{-1}, & z\in \Gamma_{\bm R}\setminus (U_\ga\cup \partial U_\gb).
\end{cases}
\end{equation}

\item As $z\to \infty$,
$$
\bm R(z)=\bm I+\Boh(z^{-1}).
$$

\end{enumerate}
\end{rhp}

To conclude the asymptotic analysis, we now obtain estimates for $\bm J_{\bm R}$.

\begin{prop}\label{prop:decayJSOR}
    There exists $\eta>0$ such that
    $$
    \|\bm J_{\bm R}-\bm I\|_{L^1\cap L^\infty( \Gamma_{\bm R}\setminus ( \partial U_\ga\cup \partial U_\gb ))  }= \Boh(\ee^{-\eta n}),\quad n\to\infty,
    $$
    as well as
    $$
    \|\bm J_{\bm R}-\bm I\|_{L^1\cap L^\infty( \partial U_\gb  )}=\Boh(\gn^{-1}),\quad \gn\to \infty,
    $$
    all valid uniformly for $s$ within any of the regimes from Assumptions~\ref{assumpt:parameterregimes}.

    Furthermore, for $s$ in either the subcritical or critical regime, the estimate
    $$
    \|\bm J_{\bm R}-\bm I\|_{L^1\cap L^\infty( \partial U_\ga  )}=\Boh(\gn^{-1/3}),\quad \gn\to \infty,
    $$
    is valid uniformly, whereas for $s$ in the supercritical regime, the corresponding uniform estimate is
    $$
    \|\bm J_{\bm R}-\bm I\|_{L^1\cap L^\infty( \partial U_\ga  )}=\Boh\left(\frac{s^2}{\gn}\right),\quad \gn\to \infty.
    $$
\end{prop}
\begin{proof}
The estimates on $\Gamma_{\bm G}\setminus (\partial U_\ga\cup U_\gb)$ follow from Proposition~\ref{prop:decayJSOPoutparam} and the fact that $\bm G$ remains bounded as $n\to \infty$, uniformly for $z$ in closed sets of $\C\setminus [a,b]$. The estimates on $\partial U_\gb$ follow from \eqref{eq:matchingPGOPb}. The estimates on $\partial U_\ga$ follow from \eqref{eq:matchingPOPa}.
\end{proof}

As a result, we obtain.

\begin{theorem}\label{thm:ROPfinal}
As $\gn\to \infty$, the matrix $\bm R$ satisfies the estimates
$$
\|\bm R-\bm I\|_{L^\infty(\Gamma_{\bm R})}=\Boh(\gn^{-1/3}),\quad \|\bm R_\pm-\bm I\|_{L^2(\Gamma_{\bm R})}=\Boh(\gn^{-1/3}),
$$
as well as 
$$
\left\|\bm R'\right\|_{L^\infty(\Gamma_{\bm R})}=\Boh(\gn^{-1/3}),\quad \|\bm R_\pm'\|_{L^2(\Gamma_{\bm R})}=\Boh(\gn^{-1/3}),\quad \text{with}\quad '=\frac{\dd }{\dd z},
$$
uniformly for $s$ in either the subcritical or the critical regimes from Assumptions~\ref{assumpt:parameterregimes}. For $s$ in the supercritical regime, the estimates hold true with the error term being replaced by $\Boh\left(s^2/\gn\right)$.
\end{theorem}

Theorem~\ref{thm:ROPfinal} concludes the steepest descent analysis, and now we move to collecting its main consequences.


\section{The multiplicative statistics in terms of the Riemann-Hilbert Problem}\label{sec:unwrapOPRHP}

With the asymptotic analysis of the orthogonal polynomials completed, in this section we unwrap all the transformations $\bm X\mapsto \cdots \mapsto \bm R$, in order to express the quotient $\msf Z_\gn^\gsig(s)/\msf Z^\gsig_\gn(S)$ in terms of the error matrix $\bm R$ and the model problem $\bm \Phi_\gt$, amongst other quantities. Later, in the next section, we then use such expression to analyze $\msf Z_\gn^\gsig(s)/\msf Z^\gsig_\gn(S)$ asymptotically as $\gn\to \infty$.

Fix a small value $\delta>0$, and let $s,S$ be values satisfying
\begin{equation}\label{eq:boundssS}
-\delta \gn^{1/2}\leq s\leq S\leq \delta\gn^{1/2}.
\end{equation}
This way, we are ensured that in the integration in \eqref{eq:deffformula3} over a dummy variable $u\in [s,S]$, the whole interval of integration $[s,S]$ lies within the regimes appearing in Assumptions~\ref{assumpt:parameterregimes}, and the asymptotic formulas obtained from the RHP analysis of Section~\ref{sec:asymptanaldOP} can be employed. 

Later on, we will fix $S$ within any of the regimes of Assumptions~\ref{assumpt:parameterregimes}, and look at many of the quantities coming next as functions of $s$ with $S$ fixed, but with both $s$ and $S$ within the same regime from Assumptions~\ref{assumpt:parameterregimes}. For that reason, it is convenient to denote many of the quantities below as functions of $s$ only, and think of $S$ as fixed along the way. For the same reason, we refer to many of the upcoming estimates as being valid uniformly in $s$, without referring to the uniformity in $S$. We stress anyway that $S$ will always be kept as a fixed value satisfying \eqref{eq:boundssS}.

We use \eqref{eq:deffformula3} and decompose 
\begin{equation}\label{eq:deffformula4}
\log \frac{\msf Z_\gn^\gsig(S)}{\msf Z^\gsig_\gn(s )}
 =\int_{s}^{S} \sum_{x\in \frac{1}{\gn}\Z_{> 0}}\frac{\gsig(x\mid v)-1}{\gsig(x\mid v)} \gWd(z\mid v)\bm e_2^T \bm X(x\mid v)^{-1}\bm X'(x\mid v) \bm e_1 \; \dd v 
 = \msf I_1(S,s)+\msf I_2(S,s)+\msf I_3(S,s),
\end{equation}
with
\begin{equation}\label{deff:Ijintegrals}
\begin{aligned}
&\mcal X_1=\mcal X_1(s)\deff \frac{1}{\gn}\Z_{> 0}\cap (0,z_\gn(s)),\quad  \mcal X_2=\mcal X_2(s)\deff \frac{1}{\gn}\Z_{> 0}\cap [z_\gn(s),\gb),\quad \mcal X_3\deff \frac{1}{\gn}\Z_{> 0}\cap [\gb,+\infty), \\ 
& \qquad  \text{and}\\
&\msf I_j(S,s)\deff \int_{s}^{S} \sum_{x\in \mcal X_j(v)}\frac{\gsig(x\mid v)-1}{\gsig(x\mid v)} \gWd(z\mid v)\bm e_2^T \bm X(x\mid v)^{-1}\bm X'(x\mid v) \bm e_1 \; \dd v, \quad j=1,2,3.
\end{aligned}
\end{equation}
The sets $\mcal X_1=\mcal X_1(v)$ and $\mcal X_2=\mcal X_2(v)$ depend on the variable of integration but $\mcal X_3$ does not, and in the definition of $\msf I_3$ we convention that $\mcal X_3(v)=\mcal X_3=\frac{1}{\gn}\Z_{> 0}\cap [\gb,+\infty)$. By making $\delta>0$ sufficiently small, we can always assume that each set $\mcal X_j(v)$ is non-empty, for every $v\in [s,S ]$.

Although the matrix $\bm X$ has poles at the points of summation, the term $\bm e_2^T\bm X(x)^{-1}\bm X'(x)\bm e_1$ involves only entries which are entire functions. In particular, we can replace their values by the boundary value $\bm e_2^T\bm X_+(x)^{-1}\bm X_+'(x)\bm e_1$. 

We split the coming analysis of the $\msf I_j$'s into several parts, corresponding to the split of the terms in the sum in $x$ in the intervals $(0,z_\gn(s)), [z_\gn(s),\gb)$ and $[\gb,\infty)$. In each of these intervals, we unfold the transformations $\bm X\mapsto \bm Y \mapsto \bm T\mapsto \bm S$. This will yield new expressions for $\msf I_j$, decomposing them into new integrals involving $\bm S$, and other explicit factors that can be integrated in $v$ exactly, so that they give rise to sums.  After that, we split our analysis once again in two separate parts, one of them corresponding to the estimate of the $v$-integrated sums, and the remaining part that depends on $\bm S$.

In the analysis of the terms involving $\bm S$ we unwrap once more $\bm S\mapsto \bm R$, which now depends on whether we are inside or outside the local lenses where local parametrices have been constructed.

For the coming calculations, we recall the notation $\bm\Delta=\bm\Delta_z$ which is introduced in \eqref{deff:Deltaoper} and will be used extensively in what follows.

\subsection{Unwrapping $\bm X\mapsto \bm S$}\hfill 

We undo the transformations $\bm X\mapsto \bm Y \mapsto \bm T\mapsto \bm S$. We do it in each of the intervals $(0,z_\gn(s)), (z_\gn(s),\gb)$ and $(\gb,\infty)$ in separate subsections below.

\subsubsection{Contributions around the saturated interval} For $x\in (0,z_\gn(s))\setminus \mcal X_1$, we obtain
\begin{multline*}
\bm X_+(x\mid s)=\ee^{(\gn \ell-\gC)\sp_3}(-2\pi \ii \mn)^{\sp_3/2}\bm S_+(x\mid s)\left(\bm I+\frac{\ee^{2\gn \phi_+(x)-\gE(x)}}{(1-\ee^{2\pi \ii \gn x})\gsig(x\mid s)}\bm E_{21}\right) \\ \times (-2\pi \ii \gn)^{\sp_3/2} \Pi_\gn(x)^{\sp_3}\ee^{(-\gn\phi(x)+\pi\ii\gn z+\gn \gV(z)/2+\gC)\sp_3},
\end{multline*}
and therefore for $x$ in the same set
\begin{multline}\label{eq:Kernelsaturated}
\bm e_2^T\bm X_+(x\mid s)^{-1}\bm X_+'(x\mid s)\bm e_1=-2\pi \ii \gn \Pi_{\gn}(x)^2 \ee^{-2\gn \phi_+(x)+2\pi \ii \gn z+\gn \gV(x)+2\gC} \\ \times
\bm e_2^T\left[
\left(\bm I-\frac{\ee^{2\gn\phi_+(x)-\gE(x)}}{(1-\ee^{2\pi \ii \gn x})\gsig(x\mid s)}\bm E_{21}\right)\bm S_+(x\mid s)^{-1}
\right] 
\left[\bm S_+(x\mid s)
\left(\bm I+\frac{\ee^{2\gn\phi_+(x)-\gE(x)}}{(1-\ee^{2\pi \ii \gn x})\gsig(x\mid s)}\bm E_{21}\right)
\right]'
\bm e_1.
\end{multline}
The expression in the second line of \eqref{eq:Kernelsaturated} may be brought to the form
\begin{equation}\label{eq:KS01}
-\left(\frac{\ee^{2\gn\phi_+(x)-\gE(x)}}{(1-\ee^{2\pi \ii \gn x})\gsig(x\mid v)}\right)^2 \left[\bm \Delta\bm S_+(x\mid s)\right]_{12}
+\frac{2\pi \ii \gn \ee^{2\pi \ii \gn x+2\gn\phi_+(x)-\gE(x)}}{(1-\ee^{2\pi \ii \gn x})^2\gsig(x\mid v)}+\frac{1}{1-\ee^{2\pi \ii \gn x}}\times (\ast)+(\ast),
\end{equation}
where $\ast$ represent different terms involving $\gsig,\wt\phi_+,\gE$ and the entries of $\bm S_+^{-1}$ and $\bm S_+'$. In particular, these terms are all non-singular along the real axis. The factor $1/(1-\ee^{2\pi \ii \gn})$ has simple poles at points in $\frac{1}{\gn}\Z_{>0}$, while $\Pi_N^2$ has double zeros at the same points. Therefore, when we evaluate the right-hand side of \eqref{eq:Kernelsaturated} in points $x\in \frac{1}{\gn}\Z_{>0}$, the terms singled out explicitly in \eqref{eq:KS01} are the only ones that do not vanish, and after some further straightforward simplification we are left with
\begin{equation}\label{eq:XinvXprS0zn}
\bm e_2^T\bm X_+(x\mid s)^{-1}\bm X_+'(x\mid s)\bm e_1=-\frac{1}{\gWd(x\mid s)}+\frac{1}{2\pi \ii \gn}\frac{\ee^{2\gn\phi_+(x)-\gE(x)}}{\gWd(x\mid s)\gsig(x\mid s)} \left[\bm \Delta\bm S_+(x\mid s)\right]_{12},
\end{equation}
which is valid for $x\in (\frac{1}{\gn}\Z_{>0})\cap (0,z_\gn(s))$. We thus obtained the exact formulation
$$
\msf I_1(S,s)=-\int_{s}^{S} \sum_{x\in \mcal X_1(v)}\frac{\gsig(x\mid v)-1}{\gsig(x\mid v)} \; \dd v +\msf J_1(S,s),
$$
with
\begin{equation}\label{deff:integralJ1}
\msf J_1(S,s)\deff 
\frac{1}{2\pi \ii \gn}\int_{s}^{S} \sum_{x\in \mcal X_1(v)}\frac{\gsig(x\mid v)-1}{\gsig(x\mid v)^2} \; \ee^{2\gn\phi_+(x)-\gE(x)}\left[\bm \Delta\bm S_+(x\mid v)\right]_{12} \dd v.
\end{equation}
From \eqref{deff:tildephi} we see that $\ee^{2\gn \wt\phi_+(x)}=\ee^{-2\gn \wt\phi_+(x)}$ for $x\in \mcal X_j(v)$, for any $j$, so we can (and will) interchange between $\phi_+$ and $\wt \phi_+$ in the sum above and other sums that will follow.

\subsubsection{Contributions around the bulk interval} For $x\in (z_\gn(s),\gb)$ we obtain
\begin{multline*}
\bm X_+(x\mid s)=
(-2\pi \ii \gn)^{\sp_3/2}\ee^{(\gn\ell-\gC)\sp_3}\bm S_+(x\mid s)\left(\bm I+\frac{\ee^{2\gn\phi_+(x)}}{\gsig(x\mid s)\ee^{\gE(x)}}\bm E_{21}\right)\\
\times 
(-2\pi \ii \gn)^{-\sp_3/2}\ee^{(-2\gn \phi_+(x)+\gn \gV(x)+\msf C)\sp_3/2}
\left(\bm I+\frac{\gWd(x\mid v)}{\Pi_\gn(x)}\ee^{\pi \ii \gn x}\bm E_{12}\right)
\end{multline*}

Observe that the term $\Pi_\gn$ generates poles at the points $x\in \mcal X_2(s)$, so we should actually interpret this identity to be valid for $x\in (z_\gn(s),\gb)\setminus \mcal X_2(s)$.

After cumbersome but straightforward calculations, from this expression we obtain
\begin{multline}\label{eq:XinvXprSznb}
\bm e_2^T\bm X_+(x\mid s)^{-1}\bm X_+'(x\mid s)\bm e_1=
-\frac{\ee^{\gn\gV(x)+2\gC-\gE(x)}}{2\pi \ii \gn\gsig(x\mid s)}
\left(
2\gn\phi'(x)-\gE'(x)-\frac{\gsig'(x\mid s)}{\gsig(x\mid s)}
\right)
 \\
- 
\frac{\ee^{-2\gn\phi_+(x)+\gn \gV(x)+2\gC}}{2\pi \ii \gn}\left[
\left(\bm I-\frac{\ee^{2\gn\phi_+(x)-\gE(x)}}{\gsig(x\mid s)}\bm E_{21}\right)
\bm \Delta\bm S_+(x\mid s)
\left(\bm I+\frac{\ee^{2\gn\phi_+(x)-\gE(x)}}{\gsig(x\mid s)}\bm E_{21}\right)
\right]_{21}.
\end{multline}
In principle, this identity is valid for $x\in (z_\gn(s),\gb)\setminus \mcal X_2(s)$. However, in this process of taking the $(2,1)$-entry the poles coming from $(1-\ee^{-2\pi \ii \gn x})$ disappeared, and the right-hand side is regular on the whole interval $(z_\gn(s),\gb)$. Thus, by continuity this identity extends from $(z_\gn(s),\gb)\setminus \mcal X_2(s)$ to the whole interval $[z_\gn(s),\gb)$.

Using this identity back in \eqref{deff:Ijintegrals}, we obtain
\begin{equation}\label{eq:I2intstep2}
\msf I_2(S,s)=-\frac{1}{2\pi \ii \gn}\int_{s}^{S} \sum_{x\in \mcal X_2(v)}\frac{\gsig(x\mid v)-1}{\gsig(x\mid v)}\left(2\gn \phi_+'(x)-\frac{\gsig'(x\mid v)}{\gsig(x\mid v)}-\gE'(x)\right)\dd v+\msf J_2(S,s),
\end{equation}
with
\begin{multline}\label{deff:integralJ2}
\msf J_2(S,s)\deff -\frac{1}{2\pi \ii\gn} \\
\times \int_{s}^{S} \sum_{x\in \mcal X_2(v)} \frac{\gsig(x\mid v)-1}{\ee^{2\gn\phi_+(x)-\gE(x)}} \left[\left(\bm I-\frac{\ee^{2\gn\phi_+(x)-\gE(x)}}{\gsig(x\mid v)}\bm E_{21}\right)\bm \Delta \bm S_+(x\mid v)\left(\bm I+\frac{\ee^{2\gn\phi_+(x)-\gE(x)}}{\gsig(x\mid v)}\bm E_{21}\right)\right]_{21}\dd v.
\end{multline}

\subsubsection{Contributions away from the support of the equilibrium measure} We now consider the unwrap of the transformations $\bm X\mapsto \bm Y\mapsto \bm T\mapsto \bm S$ on the interval $(\gb,\infty)$. The resulting formula is
\begin{multline*}
\bm X_+(x\mid s)=\ee^{(\gn \ell-\gC)\sp_3}(-2\pi \ii \gn)^{\sp_3/2}\bm S_+(x) \\
\times \left(\bm I-\frac{\sigma(x\mid s)\ee^{-2\gn \phi_+(x)+\gE(x)}}{1-\ee^{-2\pi \ii \gn x}}\bm E_{12}\right)
\ee^{(-\gn\phi_+(x)+\gn\gV(x)/2+\gC)\sp_3}(-2\pi \ii \gn)^{-\sp_3/2}.
\end{multline*}
As before, both the left-hand and right-hand sides have poles at points in $\mcal X_3$, so this identity should be interpreted as being valid for $x\in (\gb,\infty)\setminus \mcal X_3$.

Using this identity, a cumbersome but straightforward calculation yields that
$$
\left[\bm X_+(x\mid s)^{-1}\bm X_+'(x\mid s)\right]_{21} = -\frac{\ee^{-2\gn\phi_+(x)+\gn\gV(x)+2\gC}}{2\pi \ii \gn}\left[\bm \Delta\bm S_+(x\mid s)\right]_{21}.
$$
Similarly as before, by continuity we see that this identity is valid not only in $(\gb,\infty)\setminus \mcal X_3$ but in the whole interval $[\gb,\infty)$. Returning this identity to \eqref{deff:Ijintegrals} we obtain
\begin{equation}\label{deff:J3unwrapRHPOP}
\msf I_3(S,s)=
-\frac{1}{2\pi \ii \gn}\int_{s}^{S}\sum_{x\in \mcal X_3}\left(\gsig(x\mid v)-1\right)\ee^{-2\gn\phi_+(x)+\gE(x)}\left[\bm \Delta\bm S_+(x\mid v)\right]_{21}\dd v\revdeff \msf J_3(S,s).
\end{equation}
 The introduction of the new notation $\msf J_3$ is made just for notational convenience in the next steps.

\subsubsection{Summary of the result of the unwrap of $\bm X\mapsto \bm S$} Summarizing, from \eqref{eq:deffformula4} we obtained the representation
\begin{multline}\label{eq:defformXtoS}
\log \frac{\msf Z_\gn^\gsig(S)}{\msf Z^\gsig_\gn(s )}=  \msf J_1(S,s)+\msf J_2(S,s)+\msf J_3(S,s)-\int_{s}^{S} \sum_{x\in \mcal X_1(v)}\frac{\gsig(x\mid v)-1}{\gsig(x\mid v)} \; \dd v, \\
 -\frac{1}{2\pi \ii \gn}\int_{s}^{S} \sum_{x\in \mcal X_2(v)}\frac{\gsig(x\mid v)-1}{\gsig(x\mid v)}\left(2\gn \phi_+'(x)-\frac{\gsig'(x\mid v)}{\gsig(x\mid v)}-\gE'(x)\right)\dd v,
\end{multline}
with
\begin{equation}\label{eq:summaryJk}
\begin{aligned}
 \msf J_1(S,s) & =
\frac{1}{2\pi \ii \gn}\int_{s}^{S} \sum_{x\in \mcal X_1(v)}\frac{\gsig(x\mid v)-1}{\gsig(x\mid v)^2} \; \ee^{2\gn\phi_+(x)-\gE(x)}\left[\bm \Delta\bm S_+(x\mid v)\right]_{12} \dd v.\\
\msf J_2(S,s) & = -\frac{1}{2\pi \ii\gn} 
 \int_{s}^{S} 
 \begin{multlined}[t]
 \sum_{x\in \mcal X_2(v)} \frac{\gsig(x\mid v)-1}{\ee^{2\gn\phi_+(x)-\gE(x)}} \\ 
 \times \left[\left(\bm I-\frac{\ee^{2\gn\phi_+(x)-\gE(x)}}{\gsig(x\mid v)}\bm E_{21}\right)\bm \Delta \bm S_+(x\mid v)\left(\bm I+\frac{\ee^{2\gn\phi_+(x)-\gE(x)}}{\gsig(x\mid v)}\bm E_{21}\right)\right]_{21}\dd v
 \end{multlined}
 \\
\msf J_3(S,s) & = -\frac{1}{2\pi \ii \gn}\int_{s}^{S}\sum_{x\in \mcal X_3}\left(\gsig(x\mid v)-1\right)\ee^{-2\gn\phi_+(x)+\gE(x)}\left[\bm \Delta\bm S_+(x\mid v)\right]_{21}\dd v
\end{aligned}
\end{equation}

\subsection{Unwrap of the last transformation}\hfill 

Expressions \eqref{eq:defformXtoS}--\eqref{eq:summaryJk} now involve the matrix $\bm S$, and now we unwrap the transformation $\bm S\mapsto \R$ from \eqref{eq:transfSRRHPOP}. Inside $U_\ga$ and $U_\gb$, the unwrapping yields
\begin{equation}\label{eq:unwrSRP}
\bm\Delta\bm S_+(x)=\bm\Delta\bm P_+(x)+\bm P(x)^{-1}\bm\Delta\bm R(x)\bm P(x),
\end{equation}
whereas outside these neighborhoods, the unwrapping gives
\begin{equation}\label{eq:unwrSRG}
\bm\Delta\bm S_+(x)=\bm\Delta\bm G_+(x)+\bm G(x)^{-1}\bm\Delta\bm R(x)\bm G(x).
\end{equation}
To continue, let us assume without loss of generality that $U_\ga\cap \R =(\ga-\epsilon,\ga+\epsilon)$, $U_\gn\cap \R=(\gb-\epsilon,\gb+\epsilon)$, for some $\epsilon>0$ which is independent of $\gn,s,t$. 
We have to explore formulas \eqref{eq:unwrSRP}--\eqref{eq:unwrSRG} separately, depending on whether $x$ is in each of the components  
\begin{equation}\label{eq:deffmcalXjinout}
\begin{aligned}
  &  \mcal X_1^\gout\deff \mcal X_1(s)\setminus \overline U_\ga=\frac{1}{\gn}\Z_{>0}\cap (0,\ga-\epsilon),\quad && \mcal X_1^\gin(s)\deff \mcal X_1(s)\cap  \overline U_\ga=\frac{1}{\gn}\Z_{>0}\cap [\ga-\epsilon,z_\gn(s)], \\
  &  \mcal X_2^\gin(s)\deff \mcal X_2(s)\cap U_\ga =\frac{1}{\gn}\Z_{>0}\cap (z_\gn(s),\ga+\epsilon),\quad && \mcal X_2^\gout\deff \mcal X_2(s)\setminus  U_\ga=\frac{1}{\gn}\Z_{>0}\cap [\ga+\epsilon,\gb ), \\
  &  \mcal X_3^\gin(s)\deff \mcal X_3(s)\cap U_\gb = \frac{1}{\gn}\Z_{>0}\cap [\gb,\gb+\epsilon),\quad && \mcal X_3^\gout\deff \mcal X_3(s)\setminus  U_\gb=\frac{1}{\gn}\Z_{>0}\cap [\gb+\epsilon,+\infty). \\
\end{aligned}
\end{equation}
We emphasize that the sets $\mcal X_j^\gout$ are independent of $s,t$, and hence of $v$. This is so because the neighborhoods $U_\ga$ and $U_\gb$ are independent of $s,t$. We observe also that $\mcal X_2^\gout$ involves nodes both outside and inside $U_\gb$, but such distinct nature of points will not play any major role. In line with such split, we also split the integrals in \eqref{eq:summaryJk} into
\begin{equation}\label{eq:deffJoutJin}
\msf J_k(S,s)=\msf J_k^\gin(S,s)+\msf J_k^\gout(S,s),\quad k=1,2,3,
\end{equation}
where the upper index $\gin$ corresponds to the integral whose integrand sums over $x\in \mcal X_j^\gin(v)$ and the index $\gout$ corresponds to the integral whose integrand sums over $x\in \mcal X_j^\gout(v)$.

As we will see, the major contributions will come from $\mcal X_2^\gin(s)$ and $\mcal X_2^\gin(s)$, which correspond to the node points falling inside $U_\ga$. So in the next sections, we first account for such contributions, and then express the right-hand side \eqref{eq:defformXtoS} in a different way, to single out the terms that will turn out to either be relevant or asymptotic negligible.

\subsubsection{Contributions near $\ga$} Let us focus now on rewriting $\msf J_1^\gin$ and $\msf J_2^\gin$, which correspond to the terms involving nodes in $U_\ga$, where we will use \eqref{eq:unwrSRP}. Recalling \eqref{deff:conformalzeta}, near $x=\msf a$ it is convenient to use the local coordinate
$$
\zeta=\zeta(x)=\gn^{2/3}\varphi(x).
$$
Still thanks to \eqref{deff:conformalzeta}, we recall that the map $z\mapsto \zeta$ sends the $+$-boundary value of $\R$ near $\ga$ to the $-$-boundary value of $\R$ near the origin. We emphasize that in such a case
\begin{equation}\label{eq:DeltazDeltazeta}
\zeta'=\frac{\dd \zeta}{\dd z}\qquad \text{and}\qquad \bm \Delta \bm M(\zeta)=\frac{\dd \zeta}{\dd z} \bm M(\zeta)^{-1}\frac{\dd \bm M}{\dd \zeta}(\zeta)=\zeta' \bm\Delta_\zeta \bm M(\zeta), 
\end{equation}
see for instance \eqref{eq:changeofvariablesDeltaDelta}.

From the expression \eqref{deff:modelPendpointacrit} we obtain
\begin{multline*}
\bm \Delta\bm P_+(x)=-\left(\frac{\gE'(x)}{2}+\gn \wt\phi_+'(x)\right)\sp_3
+\ee^{(\gn \wt \phi_+(x)+\pi \ii/2+\gE(x)/2)\sp_3}\sp_1 \\
\times 
\left[ 
\bm \Delta \bm \Phi_{\gt,-}(\zeta)
+\bm \Phi_{\gt,-}(\zeta)^{-1}\bm \sp_1\Delta\bm A(x)\sp_1 \bm \Phi_{\gt,-}(\zeta)
\right]  \sp_1\ee^{-(\gn \wt \phi_+(x)+\pi \ii/2+\gE(x)/2)\sp_3}
\end{multline*}

In particular, a simple calculation shows that
\begin{equation}\label{eq:DeltaPga1}
\left[\bm\Delta \bm P_+(x)\right]_{12}= -\ee^{2\gn \wt\phi_+(x)+\gE(x)}
\left[ 
\bm \Delta\bm \Phi_{\gt,-}(\zeta)
+\bm \Phi_{\gt,-}(\zeta)^{-1}\sp_1\bm \Delta\bm A(x)\sp_1 \bm \Phi_{\gt,-}(\zeta))
\right]_{21}.
\end{equation}

Exploring the relation \eqref{deff:tildephi} and the fact that $2\gn \pi \ii x\in 2\pi \Z$ for $x\in \mcal X_2$, a more cumbersome but still straightforward calculation yields
\begin{multline}\label{eq:DeltaPga2}
\left[\left(\bm I-\frac{\ee^{2\gn\phi_+(x)-\gE(x)}}{\gsig(x\mid v)}\bm E_{21}\right)\bm \Delta \bm P_+(x)\left(\bm I+\frac{\ee^{2\gn\phi_+(x)-\gE(x)}}{\gsig(x\mid v)}\bm E_{21}\right)\right]_{21}
= 
\\ 
\begin{aligned}[t]
 & \frac{2\gn\wt\phi_+'(x)+\gE'(x)}{\gsig(x\mid v)}\ee^{2\gn\phi_+(x)-\gE(x)} -  \ee^{-2\gn\wt\phi(x)-\gE(x)}
\left[\left(\bm I+\frac{\bm E_{12}}{\gsig(x\mid v)}\right)\bm\Delta \bm\Phi_{\gt,-}(\zeta)\left(\bm I-\frac{\bm E_{12}}{\gsig(x\mid v)}\right)\right]_{12}\\
& -\ee^{-2\gn\wt\phi(x)-\gE(x)}
\left[\left(\bm I+\frac{\bm E_{12}}{\gsig(x\mid v)}\right)\bm\Phi_{\gt,-}(\zeta)^{-1}\sp_1\bm\Delta\bm A(x)\sp_1\bm\Phi_{\gt,-}(\zeta)\left(\bm I-\frac{\bm E_{12}}{\gsig(x\mid v)}\right)\right]_{12}.
\end{aligned}
\end{multline}
These identities are valid for any $x\in \frac{1}{\gn}\Z_{>0 }\cap U_\ga$.

From \eqref{eq:summaryJk}, \eqref{eq:DeltaPga1} and the explicit expression for $\bm P$ from \eqref{deff:modelPendpointacrit} we thus split
$$
\msf J^{\gin}_1(S,s)=\msf S_1(S,s)+\msf K_{1,1}^\gin(S,s)+\msf K_{1,2}^\gin(S,s),
$$
with
\begin{equation}\label{eq:S1deff}
\begin{aligned}
\msf S_1(S,s) & \deff -\frac{1}{2\pi \ii \gn}\int_s^{S} \sum_{x\in \mcal X_1^\gin(v)} \frac{\gsig(x\mid v)-1}{\gsig(x\mid v)^2}\left[ \bm\Delta \bm\Phi_{\gt,-}(\zeta) \right]_{21} \dd v, \\
\msf K_{1,1}^{\gin}(S,s) & 
\deff -\frac{1}{2\pi \ii \gn}\int_s^{S} 
\sum_{x\in \mcal X_1^\gin(v)} \frac{\gsig(x\mid v)-1}{\gsig(x\mid v)^2}
\left[ \bm\Phi_{\gt,-}(\zeta)^{-1}\sp_1\bm\Delta\bm A(x)\sp_1\bm\Phi_{\gt,-}(\zeta )\right]_{21} \dd v, 
\\
\msf K_{1,2}^{\gin}(S,s) & \deff -\frac{1}{2\pi \ii \gn}\int_s^{S} \sum_{x\in \mcal X_1^\gin(v)} \frac{\gsig(x\mid v)-1}{\gsig(x\mid v)^2}\left[\bm \Phi_{\gt,-}(\zeta)^{-1}\sp_1\bm A(x)^{-1}\bm\Delta\bm R(x)\bm A(x)\sp_1\bm\Phi_{\gt,-}(\zeta) \right]_{21} \dd v,
\end{aligned}
\end{equation}
where we used that $\widetilde \phi$ is analytic on $(0,\ga)$, and therefore $\widetilde\phi_+=\widetilde\phi$ in this interval. We stress that $\bm \Phi_\gt=\bm\Phi_\gt(\cdot\mid v)$ as well as all the other matrix functions above depend on the integration variable $v$, but in the above and in what follows we omit this dependence on matrix-valued functions, to lighten notation.

Likewise, we also obtain from \eqref{eq:XinvXprS0zn}, \eqref{eq:unwrSRP} and \eqref{eq:DeltaPga1},
\begin{multline}\label{eq:XXprX1in}
\bm e_2^T \bm X(x\mid s)^{-1}\bm X'(x\mid s)\bm e_1=-\frac{1}{\gWd(x\mid s)}-\frac{1}{2\pi\ii\gn}\frac{1}{\gWd(x\mid s)\gsig(x\mid s)}\Big\{ 
\left[ \bm\Delta \bm\Phi_{\gt,-}(\zeta) \right]_{21} \\
+\left[ \bm\Phi_{\gt,-}(\zeta)^{-1}\sp_1 \left( \bm\Delta\bm A(x)+\bm A(x)^{-1}\bm\Delta\bm R(x)\bm A(x) \right)\sp_1 \bm\Phi_{\gt,-}(\zeta)\right]_{21}
\Big\},\quad x\in \mcal X_1^\gin.
\end{multline}

In a similar manner, using \eqref{deff:integralJ2} and \eqref{eq:DeltaPga2} we split
\begin{equation}\label{eq:splitJ2in}
\msf J_2^\gin (S,s)=-\frac{1}{2\pi \ii \gn}\int_s^{S}\sum_{x\in \mcal X_2^\gin(v)}\frac{\gsig(x\mid v)-1}{\gsig(x\mid v)}\left(2\gn \wt\phi'_+(x)+\gE'(x)\right)\dd v  
+\msf S_2(S,s)+ \msf K_{2,1}^\gin(S,s)+\msf K_{2,2}^\gin (S,s),
\end{equation}
where
\begin{equation}\label{eq:S2deff}
\msf S_2(S,s)\deff \frac{1}{2\pi \ii \gn}\int_s^{S} \sum_{x\in \mcal X_2^\gin(v)} \left(\gsig(x\mid v)-1\right)
\left[
\left(\bm I+\frac{\bm E_{12}}{\gsig(x\mid v)}\right)
\bm \Delta\bm \Phi_{\gt,-}(\zeta)
\left(\bm I-\frac{\bm E_{12}}{\gsig(x\mid v)}\right)
\right]_{12}\dd v
\end{equation}
and
\begin{equation}\label{eq:deffK2122in}
\begin{aligned}
&
\msf K_{2,1}^\gin(S,s) \deff 
\begin{multlined}[t]
\frac{1}{2\pi \ii \gn}\int_s^{S} 
\sum_{x\in \mcal X_2^\gin(v)} \left(\gsig(x\mid v)-1\right) \\
\times
\left[
\left(\bm I+\frac{\bm E_{12}}{\gsig(x\mid v)}\right)
\bm \Phi_{\gt,-}(\zeta)^{-1}\sp_1 \bm \Delta \bm A(x)\sp_1\bm \Phi_{\gt,-}(\zeta)
\left(\bm I-\frac{\bm E_{12}}{\gsig(x\mid v)}\right)
\right]_{12}\dd v
\end{multlined} \\
&
\msf K_{2,2}^\gin(S,s)  \deff 
\begin{multlined}[t]
\frac{1}{2\pi \ii \gn}\int_s^{S}
\sum_{x\in \mcal X_2^\gin(v)} \left(\gsig(x\mid v)-1\right) \\
\times \left[
\left(\bm I+\frac{\bm E_{12}}{\gsig(x\mid v)}\right)
\bm \Phi_{\gt,-}(\zeta)^{-1}\sp_1\bm A(x)^{-1}\bm \Delta\bm R(x) \bm A(x)\sp_1\bm\Phi_{\gt,-}(\zeta) 
\left(\bm I-\frac{\bm E_{12}}{\gsig(x\mid v)}\right)
\right]_{12}\dd v,
\end{multlined}
\end{aligned}
\end{equation}
as well as
\begin{multline}\label{eq:XXprX2in}
\bm e_2^T \bm X(x\mid s)^{-1}\bm X'(x\mid s)\bm e_1=
\frac{1}{2\pi\ii \gn}\frac{1}{\gWd(x\mid s)}\left(\frac{\gsig'(x\mid s)}{\gsig(x\mid s)}-2\pi\ii\right)
- \frac{1}{2\pi\ii\gn}\frac{1}{\gW(x\mid s)} \Bigg[
\left(\bm I+\frac{\bm E_{12}}{\gsig(x\mid s)}\right)
\\
\times 
\Big(
 \bm\Delta \bm\Phi_{\gt,-}(\zeta) +\bm\Phi_{\gt,-}(\zeta)^{-1}\sp_1 \left( \bm\Delta\bm A(x)+\bm A(x)^{-1}\bm\Delta\bm R(x)\bm A(x) \right)\sp_1 \bm\Phi_{\gt,-}(\zeta)
 \Big)
 \left(\bm I+\frac{\bm E_{12}}{\gsig(x\mid s)}\right)
\Bigg]_{12},
\end{multline}
valid for $x\in \mcal X_2^\gin$.

\subsubsection{Summary of the contributions}\label{sec:summary_RHP} The identity \eqref{eq:defformXtoS} now updates to
\begin{multline}\label{eq:defformulaalmostfinal1}
\log \frac{\msf Z_\gn^\gsig(S)}{\msf Z^\gsig_\gn(s )}=\msf S_1(S,s)+\msf S_2(S,s)\\
-\int_{s}^{S} \sum_{x\in \mcal X_1(v)}\frac{\gsig(x\mid v)-1}{\gsig(x\mid v)} \; \dd v 
-\frac{1}{2\pi \ii \gn}\int_s^{S}\sum_{x\in \mcal X_2^\gin(v)}\frac{\gsig(x\mid v)-1}{\gsig(x\mid v)}\left(2\gn \wt\phi'_+(x)+\gE'(x)\right)\dd v  \\
 -\frac{1}{2\pi \ii \gn}\int_{s}^{S} \sum_{x\in \mcal X_2(v)}\frac{\gsig(x\mid v)-1}{\gsig(x\mid v)}\left(2\gn \phi_+'(x)-\frac{\gsig'(x\mid v)}{\gsig(x\mid v)}-\gE'(x)\right)\dd v\\
 +\msf J_1^{\gout}+\msf J_2^{\gout}+\msf J_3^\gout+\msf J_3^\gin+\msf K_{1,1}^\gin+\msf K_{1,2}^\gin+\msf K_{2,1}^\gin+\msf K_{2,2}^\gin,
\end{multline}
where we stress that in the above $\msf J_k^\gout =\msf J_k^\gout(S,s)$, $\msf J_k^\gin=\msf J_k^\gin(S,s)$, $\msf K_{j,k}^\gin=\msf K_{j,k}^\gin(S,s)$.

Next, and as a last step before the asymptotic analysis of each of the terms above, we do some rewriting of some of the explicit terms, as well as of the sum $\msf S_1+\msf S_2$. 

Thanks to the relation between $\phi$ and $\wt\phi$ in \eqref{deff:tildephi}, the sums over $\mcal X_2$ and $\mcal X_2^\gin$ in \eqref{eq:defformulaalmostfinal1} combine into
\begin{multline}\label{eq:sumsimplif1}
2\pi \ii \gn\sum_{x\in \mcal X_2^\gin(v)}\frac{\gsig(x\mid v)-1}{\gsig(x\mid v)}- \sum_{x\in \mcal X_2(v)}\frac{(\gsig(x\mid v)-1)\gsig'(x\mid v)}{\gsig(x\mid v)^2}\\ + \sum_{x\in \mcal X_2^\gout}\frac{\gsig(x\mid v)-1}{\gsig(x\mid v)}\left(2\gn \phi_+'(x)-\gE'(x)\right).
\end{multline}
When plugging this expression back into \eqref{eq:defformulaalmostfinal1}, the sum above over $\mcal X_2^\gin(v)$ combines with the sum over $\mcal X_1(v)$ of \eqref{eq:defformulaalmostfinal1}, yielding a final sum over $\frac{1}{\gn}\Z_{>0}\cap (0,\ga+\epsilon)$, where $\ga+\epsilon$ is the point of intersection of $\partial U_\ga$ with $\R$ which lies to the right of $\ga$. Observe now that this set is independent of $v$. The sum runs over finitely many terms, and we interchange it with the integration, performing it explicitly. As a result we obtain
$$
\int_s^{S}\sum_{x\in \frac{1}{\gn}\Z_{>0}\cap (0,\ga+\epsilon)}\frac{\gsig(x\mid v)-1}{\gsig(x\mid v)}
\dd v=\sum_{x\in \frac{1}{\gn}\Z_{>0}\cap (0,\ga+\epsilon)}
\log \frac{\sigma(x\mid s)}{\sigma(x\mid S)}.
$$
With
\begin{equation}\label{deff:mcalXa}
\mcal X_\ga\deff \frac{1}{\gn}\Z_{>0}\cap (0,\ga+\epsilon),
\end{equation}
which is independent of $v$, our formula \eqref{eq:defformulaalmostfinal1} is now updated to
\begin{multline}\label{eq:defformulaalmostfinal}
\log \frac{\msf Z_\gn^\gsig(S)}{\msf Z^\gsig_\gn(s )}=\msf S_1(S,s)+\msf S_2(S,s) +\frac{1}{2\pi\ii \gn}\int_s^{S}\sum_{x\in \mcal X_2(v)}\frac{(\gsig(x\mid v)-1)\gsig'(x\mid v)}{\gsig(x\mid v)^2}\dd v -\sum_{x\in \mcal X_\ga} \log \frac{\sigma(x\mid s)}{\sigma(x\mid S)}\\ 
-\frac{1}{2\pi \ii\gn}\int_s^{S}\sum_{x\in \mcal X_2^\gout}\frac{\gsig(x\mid v)-1}{\gsig(x\mid v)}\left(2\gn \phi_+'(x)-\gE'(x)\right)\dd v +\msf R(S,s),
\end{multline}
where we set
\begin{multline}\label{deff:msfRunwrapRHPOP}
\msf R=\msf R(S,s)\deff \msf J_1^{\gout}(S,s)+\msf J_2^{\gout}(S,s)+\msf J_3^\gout(S,s)+\msf J_3^\gin(S,s) \\ 
+\msf K_{1,1}^\gin(S,s)+\msf K_{1,2}^\gin(S,s)+\msf K_{2,1}^\gin(S,s)+\msf K_{2,2}^\gin(S,s).
\end{multline}
To conclude this section, we still need to work out the sum $\msf S_1+\msf S_2$ a little more. For that, we recall the precise form $\gsig(x)=1+\ee^{\gt\msf P_\gt(\zeta)}$ (see \eqref{eq:gsigPzeta}), which is valid for $z\in U_\ga$, and the function $\mcal H_\gt$ introduced in \eqref{deff:bmHtau}. Using Proposition~\ref{prop:fundamentalbmHtau} and the relation in \eqref{eq:DeltazDeltazeta}, we obtain
\begin{equation}\label{eq:PhitauX1in}
\frac{\gsig(x\mid v)-1}{\gsig(x\mid v)^2}\left[ \bm\Delta \bm \Phi_{\gt,-}(\zeta) \right]_{21}=\zeta' \frac{\ee^{\gt\msf P_\gt(\zeta)}}{(1+\ee^{\gt\msf P_\gt(\zeta)})^2} \mcal H_\gt(\zeta-\zeta_\gt(v)\mid v) ,\quad x\in \mcal X_1^\gin,
\end{equation}
and
\begin{multline}\label{eq:PhitauX2in}
(\gsig(x\mid v)-1)\left[
\left(\bm I+\frac{\bm E_{12}}{\gsig(x\mid v)}\right)
\bm \Delta\bm \Phi_{\gt,-}(\zeta)
\left(\bm I-\frac{\bm E_{12}}{\gsig(x\mid v)}\right)
\right]_{12}= \\
-\zeta' \frac{\ee^{\gt\msf P_\gt(\zeta)}}{(1+\ee^{\gt\msf P_\gt(\zeta)})^2}\mcal H_\gt(\zeta-\zeta_\gt(v)\mid v) 
-\frac{(\gsig(x\mid v)-1)\sigma'(x\mid v)}{\sigma(x\mid v)^2},\quad x\in \mcal X_2^\gin.
\end{multline}

Thus, from the very definition of $\msf S_1$ and $\msf S_2$ in \eqref{eq:S1deff} and \eqref{eq:S2deff}, we obtain
\begin{equation}\label{eq:sumSS1S2}
\msf S_1(S,s)+\msf S_2(S,s)=\msf S(S,s)-\frac{1}{2\pi \ii\gn} \int_s^{S}\sum_{x\in \mcal X^\gin_2(v)} \frac{(\gsig(x\mid v)-1)\sigma'(x\mid v)}{\sigma(x\mid v)^2} \dd v,
\end{equation}
where, setting
\begin{equation}\label{deff:Xin}
\mcal X^\gin\deff \mcal X_1^\gin(v)\cup \mcal X_2^\gin(v)= \frac{1}{\gn}\Z_{>0} \cap U_\ga=\frac{1}{\gn}\Z_{>0}\cap (\ga-\epsilon,\ga+\epsilon),
\end{equation}
which is independent of $v$, the term $\msf S$ is 
\begin{equation}\label{deff:Ssum}
\msf S(S,s)\deff -\frac{1}{2\pi \ii \gn}\int_s^{S} \sum_{x\in \mcal X^\gin} \zeta' \frac{\ee^{\gt\msf P_\gt(\zeta)}}{(1+\ee^{\gt\msf P_\gt(\zeta)})^2}\mcal H_\gt(\zeta-\zeta_\gt(v)) \dd v.
\end{equation}
We emphasize that both $\msf P_\gt=\msf P_\gt(\zeta\mid v)$ and $\mcal H_\gt=\mcal H_\gt(\cdot\mid v)$ depend on the variable of integration $v$. When we plug \eqref{eq:sumSS1S2} into \eqref{eq:defformulaalmostfinal}, the explicit sum in \eqref{eq:sumSS1S2} cancels the terms on $\mcal X_2^\gin(v)$ of the first explicit sum in \eqref{eq:defformulaalmostfinal}, and we obtain the final expression
\begin{multline}\label{eq:defformreadyforasymp}
\log \frac{\msf Z_\gn^\gsig(S)}{\msf Z^\gsig_\gn(s )}=\msf S(S,s)  -\sum_{x\in \mcal X_\ga} \log \frac{\sigma(x\mid s)}{\sigma(x\mid S)}\\ 
+\msf R(S,s)-\frac{1}{2\pi \ii\gn}\int_s^{S}\sum_{x\in \mcal X_2^\gout}\frac{\gsig(x\mid v)-1}{\gsig(x\mid v)}\left(2\gn \phi_+'(x)-\gE'(x) -\frac{\gsig'(x\mid v)}{\gsig(x\mid v)} \right)\dd v.  
\end{multline}

This estimate concludes the calculations of the current section. Our next and final task is to analyze $\msf S(s)$ asymptotically as $\gn\to \infty$.



\section{Asymptotic analysis of the multiplicative statistics}\label{sec:AsymptoticMultiplicative}

In this section, we carry out the asymptotic analysis of the right-hand side of \eqref{eq:defformreadyforasymp}. Such term consists of three main parts, namely the term $\msf R$, the term $\msf S$ and the remaining explicit terms. This explicit term will turn out to consist of relevant terms near $z=\ga$ plus asymptotically negligible terms; this is proven in Section~\ref{sec:explterms}. Next, in Section~\ref{sec:analysisR} we provide estimates that will show that $\msf R$ is asymptotically negligible as well. Such estimates will depend on each of the regimes from Assumptions~\ref{assumpt:parameterregimes}, and heavily rely on the estimates provided by Proposition~\ref{prop:fundestphitau}. Finally, in Section~\ref{sec:analysisS} we perform the asymptotic analysis of $\msf S$. It is in this step that the Poisson summation formula plays a role. The analysis itself also depends on which regime of $s$ we are, and it heavily relies on Theorem~\ref{prop:fundmcalHallregimes}.

\subsection{Analysis of the explicit terms}\label{sec:explterms}\hfill 

On the right-hand side of \eqref{eq:defformreadyforasymp} there are two explicit terms, one involving a sum over the set $\mcal X_\ga$ which was defined in \eqref{deff:mcalXa}, and the remaining term which is a sum over $\mcal X_2^\gout$. As we will see later, the first term is not asymptotically negligible. But the second term is negligible, and to establish this fact we first prove an auxiliary result. 

\begin{lemma}
    For some $\eta>0$, the estimate
    \begin{equation}\label{eq:sumX2outestimate}
    \sum_{x\in \mcal X_2^\gout} \frac{\gsig(x\mid s)-1}{\gsig(x\mid s)}=\Boh(\ee^{-\eta \gn}),\quad \gn \to \infty,
    \end{equation}
    holds true uniformly for $s$ and $t$ as in Assumptions~\ref{assumpt:parameterregimes}.
\end{lemma}

\begin{proof}
    Recall that $\gsig$ was given in \eqref{deff:gsiggW}. For the proof of \eqref{eq:sumX2outestimate}, we recall that the definition of $\mcal X_2^\gout$ implies that $(x-\ga)\geq \epsilon$, where $\epsilon>0$ is such that $\ga +\epsilon$ is the point of intersection of $U_\ga$ with $(\ga,\gb)$. Likewise, the very definition of $\mcal X_2^\gout$ says that this set has at most $\gn (\gb-\ga)$ points. This way
    $$
    \sum_{x\in \mcal X_2^\gout} \frac{\gsig(x\mid s)-1}{\gsig(x\mid s)}=\sum_{x\in \mcal X_2^\gout}\frac{1}{1+\ee^{t\gn(x-\ga)+s}}\leq \gn (\gb-\ga)\ee^{-t\epsilon\gn -s}.
    $$
    For us, we always have $s=\Boh(\gn^{1/2})$, and the inequality above implies \eqref{eq:sumX2outestimate} immediately.    
\end{proof}

To estimate the last term on the right-hand side of \eqref{eq:defformreadyforasymp}, observe first that $\phi_+'$ is continuous on $(\ga,\gb)\supset \mcal X_2^\gout$ and independent of $\gn$, and therefore this function is bounded for $x\in \mcal X_2^\gout$, uniformly in $\gn$. Although $\gE$ does depend on $\gn$, Assumptions~\ref{assumpt:potential_formal} ensures $\gE$ remains bounded uniformly as $\gn\to \infty$. Finally, an exact calculation shows that
$$
\left|\frac{\gsig'(x\mid s)}{\gsig(x\mid s)}\right|=\frac{\gn t |x-\ga| }{1+\ee^{t\gn(x-\ga)+s}}\leq \gn t|x-\ga|.
$$
Thus,
$$
2\gn \phi_+'(x)-\gE'(x) -\frac{\gsig'(x\mid s)}{\gsig(x\mid s)} =\Boh(\gn),\quad \gn \to \infty,
$$
uniformly for $x\in \mcal X_2^\gout$ and uniformly in $s$, and combining this estimate with \eqref{eq:sumX2outestimate}, we update \eqref{eq:defformreadyforasymp} to
\begin{equation}\label{eq:defformreadyforasymp2}
\log \frac{\msf Z_\gn^\gsig(S)}{\msf Z^\gsig_\gn(s )}=\msf S(S,s)  -\sum_{x\in \mcal X_\ga} \log \frac{\sigma(x\mid s)}{\sigma(x\mid S)}
+\msf R(S,s)+\Boh(\ee^{-\eta \gn}),\quad \gn \to \infty,
\end{equation}
where the error term is uniform for $S,s$ as in Assumptions~\ref{assumpt:parameterregimes}. Next, we carry out the asymptotic analysis of $\msf R$.

\subsection{Analysis of $\msf R$}\label{sec:analysisR}\hfill 

Recall that $\msf R$ was defined in \eqref{deff:msfRunwrapRHPOP}, and we start estimating the terms $\msf J_k^\gout$ therein. These terms involve sums of values of $x$ in each of the intervals 
$$
(0,\ga-\epsilon]\supset \mcal X_1^\gout, \quad 
[\ga+\epsilon,\gb-\epsilon]\supset \mcal X_2^\gout , \quad [\gb +\epsilon,\infty)\supset \mcal X_3^\gout.
$$ 
where we recall that $\mcal X_k^\gin$, $\mcal X_k^\gout$ were introduced in \eqref{eq:deffmcalXjinout}.

On $(0,\ga-\epsilon]\cup [\ga+\epsilon,\gb-\epsilon]\cup [\gb+\epsilon,\infty)$, we use \eqref{eq:transfSRRHPOP} and obtain
$$
\bm \Delta \bm S_+(x)=\bm \Delta \bm G_+(x)+ \bm G_+(x)^{-1}\bm \Delta\bm R_+(x)\bm G_+(x).
$$

Using the expansion \eqref{eq:RHPOPGexpansion} and the boundedness of $\bm R$, for $x\in (0,\ga-\delta)\cup (\ga+\delta,\gb-\delta)\cup (\gb+\delta,\infty)$ we obtain
\begin{equation}\label{eq:boundDeltaSOPout}
\bm \Delta \bm S_+(x)=x^{-\sp_3/2}\left( \frac{1}{2x}\sp_3+\Boh(1) \right)x^{\sp_3/2},
\end{equation}
where the error term above is uniform for both $\gn\to \infty$ and also uniform for $x$ in the refereed intervals. Recall that $\msf J_k^\gout$ are defined by the integrals in \eqref{eq:summaryJk}, with the sums over $\mcal X_k$ being replaced by sums over $\mcal X_k^\gout$. Using that $\gE$ is bounded on compacts, and that $\phi_+$ is purely imaginary on $(\ga,\gb)\supset  \mcal X_2^\gout$ (see Proposition~\ref{prop:properties_phi_function}), from \eqref{eq:boundDeltaSOPout} we obtain
\begin{align*}
    & \msf J_1^\gout (S,s)=\frac{1}{\gn}\int_s^{S}\Boh(1)\sum_{x\in \mcal X_1^\gout}\frac{\ee^{2\gn \phi_+(x)}(\gsig(x\mid v)-1)}{\gsig(x\mid v)^2} \dd v, \\
    & \msf J_2^\gout (S,s)=-\frac{1}{\gn}\int_s^{S}\Boh(1)\sum_{x\in \mcal X_2^\gout}(\gsig(x\mid v)-1) \dd v, \\ 
    & \msf J_3^\gout (S,s)=-\frac{1}{\gn}\int_s^{S}\Boh(1)\sum_{x\in \mcal X_3^\gout}\ee^{-2\gn \phi_+(x)+\gE(x)}(\gsig(x\mid v)-1) \dd v, 
\end{align*}
where the error term is uniform in all the parameters. Along $\mcal X_1^\gout$, the function $2\re \phi_+$ is strictly negative (see Proposition~\ref{prop:properties_phi_function} and Proposition~\ref{prop:properties_g_function}--(d)), say $2\re\phi_+(x)\leq -\eta$. On the other hand, using the basic inequality $\ee^{u}/(1+\ee^{u})^2\leq 1$, valid for every real value $u$, we see that the factor involving $\sigma$ in $\msf J_1^\gout$ above is bounded by $1$. Thus, we just concluded that
$$
\msf J_1^\gout(S,s)=\Boh(\ee^{-\gn \eta}),\quad \gn\to \infty,
$$
uniformly for $S,s$ within any of the regimes from Assumptions~\ref{assumpt:parameterregimes}. Next, for $\msf J_2^\gout$, we use the explicit expression for $\gsig$ to see that
$0\leq \gsig(x\mid v)-1\leq \ee^{-t\gn \delta -s }$ for every $x\in \mcal X_2^\gout\subset (\ga,\gb)$, and therefore once again
$$
\msf J_2^\gout(S,s)=\Boh(\ee^{-\gn \eta}),\quad \gn\to \infty,
$$
again uniformly for $S,s$. Finally, for $\msf J_3$, we use that $\phi(x)\sim \gV(x)$ as $x\to \infty$, and Assumptions~\ref{assumpt:potential_formal} to ensure that $\sum_{x\in \mcal X_3^\gout} \ee^{-2\gn\phi(x)+\msf E(x)}$ is bounded, and then use again the explicit expression for $\gsig(x\mid v)-1$ and obtain once more that
$$
\msf J_3^\gout(S,s)=\Boh(\ee^{-\gn \eta}),\quad \gn\to \infty,
$$
uniformly for $S,s$ within the regimes from Assumptions~\ref{assumpt:parameterregimes}. Recall that $\msf R$ was defined in \eqref{deff:msfRunwrapRHPOP}. This estimate concludes the analysis of the terms in $\msf R$ involving the outer sets $\mcal X_k^\gout$, and we now move to the terms involving the inner sets.

To cope with $\msf J_3^\gout$ we unwrap $\bm S\mapsto \bm R$ on $(\gb-\epsilon,\gb+\epsilon)\supset \mcal X_3^\gout$. Therein, we use \eqref{eq:transfSRRHPOP} and \eqref{deff:parmPbdOP} and obtain
\begin{multline*}
\bm \Delta \bm S_+(x)=-\left(\frac{1}{2} \frac{\gsig'(x)}{\gsig(x)}+\frac{\gE'(x)}{2}+\gn\phi_+(x)\right)\sp_3 + \gsig(x)^{\sp_3/2}\ee^{(\gE(x)/2+\gn\phi_+(x))\sp_3}\\ 
\times \left[ \bm\Delta \wt{\bm A}(x)+\wt{\bm A}(x)^{-1}\left((\bm\Delta_z\wt\bai)(\xi)+ \wt\bai(\xi)^{-1}\bm\Delta \bm R(x)\wt\bai(\xi) \right)\wt{\bm A}(x)  \right]
\gsig(x)^{-\sp_3/2}\ee^{-(\gE(x)/2+\gn\phi_+(x))\sp_3},
\end{multline*}
where $\xi=\gn\varphi(z)$ is as in \eqref{deff:parmPbdOP}, and we wrote $\bm \Delta=\bm\Delta_z$ to emphasize that it involves a derivative with respect to $z=x$ instead of with respect to $\xi$. Therefore
\begin{multline*}
\msf J_3^\gin(S,s)=-\frac{1}{2\pi \ii \gn}\int_{s}^{S}\sum_{x\in \mcal X_3^\gin}\frac{\gsig(x\mid v)-1}{\gsig(x\mid v)}\ee^{-2\gn \phi_+(x)-\gE(x)} \\ 
\times \left[ \bm\Delta \wt{\bm A}(x)+\wt{\bm A}(x)^{-1}\left((\bm\Delta_z\wt\bai)(\xi)+ \wt\bai(\xi)^{-1}\bm\Delta \bm R(x)\wt\bai(\xi) \right)\wt{\bm A}(x)  \right]_{21,+}\dd v.
\end{multline*}

Using that $\gsig$ as well as its derivatives remain bounded as $\gn \to \infty$ uniformly for $z$ near $\gb$, one can see that both $\wt(\bm A)^{\pm 1}$ and $\wt(\bm A)$ grows at most polynomially in $\gn$ as $\gn\to \infty$. The same considerations also hold for the terms involving $\wt{\bai}$ above. All in all, we see that the $(2,1)$-entry inside the integral grows at most polynomially in $\gn$. On the other hand, exploring the decay of the remaining terms in a similar way as  did for $\mcal X_k^\gout$ previously, we obtain
$$
\msf J_3^\gin(S,s)=\Boh(\ee^{-\eta \gn}),\quad \gn\to \infty,
$$
again uniformly in $S,s$. Recall that $\msf R$ is defined in \eqref{deff:msfRunwrapRHPOP}. This last estimate with the ones obtained previously for $\msf J_k^\gin$ show that
\begin{equation}\label{eq:errorRfinalpresplit}
\msf R(S,s)=\msf K_{1,1}^\gin(S,s)+\msf K_{1,2}^\gin(S,s)+\msf K_{2,1}^\gin(S,s)+\msf K_{2,2}^\gin(S,s)+\Boh(\ee^{-\eta \gn}),\quad \gn\to \infty.
\end{equation}
We now estimate the remaining terms, which we recall were introduced in \eqref{eq:S1deff} and \eqref{eq:deffK2122in}.

Thanks to the relation \eqref{eq:analprefac}, Proposition~\ref{prop:boundA0prefactor} and Theorem~\ref{thm:ROPfinal}, the expressions for $\msf K_{j,k}^\gin$'s in \eqref{eq:S1deff} and \eqref{eq:deffK2122in} may be updated to
\begin{equation}\label{eq:msfKjkupdate1}
\begin{aligned}
& \msf K_{1,1}^\gin=\frac{1}{\gn^{2/3}}\int_{s}^S \sum_{x\in \mcal X_1^\gin(v)} \frac{\gsig(x\mid v)-1}{\gsig(x\mid v)^2}\left[ \bm\Phi_{\gt,-}(\zeta)^{-1}
\Boh(1)\bm\Phi_{\gt,-}(\zeta) \right]_{21}\dd v, \\ 
& \msf K_{1,2}^\gin=\frac{1}{\gn^{2/3}}\int_{s}^S \sum_{x\in \mcal X_1^\gin(v)} \frac{\gsig(x\mid v)-1}{\gsig(x\mid v)^2}\left[ \bm\Phi_{\gt,-}(\zeta)^{-1}
\Boh(1)\bm\Phi_{\gt,-}(\zeta) \right]_{21}\dd v,
\end{aligned}
\end{equation}
and
\begin{equation}\label{eq:msfKjkupdate2}
\begin{aligned}
& 
\msf K_{2,1}^\gin=\frac{1}{\gn^{2/3}}\int_{s}^S \sum_{x\in \mcal X_2^\gin(v)} (\gsig(x\mid v)-1) \left[\left(\bm I+\frac{\bm E_{12}}{\gsig(x\mid v)}\right) \bm\Phi_{\gt,-}(\zeta)^{-1}
\Boh(1)\bm\Phi_{\gt,-}(\zeta)\left(\bm I-\frac{\bm E_{12}}{\gsig(x\mid v)}\right) \right]_{12}\dd v, 
\\ 
&
\msf K_{2,2}^\gin=\frac{1}{\gn^{2/3}}\int_{s}^S \sum_{x\in \mcal X_2^\gin(v)} (\gsig(x\mid v)-1)\left[\left(\bm I+\frac{\bm E_{12}}{\gsig(x\mid v)}\right) \bm\Phi_{\gt,-}(\zeta)^{-1}
\Boh(1)\bm\Phi_{\gt,-}(\zeta)\left(\bm I-\frac{\bm E_{12}}{\gsig(x\mid v)}\right) \right]_{12}\dd v.
\end{aligned}
\end{equation}
In fact, the estimates for $\msf K_{j,2}^\gin$ are not optimal: the term $\Boh(1)$ could be replaced by $\Boh(s^2/N)$ in the supercritical regime or by $\Boh(N^{-1/3})$ in the other regimes, but it suffices to consider them as above.

These expressions now involve only $\bm \Phi_\gt$, and with the help of the estimates already provided for this term we now indicate why these terms are in fact decaying; such analysis will depend on whether we are in the subcritical, critical or supercritical regimes of $S,s$.

For the calculations in the coming subsections, we recall that the change of variables 
$$ 
x\mapsto \zeta=\gt^2 \varphi(x), \quad \gt=\gn^{1/3},
$$ 
was defined in \eqref{deff:conformalzeta}, and it is such that 
\begin{equation}\label{eq:confmappropfinalunwrap}
\zeta(\zzn)=\gn^{2/3}\varphi(\zzn)=\zn, \quad \text{and}\quad \gsig(x\mid s)-1=\ee^{-t \gn  (x-\zzn)}=\ee^{\gn^{1/3}\msf P_\gn(\zeta\mid s)},
\end{equation}
see \eqref{eq:gsigPzeta}. Recall that $\mcal X^\gin$ was defined in \eqref{deff:Xin}, and introduce
\begin{equation}\label{deff:xpmendpoints}
x^-\deff \min \mcal X^\gin,\quad x^+\deff \max \mcal X^\gin,\quad \kappa \deff -\varphi(x^\pm)>0.
\end{equation}
With an adjustment of $x^\pm$ if necessary, we may without lost of generality assume that $\varphi(x^-)=\kappa=-\varphi(x^+)$. The conformal map sends points in $\mcal X_1^\gin$ to $[\zn,\kappa\gn^{2/3}]$, and points in $\mcal X_2^\gin$ to $[-\kappa\gn^{2/3},\zn]$. We use the identifications above extensively in what comes next.

\subsubsection{Rough estimates on $\msf R$ in the subcritical case}\label{sec:Rsubcritical} We start with the estimates in the subcritical case, in which $\zn<\ga$. In this case, in particular, $0<\zeta(x)<\zn$ if, and only if, $x\in (z_\gn(s),\ga)\subset\mcal X_2^\gin$.

With an application of Proposition~\ref{prop:fundestphitau}--(ii) to \eqref{eq:msfKjkupdate1} we are ensured of the estimate
\begin{equation}
\msf K_{1,1}^\gin+\msf K_{1,2}^\gin=
\Boh(\gn^{-2/3})\int_s^S \sum_{x\in \mcal X_1^\gin(v)}\frac{\sigma(x\mid v)-1}{\sigma(x\mid v)^2}(1+|\zeta|^{1/2})\ee^{-\frac{4}{3}\zeta^{3/2}} \; \dd v.
\end{equation}
Having in mind that $0<\zeta(x)<\zn$ if, and only if, $x\in (\zzn,\ga)\subset\mcal X_2^\gin$, we use Proposition~\ref{prop:fundestphitau}--(ii) again, now in \eqref{eq:msfKjkupdate2}, and write in a similar way
\begin{multline*}
\msf K_{2,1}^\gin+\msf K_{2,2}^\gin=\Boh(\gn^{-2/3})\int_s^S \sum_{\substack{x\in \mcal X_2^\gin(v)\cap (z_\gn(v),\ga) }}\frac{\sigma(x\mid v)-1}{\gsig(x\mid v)^2}\ee^{-\frac{4}{3}\zeta^{3/2}}(1+|\zeta|^{1/2}) \; \dd v \\ + 
\Boh(\gn^{-2/3})\int_s^S \sum_{\substack{x\in \mcal X_2^\gin(v)\cap [\ga,\ga+\epsilon) }}(\sigma(x\mid v)-1)(1+|\zeta|^{1/2}) \; \dd v 
\end{multline*}
To obtain this identity, we emphasize that we used the algebraic identity \eqref{eq:coolmatrixidentity} with the choice $\beta=\gsig=1+\ee^{\gn^{1/3}\msf P_\gn}$ for the estimate on $\mcal X_2^\gin(v)\cap (z_\gn(s),\ga)$, and we used that $\zeta(x)_-^{3/2}\in \ii \R$ and that $1\leq \gsig(x\mid v)\leq 2$ for $x\in \mcal X_2^\gin(v)\setminus (z_\gn(s),\ga)$.

Observe that $\msf P_\gn(\zeta)=-|\msf P_\gn(\zeta)|$ for $\zeta<\zn$, that is, for $x\in \mcal X_2^\gin$, see Proposition~\ref{prop:zeroPt}. Also, using the simple inequality $\ee^{-u}/(1+\ee^{-u})^2\leq \ee^{-|u|}$ for $u\in \R$, we bound $(\gsig-1)/\gsig^2\leq \ee^{-\gn^{1/3}|\msf P_\gn|}=\ee^{-\gn^{1/3}\msf P_\gn}$ for $\zeta>\zs$, that is, for $x\in \mcal X_1^\gin$. Hence, the latter two estimates on $\msf K_{j,k}$ can be combined to
\begin{equation}\label{eq:msfKijsubcritical}
\msf K_{1,1}^\gin+\msf K_{1,2}^\gin+\msf K_{2,1}^\gin+\msf K_{2,2}^\gin = \Boh(\gn^{-2/3})
\int_{s}^S\sum_{x\in \mcal X^\gin}(1+|\zeta|^{1/2})\ee^{-\gn^{1/3} |\msf P_\gn(\zeta)|-\frac{4}{3}\re \zeta^{3/2}_-}\, \dd v,
\end{equation}
recalling again that $\mcal X^\gin$ is as in \eqref{deff:Xin}.

To estimate these sums, let us introduce 
\begin{equation}\label{deff:mcalXjpcKbounds}
\begin{aligned}
& {\mcal X}_{1}^\pc ={\mcal X}_{k}^\pc(s)\deff \left\{x \in \mcal X^\gin\mid \zeta(x)<\frac{\zn}{3}  \right\}, \\ 
& {\mcal X}_{2}^\pc ={\mcal X}_{2}^\pc(s) \deff \left\{x \in \mcal X^\gin\mid  \frac{\zn}{3}\leq \zeta(x)\leq 3\zn  \right\}, \\
& {\mcal X}_{3}^\pc={\mcal X}_{3}^\pc(s) \deff \left\{x \in \mcal X^\gin\mid \zeta(x)>3\zn\right\}.
\end{aligned}
\end{equation}

Clearly these sets are pairwise distinct and 
$$
\mcal X^\gin={\mcal X}_{1}^\pc\cup {\mcal X}_{2}^\pc\cup {\mcal X}_{3}^\pc.
$$
We estimate the sum in \eqref{eq:msfKijsubcritical} by estimating it in each of the sets $\mcal X_k^\pc$.

We start bounding the sum over $\mcal X_1^\pc$. Over there, we simply bound $\re \zeta_-^{3/2}\geq 0$ to simplify
\begin{equation}\label{eq:sumX1pc1}
\sum_{x\in \mcal X_1^\pc(v)}(1+|\zeta|^{1/2})\ee^{-\gn^{1/3} |\msf P_\gn(\zeta)|-\frac{4}{3}\re \zeta^{3/2}_-}\leq \sum_{x\in \mcal X_1^\pc(v)}(1+|\zeta|^{1/2})\ee^{-\gn^{1/3} |\msf P_\gn(\zeta)|}.
\end{equation}
Points in $\mcal X_1^\pc(v)$ may be parametrized as $\{x_k\}_{k\geq 1}$ with
$$
x_k\leq \frac{z_\gn(v)}{2}-\frac{k}{\gn},\quad \text{and therefore}\quad \gn^{1/3}|\msf P_\gn(\zeta(x_k))|=\gn t|x_k-z_\gn(v)|\geq t\gn \frac{z_\gn(v)}{2}+ t k.
$$
Also, clearly $|\zeta|^{1/2}=\gn^{1/3}|\varphi(x)|\leq C\gn^{1/3}$, for some constant $C>0$ independent of $v,\gn$, simply because $\varphi$ is conformal. We update \eqref{eq:sumX1pc1} to
$$
\sum_{x\in \mcal X_1^\pc(v)}(1+|\zeta|^{1/2})\ee^{-\gn^{1/3} |\msf P_\gn(\zeta)|-\frac{4}{3}\re \zeta^{3/2}_-} \leq C \gn^{1/3}\ee^{-\eta\gn  z_\gn(v)}\sum_{k\geq 1} \ee^{-t k} \leq C'\gn^{1/3}\ee^{-\eta\gn z_\gn(v) }
$$
for some positive constants $C,C',\eta$, all independent of $\gn,v$. 

Similar arguments also apply to the sum over $\mcal X_3^\pc$, and we conclude that for some $\eta>0$,
\begin{equation}\label{eq:sumX1pcfinal1}
\frac{1}{\gn^{2/3}}\sum_{x\in \mcal X_1^\pc(v)\cup \mcal X_3^\pc(v)}(1+|\zeta|^{1/2})\ee^{-\gn^{1/3} |\msf P_\gn(\zeta)|-\frac{4}{3}\re \zeta^{3/2}_-}=\Boh\left( \frac{1}{\gn^{1/3}}\ee^{-\eta\gn z_\gn(v)} \right),
\end{equation}
uniformly for $v$ in the subcritical regime.

Next, we estimate the sum over $\mcal X_2^\pc$. Therein, we always have $\re \zeta^{3/2}_-=\zeta^{3/2}>0$, and we use \eqref{eq:ineqPrezetazs} to bound $\ee^{-|\gn^{1/3}\msf P_\gn|}\leq C \ee^{-\eta \gn^{1/3}|\zeta-\zeta_\gn(v)|}$, for some constants $\eta,C>0$. All in all, and using in addition that $|\zeta|\geq \frac{1}{3}\zeta_\gn(v)$ for $x\in \mcal X_2^\pc$, we obtain
\begin{equation}\label{eq:sumX1pc2}
\sum_{x\in \mcal X_2^\pc(v)}(1+|\zeta|^{1/2})\ee^{-\gn^{1/3} |\msf P_\gn(\zeta)|-\frac{4}{3}\re \zeta^{3/2}_-}\leq C \zeta_\gn(v)^{1/2} \sum_{x\in \mcal X_2^\pc(v)}\ee^{-2\eta \gn^{1/3} |\zeta-\zeta_\gn(v)|-\frac{4}{3}\zeta^{3/2}}
\end{equation}
for yet new constants $C,\eta>0$; the constant $2$ in front of $\eta$ is not an accident, and its presence is explored in what follows. We now seek for the minimum of $\zeta\mapsto \frac{4}{3}\zeta^{3/2}+\eta \gn^{1/3}|\zeta-\zeta_\gn(v)|$ for $x\in \mcal X_2^\pc(v)$. A direct calculation shows that its critical points are solutions to $\frac{2}{3}\zeta^{1/2}=\pm \eta \gn^{1/3}$, where the sign $\pm$ depends on whether $\zeta>\zeta_\gn(v)$ or $\zeta<\zeta_{\gn}(v)$. Such critical points are $\Boh(\gn^{1/3})$, and because we are in the subcritical regime none of them belongs to the interval $[\zeta_\gn(v)/3,3\zeta_\gn(v)]\supset \mcal X_2^\pc$. Hence, 
$$
\frac{4}{3}\zeta^{3/2}+\eta \gn^{1/3}|\zeta-\zeta_\gn(v)|\geq \min_{\xi=\frac{\zeta_\gn(v)}{3},\zeta_\gn(v),3\zeta_\gn(v)}\left(\frac{4}{3}\xi^{3/2}+\eta \gn^{1/3}|\xi-\zeta_\gn(v)|\right),\quad \text{for every }\zeta\in \mcal X_2^\pc.
$$
Once again because we are in the subcritical regime, we have that $\gn^{1/3}\zeta_\gn(v)\gg \zeta_\gn(v)^{3/2}$, and the minimum above is attained precisely at $\xi=\zeta_\gn(v)$. Returning this information into \eqref{eq:sumX1pc2}, we further estimate
\begin{equation}\label{eq:sumX1pc3}
\sum_{x\in \mcal X_2^\pc(v)}(1+|\zeta|^{1/2})\ee^{-\gn^{1/3} |\msf P_\gn(\zeta)|-\frac{4}{3}\re \zeta^{3/2}_-}\leq C \zeta_\gn(v)^{1/2} \ee^{-\frac{4}{3}\zeta_\gn(v)^{3/2}} \sum_{x\in \mcal X_2^\pc(v)}\ee^{-\eta \gn^{1/3} |\zeta-\zeta_\gn(v)|}.
\end{equation}
Writing $\zeta-\zeta_\gn(v)=\gn^{2/3}(\varphi(x)-\varphi(z_\gn(v)))$ and using that $\varphi$ is a conformal map independent of $\gn$, we in fact may bound $|\zeta-\zeta_\gn(v)|\geq c \gn^{2/3}|x-z_\gn(v)|$. Using then that the points $x$'s belong to $\frac{1}{\gn}\Z_{>0}$, we obtain that the remaining sum on the right-hand side above is finite. Hence, we just concluded that
\begin{equation}\label{eq:sumX1pcfinal2}
\frac{1}{\gn^{2/3}}\sum_{x\in \mcal X_2^\pc(v)}(1+|\zeta|^{1/2})\ee^{-\gn^{1/3} |\msf P_\gn(\zeta)|-\frac{4}{3}\re \zeta^{3/2}_-} =\Boh\left( \frac{\zeta_\gn(v)^{1/2}}{\gn^{2/3}} \ee^{-\frac{4}{3}\zeta_\gn(v)^{3/2}}\right).
\end{equation}

We now plug \eqref{eq:sumX1pcfinal1} and \eqref{eq:sumX1pcfinal2} into \eqref{eq:msfKijsubcritical}. Using once again that $\gn z_\gn(v)=\gn^{1/3}\zeta_\gn(v)\gg \zeta_\gn(v)^{3/2}$ for $v$ in the subcritical regime, we concluded the next result.
\begin{prop}\label{prop:estKijsubcr}
    The estimate
$$
    \msf K_{1,1}^\gin(S,s)+\msf K_{1,2}^\gin(S,s)+\msf K_{2,1}^\gin(S,s)+\msf K_{2,2}^\gin(S,s)=\Boh(\gn^{-2/3}) \int_{s}^S \zeta_\gn(v)^{1/2}\ee^{-\frac{4}{3}\zeta_\gn(v)^{3/2}}\dd v,\quad \gn\to \infty,
    $$
    is valid uniformly for $s\leq S$ and both within the subcritical regime.
    is valid uniformly for $s\leq S$ and both within the subcritical regime.
\end{prop}

\subsubsection{Rough estimates on $\msf R$ in the critical case} \label{sec:Rcritical} As before, we first summarize the contributions of $\msf K_{j,k}^\gin$ into simpler estimates. 

Recall that $\mcal X_1^\gin(v)$ and $\mcal X_2^\gin(v)$ were introduced in \eqref{eq:deffmcalXjinout}. Using Proposition~\ref{prop:fundestphitau}--(i) in the representation \eqref{eq:msfKjkupdate1}, we obtain
\begin{equation}\label{eq:Kijinrcaux}
\msf K_{1,1}^\gin(S,s)+\msf K_{1,2}^\gin(S,s)=\Boh(\gn^{-2/3})\int_s^S\sum_{x\in \mcal X_1^\gin(v)}\frac{\gsig(x\mid v)-1}{\gsig(x\mid v)^2}(1+|\zeta|^{1/2})\ee^{-\frac{4}{3}\zeta_-^{3/2}-2\zeta_\gt(v) \zeta_-^{1/2}} \dd v,
\end{equation}
where we used that $\zeta=\zeta(x)\geq \zeta_\gn(v)$ for $x\in \mcal X_1^\gin(v)$. 

Next, using \eqref{eq:coolmatrixidentity} and again Proposition~\ref{prop:fundestphitau}--(i), for $\zeta_\gn(v)-r<\zeta<\zeta_\gn(v)$ we estimate
\begin{equation}\label{eq:Kijboundrcaux3}
\left[\left(\bm I+\frac{\bm E_{12}}{\gsig(x\mid v)}\right)\bm\Phi_{\gt,-}(\zeta)^{-1}\Boh(1)\bm\Phi_{\gt,-}(\zeta)\left(\bm I-\frac{\bm E_{12}}{\gsig(x\mid v)}\right)\right]_{12}
= \Boh(1) \frac{(1+|\zeta|^{1/2})\ee^{-\frac{4}{3}\zeta_-^{3/2}-2\zeta_\gn(v)\zeta_-^{1/2}}}{\gsig(x\mid v)^2}.
\end{equation}
The exponent $\frac{2}{3}\zeta_-^{3/2}+\zeta_\gn(v)\zeta_-^{1/2}$ is purely imaginary on the negative axis, and it is bounded on bounded sets of the positive real axis. Hence, because we are assuming $\zeta<\zeta_\gn(v)$ and the latter is bounded, in this last estimate the exponential terms are in fact bounded, and we just concluded that
\begin{equation}\label{eq:Kijboundrcaux2}
\left[\left(\bm I+\frac{\bm E_{12}}{\gsig(x\mid v)}\right)\bm\Phi_{\gt,-}(\zeta)^{-1}\Boh(1)\bm\Phi_{\gt,-}(\zeta)\left(\bm I-\frac{\bm E_{12}}{\gsig(x\mid v)}\right)\right]_{12}
= \Boh(1) \frac{(1+|\zeta|^{1/2})}{\gsig(x\mid v)^2},
\end{equation}
valid for $\zeta_\gn(v)-r<\zeta<\zeta_\gn(v)$.

We now extend this estimate to $\zeta<\zeta_\gn(v)-r$ as well. First we write, again using \eqref{eq:coolmatrixidentity}
\begin{multline}\label{eq:Kijboundrcaux1}
\left[\left(\bm I+\frac{\bm E_{12}}{\gsig(x\mid v)}\right)\bm\Phi_{\gt,-}(\zeta)^{-1}\Boh(1)\bm\Phi_{\gt,-}(\zeta)\left(\bm I-\frac{\bm E_{12}}{\gsig(x\mid v)}\right)\right]_{12}
 \\ 
 =-\frac{1}{\gsig(x\mid v)^2}\left[\left(\bm I+\gsig(x\mid v)\bm E_{21} \right)\bm\Phi_{\gt,-}(\zeta)^{-1}\Boh(1)\bm\Phi_{\gt,-}(\zeta)\left(\bm I-\gsig(x\mid v)\bm E_{21}\right)\right]_{21}.
\end{multline}
Thanks to Proposition~\ref{prop:canonicestPF}--(ii), the factor $\gsig(x\mid v)=1+\ee^{\gt\msf P_\gt(\zeta)}$ is bounded for $\zeta=\zeta(x)\in (-\infty,\zeta_\gn(v))=\Gamma_4^{\zeta_\gn(v)}$. Plugging the asymptotic formula from Proposition~\ref{prop:fundestphitau}--(i) into \eqref{eq:Kijboundrcaux1}, we conclude that
\begin{multline*}
\left[\left(\bm I+\frac{\bm E_{12}}{\gsig(x\mid v)}\right)\bm\Phi_{\gt,-}(\zeta)^{-1}\Boh(1)\bm\Phi_{\gt,-}(\zeta)\left(\bm I-\frac{\bm E_{12}}{\gsig(x\mid v)}\right)\right]_{12}\\ 
=
\frac{1}{\gsig(x\mid v)^2}\left[\Boh(1)\ee^{(\frac{2}{3}\zeta_-^{3/2}+\zeta_\gn(v)\zeta_-^{1/2})\sp_3}\Boh(1+|\zeta|^{1/2})\ee^{-(\frac{2}{3}\zeta_-^{3/2}+\zeta_\gn(v)\zeta_-^{1/2})\sp_3}\Boh(1)\right]_{21},
\end{multline*}
valid for $\zeta<\zeta_\gn(v)-r$. Proceeding as we did to obtain \eqref{eq:Kijboundrcaux2} from \eqref{eq:Kijboundrcaux3}, we see that the exponential terms above are bounded, in turn extending \eqref{eq:Kijboundrcaux2} to the whole interval $\zeta<\zeta_\gn(v)$. Plugging the result asymptotic formula into \eqref{eq:msfKjkupdate2}, we obtain the estimate
$$
\msf K_{2,1}^\gin(S,s)+\msf K_{2,2}^\gin(S,s)=\Boh(\gn^{-2/3})
\int_{s}^S \sum_{x\in \mcal X_2^\gin(v)} \frac{\gsig(x\mid v)-1}{\gsig(x\mid v)^2}(1+|\zeta|^{1/2}) \dd v.
$$
We combine this estimate with \eqref{eq:Kijinrcaux}, summarizing them as
\begin{multline*}
\msf K_{1,1}^\gin(S,s)+\msf K_{1,2}^\gin(S,s)+\msf K_{2,1}^\gin(S,s)+\msf K_{2,2}^\gin(S,s)  =
\Boh(\gn^{-2/3}) \\
\times\int_{s}^S\left[ \sum_{x\in \mcal X_1^\gin(v)}\frac{\gsig(x\mid v)-1}{\gsig(x\mid v)^2}(1+|\zeta|^{1/2})\ee^{-\frac{4}{3}\zeta_-^{3/2}-2\zeta_\gn(v)\zeta_-^{1/2}}+\sum_{x\in \mcal X_2^\gin(v)}\frac{\gsig(x\mid v)-1}{\gsig(x\mid v)^2}(1+|\zeta|^{1/2}) \right]\dd v.
\end{multline*}

Using the simple inequality $\ee^{-u}/(1+\ee^{-u})^2\leq \ee^{-|u|}$ for $u\in \R$, we bound $(\gsig-1)/\gsig^2\leq \ee^{-\gn^{1/3}|\msf P_\gn|}$ for $x\in \mcal X_2^\gin(v)$, which corresponds to $\zeta<\zeta_\gn(v)$. Recalling \eqref{eq:ineqPrezetazs}, we obtain
\begin{equation}\label{eq:sumX1pc3critical}
\sum_{x\in \mcal X_2^\gin(v)}\frac{\gsig(x\mid v)-1}{\gsig(x\mid v)^2}(1+|\zeta|^{1/2}) \leq  \sum_{x\in \mcal X_2^\gin(v)}(1+|\zeta|^{1/2})\ee^{-\eta \gn^{1/3} |\zeta-\zeta_\gn(v)|}.
\end{equation}
Writing $\zeta-\zeta_\gn(v)=\gn^{2/3}(\varphi(x)-\varphi(z_\gn(v)))$ and using that $\varphi$ is a conformal map independent of $\gn$, we in fact may bound $|\zeta-\zeta_\gn(v)|\geq c \gn^{2/3}|x-z_\gn(v)|$. Since $\mcal X_2^\gin(v)\subset \frac{1}{\gn}\Z_{>0}$, we obtain that the remaining sum on the right-hand side above is finite, and in fact
\begin{equation}\label{eq:sumX1pcfinal2critical}
\frac{1}{\gn^{2/3}}\sum_{x\in \mcal X_2^\gin(v)}\frac{\gsig(x\mid v)-1}{\gsig(x\mid v)^2}(1+|\zeta|^{1/2}) =\Boh\left( \frac{|\zeta_\gn(v)|^{1/2}}{\gn^{2/3}} \right).
\end{equation}

Next, using again that $(\gsig-1)/\gsig^2\leq \ee^{-\gn^{1/3}|\msf P_\gn|}$ for $\zeta<\zs$ (that is, for $x\in \mcal X_2^\gin(v)$) and the estimate \eqref{eq:ineqPrezetazs} once again, we bound
\begin{equation}\label{eq:sumX1pcfinal2critical_further}
\sum_{x\in \mcal X_1^\gin(v)}\frac{\gsig(x\mid v)-1}{\gsig(x\mid v)^2}(1+|\zeta|^{1/2})\ee^{-\frac{4}{3}\zeta_-^{3/2}-2\zeta_\gn(v)\zeta_-^{1/2}} \leq  \sum_{x\in \mcal X_1^\gin(v)}(1+|\zeta|^{1/2})\ee^{-\eta \gn^{1/3} |\zeta-\zeta_\gn(v)|-\frac{4}{3}\zeta_-^{3/2}-2\zeta_\gn(v)\zeta_-^{1/2}}.
\end{equation}
We stress once again that $\zeta>\zeta_\gn(v)$ for $x\in \mcal X_1^\gin(v)$. If $\zeta_\gn(v)>0$, then $\frac{4}{3}\zeta_-^{3/2}+2\zeta_\gn(v)\zeta_-^{1/2}>0$ for every $x\in \mcal X_1^\gin(v)$. For $\zeta_\gn(v)<0$, we have $\frac{4}{3}\zeta_-^{3/2}+2\zeta_\gn(v)\zeta_-^{1/2}\in \ii \R$ for $\zeta_\gn(v)\leq \zeta<0$. Furthermore, still for $\zeta_\gn(v)<0$, the function $[0,+\infty)\ni \zeta\mapsto \frac{4}{3}\zeta^{3/2}+2\zeta_\gn(v)\zeta^{1/2}$ has a unique local minimum at the point $\zeta_\gn(v)/2$, where it attains the value $-2^{3/2}|\zeta_\gn(v)|^{3/2}/3$. Since we are in the critical regime, we obtain in particular that $\frac{4}{3}\zeta^{3/2}+2\zeta_\gn(v)\zeta^{1/2}$ remains bounded from below for $\zeta>0$. With these information in mind, we update \eqref{eq:sumX1pcfinal2critical_further} to the estimate
\begin{equation}\label{eq:sumX1pcfinal2critical_further2}
\sum_{x\in \mcal X_1^\gin(v)}\frac{\gsig(x\mid v)-1}{\gsig(x\mid v)^2}(1+|\zeta|^{1/2})\ee^{-\frac{4}{3}\zeta_-^{3/2}-2\zeta_\gn(v)\zeta_-^{1/2}} \leq  C \sum_{x\in \mcal X_1^\gin(v)}(1+|\zeta|^{1/2})\ee^{-\eta \gn^{1/3} |\zeta-\zeta_\gn(v)|}.
\end{equation}
Proceeding as we did for \eqref{eq:sumX1pcfinal2critical}, we see that the remaining sum above is $\Boh(|\zeta_\gn(v)|^{1/2})$. 

In summary, we concluded
\begin{prop}\label{prop:estKijcrit}
    The estimate
    $$
    \msf K_{1,1}^\gin(S,s)+\msf K_{1,2}^\gin(S,s)+\msf K_{2,1}^\gin(S,s)+\msf K_{2,2}^\gin(S,s)  =\Boh(\gn^{-2/3}) \int^S_s |\zeta_\gn(v)|^{1/2}\dd v,\quad \gn\to \infty,
    $$
    is valid uniformly for $s\leq S$ and both within the critical regime.
\end{prop}

\subsubsection{Rough estimates on $\msf R$ in the supercritical case} \label{sec:Rsupercritical} In addition to denoting $\zeta=\zeta(x)$, in what follows it is also convenient to denote 
$$
\xi=\xi(x)=\xi(x\mid s)=\zn^2(\zeta(x)-\zn).
$$

Finally, we now rewrite the error terms $\msf K_{j,k}^\gin$ in the supercritical case. For that, we will use once again the identities \eqref{eq:msfKjkupdate1}--\eqref{eq:msfKjkupdate2}, combined with the asymptotics from Proposition~\ref{prop:fundestphitau}. In particular, Proposition~\ref{prop:fundestphitau}--(iii) provides asymptotics for $\bm \Phi_{\gt,-}$ which are different depending on certain neighborhoods of $\zs=\zeta(\zzn)$. Based on these neighborhoods we introduce the pairwise disjoint sets
\begin{equation}\label{deff:Xkjnc}
\begin{aligned}
& {\mcal X}_{k,1}^\nc ={\mcal X}_{k,1}^\nc(s)\deff \left\{x \in \mcal X_k^\gin\mid |\zeta(x)-\zn|\leq \frac{\delta}{\zn^2}\right\}=\left\{x \in \mcal X_k^\gin\mid |\xi(x)|\leq \delta \right\}, \\ 
& {\mcal X}_{k,2}^\nc ={\mcal X}_{k,2}^\nc(s) \deff \left\{x \in \mcal X_k^\gin\mid  \frac{\delta}{\zn^2}< |\zeta(x)-\zn|\leq \delta\right\}=\left\{x \in \mcal X_k^\gin\mid  \delta < |\xi(x)|\leq \delta \zn^2 \right\}, \\
& {\mcal X}_{k,3}^\nc={\mcal X}_{k,3}^\nc(s) \deff \left\{x \in \mcal X_k^\gin\mid |\zeta(x)-\zn|> \delta\right\}=\left\{x \in \mcal X_k^\gin\mid |\xi(x)|> \delta \zn^2  \right\}.
\end{aligned}
\end{equation}
The value $\delta>0$ may be chosen arbitrarily small but fixed.

Clearly
$$
\mcal X_k^\gin={\mcal X}_{k,1}^\nc(s)\cup {\mcal X}_{k,2}^\nc(s)\cup {\mcal X}_{k,3}^\nc(s),\quad \text{for every } s.
$$
Furthermore, $\xi\geq 0$ for $x\in \mcal X_{1,j}^\nc$, whereas $\xi\leq 0$ for $x\in \mcal X_{2,j}^\nc$, $j=1,2,3$.

We simplify the expressions in Proposition~\ref{prop:fundestphitau}--(iii) to
\begin{multline*}
\bm \Phi_{\gt,-}(\zeta)= \\ 
\begin{cases}
    \Boh(|\zn|^{5/4})\bm B_-(\xi)(\bm I+\msf B_{\bm\Upsilon,-}(\xi)\bm E_{12}), & x\in \mcal X_{1,1}^\nc(s), \\
    \Boh(|\zn|^{5/4})\bm B_-(\xi)(\bm I-\bm E_{21})(\bm I+\msf B_{\bm\Upsilon,-}(\xi)\bm E_{12})(\bm I-(1+\ee^{\gn^{1/3}\msf P_\gn(\zeta)})\bm E_{21}), & x\in \mcal X_{2,1}^\nc(s), \\
    \Boh(|\zn|^{5/4})\bm B_-(\xi), & x\in \mcal X_{1,2}^\nc(s)\cup \mcal X_{2,2}^\nc(s), \\
    \Boh(\zn(\zeta-\zn)^{1/4})\ee^{-\left(\frac{2}{3}(\zeta-\zn)_-^{3/2}+\zn (\zeta-\zn)^{1/2}_-\right)\sp_3}, & x\in \mcal X_{1,3}^\nc(s)\cup \mcal X_{2,3}^\nc(s).
\end{cases}
\end{multline*}
In obtaining such expressions, we used the explicit expression \eqref{deff:PLUpsi} for $\bm L_{\bm \Upsilon}$, as well as the relation \eqref{deff:wtcalPF} between $\wt{\mcal P}_\gt$ and $\msf P_\gt=\msf P_\gn$. We also recall that $\bm B$ and $\msf B_{\bm \Upsilon}$ are given in \eqref{deff:BesselB} and \eqref{eq:BUpsilon}, respectively.

When $|\xi|$ is large, the behavior of $\bm B$ is governed by \eqref{eq:asymptBessel}, when $|\xi|$ gets close to $0$ the behavior of $\bm B$ is given in \eqref{eq:behBesselParOrigin}, and for $|\xi|$ in compacts of $(0,+\infty)$ the matrix $\bm B$ is bounded. With these observations in mind, we update the above estimates to

\begin{multline*}
\bm \Phi_{\gt,-}(\zeta)= \\ 
\begin{cases}
    \Boh(|\zn|^{5/4})(\bm I+\left(\frac{1}{2\pi\ii}\log \xi +\msf B_{\bm\Upsilon,-}(\xi)\right)\bm E_{12}), & x\in \mcal X_{1,1}^\nc(s), \\
    \Boh(|\zn|^{5/4})(\bm I+\left(\frac{1}{2\pi\ii}\log \xi +\msf B_{\bm\Upsilon,-}(\xi)\right)\bm E_{12})(\bm I-(1+\ee^{\gn^{1/3}\msf P_\gn(\zeta)})\bm E_{21}), & x\in \mcal X_{2,1}^\nc(s), \\
    \Boh(|\zn|^{5/4}(1+|\xi|^{1/4}))\ee^{\xi^{1/2}_-\sp_3}, & x\in \mcal X_{1,2}^\nc(s)\cup \mcal X_{2,2}^\nc(s), \\
    \Boh(\zn(\zeta-\zn)^{1/4})\ee^{-\left(\frac{2}{3}(\zeta-\zn)_-^{3/2}+\zn (\zeta-\zn)^{1/2}_-\right)\sp_3}, & x\in \mcal X_{1,3}^\nc(s)\cup \mcal X_{2,3}^\nc(s).
\end{cases}
\end{multline*}

We now use these asymptotic formulas to estimate the factors $\msf K_{k,j}^\gin$. First off, using these formulas along $\mcal X_{1,k}^\nc$ and the already familiar estimate $\ee^{u}/(1+\ee^{u})^2\leq \ee^{-|u|}$, we update \eqref{eq:msfKjkupdate1} to
\begin{multline}\label{eq:KKijspcr1}
\msf K_{1,1}^\gin(S,s)+\msf K_{1,2}^\gin(S,s)=\Boh(\gn^{-2/3})\int_{s}^S 
\left(
\sum_{x\in \mcal X_{1,1}^\nc(v)\cup \mcal X_{1,2}^\nc(v)}|\zeta_\gn(v)|^{5/4}(1+|\xi|^{1/4})\ee^{-\gn t |x-z_\gn(v)|+2\xi^{1/2}} \right.\\
\left.
+ \sum_{x\in \mcal X_{1,3}^\nc(v)}|\zeta_\gn(v)|(\zeta-\zeta_\gn(v))^{1/4}\ee^{-\gn t |x-z_\gn(v)|-\frac{4}{3}(\zeta-\zeta_\gn(v))^{3/2}-2\zeta_\gn(v)(\zeta-\zeta_\gn(v))^{1/2}} 
\right)\dd v.
\end{multline}

In the above and in what follows, we emphasize that $\xi=\xi(x\mid v)$ depends both in $x$ and also on the variable of integration $v$, but we omit this dependence.

Likewise, we now apply the asymptotic formulas to \eqref{eq:msfKjkupdate2}. In addition to plugging them in, we also use that $(\zeta-\zeta_\gn(v))_-^{1/2}, (\zeta-\zeta_\gn(v))_-^{3/2}\in \ii\R$ for $x\in \mcal X_{2,2}^\nc(v)$, $\xi_-^{1/2}\in \ii\R$ for $x\in \mcal X_{2,3}^\nc(v)$, and also identity \eqref{eq:coolmatrixidentity} to simplify the expression on $\mcal X_{2,1}^\nc(v)$. The resulting estimate is
\begin{multline}\label{eq:KKijspcr2}
\msf K_{2,1}^\gin(S,s)+\msf K_{2,2}^\gin(S,s)=\Boh(\gn^{-2/3})\int_{s}^S 
\left(
\sum_{x\in \mcal X_{2,1}^\nc(v)\cup \mcal X_{2,2}^\nc(v)}|\zeta_\gn(v)|^{5/4}(1+|\xi|^{1/4})\ee^{-\gn t |x-z_\gn(v)|} \right.\\
\left.
+ \sum_{x\in \mcal X_{2,3}^\nc(v)}|\zeta_\gn(v)||\zeta-\zeta_\gn(v)|^{1/4}\ee^{-\gn t |x-z_\gn(v)|} 
\right)\dd v.
\end{multline}

We now bound each of the sums appearing in \eqref{eq:KKijspcr1}--\eqref{eq:KKijspcr2}.

For $x\in \mcal X_{1,1}^\nc(v)\cup \mcal X_{2,1}^\nc(v)$, notice that $|\xi(x)|\leq \delta$ for some positive universal constant $c>0$, and therefore
\begin{multline}\label{eq:boundXijspcr1}
\sum_{x\in \mcal X_{1,1}^\nc(v)} |\zeta_\gn(v)|^{5/4} (1+|\xi|^{1/4})\ee^{-\gn t |x-z_\gn(v)|+2\xi^{1/2}}+\sum_{x\in \mcal X_{2,1}^\nc(v)} |\zeta_\gn(v)|^{5/4} (1+|\xi|^{1/4})\ee^{-\gn t |x-z_\gn(v)|} \\ \leq \tilde C |\zeta_\gn(v)|^{5/4} \sum_{x\in \mcal X_{1,1}^\nc(v)\cup \mcal X_{2,1}^\nc(v)}  \ee^{-\gn t|x-z_\gn(v)| } \leq C |\zeta_\gn(v)|^{5/4}
\end{multline}
for some new constants $C, \tilde C>0$, which are uniform for $s$ within the supercritical regime. In the last step, we simply used that the remaining series is finite because $\mcal X_{j,k}^\nc(v)\subset \frac{1}{\gn} \Z_{>0}$.

Next, we analyze the sums over $\chi_{k,2}^\nc(v)$. For that, we use the explicit form \eqref{deff:conformalzeta} of the conformal map $x\mapsto \zeta$ to bound $\gn t |x-z_\gn(v)|\geq 2 c N^{1/3} |\zeta-\zeta_\gn(v)|$ for some constant $c>0$ independent of $s$, so that
$$
\gn t |x-z_\gn(v)|-2\xi^{1/2}\geq \frac{t}{2} \gn |x-z_\gn(v)|+\gn^{1/3} c|\zeta-\zeta_\gn(v)|-2\xi^{1/2}= \frac{t}{2} \gn |x-z_\gn(v)|+\frac{c \gn^{1/3}}{\zeta_\gn(v)^2} \xi-2\xi^{1/2}.
$$
The map $(0,\infty)\ni \xi \mapsto c \gn^{1/3} \xi^2/ \zeta_\gn(v)^2 \xi-2\xi^{1/2}$ has a unique critical point at $(2\zeta_\gn(v)/c \gn^{1/3})^2$. Since we are in the supercritical regime, when $\zeta_\gn(s)=\Boh(s/\gn^{1/3})$, we can choose $s_0>0$ sufficiently large in such a way that this critical point is smaller than $\delta/2$. Consequently, the function $\xi \mapsto c \gn^{1/3} \xi/ \zeta_\gn(v)^2-2\xi^{1/2}$ is increasing on $(\delta,+\infty)\supset \mcal X_{2,1}^\nc(v)$, attaining its smallest value precisely at $\xi=\delta$, where its value is bounded bounded from below by a constant $c>0$ independent of $\gn$. We just concluded that
\begin{equation}\label{eq:mcalX12spcr}
\gn t |x-z_\gn(v)|-2\xi^{1/2}\geq \frac{t}{2} \gn |x-z_\gn(v)|+c
\end{equation}
for every $x\in \mcal X_{2,1}^\nc(v)$, and therefore
\begin{multline}\label{eq:boundXijspcr2}
\sum_{x\in \mcal X_{1,2}^\nc(v)} |\zeta_\gn(v)|^{5/4} (1+|\xi|^{1/4})\ee^{-\gn t |x-z_\gn(v)|+2\xi^{1/2}}+\sum_{x\in \mcal X_{2,2}^\nc(v)} |\zeta_\gn(v)|^{5/4} (1+|\xi|^{1/4})\ee^{-\gn t |x-z_\gn(v)|} \\ \leq \tilde C |\zeta_\gn(v)|^{5/4} \sum_{x\in \mcal X_{1,2}^\nc(v)\cup \mcal X_{2,2}^\nc(v)}(1+|\xi|^{1/4})  \ee^{-\frac{t}{2}\gn |x-z_\gn(v)| } \leq C |\zeta_\gn(v)|^{5/4}
\end{multline}
for some new constants $\tilde C,C>0$, and where in the last inequality we used again the fact that the summand involves $x\in \frac{1}{N}\Z_{>0}$.

Finally, for $x\in \mcal X_{1,3}^\nc(v)$ we have $|\zeta-\zeta_\gn(v)|\geq \delta$ which implies $\gn|x-z_\gn(v)|\geq 4c \gn^{1/3}$ for some new value $c>0$. Using the same arguments we applied for $x\in \mcal X_{1,2}^\nc(v)$ (see \eqref{eq:mcalX12spcr}), we bound
\begin{multline*}
\gn t |x-z_\gn(v)|+\frac{4}{3}(\zeta-\zeta_\gn(v))^{3/2}+2\zeta_\gn(v)(\zeta-\zeta_\gn(v))^{1/2} \\ 
\geq \frac{t}{4}\gn  |x-z_\gn(v)|+ c \gn^{1/3} +\frac{t}{2} \gn  |x-z_\gn(v)|-2\xi^{1/2} \geq \frac{t}{4}\gn  |x-z_\gn(v)|+ c \gn^{1/3}+C,
\end{multline*}
valid for every $x\in \mcal X_{1,3}^\nc(v)$. This way,
\begin{multline}\label{eq:boundXijspcr3}
\sum_{x\in \mcal X_{1,3}^\nc(v)}|\zeta_\gn(v)|(\zeta-\zeta_\gn(v))^{1/4}\ee^{-\gn t |x-z_\gn(v)|-\frac{4}{3}(\zeta-\zeta_\gn(v))^{3/2}-2\zeta_\gn(v)(\zeta-\zeta_\gn(v))^{1/2}} 
 \\ 
 + \sum_{x\in \mcal X_{2,3}^\nc(v)}|\zeta_\gn(v)||\zeta-\zeta_\gn(v)|^{1/4}\ee^{-\gn t |x-z_\gn(v)|} \\
 \leq C\ee^{-c \gn^{1/3}}|\zeta_\gn(v)|\sum_{x\in \mcal X_{1,3}^\nc(v)\cup \mcal X_{2,3}^\nc(v)}|\zeta-\zeta_\gn(v)|^{1/4}\ee^{-\frac{t}{4}\gn |x-z_\gn(v)|}
 \leq C \gn^{1/6}\ee^{-c\gn^{1/3} }|\zeta_\gn(v)|
\end{multline}
and where in the last step we also used that $|\zeta|\leq \kappa\gn^{2/3}$ for every $x\in \mcal X_{1,3}^\nc(v)\cup \mcal X_{2,3}^\nc(v)$ (see the comments after \eqref{eq:confmappropfinalunwrap}), and the constants appearing in each step may be different.

To conclude the current section, using \eqref{eq:boundXijspcr1},\eqref{eq:boundXijspcr2}, \eqref{eq:boundXijspcr3} in \eqref{eq:KKijspcr1} and \eqref{eq:KKijspcr2}, we obtain
\begin{prop}\label{prop:estKijspcrit}
    The estimate
    $$
    \msf K_{1,1}^\gin(S,s)+\msf K_{1,2}^\gin(S,s)+\msf K_{2,1}^\gin(S,s)+\msf K_{2,2}^\gin(S,s)  =\Boh(\gn^{-2/3}) \int^S_s |\zeta_\gn(v)|^{5/4}\dd v
    $$
    is valid uniformly for $s\leq S$ and both within the supercritical regime.
\end{prop}

Observe that when writing this bound, we can simply ignore the exponentially small contributions coming from \eqref{eq:boundXijspcr3} because $\zeta_\gn(v)$ grows algebraically in $\gn$.

\subsubsection{Conclusion of the estimates on $\msf R$} For ease of reference for later, we now summarize the obtained estimates on $\msf R$.

\begin{prop}\label{prop:estRfinal}
    The error term $\msf R=\msf R(S,s)$ in \eqref{deff:msfRunwrapRHPOP} satisfies the following estimates.
    \begin{enumerate}[(i)]
        \item With
\begin{equation}\label{eq:changevariablessnSn}
s_\gn\deff \frac{\msf c_{\varphi}}{t} \frac{s}{\gn^{1/3}},\quad S_\gn\deff \frac{\msf c_{\varphi}}{t}\frac{S}{\gn^{1/3}},
\end{equation}
        the estimate
        $$
        \msf R(S,s)=\Boh\left(\frac{\ee^{-\frac{4}{3}s_\gn^{3/2}}-\ee^{-\frac{4}{3}S_\gn^{3/2}}}{\gn^{1/3}}\right)
        $$
        is valid uniformly for $s\leq S$ within the subcritical regime.
        
        \item With 
        $$
        \msf R(S,s)=\Boh\left( \frac{|S|^{3/2}+|s|^{3/2}}{\gn^{5/6}} \right)
        $$
        is valid uniformly for $s\leq S$ within the critical regime.

        \item The estimate
        $$
        \msf R(S,s)=\Boh\left( \frac{|s|^{9/4}-|S|^{9/4}}{\gn^{13/12}} \right)
        $$
        is valid uniformly for $s\leq S$ within the supercritical regime.
    \end{enumerate}
\end{prop}

\begin{proof}
    Recall that $\zeta_\gn=\Boh(s/\gn^{1/3})$ (see \eqref{eq:zetasseries}). For the proofs in any of the regimes, our starting point is \eqref{eq:errorRfinalpresplit}.

    For $s,S$ in the subcritical regime, we use the expansion \eqref{eq:zetasseries}, the first identity in \eqref{eq:corrparameters}, and the change of variables $u=\msf c_\varphi v/(t \gn^{1/3})$, and as a result obtain
    $$
    \int_{s}^S \zeta_\gn(v)^{1/2}\ee^{-\frac{4}{3}\zeta_\gn(v)^{4/3}}\dd v= \frac{t \gn^{1/3}}{\msf c_{\varphi}} \left(1+\Boh\left(\frac{1}{\gn^{1/6}}\right)\right)\int_{s_\gn}^{S_\gn}u^{1/2} \ee^{-\frac{4}{3}u^{3/2}}\dd u,\quad \gn\to \infty.
    $$
    valid uniformly for $s,S$ in the supercritical regime. The remaining integral can be computed exactly through the change of variables $\frac{4}{3}u^{3/2}=v$, and using \eqref{eq:errorRfinalpresplit} and Proposition~\ref{prop:estKijsubcr} we conclude (i).

    For $s,S$ in the critical case, we estimate $|\zeta_\gn(v)|^{1/2}=\Boh(|v|^{1/2}/\gn^{1/6})$ uniformly for $v\in [s,S]$. Integrating the estimate from Proposition~\ref{prop:estKijcrit} and using \eqref{eq:errorRfinalpresplit}, part (ii) follows. With Proposition~\ref{prop:estKijspcrit} in mind, part (iii) follows analogously.
\end{proof}

For $s,S$ in the critical regime and {\it with the same sign}, the term $|S|^{3/2}+|s|^{3/2}$ in Proposition~\ref{prop:estRfinal}--(ii) could be replaced by the more refined term $||S|^{3/2}-|s|^{3/2}|$. But for our purposes, the estimate as stated is sufficient.

\subsection{Asymptotic analysis of $\msf S$}\label{sec:analysisS}\hfill

Finally, we now analyze the quantity $\msf S$. Once again with the identifications \eqref{eq:confmappropfinalunwrap}--\eqref{deff:xpmendpoints} in mind, we rewrite the definition of $\msf S$ from \eqref{deff:Ssum} as
\begin{equation}\label{eq:SsumH}
\begin{aligned}
\msf S(S,s) & 
=-\frac{1}{2\pi\ii \gn} \int_s^S \sum_{x\in \mcal X^\gin} \zeta'(x)\mcal H_\gn(\zeta(x)-\zeta_\gn(v)\mid v) \frac{\ee^{-\gn t(x-z_\gn(v))}}{(1+\ee^{-\gn t(x-z_\gn(v))})^2}\dd v  \\
& =-\frac{1}{2\pi\ii \gn^{1/3}}\int_s^S\sum_{k \in \gn \mcal X^{\gin}} \varphi'\left( \frac{k}{\gn} \right)\mcal H_\gn\left( \zeta\left( \frac{k}{\gn} \right) -\zeta_\gn(v)\Big|\, v \right) \frac{\ee^{-tk+\gn t\ga -v}}{(1+\ee^{-tk+\gn t\ga -v})^2}\; \dd v.
\end{aligned}
\end{equation}

Even though the argument $k/\gn$ of the terms $\varphi'$ and $\mcal H_\gn$ indicate that the sum in \eqref{eq:SsumH} may be seen as a Riemann sum, the weight-like term $\ee^{(\cdot)}/(1+\ee^{(\cdot)})^2$ already lives on a correct dominant scale, and the sum cannot be approximated by an integral. So our next goal is to estimate this sum directly.

The first ingredient is a key identity for the series obtained after formally setting $\varphi' \mcal H_\gn\equiv 1$.

\begin{lemma}\label{lem:PoissonSummation}
    For any $u\in \R$ and any $t>0$, the identity
    \begin{equation}\label{eq:PoissonIdentity}
    \sum_{k\in \Z}\frac{\ee^{-tk+u}}{(1+\ee^{-tk+u})^2}=\frac{1}{t}+\frac{8\pi}{t^2}\sum_{k=1}^\infty \frac{k\cos(2\pi k u/t)}{\ee^{2\pi^2 k/t}-\ee^{-2\pi^2k/t}}
    \end{equation}
    holds true.
\end{lemma}
\begin{proof}
Setting
$$
f(x)\deff \frac{\ee^{-tx+u}}{(1+\ee^{-tx+u})^2},\quad \text{with Fourier transform}\quad \wh f(w)\deff \int_{-\infty}^\infty f(x) \ee^{-2\pi \ii w x}\dd x,
$$
Poisson Summation Formula yields the identity
\begin{equation}\label{eq:poissonformula}
\sum_{k\in \Z}\frac{\ee^{-tk+u}}{(1+\ee^{-tk+u})^2}=\sum_{k\in \Z} f(k)=\sum_{k\in \Z} \wh f(k)=\wh f(0)+\sum_{k= 1}^\infty (\wh f(k)+\wh f(-k)).
\end{equation}
Through the change of variables $x\mapsto y=tx-u$, we obtain
$$
\wh f(w)=\frac{\ee^{-2\pi \ii w u /t}}{t}\int_{-\infty}^\infty \frac{\ee^y}{(1+\ee^y)^2}\ee^{-2\pi \ii w y/t}\dd t.
$$
For $w=0$, this integral equals $1$. For $w\neq 0$, we use the identity
$$
\int_{-\infty}^\infty \frac{\ee^y}{(1+\ee^y)^2}\ee^{-\ii \omega y}\dd y=\frac{2 \omega}{\ee^{2\pi \omega}-\ee^{-2\pi \omega}},
$$
which is the Fourier transform of the logistic distribution. Replacing these identities in \eqref{eq:poissonformula}, we obtain the result.
\end{proof}

\begin{remark}\label{rmk:JacobiTheta}
    The $4$-th Jacobi Theta function $\theta_4(z\mid q)$ satisfies the identity \cite[Equation~(20.5.13)]{NIST:DLMF}
    $$
    \partial_z \log \theta_4(z\mid q)=4\sum_{k=1}^\infty \frac{q^n}{1-q^{2n}}\sin(2n z).
    $$
    Differentiating this identity one more time and using the result into \eqref{eq:PoissonIdentity}, we obtain the alternative representation
    $$
    \sum_{k\in \Z}\frac{\ee^{-tk+u}}{(1+\ee^{-tk+u})^2}=\frac{1}{t}+\frac{\pi}{t} \partial_z^2 \log \theta_4(z\mid q)\Bigg|_{z=\pi u/t,\, q=\ee^{-2\pi^2/t}},
    $$
    which provides a more compact way of expressing \eqref{eq:poissonformula}.
\end{remark}

Let us briefly describe how we will proceed to estimate $\msf S(S,s)$. The term
$$
\frac{\ee^{-tk+\gn t \ga-v}}{(1+\ee^{-tk+\gn t\ga-v})^2}=\frac{\ee^{-\gn t|\frac{k}{\gn}-z_\gn(v)|}}{(1+\ee^{-\gn t|\frac{k}{\gn}-z_\gn(v)|})^2}
$$
has non-negligible mass for $k/\gn$ around $z_\gn(v)$, and it decays exponentially fast to zero elsewhere. Thus, at the formal level, we expect that the major contribution should be coming from this point $\frac{k}{\gn}\approx z_\gn(v)$, corresponding to $\zeta \approx \zeta_\gn(v)$. Near these points,
$$
\varphi'\left(\frac{k}{\gn}\right)=-\msf c_{\msf P}+\boh(1)\quad \text{and}\quad \mcal H_\gn(\zeta-\zs\mid s)=\mcal H_\gn(0\mid s)+\boh(1).
$$
As we will show later, such heuristics will in turn imply that
\begin{multline*}
\sum_{k \in \gn \mcal X^\gin} \varphi'\left( \frac{k}{\gn} \right)\mcal H_\gn\left( \zeta\left( \frac{k}{\gn}\right)-\zeta_\gn(v) \Big|\, v \right) \frac{\ee^{-tk+\gn t\ga -v}}{(1+\ee^{-tk+\gn t\ga -v})^2} \\ 
\approx -\msf c_{\varphi}\mcal H_\gn(0\mid v) \sum_{k\in \Z}\frac{\ee^{-tk+\gn t\ga -v}}{(1+\ee^{-tk+\gn t\ga -v})^2} \dd v.
\end{multline*}
The remaining sum is then computed exactly through \eqref{lem:PoissonSummation} and the choice $u=\gn t\ga-v$, and it yields that
\begin{multline*}
\msf S(S,s)\approx \\ 
\frac{\msf c_{\varphi}}{\gn^{1/3} 2\pi\ii t}\left[\int_{s}^S \mcal H_\gn(0\mid v) \dd v+ \frac{8\pi}{t}\sum_{k=1}^\infty \frac{k}{\ee^{2\pi^2 k/t}-\ee^{-2\pi^2 k/t}}\int_{s}^S \mcal H_\gn(0\mid v) \cos\left( \frac{2\pi k}{t}(\gn t \ga -v ) \right)\dd v \right].
\end{multline*}
In line with the classical stationary phase method for estimation of oscillatory integrals, we then show that the integrals of $\mcal H_\gn(0\mid v)$ against the cosine terms are of smaller order than the integral of $\mcal H_\gn(0\mid v)$ itself. Furthermore, with the asymptotics for $\mcal H_\gn(0\mid v)$ that we already have, we then also analyze this leading integral asymptotically, which ultimately will yield the leading term asymptotics for $\msf S(S,s)$ - in fact, we will introduce these asymptotics for $\msf H_\gn$ even before turning this last approximation into a rigorous estimate.

However, to turn this outline into rigorous results, we once again need to split the analysis into the subcritical, critical and supercritical cases, as in each of these cases we have different asymptotic behavior for $\mcal H_\gn$. We now proceed to the details.

\subsubsection{Analysis of $\msf S$ in the subcritical regime}\hfill

The asymptotic behavior of $\mcal H_\gn$ in the subcritical regime is described by Theorem~\ref{prop:fundmcalHallregimes}--(i). We use it and rewrite \eqref{eq:SsumH} as
\begin{equation}\label{eq:mcalSpc}
\msf S(S,s)=-\frac{1}{\gn}\int_s^S \sum_{x\in \mcal X^\gin}\zeta'\left(x\right)\left[\msf A\left(\zeta\left(x\right),\zeta\left(x\right)\right) +\Boh\left(\ee^{-\eta\zeta_\gn(v)-\frac{4}{3}\zeta(x)^{3/2}_+ }\right)\right]  \frac{\ee^{\gn^{1/3}\msf P_\gn(\zeta(x)\mid v)}}{\left(1+\ee^{\gn^{1/3}\msf P_\gn(\zeta(x)\mid v)}\right)^2}\dd v,
\end{equation}
valid as $\gn\to \infty$. Based on the heuristics that the major contribution should be coming from points $\zeta(x)\approx \zeta_\gn(v)$, and as we will see also guided by the asymptotc behavior of $\msf A$ in \eqref{eq:diagAiryKernelAsymp} and \eqref{eq:diagAiryKernelAsympNeg}, for fixed $\delta>0$ we introduce the sets
\begin{align*}
& \mcal X^\pc_<\deff \{x\in \mcal X^\gin  \mid \zeta(x)\leq 0  \}, \\
& \mcal X^\pc_0=\mcal X^\pc_0(v)\deff \{ x\in \mcal X^\gin \mid |\zeta(x)- \zeta_\gn(v)|<\delta   \}, \\
& \mcal X^\pc_>=\mcal X^\pc_>(v)\deff \{ x\in \mcal X^\gin \mid 0\leq \zeta(x)\leq \zeta_\gn(v)-\delta \text{ or }  \zeta\geq \zeta_\gn(v)+\delta  \},
\end{align*}
and split
$$
\sum_{x\in \mcal X^\gin}\zeta'\left(x\right)\left[\msf A\left(\zeta\left(x\right),\zeta\left(x\right)\right) +\Boh\left(\ee^{-\eta\zeta_\gn(v)-\frac{4}{3}\zeta(x)^{3/2}_+ }\right)\right]  \frac{\ee^{\gn^{1/3}\msf P_\gn(\zeta(x)\mid v)}}{\left(1+\ee^{\gn^{1/3}\msf P_\gn(\zeta(x)\mid v)}\right)^2}=S_<^\pc+S_0^\pc+S_>^\pc,
$$
where $S^\pc_j$ corresponds to the sum over $\mcal X^\pc_j$, $j=0,<,>$. 

Recall that $\zeta(x),\zeta'(x)=\Boh(\gn^{2/3})$ uniformly for $x\in \mcal X^\gin$. In the subcritical regime, we have that $\zeta_\gn(v)$ is sufficiently large (and) positive, and therefore using Proposition~\ref{prop:canonicestPF} we see that the exponential factor $\ee^\bullet/(1+\ee^\bullet)^2$ in the sum is $\Boh(\ee^{-\eta\gn^{1/3}|\zeta-\zeta_\gn(v)|})$. From \eqref{eq:diagAiryKernelAsympNeg}, the Airy factor is $\Boh(\zeta(x)^{1/2})=\Boh(\gn^{1/3})$. All in all, we obtain
$$
S_<^\pc=\sum_{x\in \mcal X^\pc_<}\Boh\left(\gn \ee^{-\eta' \gn^{1/3}|\zeta(x)-\zeta_\gn(v)|}\right)=\Boh\left(\ee^{-\eta \gn^{1/3}\zeta_\gn(v)}\right),\quad \gn\to \infty
$$
where for the last equality we used that $|\zeta-\zeta_\gn(v)|=|\zeta|+\zeta_\gn(v)$ because $\zeta_\gn(v)>0$ in the subcritical regime, and also that $\mcal X_<^\pc\subset \mcal X^\gin$ has $\Boh(\gn)$ elements. From the very definition of the subcritical regime in Assumptions~\ref{assumpt:parameterregimes} and the expansion \eqref{eq:zetasseries} with $\gt=\gn^{1/3}$, we see that
\begin{equation}\label{eq:ineqN13zeta}
\gn^{1/3}\zeta_\gn(v)\geq \gn^{1/3}\gg \zeta_\gn(v)^{3/2},\quad \gn\to \infty,
\end{equation}
and we update the last estimate to
\begin{equation}\label{eq:Sminuspcasympt}
S_<^\pc=\sum_{x\in \mcal X_<^\pc}\Boh\left(\gn \ee^{-\eta \gn^{1/3}|\zeta(x)| -\eta\gn^{1/3}\zeta_\gn(v)} \right)=\Boh\left(\ee^{-3\zeta_\gn(v)^{3/2}}\right),
\end{equation}
which will be more useful for us later. The constant $3$ is chosen above solely for notational convenience, any constant strictly larger than $4/3$ would work for what comes next, with suitable adjustments of other constants appearing later on.

Next, we estimate the sum over $\mcal X_>^\pc$. In this set, we estimate the factor $\msf A$ simply as a bounded factor. Using Proposition~\ref{prop:canonicestPF} we then estimate
\begin{equation}\label{eq:Spluspcasympt}
S_>^\pc=\sum_{x\in \mcal X_<^\pc}\Boh\left(\gn^{2/3}\ee^{-\eta\gn^{1/3}|\zeta-\zeta_\gn(v)|} \right)=\Boh\left( \gn^{\frac{5}{3}}\ee^{-\gn^{1/3}\eta \delta} \right)=\Boh(\ee^{-3\zeta_\gn(v)^{3/2}}),
\end{equation}
where in the last estimate we used again \eqref{eq:ineqN13zeta}.

Lastly, we need to estimate the sum over $\mcal X^\pc_0$. We assume that $s$ is sufficiently large so that the asymptotics \eqref{eq:diagAiryKernelAsymp} are valid for every $\zeta(x)\in \mcal X_0^\pc(v)$ and any $v\geq s$. This way, we estimate immediately
$$
\msf A(\zeta(x),\zeta(x))+\Boh(\ee^{-\eta\zeta_\gn(v)-\frac{4}{3}\zeta(x)_+^{3/2}})=\msf A(\zeta(x),\zeta(x))\left(1+\Boh(\ee^{-\eta' \zeta_\gn(v)})\right),\quad x\in \mcal X_0^\pc,\quad \gn\to \infty,
$$
where, as always, this estimate is valid uniformly for $s,S$ in the subcritical regime.

To continue, let us fix $\nu\in (0,1/4)$ and further split
\begin{align*}
& \mcal X^\pc_\gout=\mcal X^\pc_\gout(v)\deff \left\{ x\in \mcal X^\pc_0 \mid \frac{\delta}{\gn^{\nu}\zeta_\gn(v)^{1/2}} \leq |\zeta(x)- \zeta_\gn(v)|<\delta   \right\}, \\
& \mcal X^\pc_\gin=\mcal X^\pc_\gin(v)\deff \left\{ x\in \mcal X^\pc_0 \mid |\zeta(x)- \zeta_\gn(v)|<\frac{\delta}{\gn^{\nu}\zeta_\gn(v)^{1/2}}   \right\},
\end{align*}
with corresponding split of sum $S_0^\pc=S_\gin^\pc+S_\gout^\pc$ with
$$
S_j^\pc= \left(1+\Boh(\ee^{-\eta \zeta_\gn(v)})\right)\sum_{x\in \mcal X^\pc_j} \zeta'(x)\msf A(\zeta(x),\zeta(x))\frac{\ee^{\gn^{1/3}\msf P_\gn(\zeta(x))}}{(1+\ee^{\gn^{1/3}\msf P_\gn(\zeta(x))})^2},\quad j=({\rm in}),({\rm out}).
$$
Along $\mcal X^\pc_\gout$, we use again Proposition~\ref{prop:canonicestPF} and \eqref{eq:diagAiryKernelAsymp} to estimate 
$$
S_\gout^\pc=\sum_{x\in \mcal X^\pc_\gout}\Boh\left( \ee^{-\frac{4}{3}\zeta(x)^{3/2}-2\eta \gn^{1/3}|\zeta(x)-\zeta_\gn(v)|} \right).
$$
It is simple to see that $\mcal X^\pc_\gout\subset \mcal X_2^\pc$, where the latter set was introduced in \eqref{deff:mcalXjpcKbounds}. Following the same arguments carried out in \eqref{eq:sumX1pc2}--\eqref{eq:sumX1pc3}, we see that the minimum of the function $\zeta\mapsto \frac{4}{3}\zeta^{3/2}+\eta \gn^{1/3}|\zeta-\zeta_\gn(v)|$ on $[\zeta_\gn(v)-\delta,\zeta_\gn(v)+\delta]\supset \zeta(\mcal X^\pc_\gout)$ is attained at $\zeta =\zeta_\gn(v)$. This argument yields the bound
\begin{align}
S_\gout^\pc & =\sum_{x\in \mcal X^\pc_\gout}\Boh\left( \ee^{-\frac{4}{3}\zeta_\gn(v)^{3/2}-\eta \gn^{1/3}|\zeta(x)-\zeta_\gn(v)|} \right)=\ee^{-\frac{4}{3}\zeta_\gn(v)^{3/2}}\sum_{x\in \mcal X^\pc_\gout}\Boh(\ee^{-\eta \gn^{1/3}|\zeta(x)-\zeta_\gn(v)|}) \nonumber \\ 
& =\Boh(\ee^{-\frac{4}{3}\zeta_\gn(v)^{3/2}-\eta' \gn^{1/3-\nu}\zeta_\gn(v)^{-1/2}}). \label{eq:Soutpcest}
\end{align}
We notice that because we are in the subcritical regime, we have $\gn^{1/3-\nu}\zeta_\gn(v)^{-1/2}\geq  \delta \gn^{1/4-\nu}$, so the corresponding term in the estimate above decays exponentially with a power of $\gn$.

Finally, we now want to estimate the sum over $\mcal X_\gin^\pc$, which will yield the leading contribution. From the very definition of $\mcal X_\gin^\pc$, the estimate
$$
\zeta(x)^{3/2}=\zeta_\gn(v)^{3/2}\left(1-\frac{\zeta(x)-\zeta_\gn(v)}{\zeta_\gn(v)}\right)^{3/2}=\zeta_\gn(v)^{3/2}+\Boh\left(\zeta_\gn^{1/2}(\zeta(x)-\zeta_\gn)\right)=\zeta_\gn(v)^{3/2}+\Boh(\gn^{-\nu})
$$
is valid as $\gn\to \infty$ uniformly for $x\in \mcal X_\gin^\pc$. This estimate combined with \eqref{eq:diagAiryKernelAsymp} yields
$$
\msf A(\zeta(x),\zeta(x))=\msf A(\zeta_\gn(v),\zeta_\gn(v))\left(1+\Boh(\gn^{-\nu})\right),\quad \gn\to \infty,
$$
also uniformly for $x\in \mcal X^\pc_\gin$. Also, by the analyticity of $x\mapsto \zeta=\gn^{2/3}\varphi(x)$,
\begin{equation}\label{eq:ineqzetax}
\frac{C}{\gn^{2/3}}|\zeta_\gn(v)-\zeta(x)|\leq |x-z_\gn(v)|\leq \frac{C}{\gn^{2/3}}|\zeta_\gn(v)-\zeta(x)|
\end{equation}
and therefore we estimate $\zeta'(x)=-\gn^{2/3}\msf c_\varphi (1+\Boh(\gn^{-2/3-\nu}\zeta_\gn(v)^{-1/2}))$. Hence,
$$
S^\pc_\gin =-\msf c_\varphi\gn^{2/3} \left(1+\Boh(\ee^{-\eta\zeta_\gn(v)}+\gn^{-\nu})\right)\msf A(\zeta_\gn(v),\zeta_\gn(v))\sum_{x\in \mcal X_\gin^\pc}\frac{\ee^{\gn^{1/3}\msf P_\gn(\zeta(x))}}{(1+\ee^{\gn^{1/3}\msf P_\gn(\zeta(x))})^2}.
$$
Writing $\gn^{1/3}\msf P_\gn(\zeta(x))=-\gn t (x-z_\gn(v))$, and using the first inequality in \eqref{eq:ineqzetax}, we see that the sum above may be extended to a sum over the whole set $\frac{\gn}{\Z}$, at the cost of an error of order at most $\ee^{-\eta \gn^{1/4-\nu}}$. Therefore,
\begin{equation}\label{eq:Sinpcest}
S^\pc_\gin =-\msf c_\varphi\gn^{2/3} \left(1+\Boh(\ee^{-\eta\zeta_\gn(v)}+\gn^{-\nu})\right)\msf A(\zeta_\gn(v),\zeta_\gn(v))\sum_{k\in \Z} \frac{\ee^{-tk+\gn t\ga-v}}{(1+\ee^{-tk+\gn t\ga-v})^2},
\end{equation}
valid as $\gn\to \infty$, uniformly for $v\in [s,S]$ and $s\leq S$ in the subcritical regime.

Finally, using \eqref{eq:diagAiryKernelAsymp} we compare \eqref{eq:Sminuspcasympt}, \eqref{eq:Spluspcasympt}, \eqref{eq:Soutpcest} and \eqref{eq:Sinpcest} asymptotically, concluding that the dominant term is precisely $S^\pc_\gin$. Plugging this info into \eqref{eq:mcalSpc}, we obtain in summary that
$$
\msf S(S,s)=\frac{\msf c_\varphi }{\gn^{1/3}}\left(1+\Boh(\ee^{-\eta\zeta_\gn(s)}+\gn^{-\nu})\right) \int_s^S \msf A(\zeta_\gn(v),\zeta_\gn(v))  \sum_{k\in \Z} \frac{\ee^{-tk+\gn t\ga-v}}{(1+\ee^{-tk+\gn t\ga-v})^2}  \dd v.
$$
Using once more \eqref{eq:diagAiryKernelAsymp} and the expansion \eqref{eq:zetasseries} for $\zeta_\gn$, we see that
$$
\msf A(\zeta_\gn(v),\zeta_\gn(v))=\msf A\left( \frac{v}{\msf c_{\msf P}\gn^{1/3} }, \frac{v}{\msf c_{\msf P}\gn^{1/3} } \right)\left(1+\Boh\left( \frac{v^{5/2}}{\gn^{3/2}} \right)\right).
$$
In the subcritical regime this last error term is of order at least $\gn^{-1/6}$, and therefore
$$
\msf S(S,s)=\frac{\msf c_\varphi }{\gn^{1/3}}\left(1+\Boh(\ee^{-\eta s/\gn^{1/3}}+\gn^{-\nu})\right) \int_s^S \msf A\left( \frac{v}{\msf c_{\msf P}\gn^{1/3} }, \frac{v}{\msf c_{\msf P}\gn^{1/3} } \right)\sum_{k\in \Z} \frac{\ee^{-tk+\gn t\ga-v}}{(1+\ee^{-tk+\gn t\ga-v})^2} \dd v,
$$
valid as $\gn\to \infty$, for any $\nu\in (0,1/6)$ and uniformly for $s,S$ in the subcritical regime.

Having in mind that $\msf c_{\msf P}=t/\msf c_\varphi$ (see \eqref{eq:corrparameters}), using Lemma~\ref{lem:PoissonSummation}, changing variables
\begin{equation}\label{eq:changevariablespoisson}
v=\frac{t\gn^{1/3}y}{\msf c_\varphi}
\end{equation}
recalling the change $s\mapsto s_\gn, S\mapsto S_\gn$ from \eqref{eq:changevariablessnSn} and setting for simplicity
\begin{equation}\label{eq:coeffdk}
d_k\deff \frac{8\pi}{t}\frac{k}{\ee^{2\pi^2k t}-\ee^{-2\pi^2 k t}},\quad k\geq 1,\
\end{equation}
we obtain
\begin{multline*}
\msf S(S,s)=\left(1+\Boh(\ee^{-\eta s/\gn^{1/3}}+\gn^{-\nu})\right) \\ 
\times \left[ \int_{s_\gn}^{S_\gn} \msf A(y,y)\dd y + \sum_{k=1}^\infty d_k \int_{s_\gn}^{S_\gn} \msf A(y,y)\cos\left( 2\pi k(\gn t \ga-\msf c_{\varphi}^{-1}\gn^{1/3}y) \right)\dd y\right].
\end{multline*}

Using \eqref{eq:diagAiryKernelAsymp}, a simple application of Laplace's method for asymptotics of integrals shows that
\begin{equation}\label{eq:asymptAirydiagint}
\int_{t}^{\infty}\msf A(y,y)\dd y=\frac{1}{16\pi}\frac{\ee^{-\frac{4}{3}t^{3/2}}}{t^{3/2}}\left(1+\Boh(t^{-3/2})\right),\quad t\to \infty,
\end{equation}
and we apply this estimate for $t=s_n$ and $t=S_N$. In turn, integration by parts and again an application of Laplace's method yields
$$
\int_{s_\gn}^{S_\gn} \msf A(y,y)\cos\left( 2\pi k(\gn t \ga-\msf c_{\varphi}^{-1}\gn^{1/3}y) \right)\dd y=\Boh\left( \frac{\ee^{-\frac{4}{3}s_\gn^{3/2}}}{k \gn^{1/3}s_\gn} \right),\quad \gn\to \infty,
$$
both being valid uniformly for $s<S$ in the subcritical regime. This way, and using in addition that $s\gn^{-1/3}\ll \gn^{1/6}$ to suppress a further error term when comparing these last two terms, and we conclude the following result.

\begin{theorem}\label{thm:Ssubfinal}
For any $\nu\in (0,1/6)$, the estimate
$$
\msf S(s,S)=\left(1+\Boh(\ee^{-\eta s/\gn^{1/3}}+\gn^{-\nu}) \right)\int_{s_\gn}^{S_\gn}\msf A(y,y) \dd y,\quad \gn\to \infty,
$$
is valid uniformly for $s<S$ and both in the subcritical regime.
\end{theorem}

\subsubsection{Analysis of $\msf S$ in the critical regime}\hfill

We start recalling that in the critical regime, the asymptotic behavior of $\mcal H_\gn$ is described by Theorem~\ref{prop:fundmcalHallregimes}--(ii). Based on the asymptotic behavior described therein, we fix $L>0$ and introduce the sets
\begin{align*}
& \mcal X^\rc_0=\mcal X^\rc_0(v)\deff \{x\in \mcal X^\gin  \mid |\zeta(x)-\zeta_\gn(v)|\leq L  \}, \\
& \mcal X^\rc_<=\mcal X^\rc_<(v)\deff \{ x\in \mcal X^\gin \mid \zeta(x)-\zeta_\gn(v) <-L  \}, \\
& \mcal X^\rc_>=\mcal X^\rc_>(v)\deff \{ x\in \mcal X^\gin \mid \zeta(x)-\zeta_\gn(v)>L  \}.
\end{align*}
For what comes next, we also recall that $\zeta_\gn(v)$ remains bounded for $v$ in the critical regime. In particular, for any $x\in \mcal X^\gin$ we have $|\zeta(x)-\zeta_\gn(v)|\leq \kappa \gn^{2/3}$ for some constant $\kappa>0$ (see \eqref{deff:xpmendpoints}).

We express the sum on the first line of \eqref{eq:SsumH} as
$$
\sum_{x\in \mcal X^\gin}\zeta'(x)\mcal H_\gn(\zeta(x)-\zeta_\gn(v)\mid v)\frac{\ee^{-\gn t(x-z_\gn(v))}}{(1+\ee^{-\gn t(x-z_\gn(v))})^2}=S^\rc_0+S^\rc_<+S^\rc_>,
$$
where the term $S^\rc_j$ corresponds to the sum over $\mcal X^\rc_j=\mcal X^\rc_j(v)$, $j=0,<,>$. 

To estimate each of these sums, observe that the very definition of $\zeta=\gn^{2/3}\varphi$ and  \eqref{eq:expansionvarphi} imply that there exist $c,C>0$ such that
$$
c\gn^{2/3}|x-\ga|\leq |\zeta(x)|\leq C\gn^{2/3}|x-\ga|\quad \text{as well as}\quad c\gn^{2/3}\leq |\zeta'(x)|\leq C\gn^{2/3}.
$$
Consequently, in the critical regime,
\begin{align*}
& |x-\ga|-\frac{M}{\gn^{2/3}}\leq |x-\ga|-|x-z_\gn(v)|\leq  |x-z_\gn(v)|\leq |x-\ga|+|x-z_\gn(v)|\leq |x-\ga|+\frac{M}{\gn^{2/3}}, 
\end{align*}
as well as
$$
c\gn^{2/3}|x-\ga|-\wt M\leq |\zeta(x)|-|\zeta_\gn(v)|\leq |\zeta(x)-\zeta_\gn(v)|\leq |\zeta(x)|+|\zeta_\gn(v)|\leq C\gn^{2/3}|x-\ga|+\wt M
$$
for some $M,\wt M>0$.
These estimates together with Theorem~\ref{prop:fundmcalHallregimes}--(ii) imply that
$$
S^\rc_>=\Boh\left( \gn^{2/3}\sum_{x\in \mcal X^\rc_>}\ee^{-\eta \gn^{1/3}  |\zeta(x)-\zeta_\gn(v)| } \ee^{-\gn C|x-z_\gn(v)|} \right)=\Boh(\gn^{1/3}\ee^{-\eta\gn^{2/3}L}|\mcal X^\rc_>|),
$$
for yet new constants $\eta,C>0$. Now, the set $\mcal X^\gin\supset \mcal X^\rc_<$ has $\Boh(\gn)$ elements, and from the estimate above we conclude that
\begin{equation}\label{eq:estScrtgeq}
S_>^\rc=\Boh(\ee^{-\eta \gn^{1/3}}),\quad \gn\to \infty,
\end{equation}
for some constant $\eta>0$, uniformly for $v$ in the critical regime.

In a completely similar way,
\begin{equation}\label{eq:estScrtleq}
S^\rc_<=\Boh(\ee^{-\eta\gn^{1/3}}),\quad \gn\to \infty,
\end{equation}
uniformly for $v$ in the critical regime.

We now estimate the sum $S_0^\rc$. Observe that $|x-\ga|=\Boh(\gn^{-2/3})$ for $x\in \mcal X^\rc_0$, and therefore
$$
\zeta'(x)=\gn^{2/3}\varphi'(x)=-\gn^{2/3}\msf c_{\varphi}(1+\Boh(\gn^{-2/3})),
$$
where we used \eqref{eq:expansionvarphi}, and where the error term is independent of $v$. Combining with Theorem~\ref{prop:fundmcalHallregimes}--(ii), we obtain
$$
S_0^\rc= -2\pi \ii \msf c_{\varphi} \gn^{2/3}\left(1+\Boh(\gn^{-2/3})\right)\sum_{x\in \mcal X^\rc_0}\msf K_{\ptf}(\zeta(x)-\zeta_\gn(v),\zeta(x)-\zeta_\gn(v)\mid \zeta_\gn(v))\frac{\ee^{-\gn t(x-z_\gn(v))}}{(1+\ee^{-\gn t(x-z_\gn(v))})^2},
$$
valid for $\gn\to \infty$, again with error term uniform for $v$ in the positive critical regime. Next, for any fixed $\nu\in (0,1/3)$ we split the set $\mcal X_0^\rc$ into two further subsets,
$$
\mcal X_0^\rc=\mcal X^\rc_\gin\cup\mcal X^\rc_\gout,
$$
with
$$
\mcal X^\rc_\gin=\{x\in \mcal X^\rc_0\mid |x-z_\gn(v)|\leq \gn^{-2/3-\nu} \},\quad \mcal X^\rc_\gout=\{x\in \mcal X^\rc_0\mid |x-z_\gn(v)|\geq \gn^{-2/3-\nu} \}.
$$
The term $\msf K_\ptf$ is bounded on $\mcal X_\gin^\rc$, and from the very definition of $\mcal X^\rc_\gout$ we obtain that the corresponding sum over this set is $\Boh(\ee^{-\eta \gn^{1/3-\nu}})$, for some $\eta>0$. On the other hand, $|\zeta(x)-\zeta_\gn(v)|=\Boh(\gn^{-\nu})$ for $x\in \mcal X^\rc_\gin$, and the analyticity of $\msf K_\ptf(\zeta,\zeta\mid y)$ as a function of $\zeta$ implies that
$$
\msf K_{\ptf}(\zeta(x)-\zeta_\gn(v),\zeta(x)-\zeta_\gn(v)\mid \zeta_\gn(v))=\msf K_{\ptf}\left(0,0\mid \zeta_\gn(v)\right)(1+\Boh(\gn^{-\nu}),\quad \gn\to \infty,
$$
again valid uniformly for $v$ in the critical regime. All in all, we obtain
$$
S_0^\rc =-2\pi\ii \gn^{2/3}\msf c_\varphi \msf K_{\ptf}\left(0,0\mid \zeta_\gn(v)\right) (1+\Boh(\gn^{-\nu}))\sum_{x\in \mcal X^\rc_\gin}\frac{\ee^{-\gn t(x-z_\gn(v))}}{(1+\ee^{-\gn t(x-z_\gn(v))})^2},
$$
valid as $\gn\to \infty$, for any $\nu\in (0,1/3)$, uniformly for $v$ in the critical regime. The sum over $\mcal X^\rc_\gin$ above may be replaced by a sum over the whole set $\frac{1}{\gn}\Z$, at a cost of an exponentially small error. Consequently, combining this last estimate with \eqref{eq:estScrtgeq} and \eqref{eq:estScrtleq} into \eqref{eq:SsumH}, we conclude
$$
\msf S(S,s)=\frac{\msf c_\varphi}{\gn^{1/3}}(1+\Boh(\gn^{-\nu}))\int_s^S \msf K_\ptf(0,0\mid \zeta_\gn(v))\sum_{k\in \Z} \frac{\ee^{-tk+\gn t\ga-v}}{(1+\ee^{-tk+\gn t\ga-v})^2}\dd v,
$$
where the error term is valid uniformly for $s<S$ and both in the critical regime, and where the error is valid for any $\nu\in (0,1/3)$. Now, for $s,S$ in the critical regime we have $|v|\leq C\gn^{1/3}$ in the integration interval. From the analyticity of $y\mapsto \msf K_\ptf(0,0\mid y)$ and the expansion \eqref{eq:zetasseries}, we expand
$$
\msf K_\ptf(0,0\mid \zeta_\gn(v))=\msf K_\ptf\left(0,0\mid \frac{v}{\msf c_{\msf P}\gn^{1/3}}\right)\left(1+\Boh(\gn^{-2/3})\right),
$$
uniformly for $v$ in the critical regime. Recalling that $\msf c_{\msf P}=t/\msf c_{\varphi}$ (see again \eqref{eq:corrparameters}), we update the last estimate on $\msf S(S,s)$ to
$$
\msf S(S,s)=\frac{\msf c_\varphi}{\gn^{1/3}}(1+\Boh(\gn^{-\nu}))\int_s^S \msf K_\ptf\left(0,0\mid \frac{\msf c_{\varphi} v}{t \gn^{1/3}} \right)\sum_{k\in \Z} \frac{\ee^{-tk+\gn t\ga-v}}{(1+\ee^{-tk+\gn t\ga-v})^2}\dd v,
$$

Changing variables $v=t\gn^{1/3}y/\msf c_{\varphi}$ as in \eqref{eq:changevariablespoisson} {\it et seq.}, we update this estimate to
\begin{multline*}
\msf S(S,s)=\left(1+\Boh(\gn^{-\nu})\right) \\ \times \left[\int_{s_\gn}^{S_\gn} \msf K_{\ptf}(0,0\mid y)\dd y+\sum_{k=1}^\infty d_k\int_{s_\gn}^{S_\gn} \msf K_{\ptf}(0,0\mid y) \cos\left( 2\pi k(\gn t \ga- \msf c_\varphi^{-1}\gn^{1/3} y) \right)\dd y\right].
\end{multline*}
where we stress that $s_\gn,S_\gn$ are as in \eqref{eq:changevariablessnSn}. In the critical regime both $S_\gn$ and $s_\gn$ remain bounded. Because $\partial_y \msf K_\ptf(0,0\mid y)$ is bounded on bounded subsets of the real line (see \eqref{eq:KptfPXXXIV}), a simple integration by parts shows now that
$$
\int_{s_\gn}^{S_\gn} \msf K_{\ptf}(0,0\mid y) \cos\left( 2\pi k(\gn t \ga- \msf c_\varphi^{-1}\gn^{1/3} y) \right)\dd y=\Boh\left(\frac{1}{k\gn^{1/3}}\right),
$$
uniformly in $s,S$, and we finally concluded the following.

\begin{theorem}\label{thm:Scritfinal}
For any $\nu\in (0,1/3)$, the estimate
$$
\msf S(S,s)=\left(1+\Boh(\gn^{-\nu})\right)\int_{s_\gn}^{S_\gn} \msf K_{\ptf}\left(0,0\mid y\right)\dd y,\quad \gn\to \infty,
$$
is valid uniformly for $s<S$ in the critical regime.
\end{theorem}

We finally move to the last case, namely the supercritical regime.

\subsubsection{Analysis of $\msf S$ in the supercritical regime}\hfill

The asymptotic behavior of $\mcal H_\gn$ in the supercritical regime is described by Theorem~\ref{prop:fundmcalHallregimes}--(iii). Based on the estimates therein, we fix $\delta>0$ and set
$$
\mcal X^\nc=\mcal X^\nc(s)\deff \{x\in \mcal X^\gin\mid |\zeta-\zn|\leq \delta\},
$$
and
$$
S^\nc=S^\nc(s)\deff \sum_{x\in \mcal X^\nc}\zeta'(x)\mcal H_\gn(\zeta(x)-\zn\mid v)\frac{\ee^{-\gn t(x-\zzn)}}{(1+\ee^{-\gn t(x-\zzn)})^2}.
$$
We insert the first estimate in Theorem~\ref{prop:fundmcalHallregimes}--(iii), obtaining in a straightforward manner that the estimate
\begin{equation}\label{eq:SSnc}
\msf S(S,s)=-\frac{1}{2\pi \ii \gn}\int_{s}^S S^\nc(v) \dd v+\Boh(\ee^{-\eta \gn^{1/3}}),\quad \gn\to \infty,
\end{equation}
is valid uniformly for $s,S$ in the supercritical regime. To analyze the withstanding term, recall the chain of change of variables
$$
x\mapsto \zeta=\zeta(x)\deff \gn^{2/3}\varphi(x),\qquad \text{and}\qquad \zeta\mapsto \xi=\xi(\zeta)\deff \zn^2\zeta,
$$
which induces the change of variables
$$
x\mapsto \xi(x)=\xi(\zeta(x))=\gn^{2/3}\zn^2\varphi(x),\quad \text{so that}\quad \xi_\gn(s)\deff \zeta(\zn)=\zn^3.
$$
We use this change of variables extensively in the coming calculations.

Let us split $\mcal X^\nc=\mcal X^\nc_\gin\cup \mcal X^\nc_\gout$, with
\begin{align*}
\mcal X^\nc_\gin & =\mcal X^\nc_\gin(s)\deff \left\{x\in \mcal X^\nc\mid |\xi(x)-\xi_\gn(s)|\leq \frac{|\zn|}{\gn^{1/6}} \right\},\\
\mcal X^\nc_\gout & =\mcal X^\nc_\gout(s)\deff \left\{x\in \mcal X^\nc\mid \frac{|\zn|}{\gn^{1/6}}< |\xi(x)-\xi_\gn(s)|\leq \zn^2 \delta  \right\}.
\end{align*}
This split of the set $\mcal X^\nc$ induces a split $S^\nc=S^\nc_\gin+S^\nc_\gout$, and we now estimate each of these terms. For these estimates, we notice that $|\zn|/\gn^{1/6}=\Boh(|s|/\gn^{1/2})$ is arbitrarily small for $s$ in the supercritical regime (see \eqref{eq:zetasseries}), and once again we recall Theorem~\ref{prop:fundmcalHallregimes}--(iii), which will guide our arguments.

We start from $S^\nc$. Estimates \eqref{eq:asymptBesseldiag1}--\eqref{eq:asymptBesseldiag2} give us the rough bound
$$
\msf J_0(-(\xi-\xi_\gn(s)),-(\xi-\xi_\gn(s)))=\Boh\left(\ee^{2(\xi-\xi_\gn(s))^{1/2}_+}\right),\quad \xi=\xi(x)\in \mcal X_\gout^\nc.
$$
Furthermore, $\zeta'(x)=\Boh(\gn^{2/3})$ uniformly for $x\in \mcal X^\nc$. Combining these estimates with Theorem~\ref{prop:fundmcalHallregimes}--(iii), we obtain that for some $\eta>0$, independent of $s$, the estimate
$$
S^\nc_\gout = \Boh\left( \gn^{2/3}\zn^2 \sum_{x\in \mcal X^\nc_\gout} \ee^{2(\xi(x)-\xi_\gn(s))^{1/2}_+-3\eta \gn^{1/3} |\zeta-\zn|}\right),
$$
is valid uniformly for $s$ in the supercritical regime.

Let us look at the exponent of the exponential in the sum. For $\xi(x)-\xi_\gn(s)<0$, we have $(\xi(x)-\xi_\gn(x))^{1/2}_+=0$ and clearly the exponential in the sum is bounded by $\ee^{-2\eta\gn^{1/3}|\zeta-\zn|}$. For $\xi(x)-\xi_\gn(s)\geq 0$, we write (part of the) exponent as
\begin{equation}\label{eq:xisumcontrol}
-2(\xi(x)-\xi_\gn(s))_+^{1/2}+\eta \gn^{1/3}|\zeta-\zn|=-2u+\frac{\eta \gn^{1/3}}{\zn^2}u^2,\quad u=(\xi(x)-\xi_\gn(s))^{1/2}.
\end{equation}
A calculus exercise shows that this function of $u$ attains its minimum value on the positive axis at the point $u_c=\zn^2/(\eta \gn^{1/3})$, and evaluating at this point we conclude that
$$
-2(\xi(x)-\xi_\gn(s))_+^{1/2}+\eta \gn^{1/3}|\zeta-\zn|\geq -u_c=-\zn^2/(\eta \gn^{1/3})
$$
This latter value is bounded from below uniformly for $s$ in the supercritical regime. As a consequence, we obtain that
$$
S^\nc_\gout = \Boh\left( \gn^{2/3}\zn^2 \sum_{x\in \mcal X^\nc_\gout} \ee^{-2\eta \gn^{1/3} |\zeta-\zn|}\right).
$$
On the set $\mcal X_\gout^\nc$, we bound 
$$
\gn^{1/3}|\zeta(x)-\zn|=\frac{\gn^{1/3}}{\zn^2}|\xi(x)-\xi_\gn(s)|\geq \frac{\gn^{1/6}}{|\zs|}.
$$
Similarly, we also have $|\zeta(x)-\zeta_\gn(s)|\geq c \gn^{2/3}|x-\zzn|$, for some $c>0$, and therefore
$$
S^\nc_\gout =\Boh\left( \gn^{2/3}
\zn^2\ee^{-\eta \gn^{1/6}/|\zn|} \sum_{x\in \mcal X_\gout^\nc} \ee^{-\eta c \gn |x-\zzn|}
\right).
$$
Writing $x=k/\gn$, it is straightforward to show that the remaining integral is finite. Our findings so far are summarized as
\begin{equation}\label{eq:SncSncinout}
S^\nc(s)=S^\nc_\gin(s)+\Boh\left(\gn^{2/3}
\zn^2 \ee^{-\eta \gn^{1/6}/\zn} \right),\quad \gn\to \infty,
\end{equation}
uniformly for $s$ in the supercritical regime, and where we stress that $\eta>0$ is independent of $s$.

Finally, we now analyze the withstanding term $S^\nc_\gin$. In this set, for some $C>0$,
\begin{align*}
|x-\ga| & \leq |\zzn-\ga|+|x-\zzn|\leq \frac{|s|}{t\gn} +\frac{C}{N^{2/3}}|\zeta(x)-\zs| \\ 
& \leq \frac{|s|}{t\gn}+ C \frac{\zn^2}{\gn^{2/3}} |\xi(x)-\xi_\gn(s)|\leq \frac{|s|}{t\gn} +  \frac{C}{|\zn|\gn^{5/6}}.
\end{align*}
The constant $C>0$ depends only on the conformal map $x\mapsto \varphi(x)$, hence it is independent of $s$. In particular, this inequality shows that $|x-\ga|=\Boh(s/\gn)$, uniformly for $s$ in the supercritical regime and $x\in \mcal S^\nc_\gin$, and therefore
$$
\zeta'(x)=-\gn^{2/3}\msf c_\varphi\left(1+\Boh\left(\frac{s}{\gn}\right)\right),\quad x\in \mcal X_\gin^\nc,\quad \gn\to \infty,
$$
again uniformly for $s$ in the supercritical regime.

Using \eqref{eq:asymptBesseldiag1}, we see that
$$
\msf J_0(-(\xi-\xi_\gn(s)),-(\xi-\xi_\gn(s)))=\frac{1}{4}\left(1+\Boh\left(\frac{s}{\gn^{1/2}}\right)\right),\quad \gn \to \infty,
$$
uniformly for $\xi\in \mcal X^\nc_\gin$ and $s$ in the supercritical regime. With the asymptotics given by Theorem~\ref{prop:fundmcalHallregimes}--(iii), we obtain
\begin{multline}\label{eq:Sncin01}
S^\nc_\gin(s)=-\gn^{2/3}\zn^2 \frac{\pi\ii \msf c_{\varphi}}{2}\left(1+\Boh\left(\frac{s}{\gn^{1/2}}\right)\right)\sum_{x\in \mcal X^\nc_\gin} \frac{\ee^{-\gn t|x-\zzn|}}{(1+\ee^{-\gn t|x-\zzn|})^2} \\ 
+\Boh\left(\gn^{2/3}\left( |\zn|+\frac{\zn^4}{\gn^{1/3}} \right)\right)\sum_{x\in \mcal X^\nc_\gin}\frac{\ee^{-\gn t|x-\zzn|+2(\xi(x)-\xi_\gn(s))_+^{1/2}}}{(1+\ee^{-\gn t|x-\zzn|})^2},
\end{multline}

The same arguments provided in \eqref{eq:xisumcontrol} {\it et seq.} show that the sum on the second line above is finite. As for the first sum, we bound as before $\gn t |x-\zzn|/2\geq \eta \frac{\gn^{1/3}}{\zn^2}|\xi-\xi_\gn(s)|\geq \eta \gn^{1/6}/|\zn|$ for $x\in \frac{1}{\gn}\Z \setminus \mcal X^\nc_\gin$ and a constant $\eta>0$ independent of $s$, obtaining as a consequence
$$
\sum_{x\in \mcal X^\nc_\gin} \frac{\ee^{-\gn t|x-\zzn|}}{(1+\ee^{-\gn t|x-\zzn|})^2}=\sum_{x\in \frac{1}{\gn}\Z} \frac{\ee^{-\gn t|x-\zzn|}}{(1+\ee^{-\gn t|x-\zzn|})^2}+\Boh\left( \ee^{-\eta \gn^{1/6}/|\zn|^{1/6} } \right),\quad \gn\to \infty.
$$
Using these bounds in \eqref{eq:Sncin01} and returning the resulting asymptotic expression in \eqref{eq:SncSncinout}, we obtain
\begin{multline}
S^\nc(s)=\\ 
-\gn^{2/3}\zn^2 \frac{\pi\ii \msf c_{\varphi}}{2}\left(1+\Boh\left(\frac{s}{\gn^{1/2}}\right)\right)\sum_{x\in \frac{1}{\gn}\Z} \frac{\ee^{-\gn t|x-\zzn|}}{(1+\ee^{-\gn t|x-\zzn|})^2}+\Boh\left( \gn^{2/3}|\zn|+\gn^{1/3}\zn^4 \right),
\end{multline}
which is valid as $\gn\to \infty$, uniformly for $s$ in the supercritical regime.

There are three error terms in the expression above, one of order $\Boh(\gn^{2/3}\zn^2 s/\gn^{1/2})$ coming in front of the sum (which is finite), and the remaining two on the right-most side. A direct calculation shows that, in the supercritical regime, the first one is dominant over the latter two. With this fact in mind, we use this estimate in \eqref{eq:SSnc}, obtaining
$$
\msf S(S,s)= \frac{\msf c_\varphi}{4\gn^{1/3}}\int_s^S \zeta_\gn(v)^2\left(1+\Boh\left(\frac{v}{\gn^{1/2}}\right)\right) \sum_{k\in \Z}\frac{\ee^{-tk+\gn t \ga -v}}{(1+\ee^{-tk+\gn t \ga -v})^2}\dd v,
$$
where we parametrized the sum as $x=k/\gn$. Thanks to \eqref{eq:zetasseries} and again the identity $\msf c_{\msf P}=t/\msf c_\varphi$, we rewrite this estimate as
$$
\msf S(S,s)=\frac{\msf c_\varphi^3}{4t^2 \gn}\int_s^S v^2\left(1+\frac{v}{\gn^{1/2}}\right) \sum_{k\in \Z}\frac{\ee^{-tk+\gn t \ga -v}}{(1+\ee^{-tk+\gn t \ga -v})^2}\dd v.
$$
The remaining integrals could be expressed in terms of logarithms and polylogarithms. Instead, we follow the same route as before and apply \eqref{eq:PoissonIdentity}, obtaining
$$
\msf S(S,s)= \frac{\msf c_\varphi^3}{4t^3 \gn}\left[\int_s^S v^2\left(1+\Boh\left(\frac{v}{\gn^{1/2}}\right)\right) \dd v+\sum_{k=1}^\infty d_k \int_s^S v^2\left(1+\Boh\left(\frac{v}{\gn^{1/2}}\right)\right)\cos\left( 2\pi k(\gn a -t^{-1}v ) \right)\dd v\right],
$$
valid as $\gn\to \infty$, uniformly for $s\leq S$ in the supercritical regime.
A direct integration shows that the $k$-th (indefinite) integral inside the series is $\Boh(v^2/k)$. In the supercritical regime, this contribution is smaller in order than the contribution $\Boh(v^4/\gn^{1/2})$ arising from the error term in the first integral. We concluded
\begin{theorem}\label{thm:Ssuperfinal}
The estimate
$$
\msf S(S,s)= \frac{\msf c_{\varphi}^3}{12t^3\gn}(S^3-s^3)+\Boh\left(\frac{S^4-s^4}{\gn^{3/2}}\right)
$$
is valid uniformly for $s\leq S$ in the supercritical regime.
\end{theorem}

Observe that the leading contribution and the error term become of the same order precisely when $s=\Boh(\gn^{1/2})$.

This result ends our asymptotic analysis, and we now collect all the ingredients to complete the proof of our main results.

\section{Proof of main results on dOPEs}\label{sec:ProofOfMain}

We are finally ready to conclude the proof of our main results on discrete orthogonal polynomial ensembles.

\subsection{Proof of Theorem~\ref{thm:multstat_formal}}\hfill 

Our starting point is \eqref{eq:SLrelation}. With the help of \eqref{eq:SNZNrelation} and \eqref{eq:defformreadyforasymp2}, and recalling the definition of $\mcal X_\ga$ in \eqref{deff:mcalXa}, we obtain
\begin{equation}\label{eq:logLNsSest01}
\log\frac{\msf L_\gn(s)}{\msf L_\gn(S)}=\sum_{ k=\lceil \gn(\ga+\epsilon)\rceil }^\infty \log \frac{1+\ee^{-t(k-\gn\ga)-S}}{1+\ee^{-t(k-\gn\ga)-s}} -\msf S(S,s)-\msf R(S,s)+\Boh(\ee^{-\eta\gn}),\quad \gn\to \infty,
\end{equation}
and this identity is valid for $s,S$ in any of the regimes from Assumptions~\ref{assumpt:parameterregimes}. To bound the sum involving logs, observe that $t(k+\gn \ga)+s> 0$ whenever $s$ is in any of the regimes from Assumptions~\ref{assumpt:parameterregimes}. Using the basic inequalities $0\leq \log(1+u)\leq u$, which are valid for $u\geq 0$, we obtain
$$
\sum_{ k=\lceil \gn(\ga+\epsilon)\rceil }^\infty \log \left(1+\ee^{-t(k-\gn\ga)-s}\right) \leq \sum_{ k=\lceil \gn(\ga+\epsilon)\rceil }^\infty \ee^{-t(k-\gn\ga)-s}=\frac{\ee^{-t(\lceil \gn(\ga+\epsilon) \rceil -\gn \ga)-s}}{1+\ee^{-t}}\leq \frac{\ee^{-\gn t \epsilon-s}}{1+\ee^{-t}}.
$$
This shows that the sum of logs above is $\Boh(\ee^{-\eta\gn})$. Obviously, we may replace $s$ by $S$ in the argument just performed. As a consequence we update \eqref{eq:logLNsSest01} to
$$
\log\frac{\msf L_\gn(s)}{\msf L_\gn(S)}=-\msf S(S,s)-\msf R(S,s)+\Boh(\ee^{-\eta\gn}),\quad \gn\to \infty.
$$
Recall that $\msf S(S,s)$ is as in \eqref{eq:SsumH}. In that expression, notice that $\gn \mcal X^\gin=\Z \cap (\gn(\ga-\epsilon),\gn(\ga+\epsilon))$ (see \eqref{deff:Xin}), and the explicit term in the right-hand side of \eqref{eq:multstat_asymp} coincides with $-\msf S(S,s)$ with the identification $\msf H_\gn=\frac{1}{2\pi\ii}\mcal H_\gn$. The proof of Theorem~\ref{thm:multstat_formal} is now completed directly from Theorem~\ref{prop:fundmcalHallregimes} and Proposition~\ref{prop:estRfinal}.

\subsection{Proof of Theorem~\ref{thm:integrated_formal}}\hfill 

To prove Theorem~\ref{thm:integrated_formal}, we first establish the following result.

\begin{theorem}\label{thm:integrated_formal_pre}
Under the same hypotheses of Theorem~\ref{thm:multstat_formal},
\begin{enumerate}[(i)]
\item \textbf{Subcritical regime:} There exists $\eta>0$ such that for $s,S$ in the subcritical regime, with $s<S$, the estimate
\begin{equation}\label{eq:integratedsubcriticalsS}
\log\frac{\msf{L}_\gn(s)}{\msf L_\gn(S)} = -\left(1+\Boh\left(\ee^{-\eta s/\gn^{1/3}}+\gn^{-\nu}\right)\right)\int_{\msf{c}_\gV s/(t\gn^{1/3})}^{\msf{c}_\gV S/(t\gn^{1/3})} \msf{A}(y, y)\,\dd y,\quad \gn\to \infty,
\end{equation}
is valid for any $\nu\in (0,1/6)$.

\item \textbf{Critical regime:} For $s,S$ in the critical regime, with $s<S$, the estimate
\begin{equation}\label{eq:integratedcriticalsS}
\log\frac{\msf{L}_\gn(s)}{\msf L_\gn(S)} = -\left(1+\Boh(\gn^{-\nu})\right)\int_{\msf{c}_\gV s/(t\gn^{1/3})}^{\msf{c}_\gV S/(t\gn^{1/3})} \msf{K}_{\rm P34}(0, 0 \mid y)\, \dd y,\quad \gn\to \infty,
\end{equation}
is valid for any $\nu\in(0,1/3)$.

\item \textbf{Supercritical regime:} For $s,S$ in the supercritical regime, with $s<S$, the estimate
\begin{equation}\label{eq:integratedsupercriticalsS}
\log\frac{\msf{L}_\gn(s)}{\msf L_\gn(S)} = -\frac{\msf c_\gV^3}{12t^3\gn}\left(S^3-s^3\right)+\Boh\left(\frac{S^4-s^4}{\gn^{3/2}}\right),\quad \gn\to \infty,
\end{equation}
is valid.
\end{enumerate}
\end{theorem}

\begin{proof}
Let $S>s$ be given. Using \eqref{eq:SLrelation}, \eqref{eq:SNZNrelation} and \eqref{eq:defformreadyforasymp2}, we express
$$
\log\frac{\msf L_\gn(s)}{\msf L_\gn(S)}=-\msf S(S,s)+\sum_{x\in \mcal X_\ga}\log\frac{\gsig(x\mid s)}{\gsig(x\mid S)}-\sum_{k=1}^\infty \log \frac{1+\ee^{-t(k-\gn\ga)-s}}{1+\ee^{-t(k-\gn\ga)-S}} -\msf R(S,s)+\Boh\left(\ee^{-\eta \gn}\right),\quad \gn\to \infty.
$$
From the very definition of $\gsig$ and $\mcal X_\ga$ in \eqref{deff:gsiggW} and \eqref{deff:mcalXa}, the two sums combine into the single sum
$$
\sum_{k=\lceil \gn(\ga+\epsilon) \rceil }^\infty \log \frac{1+\ee^{-t(k-\gn\ga)-s }}{1+\ee^{-t(k-\gn\ga)-S }},
$$
and for $s,S$ in any of the regiems from Assumptions~\ref{assumpt:parameterregimes} standard arguments show that this sum is $\Boh(\ee^{-\eta\gn})$, for some $\eta>0$. Therefore
$$
\log\frac{\msf L_\gn(s)}{\msf L_\gn(S)}=-\msf S(S,s) -\msf R(S,s)+\Boh\left(\ee^{-\eta \gn}\right),\quad \gn\to \infty,
$$
for some $\eta>0$, uniformly for $s,S$ in any of the regimes from Assumptions~\ref{assumpt:parameterregimes}. The proof now follows directly from the estimates for $\msf R$ and $\msf S$  provided by Proposition~\ref{prop:estRfinal} and Theorems~\ref{thm:Ssubfinal}, \ref{thm:Scritfinal} and \ref{thm:Ssuperfinal}.
\end{proof}

We are now ready to prove Theorem~\ref{thm:integrated_formal}.

Start first with $s$ in the subcritical regime, and choose $\delta>0$ such that $S\deff\delta\gn^{1/2}>s$. Applying Theorem~\ref{thm:integrated_formal_pre}--(i) and using \eqref{eq:asymptAirydiagint}, we obtain
$$
\msf L_\gn(s)=\msf L_\gn(\delta\gn^{1/2})-\left(1+\Boh\left(\ee^{-\eta s/\gn^{1/3}}\right)+\Boh\left(\gn^{-\nu}\right)\right)\int_{\msf c_\gV s/(t\gn^{1/3})}^\infty \msf A(y,y)\dd y.
$$
The term $\msf L_\gn(\delta\gn^{1/2})$ is now estimated by Corollary~\ref{cor:roughbound}, which concludes the proof of Theorem~\ref{thm:integrated_formal}--(i).

Next, assume that $s$ is in the critical regime. Choose $M>0$ such that $s<S\deff M\gn^{1/3}$, and make sure also that $S$ belongs to both the sub and critical regimes. Relation~\eqref{eq:FGUEHamilt} and asymptotics \eqref{eq:tailsTW} ensures $\msf K_{\ptf}(0,0\mid y)$ decays exponentially fast as $y\to +\infty$, and from Theorem~\ref{thm:integrated_formal_pre}--(ii) we obtain
$$
\msf L_\gn(s)=\msf L_\gn(M\gn^{1/3})-\left(1-\Boh\left(\gn^{-\nu}\right)\right)\int_{\msf c_\gV s/(t\gn^{1/3})}^\infty \msf K_{\ptf}(0,0\mid y)\dd y,\quad \gn\to\infty.
$$
We now apply Theorem~\ref{thm:integrated_formal}--(i) to the term $\msf L_\gn(M\gn^{1/3})$, which ensures that it decays exponentially fast as $\gn\to\infty$, and therefore it can be suppressed as an error term in the above. The proof of Theorem~\ref{thm:integrated_formal}--(ii) is then completed using \eqref{eq:FGUEHamilt}.

Assume at last that $s$ is in the supercritical regime. Choose $M>0$ ensuring that $S\deff -M\gn^{1/3}>s$ and that it is in both the critical and supercritical regimes. This way, the integral
$$
\int_{\msf c_\gV S/(t\gn^{1/3})}^\infty \msf K_\ptf(0,0\mid y)\dd y=\int_{-\msf c_\gV M/t}^\infty \msf K_\ptf(0,0\mid y)\dd y
$$
is finite, and Theorem~\ref{thm:integrated_formal}--(ii) then ensures that $\msf L_\gn(-M\gn^{1/3})$ is $\Boh(1)$. Theorem~\ref{thm:integrated_formal}--(iii) then follows from Theorem~\ref{thm:integrated_formal_pre}--(iii).

\appendix

\section{A useful lemma}

The following lemma is extremely simple, but it turned out out to be useful in certain steps in the main text.

\begin{lemma}\label{lem:basicestimate}
Fix $r>0$. There exists $M=M(r)>0$ such that
$$
\left|\frac{1}{1-w}-1\right|\leq M|w|\quad \text{and}\quad \frac{1}{|1-w|}\leq \frac{M}{|w|},
$$
for every $w\in \C$ which satisfies
$$
|1-w|\geq r.
$$
\end{lemma}
\begin{proof}
For the first claimed inequality, simply write
$$
\left|\frac{1}{1-w}-1\right|=\left| \frac{w}{1-w} \right|\leq \frac{1}{r}|w|.
$$

The second claimed inequality is equivalent to showing that
$$
\left|\frac{w}{1-w}\right|\leq M.
$$
Suppose first that $|w|\leq 2$. Then this inequality is trivially satisfied for $M=2/r$.

Suppose now that $|w|\geq 2$. We write
$$
\left|\frac{w}{1-w}\right|=\frac{1}{|1-w^{-1}|},
$$
and we bound the denominator in the right-hand side as follows
$$
\left|1-\frac{1}{w}\right|\geq 1-\frac{1}{|w|}\geq \frac{1}{2}.
$$
Hence, the final result follows with the choice $M=\max\{ 2/r,2 \}$.
\end{proof}

\section{Residue conditions in discrete RHPs}

In the course of the paper, we came across several RHPs with residue conditions. We now book-keep a lemma which was useful in transiting between different formulations of residue conditions that appeared.

\begin{lemma}\label{lem:resRHPaux}
Let $\bm X$ be an $m\times m$ matrix-valued function with a simple pole at a point $p\in \C$ and analytic elsewhere near $p$, and let $\bm R$ be a fixed matrix with $\bm R^2=0$. Furthermore, let $\zeta=\varphi(z)$ be a conformal map from a neighborhood of the point $p$ to a neighborhood of a point $q=\varphi(p)$. Set $\bm \Phi(\zeta)\deff \bm X(z)$, which has a simple pole at $\zeta=q$ and it is analytic elsewhere near $q$. 

The following statements are equivalent.
\begin{enumerate}[(i)]
    \item The product $\bm X(z)\left(\bm I-\frac{1}{z-p}\bm R\right)$ is analytic at $z=p$.
    \item $\displaystyle{\res_{z=p}\,\bm X(z)=\lim_{z\to p}\bm X(z)\bm R}$. 
    \item The product $\bm \Phi(\zeta)\left( \bm I-\frac{\varphi'(p)}{\zeta-q}\bm R \right)$ is analytic at $\zeta=q$.
    \item $\displaystyle{\res_{\zeta=q}\;\bm \Phi(\zeta)=\varphi'(p)\lim_{\zeta\to q}\bm \Phi(z)\bm R}$.
\end{enumerate}
\end{lemma}
\begin{proof}
The equivalence of (i)--(ii) is essentially in \cite[Lemma~4.4]{Borodin2000Bessel}, see also \cite[Lemma~7.1]{BaikLiuSilva} where the equivalence statement is explicitly written. We reproduce the proof from \cite{BaikLiuSilva} because it will be instructive for the other parts.

Without loss of generality we assume $p=0$. Write
\begin{equation}\label{appeq:Xexpansion}
\bm X(z)=\frac{\bm X_{-1}}{z}+\bm X_0+\Boh(z),\quad z\to 0.
\end{equation}
A direct expansion shows that
$$
\bm X(z)\left(\bm I-\frac{1}{z-p}\bm R\right)=-\frac{\bm X_{-1}\bm R}{z^2}+\frac{\bm X_{-1}-\bm X_0\bm R}{z}+\Boh(1),\quad z\to 0.
$$
Hence, (ii) is equivalent of saying that $\bm X_{-1}\bm R=0$ and $\bm X_{-1}=\bm X_0\bm R$. Under the condition $\bm R^2=0$, these two equalities are equivalent to the single equality $\bm X_{-1}=\bm X_0\bm R$, which in turn is equivalent to (i).

Next, the equivalence between (iii) and (iv) is the same as the equivalence between (i) and (ii). Consequently, once we prove that (ii) implies (iv), then we get that (iii) implies (i) and all equivalences are established. 

To prove that (ii) implies (iv), we compute from \eqref{appeq:Xexpansion} (already using (i) as well)
$$
\bm \Phi(\zeta)=\frac{\bm X_0\bm R}{\varphi^{-1}(z)}+\bm X_0+\Boh(\zeta-q),\quad \zeta\to q.
$$
Therefore, using that $\bm R^2$,
$$
\bm \Phi(\zeta)\left(\bm I-\frac{\varphi'(0)}{\zeta-q}\bm R\right)=\left( \frac{1}{\varphi^{-1}(\zeta)}-\frac{\varphi'(0)}{\zeta-q} \right)\bm X_0\bm R+\Boh(1),\quad \zeta\to q.
$$
The result now follows from the calculation
$$
\varphi^{-1}(\zeta)=\frac{1}{\varphi'(0)}(\zeta-q)(1+\Boh(\zeta-q)),\; \zeta\to q, \quad \text{which implies}\quad \frac{1}{\varphi^{-1}(\zeta)}-\frac{\varphi'(0)}{\zeta-q}=\Boh(1),\quad \zeta\to q.
$$

\end{proof}

\section{The Airy parametrix}\label{sec:AiryRHP}

With $\ai(\zeta)$ being the Airy function, introduce
\begin{equation}\label{eq:bai1}
\bai_0(\zeta)\deff \sqrt{2\pi}\ee^{-\pi \ii /12}
\begin{pmatrix}
    \ai(\zeta) & \ai(\ee^{4\pi \ii /3} \zeta) \\
    \ai'(\zeta) & \ee^{4\pi \ii /3} \ai'(\ee^{4\pi \ii /3}\zeta)
\end{pmatrix} \ee^{-\pi \ii \sp_3/6},
\end{equation}
which satisfies $\det \bai_0(\zeta)\equiv 1$, and assemble
\begin{equation}\label{eq:bai2}
\bai(\zeta)\deff
\ee^{\pi \ii \sp_3/4}\bai_0(\zeta) \times
\begin{cases}
\bm I, & 0<\arg \zeta<\frac{2\pi}{3}, \\
(\bm I-\bm E_{21}), & \frac{2\pi}{3}<\arg \zeta<\pi, \\ 
(\bm I-\bm E_{12}), & -\frac{2\pi}{3}<\arg\zeta<0, \\ 
(\bm I-\bm E_{12})(\bm I+\bm E_{21}), & -\pi<\arg\zeta <-\frac{2\pi }{3}.  
\end{cases}
\end{equation}

With
$$
\Gamma_\bai\deff \Gamma_0^0\cup \Gamma_2^0\cup \Gamma_3^0\cup \Gamma_4^0,
$$
the matrix ${\bm A}$ is introduced so as to solve the following RHP.

\begin{rhp}\label{rhp:airy}
Find a $2\times 2$ matrix-valued function $\bai$ with the following properties.
\begin{enumerate}[(1)]
\item $\bai$ is analytic on $\C\setminus \Gamma_{\bai}$.
\item The matrix $\bai$ has continuous boundary values $\bai_\pm$ along $\Gamma_{\bai}$, and they are related by ${\bai}_+(\zeta)={\bai}_-(\zeta)\bm J_{{\bai}}(\zeta)$, $\zeta\in \Gamma_{\bai}$, with
\begin{equation}\label{eq:jumpsAiry}
\bm J_{\bai}(\zeta)\deff
\begin{dcases}
\bm I+\bm E_{12}, & \zeta\in \Gamma_0^{0}=(0,+\infty),\\
\bm I+\bm E_{21}, & \zeta\in \Gamma_2^{0}\cup \Gamma_4^{0}, \\
\bm E_{12}-\bm E_{21}, & \zeta\in \Gamma_3^{0}=(-\infty,0).
\end{dcases}
\end{equation}
\item As $\zeta\to \infty$, 
\begin{equation}\label{eq:asympAiryparam}
\bai(\zeta)=\left(\bm I+\Boh(\zeta^{-1})\right)\zeta^{-\sp_3/4}\bm U_0\ee^{-\frac{2}{3}\zeta^{3/2}\sp_3}.
\end{equation}
\item The matrix $\bai$ remains bounded as $\zeta\to 0$.
\end{enumerate}
\end{rhp}

To our knowledge, this RHP was first considered in \cite{DMKVZ99OPs}, where they also showed that the solution is precisely given by $\bai$ as defined in \eqref{eq:bai1}--\eqref{eq:bai2}, see \cite[Lemma~7.1 and Equation~(7.30)]{DMKVZ99OPs}\footnote{Our construction of $\bai$ matches the matrix $\Psi^\sigma$ from \cite[Equation~(7.9)]{DMKVZ99OPs} through the simple transformation $\Psi^\sigma=\frac{1}{\sqrt{2\pi}} \ee^{\pi \ii /12}\ee^{-\pi \ii \sp_3/4}\bai$.}. For the record, the asymptotic condition \eqref{eq:asympAiryparam} can alternatively be expressed as
$$
\bai(\zeta)=\zeta^{-\sp_3/4}\bm U_0\left(\bm I+\Boh(\zeta^{-3/2})\right)\ee^{-\frac{2}{3}\zeta^{3/2}\sp_3},\quad \zeta\to \infty.
$$

\bibliographystyle{abbrv}  
\bibliography{bibliography}

\end{document}